\documentclass[10pt]{book}
\usepackage[height=8.85in,width=6.45in]{geometry}
\usepackage{macros}
\usepackage{box_environment}
\graphicspath{{./Figures/}}

\usepackage{times}

\newcommand{\oGamma}{\overline{\Gamma}}

\newcommand{\sL}{\scalebox{0.6}{\text{L}}}
\newcommand{\sR}{\scalebox{0.6}{\text{R}}}

\usepackage{CJKutf8}

\usepackage{arydshln}

\usepackage{fancyhdr}
\usepackage[multiple]{footmisc}

\begin{document}

\begin{titlepage}
    \vspace*{10pt}
	\begin{center}
		\LARGE Doctoral Dissertation\vspace{2mm}\\
		\begin{CJK}{UTF8}{ipxm}
		\LARGE	博士論文
		\end{CJK}
	\end{center}
	
	\bigskip
	\vskip 1cm
	
	\begin{center}
	{\LARGE {
		Gauge Origami and BPS/CFT correspondence}\\
		\vspace*{1pt}
		\begin{CJK}{UTF8}{ipxm}
			\LARGE
			{(ゲージ折紙とBPS/CFT対応)}
		\end{CJK}
	}\\
		\vspace*{50pt}

	\LARGE A Dissertation Submitted for the Degree of Doctor of Philosophy\\ 
    \LARGE December 2024\vspace{2mm}\\
	\begin{CJK}{UTF8}{ipxm}
	\LARGE	令和 6 年 12 月 博士 (理学) 申請
	\end{CJK}
	
	\vspace*{60pt}
	
	\LARGE Department of Physics, Graduate School of Science,\\
    The University of Tokyo\vspace{2mm}\\
	\begin{CJK}{UTF8}{ipxm}
		東京大学大学院理学系研究科 \\
        物理学専攻
	\end{CJK}
	
	\bigskip
	\vspace*{60pt}
	\LARGE
	Go Noshita\vspace{2mm}\\
	\begin{CJK}{UTF8}{ipxm}
	 \LARGE  野下　剛
	\end{CJK}
	\end{center}
	
\end{titlepage}

\cleardoublepage
\thispagestyle{empty}

\vspace*{40pt}
\begin{center}
	\textbf{Abstract}
\end{center}
Gauge origami is a generalized supersymmetric gauge theory defined on several intersecting space-time components. It provides a systematic way to consider generalizations of instantons. In this thesis, we explore the gauge origami system in $\mathbb{R}^{1,1}\times \mathbb{C}^{4}$ and its BPS/CFT correspondence. String theoretically, instantons of the gauge origami system arise from D0-branes bound to D$(2p)$-branes wrapping cycles in $\mathbb{C}^{4}$. The low energy theory of the D0-branes is understood as an $\mathcal{N}=2$ supersymmetric quiver quantum mechanical system and the Witten index of it produces the instanton partition function. We define a $q$-deformed quiver Cartan matrix associated to this quiver structure and introduce vertex operators associated with the D-branes and show that the contour integral formula for the Witten index has a nice free field realization. Such free field realization leads to the concept of BPS $qq$-characters or BPS quiver W-algebras, which are generalizations of the conventional deformed W-algebras. The $qq$-characters of D2 and D4-branes correspond to screening charges and generators of the affine quiver W-algebra, respectively. On the other hand, the $qq$-characters of D6 and D8-branes represent novel types of $qq$-characters, where monomial terms are characterized by plane partitions and solid partitions. The composition of these $qq$-characters yields the instanton partition functions of the gauge origami system, eventually establishing the BPS/CFT correspondence. Additionally, we demonstrate that the fusion of $qq$-characters of D-branes in lower dimensions results in higher-dimensional D-brane $qq$-characters. We also investigate quadratic relations among these $qq$-characters. Furthermore, we explore the relationship with the representations of the quantum toroidal~$\mathfrak{gl}_{1}$.

\setcounter{tocdepth}{2}
\setcounter{secnumdepth}{3}
\tableofcontents
\thispagestyle{empty}
\addtocontents{toc}{\protect\thispagestyle{empty}}


\chapter{Introduction}\label{chap:intro}

\paragraph{Rise of BPS/CFT correspondence}
Towards a non-perturbative understanding of string theory and quantum field theory, exact computations of physical observables have played significant roles. Among these observables, the instanton partition functions of supersymmetric gauge theories have emerged as significant quantities that can be precisely computed through localization techniques~\cite{Nekrasov:2002qd,Nekrasov:2003rj,Nakajima:2003pg,Nakajima:2003uh,Nakajima:2005fg,Pestun:2007rz,Pestun:2016zxk}. These partition functions exhibit remarkable characteristics, particularly in the realm of rich geometric and algebraic properties. Our specific interest lies in the algebraic structures of these partition functions, which have, in turn, led to the revelation of a novel duality known as the BPS/CFT correspondence~\cite{Nakajima:1999,Nekrasov:2013xda,Nekrasov:2015wsu,Nekrasov:2016qym,Nekrasov:2016ydq,Nekrasov:2017rqy,Nekrasov:2017gzb,Nekrasov:2017cih,Nekrasov:2018xsb}. The BPS/CFT correspondence claims that correlation functions of BPS observables in supersymmetric gauge theories are dual to correlation functions of infinite-dimensional algebras. A well-known example of this duality is the Alday--Gaiotto--Tachikawa (AGT) duality~\cite{Gaiotto:2009we,Alday:2009aq,Wyllard:2009hg} (see also~\cite{LeFloch:2020uop} for a review), which established a connection between instanton partition functions of 4d $\mathcal{N}=2$ theories and conformal blocks of 2d conformal field theories (CFT), where both theories originate from a 6d $\mathcal{N}=(2,0)$ theory compactified on a Riemann surface. A 5d lift-up of the AGT duality called the 5d AGT duality~\cite{Awata:2009ur,Awata:2010yy,Taki:2014fva}, was subsequently discovered and the partition functions are dual to correlation functions of quantum algebras. Over the last decade, efforts have been dedicated to generalizing both the gauge theory side and the infinite-dimensional algebra side to achieve a more comprehensive understanding of the BPS/CFT correspondence~\cite{miki2007q,FFJMM1,Litvinov:2016mgi,Gaiotto:2017euk,Prochazka:2017qum,Prochazka:2018tlo,Rapcak:2018nsl,Rapcak:2020ueh,Li:2020rij,Galakhov:2021vbo,Noshita:2021ldl}.   

\paragraph{Generalized gauge theory}
The usual instantons appear as topologically nontrivial field configurations minimizing the action of the Yang--Mills theory on $\mathbb{R}^{4}$ and the moduli space comes from the ADHM construction \cite{Atiyah:1978ri,Atiyah:1984px}. After introducing a nontrivial background flux called $\Omega$-background, Nekrasov showed that the instanton partition function takes the form of $\mathcal{Z}=\sum_{\lambda}\mathfrak{q}^{|\lambda|}\mathcal{Z}[\lambda]$ where the terms are a summation of two-dimensional partitions \cite{Lossev:1997bz,Moore:1997dj,Nekrasov:2002qd,Nekrasov:2003rj}. Over the last few years, researches on generalizations of the Yang--Mills instantons have been conducted and we now know that there are higher dimensional generalizations of instantons having ADHM-like constructions \cite{Acharya:1997gp,Baulieu:1997jx,Nekrasov:2009JJM,Szabo:2022zyn,Nekrasov:2023xzm}. The instanton partition functions are given as a statistical sum of random partitions as $\mathcal{Z}=\sum_{\Lambda}\mathfrak{q}^{|\Lambda|}\mathcal{Z}[\Lambda]$, where $\Lambda$ is a random partition. Considering 6d and 8d theories, the random partitions $\Lambda$ will be plane partitions \cite{Jafferis:2007sg,Cirafici:2008sn,Nekrasov:2009JJM,Awata:2009dd} and solid partitions \cite{Nekrasov:2017cih,Nekrasov:2018xsb}, respectively (see also \cite{Kanno:2020ybd}). In type IIB string theory on $\mathbb{R}^{10}$, such kind of setup appears as the low energy limit of the worldvolume theory of the D1-branes probing the $\D(2p+1)$-branes $(p=2,3,4)$~\cite{Witten:1994tz,Witten:1995gx,Douglas:1995bn,Douglas:1996uz}: 2d ($p=2$), 3d ($p=3$), 4d $(p=4)$ partitions. Mathematically, the $p=3$ case gives the equivariant Donaldson--Thomas invariants of $\mathbb{C}^{3}$, while the $p=4$ case is called the magnificent four model and gives the equivariant Donaldson--Thomas invariants of $\mathbb{C}^{4}$.

Another direction of generalizations of the Yang--Mills instantons is the \emph{generalized gauge theory}~\cite{Nekrasov:2015wsu,Nekrasov:2016qym,Nekrasov:2016ydq}, which is a theory defined on several, generally, intersecting components as $\mathcal{S}=\bigcup_{i}\mathcal{S}_{i}$ where on each space-time component $\mathcal{S}_{i}$ there is an original field theory. On each space-time component $\mathcal{S}_{i}$, we have a gauge group $G_{i}$ and at the intersection $\mathcal{S}_{i}\cap\mathcal{S}_{j}$ we have bifundamental fields transforming under $G_{i}\times G_{j}$, and thus it could be understood as a generalized quiver gauge theory. From the viewpoint of each field theory on $\mathcal{S}_{i}$, the intersection of other components plays the role of defects. The first example of such general gauge theory is the so-called spiked instanton system, which was introduced in~\cite{Nekrasov:2015wsu,Nekrasov:2016qym}. The spiked instantons arise from the low energy limit of D1-branes probing \emph{intersecting} D5 (and anti-D5)-branes. Later, it was generalized to D1-branes probing intersecting D7-branes in~\cite{Pomoni:2021hkn} (see also~\cite{Fasola:2023ypx,Cao:2023lon}), and the arising instantons are called the tetrahedron instantons. Note that these instantons are generalizations of the higher dimensional instantons introduced in the previous paragraph. 

\paragraph{Gauge origami}
The setups discussed previously are collectively called the \emph{gauge origami} and the arising partition function is called the \emph{gauge origami partition function} \cite{Nekrasov:2016ydq}. Consider a type IIB theory on $Z\times \mathcal{C}$ where $Z$ is a toric Calabi--Yau four-fold and $\mathcal{C}=\mathbb{C},\,\mathbb{C}^{\times},\,\mathbb{T}^{2}$ and the low energy limit of the D1-branes probing $\D(2p+1)$-branes. The D1-branes wrap $\mathcal{C}$ while the $\D(2p+1)$-branes wrap the product of $\mathcal{C}$ and non-compact toric submanifolds of $Z$ in a way preserving a suitable number of supersymmetries. Depending on $\mathcal{C}$, the arising partition function becomes rational, trigonometric, and elliptic, respectively. The gauge origami partition function generally takes the form as\footnote{Origami is ``折紙" in Japanese.}
\bea
\mathcal{Z}_{\text{折紙}}=\sum_{\{\Lambda_{i}^{(\alpha)}\}}\underline{\mathfrak{q}}^{|\underline{\vec{\Lambda}}|}\prod_{(i,\alpha)}\mathcal{Z}[\Lambda_{i}^{(\alpha)}]\prod_{(i,\alpha)\neq (j,\beta)}\mathcal{Z}(\Lambda_{i}^{(\alpha)}\,|\,\Lambda_{j}^{(\beta)})
\eea
where $i$ labels the possible types of toric submanifolds and $\alpha$ labels the number of D-branes wrapping them. The partition function is a summation of random \emph{BPS crystals}, which are generalizations of the partitions, and they are denoted $\Lambda^{(\alpha)}_{i}$. These crystals are expected to be, generally, truncations of four-dimensional BPS crystals which are generalizations of the three-dimensional BPS crystals \cite{Szendroi:2007nu,Mozgovoy2008OnTN,Ooguri:2009ijd,Nagao:2010kx} (see \cite{ Bao:2024ygr, Franco:2023tly} for recent discussions). The $\mathcal{Z}[\Lambda_{i}^{(\alpha)}]$ part comes from the contribution of each $\D(2p+1)$-branes while the $\mathcal{Z}(\Lambda_{i}^{(\alpha)}\,|\,\Lambda_{j}^{(\beta)})$ part comes from the bifundamental contributions at the junctions. In this paper, we mainly focus on the case when $Z=\mathbb{C}^{4}$ and $\mathcal{C}=\mathbb{C}^{\times }\simeq \mathbb{R}\times \mathbb{S}^{1}$ which gives the K-theoretic magnificent four \cite{Nekrasov:2017cih,Nekrasov:2018xsb}, tetrahedron instanton \cite{Pomoni:2021hkn}, and spiked instanton \cite{Nekrasov:2016gud} setups. From the string theory viewpoint, we take the T-duality of the D1--$\D(2p+1)$ system and consider a $\D0$--$\D(2p)$ system, where each $\D(2p)$-brane gives a $(2p+1)$-dimensional gauge theory and the $\D0$-branes play the roles of instantons. 

\paragraph{Quantum algebra of gauge origami}
An interesting property of the gauge origami partition function is the existence of an infinite set of non-perturbative Dyson--Schwinger equations related to the symmetries of adding and removing instantons~\cite{Nekrasov:2015wsu}. The $qq$-characters are physical observables characterizing them, and interestingly, there is an operator formalism of them called the quiver W-algebra \cite{Kimura:2015rgi,Kimura:2016dys,Kimura:2017hez,Kimura:2019xzj} (see \cite{Kimura:2020jxl} for a review). From the gauge theoretic viewpoint, such algebras are associated with the Dynkin diagram corresponding to the quiver structure of the gauge theory. In the gauge origami formalism, the $qq$-characters and quiver W-algebras appear from the gauge origami system in $Z=\mathbb{C}^{2}\times \mathbb{C}^{2}/\Gamma$, where $\Gamma$ denotes the finite subgroup of $\mathrm{SU}(2)$ associated with the quiver structure through the McKay correspondence. Placing $\D4$-branes wrapping the $\mathbb{C}^{2}$ part gives the 5d $\mathcal{N}=1$ (affine) quiver gauge theories and the $\D4$-branes wrapping the $\mathbb{C}^{2}/\Gamma$ give the $qq$-characters or the quiver W-algebras of the theory. Physically, they are codimension four defects of the quiver gauge theory \cite{Kim:2016qqs}.

One of the reasons why such algebras are considered important is because they give the BPS/CFT correspondence for the present case. One can construct screening currents from the quiver structure and the vacuum expectation value of them gives the instanton partition function of the quiver gauge theory. Moreover, after defining the highest weight, the commutativity with the screening charges determines the generator of the quiver W-algebra uniquely and they are the operator version of the $qq$-characters. These quiver W-algebras include the well-known $q$-Virasoro \cite{Shiraishi:1995rp}, $q$-W$_{N}$ \cite{Awata:1995zk,Awata:1996dx}, and Frenkel--Reshetkhin's deformed W-algebras \cite{Frenkel:1997CMP,Frenkel:1998ojj}.

Although the quiver W-algebras indeed gave a way to discuss the BPS/CFT correspondence from the quantum algebraic viewpoint, the applicable theory is still limited and needs to be extended. For example, while the $qq$-characters associated with $Z=\mathbb{C}^{2}/\Upsilon\times \mathbb{C}^{2}/\Gamma$ were studied in terms of partition functions in \cite{Nekrasov:2012xe,Nekrasov:2013xda,Nekrasov:2015wsu,Jeong:2018qpc,Jeong:2021rll}, it seems that the complete operator formalism of such cases is still missing in the literature. Moreover, the quiver W-algebra is only applicable to discuss two stacks of D-branes in transverse directions while we have multiple intersecting D-branes in the gauge origami system. Based on recent studies such as the tetrahedron instanton system we also have D-branes wrapping not only complex two-dimensional surfaces but also complex three-dimensional manifolds. Most importantly, we need to generalize the operator formalism to describe gauge origami systems associated with general toric Calabi--Yau four-folds~$Z$. See \cite{Cao:2019tvv,Kimura:2022zsm,Cao:2023gvn,Szabo:2023ixw,Nekrasov:2023nai} for discussions along this direction. 

The goal of this thesis is to fill in this gap by generalizing the concept of quiver W-algebras and showing the BPS/CFT correspondence of the gauge origami system associated with $\mathbb{C}^{4}$. Some generalizations to toric Calabi--Yau four-folds were already proposed in \cite{Kimura:2023bxy} (see also \cite{Bao:2024ygr, Franco:2023tly}) but in this thesis we focus only on $\mathbb{C}^{4}$.


\section*{Organization and summary of this thesis}

The rest of this thesis consists of six chapters and five appendices. Chap.~\ref{chap:ADHM-localization}, \ref{chap:gauge-origami}, \ref{chap:gaugeorigamipartitionfunction} are preliminaries while Chap.~\ref{chap:freefield-vertexop} and \ref{chap:quantum-algebra-BPSqq} are the original part which are based on the author's papers \cite{Kimura:2023bxy, Kimura:2024osv}. Appendix~\ref{app:spinor_susy}, \ref{app-chap:openstring}, \ref{app:specialfunct}, \ref{app-chap:Partition-formulas}, and \ref{app-chap:proofs} are supplementary materials.

In Chap.~\ref{chap:ADHM-localization}, we review basic aspects of instanton counting and supersymmetric localization using the pure super Yang--Mills theory as a toy model. The ADHM data, contour integral formulas, JK residues and Nekrasov factors are reviewed there. In Chap.~\ref{chap:gauge-origami}, we move on to the gauge origami system and discuss generalizations of the ADHM instantons. Such instantons appear in string theory as low energy limits of D1-branes probing the D($2p+1$)-branes. We discuss the number of preserved supersymmetries, open string spectrum, and show that the low energy limit of the D1-branes can be described by a 2d $\mathcal{N}=(0,2)$ quiver gauge theory. Moreover, we show that the vacuum moduli space of such quiver gauge theories give the instanton moduli space. We then take the T-duality and consider the 1d $\mathcal{N}=2$ supersymmetric quantum mechanics and evaluate the Witten index in Chap.~\ref{chap:gaugeorigamipartitionfunction}. Similar to the pure SYM case, after evaluation of the contour integral formulas, the poles are classified by multi-dimensional partitions, which are generalizations of Young diagrams. We also study the properties of the gauge origami partition functions. In Chap.~\ref{chap:freefield-vertexop}, we propose a systematical way to define vertex operators reproducing the contour integral formulas of a 1d $\mathcal{N}=2$ supersymmetric quiver quantum mechanics. The data of these vertex operators are all encoded in a matrix which we call \textit{qquiver Cartan matrix} (see section~\ref{sec:vertexop-quiver} for the definitions). This matrix includes the information of the quiver diagram and the associated flavor charges of the matter components in the low energy field theory. We then apply this formalism to the gauge origami system and derive the free field realizations of the contour integral formulas. Such free field realization implies the existence of a quantum algebraic structure. In Chap.~\ref{chap:quantum-algebra-BPSqq}, we introduce the BPS $qq$-characters and show the BPS/CFT correspondence of the gauge origami system. Furthermore, we study the quantum algebraic properties of these BPS $qq$-characters. Finally, in Chap.~\ref{chap:conclusion}, we give the conclusion.

\subsection*{Special notes}
This thesis is based on the following papers.
\begin{center}
\begin{tabular}{llllll}
 \cite{Kimura:2023bxy} &\,\,T.\,Kimura and G.\,Noshita, ``Gauge origami and quiver W-algebras",\,\, \\
    &\href{https://arxiv.org/abs/2310.08545}{\magenta{arXiv:2310.08545 [hep-th]}}, \,\,\href{https://link.springer.com/article/10.1007/JHEP05(2024)208}{{JHEP05(2024)208}}\\
    &\\
    \cite{Kimura:2024osv}& T.\,Kimura and G.\,Noshita, ``Gauge origami and quiver W-algebras III: Donaldson--Thomas $qq$-characters",\\ &\href{https://arxiv.org/abs/2411.01987}{\magenta{arXiv:2411.01987 [hep-th]}}
\end{tabular}
\end{center}

\subsection*{Summary of main results}

\paragraph{QQuiver Cartan matrix}We give the definition of the $q$-deformed quiver Cartan matrix and propose a general way to derive vertex operators given a 2 SUSY quiver.

{\begin{defmod}[Def.~\ref{def:totalqquiverCartan}]
    We define the total qquiver Cartan matrix $c=(c_{ab})_{a,b\in\overline{Q}_{0}}$ as
    \bea
    c_{ab}&=c_{ab}^{+}+c_{ab}^{-}=2\delta_{ab}-\left(\sum_{b\rightarrow a}q(\Phi_{b\rightarrow a})+\sum_{a\rightarrow b}q(\Phi_{a\rightarrow b})^{-1} \right)\\
    &\qquad +\left(\sum_{b\textcolor{red}{\rightarrow} a}q(\Lambda_{b\textcolor{red}{\rightarrow} a})+\sum_{a\textcolor{red}{\rightarrow} b}q(\Lambda_{a\textcolor{red}{\rightarrow} b})^{-1}\right).
    \eea
\end{defmod}

\paragraph{Definition of D-brane vertex operators}
We introduce vertex operators $\mathsf{A}(x),\mathsf{S}_{a}(x),\mathsf{X}_{A}(x),\mathsf{W}_{\bar{a}}(x),\mathsf{Z}(x)$ corresponding to D0, D2, D4, D6, and D8-branes wrapping the possible subspaces $\text{pt}$, $\mathbb{C}_{a}$, $\mathbb{C}_{A}^{2}$, $\mathbb{C}_{\bar{a}}^{3}$, $\mathbb{C}^{4}$ ($a\in\four,A\in\six$), respectively.
The main result of section~\ref{sec:freefieldintegral} is as follows.\footnote{In this thesis, we only focus on the D4, D6, and D8 cases. Discussion on the D2-branes are left for future work.}
\begin{theorem}[Thm.~\ref{thm:freefieldconclusion}]
For each D-brane (D0, D2, D4, D6, D8), we define the corresponding vertex operators as
\begin{align}
    \renewcommand\arraystretch{1.2}{
        \begin{tabular}{|c|c|c|}\hline
            D-brane & space-time & vertex operator \\
           \hline\hline  D0-brane  & $\text{pt}\times\mathbb{S}^{1}$ & $\mathsf{A}(x)$ \\
           \hline D2-brane   &  $\mathbb{C}_{a}\times \mathbb{S}^{1}$ ($a\in\four$)  & $\mathsf{S}_{a}(x)$\\
           \hline D4-brane &  $\mathbb{C}^{2}_{A}\times \mathbb{S}^{1}$ ($A\in\six$) & $\mathsf{X}_{A}(x)$\\
           \hline D6-brane &  $\mathbb{C}^{3}_{\bar{a}}\times \mathbb{S}^{1}$ ($a\in\four$) &  $\mathsf{W}_{\bar{a}}(x)$\\
           \hline D8-brane & $\mathbb{C}^{4}\times \mathbb{S}^{1}$  & $\mathsf{Z}(x)$ \\\hline
         \end{tabular}}
\end{align}
We have multiple copies of vertex operators if there are multiple ways that the D-branes can wrap. Then, the contour integral formula of the $k$-instanton partition function takes the form as
\bea
    \mathcal{Z}_k =\frac{\mathcal{G}^{k}}{k!} \oint \prod_{I=1}^{k}\frac{dx_{I}}{2\pi i x_{I}} \bra{0} \prod_{I=1}^{k}\mathsf{A}(x_{I})^{-1} :\prod_{i}\mathsf{V}_{i}(v_{i}): \ket{0},
\eea
where $\mathsf{V}_{i}(x)$ is an operator written from $\{\mathsf{X}_{A}(x),\mathsf{W}_{\bar{a}}(x),\mathsf{Z}(x)\}$.
\end{theorem}

\paragraph{Operator formalism of $qq$-characters}
We introduce the operator formalism of the $qq$-characters of the gauge origami system and show the BPS/CFT correspondence. For each D$(2p)$-brane, we can associate a $qq$-character. The D2 $qq$-characters $\mathscr{Q}_{a}(x)\,(a\in\four)$ are four copies of the screening charge of the $\widehat{A}_{0}$ quiver W-algebra \cite{Kimura:2015rgi} and the D4 $qq$-characters $\mathsf{T}_{A}(x)\,(A\in\six)$ are six copies of the generator of the $\widehat{A}_{0}$ quiver W-algebra. The D6 $qq$-characters $\mathsf{T}_{\bar{a}}(x)\,(a\in\four)$ are the new $qq$-characters where the monomial terms are labeled by plane partitions. These $qq$-characters represent the quantum algebras associated with complex $1,2,3$-dimensional submanifolds. We will see that their compositions indeed give the gauge origami partition function which shows the BPS/CFT correspondence. 
\begin{theorem}[Thm.~\ref{thm:spiked-qq-BPSCFT}, \ref{thm:tetra-origamiBPSCFT}]\label{thm:BPS/CFTintro}
The gauge origami partition function is in general given as a correlation function of the $qq$-characters,
\bea
    \mathcal{Z}_{\text{折紙}}=\bra{0}\prod_{(i,\alpha)}\mathsf{T}_{i}(x_{i,\alpha})\ket{0}
\eea
where $i$ takes values in $i\in\four^{\vee},\six $ and $i=1234$.
\end{theorem}

In the context of quiver W-algebra and also the quantum integrable system, the commutation relation between the vertex operators plays a fundamental role.
We obtain an interesting commutativity among $qq$-characters: The D2, D4, D6 $qq$-characters associated with subspaces of $\mathbb{C}^{4}$ that are transverse with each other commutes.  We also give a conjecture claiming that the D6 and D8 $qq$-characters all commute with each other.
\begin{theorem}[Thm.~\ref{thm:tetrascreening}]
The $qq$-characters associated with the elements $i,j\in\four\oplus\six\oplus\four^{\vee}$ commute with each other up to trivial zero modes (see \eqref{eq:weakcommute}) when $i$ and $j$ are transverse with each other:
\beq
    \mathsf{T}_{i}(x)\mathsf{T}_{j}(x')-f_{ij}(x,x')\mathsf{T}_{j}(x')\mathsf{T}_{i}(x)=0 \quad \Longleftrightarrow \quad i \cap j =\emptyset,
\eeq
where $f_{ij}(x,x')$ are zero mode factors.
\end{theorem}    

\begin{conjecture}[Conj.~\ref{conj:D8commutativity}]
The quadratic relations of the D8 $qq$-characters are
    \bea
\mathsf{f}^{K_{2}|K_{1}}_{\four\four}\left(x_{2}/x_{1}\right)\mathsf{T}_{\four}^{K_{1}}(x_{1})\mathsf{T}^{K_{2}}_{\four}(x_{2})-\mathsf{f}^{K_{1}|K_{2}}_{\four\four}\left(x_{1}/x_{2}\right)\mathsf{T}_{\four}^{K_{2}}(x_{2})\mathsf{T}^{K_{1}}_{\four}(x_{1})=0,
    \eea
    where $K_{1},K_{2}$ are arbitrary. Moreover, this commutativity uniquely determines the sign factor $(-1)^{\sigma_{4}(\rho)}$ up to a global $\mathbb{Z}_{2}$ symmetry. 
\end{conjecture}

\begin{corollary}[Cor.~\ref{cor:D6commutativity}]
    The D6 $qq$-characters all commute with each other:
    \bea
\mathsf{f}^{q_{b}|q_{a}}_{\four\four}\left(x_{2}/x_{1}\right)\mathsf{T}_{\bar{a}}(x_{1})\mathsf{T}_{\bar{b}}(x_{2})-\mathsf{f}^{q_{a}|q_{b}}_{\four\four}\left(x_{1}/x_{2}\right)\mathsf{T}_{\bar{b}}(x_{2})\mathsf{T}_{\bar{a}}(x_{1})=0,\quad a,b\in\four.
    \eea
    Moreover, they also commute with the D8 $qq$-characters:
    \bea
    \mathsf{f}^{q_{a}|K}_{\four\four}\left(x_{2}/x_{1}\right)\mathsf{T}_{\four}^{K}(x_{1})\mathsf{T}_{\bar{a}}(x_{2})-\mathsf{f}^{K|q_{a}}_{\four\four}\left(x_{1}/x_{2}\right)\mathsf{T}_{\bar{a}}(x_{2})\mathsf{T}^{K}_{\four}(x_{1})=0,\quad a\in\four.
    \eea
\end{corollary}

\chapter{ADHM theory and supersymmetric localization}\label{chap:ADHM-localization}

Yang--Mills instantons are topologically nontrivial field configurations minimizing the Yang--Mills action. ADHM constructions give a systematical way to construct instantons and the instanton moduli space is encoded in the ADHM data. With suitable number of supersymmetries, we can explicitly compute the instanton partition functions. In this chapter, we review basic aspects of instanton counting and supersymmetric localization using the pure super Yang--Mills (SYM) as a toy model. For a reviews, see for example \cite{Pestun:2016zxk} or \cite{Shadchin:2005mx} (see also \cite{Nakajima:2003uh}). 

We first review basic facts for instantons and supersymmetric localization in section~\ref{sec:YMtheory-instanton}. We then give the ADHM data and discuss that the moduli space is constructed from the ADHM equations in section~\ref{sec:ADHMdata-LMNSformula}. Studying the symmetries of the moduli space, we provide the so-called LMNS formula which is a contour integral formula giving the $k$-instanton contribution. In section~\ref{sec:pureSYM_JK}, we discuss on how to evaluate the contour integral using the JK-residue formalism. We will see that the poles are classified by Young diagrams. Finally, in section~\ref{sec:character-index}, we provide residue formulas and explicitly evaluate the residues.  

\section{Yang--Mills theory and instantons}\label{sec:YMtheory-instanton}
Let us briefly review what a gauge theory is. Let $G$ be a Lie group, where in this thesis it is taken to be a unitary group $G=\U(n)$. We denote $\mathfrak{g}=\operatorname{Lie} (G)$ as the Lie algebra. We take the space-time to be Euclidean. The gauge field $A_{\mu}$ is a $\mathfrak{g}$-valued one-form and the field strength is defined as $F=dA+A\wedge A$. The gauge transformation of them are given as $A\mapsto gAg^{-1}+gdg^{-1}$ and $F\mapsto gFg^{-1}$ for $g\in G$. The $G$-invariant action of the Yang--Mills (YM) theory is defined as
\bea\label{eq:YMlagrangian}
S_{\text{YM}}=-\frac{1}{g^{2}_{\text{YM}}}\int d^{4}x \,\Tr\left(F\wedge \star F\right),
\eea
where $\star$ is the Hodge dual. Given this Lagrangian, we are interested in field configurations minimizing the YM action as
\bea
\frac{\delta S_{\text{YM}}}{\delta A}=0.
\eea
Since the Hodge star acts as $\star^{2}=1$ on two-forms in four dimensions, we can decompose the field strength into self-dual (SD) and anti-self-dual (ASD) parts, with respect to the eigenvalue of the Hodge star:
\bea
F_{\pm}=\frac{1}{2}(F\pm \star F).
\eea
The vanishing conditions of the SD (ASD) part is the ASD (SD) YM equation: $\star F=\mp F$. Fields obeying the ASD (SD) equations are called the instantons (anti-instantons).

The action \eqref{eq:YMlagrangian} is bounded by the SD, ASD YM equations:
\bea
S_{\text{YM}}=-\frac{1}{2g^{2}_{\text{YM}}}\int d^{4}x\, \Tr\left(F\pm \star F\right)\wedge \star(F\pm \star F)\pm \frac{1}{g^{2}_{\text{YM}}}\int \Tr F\wedge F \geq \frac{8\pi^{2}|k|}{g^{2}_{\text{YM}}},
\eea
where the instanton number $k\in\mathbb{Z}$ is defined as
\bea
k=\frac{1}{8\pi^{2}}\int \Tr F\wedge F
\eea
and the equality is at $F=\mp \star F$. For the ASD (SD) YM equations, we have $k\geq 0\,(k\leq 0)$. We may also include the $\theta$-term as
\bea
S_{\theta}=-\frac{i\theta}{2\pi^{2}}\int d^{4}x\, \Tr F\wedge F
\eea
and the total action is given as $S_{\text{tot.}}=S_{\text{YM}}+S_{\theta}$. The total action around the $k$-instanton contribution is schematically expanded as
\bea
S_{\text{tot.}}[A]=\frac{8\pi^{2}}{g_{\text{YM}}^{2}}|k|-i\theta k+S_{\text{fluc.}}[\delta A]=S_{\text{top.}}[A]+S_{\text{fluc.}}[\delta A]
\eea
with $A=A^{(k)}_{\text{inst.}}+\delta A$. In this thesis, we are interested in the partition function defined as
\bea
\mathcal{Z}=\int [\mathcal{D}A]e^{-S_{\text{tot.}}}.
\eea

Explicitly computing this quantity is difficult but with the power of supersymmetry, the fluctuation part will be suppressed and we can compute it explicitly using supersymmetric localization. Suppose that $S$ is an action of a given theory and that we have a non-anomalous fermionic symmetry charge $\mathcal{Q}$. We collectively denote $\Phi$ as the field variables. The charge $\mathcal{Q}$ obeys $\mathcal{Q}S=0$ and $\mathcal{Q}^{2}=B$ where $B$ is some bosonic symmetry. Let $\mathcal{O}$ be a $\mathcal{Q}$-closed operator and the vacuum expectation value of it is defined as
\bea
\langle \mathcal{O}\rangle =\int [\mathcal{D}\Phi]e^{-S[\Phi]}\mathcal{O}[\Phi].
\eea
Choosing a functional $V[\Phi]$ such that $BV=0$, we can add the localizing action $S_{\text{loc.}}=\mathcal{Q}V[\Phi]$ and deform the VEV as
\bea
\langle \mathcal{O}\rangle(t)=\int [\mathcal{D}\Phi]e^{-S[\Phi]-t S_{\text{loc.}}[\Phi]}\mathcal{O}[\Phi].
\eea
This VEV actually does not depend on $t$:
\bea
\frac{\partial \langle \mathcal{O}\rangle (t)}{\partial t}=-\int [\mathcal{D}X]\mathcal{Q}V[\Phi]e^{-S[\Phi]-tS_{\text{loc.}}[\Phi]}\mathcal{O}[\Phi]=-\int [\mathcal{D}X]\mathcal{Q}\left(V[\Phi]e^{-S[\Phi]-tS_{\text{loc.}}[\Phi]}\mathcal{O}[\Phi]\right)=0.
\eea
We then have
\bea
\langle \mathcal{O}\rangle=\lim_{t\rightarrow \infty}\langle\mathcal{O}\rangle(t).
\eea
Note that the VEV of the operator does not change even if we add a $\mathcal{Q}$-exact term. Moreover, the coupling constants of such $\mathcal{Q}$-exact terms do not depend. Using this property, taking the limit $t\rightarrow \infty$ fields obeying $S_{\text{loc.}}[\Phi_{0}]=0$ only remain and we can expand the fields around $\Phi_{0}$ as
\bea
\Phi=\Phi_{0}+t^{-1/2}\tilde{\Phi}
\eea
and we have 
\bea
S[\Phi]+t\mathcal{Q}V[\Phi]=S[\Phi_{0}]+\int \left(\tilde{\Phi}\left.\frac{\delta^{2}S_{\text{loc.}}[\Phi]}{\delta \Phi^{2}}\right|_{\Phi=\Phi_{0}}\tilde{\Phi}\right)+\mathcal{O}(t^{-1/2}).
\eea
Thanks to this property, only quadratic terms in $\tilde{\Phi}$ will appear in the VEV. Denoting the space of the fields obeying the condition $S_{\text{loc.}}[\Phi_{0}]=0$ as $\mathfrak{M}$, the VEV is then given as
\bea\label{eq:susylocal-opVEV}
\langle \mathcal{O} \rangle= \lim_{t\rightarrow \infty} \langle\mathcal{O}\rangle(t)=\int_{\mathfrak{M}}[\mathcal{D}\Phi_{0}]e^{-S[\Phi_{0}]}\left[\operatorname{sdet}\left(\left.\frac{\delta^{2}S_{\text{loc.}}[\Phi]}{\delta\Phi^{2}}\right|_{\Phi=\Phi_{0}}\right)\right]^{-1}\mathcal{O}[\Phi_{0}].
\eea
The $\operatorname{sdet}$ is the one-loop functional determinants over the fermions and bosons.

In particular, considering the 4d $\mathcal{N}=2$ pure super Yang--Mills theory with gauge group $G=\U(n)$, the action schematically takes the form as
\bea
S=S_{\text{top.}}+\mathcal{Q} V.
\eea
For the detailed expression, see for example \cite{Shadchin:2005mx}. Expanding around the $k$-instanton configuration, the partition function is then given by 
\bea
\mathcal{Z}=\langle 1\rangle=\int [\mathcal{D}\Phi]e^{-S_{\text{top.}}-\mathcal{Q}V
}=\mathcal{Z}_{\text{pert.}}\sum_{k=0}^{\infty}\mathfrak{q}^{k}\int_{\mathfrak{M}_{n,k}} [\mathcal{D}A^{(k)}_{\text{inst.}}] \, 1
\eea
where $\mathfrak{q}=e^{2\pi i \tau}$ and 
\bea
\tau=\frac{4\pi i }{g_{\text{YM}}^{2}}+\frac{\theta}{2\pi}.
\eea
The space $\mathfrak{M}_{n,k}$ is the space of fields obeying the ASD YM equation with the topological number $k$ and it is called the \textit{instanton moduli space}. The perturbative part $\mathcal{Z}_{\text{pert.}}$ comes from the path integral of the quadratic terms coming from the fluctuations around the saddle point. The instanton partition function is then defined by omitting the perturbative part
\bea
\mathcal{Z}_{\text{inst.}}=\sum_{k=0}^{\infty}\mathfrak{q}^{k}\mathcal{Z}_{k},\quad \mathcal{Z}_{k}=\int_{\mathfrak{M}_{n,k}}1.
\eea

\section{ADHM data and LMNS formula}\label{sec:ADHMdata-LMNSformula}
\paragraph{ADHM data}
To evaluate the instanton partition function, we need the information of the instanton moduli space $\mathfrak{M}_{n,k}$. A systematic construction for instantons is known as the ADHM construction \cite{Atiyah:1978ri}. To construct the $k$-instanton solution in $G=\U(n)$ Yang--Mills theory on $\mathbb{R}^{4}=\mathbb{C}^{2}$, we introduce two complex vector spaces
\bea
\bfN=\mathbb{C}^{n},\quad \bfK=\mathbb{C}^{k}
\eea
and define
\bea
X=\Hom(\bfK,\bfK)\oplus \Hom(\bfK,\bfK)\oplus \Hom(\bfN,\bfK)\oplus \Hom(\bfK,\bfN).
\eea
The ADHM variables $\mathsf{B}_{1,2},\mathsf{I,J}$ are linear maps relating the vector spaces $\bfN,\bfK$ and are elements of $X$:
\bea
\mathsf{B}_{1,2}:\bfK\rightarrow \bfK,\quad \mathsf{I}:\bfN\rightarrow \bfK,\quad \mathsf{J}:\bfK\rightarrow \bfN.
\eea
The space $X$ has a natural action coming from $G=\U(n)$ and $G^{\vee}=\U(k)$. The actions on the ADHM variables are given as
\bea\label{eq:ADHMflavorsymmetry}
G^{\vee}:&\quad (g)\cdot (\mathsf{B}_{1,2},\mathsf{I},\mathsf{J})=(g\,\mathsf{B}_{1,2}g^{-1},g\,\mathsf{I},\mathsf{J}g^{-1}),\quad g\in\U(k),\\
G:&\quad (\nu)\cdot (\mathsf{B}_{1,2},\mathsf{I},\mathsf{J})=(\mathsf{B}_{1,2},\mathsf{I}\,\nu^{-1},\nu \mathsf{J}),\quad \nu\in\U(n).
\eea
They play the role of the flavor symmetry of the ADHM variables. The latter symmetry corresponds to the gauge transformation of the framing of the $\U(n)$ theory.\footnote{Since we are interested in instantons with finite action, the curvature needs to vanish at the infinity $F\rightarrow 0$ implying a pure gauge $A\rightarrow gdg^{-1}$. We still have a residual global symmetry here and it is called the framing of the gauge connection.
}

Atiyah--Drinfeld--Hitchin--Manin identified that the $k$-instanton moduli space of the $\U(n)$ connection is identified as the hyper-K\"{a}hler quotient 
\bea\label{eq:ADHMmodulispace}
\mathfrak{M}_{n,k}=\{(\mathsf{B}_{1},\mathsf{B}_{2},\mathsf{I},\mathsf{J})\mid (\mu_{\mathbb{R}},\mu_{\mathbb{C}})=(0,0)\}/\U(k)
\eea
where $\mu_{\mathbb{R}},\mu_{\mathbb{C}}$ are the moment maps defined as
\bea\label{eq:ADHMeq}
\mu_{\mathbb{R}}&=[\mathsf{B}_{1},\mathsf{B}_{1}^{\dagger}]+[\mathsf{B}_{2},\mathsf{B}_{2}^{\dagger}]+\mathsf{I}\mathsf{I}^{\dagger}-\mathsf{J}^{\dagger}\mathsf{J},\\
\mu_{\mathbb{C}}&=[\mathsf{B}_{1},\mathsf{B}_{2}]+\mathsf{I}\mathsf{J}.
\eea
These equations $(\mu_{\mathbb{R}},\mu_{\mathbb{C}})=(0,0)$ are called the ADHM equations.

In addition to \eqref{eq:ADHMflavorsymmetry}, we have another global symmetry preserving the ADHM equations \eqref{eq:ADHMeq} corresponding to the spacetime rotation of the gauge theory on $\mathbb{C}^{2}$:
\bea\label{eq:ADHMomegasymmetry}
(q_{1},q_{2})\cdot (\mathsf{B}_{1},\mathsf{B}_{2},\mathsf{I},\mathsf{J})=(q_{1}\mathsf{B}_{1},q_{2}\mathsf{B}_{2},\mathsf{I},q_{12}\mathsf{J})
\eea
where $q_{1,2}=e^{\epsilon_{1,2}}\in\U(1)^{2}$, $q_{12}=q_{1}q_{2}=e^{\epsilon_{12}}$, and $\epsilon_{12}=\epsilon_{1}+\epsilon_{2}$. Under this symmetry, the ADHM equations transform as
\bea
\mu_{\mathbb{R}}\rightarrow \mu_{\mathbb{R}},\quad \mu_{\mathbb{C}}\rightarrow q_{12}\mu_{\mathbb{C}}.
\eea

The moduli space \eqref{eq:ADHMmodulispace} itself is actually non-compact and singular. Naively, the volume of it will diverge and to resolve this we can modify the ADHM equations as
\bea\label{eq:ADHMresolvedmodulispace}
\mathfrak{M}_{n,k}^{\zeta}=\mu^{-1}(\zeta)/\U(k)=s^{-1}(0)/\U(k)
\eea
where $\mu=(\mu_{\mathbb{R}},\Re \mu_{\mathbb{C}},\Im \mu_{\mathbb{C}})$, $\zeta=(\zeta 1_{k},0,0)$, and $\vec{s}=\vec{\mu}-\vec{\zeta}$ $(=s)$. 


\paragraph{LMNS formula}
The integration over $\mathfrak{M}_{n,k}^{\zeta}$ can be written in a path integral over the variables which eventually leads to a finite-dimensional integral. Let $(\psi_{\mathsf{B}_{1,2}},\psi_{\mathsf{I}},\psi_{\mathsf{J}})$ be fermionic variables and $(\vec{\chi},\vec{H}),(\lambda,\eta)$ be the anti-ghost multiplets. We define the BRST transformation as
\bea
\mathcal{Q}\mathsf{B}_{1,2}=\psi_{1,2},&\qquad \mathcal{Q}\psi_{1,2}=[\phi,\mathsf{B}_{1,2}]-\epsilon_{1,2}\mathsf{B}_{1,2},\\
\mathcal{Q}\mathsf{I}=\psi_{\mathsf{I}},&\qquad \mathcal{Q}\psi_{\mathsf{I}}=\phi\,\mathsf{I}-\mathsf{I}\mathfrak{a},\\
\mathcal{Q}_{\mathsf{J}}=\psi_{\mathsf{J}},&\qquad \mathcal{Q}\psi_{\mathsf{J}}=-\mathsf{J}\phi+\mathfrak{a}\mathsf{J}-\epsilon_{12}\mathsf{J},\\
\mathcal{Q}\chi_{\mathbb{R}}=H_{\mathbb{R}},&\qquad  \mathcal{Q}H_{\mathbb{R}}=[\phi,\chi_{\mathbb{R}}],\\
\mathcal{Q}\chi_{\mathbb{C}}=H_{\mathbb{C}},&\qquad \mathcal{Q}H_{\mathbb{C}}=[\phi,\chi_{\mathbb{C}}]-\epsilon_{12}\chi_{\mathbb{C}},\\
\mathcal{Q}\lambda=\eta,&\qquad \mathcal{Q}\eta=[\phi,\lambda] 
\eea
where $(\phi,\mathfrak{a},\epsilon_{1,2})\in \mathfrak{u}(k),\mathfrak{u}(n),\mathfrak{u}(1)^{2}$. The equivariant integral over the ADHM variables is given as
\bea
\mathcal{Z}_{k}(\mathfrak{a},\epsilon_{1,2})\coloneqq\int_{\mathfrak{M}_{n,k}}1=\int_{\mathfrak{g}^{\vee}}\frac{d\phi}{\operatorname{vol} G^{\vee}}\int_{X}e^{-S}
\eea
with the action defined as
\bea
S=i\mathcal{Q}\Tr_{K}\left[\vec{\chi}\cdot \vec{s}+t\vec{\chi}\cdot \vec{H}+\psi\cdot V(\lambda)+\frac{1}{g_{\eta}}\eta[\phi,\lambda]\right]
\eea
where
\bea
\vec{\chi}\cdot\vec{s}&=\chi_{\mathbb{R}}s_{\mathbb{R}}+\frac{1}{2}\left(\chi_{\mathbb{C}}^{\dagger}s_{\mathbb{C}}+\chi_{\mathbb{C}}s_{\mathbb{C}}^{\dagger}\right),\\
\vec{\chi}\cdot \vec{H}&=\chi_{\mathbb{R}}H_{\mathbb{R}}+\frac{1}{2}\left(\chi_{\mathbb{C}}^{\dagger}H_{\mathbb{C}}+\chi_{\mathbb{C}}H_{\mathbb{C}}^{\dagger}\right),\\
\psi\cdot V(\lambda)&=\psi_{1}[\lambda,\mathsf{B}_{1}]+\psi_{2}[\lambda,\mathsf{B}_{2}]+\psi_{\mathsf{I}}\lambda\mathsf{I}-\mathsf{J}\lambda\psi_{\mathsf{J}}\\
&-\bar{\psi}_{1}[\lambda,\mathsf{B}_{1}^{\dagger}]-\bar{\psi}_{2}[\lambda,\mathsf{B}_{2}^{\dagger}]-\mathsf{I}^{\dagger}\lambda\bar{\psi}_{\mathsf{I}}+\bar{\psi}_{\mathsf{J}}\lambda\mathsf{J}^{\dagger}.
\eea
Note that since this action is $\mathcal{Q}$-exact, the path integral is independent of the coupling constants. To cure zero-modes coming from $\chi_{\mathbb{R}}$, we may add another $\mathcal{Q}$-exact term $\mathcal{Q}\Tr_{K}\chi_{\mathbb{R}}\lambda$. We also can further introduce kinetic terms of $(\mathsf{B}_{1,2},\mathsf{I},\mathsf{J})$ as
\bea
\frac{g_{\text{kin}}}{2}\mathcal{Q}\Tr_{K}\left[\mathsf{B}_{1,2}^{\dagger}\psi_{1,2}-\bar{\psi}_{1,2}\mathsf{B}_{1,2}+\mathsf{I}^{\dagger}\psi_{\mathsf{I}}-\bar{\psi}_{\mathsf{I}}\mathsf{I}+\mathsf{J}^{\dagger}\psi_{\mathsf{J}}-\bar{\psi}_{\mathsf{J}}\mathsf{J}\right].
\eea

Using the property that the path integral does not depend on the coupling constants, at the end the resulting terms only come from $\mathsf{B}_{1,2},\mathsf{I},\mathsf{J},\chi_{\mathbb{C}}$ (see for example \cite{Shadchin:2005mx} for details of this computation). To see this, we first diagonalize $\phi$ as $\phi=\operatorname{diag}(\phi_{1},\ldots, \phi_{k})$ and then the measure is given as
\bea
\frac{d\phi}{\operatorname{vol}G^{\vee}}=\frac{1}{k!}\prod_{I=1}^{k}\frac{d\phi_{I}}{2\pi i } \prod_{I\neq J}^{k}(\phi_{I}-\phi_{J}).
\eea
Under this basis, the bilinear terms of $\chi_{\mathbb{C}},\mathsf{B}_{1,2},\mathsf{I},\mathsf{J}$ are given as
\bea
\sum_{I,J=1}^{k}(\phi_{IJ}-\epsilon_{12})|\chi_{\mathbb{C},IJ}|^{2},&\quad \sum_{I,J=1}^{k}(\phi_{IJ}-\epsilon_{a})|\mathsf{B}_{a,IJ}|^{2},\\
\sum_{\alpha=1}^{n}\sum_{I=1}^{k}(\phi_{I}-\mathfrak{a}_{\alpha})|\mathsf{I}_{I\alpha}|^{2},&\quad \sum_{I=1}^{k}\sum_{\alpha=1}^{n}(-\phi_{I}+\mathfrak{a}_{\alpha}-\epsilon_{12})|\mathsf{J}_{\alpha I}|^{2}.
\eea
Performing the Gaussian integral gives
\bea
\mathcal{Z}_{k}=\frac{1}{k!}\frac{(-\epsilon_{12})^{k}}{\epsilon_{1}^{k}\epsilon_{2}^{k}}\oint \prod_{I=1}^{k}\frac{d\phi_{I}}{2\pi i}\frac{1}{P_{\text{4d}}(\phi_{I})\widetilde{P}_{\text{4d}}(\phi_{I}+\epsilon_{12})}\prod_{I\neq J}^{k}\mathscr{S}(\phi_{IJ})^{-1}
\eea
where
\bea
P_{\text{4d}}(\phi)=\prod_{\alpha=1}^{n}(\phi-\mathfrak{a}_{\alpha}),\quad \widetilde{P}_{\text{4d}}(\phi)=\prod_{\alpha=1}^{n}(-\phi+\mathfrak{a}_{\alpha}),\quad \mathscr{S}(\phi)=\frac{(\phi-\epsilon)(\phi-\epsilon)}{\phi(\phi-\epsilon_{12})}.
\eea
This is the famous Losev--Moore--Nekrasov--Shatashvili (LMNS) formula \cite{Losev:1997tp,Moore:1997dj,Lossev:1997bz}. 


\paragraph{Five-dimensional lift up}
Although we considered the 4d pure SYM defined on $\mathbb{C}^{2}$, we can also study the 5d lift up version of it defined on $\mathbb{C}^{2}\times \mathbb{S}^{1}$. There is a simple prescription to do this and it is to replace variables such as $(\phi_{I}-\mathfrak{a})$ to $\sh(\phi_{I}-\mathfrak{a})$ (see section~\ref{sec:susyqm-index} for motivations), where
\bea
\sh(x)\coloneqq e^{x/2}-e^{-x/2}.
\eea
Physically, we are considering the supersymmetric quantum mechanics instead of the matrix model. The instanton partition function for the 5d $\mathcal{N}=1$ $\U(n)$ pure super Yang--Mills theory is then given by 
\bea
\mathcal{Z}&=\sum_{k=0}^{\infty}\mathfrak{q}^{k}\mathcal{Z}_{k},\quad \mathcal{Z}_{k}=\frac{1}{k!}\oint_{\text{JK}} \prod_{I=1}^{k}\frac{d\phi_{I}}{2\pi i }\,\,\mu_{k}(\mathfrak{a}_{\alpha},\phi_{I},\epsilon_{1,2})
\eea
where
\bea\label{eq:pureSYMcontour}
\mu_{k}&=\left(\frac{\sh(\epsilon_{12})}{\sh(-\epsilon_{1,2})}\right)^{k}\prod_{I=1}^{k}\prod_{\alpha=1}^{n}\frac{1}{\sh(\phi_{I}-\mathfrak{a}_{\alpha})\sh(\mathfrak{a}_{\alpha}-\phi_{I}-\epsilon_{12})}\prod_{I\neq J}^{k}\frac{\sh(\phi_{I}-\phi_{J})\sh(\phi_{I}-\phi_{J}+\epsilon_{12})}{\sh(\phi_{I}-\phi_{J}-\epsilon_{1})\sh(\phi_{I}-\phi_{J}-\epsilon_{2})}.
\eea
In this thesis, we will be interested in these type of partition functions where trigonometric functions appear.


\section{Jefferey--Kirwan residue and Young diagrams}\label{sec:pureSYM_JK}

The formula in \eqref{eq:pureSYMcontour} is a contour integral in multiple variables and we need to specify how to take the contour integral to evaluate it explicitly \cite{Nekrasov:2002qd,Nakajima:1999,Nakajima:2003uh}. In this thesis, we use the Jefferey--Kirwan (JK) residue prescription \cite{Jeffrey1993LocalizationFN} (see \cite{Benini:2013nda,Benini:2013xpa} for a physical derivation of it). Since it is out of the scope of this thesis, we will not give any reasoning why we should use this prescription, but we will simply use the formalism to evaluate all the contour integral formulas appearing in this thesis. 

Let $\omega$ be a meromorphic $(k,0)$-form in variables $\phi=(\phi_{1},\ldots ,\phi_{k})$
\bea
\omega=\frac{f(\phi)}{\prod_{i}(\bfQ_{i}\cdot \phi+b_{i})}d\phi_{1}\wedge d\phi_{2}\wedge\cdots \wedge d\phi_{k},
\eea
where $\bfQ_{i}$ is written in the standard basis $\{\bfe_{i}\}_{i=1,\ldots,k}$ as
\bea
\bfQ_{i}=\sum_{j=1}^{k}\bfQ_{i,j}\bfe_{j},\quad \bfe_{j}=(0,\ldots, 0,\overset{j}{1},0,\ldots,0).
\eea
For simplicity, we omit the $2\pi i$ factor in the denominator of $d\phi$. Denoting the singular hyperplane $H_{i}$ as
\bea
H_{i}=\{\phi\in\mathbb{C}^{k}\mid \bfQ_{i}\cdot\phi+b_{i}=0\},
\eea
then $\omega$ is singular at $\mathcal{M}_{\text{sing}}=\bigcup_{i} H_{i}$ and holomorphic at the complement of it. We define $\mathcal{M}^{\ast}_{\text{sing}}$ to be the set of isolated points where $n\geq k$ linearly independent singular hyperplanes meet. When $n=r$ hyperplanes meet the point $\phi_{\ast}\in\mathcal{M}^{\ast}_{\text{sing}}$, the intersection is called \textit{non-degenerate}, while when $n>r$ hyperplanes meet it, the intersection is called \textit{degenerate}. For $\phi_{\ast}\in\mathcal{M}^{\ast}_{\text{sing}}$, we denote $\bfQ(\phi_{\ast})$ as the set of charge vectors of singular hyperplanes meeting at $\phi_{\ast}$:
\bea
\bfQ(\phi_{\ast})=\{\bfQ_{i}\mid \phi_{\ast}\in H_{i},\,i=1,\ldots, n\}.
\eea
The hyperplane arrangements are called projective when the $n$ charge vectors $\bfQ(\phi_{\ast})$ are contained in the half space of $\mathbb{R}^{k}$. We assume that the hyperplane arrangements are projective, then the residue of $\omega$ at $\phi_{\ast}$ is given as follows. Near the pole $\phi_{\ast}$, we can expand $\omega$ in negative powers of $\bfQ_{i}\cdot (\phi-\phi_{\ast})$ giving linear combinations of 
\bea
\frac{1}{\bfQ_{1}\cdot (\phi-\phi_{\ast})\cdots \bfQ_{k}\cdot (\phi-\phi_{\ast})}
\eea
for some $k$ charge vectors $\bfQ_{1},\ldots,\bfQ_{k}$ in $\bfQ(\phi_{\ast})$. We denote the cone spanned by $\bfQ_{1},\ldots,\bfQ_{k}$ as
\bea\label{eq:JKconedef}
\operatorname{Cone}(\bfQ_{1},\ldots,\bfQ_{k})=\left\{\sum_{i=1}^{k}\lambda_{i}\bfQ_{i}=\bm{\eta}\mid \lambda_{i}>0\right\},
\eea
where $\bm{\eta}$ is some generic vector which we call the \textit{reference vector}. In this thesis, we always keep it to be $\bm{\eta}=(1,\ldots,1)=\sum_{i=1}^{k}\bfe_{i}$. The JK-residue formula claims that the contour integral is given as
\bea
\int \omega =\sum_{\phi_{\ast}\in\mathcal{M}_{\text{sing}}^{\ast}}\underset{\phi=\phi_{\ast}}{\operatorname{JK-Res}}(\bfQ(\phi_{\ast}),\bm{\eta})\omega
\eea
where the JK-residue operator is defined as
\bea
\underset{\phi=\phi_{\ast}}{\operatorname{JK-Res}}(\bfQ(\phi_{\ast}),\bm{\eta})\frac{d\phi_{1}\wedge \cdots\wedge d\phi_{k}}{\prod_{i=1}^{k}(\bfQ_{i}\cdot (\phi-\phi_{\ast}))}=\begin{dcases}
    |\det (\bfQ_{1},\ldots,\bfQ_{k})|^{-1},\quad \bm{\eta}\in \operatorname{Cone}(\bfQ_{1},\ldots,\bfQ_{k})\\
    0,\quad\quad \quad \quad  \text{otherwise}.
\end{dcases}
\eea

Although the above description gives an explicit way to evaluate the residues of the contour integral, finding a way to evaluate the integral recursively is useful. Actually, the JK residue can be equivalently given by a sum of iterated residues \cite{Szenes2003ToricRA}. For each $\phi_{\ast}$ with $(\bfQ_{1},\ldots,\bfQ_{n})$, we consider a flag
\bea
\mathcal{F}=\left[\mathcal{F}_{0}=\{0\}\subset \mathcal{F}_{1}\subset \mathcal{F}_{2}\cdots \subset \mathcal{F}_{k}=\mathbb{R}^{k}\right],\quad \operatorname{dim}\mathcal{F}_{j}=j
\eea
where $\mathcal{F}_{j}$ is a vector space spanned by $\{\bfQ_{1},\ldots,\bfQ_{j}\}$. Denote the set of flags as $\mathcal{FL}(\bfQ(\phi_{\ast}))$ and define the subset 
\bea
\mathcal{FL}^{+}(\bfQ(\phi_{\ast}))=\left\{\mathcal{F}\in \mathcal{FL}(\bfQ(\phi_{\ast}))\mid \eta\in \operatorname{Cone}(\kappa_{1}^{\mathcal{F}},\ldots,\kappa_{r}^{\mathcal{F}})\right\},
\eea
where 
\bea
\kappa_{j}^{\mathcal{F}}=\sum_{\bfQ_{i}\in\bfQ(\phi_{\ast})\cap \mathcal{F}_{j}}\bfQ_{i}.
\eea
We also define the sign of $\mathcal{F}$, denoted as $\operatorname{sgn}\mathcal{F}$, as the sign of the determinant of the matrix 
\bea
\kappa(\mathcal{F},\bfQ(\phi_{\ast}))=(\kappa_{1}^{\mathcal{F}},\ldots,\kappa_{r}^{\mathcal{F}}).
\eea
The JK residue is then obtain by an iterated residue :
\bea\label{eq:JK-iteratedmethod}
\underset{\phi=\phi_{\ast}}{\operatorname{JK-Res}}(\bfQ(\phi_{\ast}),\bm{\eta})=\sum_{\mathcal{F}}\operatorname{sgn}(\det \kappa(\mathcal{F},\bfQ(\phi_{\ast}))) \underset{\mathcal{F}}{\operatorname{Res}},
\eea
where the iterated residue is simply the residue computed in the basis given by $\{\bfQ_{i}\}_{i=1}^{k}$. Namely, given a $k$-form $\omega=\omega_{1,\ldots,k}d\phi_{1}\wedge d\phi_{2}\wedge \cdots \wedge d\phi_{k}$, we define new coordinates $\tilde{\phi}_{j}=\bfQ_{j}\cdot \phi$ such that $\omega=\tilde{\omega}_{1,\ldots,k}d\tilde{\phi}_{1}\wedge \cdots d\tilde{\phi}_{k}$, then the residue is given by 
\bea\label{eq:JK-iteratedmethod2}
\underset{\mathcal{F}}{\operatorname{Res}}=\underset{\tilde{\phi}_{r}=\tilde{\phi}_{r}^{\ast}}{\Res}\cdots \underset{\tilde{\phi}_{1}=\tilde{\phi}_{1}^{\ast}}{\Res}=J\left(\frac{\partial \widetilde{\phi}_{i}}{\partial \phi_{j}}\right)\underset{\phi_{r}=\phi_{r}^{\ast}}{\Res}\cdots \underset{\phi_{1}=\phi_{1}^{\ast}}{\Res}
\eea
where $J$ is the Jacobian. We note that when taking the iterative residue, at each step the other variables are kept constant and generic. Namely, when evaluating the pole at $\phi_{i}=\phi_{i}^{\ast}$, residues for $\phi_{j}=\phi_{j}^{\ast}$ ($j<i$) are already taken while variables $\phi_{j}\,(j>i)$ are kept generic.


The JK-residue formalism discussed above is quite abstract so for practical use let us study the application to the pure super Yang--Mills case \eqref{eq:pureSYMcontour}. We follow the discussion in \cite{Hwang:2014uwa}. Looking at the denominator of \eqref{eq:pureSYMcontour}, the poles come from the hyperplanes
\bea
H_{I,\alpha}=\{\phi_{I}-\mathfrak{a}_{\alpha}=0\},\quad \widetilde{H}_{I,\alpha}=\{\mathfrak{a}_{\alpha}-\phi_{I}-\epsilon_{12}=0\},\quad H_{IJ,a=1,2}=\{\phi_{I}-\phi_{J}-\epsilon_{a}=0\}.
\eea
The corresponding charge vectors are $\{\pm\mathbf{e}_{I}\}$ and $\{\mathbf{e}_{I}-\mathbf{e}_{J}\}$, respectively and thus $\{\bfQ_{I}\}$ are chosen from them. Not all of the poles are chosen but the JK residue procedure restricts which poles to evaluate. Again, we first choose $\bm{\eta}=(1,1,\ldots,1)=\sum_{i=1}^{k}\mathbf{e}_{i}$. The poles picked up are determined by the possible $k$ charges $\{\mathbf{Q}_{1},\ldots,\mathbf{Q}_{k}\}$ such that the reference vector obeys the condition \eqref{eq:JKconedef}, which is 
\bea\label{eq:JKconecond}
\bm{\eta}=\sum_{I=1}^{k}\bfe_{I}=\sum_{I=1}^{k}\lambda_{I}\bfQ_{I},\quad \lambda_{I}>0.
\eea

\paragraph{Two-instanton case}
Before moving on to higher instantons cases, let us see what will happen for the $k=2$ case. The poles come from $\{\pm \bfe_{1,2}\}$ and $\{\bfe_{1}-\bfe_{2},\bfe_{2}-\bfe_{1}\}$. The possible choices obeying $\eta\in\operatorname{Cone}(\bfQ_{1},\bfQ_{2})$ are illustrated as
\bea
\begin{tikzpicture}[scale=1.5]
\fill[lightgray!30!white] (0,0)--(1,0)--(1,1)--(0,1)--(0,0);
\draw[gray!80!black,thick] (-1,0)--(1,0);
\draw[gray!80!black,thick] (0,1)--(0,-1);
\draw[red,thick,-Triangle] (0,0)--(0.7,0.7);
\node[above right,red] at (0.6,0.6){$\eta$};
\draw[thick,-Triangle,blue](0,0)--(0.7,0);
\draw[thick,-Triangle](0,0)--(-0.7,0);
\draw[thick,-Triangle,blue](0,0)--(0,0.7);
\draw[thick,-Triangle](0,0)--(0,-0.7);
\draw[thick, -Triangle] (0,0)--(0.65,-0.65);
\draw[thick, -Triangle] (0,0)--(-0.65,0.65);
\node[below right,blue] at (0.6,0){$\mathbf{e}_{1}$};
\node[above left,blue] at (0,0.6){$\mathbf{e}_{2}$};
\end{tikzpicture}\qquad
\begin{tikzpicture}[scale=1.5]
\fill[lightgray!30!white] (0,0)--(1,0)--(1,1)--(-1,1)--(0,0);
\draw[gray!80!black,thick] (-1,0)--(1,0);
\draw[gray!80!black,thick] (0,1)--(0,-1);
\draw[red,thick,-Triangle] (0,0)--(0.7,0.7);
\node[above right,red] at (0.6,0.6){$\eta$};
\draw[thick,-Triangle,blue](0,0)--(0.7,0);
\draw[thick,-Triangle](0,0)--(-0.7,0);
\draw[thick,-Triangle](0,0)--(0,0.7);
\draw[thick,-Triangle](0,0)--(0,-0.7);
\draw[thick, -Triangle] (0,0)--(0.65,-0.65);
\draw[thick, -Triangle,blue] (0,0)--(-0.65,0.65);
\node[below right,blue] at (0.6,0){$\mathbf{e}_{1}$};
\node[above,blue] at (-0.65,0.65){$-\mathbf{e}_{1}+\mathbf{e}_{2}$};
\end{tikzpicture}
\qquad 
\begin{tikzpicture}[scale=1.5]
\fill[lightgray!30!white] (0,0)--(1,-1)--(1,1)--(0,1)--(0,0);
\draw[gray!80!black,thick] (-1,0)--(1,0);
\draw[gray!80!black,thick] (0,1)--(0,-1);
\draw[red,thick,-Triangle] (0,0)--(0.7,0.7);
\node[above right,red] at (0.6,0.6){$\eta$};
\draw[thick,-Triangle](0,0)--(0.7,0);
\draw[thick,-Triangle](0,0)--(-0.7,0);
\draw[thick,-Triangle,blue](0,0)--(0,0.7);
\draw[thick,-Triangle](0,0)--(0,-0.7);
\draw[thick, -Triangle,blue] (0,0)--(0.65,-0.65);
\draw[thick, -Triangle] (0,0)--(-0.65,0.65);
\node[below,blue] at (0.65,-0.65){$\mathbf{e}_{1}-\mathbf{e}_{2}$};
\node[above left,blue] at (0,0.6){$\mathbf{e}_{2}$};
\end{tikzpicture}
\eea
which give
\bea
\{\bfe_{1},\bfe_{2}\},\quad \{\bfe_{1},-\bfe_{1}+\bfe_{2}\},\quad \{\bfe_{2},\bfe_{1}-\bfe_{2}\}.
\eea

\paragraph{General case} Using the condition $\bm{\eta}\in\operatorname{Cone}(\bfQ_{1},\ldots,\bfQ_{k})$, we can further restrict the choice of the poles.
\begin{lemma}\label{lem:JKantifund}
    The charge vectors $\{\bfQ_{I}\}$ do not contain elements of $\{-\mathbf{e}_{I}\}$.
\end{lemma}
\begin{proof}
    Due to the Weyl invariance of the $\{\mathbf{e}_{I}\}$, it is enough to show that we can not have $\bfQ_{1}=-\mathbf{e}_{1}$. Assume that we have this element and we have
    \bea\label{eq:JKantifundprf1}
    \bm{\eta}=\sum_{I=1}^{k}\mathbf{e}_{I}=\sum_{I=1}^{k}\lambda_{k}\bfQ_{I}=-\lambda_{1}\mathbf{e}_{1}+\sum_{I=2}^{k}\lambda_{I}\bfQ_{I},\quad \lambda_{i}>0.
    \eea
    For this condition to be obeyed, $\{\bfQ_{I}\}_{I=2}^{k}$ needs to contain a vector taking the form as $+\mathbf{e}_{1}$. From the linear independence, and the Weyl invariance, we can set $\bfQ_{2}=\mathbf{e}_{1}-\mathbf{e}_{2}$, which gives
    \bea
    \sum_{I=1}^{k}\bfe_{I}=(\lambda_{2}-\lambda_{1})\bfe_{1}-\lambda_{2}\bfe_{2}+\sum_{I=3}^{k}\lambda_{I}\bfQ_{I}.
    \eea
    Similarly $\{\bfQ_{I}\}_{I=3}^{k}$ needs to contain a vector taking the form as $+\bfe_{2}$. Doing this recursively, we can set
    \bea
(\bfQ_{1},\bfQ_{2},\ldots ,\bfQ_{k-1},\bfQ_{k})=(-\bfe_{1},\bfe_{1}-\bfe_{2},\ldots \bfe_{k-2}-\bfe_{k-1},\bfe_{k-1}-\bfe_{k})
    \eea
    which gives
    \bea
    \sum_{I=1}^{k}\bfe_{I}=\sum_{I=1}^{k-1}(\lambda_{I+1}-\lambda_{I})\bfe_{I}-\lambda_{k}\bfe_{k}.
    \eea
    The first $k-1$ terms can be tuned properly but the last term contradicts due to $\lambda_{k}>0$. Therefore, we can not have $\{-\bfe_{I}\}$ as the charge vectors.
\end{proof}

Since $\{\bfe_{I}-\bfe_{J}\}$ generates only the $\mathbb{R}^{k-1}$ subspace, there should be one or more charges coming from $\{\bfe_{I}\}$. Using the Weyl symmetry, let us choose $\bfe_{1},\ldots,\bfe_{p}$ to be the charges chose from $\{\bfe_{I}\}$. The other $k-p$ elements will be divided into $p$ groups such that only one of $\bfe_{I}\,\,(1\leq I\leq p)$ is included. For example, assume that we first choose $\bfe_{1}$. Elements such as $\bfe_{1}-\bfe_{J}$ are not included because if it is included, the condition~\eqref{eq:JKconecond} is translated to
\bea
\sum_{I=1}^{k}\bfe_{I}=(\lambda_{1}+\lambda_{2})\bfe_{1}-\lambda_{2}\bfe_{2}+\sum_{I=3}^{k}\lambda_{I}\bfQ_{I}
\eea
where we set $J=2$ for simplicity. From the linear independence, $\bfQ_{3\leq I}$ can not include $\{\bfe_{1},\bfe_{2}\}$ and after setting $\lambda_{1}+\lambda_{2}=1$, the above equation reduces to the computation in \eqref{eq:JKantifundprf1} and thus $\bfe_{1}-\bfe_{J}$ is forbidden. Charges such as $\bfe_{J}-\bfe_{1}$ will be only allowed and due to the linear independence, the charge $\bfe_{J}-\bfe_{1}$ needs to obey $J\notin\{ 2,\ldots, p\}$. Starting from the two elements $\bfe_{1},\bfe_{J}-\bfe_{1}$, if the remaining elements contain $\bfe_{J}$, then such an element must be $\bfe_{K}-\bfe_{J}$ for $K\notin\{1,\ldots,p\}\cup \{J\}$. This is because $\bfe_{J}$ breaks the linear independence and $\bfe_{J}-\bfe_{K}$ will not satisfy the condition \eqref{eq:JKconecond} as explained before. Doing this procedure recursively, the charge vectors obey a tree structure.
\begin{lemma}\label{lem:JKtree-structure}
    The charge vectors $\{\bfQ_{I}\}$ form an oriented tree structure. 
    \begin{itemize}[topsep=1ex,itemsep=-0.5ex,partopsep=1ex,parsep=1ex]
    \item The vertices are labeled with $\{\bfe_{I}\}_{I=1}^{k}$ and there are arrows connecting the vertices.
    \item The root vertex labeled with $\bfe_{J}$ correspond with the charge vector $+\bfe_{J}$. 
    \item The tree grows by adding vertices $\bfe_{K}$ to $\bfe_{J}$ with an arrow $J\rightarrow K$ and this correspond with the charge vector $\bfe_{K}-\bfe_{J}$. We call these arrows branches.
    \item There are only one vertex for each $\{\bfe_{I}\}_{I=1}^{k}$.
    \end{itemize}
\end{lemma}
For example, for $k=3$, we have
\bea
\adjustbox{valign=c}{\includegraphics[width=3cm]{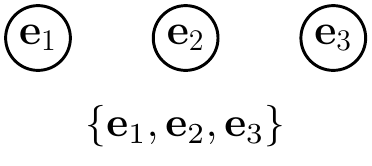}}\quad 
\adjustbox{valign=c}{\includegraphics[width=3cm]{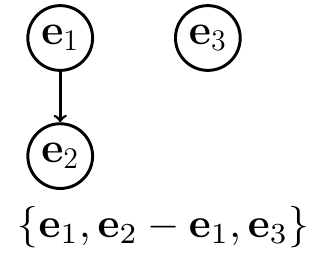}}\quad
\adjustbox{valign=c}{\includegraphics[width=4cm]{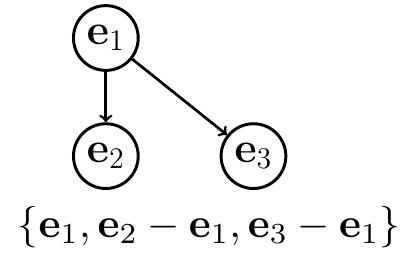}}\quad 
\adjustbox{valign=c}{\includegraphics[width=4cm]{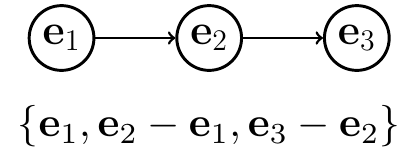}}
\eea
and other possible cases are obtained by the Weyl reflection.

\paragraph{Classification of poles and Young diagrams}
The above tree structure determines the poles of the contour integral \eqref{eq:pureSYMcontour}. However, not all of the poles are non-zero and some of them still give zero residues. The poles that will give non-zero residues have a nice combinatorial expression and they are actually labeled by colored Young diagrams \cite{Nakajima:1999,Nekrasov:2002qd,Nekrasov:2003rj}. For the $\U(n)$ pure super Yang--Mills case, they are $n$-tuples of Young diagrams.

Let us briefly summarize the notations we use in this thesis. A Young diagram (two-dimensional partition) is a sequence of non-negative integers obeying
\bea
\lambda=(\lambda_{1},\lambda_{2},\ldots,\lambda_{i},\ldots),\quad \lambda_{i}\geq \lambda_{i+1}\geq 0,\quad
    |\lambda|=\lambda_{1}+\cdots+\lambda_{\ell(\lambda)},\label{eq:Youngcond}
\eea
where $|\lambda|$ is the size, and $\ell(\lambda)$ is the length of the partition. The set of all possible 2d partitions/Young diagrams is denoted as $\mathcal{P}$ and the transpose is denoted as $\lambda^{\rmT}$ in the usual sense. A box $\Bbox$ positioned in the $i'$th row counted from the bottom 
and $j'$th column counted from the left in the Young diagram is assigned a coordinate $(i,j)$, $1\leq i\leq \ell(\lambda),\,1\leq j\leq \lambda_{i}$. For a box $\Bbox=(i,j)\in\lambda$, we assign the coordinate as\footnote{We call both $(i,j)$ and $c_{12}(\Bbox)$ coordinates.}
    \begin{equation}\label{eq:2dpartition-coordinate}
        \adjustbox{valign=c}{\begin{tikzpicture}
         \fill[red!20!white] (0.9,1.4)--(1.6,1.4)--(1.6,2.1)--(0.9,2.1)--(0.9,1.4);
        \draw[->] (-1,0)--(4,0);
        \node[above] at (-0.5,4){$\epsilon_{2}$};
        \node [right] at (4,0){$\epsilon_{1}$};
        \node[below] at (-0.15,0) {$1$};
        \node [below] at (0.55,0){$\cdots$};
        \node [below] at (1.25,0){$i$};
         \draw[->]   (-0.5,-0.5)--(-0.5,4);
         \draw (0.2,3.5)--(0.2,0.7);
         \draw (0.9,2.8)--(0.9,0.7);
         \draw (1.6,2.1)--(1.6,0.7);
         \draw (2.3,1.4)--(2.3,0.7);
         \draw (2.3,0.7)--(-0.5,0.7);
         \draw (2.3,1.4)--(-0.5,1.4);
         \draw (1.6,2.1)--(-0.5,2.1);
         \draw (0.9,2.8)--(-0.5,2.8);
         \draw (0.2,3.5)--(-0.5,3.5);
        \draw (-0.5,0.7)--(3,0.7);
        \draw (3,0)--(3,0.7);
        \draw (0.2,0)--(0.2,0.7);
        \draw (0.9,0)--(0.9,0.7);
         \draw (1.6,0)--(1.6,0.7);
          \draw (2.3,0)--(2.3,0.7);
          \draw (-0.15,0.35)--++(-0.7,-1);
          \node[left] at (-0.85,-0.65){$\mathfrak{a}$};
          \node [left] at (-0.5,0.35) {$1$};
          \node [left] at (-0.5,1.05){$\vdots$};
          \node [left] at (-0.5,1.75){$j$};
          \draw  (1.25,1.75)--++(0.9,0.9);
          \node[right] at (2.2,2.65) {$c_{12}(\Bbox)=c_{12}(\mathfrak{a},(i,j))=\mathfrak{a}+(i-1)\epsilon_{1}+(j-1)\epsilon_{2}$};
        \end{tikzpicture}
        }
    \end{equation}
The arm and leg length are defined as
\bea\label{eq:armleglength}
a_{\lambda}(\Bbox)=\lambda_{i}-j,\quad \ell_{\lambda}(\Bbox)=\lambda_{j}^{\rmT}-i,\quad \Bbox=(i,j).
\eea
For later use, we also define the multiplicative coordinates ($q$-coordinates) as
\bea
\chi_{12,v}(\Bbox)=e^{c_{12}(\Bbox)}=vq_{1}^{i-1}q_{2}^{j-1},\quad v=e^{\mathfrak{a}},\quad q_{1,2}=e^{\epsilon_{1,2}}.
\eea

\begin{theorem}\label{thm:pureSYMJKpoles}
    The poles of \eqref{eq:pureSYMcontour} are classified by $n$-tuples of Young diagrams $\vec{\lambda}=\{\lambda^{(\alpha)}\}_{\alpha=1}^{n}$:
    \bea
    \{\phi_{I}\}_{I=1}^{k}\rightarrow  \{\mathfrak{a}_{\alpha}+(i-1)\epsilon_{1}+(j-1)\epsilon_{2}\mid (i,j)\in\lambda^{(\alpha)}\}_{\alpha=1}^{n}
    \eea
    where $k$ corresponds to the total number of boxes in the Young diagrams $k=\sum_{\alpha=1}^{n}|\lambda^{(\alpha)}|$. Namely, each pole corresponds to the coordinate of the box in the Young diagram.
\end{theorem}
\begin{proof}
We show this by induction. For $k=1$, the pole come from $\phi_{1}-\mathfrak{a}_{\alpha}=0$ for some $\alpha$ which corresponds to a configuration where there is only one nonempty Young diagram with one box and others are empty. 

Assume that for rank $k-1$, the Young diagram rule is true and the $k$ independent hyperplane equations are $\{\bfQ_{1},\ldots,\bfQ_{k}\}$. When the trees have no branches, the poles will be $\{\phi_{I}-\mathfrak{a}_{\alpha}=0\}$, which is the situation when the nonempty Young diagrams contain only one box. Let us consider the case when $\bfQ_{k}=\bfe_{I}-\bfe_{J}$ for some $I,J$ is a branch at the end of a tree. This condition determines the pole for $\phi_{I}$. Using the Weyl invariance, we can assume that $\{\phi_{1},\ldots,\phi_{k-1}\}$ are determined by $\{\bfQ_{1},\ldots,\bfQ_{k-1}\}$ and their corresponding hyperplanes, and $\phi_{k}$ is the last pole determined by $\bfQ_{k}=\bfe_{k}-\bfe_{J}$ for some $J$ and its corresponding hyperplane $H_{k}$. By the induction hypothesis, the poles $\{\phi_{1},\ldots,\phi_{k-1}\}$ are classified by Young diagrams with $k-1$ boxes denoted as $\vec{\lambda}$. We thus need to show that the possible choices of $\bfQ_{k}$ have a one to one correspondence with the possible way to add boxes to $\vec{\lambda}$ giving $\vec{\lambda}'=\vec{\lambda}+\Bbox$. 

We can assume that $\bfe_{k}$ belongs to a tree whose root vertex corresponds to the pole at $\mathfrak{a}_{\ast}$. We denote the corresponding Young diagrams as $\lambda_{\ast},\lambda'_{\ast}=\lambda_{\ast}+\Bbox$. The factors related when evaluating $\phi_{k}$ are
\bea
\,&\prod_{\alpha=1}^{n}\frac{1}{\sh(\phi_{k}-\mathfrak{a}_{\alpha})\sh(\mathfrak{a}_{\alpha}-\phi_{k}-\epsilon_{12})}\prod_{\Abox\in\lambda_{\ast}}\frac{\sh(\phi_{k}-c_{12}(\Bbox))\sh(\phi_{k}-c_{12}(\Bbox)+\epsilon_{12})}{\sh(\phi_{k}-c_{12}(\Bbox)-\epsilon_{1})\sh(\phi_{k}-c_{12}(\Bbox)-\epsilon_{2})}\\
& \qquad \times\prod_{\Abox\in\lambda_{\ast}}\frac{\sh(c_{12}(\Bbox)-\phi_{k})\sh(c_{12}(\Bbox)-\phi_{k}+\epsilon_{12})}{\sh(c_{12}(\Bbox)-\phi_{k}-\epsilon_{1})\sh(c_{12}(\Bbox)-\phi_{k}-\epsilon_{2})}.
\eea
Other factors in \eqref{eq:pureSYMcontour} will be regular in $\phi_{k}$. The pole $\phi_{k}$ is then determined by $\phi_{k}-\phi_{J}=\epsilon_{1,2}$. Due to the symmetry between $\epsilon_{1}\leftrightarrow \epsilon_{2}$ we can focus on the pole coming from $\phi_{k}-\phi_{J}=\epsilon_{1}$. We denote $\phi_{J}=\mathfrak{a}_{\ast}+(x-1)\epsilon_{1}+(y-1)\epsilon_{2}$ and $\phi_{k}=\mathfrak{a}_{\ast}+x\epsilon_{1}+(y-1)\epsilon_{2}$. Under this situation, we have the following properties.
\begin{itemize}[topsep=1ex,itemsep=-0.5ex,partopsep=1ex,parsep=1ex]
    \item $\phi_{J}$ needs to be in the boundary which means that $\phi_{k}$ can not be included in the Young diagram. If $(x+1,y)\in\lambda_{\ast}$, the zero coming from $\prod_{I\neq J}\sh(\phi_{I}-\phi_{J})$ cancels the pole
    \bea
    \frac{\sh(\phi_{k}-c_{12}(\mathfrak{a}_{\ast},x+1,y))\sh(c_{12}(\mathfrak{a}_{\ast},x+1,y)-\phi_{k})}{\sh(\phi_{k}-c_{12}(\mathfrak{a}_{\ast},x,y)-\epsilon_{1})}\rightarrow 0.
    \eea
    \item When $(x,y)\in\lambda_{\ast}$ belongs to the boundary and $x>1,y=1$, no pole cancellation occurs and we have a simple pole. Namely, we can always add a box to this configuration without breaking the Young diagram condition.
    \item When $(x,y)$ belongs to the boundary and $x>1,y>1$, the induction hypothesis says that $\lambda_{\ast}$ obeys the Young diagram condition, i.e. $(x,y-1)\in\lambda_{\ast}$. We then have two possibilities depending on the position of $(x+1,y)$.
    \begin{itemize}[topsep=1ex,itemsep=-0.5ex,partopsep=1ex,parsep=1ex]
        \item If $(x+1,y-1)\notin\lambda_{\ast}$, i.e. $\lambda'_{\ast}$ does not obey the Young diagram condition, then the simple pole vanishes and the JK residue is zero. 
        \item If $(x+1,y-1)\in\lambda_{\ast}$, i.e. $\lambda'_{\ast}$ obeys the Young diagram condition, then there is a simple pole and the JK residue is nonzero.
    \end{itemize}
    This comes from
    \bea
    &\frac{\sh(c_{12}(\mathfrak{a}_{\ast},x,y-1)-\phi_{k}+\epsilon_{12})}{\sh(\phi_{k}-c_{12}(\mathfrak{a}_{\ast},x,y)-\epsilon_{1})}\times\begin{dcases}
        1\\
        \frac{1}{\sh(\phi_{k}-c_{12}(\mathfrak{a}_{\ast},x+1,y-1)-\epsilon_{2})}
    \end{dcases}\\
    \propto &\begin{dcases}
        1,\quad  (x+1,y-1)\notin\lambda_{\ast}\\
        \frac{1}{\sh(\phi_{k}-c_{12}(\mathfrak{a}_{\ast},x,y)-\epsilon_{1})},\quad (x+1,y-1)\in\lambda_{\ast}
    \end{dcases}
    \eea
\end{itemize}
This concludes the inductive proof and the non-zero JK residues are indeed classified by $n$-tuples of Young diagrams.

\end{proof}

The above discussion of the classification of the poles already gives the iterative process to take the JK-residue, which was given in an abstract way in \eqref{eq:JK-iteratedmethod}. Focusing on the $n=1$ case, the poles are classified by the boxes included in a Young diagram. Given a Young diagram $\lambda$ with $k$-boxes, the non-zero JK-residue is computed by taking the iterated integral in the order of the stacking of the boxes. Let $\{\phi_{1\ast},\phi_{2\ast},\cdots \phi_{k\ast}\}$ be the sequence of the coordinates of the $k$-boxes in the Young diagram. For each step, the collection of boxes $\{\phi_{1\ast},\cdots,\phi_{i\ast}\}$ is a Young diagram with $i$-boxes obeying the condition \eqref{eq:Youngcond}. Using the Weyl invariance, we can say that the contour integral $\oint d\phi_{I}$ picks up the pole at $\phi_{I\ast}$. Then, the multi-contour integral is evaluated in the order $\oint_{\phi_{k}=\phi_{k\ast}} d\phi_{k}\cdots \oint_{\phi_{1}=\phi_{1\ast}} d\phi_{1}$. The resulting residue actually does not depend on the ordering of $\{\phi_{1\ast},\cdots \phi_{k\ast} \}$ as long as the ordering is compatible with the condition that for each step the Young diagram condition is obeyed. Thus, it is useful to fix the ordering in the boxes of the Young diagram and take the iterative residue in such ordering. We define an ordering of the boxes in the Young diagram as
\begin{equation}
        \adjustbox{valign=c}{\begin{tikzpicture}
         \draw[->] (-0.5,0)--(3,0);
        \draw[->] (0,-0.5)--(0,3);
        \node[above] at (0,3){$\epsilon_{2}$};
        \node [right] at (3,0){$\epsilon_{1}$};
        \draw (0.7,0)--(0.7,2.1);
        \draw (1.4,0)--(1.4,1.4);
        \draw (2.1,0)--(2.1,0.7);
        \draw (0,2.1)--(0.7,2.1);
        \draw (0,1.4)--(1.4,1.4);
        \draw (0,0.7)--(2.1,0.7);
        \node at (0.35,0.35){$1$};
        \node at (1.05,0.35){$2$};
        \node at (1.75,0.35){$3$};
         \node at (0.35,1.05){$4$};
        \node at (1.05,1.05){$5$};
        \node at (0.35,1.75){$6$};
        \end{tikzpicture}
        }\qquad (i,j)<(i',j')\Leftrightarrow  (j<j')\vee (j=j',i<i').
\end{equation}
This ordering induces an ordering in the coordinates $c_{12}(\Bbox),\chi_{12,v}(\Bbox)$
\bea\label{eq:pureSYM-qcoordinates-ordering}
c_{12}(\Bbox) <c_{12}(\BboxF),\quad \chi_{12,v}(\Bbox)<\chi_{12,v}(\BboxF)
\eea
for $\Bbox <\BboxF$.

The $\U(1)$ partition function is then written as
\bea
\mathcal{Z}_{\text{SYM}}=\sum_{\lambda}\mathfrak{q}^{|\lambda|}\mathcal{Z}_{\text{SYM}}[\lambda],\quad \mathcal{Z}_{\text{SYM}}[\lambda]= \underset{\phi=\phi_{\lambda}}{\Res}\mu_{k}(\mathfrak{a}_{\alpha},\phi_{I},\epsilon_{1,2})
\eea
where we denoted the iterated residue as
\bea\label{eq:pureSYM-iterativeresidue}
\underset{\phi=\phi_{\lambda}}{\Res}\mu_{k}(\mathfrak{a}_{\alpha},\phi_{I},\epsilon_{1,2})=\underset{\phi_{k}=\phi_{k}^{\ast}}{\Res}\cdots \underset{\phi_{1}=\phi_{1}^{\ast}}{\Res}\mu_{k}(\mathfrak{a}_{\alpha},\phi_{I},\epsilon_{1,2})
\eea
and $k=|\lambda|$. The ordering in the sequence $\{\phi_{I}^{\ast}\}_{I=1}^{k}$ is given by the ordering above. 

For higher rank cases, we first decompose the $k$-variables $\{\phi_{I}\}_{I=1}^{k}$ into $n$ groups $\{\phi_{\alpha,1},\ldots \phi_{\alpha,k_{\alpha}}\}_{\alpha=1}^{n}$ with $\sum_{\alpha}k_{\alpha}=k$. For each $k_{\alpha}$ variables $\{\phi_{\alpha,I}\}_{I=1}^{k_{\alpha}}$, we assign the Young diagram $(\mathfrak{a}_{\alpha},\lambda^{(\alpha)})$ and the coordinates of the boxes $\{\phi_{\alpha,I}^{\ast}\}$. The iterated residue can then be taken as
\bea\label{eq:pureSYM-higherrank-iterativeresidue}
\underset{\phi_{n,k_{n}}=\phi_{n,k_{n}}^{\ast}}{\Res}\cdots \underset{\phi_{n,1}=\phi_{n,1}^{\ast}}{\Res}\cdots \underset{\phi_{2,k_{2}}=\phi_{2,k_{2}}^{\ast}}{\Res}\cdots \underset{\phi_{2,1}=\phi_{2,1}^{\ast}}{\Res}\,\,\underset{\phi_{1,k_{1}}=\phi_{1,k_{1}}^{\ast}}{\Res}\cdots \underset{\phi_{1,1}=\phi_{1,1}^{\ast}}{\Res}
\mu_{k}(\mathfrak{a}_{\alpha},\phi_{I},\epsilon_{1,2}).
\eea

\section{Characters and index}\label{sec:character-index}
To get the instanton partition function, we need to evaluate the iterative residue \eqref{eq:pureSYM-iterativeresidue} explicitly. Moreover, for practical use, it is convenient to find a compact way to write down the partition function for each Young diagram $\lambda$. In section~\ref{sec:residue-formula}, we discuss basic residue formulas and the relation with \textit{characters}. We then rewrite the partition function of pure SYM in terms of characters at section~\ref{sec:pureSYM-indexformalism}.

\subsection{Residue formulas}\label{sec:residue-formula}
\paragraph{Index functor}
For a character defined as
\bea
    \bfX=\sum_{i}n_{i}e^{x_{i}},
\eea
where $n_{i}\in\mathbb{Z}$, the index functor to convert the additive components to multiplicative components is defined as\footnote{In this thesis, we will use two notations of the trigonometric index. One is the index here $\llbracket x\rrbracket=1-e^{-x}$ and the other is the symmetric index $\sh(x)=e^{x/2}-e^{-x/2}$ used in the JK-residue. For quantum algebraic use, we will use the former notation so that we do not need to deal with the square root factors. At the level of partition function, the difference will only appear as topological terms and Chern--Simons terms.} 
\bea
    \mathbb{I}\left[\bfX\right]=\prod_{i}\llbracket x_{i}\rrbracket ^{n_{i}},\qquad \llbracket x\rrbracket=\begin{dcases}
        x,\quad &(\text{4d})\\
        1-e^{-x},\quad &(\text{5d})\\
        \theta(e^{-x};p),\quad &(\text{6d})
    \end{dcases}\label{eq:rat_trig_ell}
\eea
where the theta function is $\theta(x;p)=(x;p)_{\infty}(px^{-1};p)_{\infty}$ (see Appendix~\ref{app:specialfunct} for details of the notation).

The hierarchical structure between rational, trigonometric, and elliptic functions appears here by taking the limit as 
\bea
    \theta(e^{-x};p)\xrightarrow{p\rightarrow0} 1-e^{-x}=x+\mathcal{O}(x^{2}).
\eea

Most of the computations in this paper will be done explicitly using the trigonometric notation, so when not mentioned we are using the following convention:
\bea
    \mathbb{I}[x]=(1-x^{-1})=\exp\left(-\sum_{n=1}^{\infty}\frac{1}{n}x^{-n}\right)
\eea
but using the formula in \eqref{eq:rat_trig_ell} one can convert the results to rational and elliptic ones. 

The dual of the character is defined as
\bea\label{eq:dualcharacter}
   \bfX^{\vee}=\sum_{i}n_{i}e^{-x_{i}}
\eea
and we have the reflection property
\bea
    \mathbb{I}\left[\bfX^{\vee}\right]=(-1)^{\text{rk}\bfX}\det\bfX\,\,\mathbb{I}[\bfX]\label{eq:index_reflecprop}
\eea
where $\text{rk}\bfX=\sum_{i}n_{i}$ and $\det\bfX=\prod_{i}e^{n_{i}x_{i}}$. 

For a character $\bfX$, we also define an operation
\beq\label{eq:pdegree-op}
   \bfX^{[p]}=\sum_{i}n_{i}e^{px_{i}}
\eeq
which is called the $p$-th Adams operation.

\begin{definition}\label{app-def:movable}
    Let $\mathbf{A}$ be a Laurent polynomial
    \bea
    \mathbf{A}=\sum_{(n_{1},\ldots ,n_{p})}A_{\vec{n}}x_{1}^{n_{1}}\cdots x_{p}^{n_{p}},
    \eea
    where the sum $(n_{1},\ldots,n_{p})$ is taken over some subset in $\mathbb{Z}^{p}$. If there is no constant term $A_{\vec{n}}=0$ for $n_{1}=n_{2}=\cdots n_{p}=0$, then $\mathbf{A}$ is \textbf{movable}. The constant term is called the \textbf{unmovable} part. We denote the movable part and unmovable part as
    \bea
\left[\mathbf{A}\right]^{(\neq0)},\quad \left[\mathbf{A}\right]^{(0)},
    \eea
    respectively. Note that this means
    \bea
    \mathbf{A}=\left[\mathbf{A}\right]^{(\neq0)}+\left[\mathbf{A}\right]^{(0)}.
    \eea
\end{definition}


\paragraph{Residue formulas}In the following sections, we will need to take residues of rational functions. We summarize the properties of the residue formulas.
\begin{definition}\label{def:residuedef-normal}
    Let $f(x)$ be a rational function whose zeros and poles are $\{p^{+}_{i}\}$ and $\{p^{-}_{j}\}$ respectively:
    \bea
f(x)=\frac{\prod_{j}(x-p_{j}^{+})}{\prod_{i}(x-p_{i}^{-})}.
    \eea
    The residue at $p_{a}^{-}$ is defined as
    \bea
    \underset{x=p_{a}^{-}} {\Res} f(x)=\lim_{x\rightarrow p_{a}^{-}}(x-p_{a}^{-})f(x)=\frac{\prod_{j}(p_{a}^{-}-p_{j}^{+})}{\prod_{i\neq a}(p_{a}^{-}-p_{i}^{-})}.
    \eea
\end{definition}
This is the usual definition of the residue of a rational function. In the equivariant index formalism, it is natural to deal with functions $\tilde{f}(x)$ written in products $(1-a/x)^{-1}$. In such cases, the residue is defined as $\underset{x=a}{\Res}x^{-1}\tilde{f}(x)$.
\begin{definition}\label{def:residuedef}
Given a rational function $\tilde{f}(x)$ written as
\bea
\tilde{f}(x)=\mathbb{I}[\mathbf{X}^{\vee}x]=\frac{\prod_{j}(1-p_{j}^{+}/x)}{\prod_{i}(1-p_{i}^{-}/x)},\quad \mathbf{X}=\sum_{j}p_{j}^{+}-\sum_{i}p_{i}^{-},
\eea
the residue is defined as
\bea
\underset{x=p_{a}^{-}}{\Res}x^{-1}\tilde{f}(x)\coloneqq\lim_{x\rightarrow p_{a}^{-}}\left(1-\frac{p_{a}^{-}}{x}\right)\tilde{f}(x)=\frac{\prod_{j}(1-p_{j}^{+}/p_{a}^{-})}{\prod_{i\neq a}(1-p_{i}^{-}/p_{a}^{-})}.
\eea
\end{definition}
Note that the right hand side can be written also as
\bea
\mathbb{I}[\mathbf{X}^{\vee}p_{a}^{-}+1]=\frac{\prod_{j}(1-p_{j}^{+}/p_{a}^{-})}{\prod_{i\neq a}(1-p_{i}^{-}/p_{a}^{-})}
\eea
and thus we have
\bea\label{eq:residueindexformula}
\underset{x=p_{a}^{-}}{\Res}x^{-1}\tilde{f}(x)=\lim_{x\rightarrow p_{a}^{-}}\mathbb{I}\left[\mathbf{X}^{\vee}x+1\right].
\eea

We may also study the residue formula for the dual rational function $\tilde{f}^{\vee}(x)$, which is defined as
\bea
\tilde{f}^{\,\vee}(x)=\mathbb{I}[\mathbf{X}x^{\vee}]=\frac{\prod_{j}(1-x/p_{j}^{+})}{\prod_{i}(1-x/p_{i}^{-})}.
\eea
\begin{proposition}\label{prop:residuedual}
    The residue formula is given as
    \bea
    \underset{x=p_{a}^{-}}{\Res} x^{-1}\tilde{f}^{\,\vee}(x)=\lim_{x\rightarrow p_{a}^{-}}\left(1-\frac{p_{a}^{-}}{x}\right)\tilde{f}^{\,\vee}(x)=-\frac{\prod_{j}(1-p_{a}^{-}/p_{j}^{+})}{\prod_{i\neq a}(1-p_{a}^{-}/p_{i}^{-})}.
    \eea
    Moreover the right hand side can be written also as
    \bea\label{eq:dualresidueindexformula}
 \underset{x=p_{a}^{-}}{\Res} x^{-1}\tilde{f}^{\,\vee}(x)=- \mathbb{I}\,[\mathbf{X}p_{a}^{-\vee}+1],
    \eea
    where an extra \textit{sign} appears.
\end{proposition}
The sign appearing above will be important in later sections and one needs to be careful with these signs otherwise the partition function obtained from the index formalism will not give the correct partition function.

Generally, we will have the following relation between the residue formula and the index computations.
\begin{proposition}\label{prop:residuegeneral}
    Let $g(x)$ be a rational function defined as
    \bea
    g(x)=\mathbb{I}[\mathbf{X}^{\vee}x+\mathbf{Y}x^{-1}],\quad \mathbf{X}=\sum_{j}p_{j}^{+}-\sum_{i}p_{i}^{-},\quad \mathbf{Y}=\sum_{j}q_{j}^{+}-\sum_{i}q_{i}^{-}.
    \eea
    Explicitly, this is written as
    \bea
    g(x)=\frac{\prod_{j}(1-p_{j}^{+}/x)}{\prod_{i}(1-p_{i}^{-}/x)}\frac{\prod_{j}(1-x/q_{j}^{+})}{\prod_{i}(1-x/q_{i}^{-})},
    \eea
    The residue formula has the following relation
    \bea
        \underset{x=y}{\Res}x^{-1}g(x)=\begin{dcases}
            +\mathbb{I}\,[\mathbf{X}^{\vee}y+\mathbf{Y}y^{-1}+1],\quad y\in\{p_{i}^{-}\}\\
            -\mathbb{I}\,[\mathbf{X}^{\vee}y+\mathbf{Y}y^{-1}+1],\quad y\in\{q_{i}^{-}\}
        \end{dcases}.
    \eea
\end{proposition}


\subsection{Pure SYM}\label{sec:pureSYM-indexformalism}
For later use, we change the notations used in the previous section. Instead of using the $\sh(\phi)$ notation, we use the $\mathbb{I}[x]$ notation. 

\begin{proposition}\label{prop:LMNSformula-mod}
The LMNS formula is rewritten as
\bea
\mathcal{Z}_{k}&=\frac{1}{k!}\oint_{\text{JK}} \prod_{I=1}^{k}\frac{dx_{I}}{2\pi ix_{I}}z_{k}(v_{\alpha},x_{I},q_{1,2}),\\
z_{k}(v_{\alpha},x_{I},q_{1,2})&=\left(\frac{(1-q_{12})}{(1-q_{1})(1-q_{2})}\right)^{k}\prod_{I=1}^{k}\frac{1}{P(x_{I})\widetilde{P}(q_{12}x_{I})}\prod_{I\neq J}\mathscr{S}_{12}\left(\frac{x_{I}}{x_{J}}\right)^{-1}
\eea
where
\bea
P(x)=\prod_{\alpha=1}^{n}\left(1-\frac{v_{\alpha}}{x}\right),\quad \widetilde{P}(x)=\prod_{\alpha=1}^{n}\left(1-\frac{x}{v_{\alpha}}\right),\quad \mathscr{S}_{12}(x)=\frac{(1-q_{1}x)(1-q_{2}x)}{(1-x)(1-q_{12}x)},\quad q_{12}=q_{1}q_{2}
\eea
and $q_{1,2}=e^{\epsilon_{1,2}},\,v_{\alpha}=e^{\mathfrak{a}_{\alpha}}, x_{I}=e^{\phi_{I}}$. 
\end{proposition}
Note that $d\phi_{I}=x_{I}^{-1}dx_{I}$. The notation here is the same with \eqref{eq:pureSYMcontour} up to redefinition of the topological terms and inclusion of Chern--Simons terms.

Let us understand the rational functions inside the contour integral by using the characters. Define
\bea
\bfN=\sum_{\alpha=1}^{n}e^{\mathfrak{a}_{\alpha}}=\sum_{\alpha=1}^{n}v_{\alpha},\qquad \bfK=\sum_{I=1}^{k}e^{\phi_{I}}=\sum_{I=1}^{k}x_{I},\quad \bfP_{1,2}=1-q_{1,2},\quad \bfP_{12}=\bfP_{1}\bfP_{2}
\eea
and one can see that 
\bea\label{eq:pureSYMcharacter}
   \mathcal{Z}_{k}=\frac{1}{k!}\oint\prod_{I=1}^{k}\frac{dx_{I}}{2\pi ix_{I}} \mathbb{I}[\mathbf{v}_{\text{inst.}}-k],\quad \mathbf{v}_{\text{inst.}}=-\bfN^{\vee}\bfK-q_{12}^{\vee}\bfK^{\vee}\bfN+\bfP_{12}^{\vee}\bfK^{\vee}\bfK,
\eea
where we used
\bea
P(x)=\mathbb{I}[\bfN^{\vee} x],\quad \widetilde{P}(x)=\mathbb{I}[x^{-1}\bfN],\quad \mathscr{S}_{12}(x)=\mathbb{I}[-\bfP_{12}^{\vee}x^{\vee}].
\eea

Thm.~\ref{thm:pureSYMJKpoles} tells us that the poles are classified by $n$-tuples of Young diagrams as
\begin{equation}
\{x_{I}\}_{I=1,\ldots,k}\longrightarrow \{\chi_{12,v_{\alpha}}(\Bbox)=v_{\alpha}q_{1}^{i-1}q_{2}^{j-1}\}_{\alpha=1,\ldots,n,\,\Abox=(i,j)\in\lambda^{(\alpha)}}
\end{equation}
and we obtain
\bea
\mathcal{Z}_{\text{SYM}}=\sum_{\vec{\lambda}}\mathfrak{q}^{|\vec{\lambda}|}\mathcal{Z}_{\text{SYM}}[\vec{v},\vec{\lambda}],\quad \mathcal{Z}_{\text{SYM}}[\vec{v},\vec{\lambda}]=\underset{x=x_{\vec{\lambda}}}{\Res}x^{-1} z_{k}(v_{\alpha},x_{I},q_{1,2})
\eea
where $\vec{v}=(v_{\alpha})_{\alpha=1,\ldots,n},\,\,  \vec{\lambda}=(\lambda^{(\alpha)})_{\alpha=1,\ldots,n},\,\,|\vec{\lambda}|=\sum_{\alpha=1}^{n}|\lambda^{(\alpha)}|=k$ and the iterative residue \eqref{eq:pureSYM-iterativeresidue} is defined as
\bea\label{eq:pureSYM-mult-residue}
\underset{x=x_{\vec{\lambda}}}{\Res}x^{-1} z_{k}(v_{\alpha},x_{I},q_{1,2})&\coloneqq \underset{x_{k}=x_{k}^{\ast}}{\Res}x_{k}^{-1}\cdots \underset{x_{1}=x_{1}^{\ast}}{\Res}x_{1}^{-1}\,z_{k}(v_{\alpha},x_{I},q_{1,2}).
\eea
The $x_{1}^{\ast},\ldots ,x_{k}^{\ast}$ are the exponent of the coordinates $c_{12}(\Bbox)$ of the $k$-boxes in the $n$-tuples of Young diagrams $\vec{\lambda}$ in the ordering \eqref{eq:pureSYM-higherrank-iterativeresidue}. Let us evaluate this residue explicitly and rewrite the partition function using characters. We first consider the case $n=1$ giving
\bea
z_{k}(v,x_{I},q_{1,2})&=\left(\frac{(1-q_{12})}{(1-q_{1})(1-q_{2})}\right)^{k}\prod_{I=1}^{k}\frac{1}{(1-v/x)(1-q_{12}x/v)}\prod_{I>J}\mathscr{S}_{12}\left(\frac{x_{I}}{x_{J}}\right)^{-1}\mathscr{S}_{12}\left(\frac{x_{J}}{x_{I}}\right)^{-1}.
\eea
When taking the residue at $x_{I}$, residues at $x_{J}$ $(J<I)$ are already taken and they will be some variables taking in value of $\{\chi_{12,v}(\Bbox)\mid  \Bbox\in\lambda\}$. Moreover, since we are assuming that $x_{I}^{\ast}>x_{J}^{\ast}$ obeys, the poles are always picked up from the $\prod_{I>J}\mathscr{S}_{12}(x_{J}/x_{I})^{-1}$. This property is also true for generic $n$ and after using the residue formula \eqref{eq:residueindexformula}, we then have
\bea
\underset{x=x_{\vec{\lambda}}}{\Res}x^{-1} z_{k}(v_{\alpha},x_{I},q_{1,2})=\mathbb{I}[\,\mathbf{v}_{\text{inst.}}|_{\vec{\lambda}}]
\eea
where $\mathbf{v}_{\text{inst.}}|_{\vec{\lambda}}$ means we insert the information of the poles to the character $\bfK$ as
\bea
\bfK|_{\vec{\lambda}}=\sum_{\alpha=1}^{n}\sum_{\Abox\in\lambda^{(\alpha)}}\chi_{12,v_{\alpha}}(\Bbox).
\eea
Therefore, the instanton partition function is given as
\bea
\mathcal{Z}_{\text{SYM}}[\vec{v},\vec{\lambda}]=\mathbb{I}\left[\left.\mathbf{v}_{\text{inst.}}\right|_{\vec{\lambda}}\right].
\eea
We can also decompose this partition function using fundamental factors called the \textit{Nekrasov factor} as
\bea
\mathcal{Z}_{\text{SYM}}[\vec{v},\vec{\lambda}]=\prod_{\alpha,\beta=1}^{n}\frac{1}{\mathsf{N}_{12}(v_{\alpha},\lambda^{(\alpha)}\mid v_{\beta},\lambda^{(\beta)})}.
\eea
where
\bea
    \mathsf{N}_{12}(v_{1},\lambda^{(1)}\,|\,v_{2},\lambda^{(2)})&=\prod_{\Abox\in\lambda^{(1)}}\left(1-\frac{q_{12}\chi_{12,v_{1}}(\Bbox)}{v_{2}}\right)\prod_{\Abox\in\lambda^{(2)}}\left(1-\frac{v_{1}}{\chi_{12,v_{2}}(\Bbox)}\right)\prod_{\substack{\Abox\in\lambda^{(1)}\\\AboxF\in\lambda^{(2)}}}\mathscr{S}_{12}\left(\frac{\chi_{12,v_{1}}(\Bbox)}{\chi_{12,v_{2}}(\BboxF)}\right).
\eea

\chapter{Gauge Origami}\label{chap:gauge-origami}
In this chapter, we discuss a physical setup called \textit{gauge origami}, originally introduced by Nekrasov \cite{Nekrasov:2015wsu,Nekrasov:2016qym,Nekrasov:2016ydq}. The gauge origami setup is a physical setup where generalizations of the instantons in Chap.~\ref{chap:ADHM-localization} appear. Two generalizations of instantons will be discussed. One is instantons of higher dimensional gauge theories, and the other is instantons of \textit{intersecting} gauge theories. String theoretically, we will show that they appear as the low energy worldvolume theories of higher dimensional and generally intersecting D-branes. The goal of this chapter is to derive the ADHM-like equations of these generalized instantons from the string theoretic viewpoint. 

In particular, in section~\ref{sec:general-instanton}, we briefly summarize the ADHM data of the instantons we are interested in. We then explain the physical setup of the gauge origami system in section~\ref{sec:gaugeorigamiphysicalsetup}. Similar to the pure SYM case, supersymmetries will play important roles and they are reviewed in section~\ref{sec:susyBfield}. The NS $B$-field is inevitable to consider such systems and it is also reviewed. At the end, we will see that two supersymmetries are preserved in all the setups.  We then briefly study the open string spectrum in section~\ref{sec:opstringspectrum} and construct 2d $\mathcal{N}=(0,2)$ quiver gauge theories appearing at the low energy limit. The vacuum moduli space at the end coincides with the ADHM equations.

\section{Generalized instantons}\label{sec:general-instanton}
In this section, we discuss generalizations of instantons in higher dimensional and intersecting gauge theories. We give the information of the instanton moduli space in a top-down way. The physical interpretations of them are explained in the following sections. See \cite{Szabo:2022zyn,Kanno:2020ybd} for reviews and references.

\paragraph{Higher dimensional instantons}Let $M_{d}$ be a connected oriented Riemannian manifold of even dimension $2d\geq 4$ and $\star$ the corresponding Hodge dual. Generalized instantons are defined as localized finite-action solutions of the gauge theory on $M_{n}$ obeying
\bea
\Sigma\wedge F=\star F
\eea
where $A$ is the 1-form gauge connection, $F=dA+A\wedge A$ is the curvature, and $\Sigma$ is a differential form of degree $2d-4$. By the famous Derrick's Theorem \cite{Derrick:1964ww}, localized solutions with finite YM action is forbidden for $n>2$. For a localized solution $A=A_{\mu}dx^{\mu}$, we can rescale the coordinates $x\rightarrow\lambda x$ with $\lambda\in\mathbb{R}_{>0}$. The action changes as $S_{\text{YM}}(\lambda)=\lambda^{4-2d}S_{\text{YM}}$. Since $S'_{\text{YM}}(1)<0$ for $d>2$, we do not have a localized solution. In this sense, naively such kind of instantons do not exist. We can cure the problem by considering \textit{noncommutative} instantons because Derrick's Theorem does not hold for such nontrivial background. The instantons we are considering will be noncommutative instantons.

Let us quote the ADHM type equations for instantons in higher dimensional theories. We introduce $\bfN=\mathbb{C}^{n}, \bfK=\mathbb{C}^{k}$ similar to Chap.~\ref{chap:ADHM-localization}.
\begin{itemize}
    \item $d=3$, $M_{3}=\mathbb{C}^{3}$ \cite{Nekrasov:2009JJM,Jafferis:2007sg,Cirafici:2008sn,Kanno:2020ybd,Benini:2018hjy}: 
    \bea\label{eq:D6-ADHM}
    \mu_{\mathbb{C}}&=[\mathsf{B}_{i},\mathsf{B}_{j}]+\frac{1}{2}\varepsilon_{ijk}[\mathsf{B}_{k}^{\dagger},\mathsf{Y}],\quad 
    \mu_{\mathbb{R}}=\sum_{i=1}^{3}[\mathsf{B}_{i},\mathsf{B}_{i}^{\dagger}]+[\mathsf{Y},\mathsf{Y}^{\dagger}]+\mathsf{I}\mathsf{I}^{\dagger},\quad \sigma=\mathsf{Y} \mathsf{I}=0,
    \eea
    where $\mathsf{B}_{1,2,3},\mathsf{Y}\in\operatorname{Hom}(\bfK,\bfK)$ and $\mathsf{I}\in \operatorname{Hom}(\bfN,\bfK)$. The moduli space is given by
    \bea
    \mathfrak{M}_{n,k}=\left\{(\mathsf{B}_{1,2,3},\mathsf{Y})\,|\, \mu_{\mathbb{C}}=\mu_{\mathbb{R}}-\zeta\cdot 
    1_{k},\sigma=0\right\}/\U(k).
    \eea
    \item $d=4$, $M_{4}=\mathbb{C}^{4}$ \cite{Nekrasov:2017cih,Nekrasov:2018xsb}:
    \bea
\mu_{\mathbb{R}}&=\sum_{a\in\four}[\mathsf{B}_{a},\mathsf{B}_{a}^{\dagger}]+\mathsf{I}\mathsf{I}^{\dagger},\quad s_{ab}=[\mathsf{B}_{a},\mathsf{B}_{b}]+\frac{1}{2}\varepsilon_{abcd}[\mathsf{B}_{c}^{\dagger},\mathsf{B}_{d}^{\dagger}],
\eea
where $\mathsf{B}_{a}\in \operatorname{Hom}(\bfK,\bfK)\,(a=1,2,3,4)$ and $\mathsf{I}\in\Hom(\bfN,\bfK)$. The $\varepsilon_{abcd}$ is the total antisymmetric tensor with $\varepsilon_{1234}$. The instanton moduli space is defined as 
\bea\label{eq:M4ADHM1}
\mathfrak{M}_{n,k}=\left.\left\{(\vec{\mathsf{B}},\mathsf{I})\mid \mu_{\mathbb{R}}-\zeta\cdot 1_{k}=s_{A}=0\right\}\right/\U(k).
\eea
This setup is called the \textbf{magnificent four} system. Note that to obtain the $d=3$ result, we can set $\mathsf{Y}=\mathsf{B}_{4}$ and impose $\mathsf{Y}\mathsf{I}=0$. 

Using the identity 
\bea\label{eq:D6Ftermconverse}
\sum_{a<b}\Tr[\mathsf{B}_{a},\mathsf{B}_{b}][\mathsf{B}_{a},\mathsf{B}_{b}]^{\dagger}=\frac{1}{2}\sum_{a<b}\Tr s_{ab}s_{ab}^{\dagger},
\eea
we can replace $s_{A}=0$ with $\mu_{A}=0$ where
\bea\label{eq:M4ADHM2}
\mu_{A}=[\mathsf{B}_{a},\mathsf{B}_{b}].
\eea


\end{itemize}

\paragraph{Instantons of intersecting gauge theories} Let us discuss the generalization of instantons that arise from intersecting gauge theories. There are two setups discussed in this thesis.
\begin{itemize}
    \item \textbf{Spiked instanton} \cite{Nekrasov:2016gud,Nekrasov:2016qym,Nekrasov:2015wsu}: Roughly speaking, it is a setup of six copies of \eqref{eq:ADHMeq}. We introduce six vector spaces $\bfN_{A}=\mathbb{C}^{n_{A}}\,(A\in \six)$ (see section~\ref{sec:gaugeorigamiphysicalsetup} for the notations) and $\bfK=\mathbb{C}^{k}$. The ADHM variables are
    \bea
\mathsf{I}_{A}:\bfN_{A}\rightarrow \bfK,\quad \mathsf{J}_{A}: \bfK\rightarrow \bfN_{A},\quad \mathsf{B}_{a}:\bfK\rightarrow \bfK.
\eea
The moment maps are defined as
\bea\label{eq:spikedmomentdef}
&\mu_{A}=[\mathsf{B}_{a},\mathsf{B}_{b}]+\mathsf{I}_{A}\mathsf{J}_{A},\quad A=ab,\,\,(a<b),\\
&s_{A}=\mu_{A}+\varepsilon_{A\bar{A}}\mu_{\bar{A}}^{\dagger},\quad A\in\six,\\
&\mu_{\mathbb{R}}=\sum_{a\in\four}[\mathsf{B}_{a},\mathsf{B}_{a}^{\dagger}]+\sum_{A\in\six}(\mathsf{I}_{A}\mathsf{I}_{A}^{\dagger}-\mathsf{J}_{A}^{\dagger}\mathsf{J}_{A})
\eea
where $\varepsilon_{A\bar{A}}=\varepsilon_{abcd}\,\,(A=ab,\,\bar{A}=cd)$ for $a<b,\,c<d$ is the total antisymmetric tensor with $\varepsilon_{1234}=1$. We also have the property $s_{A}^{\dagger}=\varepsilon_{A\bar{A}}s_{\bar{A}}$ which implies $s_{A}$ is a real map giving six real conditions. Note also that $\mu_{A},s_{A},\mu_{\mathbb{R}}\in\operatorname{Hom}\,(\bfK,\bfK)$. The moment map equations are then given as 
\bea\label{eq:spikedKK}
\{{s}_{A}=0\}/\U(k),\quad (A\in\six),\qquad 
\{\mu_{\mathbb{R}}=\zeta\cdot1_{k}\}/\U(k).
\eea
We additionally have 
\bea\label{eq:spikedNK}
\{\sigma_{aA}=0\}/\U(k),\quad \sigma_{aA}=B_{a}I_{A}+\varepsilon_{abA}B_{b}^{\dagger}J_{A}^{\dagger}
\eea
where $a,b\in\bar{A},\,(a\neq b)$ and $\sigma_{aA}\in\operatorname{Hom}\,(\bfN_{A},\bfK)$ and 
\beq\label{eq:spikedNN}
\{\Upsilon_{A}=0\}/\U(k),\quad \Upsilon_{A}=\mathsf{J}_{\bar{A}}\mathsf{I}_{A}-\mathsf{I}^{\dagger}_{\bar{A}}\mathsf{J}_{A}^{\dagger}:\bfN_{A}\rightarrow \bfN_{\bar{A}}
\eeq
obeying $\Upsilon_{A}=-\Upsilon_{\bar{A}}$.

The instanton moduli space is given as
\bea\label{eq:spikedinstantonmoduli}
\mathfrak{M}_{\vec{n},k}=\left.\{(\vec{\mathsf{B}},\vec{\mathsf{I}},\vec{\mathsf{J}}\,)\mid s_{A}=0,\,\,\mu_{\mathbb{R}}=\zeta\cdot 1_{k},\,\, \sigma_{aA}=0,\,\,\Upsilon_{A}=0\}\right/\U(k).
\eea
Using the following identity \cite[eq.~(54)]{Nekrasov:2016qym}:
\bea
&\sum_{A\in\six}\text{Tr}\,s_{A}s_{A}^{\dagger}+\sum_{A\in\six,\,a\in\bar{A}}\Tr\sigma_{aA}\sigma_{aA}^{\dagger}+\sum_{A\in\six}\Tr\Upsilon_{A}\Upsilon_{A}^{\dagger} \\
&\quad =2\sum_{A\in\six}(||\mu_{A}||^{2}+||\mathsf{J}_{\bar{A}}\mathsf{I}_{A}||^{2})+\sum_{A\in\six,\,a\in\bar{A}}(||\mathsf{B}_{a}\mathsf{I}_{A}||^{2}+||\mathsf{J}_{A}\mathsf{B}_{a}||^{2})
\eea
the conditions \eqref{eq:spikedKK}, \eqref{eq:spikedNK}, \eqref{eq:spikedNN} are transformed to
\bea\label{eq:spikedJEterm}
s_{A}=0&\longrightarrow \mu_{A}=0\,\,(A\in\six),\\
\sigma_{aA}=0 &\longrightarrow  \mathsf{B}_{a}\mathsf{I}_{A}=0,\quad \mathsf{J}_{A}\mathsf{B}_{a}=0\,\,(A\in\six,\,a\in\bar{A}),\\
\Upsilon_{A}=0&\longrightarrow \mathsf{J}_{\bar{A}}\mathsf{I}_{A}=0\,\,(A\in\six).
\eea

\item \textbf{Tetrahedron instanton} \cite{Pomoni:2021hkn,Fasola:2023ypx}: Roughly speaking, this is a setup with four copies of \eqref{eq:D6-ADHM}. We introduce four vector spaces $\bfN_{\bar{a}}=\mathbb{C}^{n_{\bar{a}}}\,(a\in\four)$ and $\bfK=\mathbb{C}^{k}$. The ADHM variables are $\mathsf{B}_{a}\in\operatorname{Hom}\,(\bfK,\bfK)$ and $\mathsf{I}_{\bar{a}}\in\operatorname{Hom}\,(\bfN_{\bar{a}},\bfK)$. We then introduce the following moment maps:
\bea
\mu_{A}&=[\mathsf{B}_{a},\mathsf{B}_{b}],\quad A=ab\,(a<b)\\
s_{A}&=\mu_{A}+\varepsilon_{A\bar{A}}\mu_{\bar{A}}^{\dagger}=[\mathsf{B}_{a},\mathsf{B}_{b}]+\frac{1}{2}\varepsilon_{abcd}[\mathsf{B}_{c}^{\dagger},\mathsf{B}_{d}^{\dagger}],\quad A=ab\\
\mu_{\mathbb{R}}&=\sum_{a\in\four}[\mathsf{B}_{a},\mathsf{B}_{a}^{\dagger}]+\sum_{a\in\four}\mathsf{I}_{\bar{a}}\mathsf{I}_{\bar{a}}^{\dagger}
\eea
where all of them belong to $\operatorname{Hom}\,(\bfK,\bfK)$. The moment map equations are given 
\bea\label{eq:D6KKcondition}
\{\mu_{\mathbb{R}}=\zeta\cdot 1_{k}\}/\U(k)\,\, (\zeta>0),\quad \{s_{A}=0\}/\U(k). 
\eea
We also have the contributions:
\bea\label{eq:D6NKcondition}
\{\sigma_{\bar{a}}=0\}_{a\in\four}/\U(k),\quad \sigma_{\bar{a}}=\mathsf{B}_{a}\mathsf{I}_{\bar{a}}\in\operatorname{Hom}\,(\bfN_{\bar{a}},\bfK).
\eea
The instanton moduli space of the tetrahedron instanton system is defined as 
\bea\label{eq:tetraADHM}
\mathfrak{M}_{\vec{n},k}=\left.\left\{(\vec{\mathsf{B}},\vec{\mathsf{I}}\,)\mid \mu_{\mathbb{R}}-\zeta\cdot 1_{k}=s_{A}=\sigma_{\bar{a}}=0\right\}\right/\U(k).
\eea
Note that similar to the spiked instanton case, using \eqref{eq:D6Ftermconverse}, the condition $s_{A}=0$ is replaced with a stronger condition $\mu_{A}=0$.

\end{itemize}

\paragraph{Symmetries}
Let us summarize the symmetries of the ADHM variables and moment maps. The derivation of these will be discussed in section~\ref{sec:2susy2d-flavorsymmetry}.
\begin{itemize}

\item Magnificent four:
\bea\label{eq:D8ADHMvariablesymmetry}
(\mathsf{B}_{a},\mathsf{I})_{a\in\four}&\longmapsto(g^{-1}\mathsf{B}_{a}g,g^{-1}\mathsf{I}),\quad g\in\U(k),\\
(\mathsf{B}_{a},\mathsf{I})&\longmapsto (\mathsf{B}_{a},\mathsf{I}h),\quad h\in\U(n),\\
(\mathsf{B}_{a},\mathsf{I})&\longmapsto (q_{a}\mathsf{B}_{a},\mathsf{I}),\quad q_{a}\in\U(1)^{3}
\eea
and
\bea\label{eq:D8ADHMconstrsymmetry}
s_{A}\longmapsto q_{A}g^{-1}s_{A}g,\quad g\in\U(k).
\eea

\item Tetrahedron instanton:
\bea\label{eq:D6gaugesymmetry}
(\mathsf{B}_{a},\mathsf{I}_{\bar{b}})_{a,b\in\four}&\longmapsto(g^{-1}\mathsf{B}_{a}g,g^{-1}\mathsf{I}_{\bar{b}}),\quad g\in\U(k),\\
(\mathsf{B}_{a},\mathsf{I}_{\bar{b}})&\longmapsto (\mathsf{B}_{a},\mathsf{I}_{\bar{b}}h_{\bar{b}}),\quad \underline{h}=(h_{\bar{a}})_{a\in\four}\in\prod_{a\in\four}\U(n_{\bar{a}}),\\
(\mathsf{B}_{a},\mathsf{I}_{\bar{b}})&\longmapsto (q_{a}\mathsf{B}_{a},\mathsf{I}_{\bar{b}}).
\eea
We also have
\bea\label{eq:D6ADHMconstrsymmetry}
s_{A}\longmapsto q_{A}g^{-1}s_{A}g,\quad \sigma_{\bar{a}}\longmapsto q_{a}g^{-1}\sigma_{\bar{a}}h_{\bar{a}}
\eea
where $g\in\U(k),\,h_{\bar{a}}\in\U(n_{\bar{a}})$.

\item Spiked instanton: We have
\bea\label{eq:instantongaugesymmetry}
(\mathsf{B}_{a},\mathsf{I}_{A},\mathsf{J}_{A})_{a\in\four,A\in\six}&\longmapsto (g^{-1}\mathsf{B}_{a}g,g^{-1}\mathsf{I}_{A},\mathsf{J}_{A}g),\quad g\in\U(k),\\
(\mathsf{B}_{a},\mathsf{I}_{A},\mathsf{J}_{A})&\longmapsto (\mathsf{B}_{a},\mathsf{I}_{A}h_{A},h_{A}^{-1}\mathsf{J}_{A}),\quad \underline{h}=(h_{A})_{A\in\six}\in \prod_{A\in\six}\U(n_{A}),\\
(\mathsf{B}_{a},\mathsf{I}_{A},\mathsf{J}_{A})&\longmapsto (q_{a}\mathsf{B}_{a}, \mathsf{I}_{A},q_{A}\mathsf{J}_{A})
\eea
and
\bea\label{eq:D4ADHMconstraintsymmetry}
s_{A}&\longmapsto q_{A}g^{-1}s_{A}g,\quad \sigma_{aA}\longmapsto q_{a}g^{-1}\sigma_{aA}h_{A},\quad \Upsilon_{A}\longmapsto q_{A}^{-1}h_{\bar{A}}^{-1}\Upsilon_{A}h_{A},
\eea
\end{itemize}

\section{Physical setup of gauge origami}\label{sec:gaugeorigamiphysicalsetup}
In this section, we will consider string theoretic realizations of instantons. We will generalize the concept of instantons in two directions. One is instantons of higher dimensional gauge theories and the other will be instantons of \textit{generalized gauge theories}. Both of these generalizations are naturally incorporated in the framework introduced by Nekrasov, which is called \textit{gauge origami}.

Roughly speaking, the gauge origami system is a gauge theory whose space-time $\mathcal{S}$ contains several (generally) intersecting components as
\bea
\mathcal{S}=\bigcup_{i}\mathcal{S}_{i}.
\eea
For each component $\mathcal{S}_{i}$, there is a gauge group $G_{i}$. Matter fields appear at the intersection $\mathcal{S}_{i}\cap \mathcal{S}_{j}$ and they transform under the gauge group $G_{i}\times G_{j}$ and thus they are bifundamental multiplets. In this sense, we can also understand them as generalized quiver gauge theories.

Let the ten-dimensional spacetime be $\mathbb{R}^{1,9}=\mathbb{R}^{1,1}\times \mathbb{R}^{8}$ and we write it as $\mathbb{R}^{1,1}\times \mathbb{C}^{4}$ by choosing a complex structure on $\mathbb{R}^{8}$. Generally, we can consider the spacetime to be $\mathbb{R}^{1,1}\times Z$, where $Z$ is a toric Calabi--Yau four-fold but we only focus on the case when $Z=\mathbb{C}^{4}$ in this thesis. The coordinates are assigned as 
\bea\label{eq:string-background}
\renewcommand{\arraystretch}{1.3}
\begin{tabular}{|c|c|c|c|c|c|c|c|c|c|}
\hline
 \multicolumn{2}{|c|}{$\mathbb{C}_{1}$} & \multicolumn{2}{c|}{$\mathbb{C}_{2}$} & \multicolumn{2}{c|}{$\mathbb{C}_{3}$} & \multicolumn{2}{c|}{$\mathbb{C}_{4}$} & \multicolumn{2}{c|}{$\mathbb{R}^{1,1}$} \\
\cline{1-10}   1 & 2 & 3 & 4& 5 & 6 & 7 & 8 & 9& 0\\
\cline{1-10} $x^{1}$ & $x^{2}$ & $x^{3}$ & $x^{4}$& $x^{5}$ & $x^{6}$& $x^{7}$ &$x^{8}$ & $x^{9}$& $x^{0}$\\
\hline \multicolumn{2}{|c|}{$z_{1}$} & \multicolumn{2}{c|}{$z_{2}$} & \multicolumn{2}{c|}{$z_{3}$} & \multicolumn{2}{c|}{$z_{4}$} & \multicolumn{2}{c|}{$x,t$} \\
\hline
\end{tabular}
\eea
where we denoted the real coordinates as $\{x^{i}\}_{i=0}^{9}$ and the four complex coordinates as $\{z_{a}\}_{a=1}^{4}$ with $z_{a}=x^{2a-1}+ix^{2a}$. 

There are four complex one-planes and three-planes which we denote as $\mathbb{C}_{a},\,\, \mathbb{C}_{\bar{a}}^{3}$ for $a\in\four,\bar{a}\in\four^{\vee}$, where
\bea
 \four=\{1,2,3,4\},\quad \four^{\vee}=\{123,124,134,234\}.
\eea
In later sections, we will frequently use the symmetry between the four variables $1\leftrightarrow 2 \leftrightarrow 3 \leftrightarrow 4$. We call this symmetry the \textit{quadrality}\footnote{Duality, triality, quadrality, pentality,...}.

The complement $\bar{A}$ of $A\in\six$ is defined for example as $A=12,\,\, \bar{A}=34$. Note that we have $\four\simeq \four^{\vee}$ under the map $a\in\four\leftrightarrow\bar{a}\in\four^{\vee} $. We introduce the set $\three$ as the quotient $\six/\sim$ where $A\sim\bar{A}$ is
\bea
    (12)\sim (34),\quad (13)\sim(24),\quad (23)\sim(14)
\eea
and choose the representative as $A=a4,\,\, a\in[3]$. We also use 
\bea
    A=(ab),\quad \text{sup}(A)=b,\quad \text{inf}(A)=a
\eea
for $a<b$ and introduce a lexicographic order as $12<13<14<23<24<34$. 
\begin{figure}[t]
\begin{minipage}[b]{0.45\linewidth}
    \centering
    \includegraphics[width=5cm]{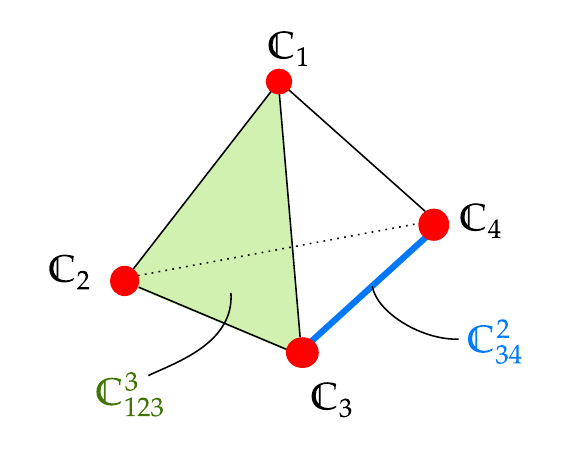}
\end{minipage}\hspace{1cm}{
\begin{minipage}[b]{0.45 \linewidth}
\centering
    \includegraphics[width=4cm]{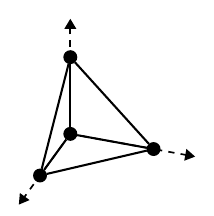}
    \vspace{1cm}
\end{minipage}}
    \caption{Left: The four vertices of the tetrahedron correspond to the $\mathbb{C}_{a}\,\,(a\in\four)$, the six edges connecting two vertices of the tetrahedron correspond to the $\mathbb{C}^{2}_{A}\,\,(A\in\six)$, the four faces surrounded by three vertices and the three edges connecting them correspond to the complex three-planes $\mathbb{C}^{3}_{\bar{a}}\,\,(a\in\four)$, and the whole tetrahedron correspond to the $\mathbb{C}^{4}_{\four}$. Right: The toric diagram of $\mathbb{C}^{4}$. The vertices correspond to the co-dimension one subvariety $\mathbb{C}^{3}$, while the edges and faces correspond to co-dimension two and three subvarieties $\mathbb{C}^{2},\mathbb{C}$, respectively. Both of the description are obtained by taking the dual of each polytopes.}
    \label{fig:complex}
\end{figure}
The six possible complex two-planes are denoted as $\mathbb{C}^{2}_{A}$ for $A\in\six$, where $\six=\{12,13,14,23,24,34\}$ corresponds to the set of the two element subsets of $\four$. All of this data can be summarized in a tetrahedron (see Figure \ref{fig:complex}).

The gauge origami system appears as the low energy limit of the world volume theory on multiple intersecting D-branes in type II string theory. Depending on which type of complex planes the D-branes wrap, different gauge origami systems appear. Generally, one may consider arbitrary D-branes to appear in the setup. However, generally the number of supersymmetries will not be enough and it will be difficult to compute physical observables explicitly. Thus, we want a nice setup that preserves suitable amount of supersymmetries. The way how to count supersymmetries will be explicitly discussed in section~\ref{sec:susyBfield}. 

The gauge origami system we will focus on this thesis is a setup where D-branes wrap the complex planes $\mathbb{C}^{2},\mathbb{C}^{3},\mathbb{C}^{4}$. Moreover, in particular, we are interested in instantons appearing in the theory. The instantons will appear as point like objects in the $\mathbb{C}^{4}$. Let us briefly summarize the setups that will be mainly discussed in this thesis.

\begin{itemize}[topsep=0pt, partopsep=0pt, itemsep=0pt]
\item Spiked instanton \cite{Nekrasov:2015wsu,Nekrasov:2016qym,Nekrasov:2016ydq,Nekrasov:2016gud}: This is a setup where the D5-branes wrap $\mathbb{C}^{2}_{A}\times \mathbb{R}^{1,1}\,(A\in\six)$ while the D1-branes wrap $\mathbb{R}^{1,1}$. The brane configuration is written as
\bea  \label{eq:2Bspikedinstanton}
\renewcommand{\arraystretch}{1.05}
\begin{tabular}{|c|c|c|c|c|c|c|c|c|c|c|}
\hline
& \multicolumn{2}{c|}{$\mathbb{C}_{1}$} & \multicolumn{2}{c|}{$\mathbb{C}_{2}$} & \multicolumn{2}{c|}{$\mathbb{C}_{3}$} & \multicolumn{2}{c|}{$\mathbb{C}_{4}$} & \multicolumn{2}{c|}{$\mathbb{R}^{1,1}$} \\
\cline{2-11} \raisebox{-0mm}{}  & 1 & 2 & 3 & 4& 5 & 6 & 7 & 8 & 9& 0\\[-2pt]
\hline D1& $\bullet$ & $\bullet$  & $\bullet$  & $\bullet$  & $\bullet$  & $\bullet$   & $\bullet$  & $\bullet$  & $-$   & $-$\\
\hline
$\D5_{12} $& $-$ & $-$ & $-$ & $-$ & $\bullet$ & $\bullet$ & $\bullet$ & $\bullet$ & $-$ & $-$ \\
\hline
\raisebox{-0mm}{$\D5_{13}$} & $-$ & $-$& $\bullet$ & $\bullet$  & $-$ & $-$ & $\bullet$ & $\bullet$ & $-$ & $-$ \\
\hline
$\D5_{14} $& $-$ & $-$  & $\bullet$ & $\bullet$ & $\bullet$ & $\bullet$& $-$ & $-$ & $-$ & $-$ \\
\hline $\D4_{23}$ & $\bullet$ & $\bullet$ & $-$ & $-$ & $-$ & $-$ & $\bullet$ & $\bullet$ & $-$ & $-$ \\
\hline
    \raisebox{0mm}{$\D5_{24}$} & $\bullet$ & $\bullet$ & $-$ & $-$ & $\bullet$ & $\bullet$ & $-$ & $-$ & $-$ & $-$ \\
\hline
$\D5_{34} $& $\bullet$ & $\bullet$ & $\bullet$ & $\bullet$& $-$ & $-$ & $-$ & $-$  & $-$ & $-$ \\
\hline
\end{tabular}
\eea
where $-$ means the brane extends in the corresponding direction and $\bullet$ means it is point-like. We may also take the T-duality in the $\mathbb{R}^{1,1}$ part and get intersecting D4-branes and D3-branes in type IIA and IIB string theory, respectively. The D1-branes will be D0 and $\D(-1)$-branes respectively under this transformation.

Depending on which D5-branes exist in the theory, we usually call them with different names (see Figure~\ref{fig:spiked-inst}). For the case when we only have a stack of D5$_{A}$ ($A\in\six$) branes, the D1-branes are the ordinary instantons. For the setup with D5-branes transverse with each other such as D5$_{12}$ and D5$_{34}$, we usually call them the \textit{crossed instantons}. For D5-brane configurations sharing a common $\mathbb{C}$-plane such as D5$_{12}$ and D5$_{23}$, we call them the \textit{folded instantons}. 
\begin{figure}
    \centering
    \includegraphics[width=13cm]{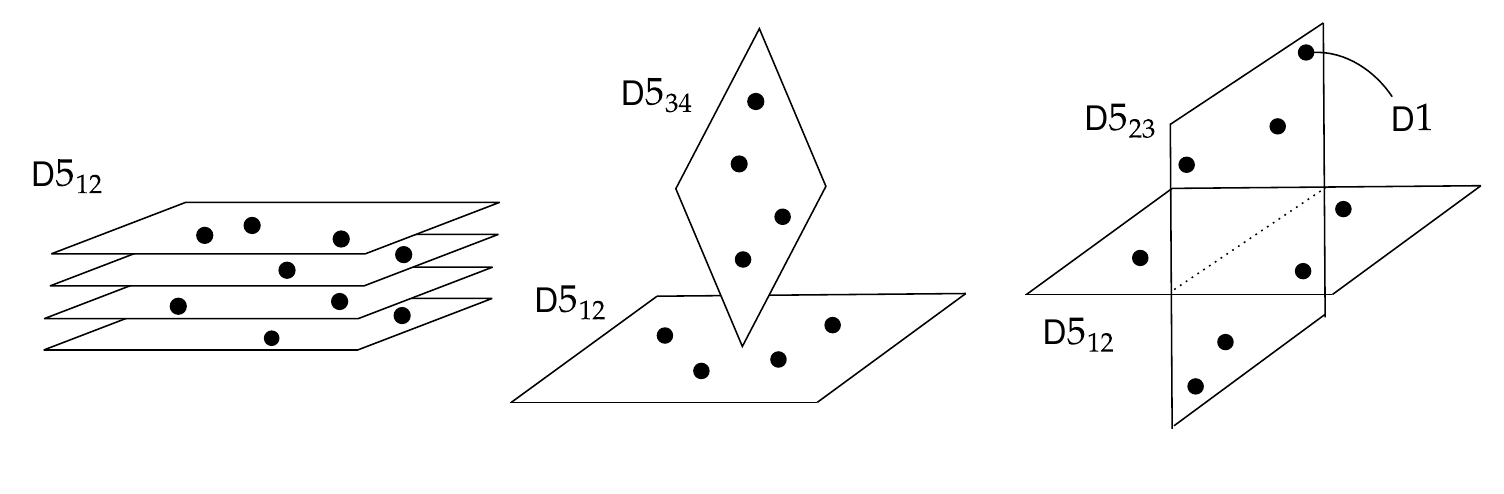}
    \caption{Spiked Instantons: The left gives the ordinary instantons, the middle gives the crossed instantons, and the right gives the folded instantons.}
    \label{fig:spiked-inst}
\end{figure}

\item Tetrahedron instanton \cite{Pomoni:2021hkn,Pomoni:2023nlf,Fasola:2023ypx}: This is a setup where the D7-branes wrap $\mathbb{C}^{3}_{\bar{a}}\times \mathbb{R}^{1,1}$ for $a\in\four$. The brane configuration is summarized as
\bea\label{eq:2Btetrahedroninstanton}
\renewcommand{\arraystretch}{1.05}
\begin{tabular}{|c|c|c|c|c|c|c|c|c|c|c|}
\hline
& \multicolumn{2}{c|}{$\mathbb{C}_{1}$} & \multicolumn{2}{c|}{$\mathbb{C}_{2}$} & \multicolumn{2}{c|}{$\mathbb{C}_{3}$} & \multicolumn{2}{c|}{$\mathbb{C}_{4}$} & \multicolumn{2}{c|}{$\mathbb{R}^{1,1}$} \\
\cline{2-11}  & 1 & 2 & 3 & 4& 5 & 6 & 7 & 8 & 9& 0\\
\hline D1& $\bullet$ & $\bullet$  & $\bullet$  & $\bullet$  & $\bullet$  & $\bullet$   & $\bullet$  & $\bullet$  & $-$   & $-$\\
\hline
$\D7_{123} $& $-$ & $-$ & $-$ & $-$ & $-$ & $-$ & $\bullet$ & $\bullet$ & $-$ & $-$ \\
\hline
$\D7_{124} $& $-$ & $-$& $-$ & $-$  & $\bullet$ & $\bullet$ & $-$ & $-$ & $-$ & $-$ \\
\hline
$\D7_{134} $& $-$ & $-$  & $\bullet$ & $\bullet$ & $-$ & $-$& $-$ & $-$ & $-$ & $-$ \\
\hline $\D7_{234}$ & $\bullet$ & $\bullet$ & $-$ & $-$ & $-$ & $-$ & $-$ & $-$ & $-$ & $-$ \\
\hline
\end{tabular}
\eea
and generally we have four stacks of D7-branes wrapping different complex three-planes in $\mathbb{C}^{4}$. Again, we can take T-duality in the $\mathbb{R}^{1,1}$ direction and obtain D0--D6 and $\D(-1)$--D5 setups.

\item Magnificent four \cite{Nekrasov:2017cih,Nekrasov:2018xsb,Nekrasov:2023nai}: The D9-branes wrap the entire $\mathbb{C}^{4}_{\four}\times \mathbb{R}^{1,1}$ and the D1-branes wrap $\mathbb{R}^{1,1}$:
\bea\label{eq:2Bmagnificentfour}
\renewcommand{\arraystretch}{1.05}
\begin{tabular}{|c|c|c|c|c|c|c|c|c|c|c|}
\hline
& \multicolumn{2}{c|}{$\mathbb{C}_{1}$} & \multicolumn{2}{c|}{$\mathbb{C}_{2}$} & \multicolumn{2}{c|}{$\mathbb{C}_{3}$} & \multicolumn{2}{c|}{$\mathbb{C}_{4}$} & \multicolumn{2}{c|}{$\mathbb{R}^{1,1}$} \\
\cline{2-11}  & 1 & 2 & 3 & 4& 5 & 6 & 7 & 8 & 9& 0\\
\hline D1& $\bullet$ & $\bullet$  & $\bullet$  & $\bullet$  & $\bullet$  & $\bullet$   & $\bullet$  & $\bullet$  & $-$   & $-$\\
\hline
$\D9 $& $-$ & $-$ & $-$ & $-$ & $-$ & $-$ & $-$ & $-$ & $-$ & $-$ \\
\hline
\end{tabular}
\eea
Similarly, T-duality gives the D0--D8 and $\D(-1)$--D7 setups.
\end{itemize}


\section{Intersecting D-branes and supersymmetries}\label{sec:susyBfield}
Naively, the setups \eqref{eq:2Bspikedinstanton}, \eqref{eq:2Btetrahedroninstanton}, \eqref{eq:2Bmagnificentfour} seem to not preserve supersymmetry. However, it seems that after introducing a suitable constant Neveu--Schwarz $B$-field to the system, all of the systems will preserve at least two supersymmetries. Moreover, under such situation, supersymmmetric localization will be applicable (see Chap.~\ref{chap:gaugeorigamipartitionfunction}). In this section, let us derive the number of preserved supersymmetries for each gauge origami setup. We focus on the type IIB string theory description with all the D-branes wrapping the two-dimensional $\mathbb{R}^{1,1}$ space. Dimensional reduction by taking the T-duality can be done similarly and the result will be the same.

Let us summarize the notations of the spinor representations we use in this section. For details, see Appendix~\ref{app:spinor_susy}. The Dirac matrix obeys
\bea
\{\Gamma^{\mu},\Gamma^{\nu}\}=2\eta^{\mu\nu},\quad \mu,\nu=0,1,\ldots,9
\eea
for $\eta^{\mu\nu}=\operatorname{diag}(-1,+1,\ldots,+1)$. We choose the matrix expression
\bea
\Gamma^{1}&=\sigma_{1}\otimes 1\otimes 1\otimes 1\otimes 1,\quad\qquad  \Gamma^{2}=\sigma_{2}\otimes 1\otimes 1\otimes 1\otimes 1 \\
\Gamma^{3}&=\sigma_{3}\otimes \sigma_{1}\otimes 1\otimes 1\otimes 1,\qquad \,\,\,\Gamma^{4}=\sigma_{3}\otimes \sigma_{2}\otimes 1\otimes 1\otimes 1,\\
\Gamma^{5}&=\sigma_{3}\otimes \sigma_{3}\otimes \sigma_{1}\otimes 1\otimes 1,\qquad \Gamma^{6}=\sigma_{3}\otimes \sigma_{3}\otimes \sigma_{2}\otimes 1\otimes 1,\\
\Gamma^{7}&=\sigma_{3}\otimes \sigma_{3}\otimes \sigma_{3}\otimes \sigma_{1}\otimes 1,\quad\,\,\, \Gamma^{8}=\sigma_{3}\otimes \sigma_{3}\otimes \sigma_{3}\otimes \sigma_{2}\otimes 1\\
\Gamma^{9}&=\sigma_{3}\otimes \sigma_{3}\otimes \sigma_{3}\otimes \sigma_{3}\otimes \sigma_{1},\quad\, \Gamma^{0}=\sigma_{3}\otimes \sigma_{3}\otimes \sigma_{3}\otimes \sigma_{3}\otimes (i\sigma_{2}).
\eea
Under this expression, the chirality matrix $\oGamma$ is defined as
\bea
\oGamma=\Gamma^{0}\Gamma^{1}\cdots \Gamma^{9}=\sigma_{3}\otimes \sigma_{3}\otimes \sigma_{3}\otimes \sigma_{3}\otimes \sigma_{3}.
\eea
Since the D-branes appearing in the gauge origami system wrap $\mathbb{R}^{1,1}$ and the complex planes $\mathbb{C}_{a}\,(a\in\four)$, it is convenient to use the bases diagonalized by 
\bea\label{eq:spin-operator-def}
(S_{0},S_{1},S_{2},S_{3},S_{4})\coloneqq(\Gamma^{0}\Gamma^{9},-i\Gamma^{1}\Gamma^{2},-i\Gamma^{3}\Gamma^{4},-i\Gamma^{5}\Gamma^{6},-\Gamma^{7}\Gamma^{8}).
\eea 
The eigenstates are denoted as $\ket{\vec{s}\,}=\ket{s_{0},s_{1},s_{2},s_{3},s_{4}}$ with 
\bea
S_{a}\ket{\vec{s}\,}=s_{a}\ket{\vec{s}\,},\quad s_{a}\in\{\pm 1\}.
\eea
For later use, we also shortly write the product of Gamma matrices as $\Gamma^{A}=\prod_{a\in A}\Gamma^{a}$ for a subset $A\subseteq\{0,1,\ldots,9\}$. We will also sometimes omit the factor $1$ and just write the signs when we discuss on the eigenvalues $\{s_{a}\}_{a=0,1,2,3,4}$.

In type IIB string theory, the supercharges originating from left and right-moving world sheet degrees of freedom transform as spinors with the same chirality: $\sum_{\alpha}\xi_{\sL,\alpha}Q^{\alpha}+\sum_{\beta}\xi_{\sR,\beta}\widetilde{Q}^{\beta}$. We then have
\bea
\oGamma\xi_{\sL}=\xi_{\sL},\quad \oGamma\xi_{\sR}=\xi_{\sR}.
\eea
The conditions for supersymmetries preserved in D-branes are summarized in Appendix~\ref{app:spinor_susy}. Let us study the situation when we only have a D-brane wrapping one of the complex planes in $\mathbb{C}^{4}$ with a D1-brane bound to them.

\paragraph{D1--D5 system}Consider the situation with D5$_{12}$ and D1-branes. The supersymmetric condition is
\bea
\D1(09):&\quad \Gamma^{09}\xi_{\sR}=\xi_{\sL},\\
\D5_{12}(091234):&\quad \Gamma^{09}\Gamma^{1234}\xi_{\sR}=\xi_{\sL}
\eea
where each condition gives $16$ supercharges. Imposing both conditions give $\Gamma^{1234}\xi_{\sL}=\xi_{\sL}$. Since $(\Gamma^{1234})^{2}=1$, the spinor belongs to the eigenspace with eigenvalue $+1$ and the possible $\xi_{\sL}$ is 8-dimensional. The other $\xi_{\sR}$ is automatically determined by the above condition. Therefore, this setup preserves 8 supercharges and it is supersymmetric.

\paragraph{D1--D7 system} Consider the situation with $\D7_{123}$ and D1-branes. The supersymmetric condition is
\bea
\D1(09):&\quad \Gamma^{09}\xi_{\sR}=\xi_{\sL},\\
\D7_{123}(09123456):&\quad \Gamma^{09}\Gamma^{123456}\xi_{\sR}=\xi_{\sL}
\eea
giving $\Gamma^{123456}\xi_{\sL}=\xi_{\sL}$. However, since $(\Gamma^{123456})^{2}=-1$, we have no such kind of spinors and thus this setup is not supersymmetric.

\paragraph{D1--D9 system}The $\D9_{\four}$-brane and D1-brane gives
\bea
\D1(09):&\quad \Gamma^{09}\xi_{\sR}=\xi_{\sL},\\
\D9_{\four}(0912345678):&\quad \Gamma^{09}\Gamma^{12345678}\xi_{\sR}=\xi_{\sL}
\eea
giving $\Gamma^{12345678}\xi_{\sL}=\xi_{\sL}$. Similarly, we have $(\Gamma^{12345678})^{2}=+1$ and we have a spinor obeying the above conditions and thus it is supersymmetric.

\paragraph{Supersymmetry under NS $B$-fields}The D1--D7 system above is not supersymmetric. Witten showed that after turning on a suitable constant NS $B$-field, the D1--D7 system can gain supersymmetries~\cite{Witten:2000mf}. We turn on a constant $B$-field along $\mathbb{C}^{4}$ taking the form
\bea\label{eq:B-field}
B=\sum_{a\in\four}b_{a}dx^{2a-1}\wedge dx^{2a},\quad b_{a}\in\mathbb{R}
\eea
and define
\bea
e^{2\pi i v_{a}}=\frac{1+ib_{a}}{1-ib_{a}},\quad -\frac{1}{2}<v_{a}<\frac{1}{2}.
\eea
Turning on this $B$-field effectively rotates the D-brane by the angles $2\pi v_{a}$ and the conditions on the spinors will be modified. For example, a D-brane extending in the $(2a-1,2a)$-plane and rotated with the angle $\omega$ gives
\bea\label{eq:Brotsusy}
\Gamma^{A}\Gamma^{2a-1,2a}\exp\left(\frac{\omega}{2}\Gamma^{2a-1,2a}\right)\xi_{\sR}=\xi_{\sL},
\eea
where $A$ is the directions except $(2a-1,2a)$ where the D-brane is extending. Using
\bea
\exp\left(\frac{i\omega}{2}\left(-i\Gamma^{2a-1,2a}\right)\right)=\cos(\omega/2)+\sin(\omega/2)\Gamma^{2a-1,2a}=\cos(\omega/2)+i\sin(\omega/2)S_{a},
\eea
we have $\exp\left(\frac{\pi}{2}\Gamma^{2a-1,2a}\right)=\Gamma^{2a-1,2a}$. The condition \eqref{eq:Brotsusy} is then translated to
\bea
\Gamma^{A}\exp\left(\frac{\omega+\pi}{2}\Gamma^{2a-1,2a}\right)\xi_{\sR}=\xi_{\sL}.
\eea
Since the rotation angle is $2\pi v_{a}$ now, after defining
\bea
R_{a}(\theta)=\exp\left(\theta_{a} \Gamma^{2a-1}\Gamma^{2a}\right),\quad \theta_{a}=\pi\left(v_{a}+\frac{1}{2}\right)\in(0,\pi)
\eea
for $a\in\four$, the condition for a D-brane extending in a complex plane $A\in\subseteq\{1,2,3,4\}$ comes from
\bea
R_{A}(\theta)=\prod_{a\in A}R_{a}(\theta)=\exp\left(\sum_{a\in A}\theta_{a}\Gamma^{2a-1}\Gamma^{2a}\right).
\eea

Using this, let us revisit the D1--D7 system. The supersymmetric condition is now modified to 
\bea
R_{123}(\theta)\xi_{\sL}=\xi_{\sL}.
\eea
Using the bases $\ket{\vec{s}\,}$, this condition is 
\bea\label{eq:D1-D7Bfield_onestack}
\exp\left(is_{1}\theta_{1}+is_{2}\theta_{2}+is_{3}\theta_{3}\right)=+1
\eea
giving $\sum_{a\in\bar{4}}s_{a}\theta_{a}\in 2\pi\mathbb{Z}$. Note also that since the chirality of $\xi_{\sL}$ is positive, the other condition is 
\bea\label{eq:chirality_cond2B}
s_{0}s_{1}s_{2}s_{3}s_{4}=+1.
\eea
 When $\{\theta_{a}\}_{a\in\four}$ are generic we have four supersymmetries. Given a $(s_{1},s_{2},s_{3})$ obeying~\eqref{eq:D1-D7Bfield_onestack}, $(-s_{1},-s_{2},-s_{3})$ is also a solution. The other condition \eqref{eq:chirality_cond2B} determines the product $s_{0}s_{4}$ uniquely and we have two choices of $(s_{0},s_{4})$. Therefore, generically, we will have four supersymmetries. When special conditions are imposed to the $B$-field parameters $\{\theta_{a}\}_{a\in\four}$, enhancement of supersymmetries may occur but we do not discuss such cases.

Explicitly, for example, we have the following solutions.
\begin{itemize}[topsep=0pt, partopsep=0pt, itemsep=0pt]
\item When $\theta_{1}+\theta_{2}+\theta_{3}=2\pi$, we have for example $(s_{1},s_{2},s_{3})=(+,+,+),(-,-,-)$. Using \eqref{eq:chirality_cond2B}, the other signs are determined as
\bea
(s_{0},s_{1},s_{2},s_{3},s_{4})=(+,+,+,+,+),\,\,(-,+,+,+,-),\,\,(+,-,-,-,-),\,\,(-,-,-,-,+)
\eea
which gives 4 supersymmetries.
\item When $\theta_{1}-\theta_{2}-\theta_{3}=0$, the solutions are
\bea
(s_{0},s_{1},s_{2},s_{3},s_{4})=(+,+,-,-,+),\,\,(-,+,-,-,-),\,\,(+,-,+,+,-),\,\,(-,-,+,+,+)
\eea
which also gives 4 supersymmetries. Other cases such as $-\theta_{1}+\theta_{2}-\theta_{3}=0$ are obtained by using the triality.
\end{itemize}
In terms of $\{v_{a}\}_{a\in\four}$, the condition is written as $\pm v_{1}\pm v_{2}\pm v_{3}=\pm1/2$. Note that when the $B$-field is absent, i.e. $v_{a}=0$ and $\theta_{a}=\pi/2$, none of the above conditions is satisfied and it is compatible with the previous discussion that the system is not supersymmetric.

\subsection{Magnificent four}\label{eq:M4-susycount}
Let us consider the magnificent four system with constant $B$-field. The supersymmetric condition is modified to
\bea
\D1(09):&\quad \Gamma^{09}\xi_{\sR}=\xi_{\sL},\\
    \D9_{\four}(0912345678):&\quad \Gamma^{09}R_{1234}(\theta)\xi_{\sR}=\xi_{\sL}
\eea
which gives $R_{1234}(\theta)\xi_{\sL}=\xi_{\sL}$. Using the basis $\ket{\vec{s}\,}$, the condition is translated to
\bea\label{eq:M4susycond}
\exp\left(i\sum_{a\in\four}s_{a}\theta_{a}\right)=+1,\quad \sum_{a\in\four}s_{a}\theta_{a}\in 2\pi\mathbb{Z}.
\eea
Since $\theta_{a}<\pi$, we have the following cases:
\bea
\sum_{a\in\four}s_{a}\theta_{a}=2\pi,0,-2\pi.
\eea
For generic $\theta_{a}$, if $(s_{1},s_{2},s_{3},s_{4})$ obeys the condition \eqref{eq:M4susycond}, then $(-s_{1},-s_{2},-s_{3},-s_{4})$ also obeys the condition \eqref{eq:M4susycond}. For each cases, the remaining $s_{0}$ is automatically determined by \eqref{eq:D1-D7Bfield_onestack} and thus we have two supersymmetries. Note that when there is no $B$-field with $\theta_{a}=\pi/2\,\,(a\in\four)$, there are solutions obeying the conditions and thus it is supersymmetric as mentioned before.

\subsection{Tetrahedron instanton}\label{sec:tetra-susycount}
As mentioned before, the D1--D7$_{123}$ system is supersymmetric with $B$-field. Let us place other D7-branes in the setup as \eqref{eq:2Btetrahedroninstanton} and analyze the condition for unbroken supersymmetry following \cite{Pomoni:2021hkn}. The conditions are
\bea
\D1:&\quad \Gamma^{09}\xi_{\sR}=\xi_{\sL},\\
\D7_{\bar{a}}:&\quad \Gamma^{09}R_{\bar{a}}(\theta)\xi_{\sR}=\xi_{\sL},\quad \forall a\in\four.
\eea
The conditions from D1--D7 gives
\bea
R_{\bar{a}}(\theta)\xi_{\sL}=\xi_{\sL},\quad a\in\four
\eea
and the conditions from $\D7_{\bar{a}}$--$\D7_{\bar{b}}$ for $a\neq b$ gives
\bea
R_{\bar{a}}(\theta)\xi_{\sL}=R_{\bar{a}}(\theta)\xi_{\sL}.
\eea
Using the bases $\ket{\vec{s}\,}$, we have
\bea
\exp\left(i\sum_{\alpha\in\bar{a}}s_{\alpha}\theta_{\alpha}\right)&=+1\,\, (\forall a\in\four),\quad \exp\left(i\sum_{\alpha\in\bar{a}}s_{\alpha}\theta_{\alpha}\right)=\exp\left(i\sum_{\beta\in\bar{b}}s_{\beta}\theta_{\beta}\right)\,\,(\forall a\neq b \in\four).
\eea
Since any two stacks of D7-branes share a complex two-plane, the second condition is translated to 
\bea
\exp\left(is_{a}\theta_{a}-is_{b}\theta_{b}\right)=+1,\quad \forall a\neq b.
\eea
This condition gives $\theta_{a}=\theta_{b}$ because $\theta_{a}\in(0,\pi)$ and $s_{a}\theta_{a}-s_{b}\theta_{b}=0$. Moreover, we need $s_{a},s_{b}$ to be the same signs. Thus, the possible solution is only
\bea
\theta_{1}=\theta_{2}=\theta_{3}=\theta_{4}=\frac{2\pi}{3}
\eea
and 
\bea
s_{1}=s_{2}=s_{3}=s_{4}=\pm 1,\quad s_{0}=1
\eea
where $s_{0}$ is automatically determined by \eqref{eq:chirality_cond2B}. Therefore, we have two supersymmetries in this setup. 

\subsection{Spiked instanton}\label{sec:spiked-susycount}
Let us discuss the spiked instanton setup. Since the D1-D5 setup preserves supersymmetry in the absence of $B$-field, let us first consider if we can have a supersymmetric system without the $B$-field. The analysis of this setup was done in \cite{Nekrasov:2016gud} and this part is based on it.

Including the six stacks of D5-branes give the following conditions
\bea
\D1:&\quad  \Gamma^{09}\xi_{\sR}=\xi_{\sL},\\
\D5_{A}\,\,(A\in\six):&\quad \Gamma^{09}\prod_{a\in A}\Gamma^{2a-1}\Gamma^{2a}\xi_{\sR}=\xi_{\sL}. 
\eea
Explicitly, we have
\bea
\Gamma^{1234}\xi_{\sL}=\xi_{\sL},\quad \Gamma^{1256}\xi_{\sL}=\xi_{\sL},\quad \Gamma^{1278}\xi_{\sL}=\xi_{\sL},\quad \Gamma^{3456}\xi_{\sL}=\xi_{\sL},\quad \Gamma^{3478}\xi_{\sL}=\xi_{\sL},\quad \Gamma^{5678}\xi_{\sL}=\xi_{\sL}.
\eea
All of these conditions can not be satisfied because for example the first and second conditions give $\Gamma^{3456}\xi_{\sL}=-\xi_{\sL}$, which breaks the fourth condition.

\paragraph{Four D5-branes and two anti D5-branes}
One way to resolve this situation is to include anti D5-branes to the setup us
\bea\label{eq:2Bspikedinstanton-mod}
\renewcommand{\arraystretch}{1.05}
\begin{tabular}{|c|c|c|c|c|c|c|c|c|c|c|}
\hline
& \multicolumn{2}{c|}{$\mathbb{C}_{1}$} & \multicolumn{2}{c|}{$\mathbb{C}_{2}$} & \multicolumn{2}{c|}{$\mathbb{C}_{3}$} & \multicolumn{2}{c|}{$\mathbb{C}_{4}$} & \multicolumn{2}{c|}{$\mathbb{R}^{1,1}$} \\
\cline{2-11}  & 1 & 2 & 3 & 4& 5 & 6 & 7 & 8 & 9& 0\\
\hline D1& $\bullet$ & $\bullet$  & $\bullet$  & $\bullet$  & $\bullet$  & $\bullet$   & $\bullet$  & $\bullet$  & $-$   & $-$\\
\hline
$\D5_{12} $& $-$ & $-$ & $-$ & $-$ & $\bullet$ & $\bullet$ & $\bullet$ & $\bullet$ & $-$ & $-$ \\
\hline
\raisebox{-0.6mm}{$\overline{\D5}_{13}$} & $-$ & $-$& $\bullet$ & $\bullet$  & $-$ & $-$ & $\bullet$ & $\bullet$ & $-$ & $-$ \\
\hline
$\D5_{14} $& $-$ & $-$  & $\bullet$ & $\bullet$ & $\bullet$ & $\bullet$& $-$ & $-$ & $-$ & $-$ \\
\hline $\D5_{23}$ & $\bullet$ & $\bullet$ & $-$ & $-$ & $-$ & $-$ & $\bullet$ & $\bullet$ & $-$ & $-$ \\
\hline
    \raisebox{-0.6mm}{$\overline{\D5}_{24}$} & $\bullet$ & $\bullet$ & $-$ & $-$ & $\bullet$ & $\bullet$ & $-$ & $-$ & $-$ & $-$ \\
\hline
$\D5_{34} $& $\bullet$ & $\bullet$ & $\bullet$ & $\bullet$& $-$ & $-$ & $-$ & $-$  & $-$ & $-$ \\
\hline
\end{tabular}
\eea
This system was introduced and studied in \cite{Nekrasov:2016gud}. Since the supersymmetry condition for the anti D-branes reverse the signs as $\Gamma^{01\cdots p}\xi_{\sR}=-\xi_{\sL}$, the conditions are written as
\bea
\Gamma^{1234}\xi_{\sL}=\xi_{\sL},\quad \Gamma^{1256}\xi_{\sL}=-\xi_{\sL},\quad \Gamma^{1278}\xi_{\sL}=\xi_{\sL},\quad \Gamma^{3456}\xi_{\sL}=\xi_{\sL},\quad \Gamma^{3478}\xi_{\sL}=-\xi_{\sL},\quad \Gamma^{5678}\xi_{\sL}=\xi_{\sL}.
\eea
This can be also written as
\bea
\prod_{a\in A}\Gamma^{2a-1,2a}\xi_{\sL}=\varepsilon_{A\bar{A}}\xi_{\sL}
\eea
where $\varepsilon_{1234}=+1$ is the anti-symmetric tensor and for example $\varepsilon_{(12)(34)}=+1$. The first and second conditions now give $\Gamma^{3456}\xi_{\sL}=\xi_{\sL}$ because of the extra sign coming from the anti-brane contribution. Similarly,the fifth and sixth condition comes from the first and third, the second and third conditions, respectively and thus we only have three independent conditions.

Taking the first three conditions to be the three independent conditions, using the operators \eqref{eq:spin-operator-def}, the conditions are
\bea
S_{1}S_{2}\,\xi_{\sL}=-\xi_{\sL},\quad S_{1}S_{3}\,\xi_{\sL}=+\xi_{\sL},\quad S_{1}S_{4}\,\xi_{\sL}=-\xi_{\sL},
\eea
where note that the extra $i$ factor gives a flip of the sign. Using the bases $\xi_{\sL}=\ket{\vec{s}\,}$, the conditions are
\bea
s_{1}s_{2}=-1,\quad s_{1}s_{3}=+1,\quad s_{1}s_{4}=-1,
\eea
which gives
\bea
(s_{1},s_{2},s_{3},s_{4})=(+,-,+,-),\,\,(-,+,-,+).
\eea
Moreover, using \eqref{eq:chirality_cond2B}, $s_{0}$ is uniquely determined and we have
\bea
(s_{0},s_{1},s_{2},s_{3},s_{4})=(+,+,-,+,-),\,(+,-,+,-,+).
\eea
In conclusion, this setup preserves two supersymmetries. 

\paragraph{4 SUSY system}
As an exercise, let us also discuss on branes systems with 4 supersymmetries. Depending on the brane systems, the preserved 4 supersymmetries will change but in any cases, we have a common two supersymmetry taking the form $(s_{0},s_{1},s_{2},s_{3},s_{4})=(+,+,-,+,-),(+,-,+,-,+)$.
\begin{itemize}[topsep=0pt, partopsep=0pt, itemsep=0pt]
\item  D5-branes wrapping complex two-planes inside $\mathbb{C}^{3}_{123}$: Namely, we have $\D5_{12},\overline{\D5}_{13},\D5_{23}$-branes and the conditions are
\bea
\Gamma^{1234}\xi_{\sL}=\xi_{\sL},\quad \Gamma^{1256}\xi_{\sL}=-\xi{\sL},\quad \Gamma^{3456}\xi_{\sL}=\xi_{\sL},
\eea
and using \eqref{eq:spin-operator-def} gives
\bea
S_{1}S_{2}\,\xi_{\sL}=-\xi_{\sL},\quad S_{1}S_{3}\,\xi_{\sL}=+\xi_{\sL},\quad S_{2}S_{3},\xi_{\sL}=-\xi_{\sL}.
\eea
Since the third condition is automatically satisfied when the first and second conditions are satisfied, we have $s_{1}s_{2}=-1,\,\,s_{1}s_{3}=+1$ giving $(s_{1},s_{2},s_{3})=(+,-,+),\,(-,+,-)$. The condition \eqref{eq:chirality_cond2B} only determines the product $s_{0}s_{4}$ and we have two possibilities for each case:
\bea
(s_{0},s_{1},s_{2},s_{3},s_{4})=(+,+,-,+,-),\,\,(-,+,-,+,+),\,\,(+,-,+,-,+),\,\,(-,-,+,-,-).
\eea
Therefore, this system preserves four supersymmetries in total. From the 2d theory on $\mathbb{R}^{1,1}$ viewpoint, this gives a 2d $\mathcal{N}=(2,2)$ theory, which can be obtained from the dimensional reduction of a 4d $\mathcal{N}=1$ theory.

\item Transverse D5-branes: A different situation when we have four supersymmetries is when we have two stacks of D-branes wrapping the complex two-planes $\mathbb{C}^{2}_{12},\,\,\mathbb{C}^{2}_{34}$:
\bea
\Gamma^{1234}\xi_{\sL}=\xi_{\sL},\quad \Gamma^{5678}\xi_{\sL}=\xi_{\sL},
\eea
which is equivalent to
\bea
S_{1}S_{2}\,\xi_{\sL}=-\xi_{\sL},\quad S_{3}S_{4}\xi_{\sL}=-\xi_{\sL}.
\eea
These give $s_{1}s_{2}=-1$ and $s_{3}s_{4}=-1$ which automatically give $s_{0}=+1$. The possible choices are then
\bea
(s_{0},s_{1},s_{2},s_{3},s_{4})=(+,+,-,+,-),\,(+,+,-,-,+),\,(+,-,+,+,-),\,(+,-,+,-,+),
\eea
which gives four supersymmetries. From the viewpoint of the 2d theory, all of them have positive chiralities and then it is a 2d $(0,4)$ supersymmetric field theory.

\item Transverse anti D5-branes: The other situation when we have anti D5-branes wrapping transverse complex two-planes $\mathbb{C}^{2}_{13}$ and $\mathbb{C}^{2}_{24}$ give
\bea
\Gamma^{1256}\xi_{\sL}=-\xi_{\sL},\quad \Gamma^{3478}\xi_{\sL}=-\xi_{\sL}
\eea
which is
\bea
S_{1}S_{3}\,\xi_{\sL}=\xi_{\sL},\quad S_{2}S_{4}\,\xi_{\sL}=\xi_{\sL}
\eea
and $s_{1}s_{3}=+1,\,s_{2}s_{4}=+1$. Therefore, we have four supersymmetries
\bea
(s_{0},s_{1},s_{2},s_{3},s_{4})=(+,+,+,+,+),\,(+,-,+,-,+),\,(+,+,-,+,-),\,(+,-,-,-,-).
\eea
Again, from the 2d theory viewpoint, this theory also gives a 2d $\mathcal{N}=(0,4)$ theory.
\end{itemize}

\paragraph{Inclusion of $B$-field}
Following the magnificent four and tetrahedron instanton system, let us include a constant $B$-field and see what will happen. We first consider the case when we have only one $\D5_{12}$-brane. The condition will be 
\bea
\D1(09):&\quad \Gamma^{09}\xi_{\sR}=\xi_{\sL},\\
\D5_{12}(091234):&\quad \Gamma^{09}R_{12}(\theta)\xi_{\sR}=\xi_{\sL}
\eea
which gives $R_{12}(\theta)\xi_{\sL}=\xi_{\sL}$:
\bea
\exp\left(is_{1}\theta_{1}+is_{2}\theta_{2}\right)=+1.
\eea
Since $\theta_{1,2}\in(0,\pi)$, the possible solution is only $s_{1}\theta_{1}+s_{2}\theta_{2}=0$, which gives $(s_{1},s_{2})=(+,-),(-,+)$ and $\theta_{1}=\theta_{2}$. In terms of $\{v_{a}\}$, this is $v_{1}-v_{2}=0$.

Note that if we use the anti D5-brane, the condition is $R_{12}(\theta)\xi_{\sL}=-\xi_{\sL}$:
\bea
\exp\left(is_{1}\theta_{1}+is_{2}\theta_{2}\right)=-1
\eea
whose solution is $s_{1}\theta_{1}+s_{2}\theta_{2}=\pm \pi$. The solution of this is $\theta_{1}+\theta_{2}=\pi$ with $(s_{1},s_{2})=(+,+),(-,-)$. Moreover, in terms of $\{v_{a}\}$, this condition is $v_{1}+v_{2}=0$, which translates to $b_{1}+b_{2}=0$. This means that the $B$-field is anti-self-dual.

Let us consider next the situation when we have the six D5-branes with no anti-branes. The supersymmetric condition is modified as
\bea
\D1:&\quad \Gamma^{09}\xi_{\sR}=\xi_{\sL},\\
\D5_{A}:&\quad \Gamma^{09}R_{A}(\theta)\xi_{\sR}=\xi_{\sL},\quad \forall A\in\six.
\eea
The D1--D5$_{A}$ gives $R_{A}(\theta)\xi_{\sL}=\xi_{\sL}$ and the $\D5_{A}$--$\D5_{B}$ ($A\neq B$) gives $R_{A}(\theta)\xi_{\sL}=R_{B}(\theta)\xi_{\sL}$. Using the bases $\ket{\vec{s}\,}$, the solution is
\bea
    \exp\left(is_{a}\theta_{a}+is_{b}\theta_{b}\right)=+1,\quad \forall (ab)\in\six.
\eea
This has no solution unless $\theta_{a}=0,\pi$ for all $a\in\four$ and thus turning on a finite $B$-field does not make the brane configuration supersymmetric.

Let us then include the anti D5-branes to the setup. The condition is then rewritten as
\bea
R_{A}(\theta)\xi_{\sL}=\varepsilon_{A\bar{A}}\xi_{\sL}.
\eea
This gives
\bea
\exp\left(is_{a}\theta_{a}+is_{b}\theta_{b}\right)=\varepsilon_{ab\overline{ab}},\quad (ab)\in\six.
\eea
Explicitly, we have
\bea
\exp(is_{1}\theta_{1}+is_{2}\theta_{2})=+1,\quad \exp(is_{1}\theta_{1}+is_{3}\theta_{4})=-1,\\
\exp(is_{1}\theta_{1}+is_{4}\theta_{4})=+1,\quad \exp(is_{2}\theta_{2}+is_{3}\theta_{3})=+1,\\
\exp(is_{2}\theta_{2}+is_{4}\theta_{4})=-1,\quad \exp(is_{3}\theta_{3}+is_{4}\theta_{4})=+1,
\eea
which gives
\bea
s_{1}\theta_{1}+s_{2}\theta_{2}=0,\quad s_{1}\theta_{1}+s_{3}\theta_{3}=\pm\pi,\\
s_{1}\theta_{1}+s_{4}\theta_{4}=0,\quad s_{2}\theta_{2}+s_{3}\theta_{3}=0,\\
s_{2}\theta_{2}+s_{4}\theta_{4}=\pm \pi,\quad s_{3}\theta_{3}+s_{4}\theta_{4}=0.
\eea
The only solution is $\theta_{a}=\pi/2$ which reduce to the situation with no $B$-field. At the end we have two supersymmetries.

\begin{remark}
Actually there is a different setup that preserves supersymmetry which somehow we lack discussions in the context of gauge origami. It is a setup when we have six stacks of anti D5-branes:
\bea
\Gamma^{09}\prod_{a\in A}\Gamma^{2a-1}\Gamma^{2a}\xi_{\sR}=-\xi_{\sL}.
\eea
Including the D1-branes gives
\bea
\Gamma^{1234}\xi_{\sL}=-\xi_{\sL},\quad \Gamma^{1256}\xi_{\sL}=-\xi_{\sL},\quad \Gamma^{1278}\xi_{\sL}=-\xi_{\sL},\quad \Gamma^{3456}\xi_{\sL}=-\xi_{\sL},\quad \Gamma^{3478}\xi_{\sL}=-\xi_{\sL},\quad \Gamma^{5678}\xi_{\sL}=-\xi_{\sL}.
\eea
These conditions are compatible because for example the first and second conditions give
\bea
\Gamma^{1234}\Gamma^{1256}\xi_{\sL}&=-\Gamma^{1234}\xi_{\sL}=\xi_{\sL},\\
-\Gamma^{3456}\xi_{\sL}&=\xi_{\sL}
\eea
which is compatible with the fourth condition. 

Including the $B$-field gives 
\bea
R_{A}(\theta)\xi_{\sL}=-\xi_{\sL},\quad A\in\six
\eea
which is translated to
\bea
\exp\left(is_{a}\theta_{a}+is_{b}\theta_{b}\right)=-1,\quad (ab)\in\six.
\eea
Using the condition $\theta_{a}\in(0,\pi)$, we have $s_{a}\theta_{a}+s_{b}\theta_{b}=\pm \pi$. The solution is then 
\bea
\theta_{1}=\theta_{2}=\theta_{3}=\theta_{4}=\pi/2,\quad (s_{0},s_{1},s_{2},s_{3},s_{4})=(+,+,+,+,+),(+,-,-,-,-)
\eea
and we have two supersymmetries. Note that $v_{1}=v_{2}=v_{3}=v_{4}=0$ and this is the situation with no $B$-field.

\end{remark}

\section{Open string spectrum}\label{sec:opstringspectrum}
In this section, we review the open string spectrum with a nontrivial $B$-field \eqref{eq:B-field} and derive the matter fields arising in the gauge origami system. For reference, see \cite{Abouelsaood:1986gd} or \cite[App.~A]{Nekrasov:2016gud}.

The worldsheet action of the open string on the background $\mathbb{R}^{1,1}\times\mathbb{C}^{4}$ in \eqref{eq:string-background} is given as
\bea
S&=\frac{1}{2\pi}\int d\tau\int_{0}^{\pi}d\sigma \,G_{\mu\nu}\left(\partial_{+}X^{\mu}\partial_{-}X^{\nu}+i\psi_{-}^{\mu}\partial_{+}\psi^{\nu}_{-}+i\psi^{\mu}_{+}\partial_{-}\psi^{\nu}_{+}\right)\\
&+\frac{1}{2\pi}\int d\tau B_{\mu\nu}\left[2(\partial_{\tau}X^{\mu})X^{\nu}+i\psi^{\mu}_{-}+i\psi^{\mu}_{+}\psi^{\nu}_{+}\right]_{\sigma=0}^{\pi}
\eea
where $G_{\mu\nu}=\eta_{\mu\nu}$, $B_{\mu\nu}$ is \eqref{eq:B-field}, and $\sigma^{\pm}=\tau\pm\sigma$, $\partial_{\pm}=(\partial_{\tau}\pm\partial_{\sigma})/2$ are the light-cone coordinates. See Appendix~\ref{app-chap:openstring} for the analysis when there is no $B$-field. The equations of motion are
\bea
\partial_{+}\partial_{-}X^{\mu}=0,\quad \partial_{+}\psi^{\mu}_{-}=\partial_{-}\psi^{\mu}_{+}=0
\eea
with the boundary conditions
\bea
&\left[(G_{\mu\nu}\partial_{\sigma}X^{\mu}+B_{\mu\nu}\partial_{\tau}X^{\mu})\delta X^{\nu}\right]_{\sigma=0}^{\sigma=\pi}=0,\\
&\left[\delta\psi^{\mu}_{-}(G_{\mu\nu}-B_{\mu\nu})\psi_{-}^{\nu}-\delta \psi^{\mu}_{+}(G_{\mu\nu}+B_{\mu\nu})\delta^{\nu}_{+}\right]_{\sigma=0}^{\sigma=\pi}=0.
\eea
Similar to the case when there is no $B$-field, we have two types of boundary conditions. One is the Dirichlet (D) boundary condition
\bea
\delta X^{\mu}|_{\sigma=0,\pi}=0\quad \Leftrightarrow \partial_{\tau}X^{\mu}|_{\sigma=0,\pi}=0
\eea
and the other is the twisted (T) boundary condition
\bea
(G_{\mu\nu}\partial_{\sigma}X^{\mu}+B_{\mu\nu}\partial_{\tau}X^{\mu})|_{\sigma=0,\pi}=0.
\eea
The twisted boundary condition is just the Neumann boundary condition when the $B$-field is absent.

In our setup, we have D1-branes wrapping $\mathbb{R}^{1,1}$, D5$_{A}$ $(A\in\six)$ branes wrapping $\mathbb{R}^{1,1}\times \mathbb{C}^{2}_{A}$, D7$_{\bar{a}}$ $(\bar{a}\in\four^{\vee})$ branes wrapping $\mathbb{R}^{1,1}\times \mathbb{C}^{3}_{\bar{a}}$, and D9-branes wrapping $\mathbb{R}^{1,1}\times \mathbb{C}^{4}_{\four}$ (see section~\ref{sec:gaugeorigamiphysicalsetup} for the notation). Thus, for the $\mathbb{R}^{1,1}$-part, the open strings will always obey the NN boundary condition
\bea
\partial_{\sigma}X^{0,9}|_{\sigma=0}=\partial_{\sigma}X^{0,9}|_{\sigma=\pi}=0.
\eea
For the remaining eight directions in $\mathbb{C}^{4}$, it is convenient to combine the real coordinates and introduce the complex coordinates
\bea
Z^{a}=X^{2a-1}+iX^{2a},\quad \bar{Z}^{a}=X^{2a-1}-iX^{2a},\quad a\in\four.
\eea
Under these complex coordinates, the twisted boundary condition where the $B$-field appears can be written as
\bea
(\partial_{+}-e^{-2\pi i\nu_{a}}\partial_{-})Z^{a}|_{\sigma=0}=(\partial_{+}-e^{-2\pi i \nu'_{a}}\partial_{-})Z^{a}|_{\sigma=\pi}=0,
\eea
where $\nu_{a},\nu_{a}'$ determines the boundary condition of the end of string $\sigma=0,\pi$, respectively. Let us focus on $\sigma=0$ for the moment. If the string ends on a D-brane wrapping $\mathbb{C}_{a}$, then $\nu_{a}=v_{a}$, otherwise it is $\nu=1/2$. We have the same relation for the $\nu'_{a}$.

The mode expansion that obeys the EOM and the boundary conditions is
\bea
\nu_{a}=\nu'_{a}:&\qquad Z^{a}=z_{a}+2p^{a}(\sigma^{+}+e^{2\pi i \nu_{a}}\sigma^{-})+i\sum_{n\in\mathbb{Z}\neq 0}\frac{\alpha_{n}^{a}}{n}(e^{-in\sigma^{+}}+e^{2\pi i \nu_{a}}e^{-in\sigma^{-}})\\
\nu'_{a}-\nu_{a}\eqqcolon\delta\neq 0:&\qquad Z^{a}=z_{a}+i\sum_{r\in\mathbb{Z}+\delta}\frac{\alpha^{a}_{r}}{r}(e^{-ir\sigma^{+}}+e^{2\pi i \nu_{a}}e^{-ir\sigma^{-}}).
\eea

Let us move on to the fermions:
\bea
\Psi^{a}_{\pm}=\psi^{2a-1}_{\pm}+i\psi^{2a}_{\pm},\quad \overline{\Psi}_{a}^{\pm}=\psi^{2a-1}_{\pm}-i\psi_{a}^{2a}
\eea
where we complexified them. The boundary conditions are transformed into
\bea
\left(\Psi^{a}_{+}-(-1)^{\xi}e^{-2\pi i \nu_{a}}\Psi^{a}_{-}\right)|_{\sigma=0}=\left(\Psi^{a}_{+}-e^{-2\pi i\nu'_{a}}\Psi^{a}_{-}\right)|_{\sigma=\pi}=0
\eea
where $\xi=0$ for the R sector and $\xi=1$ for the NS sector. Using the doubling trick (see Appendix~\ref{app-chap:openstring}), we can combine the $\Psi_{\pm}^{a}$ into a single field 
\bea
\Psi^{(a)}(\tau,\sigma)=\begin{dcases}
\Psi^{a}_{+}(\tau,\sigma)\qquad 0\leq \sigma \leq \pi\\
e^{-2\pi i\nu_{a}'}\Psi^{a}_{-}(\tau,2\pi-\sigma)\qquad \pi\leq \sigma\leq 2\pi
\end{dcases}.
\eea
We then have
\bea
\Psi^{(a)}(\tau,2\pi)&=e^{-2\pi i\nu'_{a}}\Psi^{a}_{-}(\tau,0)=(-1)^{\xi}e^{2\pi i (\nu_{a}-\nu'_{a})}\Psi^{a}_{+}(\tau,0)\\
&=\exp\left(-2\pi i \left(\delta-\frac{1}{2}\xi\right)\right)\Psi^{(a)}(\tau,2\pi).
\eea
Therefore, the mode expansions of the fermions are given as
\bea
\text{R sector $(\xi=0)$}:&\qquad \Psi^{a}(\tau,\sigma)=\sum_{r\in\mathbb{Z}+\delta}d^{a}_{r}e^{-ir\sigma^{+}},\\
\text{NS sector $(\xi=1)$}:&\qquad \Psi^{a}(\tau,\sigma)=\sum_{r\in\mathbb{Z}+\delta-1/2}b^{a}_{r}e^{-ir\sigma^{+}}.
\eea

Given the above mode expansions, we can define the zero-point energy:\footnote{The zero-mode energy is related with $\tilde{E}=-E_{0}$ compared with the Appendix~\ref{app-chap:openstring}. Moreover since we are considering the complexified bosons and fermions, the zero-point energy is twice different from the Appendix.}
\bea
\widetilde{E}_{Z}(\delta)&=+\sum_{n=0}^{\infty}(n+|\delta|)=\zeta(-1,|\delta|)=\frac{1}{24}-\frac{1}{2}\left(|\delta|-\frac{1}{2}\right)^{2}\\
\widetilde{E}_{\Psi}(\delta)&=\begin{dcases}
-\sum_{n=0}^{\infty}(n+|\delta|)=-\frac{1}{24}+\frac{1}{2}\left(|\delta|-\frac{1}{2}\right)^{2}\qquad \text{R sector}\\
-\sum_{n=0}^{\infty}(n+||\delta|-\frac{1}{2}|)=-\frac{1}{24}+\frac{1}{2}\left(||\delta|-\frac{1}{2}|-\frac{1}{2}\right)^{2}\qquad \text{NS sector}
\end{dcases}
\eea
where $\zeta(s,a)=\sum_{n=0}^{\infty}(a+n)^{-s}$. The zero-point energy coming from the open strings connected to the complex plane $\mathbb{C}_{a}$ is
\bea\label{eq:zero-point-energy-complex-plane}
\widetilde{E}(\delta)=\begin{dcases}
    0\qquad \qquad \text{R sector},\\
    \frac{1}{8}-\frac{1}{2}||\delta|-\frac{1}{2}|\qquad \text{NS sector},
\end{dcases}
\eea
where the zero-point energy of the R sector always vanishes.

Taking the light cone gauge (see Appendix~\ref{app:lightconegauge} for an example), the mass for the two-dimensional particle of $\mathbb{R}^{1,1}$ is proportional to 
\bea
m^{2}\propto (\widehat{N}+\widetilde{E}_{0})
\eea
where $\widehat{N}$ counts the mode fluctuations from the $\mathbb{C}^{4}$ part and $\widetilde{E}_{0}$ is the total zero-point energy.

\subsection{Magnificent four}\label{sec:M4-openstring}
Let us study the open string spectrum of the magnificent four setup where we have $k$ D1-branes wrapping $\mathbb{R}^{1,1}$ and $n$ D9-branes wrapping $\mathbb{R}^{1,1}\times \mathbb{C}^{4}$. Again, we turn on the $B$-field along the D9-branes as in \eqref{eq:B-field}. 

\paragraph{D1--D1 strings}Let us first study the the D1--D1 strings. It is well-known that the low energy field theory of $k$ parallel D1-branes is obtained by the dimensional reduction of 10d $\mathcal{N}=1$ $\U(k)$ SYM, which is the 2d $\mathcal{N}=(8,8)$ $\U(k)$ SYM. Let us review the derivation of this. The open strings obey the NN boundary conditions along $\mathbb{R}^{1,1}$ and DD boundary conditions along $\mathbb{C}^{4}$. The zero point energy of the NS sector is computed as $\widetilde{E}_{0}=-1/2$.

In the R sector, the zero-point energy vanishes, and the level zero $\widehat{N}=0$ states give the massless states. The ten zero-modes coming from the fermions give the 32 degenerate ground states $\ket{\vec{s}}_{\text{R}}$. After GSO projection, half of them remain giving 16 fermions.

In the NS sector, the ground state $\ket{0}_{\text{NS}}$ gives the $\widetilde{E}_{0}=-1/2$ which is tachyonic. GSO projection removes this state and the level one excitation coming from $\psi_{-1/2}^{\mu}\ket{0}_{\text{NS}}$ $(\mu=1,\ldots, 8)$ in light-cone gauge give the massless spectrum. From the 2d theory viewpoint, the eight real scalars correspond to the bosonic coordinates $x^{1},\ldots,x^{8}$. We can combine them into four complex scalars $\mathsf{B}_{a}$ $(a\in\four)$ and their complex conjugates. 

\paragraph{D1--D9 strings} The T-dualized setup of ours was analyzed in \cite{Witten:2000mf}. We follow the discussion there. In this case, the boundary conditions for the D1--D9 strings are NN along $\mathbb{R}^{1,1}$ and DT along $\mathbb{C}^{4}$. We then have $\delta=1/2-v_{a}$. Using \eqref{eq:zero-point-energy-complex-plane}, the zero-point energy of the NS sector is computed as
\bea
\widetilde{E}_{0}=\sum_{a\in\four}\widetilde{E}\left(\frac{1}{2}-v_{a}\right)=\frac{1}{2}\left(1-\sum_{a\in\four}|v_{a}|\right)
\eea
where we used $-1/2<v_{a}<1/2$. When $\sum_{a\in\four}|v_{a}|=1$, the ground state coming from $\widehat{N}=0$ is massless and unique since there are no worldsheet zero modes to be quantized. It will give a massless boson, which will be denoted $\mathsf{I}$ (see section \ref{sec:M4-2dtheory}). For the R sector, since the zero-point energy is zero, there is a unique fermionic ground state coming from $\widehat{N}=0$. Combining both of them give a 2d $\mathcal{N}=(0,2)$ chiral multiplet (see section~\ref{sec:2susygaugetheory} for the notations).

Let us consider the case when the $B$-field is absent. For this case, we can formally set $v_{a}=0\,(a\in\four)$ and the zero-point energy is $\widetilde{E}_{0}=1/2$. Since physical states need to obey $\widehat{N}+1/2=0$, there are no physical states in the NS sector. For the R-sector, the ground state survives. We thus obtain a Fermi multiplet in this case.

Note that when the $B$-field is absent or $\sum_{a\in\four}|v_{a}|<1$, since the zero-point energy is positive, we do not have any massless bosons. For the computations in later sections to be applicable, we \textit{need} the nontrivial $B$-field. On the other hand when $\sum_{a\in\four}|v_{a}|>1$, the theory is unstable and expected to decay to a supersymmetric ground state with two supersymmetries.


\subsection{Tetrahedron instanton}\label{sec:tetra-opstring}
Let us consider the tetrahedron instanton setup where we have $k$ D1-branes and $n_{\bar{a}}$ D7$_{\bar{a}}$-branes. We focus on the spectrum coming from the D1--D7 strings. Using the quadrality, we can focus on the D7$_{123}$-branes. The boundary conditions are NN in $\mathbb{R}^{1,1}$, DD in $\mathbb{C}_{4}$ and DT in $\mathbb{C}_{123}^{3}$.

In the R sector, the zero-point energy vanishes. In the light-cone gauge, we have two-zero modes coming from the $\mathbb{C}_{4}$ part give the fermionic massless states with spin $\pm1/2$ under rotations in $\mathbb{C}_{4}$. GSO projection kills one of them.

In the NS sector, the zero-point energy is computed as
\bea
\widetilde{E}_{0}=\sum_{a=1,2,3}\widetilde{E}\left(\frac{1}{2}-v_{a}\right)+\widetilde{E}(0)=\frac{1}{4}-\frac{1}{2}\sum_{a\in 123}|v_{a}|=\frac{1}{2}\left(\frac{1}{2}-\sum_{a\in 123}|v_{a}|\right).
\eea
When $\sum_{a=1,2,3}|v_{a}|=1/2$, the ground state energy vanishes. Since there are no worldsheet zero-modes, there is a unique massless ground state with vanishing spin. It survives the GSO projection and gives a real scalar field which will be denoted as $\mathsf{I}_{\bar{4}}$ (see section~\ref{sec:tetra-2dtheory}). Combining with the D7--D1 strings and the fermions coming from the Ramond sectors, we obtain a $\mathcal{N}=(2,2)$ chiral multiplet.


\subsection{Spiked instanton}\label{sec:spiked-opstring}
Let us discuss the massless spectrum for the spiked instanton setup where we have D5-branes wrapping $\mathbb{R}^{1,1}\times\mathbb{C}^{2}_{A}$ $(A\in\six)$. For a detailed analysis see \cite{Nekrasov:2016gud}. We focus on the D1--D5$_{A}$ strings. The boundary conditions are NN for $\mathbb{R}^{1,1}$, DD for $\mathbb{C}^{2}_{\bar{A}}$ and DT for $\mathbb{C}^{2}_{A}$. We have $\delta=v_{a}-1/2$ for bosons $Z^{a} (a\in A)$ and $\delta=0$ for $a\in\bar{A}$ giving $\mathbb{Z}+\delta$ modings. The R fermions similarly have the same modings with the bosons while the NS fermions have $\mathbb{Z}+v_{a}$ for $a\in A$ and $\mathbb{Z}+1/2$ for $a\in\bar{A}$.

The zero-point energy for the NS sector is given as
\bea
\widetilde{E}_{0}=\sum_{i\in A=(ab)}\widetilde{E}\left(\frac{1}{2}-v_{a}\right)+2\widetilde{E}(0)=-\frac{1}{2}(|v_{a}|+|v_{b}|).
\eea
To study the low energy spectrum, we need to study the action of fermionic operators with modes $\mathbb{Z}+v_{a}$ and deal with the parameters $v_{a}$ carefully. Since it is tedious to do so, in this thesis we simply set $v_{a}=v_{b}=0$ and consider the massless spectrum. The massless spectrum arising from this is the dimensional reduction of the 4d $\mathcal{N}=2$ hypermultiplet giving the 2d $\mathcal{N}=(4,4)$ bifundamental hypermultiplet.  



\section{2d \texorpdfstring{$\mathcal{N}=(0,2)$}{N=(0,2)} quiver gauge theories and vacuum moduli space}\label{sec:2susygaugetheory}
Given the open string spectrum, we now know that the low energy effective theory of the D1-branes is given by a 2d $\mathcal{N}=(0,2)$ theory. However, the matter contents given in the previous section are not enough to determine this theory and we need to consider interaction terms coming from the so-called $J$ and $E$-terms. To get such terms from string theory, one needs to study correlation functions of open string vertex operators. See \cite{Billo:2021xzh} for the magnificent four setup and \cite{Nekrasov:2016gud} for the spiked instanton setup. Instead of proceeding in this direction, we will rather use the symmetries such as the quadrality in the gauge origami setup and derive the corresponding $J$ and $E$-terms. 

We first review some general aspects of 2d $\mathcal{N}=(0,2)$ quiver gauge theories, and then discuss how to determine the $J$ and $E$-terms of the gauge origami system. Given such $J$ and $E$-terms, we will show that the vacuum moduli space coincides with the instanton moduli space discussed in section~\ref{sec:general-instanton}.

\subsection{General structure}\label{sec:2susyquivergaugetheory-general}
Let us briefly review basic aspects of the two-dimensional $\mathcal{N}=(0,2)$ gauge theory using the superfield formalism. The $\mathcal{N}=(0,2)$ superspace is parametrized by the space-time coordinates $x^{0},x^{1}$ and fermionic coordinates $\theta^{+},\bar{\theta}^{+}$. The $\mathcal{N}=(0,2)$ supersymmetry is generated by two chiral supercharges $\tQ_{+},\otQ_{+}$:
\bea
\text{Q}_{+}^{2}=\otQ_{+}^{2}=0,\quad \{\tQ_{+},\otQ_{+}\}=-i(\partial_{0}+\partial_{1}).
\eea
The operators have a differential operator representation
\bea
\tQ_{+}=\frac{\partial}{\partial\theta^{+}}+\frac{i}{2}\overline{\theta}^{+}\left(\partial_{0}+\partial_{1}\right),\quad \otQ_{+}=-\frac{\partial}{\partial\overline{\theta}^{+}}-\frac{i}{2}\theta^{+}\left(\partial_{0}+\partial_{1}\right)
\eea
and the supersymmetry transformation is defined as $\delta=-{\xi}_{-}\tQ_{+}+\bar{\xi}_{-}\otQ_{+}$, where $\xi_{-},\overbar{\xi}_{-}$ are Grassmann variables. These commute with another set of superderivatives 
\bea
&\tD_{+}=\frac{\partial}{\partial\theta^{+}}-\frac{i}{2}\overline{\theta}^{+}(\partial_{0}+\partial_{1} ),\quad \otD_{+}=-\frac{\partial}{\partial \overline{\theta}^{+}}+\frac{i}{2}\theta^{+} (\partial_{0}+\partial_{1}),
\eea
obeying $\tD_{+}^{2}=\otD_{+}^{2}=0,\,\,\{\tD_{+},\otD_{+}\}=i(\partial_{0}+\partial_{1})$. For simplicity, we also introduce $\partial_{\pm}=\left(\partial_{0}\pm \partial_{1}\right)/2$ and $x^{\pm}=x^{0}\pm x^{1}$.

There are three kinds of $\mathcal{N}=(0,2)$ superfields that we are interested in are the vector, chiral and Fermi superfields. 

\paragraph{Chiral superfield} An $\mathcal{N}=(0,2)$ chiral superfield $\Phi$ consists of a complex scalar $\phi$ and a fermion $\psi_{+}$ obeying
\bea
\otD_{+}\Phi=0,\quad \Phi=\phi+\theta^{+} \psi_{+}-i\theta^{+}\overline{\theta}^{+}\partial_{+}\phi
\eea
The supersymmetry transformation is given as
\bea
\delta \phi&=-\xi_{-}\psi_{+},\quad \delta \psi=2i\bar{\xi}_{-}\,\partial_{+}\phi.
\eea
The supersymmetric Lagrangian is
\bea\label{eq:2susychirallagrangian}
S_{\text{chiral}}=\int d^{2}x\int d\theta^{+} d\bar{\theta}^{+}\,i\overline{\Phi} (\partial_{0}-\partial_{1}) \Phi=\int d^{2}x\left( |\partial_{0}\phi|^{2}-|\partial_{1}\phi|^{2}+i\bar{\psi}_{+}(\partial_{0}-\partial_{1})\psi_{+}\right).
\eea

\paragraph{Fermi superfield}
A Fermi superfield $\Lambda=(\eta_{-},G)$ consists of a fermion $\eta_{-}$ and an auxiliary field $G$. It obeys
\bea
\otD_{+} \Lambda&=E(\Phi),
\eea
where $E(\Phi)$ is a holomorphic function of the chiral superfields $\Phi$. Note that this condition is compatible because
\bea
\otD_{+}\otD_{+} \Lambda=\otD_{+} E(\Phi)=0.
\eea 
In component fields, the Fermi superfield is expanded as
\bea
\Lambda&=\eta_{-}-\theta^{+} G-i\theta^{+}\bar{\theta}^{+}\partial_{+}\eta_{-}-\bar{\theta}^{+} E(\Phi).
\eea
and the supersymmetric transformations of the component fields are
\bea
\delta \eta_{-}&=\xi_{-} G+\bar{\xi}_{-} E(\phi),\quad \delta G=\bar{\xi}_{-}\left(-2i\partial_{+}\eta_{-}+\psi_{+}^{i}\partial_{i}E(\phi)\right).
\eea
The supersymmetric Lagrangian of the Fermi field is 
\bea\label{eq:2susyFermilagrangian}
S_{\text{Fermi}}=\int d^{2}x\int d\theta^{+} d\bar{\theta}^{+} \bar{\Lambda}\Lambda = \int d^{2}x\left(i\bar{\eta}_{-}(\partial_{0}+\partial_{1}){\eta_{-}}+|G|^{2}-|E(\phi)|^{2}-\bar{\eta}_{-}\partial_{i}E(\phi)\psi_{+}^{i}-\bar{\psi}_{+}^{\bar{i}}\partial_{\bar{i}}\overline{E(\phi)}\eta_{-}\right).
\eea
Including a nontrivial $E$-term induces interactions between the fermions of the chiral and Fermi superfield.

\paragraph{$J$-term}We can introduce another term called the $J$-term which will induce different interactions between the chiral and Fermi superfields. Given a set of Fermi and chiral superfields $\{\Lambda^{\alpha}, J_{\alpha}(\Phi)\}$, the superfield $\sum_{\alpha}\Lambda^{\alpha}J_{\alpha}(\Phi)$ is a Fermi superfield when
\bea\label{eq:JE-condition}
\sum_{\alpha}J_{\alpha}(\phi)E^{\alpha}(\phi)=0
\eea
is obeyed. Under this condition, the following term is also supersymmetric 
\bea\label{eq:2susyJterm}
S_{J\text{-term}}=-\int d^{2}x\int d\theta^{+}\,\left(\Lambda^{\alpha}J_{\alpha}(\Phi)\right)_{\bar{\theta}=0}+\text{c.c.}=\int d^{2}x\left(\eta_{-}^{\alpha}\partial_{i}J_{\alpha}(\phi)\psi_{+}^{i}+G^{\alpha}J_{\alpha}(\phi)\right)+\text{c.c.}.
\eea
Note that if \eqref{eq:JE-condition} is not obeyed, the Lagrangian is not $(0,2)$ supersymmetric.

\paragraph{$J$-$E$ duality}
After integrating the auxiliary field $G$, the Lagrangian becomes symmetric in $E\leftrightarrow J$ and we have a duality between the $J,E$-terms. In the superfield formalism, this can be seen as follows. Let $\Lambda$ be a Fermi superfield obeying $\otD_{+}\Lambda=E(\Phi)$. The Lagrangian is written as
\bea\label{eq:JEduality-Jdescription}
S[\Lambda]&=\int d^{2}x \left(\int d\theta^{+} d\bar{\theta}^{+}\bar{\Lambda}\Lambda-\int d\theta^{+} \Lambda J -\int d\bar{\theta}^{+}\bar{\Lambda}\bar{J}\right)\\
&=\int d^{2}x \left(-\tD_{+}\otD_{+}(\bar{\Lambda}\Lambda)-\tD_{+}(\Lambda J)+\otD_{+}(\bar{\Lambda}\bar{J})\right)
\eea
where we used $\tD_{+}\sim \partial/\partial_{\theta^{+}}$ and $\otD_{+}\sim -\partial/\partial_{\bar{\theta}^{+}}$ up to total derivatives in the bosonic coordinates. We can introduce a Grassmann superfield $\Psi$ to the Lagrangian as
\bea
S[\Psi,\Lambda]=\int d^{2}x \left(-\tD_{+}\otD_{+}(\bar{\Psi}\Psi-\Lambda\bar{\Psi}-\Psi\bar{\Lambda})-\tD_{+}(\Lambda J)+\otD_{+}(\bar{\Lambda}\bar{J})\right),
\eea
which reproduces the original Lagrangian after integrating out $\Psi$. Let us see what will happen if we integrate out $\Lambda$ first. We can push the $\D_{+},\otD_{+}$ inside the Lagrangian multiplier and obtain
\bea
S[\Psi,\Lambda]=\int d^{2}x \left(-\tD_{+}\otD_{+}(\bar{\Psi}\Psi)+\tD_{+}(E\bar{\Psi})+\tD_{+}(\Lambda\otD_{+}\bar{\Psi}-\Lambda J)+\text{c.c.}   \right).
\eea
Integrating the superfield $\Lambda$ gives
\bea
\otD_{+}\overline{\Psi}=J(\Phi)
\eea
and
\bea\label{eq:JEduality-Edescription}
S[\Psi]&=\int d^{2}x \left(\int d\theta^{+} d\bar{\theta}^{+}\overline{\Psi}\Psi+\int d\theta^{+} E(\Phi)\overline{\Psi}+\int d\bar{\theta}^{+}\Psi\overline{E(\Phi)}\right)\\
&=\int d^{2}x \left(\int d\theta^{+} d\bar{\theta}^{+}\overline{\Lambda'}\Lambda'+\int d\theta^{+} E(\Phi)\Lambda'+\int d\bar{\theta}^{+}\overline{\Lambda'}\overline{E(\Phi)}\right)
\eea
where we defined $\Lambda'=\overline{\Psi}$ at the second line. 

Note that the Fermi superfield $\Lambda'$ transforms in the conjugate representation of the symmetry groups compared with $\Lambda$. Moreover, because of
\bea
\Lambda&=\cdots -\bar{\theta}^{+}E(\Phi),\\
\Lambda'&= \cdots -\bar{\theta}^{+}J(\Phi),
\eea
the chiral superfield $E(\Phi)$ ($J(\Phi)$) has the same quantum numbers with $\Lambda$ ($\Lambda'$) and they are conjugate with each other. This can also be observed by looking at the interaction term $\Lambda J$. For it to be invariant under a symmetry transformation, $\Lambda$ and $J(\Phi)$ needs to transform in a conjugate way.

\paragraph{Vector superfield}Since we are interested in supersymmetric gauge theories, we need to include vector multiplets to the theory. A vector superfield $V=(v_{0},v_{1},\lambda_{-},\bar{\lambda}_{-},D)$ consists of a two-dimensional gauge field $v_{0,1}$, a fermion $\lambda_{-},\overbar{\lambda}_{-}$, and an auxiliary scalar field $D$ transforming in the adjoint representation of the gauge group. In this thesis, we will always keep the group structure to be unitary groups $\U(n)$.

The supersymmetry algebra will be modified to a gauge covariant one:
\bea
\mathcal{D}_{+}^{2}=\overbar{\mathcal{D}}_{+}^{2}=0,\quad \{\mathcal{D}_{+},\overbar{\mathcal{D}}_{+}\}=i(\mathcal{D}_{0}+\mathcal{D}_{1}),
\eea
where $\mathcal{D}_{0,1}$ are gauge covariant superderivatives associated with $\partial_{0},\partial_{1}$. Namely, when restricted on $x^{0,1}$, they will reduce to the gauge covariant differential operators $D_{0,1}=\partial_{0,1}+iv_{0,1}$. Using the parametrization 
\bea
\mathcal{D}_{+}=e^{-\Omega} \tD_{+} e^{\Omega},\quad \overbar{\mathcal{D}}_{+}=e^{-\Omega}\otD_{+}e^{\Omega}
\eea
and choosing the Wess--Zumino gauge
\bea
\Omega=-\frac{1}{2}\theta^{+}\bar{\theta}^{+}(v_{0}+v_{1}),
\eea
the operators are 
\bea
\mathcal{D}_{+}
&= \frac{\partial}{\partial \theta^{+}}-\frac{i}{2}\bar{\theta}^{+}(D_{0}+D_{1})    \quad 
\overbar{\mathcal{D}}_{+}
=-\frac{\partial}{\partial \bar{\theta}^{+}}+\frac{i}{2}\theta^{+}\left(D_{0}+D_{1}\right).
\eea
Under this gauge, we have $\mathcal{D}_{0}+\mathcal{D}_{1}=D_{0}+D_{1}$. The vector superfield is then defined as $\mathcal{D}_{0}-\mathcal{D}_{1}=\partial_{0}-\partial_{1}+i V$ where we can write it in component fields as
\bea
V=v_{0}-v_{1}-i\theta^{+}\bar{\lambda}_{-}-i\bar{\theta}^{+}\lambda_{-}+\theta^{+}\bar{\theta}^{+}D.
\eea
The supersymmetry transformation is given as
\bea
\delta v_{0}&=-\delta v_{1}=\frac{i}{2}\xi_{-}\bar{\lambda}_{-}+\frac{i}{2}\bar{\xi}_{-}\lambda_{-},\\
\delta \lambda_{-}&=\xi_{-}(F_{01}+iD),\\
\delta D&= \frac{1}{2}\xi_{-} (D_{0}+D_{1})\bar{\lambda}_{-}-\frac{1}{2}\bar{\xi}_{-}(D_{0}+D_{1})\lambda_{-}. 
\eea
The field strength is defined as
\bea
\Upsilon&=[\overbar{\mathcal{D}},\mathcal{D}_{0}-\mathcal{D}_{1}]=-\lambda_{-}+\theta^{+}(F_{01}+iD)+\frac{i}{2}\theta^{+}\bar{\theta}^{+}(D_{0}+D_{1})\lambda_{-}
\eea
where $F_{01}=(\partial_{0}v_{1}-\partial_{1}v_{0}+i[v_{0},v_{1}])=-i[D_{0},D_{1}]$. The gauge invariant supersymmetric Lagrangian is then given as
\bea\label{eq:2susyvectorLagrangian}
S_{\text{gauge}}&=\int d^{2}x\int d\theta^{+} d\bar{\theta}^{+} \frac{1}{2g^{2}}\Tr \overbar{\Upsilon}\Upsilon=\frac{1}{2g^{2}}\int d^{2}x\Tr\left[F_{01}^{2}+D^{2}+i\bar{\lambda}_{-}(D_{0}+D_{1})\lambda_{-}\right].
\eea
Using this vector superfield, we can also include the FI-term as
\bea\label{eq:2susyFIterm}
S_{\text{FI}}&=\frac{ir}{2}\int d^{2}x\int d\theta^{+} \Tr \Upsilon|_{\bar{\theta}^{+}=0}+\text{c.c.} =-r\int d^{2}x \Tr D
\eea
where we assumed $r$ to be real.

\paragraph{Gauge covariant chiral and Fermi superfields} To obtain gauge covariant chiral superfields and Fermi superfields, we need to modify the differential operators to gauge covariant differential operators:
\bea
\overline{\mathcal{D}}_{+}\Phi=0,\quad \overline{\mathcal{D}}_{+}\Lambda_{-} =E(\Phi).
\eea
In components fields, we just need to modify the operators as $\partial_{0,1}\rightarrow D_{0,1}$. Moreover, the component fields and superfields transform under a representation of the gauge group.

The Lagrangian of the supersymmetric gauge theory is then summarized as
\bea\label{eq:2susygaugeLagrangian}
\mathcal{L}_{\text{gauge}}&=\int d\theta^{+} d\bar{\theta}^{+} \frac{1}{2g^{2}}\Tr\overline{\Upsilon}\Upsilon= \frac{1}{2g^{2}}\Tr\left[(F_{01})^{2}+i\bar{\lambda}_{-}(D_{0}+D_{1})\lambda_{-}+D^{2}\right]\\
\mathcal{L}_{\text{chiral}}&=\int d\theta^{+} d\bar{\theta}^{+}\, i\overline{\Phi}(\mathcal{D}_{0}-\mathcal{D}_{1}){\Phi}\,\,=\overline{D_{0}\phi}D_{0}\phi-\overline{D_{1}\phi}D_{1}\phi+i\bar{\psi}_{+}(D_{0}-D_{1})\psi_{+}+\bar{\phi}D\phi-i\bar{\phi}\lambda_{-}\psi_{+}+i\overline{\psi_{+}\lambda_{-}}\phi \\
\mathcal{L}_{\text{Fermi}}&=\int d\theta^{+} d\bar{\theta}^{+}\,\, \bar{\Lambda}_{-}\Lambda_{-} = i\bar{\eta}_{-}(D_{0}+D_{1})\eta_{-}+|G|^{2}-|E(\phi)|^{2}-\bar{\eta}_{-}\partial_{i}E(\phi)\psi_{+}^{i}-\bar{\psi}^{\bar{i}}_{+}\partial_{\bar{i}}\overline{E(\phi)}\eta_{-}\\
\mathcal{L}_{J\text{-term}}&=-\int d\theta^{+}\,\left(\Lambda_{-}^{\alpha}J_{\alpha}(\Phi)\right)_{\bar{\theta}=0}+\text{c.c.}=\eta_{-}^{\alpha}\partial_{i}J_{\alpha}(\phi)\psi_{+}^{i}+G^{\alpha}J_{\alpha}(\phi)+\text{c.c.}\\
\mathcal{L}_{\text{FI}}&=\frac{ir}{2}\int d\theta^{+} \Tr \Upsilon|_{\bar{\theta}^{+}=0}+\text{c.c.} =-r \Tr D.
\eea
After integrating the auxiliary fields, the scalar potential is given as
\bea\label{eq:scalarpotential-general}
U_{\text{pot.}}=\frac{g^{2}}{2}\left(\phi\bar{\phi}-r\right)^{2}+\sum_{\alpha}\left(|E_{\alpha}(\phi)|^{2}+|J_{\alpha}(\phi)|^{2}\right).
\eea
The vacuum moduli space is then given by 
\bea
|\phi|^{2}=r,\quad E_{\alpha}(\phi)=J_{\alpha}(\phi)=0.
\eea
We call the former the $D$-term condition and the latter the $E,J$-term condition.

\paragraph{Quiver gauge theory}
Let us consider a special class of 2d $\mathcal{N}=(0,2)$ gauge theories where the data of the gauge theory is encoded in a quiver diagram. Given a quiver $\overbar{Q}=(\overbar{Q}_{0},\overbar{Q}_{1})$, where $\overbar{Q}_{0}$ is a set of nodes and $\overbar{Q}_{1}$ is a set of edges, the information of the gauge symmetry and matter fields can be obtained as follows.
\begin{itemize}[topsep=0pt, partopsep=0pt, itemsep=0pt]
    \item For each node $a$, a group $\U(N_{a})$ is assigned.
    \item Edges connect the nodes and there are two sets of oriented edges colored in black and red. We denote the set of black edges as $\overline{Q}_{1}^{(0)}$ while set of red edges are denoted as $\overline{Q}_{1}^{(1)}$.
    \item A black oriented edge from a source node $a$ to a terminal node $b$ represents a chiral superfield $\Phi_{ba}=\Phi_{b\leftarrow a}$ transforming in the bifundamental representation $(\overline{N}_{a},N_{b})$ of $\U(N_{a})\times \U(N_{b})$.
    \item A red oriented edge connecting two nodes $a,b$ represents a Fermi superfield $\Lambda_{ba}=\Lambda_{b\leftarrow a}$ transforming in the bifundamental representation $(\overline{N}_{a},N_{b})$ of $\U(N_{a})\times \U(N_{b})$.\footnote{In the context of $\mathcal{N}=(0,2)$ quiver gauge theories, the Fermi superfields are sometimes denoted using \textit{unoriented edges} reflecting the fact that for the Fermi superfields, we will always have the conjugate $\overline{\Lambda}_{ba}$ transforming in $(N_{a},\overline{N}_{b})$ of $\U(N_{a})\times \U(N_{b})$. This also can be seen from the fact that there is a duality between the $J,E$-terms as mentioned in \eqref{eq:JEduality-Edescription}, \eqref{eq:JEduality-Jdescription} and the quantum number of them are conjugate with each other. }
\end{itemize}
We note that as in usual quiver gauge theories, circle nodes give gauge symmetries while square nodes give the flavor symmetries. Moreover, self-loop nodes transform in the adjoint representation. An example of the quiver diagram is
\bea
\begin{tikzpicture}[decoration={markings,mark=at position \arrowHeadPosition with {\arrow{latex}}}]
 \tikzset{
        box/.style={draw, minimum width=0.6cm, minimum height=0.6cm, text centered,thick},
        ->-/.style={decoration={
        markings,mark=at position #1 with {\arrow[scale=1.5]{>}}},postaction={decorate},line width=0.5mm},
        -<-/.style={decoration={
        markings,
        mark=at position #1 with {\arrow[scale=1.5]{<}}},postaction={decorate},line width=0.5mm}    
    }
\begin{scope}[xshift=4cm]
    \draw[postaction={decorate}, black,thick,scale=1.3] (0.65,0) arc(0:-180:0.65 and 0.1) ;
    \draw[postaction={decorate}, black,thick,scale=1.3] (0.75,0) arc(0:-180:0.75 and 0.2) ;
    \draw[postaction={decorate}, black,thick,scale=1.3] (-0.65,0) arc(180:0:0.65 and 0.1) ;
    \draw[postaction={decorate}, black,thick,scale=1.3] (-0.75,0) arc(180:0:0.75 and 0.2) ;
    \draw[postaction={decorate}, black,thick,scale=1.3] (-0.8,0) arc(360:0:0.4 and 0.3) ;
    \draw[postaction={decorate}, black,thick,scale=1.3] (0.8,0) arc(-180:180:0.4 and 0.3) ;
    
    \draw[red,thick,postaction={decorate},scale=1.3] (0.65,0) arc(0:-180:0.65 and 0.3) ;
    \draw[red,thick,postaction={decorate},scale=1.3] (0.75,0) arc(0:-180:0.75 and 0.4) ;
    \draw[red,thick,postaction={decorate},scale=1.3] (-0.65,0) arc(180:0:0.65 and 0.3) ;
    \draw[red,thick,postaction={decorate},scale=1.3] (-0.75,0) arc(180:0:0.75 and 0.4) ;
    \draw[fill=black!20!white,thick](-0.95,0) circle(0.3cm);
    \draw[fill=black!20!white,thick](0.9,0) circle(0.3cm);
    \node[box,fill=black!20!white] at (-0.95,2.0){$ $};
    \draw[postaction={decorate},thick] (-0.95,1.7)--(-0.95,0.3);
\end{scope}
\end{tikzpicture}
\eea
The quiver diagram itself does not encode the $E,J$-terms and so we need to include them additionally.

\paragraph{$\mathcal{N}=(2,2)\rightarrow \mathcal{N}=(0,2)$ decomposition}The 2d $\mathcal{N}=(0,2)$ system discussed above actually can be understood as a reduction of the 2d $\mathcal{N}=(2,2)$ theory which is a dimensional reduction of the 4d $\mathcal{N}=1$ theory. This time, we will have four supercharges $\tQ_{\pm},\otQ_{\pm}$ obeying 
\bea
\{\tQ_{\pm},\otQ_{\pm}\}=-2i\partial_{\pm}.
\eea
Similarly, we can further introduce another set of superderivatives $\tD_{\pm},\otD_{\pm}$ obeying 
\bea
\{\tD_{\pm},\otD_{\pm}\}=2i\partial_{\pm}.
\eea
There will be three basic $\mathcal{N}=(2,2)$ superfields: chiral, twisted chiral\footnote{We will not consider this type of superfield in this thesis.}, and vector superfields. 

For simplicity, let us focus only on the chiral superfield. The chiral superfield is defined by the condition
\bea
\otD_{\pm}\Phi=0,\quad \Phi(x^{\pm},\theta^{\pm},\bar{\theta}^{\pm})=\phi(y^{\pm})+\theta^{\alpha}\psi_{\alpha}(y^{\pm})+\theta^{+}\theta^{-}F(y^{\pm})
\eea
where $y^{\pm}=x^{\pm}-i\theta^{\pm}\bar{\theta}^{\pm}$. Note that the anti-chiral superfield is given by $D_{\pm}\overline{\Phi}=0$. The kinetic term of a supersymmetric Lagrangian is given as 
\bea
S_{\text{chiral}}&=\int d^{2}x\int d^{4}\theta\,\, \overline{\Phi}\Phi=\int d^{2}x \left(|\partial_{0}\phi|^{2}-|\partial_{1}\phi|^{2}+i\bar{\psi}_{-}(\partial_{0}+\partial_{1})\psi_{-}+i \bar{\psi}_{+}(\partial _{0}-\partial_{1})\psi_{+}+|F|^{2}\right).
\eea
We can add interaction terms using the so-called \textit{superpotential}, which is a holomorphic function of the chiral superfields $\Phi$, as
\bea
S_{W}=\int d^{2}x\int d^{2}\theta W(\Phi)+\text{c.c}=\int d^{2}x \left(W'(\phi)F-W''(\phi)\psi_{+}\psi_{-}+\text{c.c.}\right).
\eea
Combining both terms and integrating out the auxiliary field $F$ gives
\bea
F=-\overline{W}'(\bar{\phi}).
\eea
The vacuum of the scalar potential then comes from
\bea
\frac{\partial W(\phi)}{\partial \phi}=0
\eea
which is usually called the $F$-term condition. 

By looking at the component fields and the supersymmetric Lagrangian and comparing with \eqref{eq:2susychirallagrangian}, \eqref{eq:2susyFermilagrangian}, one can observe that the $\mathcal{N}=(2,2)$ chiral multiplet $\Phi^{(2,2)}=(\phi,\psi_{\pm},F)$ can be decomposed into a $\mathcal{N}=(0,2)$ chiral multiplet $\Phi^{(0,2)}=(\phi, \psi_{+})$ and a Fermi multiplet $\Lambda_{-}=(\psi_{-},F)$ with the $E$-term $E(\Phi)=0$. Indeed we have
\bea
S^{(0,2)}_{ \text{chiral}} +S^{(0,2)}_{\text{Fermi}}= S^{(2,2)}_{\text{chiral}}.
\eea
Moreover, the $\mathcal{N}=(2,2)$ superpotential term is related with the $\mathcal{N}=(0,2)$ $J$-term as
\bea\label{eq:4susyJterm}
J(\phi)=\partial W(\phi)
\eea
because of 
\bea
S_{W}&=\int d^{2}x \left(W'(\phi)F-W''(\phi)\psi_{+}\psi_{-}+\text{c.c.}\right)\\
&=\int d^{2}x\,\left(\psi_{-}\partial_{i}J(\phi)\psi_{+}+FJ(\phi)+\text{c.c.}\right)=S_{J\text{-term}}.
\eea

Including the vector multiplet makes the discussion complicated so we will not discuss it explicitly, so see \cite{Witten:1993yc}. The $\mathcal{N}=(2,2)$ vector multiplet $V^{(2,2)}$ will at the end decompose into a $\mathcal{N}=(0,2)$ vector multiplet $V^{(0,2)}$ and a $\mathcal{N}=(0,2)$ chiral multiplet $\Sigma^{(0,2)}$ (of course in the adjoint representation). The gauge interactions will modify the decomposition of the $\mathcal{N}=(2,2)$ chiral multiplets $\Phi^{(2,2)}$ transforming under a some representation of the gauge group into a chiral multiplet $\Phi^{(0,2)}$ and a Fermi multiplet $\Lambda^{(0,2)}$ transforming under the same representation. This time the $E$-term will be modified to
\bea\label{eq:4susyEterm}
E=\Sigma^{(0,2)} \Phi^{(0,2)}.
\eea
The $J$-term is not modified and the relation is simply $J=\partial W$. The traceless condition \eqref{eq:JE-condition} simply means that the superpotential is gauge invariant and thus it is automatically satisfied for $\mathcal{N}=(2,2)$ theories. Let us summarize the above results in a table:
\bea\label{eq:4susy2susycorrespondence}
\renewcommand{\arraystretch}{1.05}
    \begin{tabular}{|c|c|}\hline
     $\mathcal{N}=(2,2)$    & $\mathcal{N}=(0,2)$  \\\hline
      vector $V^{(2,2)}$  &  vector $V^{(0,2)}$ \\
        & chiral $\Sigma^{(0,2)}$ \\ \hline
        chiral $\Phi^{(2,2)}$  &  chiral $\Phi^{(0,2)}$\\
         & Fermi $\Lambda^{(0,2)}$\\
          & E-term $E=\Sigma^{(0,2)}\cdot \Phi^{(0,2)}$\\ \hline
          superpotential $W$  & $J$-term  $J=\partial W$\\\hline
     \end{tabular}
\eea

We can further consider quiver gauge theories with four supersymmetries. Let $Q=(Q_{0},Q_{1})$ be a quiver, where $Q_{0}$ is a set of nodes and $Q_{1}$ is a set of edges. The matter fields are described as follows.
\begin{itemize}[topsep=0pt, partopsep=0pt, itemsep=0pt]
    \item For each node $a$, a group $\U(N_{a})$ is assigned.
    \item There is only one type of edges connecting the nodes. 
    \item An oriented edge from a source node $a$ to a terminal node $b$ represents a $\mathcal{N}=(0,2)$ chiral superfield $\Phi_{ba}=\Phi_{b\leftarrow a}$ transforming in the bifundamental representation $(\overline{N}_{a},N_{b})$ of $\U(N_{a})\times \U(N_{b})$.
\end{itemize}
Similar to the $(0,2)$ quiver gauge theory, circle nodes give gauge symmetries while square nodes give flavor symmetries.
The $(0,2)$ decomposition \eqref{eq:4susy2susycorrespondence} is then described as
\bea
\begin{tikzpicture}[decoration={markings,mark=at position \arrowHeadPosition with {\arrow{latex}}}]
 \tikzset{
        box/.style={draw, minimum width=0.6cm, minimum height=0.6cm, text centered,thick},
        ->-/.style={decoration={
        markings,mark=at position #1 with {\arrow[scale=1.5]{>}}},postaction={decorate},line width=0.5mm},
        -<-/.style={decoration={
        markings,
        mark=at position #1 with {\arrow[scale=1.5]{<}}},postaction={decorate},line width=0.5mm}    
    }
\begin{scope}{xshift=0cm}
    \draw[fill=black!20!white,thick](0,0) circle(0.3cm);
    \node[left]at (-0.3,0){$V^{(2,2)}$};
    \node[scale=2.0] at (1.4,0){$\rightsquigarrow$};
\end{scope}
\begin{scope}[xshift=3cm]
    \draw[fill=black!20!white,thick](0,0) circle(0.3cm);
    \node[right] at (0.3,0.0){$V^{(0,2)}$};
    \node[right] at (0.3,1.0){$\Sigma^{(0,2)}$};
    \chiralarc[postaction={decorate},thick](0,0.5)(-45:225:0.3:0.8)
\end{scope}

\begin{scope}{}
\draw[fill=black!20!white,thick](0.3,-2) circle(0.3cm);
\draw[fill=black!20!white,thick](-1.7,-2) circle(0.3cm);
\draw[postaction={decorate}, thick](-1.4,-2)--(0,-2);
\node[above] at (-0.6,-1.8){$\Phi^{(2,2)}$};
    \node[scale=2.0] at (1.4,-2){$\rightsquigarrow$};
\end{scope}

\begin{scope}{xshift=3cm}
\draw[fill=black!20!white,thick](4.7,-2) circle(0.3cm);
\draw[fill=black!20!white,thick](2.7,-2) circle(0.3cm);
\node[above] at (3.8,-1.8){$\Phi^{(0,2)}$};
\node[below] at (3.8,-2.2){$\Lambda^{(0,2)}$};
\draw[postaction={decorate}, thick](3,-1.9)--(4.4,-1.9);
\draw[postaction={decorate}, thick, red](3,-2.1)--(4.4,-2.1);
\end{scope}
\end{tikzpicture}
\eea

\subsection{Magnificent four}\label{sec:M4-2dtheory}
Let us consider the low energy field theory of $k$ D1-branes probing $n$ D9-branes wrapping the entire $\mathbb{C}^{4}$. Based on the discussion section~\ref{sec:M4-openstring}, we have a chiral superfield corresponding to the degrees of freedom coming from the open string connecting the D9-branes and D1-branes, which we denote $\Phi_{\four}=\mathsf{I}+\cdots$, where $\mathsf{I}$ is the first component field of the chiral superfield. As mentioned, we need to have the nontrivial $B$-field so that we have this chiral superfield.

The degrees of freedom coming from the 1-1 strings come from dimensional reduction of the 4d $\mathcal{N}=4$ SYM and we will have a vector superfield $\Upsilon$, four chiral superfields $\Phi_{a}=\mathsf{B}_{a}+\cdots$ for $a\in\four$, where $\mathsf{B}_{a}$ are the first component field and three Fermi superfields $\Lambda_{1,2,3}$. In the following discussions, we will also identify the chiral superfields with the first bosonic component fields. 

Since we have multiple D-branes now, the Chan-Paton factors will introduce a group structure and the D1-branes will give $\U(k)$ gauge symmetry. The D9-branes are heavy enough and thus the gauge symmetry is frozen and it will only induce a global flavor $\U(n)$ symmetry. The chiral superfield $\mathsf{I}$ then transforms under $(\overline{n},k)$, while the other chiral and Fermi superfields $\mathsf{B}_{a},\Lambda_{i}$ all transform in the adjoint representation of $\U(k)$. The matter components of this field theory are then summarized in the following quiver diagram:
\bea\label{eq:2SUSYquiver-magnificent}
\begin{tikzpicture}[decoration={markings,mark=at position \arrowHeadPosition with {\arrow{latex}}}]
 \tikzset{
        box/.style={draw, minimum width=0.7cm, minimum height=0.7cm, text centered,thick},
        ->-/.style={decoration={
        markings,mark=at position #1 with {\arrow[scale=1.5]{>}}},postaction={decorate},line width=0.5mm},
        -<-/.style={decoration={
        markings,
        mark=at position #1 with {\arrow[scale=1.5]{<}}},postaction={decorate},line width=0.5mm}    
    }
\begin{scope}[xshift=4cm]
    \draw[fill=black!10!white,thick](0,0) circle(0.4cm);
    \node at (0,0) {$k$};
    \node[box,fill=black!10!white] at (0,1.6) {$n$};
    \draw[postaction={decorate},thick] (0,1.25)--(0,0.4);
    \foreach \ang in {90,145,215,270} {
    \begin{scope}[rotate=\ang]
        \chiralarc[postaction={decorate},thick](0,0.5)(-45:225:0.22:0.65)
    \end{scope}
    }
    \foreach \ang in {90,145,270} {
    \begin{scope}[rotate=\ang]
    \fermiarc[postaction={decorate},thick](0,0.5)(-45:225:0.1:0.5)
    \end{scope}
    \node[right] at (0,0.8) {$\mathsf{I}$};
    \node[left] at (-1.5,0) {$\mathsf{B}_{2},\textcolor{red}{\Lambda_{2}}$};
    \node[right] at (1.6,0) {$\mathsf{B}_{1},\textcolor{red}{\Lambda_{1}}$};
    \node[below left] at (-0.9,-1){$\mathsf{B}_{3},\textcolor{red}{\Lambda_{3}}$};
    \node[below right] at (0.9,-1){$\mathsf{B}_{4}$};
    \draw[fill=black!10!white,thick](0,0) circle(0.4cm);
    \node at (0,0) {$k$};
    
    }
\end{scope}
\end{tikzpicture}
\eea
The $E, J$-terms are given by 
\bea\label{eq:2SUSYJEterm-magnificent}
    E_{i}=[\mathsf{B}_{4},\mathsf{B}_{i}]\quad J_{i}=\frac{1}{2}\varepsilon_{ijk4}[\mathsf{B}_{j},\mathsf{B}_{k}]
\eea
for $i=1,2,3$. The $E,J$-terms above are obtained by dimensional reduction of the world volume theory of the D1-branes with no D9-branes. The 2 SUSY condition \eqref{eq:JE-condition} can be checked as
\bea
\Tr\left(J_{1}E_{1}+J_{2}E_{2}+J_{3}E_{3}\right)&=\Tr\left([\mathsf{B}_{2},\mathsf{B}_{3}][\mathsf{B}_{4},\mathsf{B}_{1}]+[\mathsf{B}_{3},\mathsf{B}_{1}][\mathsf{B}_{4},\mathsf{B}_{2}]+[\mathsf{B}_{1},\mathsf{B}_{2}][\mathsf{B}_{4},\mathsf{B}_{3}]\right)\\
&=0.
\eea
We then can write down the Lagrangian of the low-energy worldvolume theory in terms of the 2d $\mathcal{N}=(0,2)$ superspace:
\bea
\mathcal{L}&=\int d\theta^{+} d\bar{\theta}^{+} \Tr\left(  \frac{1}{2g^{2}}\overline{\Upsilon}\Upsilon+i\sum_{a\in\four}\overline{\Phi}_{a}(\mathcal{D}_{0}-\mathcal{D}_{1}){\Phi}_{a}+\sum_{i=1}^{3}\overline{\Lambda}_{i}\Lambda_{i}  \right)\\
&+\left(\int d\theta^{+} \Tr \sum_{i=1}^{3}\Lambda_{i}J_{i}|_{\bar{\theta}^{+}=0}+\text{c.c} \right)+\left(\frac{ir}{2}\int d\theta^{+}\Upsilon|_{\bar{\theta}^{+}=0}+\text{c.c.}\right)\\
&+\int d\theta^{+} d\bar{\theta}^{+} i\overline{\Phi}_{\four}(\mathcal{D}_{0}-\mathcal{D}_{1}){\Phi}_{\four}
\eea
where $r$ is identified with the $B$-field as 
\bea
r=\sum_{a\in\four}|v_{a}|-1.
\eea
Integrating out the auxiliary fields, the scalar potential is given as
\bea
U_{\text{pot.}}=\frac{g^{2}}{2}\left(\sum_{a\in\four}[\mathsf{B}_{a},\mathsf{B}_{a}^{\dagger}]+\mathsf{I}\mathsf{I}^{\dagger}-r\right)^{2}+\sum_{(ab)\in\six}\Tr|[\mathsf{B}_{a},\mathsf{B}_{b}]|^{2},
\eea
which can be also read directly from \eqref{eq:scalarpotential-general}. Since $U_{\text{pot.}}\geq 0$, the ground state is always stable. When $r< 0$, the vacuum has a positive energy and will break supersymmetry. When $r=0$, it will preserve supersymmetry. When $r>0$, the original string theory is not supersymmetric but it will restore supersymmetry after tachyon condensation. The moduli space of this vacua of this theory is given as
\bea
\mathfrak{M}_{n,k}=\{(\vec{\mathsf{B}},\mathsf{I})\mid \mu_{D}-r=\mu_{A}=0\}/\U(k),\quad \vec{\mathsf{B}}=(\mathsf{B}_{a})_{a\in\four}
\eea
where
\bea
\mu_{D}=\sum_{a\in\four}[\mathsf{B}_{a},\mathsf{B}_{a}^{\dagger}]+\mathsf{I}\mathsf{I}^{\dagger},\quad \mu_{A}=[\mathsf{B}_{a},\mathsf{B}_{b}],\quad A=(ab)\in\six.
\eea
This is exactly the one introduced in \eqref{eq:M4ADHM1} and \eqref{eq:M4ADHM2} ($\mu_{D}\rightarrow \mu_{\mathbb{R}}, r\rightarrow \zeta$).

\subsection{Tetrahedron instanton}\label{sec:tetra-2dtheory}
Before moving on to the general case of tetrahedron instanton, let us focus on the case when we have $n_{\bar{4}}$ D7-branes wrapping the complex three-plane $\mathbb{C}^{3}_{123}$ and $k$ D0-branes probing them. This setup preserves 4 SUSY and the quiver and the superpotential of the 2d $\mathcal{N}=(2,2)$ theory are summarized as 
\bea\label{eq:4SUSYD7setup}
\adjustbox{valign=c}{
\begin{tikzpicture}[decoration={markings,mark=at position \arrowHeadPosition with {\arrow{latex}}}]
 \tikzset{
        box/.style={draw, minimum width=0.7cm, minimum height=0.7cm, text centered,thick},
        ->-/.style={decoration={
        markings,mark=at position #1 with {\arrow[scale=1.5]{>}}},postaction={decorate},line width=0.5mm},
        -<-/.style={decoration={
        markings,
        mark=at position #1 with {\arrow[scale=1.5]{<}}},postaction={decorate},line width=0.5mm}    
    }
\begin{scope}[xshift=4cm]
    \node[box,fill=black!10!white] at (0,1.6) {$n_{\bar{4}}$};
    \draw[postaction={decorate},thick] (0,1.25)--(0,0.4);
    \foreach \ang in {125,180,235} {
    \begin{scope}[rotate=\ang]
        \chiralarc[postaction={decorate},thick](0,0.5)(-45:225:0.22:0.65)
    \end{scope}
    }
    \node[right] at (0,0.8) {$\mathsf{I}$};
    \node[] at (-1.5,-0.9) {$\mathsf{B}_{1}$};
    \node[below] at (0,-1.6) {$\mathsf{B}_{2}$};
    \node[ ] at (1.5,-0.9){$\mathsf{B}_{3}$};
    \draw[fill=black!10!white,thick](0,0) circle(0.4cm);
    \node at (0,0) {$k$};
\end{scope}
\end{tikzpicture}}\qquad  
\mathsf{W}_{0}=\Tr\left(\mathsf{B}_{1}[\mathsf{B}_{2},\mathsf{B}_{3}]\right).
\eea
In the 4 SUSY notation ($\mathcal{N}=(2,2)$ superfields), the matter components coming from the 1-1 strings are a vector superfield, three adjoint chiral superfields $\Phi^{(2,2)}_{1,2,3}=\mathsf{B}_{1,2,3}+\cdots$. As discussed in section~\ref{sec:tetra-opstring}, the 1-7 strings give a chiral superfield $\Phi^{(2,2)}_{\bar{4}}=\mathsf{I}+\cdots$ in the bifundamental representation $(\overline{n}_{\bar{4}},k)$. Similarly, we identify the superfields with the first component field. The superpotential comes from the dimensional reduction of the 4d $\mathcal{N}=4$ SYM. The chiral superfield $\mathsf{I}$ gives no extra term because there is no other gauge invariant term.

Decomposing into the 2 SUSY convention gives the following quiver and $E,J$-terms:
\bea\label{eq:2susyD7onestack}
\adjustbox{valign=c}{
\begin{tikzpicture}[decoration={markings,mark=at position \arrowHeadPosition with {\arrow{latex}}}]
 \tikzset{
        box/.style={draw, minimum width=0.7cm, minimum height=0.7cm, text centered,thick},
        ->-/.style={decoration={
        markings,mark=at position #1 with {\arrow[scale=1.5]{>}}},postaction={decorate},line width=0.5mm},
        -<-/.style={decoration={
        markings,
        mark=at position #1 with {\arrow[scale=1.5]{<}}},postaction={decorate},line width=0.5mm}    
    }
\begin{scope}[xshift=4cm]
    \draw[fill=black!10!white,thick](0,0) circle(0.4cm);
    \node at (0,0) {$k$};
    \node[box,fill=black!10!white] at (0,1.6) {$n_{\bar{4}}$};
    \draw[postaction={decorate},thick] (-0.1,1.25)--(-0.1,0.4);
    \draw[postaction={decorate},red,thick] (0.1,1.25)--(0.1,0.4);
    \foreach \ang in {90,145,215,270} {
    \begin{scope}[rotate=\ang]
        \chiralarc[postaction={decorate},thick](0,0.5)(-45:225:0.22:0.65)
    \end{scope}
    }
    \foreach \ang in {90,145,270} {
    \begin{scope}[rotate=\ang]
    \fermiarc[postaction={decorate},thick](0,0.5)(-45:225:0.1:0.5)
    \end{scope}
    \node[left] at (-0.1,0.8) {$\mathsf{I}$};
    \node[right] at (0.1,0.8) {$\textcolor{red}{\Lambda_{\mathsf{I}}}$};
    \node[left] at (-1.5,0) {$\mathsf{B}_{2},\textcolor{red}{\Lambda_{2}}$};
    \node[right] at (1.6,0) {$\mathsf{B}_{1},\textcolor{red}{\Lambda_{1}}$};
    \node[below left] at (-0.9,-1){$\mathsf{B}_{3},\textcolor{red}{\Lambda_{3}}$};
    \node[below right] at (0.9,-1){$\mathsf{B}_{4}$};
    \draw[fill=black!10!white,thick](0,0) circle(0.4cm);
    \node at (0,0) {$k$};
    
    }
\end{scope}
\end{tikzpicture}}\qquad 
\begin{array}{l}
    E_{i}=[\mathsf{B}_{4},\mathsf{B}_{i}],\quad E_{\mathsf{I}}=\mathsf{B}_{4}\,\mathsf{I}   \\
    J_{i}=\partial \mathsf{W}_{0}/\partial \mathsf{B}_{i}=\frac{1}{2}\varepsilon_{ijk4}[\mathsf{B}_{j},\mathsf{B}_{k}]
\end{array}
\eea
As discussed in \eqref{eq:4susy2susycorrespondence}, the $\mathcal{N}=(2,2)$ chiral superfields decompose into $\mathcal{N}=(0,2)$ chiral and Fermi superfields. For this case, we have $\Phi^{(2,2)}_{i}\rightarrow \Phi_{i}^{(0,2)}=\mathsf{B}_{i}+\cdots,\Lambda_{i}\,(i=1,2,3)$ and $\Phi_{\bar{4}}^{(2,2)}\rightarrow \Phi_{\bar{4}}^{(0,2)}=\mathsf{I}+\cdots ,\,\,\Lambda_{\mathsf{I}}$. The $\mathcal{N}=(2,2)$ vector superfield decomposes into a $\mathcal{N}=(0,2)$ vector superfield and a chiral superfield $\Phi^{(0,2)}_{4}=\mathsf{B}_{4}+\cdots$. The $E$-terms are derived by using \eqref{eq:4susyEterm}, \eqref{eq:4susy2susycorrespondence}. The $J$-terms are obtained by using \eqref{eq:4susyJterm}, \eqref{eq:4susy2susycorrespondence}. Note that the traceless condition \eqref{eq:JE-condition} is similar to the magnificent four case and indeed we have $\Tr(\sum_{i=1}^{3}J_{i}E_{i})=0$.


\paragraph{Tetrahedron instanton case}Let us generalize the story to the tetrahedron instanton case. We have $n_{\bar{a}}$ D7$_{\bar{a}}$-branes and $k$ D1-branes. For each $\D7_{\bar{a}}\,\,(a\in\four)$, the 1-7 strings give a chiral superfield $\Phi_{\bar{a}}=\mathsf{I}_{\bar{a}}+\cdots$ and a Fermi superfield $\Lambda_{\bar{a}}$ transforming in the bifundamental representation of $\U(n_{\bar{a}})\times\U(k)$. The matter components coming from the 1-1 strings are one vector superfield $\Upsilon$, four adjoint chiral superfields $\Phi_{a}$, and three Fermi supefields $\Lambda_{i}$, where we used the same notation with the magnificent four case. The matter components and the $J,E$-terms are summarized as
\bea\label{eq:2SUSYquiver-tetrahedron}
\adjustbox{valign=c}{
\begin{tikzpicture}[decoration={markings,mark=at position \arrowHeadPosition with {\arrow{latex}}}]
    \tikzset{
        cir/.style={circle,fill=black!10!white, draw, minimum size=0.8cm, text centered, thick},
        point/.style={circle, fill, inner sep=1.5pt},
        box/.style={draw, fill=black!10!white,minimum width=0.7cm, minimum height=0.7cm, text centered,thick},
        ->-/.style={decoration={
        markings,mark=at position #1 with {\arrow[scale=1.5]{>}}},postaction={decorate},line width=0.5mm},
        -<-/.style={decoration={
        markings,
        mark=at position #1 with {\arrow[scale=1.5]{<}}},postaction={decorate},line width=0.5mm}    
    }

    \node[cir] (i1) at ($(0,0)$) {$k$};
    \begin{scope}[rotate=0]
    \node[box,rotate=0] (n13) at ($(0,-2)$) {\rotatebox{-0}{$n_{\bar{4}}$}};

    \draw[postaction={decorate}, thick] ($(0.06,-1.65)$) -- ($(0.06,-0.4)$);
    \draw[postaction={decorate},red, thick] ($(-0.06,-1.65)$) -- ($(-0.06,-0.4)$);
    \end{scope}

    \begin{scope}[rotate=90]
    \node[box,rotate=90] (n12) at  ($(0,-2)$) {\rotatebox{-90}{$n_{\bar{1}}$}};
    \draw[postaction={decorate}, thick] ($(0.06,-1.65)$) -- ($(0.06,-0.4)$);
    \draw[postaction={decorate},red, thick] ($(-0.06,-1.65)$) -- ($(-0.06,-0.4)$);
    \end{scope}
    
    \begin{scope}[rotate=180]
    \node[box,rotate=180] (n12) at  ($(0,-2)$) {\rotatebox{-180}{$n_{\bar{2}}$}};
    
    \draw[postaction={decorate}, thick] ($(0.06,-1.65)$) -- ($(0.06,-0.4)$);
    \draw[postaction={decorate},red, thick] ($(-0.06,-1.65)$) -- ($(-0.06,-0.4)$);
    \end{scope}
    
    \begin{scope}[rotate=270]
    \node[box,rotate=270] (n23) at  ($(0,-2)$) {\rotatebox{-270}{$n_{\bar{3}}$}};
    
    \draw[postaction={decorate}, thick] ($(0.06,-1.65)$) -- ($(0.06,-0.4)$);
    \draw[postaction={decorate},red, thick] ($(-0.06,-1.65)$) -- ($(-0.06,-0.4)$);
    \end{scope}

    \begin{scope}[rotate=45]
    \chiralarc[postaction={decorate},thick](0,0.55)(-45:225:0.22:0.65)
    \fermiarc[postaction={decorate},thick](0,0.5)(-45:225:0.1:0.5)
    \end{scope}
    \begin{scope}[rotate=135]
    \chiralarc[postaction={decorate},thick](0,0.55)(-45:225:0.22:0.65)
    \fermiarc[postaction={decorate},thick](0,0.5)(-45:225:0.1:0.5)
    \end{scope}
    \begin{scope}[rotate=225]
    \chiralarc[postaction={decorate},thick](0,0.55)(-45:225:0.22:0.65)
    \end{scope}
    \begin{scope}[rotate=315]
    \chiralarc[postaction={decorate},thick](0,0.55)(-45:225:0.22:0.65)
    \fermiarc[postaction={decorate},thick](0,0.5)(-45:225:0.1:0.5)
    \end{scope}    
\end{tikzpicture}} \qquad 
\begin{array}{l}
   E_{i}=[\mathsf{B}_{4},\mathsf{B}_{i}],\quad E_{\mathsf{I}_{\bar{a}}}=\mathsf{B}_{a}\,\mathsf{I}_{\bar{a}}\\
   J_{i}=\frac{1}{2}\varepsilon_{ijk4}[\mathsf{B}_{j},\mathsf{B}_{k}]
\end{array}
\eea
We chose a notation so that the $E,J$-terms in \eqref{eq:2susyD7onestack} when we have only one stack of D7$_{\bar{4}}$-branes setup naturally arise. Similarly, it is obvious that we have $\Tr(\sum_{i=1}^{3}J_{i}E_{i})=0$.

The low energy effective theory of the D1-branes is given as
\bea
\mathcal{L}&=\int d\theta^{+} d\bar{\theta}^{+} \Tr\left(  \frac{1}{2g^{2}}\overline{\Upsilon}\Upsilon+i\sum_{a\in\four}\overline{\Phi}_{a}(\mathcal{D}_{0}-\mathcal{D}_{1}){\Phi}_{a}+\sum_{i=1}^{3}\overline{\Lambda}_{i}\Lambda_{i}  \right)+\left(\int d\theta^{+} \Tr \sum_{i=1}^{3}\Lambda_{i}J_{i}|_{\bar{\theta}^{+}=0}+\text{c.c} \right)\\
&+\int d\theta^{+} d\bar{\theta}^{+} \left(i\sum_{a\in\four}\overline{\Phi}_{\bar{a}}(\mathcal{D}_{0}-\mathcal{D}_{1}){\Phi}_{\bar{a}}+\sum_{a\in\four}\overbar{\Lambda}_{\bar{a}}\Lambda_{\bar{a}} \right)+\left(\frac{ir}{2}\int d\theta^{+}\Upsilon|_{\bar{\theta}^{+}=0}+\text{c.c.}\right).
\eea
The FI parameter $r$ here is identified with the $B$-fields as
\bea
v_{1}=v_{2}=v_{4}=v_{4}=\frac{1}{6}+\frac{r}{3}.
\eea
We omit the reason why this should be in this way so see \cite{Pomoni:2021hkn}. Integrating out the auxiliary fields, the scalar potential is given as
\bea
U_{\text{pot.}}=\frac{g^{2}}{2}\left(\sum_{a\in\four}[\mathsf{B}_{a},\mathsf{B}_{a}^{\dagger}]+\sum_{a\in\four}\mathsf{I}_{\bar{a}}\mathsf{I}_{\bar{a}}^{\dagger}-r\right)^{2}+\sum_{(ab)\in\six}\Tr|[\mathsf{B}_{a},\mathsf{B}_{b}]|^{2}+\sum_{a\in\four}\Tr|\mathsf{B}_{a}\mathsf{I}_{\bar{a}}|^{2}.
\eea
The moduli space of this vacua is
\bea
\mathfrak{M}_{\vec{n},k}=\{(\vec{\mathsf{B}},\vec{\mathsf{I}}\,\,)\mid \mu_{D}-r=\mu_{A}=\mu_{a}=0\}/\U(k),\quad \vec{\mathsf{B}}=(\mathsf{B}_{a})_{a\in\four},\quad \vec{\mathsf{I}}=(\mathsf{I}_{\bar{a}})_{a\in\four}
\eea
where
\bea
\mu_{D}=\sum_{a\in\four}[\mathsf{B}_{a},\mathsf{B}_{a}^{\dagger}]+\sum_{a\in\four}\mathsf{I}_{\bar{a}}\mathsf{I}_{\bar{a}}^{\dagger},\quad \mu_{A}=[\mathsf{B}_{a},\mathsf{B}_{b}],\quad \mu_{a}&=\mathsf{B}_{a}\mathsf{I}_{\bar{a}}
\eea
for $A=(ab)\in\six$ and $a\in\four$. This coincides with \eqref{eq:tetraADHM}.

\subsection{Spiked instanton }
Let us consider the low energy effective field theory of the $k$ D1-branes probing $n_{12}$ D5-branes wrapping the complex two plane $\mathbb{C}^{2}_{12}$. As mentioned in secdtion~\ref{sec:spiked-susycount}, this setup preserves 8 SUSY. Using the 4 SUSY notation, the matter components are summarized as
\bea
\adjustbox{valign=c}{
\begin{tikzpicture}[decoration={markings,mark=at position \arrowHeadPosition with {\arrow{latex}}}]
 \tikzset{
        box/.style={draw, minimum width=0.7cm, minimum height=0.7cm, text centered,thick},
        ->-/.style={decoration={
        markings,mark=at position #1 with {\arrow[scale=1.5]{>}}},postaction={decorate},line width=0.5mm},
        -<-/.style={decoration={
        markings,
        mark=at position #1 with {\arrow[scale=1.5]{<}}},postaction={decorate},line width=0.5mm}    
    }
\begin{scope}[xshift=4cm]
    \node[box,fill=black!10!white] at (0,1.6) {$n_{12}$};
    \draw[postaction={decorate},thick] (-0.1,1.25)--(-0.1,0.4);
    \draw[postaction={decorate},thick] (0.1,0.4)--(0.1,1.25);
    \foreach \ang in {125,180,235} {
    \begin{scope}[rotate=\ang]
        \chiralarc[postaction={decorate},thick](0,0.5)(-45:225:0.22:0.65)
    \end{scope}
    }
    \node[left] at (-0.1,0.8) {$\mathsf{I}$};
    \node[right] at (0.1,0.8) {$\mathsf{J}$};
    \node[] at (-1.5,-0.9) {$\mathsf{B}_{1}$};
    \node[below] at (0,-1.6) {$\mathsf{B}_{2}$};
    \node[ ] at (1.5,-0.9){$\mathsf{B}_{3}$};
    \draw[fill=black!10!white,thick](0,0) circle(0.4cm);
    \node at (0,0) {$k$};
\end{scope}
\end{tikzpicture}}\qquad\qquad   
\begin{array}{l}
\mathsf{W}=\mathsf{W}_{0}+\mathsf{W}_{F}\\
\mathsf{W}_{F}=\Tr\left(\mathsf{J}\,\mathsf{B}_{3}\,\mathsf{I}\right)
\end{array}
\eea
The matter components coming from the 1-1 strings are the same as \eqref{eq:4SUSYD7setup}: a $\mathcal{N}=(2,2)$ vector superfield and three adjoint chiral superfields $\Phi^{(2,2)}_{1,2,3}=\mathsf{B}_{1,2,3}+\cdots$. In this case, the 1-5 strings give two chiral superfields $\Phi_{12}^{(2,2)}=\mathsf{I}+\cdots$ and $\widetilde{\Phi}_{12}^{(2,2)}=\mathsf{J}+\cdots$ transforming in the bifundamental representations $(\overbar{n}_{12},k)$ and $(n_{12},\bar{k})$, respectively. In this case, the superpotential is modified by an extra term $\mathsf{W}_{F}$ which is obviously a gauge invariant term and thus can be added to the superpotential.

In the 2 SUSY notation, we have
\bea
\adjustbox{valign=c}{
\begin{tikzpicture}[decoration={markings,mark=at position \arrowHeadPosition with {\arrow{latex}}}]
 \tikzset{
        box/.style={draw, minimum width=0.7cm, minimum height=0.7cm, text centered,thick},
        ->-/.style={decoration={
        markings,mark=at position #1 with {\arrow[scale=1.5]{>}}},postaction={decorate},line width=0.5mm},
        -<-/.style={decoration={
        markings,
        mark=at position #1 with {\arrow[scale=1.5]{<}}},postaction={decorate},line width=0.5mm}    
    }
\begin{scope}[xshift=4cm]
    \draw[fill=black!10!white,thick](0,0) circle(0.4cm);
    \node at (0,0) {$k$};
    \node[box,fill=black!10!white] at (0,1.6) {$n_{12}$};
    \draw[postaction={decorate},thick] (-0.2,1.25)--(-0.2,0.3);
    \draw[postaction={decorate},red,thick] (-0.07,1.25)--(-0.07,0.4);
    \draw[postaction={decorate},thick] (0.2,0.3)--(0.2,1.25);
    \draw[postaction={decorate},red,thick] (0.07,0.4)--(0.07,1.25);
    \foreach \ang in {90,145,215,270} {
    \begin{scope}[rotate=\ang]
        \chiralarc[postaction={decorate},thick](0,0.5)(-45:225:0.22:0.65)
    \end{scope}
    }
    \foreach \ang in {90,145,270} {
    \begin{scope}[rotate=\ang]
    \fermiarc[postaction={decorate},thick](0,0.5)(-45:225:0.1:0.5)
    \end{scope}}
    \node[left] at (-0.2,0.8) {$\mathsf{I},\textcolor{red}{\Lambda_{\mathsf{I}}}$};
    \node[right] at (0.2,0.8) {$\mathsf{J},\textcolor{red}{\Lambda_{\mathsf{J}}}$};
    \node[left] at (-1.5,0) {$\mathsf{B}_{2},\textcolor{red}{\Lambda_{2}}$};
    \node[right] at (1.6,0) {$\mathsf{B}_{1},\textcolor{red}{\Lambda_{1}}$};
    \node[below left] at (-0.9,-1){$\mathsf{B}_{3},\textcolor{red}{\Lambda_{3}}$};
    \node[below right] at (0.9,-1){$\mathsf{B}_{4}$};
    \draw[fill=black!10!white,thick](0,0) circle(0.4cm);
    \node at (0,0) {$k$};
\end{scope}
\end{tikzpicture}} \quad  \begin{array}{l}
    E_{i}=[\mathsf{B}_{4},\mathsf{B}_{i}],\quad E_{\mathsf{I}}=\mathsf{B}_{4}\mathsf{I},\quad E_{\mathsf{J}}=-\mathsf{J}\mathsf{B}_{4},\\
    J_{1}=[\mathsf{B}_{2},\mathsf{B}_{3}],\quad J_{2}=[\mathsf{B}_{3},\mathsf{B}_{1}],\\
     J_{3}=[\mathsf{B}_{1},\mathsf{B}_{2}]+\mathsf{I}\mathsf{J},\quad J_{\mathsf{I}}=\mathsf{J}\mathsf{B}_{3},\quad J_{\mathsf{J}}=\mathsf{B}_{3}\mathsf{I}
\end{array}
\eea
where the $\mathcal{N}=(2,2)$ chiral superfields are decomposed into $\mathcal{N}=(0,2)$ chiral superfields and Fermi superfields which are denoted as $\Phi_{12}=\mathsf{I}+\cdots,\widetilde{\Phi}_{12}=\mathsf{J}+\cdots$ and $\Lambda_{\mathsf{I}},\Lambda_{\mathsf{J}}$, respectively. Other components are the same with the tetrahedron instanton and the magnificent four cases. 

The $E,J$-terms are obtained by using the relation \eqref{eq:4susyJterm} and \eqref{eq:4susyEterm}. The existence of the extra superpotential term $\mathsf{W}_{F}$ modifies the $J$-term coupled with the Fermi superfield $\Lambda_{3}$ because of 
\bea
J_{3}=\frac{\partial \mathsf{W}}{\partial \mathsf{B}_{3}}=\frac{\partial \mathsf{W}_{0}}{\partial \mathsf{B}_{3}}+\frac{\partial \mathsf{W}_{F}}{\partial \mathsf{B}_{3}}=[\mathsf{B}_{1},\mathsf{B}_{2}]+\mathsf{I}\mathsf{J}.
\eea
Indeed, the traceless condition \eqref{eq:JE-condition} is confirmed as
\bea
\Tr\left(\sum_{i=1}^{3}J_{i}E_{i}\right)&=\Tr\left(\sum_{i=1}^{3}\varepsilon_{ijk4}[\mathsf{B}_{j},\mathsf{B}_{k}][\mathsf{B}_{4},\mathsf{B}_{i}]+\mathsf{I}\mathsf{J}[\mathsf{B}_{4},\mathsf{B}_{3}]\right)&=\Tr\left(\mathsf{I}\mathsf{J}[\mathsf{B}_{4},\mathsf{B}_{3}]\right),\\
\Tr\left(J_{\mathsf{I}}\,E_{\mathsf{I}}\right)&=\Tr\left(\mathsf{J}\mathsf{B}_{3}\mathsf{B}_{4}\mathsf{I}\right),\quad \Tr\left(J_{\mathsf{J}}E_{\mathsf{J}}\right)=\Tr\left(-\mathsf{B}_{3}\mathsf{I}\mathsf{J}\mathsf{B}_{4}\right)
\eea
which gives $\Tr\left(\sum_{i=1}^{3}J_{i}E_{i}+J_{\mathsf{I}}E_{\mathsf{I}}+J_{\mathsf{J}}E_{\mathsf{J}}\right)=0$.


\paragraph{4 SUSY setup} Let us consider a spiked instanton setup with D5-branes wrapping complex two-planes inside $\mathbb{C}^{3}_{123}$. We put $n_{A}$ D5-branes wrapping the complex two-plane $\mathbb{C}^{2}_{A}$ for $A\in\{12,13,23\}$ and $k$ D1-branes probing this D5-branes. This setup preserves four supersymmetries and the quiver diagram is modified to 
\bea
\begin{tikzpicture}[decoration={markings,mark=at position \arrowHeadPosition with {\arrow{latex}}}]
 \tikzset{
        box/.style={draw, minimum width=0.7cm, minimum height=0.7cm, text centered,thick},
        ->-/.style={decoration={
        markings,mark=at position #1 with {\arrow[scale=1.5]{>}}},postaction={decorate},line width=0.5mm},
        -<-/.style={decoration={
        markings,
        mark=at position #1 with {\arrow[scale=1.5]{<}}},postaction={decorate},line width=0.5mm}    
    }
\begin{scope}[xshift=4cm]
    \begin{scope}[]
    \node[box,fill=black!10!white] at (0,1.6) {$n_{12}$};
    \draw[postaction={decorate},thick] (-0.1,1.25)--(-0.1,0.4);
    \draw[postaction={decorate},thick] (0.1,0.4)--(0.1,1.25);
    \node[left] at (-0.1,0.8) {$\mathsf{I}_{12}$};
    \node[right] at (0.06,0.8) {$\mathsf{J}_{12}$};
    \end{scope}
    \begin{scope}[rotate=120]
    \draw[postaction={decorate},thick] (-0.1,1.25)--(-0.1,0.4);
    \draw[postaction={decorate},thick] (0.1,0.4)--(0.1,1.25);
    \node[rotate=120,left] at (-0.1,0.8) {\rotatebox{-120}{$\mathsf{I}_{13}$}};
    \node[rotate=120,right] at (0.06,0.8) {\rotatebox{-120}{$\mathsf{J}_{13}$}};
    \node[box,rotate=120,fill=black!10!white] at (0,1.6) {\rotatebox{-120}{$n_{13}$}};
    \end{scope}

    \begin{scope}[rotate=240]
    \draw[postaction={decorate},thick] (-0.1,1.25)--(-0.1,0.4);
    \draw[postaction={decorate},thick] (0.1,0.4)--(0.1,1.25);
    \node[rotate=240,left] at (-0.1,0.8) {\rotatebox{-240}{$\mathsf{I}_{23}$}};
    \node[rotate=240,right] at (0.06,0.8) {\rotatebox{-240}{$\mathsf{J}_{23}$}};
    \node[box,rotate=240,fill=black!10!white] at (0,1.6) {\rotatebox{-240}{$n_{23}$}};
    \end{scope}
    \foreach \ang in {60,180,300} {
    \begin{scope}[rotate=\ang]
        \chiralarc[postaction={decorate},thick](0,0.5)(-45:225:0.22:0.65)
    \end{scope}
    }
    \node[] at (-1.5,0.9) {$\mathsf{B}_{1}$};
    \node[below] at (0,-1.6) {$\mathsf{B}_{2}$};
    \node[ ] at (1.5,0.9){$\mathsf{B}_{3}$};
    \draw[fill=black!10!white,thick](0,0) circle(0.4cm);
    \node at (0,0) {$k$};
\end{scope}
\end{tikzpicture}
\eea
For each stack of D5$_{A}$-branes ($A\in\{12,13,23\}$), the 1-5 strings give $\mathcal{N}=(2,2)$ chiral superfields $\Phi_{A}^{(2,2)}=\mathsf{I}_{A}+\cdots$ and $\widetilde{\Phi}^{(2,2)}_{A}=\mathsf{J}_{A}+\cdots$. The matter components of the 1-1 strings are the same as the previous setup. Each stack of D5$_{A}$-branes for $A\in\{12,13,23\}$ introduces an extra term to the superpotential similar to the previous case and we have
\bea
\mathsf{W}_{F}=\Tr(\mathsf{J}_{23}\mathsf{B}_{1}\mathsf{I}_{23})+\Tr(\mathsf{J}_{13}\mathsf{B}_{2}\mathsf{I}_{13})+\Tr(\mathsf{J}_{12}\mathsf{B}_{3}\mathsf{I}_{12}).
\eea

In the 2 SUSY convention, the $E,J$-terms are determined as
\bea
E_{1,2,3}&=[\mathsf{B}_{4},\mathsf{B}_{1,2,3}],\quad E_{\mathsf{I}_{A}}=\mathsf{B}_{4}\mathsf{I}_{A},\quad E_{\mathsf{J}_{A}}=-\mathsf{J}_{A}\mathsf{B}_{4},\quad A\in\{12,23,13\}\\
J_{i}&=\frac{\partial{\mathsf{W}}}{\partial{\mathsf{B}_{i}}}=\varepsilon_{ijk4}\mathsf{B}_{j}\mathsf{B}_{k}+\mathsf{I}_{\overbar{i4}}\mathsf{J}_{\overbar{i4}},\quad J_{\mathsf{I}_{\overbar{i4}}}=\mathsf{J}_{\overbar{i4}}\mathsf{B}_{i},\quad
J_{\mathsf{J}_{\overbar{i4}}}=\mathsf{B}_{i}\mathsf{I}_{\overbar{i4}}
\eea
for $i=1,2,3$. The explicit 2 SUSY quiver diagram of this setup is obtained as a specialization of \eqref{eq:2SUSYquiver-spikedinstantonn} by setting $n_{14}=n_{24}=n_{34}=0$ there.

\paragraph{Spiked instanton case}Let us consider the most generic setup where we have six stacks of $n_{A}$ D5$_{A}$-branes wrapping the complex two-planes $\mathbb{C}^{2}_{A}$ for $A\in\six$. The 1-5$_{A}$ strings give $\mathcal{N}=(0,2)$ chiral superfields $\Phi_{A}=\mathsf{I}_{A}+\cdots $, $\widetilde{\Phi}_{A}=\mathsf{J}_{A}+\cdots $ and Fermi superfields $\Lambda_{A}=\Lambda_{\mathsf{I}_{A}},\widetilde{\Lambda}_{A}=\Lambda_{\mathsf{J}_{A}}$. The components coming from 1-1 strings are the same as the previous setups. They are summarized as the following quiver diagram:
\bea\label{eq:2SUSYquiver-spikedinstantonn}
\begin{tikzpicture}[decoration={markings,mark=at position \arrowHeadPosition with {\arrow{latex}}}]
 \tikzset{
        box/.style={draw, minimum width=0.8cm, minimum height=0.8cm, text centered,thick},
        ->-/.style={decoration={
        markings,mark=at position #1 with {\arrow[scale=1.5]{>}}},postaction={decorate},line width=0.5mm},
        -<-/.style={decoration={
        markings,
        mark=at position #1 with {\arrow[scale=1.5]{<}}},postaction={decorate},line width=0.5mm}    
    }
\begin{scope}[xshift=4cm]
    \draw[fill=black!10!white,thick](0,0) circle(0.4cm);
    \node at (0,0) {$k$};
    \begin{scope}[rotate=36]
    \node[rotate=36,box,fill=black!10!white] at (0,1.8) {\rotatebox{-36}{$n_{12}$}};
    \draw[postaction={decorate},thick] (-0.15,1.4)--(-0.15,0.3);
    \draw[postaction={decorate},red,thick] (-0.05,1.4)--(-0.05,0.4);
    \draw[postaction={decorate},thick] (0.15,0.3)--(0.15,1.4);
    \draw[postaction={decorate},red,thick] (0.05,0.4)--(0.05,1.4);
    \end{scope}
    \begin{scope}[rotate=72]
    \node[rotate=72,box,fill=black!10!white] at (0,1.8) {\rotatebox{-72}{$n_{13}$}};
    \draw[postaction={decorate},thick] (-0.15,1.4)--(-0.15,0.3);
    \draw[postaction={decorate},red,thick] (-0.05,1.4)--(-0.05,0.4);
    \draw[postaction={decorate},thick] (0.15,0.3)--(0.15,1.4);
    \draw[postaction={decorate},red,thick] (0.05,0.4)--(0.05,1.4);
    \end{scope}
    
    \begin{scope}[rotate=144]
    \node[rotate=144,box,fill=black!10!white] at (0,1.8) {\rotatebox{-144}{$n_{14}$}};
    \draw[postaction={decorate},thick] (-0.15,1.4)--(-0.15,0.3);
    \draw[postaction={decorate},red,thick] (-0.05,1.4)--(-0.05,0.4);
    \draw[postaction={decorate},thick] (0.15,0.3)--(0.15,1.4);
    \draw[postaction={decorate},red,thick] (0.05,0.4)--(0.05,1.4);
    \end{scope}
    \begin{scope}[rotate=216]
    \node[rotate=216,box,fill=black!10!white] at (0,1.8) {\rotatebox{-216}{$n_{23}$}};
    \draw[postaction={decorate},thick] (-0.15,1.4)--(-0.15,0.3);
    \draw[postaction={decorate},red,thick] (-0.05,1.4)--(-0.05,0.4);
    \draw[postaction={decorate},thick] (0.15,0.3)--(0.15,1.4);
    \draw[postaction={decorate},red,thick] (0.05,0.4)--(0.05,1.4);
    \end{scope}
    \begin{scope}[rotate=288]
    \node[rotate=288,box,fill=black!10!white] at (0,1.8) {\rotatebox{-288}{$n_{24}$}};
    \draw[postaction={decorate},thick] (-0.15,1.4)--(-0.15,0.3);
    \draw[postaction={decorate},red,thick] (-0.05,1.4)--(-0.05,0.4);
    \draw[postaction={decorate},thick] (0.15,0.3)--(0.15,1.4);
    \draw[postaction={decorate},red,thick] (0.05,0.4)--(0.05,1.4);
    \end{scope}
    \begin{scope}[rotate=-36]
    \node[rotate=-36,box,fill=black!10!white] at (0,1.8) {\rotatebox{36}{$n_{34}$}};
    \draw[postaction={decorate},thick] (-0.15,1.4)--(-0.15,0.3);
    \draw[postaction={decorate},red,thick] (-0.05,1.4)--(-0.05,0.4);
    \draw[postaction={decorate},thick] (0.15,0.3)--(0.15,1.4);
    \draw[postaction={decorate},red,thick] (0.05,0.4)--(0.05,1.4);
    \end{scope}
    \foreach \ang in {0,108,252,180} {
    \begin{scope}[rotate=\ang]
        \chiralarc[postaction={decorate},thick](0,0.7)(-45:225:0.25:1.0)
    \end{scope}
    }
    \foreach \ang in {108,180,252} {
    \begin{scope}[rotate=\ang]
    \fermiarc[postaction={decorate},thick](0,0.6)(-45:225:0.1:0.7)
    \end{scope}}
    \draw[fill=black!10!white,thick](0,0) circle(0.5cm);
    \node at (0,0) {$k$};
\end{scope}
\end{tikzpicture}
\eea
The black arrows $\U(n_{A})\rightarrow \U(k)$ and $\U(k)\rightarrow \U(n_{A})$ correspond to $\mathsf{I}_{A},\mathsf{J}_{A}$, respectively and the red arrows $\U(n_{A})\textcolor{red}{\rightarrow} \U(k)$ and $\U(k)\textcolor{red}{\rightarrow} \U(n_{A})$ correspond to $\Lambda_{\mathsf{I}_{A}},\Lambda_{\mathsf{J}_{A}}$, respectively. The $J$ and $E$-terms of this setup is given as 
\bea\label{eq:2SUSYJEterm-spiked}
E_{i}=[\mathsf{B}_{4},\mathsf{B}_{i}]+\mathsf{I}_{i4}\mathsf{J}_{i4},&\quad J_{i}=\varepsilon_{ijk4}\mathsf{B}_{j}\mathsf{B}_{k}+\mathsf{I}_{\overbar{i4}}\mathsf{J}_{\overbar{i4}},\quad E_{\mathsf{I}_{12}}=\mathsf{B}_{4}\mathsf{I}_{12},\quad J_{\mathsf{I}_{12}}=\mathsf{J}_{12}\mathsf{B}_{3},\\
E_{\mathsf{J}_{12}}=-\mathsf{J}_{12}\mathsf{B}_{4},&\quad J_{\mathsf{J}_{12}}=\mathsf{B}_{3}\mathsf{I}_{12},\quad
E_{\mathsf{I}_{13}}=\mathsf{B}_{4}\mathsf{I}_{13},\quad J_{\mathsf{I}_{13}}=\mathsf{J}_{13}\mathsf{B}_{2},\\ E_{\mathsf{J}_{13}}=-\mathsf{J}_{13}\mathsf{B}_{4},&\quad  J_{\mathsf{J}_{13}}=\mathsf{B}_{2}\mathsf{I}_{13},\quad 
E_{\mathsf{I}_{23}}=\mathsf{B}_{4}\mathsf{I}_{23},\quad J_{\mathsf{I}_{23}}=\mathsf{J}_{23}\mathsf{B}_{1},\\
E_{\mathsf{J}_{23}}=-\mathsf{J}_{23}\mathsf{B}_{4},&\quad J_{\mathsf{J}_{23}}=\mathsf{B}_{1}\mathsf{I}_{23},\quad 
E_{\mathsf{I}_{14}}=\mathsf{B}_{2}\mathsf{I}_{14},\quad J_{\mathsf{I}_{14}}=\mathsf{J}_{14}\mathsf{B}_{3},\\
E_{\mathsf{J}_{14}}=-\mathsf{J}_{14}\mathsf{B}_{2},&\quad J_{\mathsf{J}_{14}}=\mathsf{B}_{3}\mathsf{I}_{14},\quad 
E_{\mathsf{I}_{24}}=\mathsf{B}_{3}\mathsf{I}_{24},\quad J_{\mathsf{I}_{24}}=\mathsf{J}_{24}\mathsf{B}_{1},\\
E_{\mathsf{J}_{24}}=-\mathsf{J}_{24}\mathsf{B}_{3},&\quad J_{\mathsf{J}_{24}}=\mathsf{B}_{1}\mathsf{I}_{24},\quad 
E_{\mathsf{I}_{34}}=\mathsf{B}_{1}\mathsf{I}_{34},\quad J_{\mathsf{I}_{34}}=\mathsf{J}_{34}\mathsf{B}_{2},\\
E_{\mathsf{J}_{34}}=-\mathsf{J}_{34}\mathsf{B}_{1},&\quad J_{\mathsf{J}_{34}}=\mathsf{B}_{2}\mathsf{I}_{34}
\eea
where $i=1,2,3$. In fact, under these definitions, one can check that the traceless condition \eqref{eq:JE-condition} is indeed obeyed. Moreover, using the quadrality symmetry and the $J$-$E$ duality, one also confirm that the previous setups preserving four supersymmetries can be reproduced from this definition.

The low energy effective field theory of the D1-branes is given as
\bea
\mathcal{L}&=\int d\theta^{+} d\bar{\theta}^{+} \Tr\left(  \frac{1}{2g^{2}}\overline{\Upsilon}\Upsilon+i\sum_{a\in\four}\overline{\Phi}_{a}(\mathcal{D}_{0}-\mathcal{D}_{1}){\Phi}_{a}+\sum_{i=1}^{3}\overline{\Lambda}_{i}\Lambda_{i}  \right)+\left(\int d\theta^{+} \Tr \sum_{i=1}^{3}\Lambda_{i}J_{i}|_{\bar{\theta}^{+}=0}+\text{c.c} \right)\\
&+\int d\theta^{+} d\bar{\theta}^{+} \left(i\sum_{A\in\six}\overline{\Phi}_{A}(\mathcal{D}_{0}-\mathcal{D}_{1}){\Phi}_{A}+i\sum_{A\in\six}\widetilde{\Phi}_{A}(\mathcal{D}_{0}-\mathcal{D}_{1})\overline{\widetilde{{\Phi}}}_{A}+\sum_{A\in\six}\overbar{\Lambda}_{A}\Lambda_{A}+\sum_{A\in\six}\overbar{\widetilde{\Lambda}}_{A}\widetilde{\Lambda}_{A} \right)\\
&+\int d\theta^{+} \Tr \sum_{A\in\six}\left(\Lambda_{A}J_{\mathsf{I}_{A}}|_{\bar{\theta}^{+}=0}+\widetilde{\Lambda}_{A}J_{\mathsf{J}_{A}}|_{\bar{\theta}^{+}=0}+\text{c.c}\right) +\left(\frac{ir}{2}\int d\theta^{+}\Upsilon|_{\bar{\theta}^{+}=0}+\text{c.c.}\right)
\eea
So that the $B$-field can be understood as the FI-parameter, constraints on the $B$-field need to be imposed and it is actually
\bea
v_{1}=-v_{2}=v_{3}=-v_{4}.
\eea
We omit the derivation of this so see \cite{Nekrasov:2016gud}.

Integrating out the auxiliary fields, the scalar potential is given as
\bea
U_{\text{pot.}}&=\frac{g^{2}}{2}\left(\sum_{a\in\four}[\mathsf{B}_{a},\mathsf{B}_{a}^{\dagger}]+\sum_{A\in\six}(\mathsf{I}_{A}\mathsf{I}_{A}^{\dagger}-\mathsf{J}_{A}^{\dagger}\mathsf{J}_{A})-r \right)^{2}+\sum_{(ab)\in\six}\Tr|[\mathsf{B}_{a},\mathsf{B}_{b}]+\mathsf{I}_{A}\mathsf{J}_{A}|^{2}\\
&+\sum_{A\in\six}\sum_{a\in\bar{A}}\left(\Tr|\mathsf{B}_{a}\mathsf{I}_{A}|^{2}+\Tr|\mathsf{J}_{A}\mathsf{B}_{a}|^{2} \right).
\eea
The vacuum moduli space is then given as
\bea
\mathfrak{M}_{\vec{n},k}=\{(\vec{\mathsf{B}},\vec{\mathsf{I}},\vec{\mathsf{J}}\,)\mid \mu_{D}-r=\mu_{A}=\sigma_{a;A}=\widetilde{\sigma}_{a;A}=0\}/\U(k)
\eea
where
\bea
\mu_{D}&=\sum_{a\in\four}[\mathsf{B}_{a},\mathsf{B}_{a}^{\dagger}]+\sum_{A\in\six}(\mathsf{I}_{A}\mathsf{I}_{A}^{\dagger}-\mathsf{J}_{A}^{\dagger}\mathsf{J}_{A}),\\
\mu_{A}&=[\mathsf{B}_{a},\mathsf{B}_{b}]+\mathsf{I}_{A}\mathsf{J}_{A},\quad \sigma_{a;A}=\mathsf{B}_{a}\mathsf{I}_{A},\quad \tilde{\sigma}_{a;A}=\mathsf{J}_{A}\mathsf{B}_{a},\quad a\in\bar{A}.
\eea
This indeed coincides with \eqref{eq:spikedinstantonmoduli}.

\chapter{Gauge Origami partition functions}\label{chap:gaugeorigamipartitionfunction}

The goal of this chapter is to derive the gauge origami partition functions by performing as what we did for the pure SYM case in Chap.~\ref{chap:ADHM-localization}. As discussed in the previous chapter, the low energy field theory of the $k$ D1-branes is described by a $\mathcal{N}=(0,2)$ quiver gauge theory. Since these D1-branes play the roles of instantons in the D$(2p+1)$-branes wrapping holomorphic cycles in $\mathbb{C}^{4}$, the partition function is expected to be the partition function of this 2d theory.  We will derive contour integral formulas and explicitly evaluate them.

In the pure SYM case, the $\U(1)$ charges $q_{1,2}$ of the ADHM variables played an important role. This is also the case for the gauge origami system. After discussing general aspects of the flavor symmetries, we explicitly derive the flavor charges of each gauge origami system in section~\ref{sec:2susy2d-flavorsymmetry}. Instead of considering the 2d theory, we them take the T-duality and consider the supersymmetric quantum mechanics of this 2 SUSY theory in section~\ref{sec:susyqm-index}. The partition function is now the Witten index of this supersymmetric quantum mechanics. Explicit contour integral formulas for it are given in section~\ref{sec:susyqm-index}. We then use the JK residue prescription to evaluate these contour integral formulas in section~\ref{sec:gaugeorigami-multidim-partition}. The relations of the JK-residue and the index formalism are discussed in section~\ref{sec:equiv-index}. The existence of \textit{sign rules} of the magnificent four system is also discussed there.

\section{Flavor symmetries}\label{sec:2susy2d-flavorsymmetry}
As discussed in Chap.~\ref{chap:gauge-origami}, instantons in the gauge origami system are described by D1-branes wrapping $\mathbb{R}^{1,1}$ and probing D$(2p+1)$ $(p=2,3,4)$ branes wrapping $\mathbb{R}^{1,1}$ and non-compact complex planes in the $\mathbb{C}^{4}$. Moreover, the low energy world volume theory of the $k$ D1-branes is a 2d $\mathcal{N}=(0,2)$ quiver gauge theory and its vacuum moduli space coincides with the $k$-instanton moduli space of the gauge origami system (see section~\ref{sec:general-instanton}). The flavor symmetries of the 2d $\mathcal{N}=(0,2)$ quiver gauge theory is therefore directly related to the flavor symmetries of the instanton moduli space. 

Let $V_{\text{chiral}}$ and $V_{\text{Fermi}}$ be representations of the gauge group $G$. The chiral superfields $\Phi$ belong to $V_{\text{chiral}}$ and the Fermi superfields $\Lambda$ belong to $V_{\text{Fermi}}$. The $E$-term is a $V_{\text{Fermi}}$ valued holomorphic function of $\Phi$, i.e. $E:V_{\text{chiral}}\rightarrow V_{\text{Fermi}}$. The $J$-term is a $V^{\ast}_{\text{Fermi}}$ valued holomorphic function of $\Phi$, i.e. $J:V_{\text{chiral}}\rightarrow V_{\text{Fermi}}^{\ast}$. Assume that we have a set of unitary transformations $h_{\text{c}},h_{\text{f}}$ acting on $V_{\text{chiral}},V_{\text{Fermi}}$ that commute with the gauge group. We denote the flavor symmetry group as $G_{\text{F}}$ and it is defined as a symmetry obeying
\bea\label{eq:2susyflavorsymmetry}
E(h_{\text{c}}\Phi)&=h_{\text{f}}E(\Phi),\quad J(h_{\text{c}}\Phi)&=J(\Phi)h_{\text{f}}^{-1}.
\eea
The chiral and Fermi multiplets transform as $(\phi,\psi_{+})\rightarrow (h_{\text{c}}\phi,h_{\text{c}}\psi_{+})$ and $(\eta_{-},G)\rightarrow (h_{\text{F}}\eta_{-},h_{\text{F}}G)$. Note also that this is compatible with the fact that the $E,J$ terms have the conjugated quantum numbers.

For the gauge origami system, we have a $\U(1)^{4}$ global symmetry induced from the rotation $\prod_{a\in\four}\SO(2)_{a}$ rotating $\mathbb{C}^{4}$. Let us explicitly derive the flavor symmetries for each setup.

\paragraph{Magnificent four}Assuming the $\U(1)$ transformation
\bea\label{eq:gaugeorigami-U1symmetry}
\mathsf{B}_{a}\rightarrow q_{a}\mathsf{B}_{a},\quad a\in\four,
\eea
the transformations of the $E,J$-terms in \eqref{eq:2SUSYJEterm-magnificent} are summarized as
\bea
\begin{tabular}{|c|c|}
$E_{1}\rightarrow q_{41}E_{1}$ & $J_{1}\rightarrow q_{23}J_{1}$\\
$E_{2}\rightarrow q_{42}E_{2}$ & $J_{2}\rightarrow q_{13}J_{2}$\\
$E_{3}\rightarrow q_{43}E_{3}$ & $J_{3}\rightarrow q_{12}J_{3}$
\end{tabular}
\eea
where we shortly wrote $q_{ab}=q_{a}q_{b}$. Since the $J,E$-terms need to transform in a conjugate way, we have the condition
\bea\label{eq:CYcondition}
q_{1}q_{2}q_{3}q_{4}=1.
\eea
The $\U(1)^{3}$ charges of the chiral and Fermi superfields are then given as
\bea
\begin{tabular}{|c||c|}\hline
    Fields & $\U(1)^{3}$-charges  \\
\hline \hline $\mathsf{I}$     &  1\\
 \hline $\mathsf{B}_{a}$, $(a\in\four)$     &  $q_{a}$\\
\hline $\mathsf{\Lambda}_{1,2,3}$     &  $q_{4}q_{1,2,3}$ \\\hline
\end{tabular}
\eea
where we used the fact that the Fermi superfields have the same quantum numbers as the $E$-term. 

It is convenient to summarize the flavor charges in the quiver diagram as
\bea\label{eq:flavorquiver-magnificent}
\adjustbox{valign=c}{\begin{tikzpicture}[decoration={markings,mark=at position 0.7 with {\arrow{latex}}}]
 \tikzset{
        box/.style={draw, minimum width=0.7cm, minimum height=0.7cm, text centered,thick},
        ->-/.style={decoration={
        markings,mark=at position #1 with {\arrow[scale=1.5]{>}}},postaction={decorate},line width=0.5mm},
        -<-/.style={decoration={
        markings,
        mark=at position #1 with {\arrow[scale=1.5]{<}}},postaction={decorate},line width=0.5mm}    
    }
\begin{scope}[xshift=4cm]
    \draw[fill=black!10!white,thick](0,0) circle(0.4cm);
    \node at (0,0) {$k$};
    \node[box,fill=black!10!white] at (0,1.6) {$n$};
    \draw[postaction={decorate},thick] (0,1.25)--(0,0.4);
    \foreach \ang in {90,145,215,270} {
    \begin{scope}[rotate=\ang]
        \chiralarc[postaction={decorate},thick](0,0.5)(-45:225:0.22:0.65)
    \end{scope}
    }
    \foreach \ang in {90,145,270} {
    \begin{scope}[rotate=\ang]
    \fermiarc[postaction={decorate},thick](0,0.5)(-45:225:0.1:0.5)
    \end{scope}
    \node[right] at (0,0.8) {$\mathsf{I}$};
    \node[left] at (-1.5,0) {$\mathsf{B}_{2},\textcolor{red}{\Lambda_{2}}$};
    \node[right] at (1.6,0) {$\mathsf{B}_{1},\textcolor{red}{\Lambda_{1}}$};
    \node[below left] at (-0.9,-1){$\mathsf{B}_{3},\textcolor{red}{\Lambda_{3}}$};
    \node[below right] at (0.9,-1){$\mathsf{B}_{4}$};
    \draw[fill=black!10!white,thick](0,0) circle(0.4cm);
    \node at (0,0) {$k$};
    \node at (0,-2){$E_{i}=[\mathsf{B}_{4},\mathsf{B}_{i}]\quad J_{i}=\frac{1}{2}\varepsilon_{ijk4}[\mathsf{B}_{j},\mathsf{B}_{k}]$};
    }
\end{scope}
\end{tikzpicture}}\quad\quad  {\rightsquigarrow} \quad \quad 
\adjustbox{valign=c}{\begin{tikzpicture}[decoration={markings,mark=at position 0.7 with {\arrow{latex}}}]
 \tikzset{
        box/.style={draw, minimum width=0.7cm, minimum height=0.7cm, text centered,thick},
        ->-/.style={decoration={
        markings,mark=at position #1 with {\arrow[scale=1.5]{>}}},postaction={decorate},line width=0.5mm},
        -<-/.style={decoration={
        markings,
        mark=at position #1 with {\arrow[scale=1.5]{<}}},postaction={decorate},line width=0.5mm}    
    }
\begin{scope}[xshift=4cm]
    \draw[fill=black!10!white,thick](0,0) circle(0.4cm);
    \node at (0,0) {$k$};
    \node[box,fill=black!10!white] at (0,1.6) {$n$};
    \draw[postaction={decorate},thick] (0,1.25)--(0,0.4);
    \foreach \ang in {90,145,215,270} {
    \begin{scope}[rotate=\ang]
        \chiralarc[postaction={decorate},thick](0,0.5)(-45:225:0.22:0.65)
    \end{scope}
    }
    \foreach \ang in {90,145,270} {
    \begin{scope}[rotate=\ang]
    \fermiarc[postaction={decorate},thick](0,0.5)(-45:225:0.1:0.5)
    \end{scope}
    \node[right] at (0,0.8) {$1$};
    \node[left] at (-1.5,0) {$q_{2},\textcolor{red}{q_{2}q_{4}}$};
    \node[right] at (1.6,0) {$q_{1},\textcolor{red}{q_{1}q_{4}}$};
    \node[below left] at (-0.9,-1){$q_{3},\textcolor{red}{q_{3}q_{4}}$};
    \node[below right] at (0.9,-1){$q_{4}$};
    \draw[fill=black!10!white,thick](0,0) circle(0.4cm);
    \node at (0,0) {$k$};
    \node at (0,-2){ };
    }
\end{scope}
\end{tikzpicture}}
\eea
The information of the $E,J$-terms are implicitly incorporated in the flavor charges of the matter components of the quiver diagram.

\paragraph{Tetrahedron instanton} Similarly, assuming the $\U(1)$ transformation \eqref{eq:gaugeorigami-U1symmetry} and setting the $\U(1)$ action on $\mathsf{I}_{\bar{a}}$ as $\mathsf{I}_{\bar{a}}\rightarrow \mathsf{I}_{\bar{a}}$ for $a\in\four$, the $E,J$-terms in \eqref{eq:2SUSYquiver-tetrahedron} transform as
\bea
\begin{tabular}{|c|c|}
$E_{1}\rightarrow q_{41}E_{1}$ & $J_{1}\rightarrow q_{23}J_{1}$\\
$E_{2}\rightarrow q_{42}E_{2}$ & $J_{2}\rightarrow q_{13}J_{2}$\\
$E_{3}\rightarrow q_{43}E_{3}$ & $J_{3}\rightarrow q_{12}J_{3}$\\
$E_{\mathsf{I}_{\bar{a}}}\rightarrow q_{a}E_{\mathsf{I}_{\bar{a}}}$ &
\end{tabular}
\eea
The condition \eqref{eq:CYcondition} also comes from the fact that the $J,E$-terms transform in a conjugate way. The $\U(1)^{3}$ charges of the chiral and the Fermi superfields are then given as
\bea
\begin{tabular}{|c||c|}\hline
    Fields & $\U(1)^{3}$-charges  \\
\hline \hline $\mathsf{I}_{\bar{a}},\,(a\in\four)$     &  1\\
 \hline $\mathsf{B}_{a}$, $(a\in\four)$     &  $q_{a}$\\
\hline $\mathsf{\Lambda}_{1,2,3}$     &  $q_{4}q_{1,2,3}$ \\
\hline $\mathsf{\Lambda}_{\bar{a}}$ &  $q_{a}$  \\\hline
\end{tabular}
\eea

Similar to the magnificent four case, we can summarize the $\U(1)^{3}$ charges in the quiver diagram. For example, the quiver diagram for the tetrahedron instanton system with only one stack of D7$_{123}$-branes can be written as
\bea\label{eq:flavorquiver-tetrahedron}
\adjustbox{valign=c}{
\begin{tikzpicture}[decoration={markings,mark=at position \arrowHeadPosition with {\arrow{latex}}}]
 \tikzset{
        box/.style={draw, minimum width=0.7cm, minimum height=0.7cm, text centered,thick},
        ->-/.style={decoration={
        markings,mark=at position #1 with {\arrow[scale=1.5]{>}}},postaction={decorate},line width=0.5mm},
        -<-/.style={decoration={
        markings,
        mark=at position #1 with {\arrow[scale=1.5]{<}}},postaction={decorate},line width=0.5mm}    
    }
\begin{scope}[xshift=4cm]
    \draw[fill=black!10!white,thick](0,0) circle(0.4cm);
    \node at (0,0) {$k$};
    \node[box,fill=black!10!white] at (0,1.6) {$n_{\bar{4}}$};
    \draw[postaction={decorate},thick] (-0.1,1.25)--(-0.1,0.4);
    \draw[postaction={decorate},red,thick] (0.1,1.25)--(0.1,0.4);
    \foreach \ang in {90,145,215,270} {
    \begin{scope}[rotate=\ang]
        \chiralarc[postaction={decorate},thick](0,0.5)(-45:225:0.22:0.65)
    \end{scope}
    }
    \foreach \ang in {90,145,270} {
    \begin{scope}[rotate=\ang]
    \fermiarc[postaction={decorate},thick](0,0.5)(-45:225:0.1:0.5)
    \end{scope}
    \node[left] at (-0.1,0.8) {$\mathsf{I}$};
    \node[right] at (0.1,0.8) {$\textcolor{red}{\Lambda_{\mathsf{I}}}$};
    \node[left] at (-1.5,0) {$\mathsf{B}_{2},\textcolor{red}{\Lambda_{2}}$};
    \node[right] at (1.6,0) {$\mathsf{B}_{1},\textcolor{red}{\Lambda_{1}}$};
    \node[below left] at (-0.9,-1){$\mathsf{B}_{3},\textcolor{red}{\Lambda_{3}}$};
    \node[below right] at (0.9,-1){$\mathsf{B}_{4}$};
    \draw[fill=black!10!white,thick](0,0) circle(0.4cm);
    \node at (0,0) {$k$};
    \node at (0,-2.5){$\begin{array}{l}
    E_{i}=[\mathsf{B}_{4},\mathsf{B}_{i}],\quad E_{\mathsf{I}}=\mathsf{B}_{4}\,\mathsf{I}   \\
    J_{i}=\partial \mathsf{W}_{0}/\partial \mathsf{B}_{i}=\frac{1}{2}\varepsilon_{ijk4}[\mathsf{B}_{j},\mathsf{B}_{k}]\end{array}$};
    }
\end{scope}
\end{tikzpicture}}\quad \rightsquigarrow \quad 
\adjustbox{valign=c}{
\begin{tikzpicture}[decoration={markings,mark=at position \arrowHeadPosition with {\arrow{latex}}}]
 \tikzset{
        box/.style={draw, minimum width=0.7cm, minimum height=0.7cm, text centered,thick},
        ->-/.style={decoration={
        markings,mark=at position #1 with {\arrow[scale=1.5]{>}}},postaction={decorate},line width=0.5mm},
        -<-/.style={decoration={
        markings,
        mark=at position #1 with {\arrow[scale=1.5]{<}}},postaction={decorate},line width=0.5mm}    
    }
\begin{scope}[xshift=4cm]
    \draw[fill=black!10!white,thick](0,0) circle(0.4cm);
    \node at (0,0) {$k$};
    \node[box,fill=black!10!white] at (0,1.6) {$n_{\bar{4}}$};
    \draw[postaction={decorate},thick] (-0.1,1.25)--(-0.1,0.4);
    \draw[postaction={decorate},red,thick] (0.1,1.25)--(0.1,0.4);
    \foreach \ang in {90,145,215,270} {
    \begin{scope}[rotate=\ang]
        \chiralarc[postaction={decorate},thick](0,0.5)(-45:225:0.22:0.65)
    \end{scope}
    }
    \foreach \ang in {90,145,270} {
    \begin{scope}[rotate=\ang]
    \fermiarc[postaction={decorate},thick](0,0.5)(-45:225:0.1:0.5)
    \end{scope}}
    \node[left] at (-0.1,0.8) {$1$};
    \node[right] at (0.1,0.8) {$\textcolor{red}{q_{4}}$};
    \node[left] at (-1.5,0) {$q_{2},\textcolor{red}{q_{2}q_{4}}$};
    \node[right] at (1.6,0) {$q_{1},\textcolor{red}{q_{1}q_{4}}$};
    \node[below left] at (-0.9,-1){$q_{3},\textcolor{red}{q_{3}q_{4}}$};
    \node[below right] at (0.9,-1){$q_{4}$};
    \draw[fill=black!10!white,thick](0,0) circle(0.4cm);
    \node at (0,0) {$k$};
    \node at (0,-3){$ $};
\end{scope}
\end{tikzpicture}}
\eea

\vspace{-0.5cm}\paragraph{Spiked instanton}
Assuming the $\U(1)$ transformation \eqref{eq:gaugeorigami-U1symmetry} and setting the $\U(1)$ charge of $\mathsf{I}_{A}\,\,(A\in\six)$ to be trivial, the fields $\mathsf{J}_{A}$ need to transform as $\mathsf{J}_{A}\rightarrow q_{A}\mathsf{J}_{A}$ so that the $E,J$-terms transform in a homogeneous way. Under this transformation, the $E,J$-terms of \eqref{eq:2SUSYJEterm-spiked} transform as
\bea
\begin{tabular}{|c|c|}
$E_{1}\rightarrow q_{41}E_{1}$ & $J_{1}\rightarrow q_{23}J_{1}$\\
$E_{2}\rightarrow q_{42}E_{2}$ & $J_{2}\rightarrow q_{13}J_{2}$\\
$E_{3}\rightarrow q_{43}E_{3}$ & $J_{3}\rightarrow q_{12}J_{3}$\\
$E_{\mathsf{I}_{12}}\rightarrow q_{4}E_{\mathsf{I}_{12}}$ & $J_{\mathsf{I}_{12}}\rightarrow q_{123}J_{\mathsf{I}_{12}}$\\
$E_{\mathsf{J}_{12}}\rightarrow q_{124}E_{\mathsf{J}_{12}}$ & $J_{\mathsf{J}_{12}}\rightarrow q_{3}J_{\mathsf{J}_{12}}$\\
$E_{\mathsf{I}_{13}}\rightarrow q_{4}E_{\mathsf{I}_{13}}$ & $J_{\mathsf{I}_{13}}\rightarrow q_{123}J_{\mathsf{I}_{13}}$\\
$E_{\mathsf{J}_{13}}\rightarrow q_{134}E_{\mathsf{J}_{13}}$ & $J_{\mathsf{J}_{13}}\rightarrow q_{2}J_{\mathsf{J}_{13}}$\\
\end{tabular}\qquad \begin{tabular}{|c|c|}
$E_{\mathsf{I}_{14}}\rightarrow q_{2}E_{\mathsf{I}_{14}}$ & $J_{\mathsf{I}_{14}}\rightarrow q_{134}J_{\mathsf{I}_{14}}$\\
$E_{\mathsf{J}_{14}}\rightarrow q_{214}E_{\mathsf{J}_{14}}$ & $J_{\mathsf{J}_{14}}\rightarrow q_{3}J_{\mathsf{J}_{14}}$\\
$E_{\mathsf{I}_{23}}\rightarrow q_{4}E_{\mathsf{I}_{23}}$ & $J_{\mathsf{I}_{23}}\rightarrow q_{123}J_{\mathsf{I}_{23}}$\\
$E_{\mathsf{J}_{23}}\rightarrow q_{234}E_{\mathsf{J}_{23}}$ & $J_{\mathsf{J}_{23}}\rightarrow q_{1}J_{\mathsf{J}_{23}}$\\
$E_{\mathsf{I}_{24}}\rightarrow q_{3}E_{\mathsf{I}_{24}}$ & $J_{\mathsf{I}_{24}}\rightarrow q_{124}J_{\mathsf{I}_{24}}$\\
$E_{\mathsf{J}_{24}}\rightarrow q_{234}E_{\mathsf{J}_{24}}$ & $J_{\mathsf{J}_{24}}\rightarrow q_{1}J_{\mathsf{J}_{24}}$\\
$E_{\mathsf{I}_{34}}\rightarrow q_{1}E_{\mathsf{I}_{34}}$ & $J_{\mathsf{I}_{34}}\rightarrow q_{234}J_{\mathsf{I}_{34}}$\\
$E_{\mathsf{J}_{34}}\rightarrow q_{134}E_{\mathsf{J}_{34}}$ & $J_{\mathsf{J}_{34}}\rightarrow q_{2}J_{\mathsf{J}_{34}}$\\
\end{tabular}
\eea
Indeed each $E,J$ term transform in a conjugate way under the condition \eqref{eq:CYcondition}. The $\U(1)^{3}$ charges of the chiral and the Fermi superfieds are
\bea
\begin{tabular}{|c||c|}\hline
    Fields & $\U(1)^{3}$-charges  \\
\hline \hline $\mathsf{I}_{A},\,(A\in\six)$     &  1\\
\hline $\mathsf{J}_{A},\,(A\in\six)$ & $q_{A}$\\
 \hline $\mathsf{B}_{a}$, $(a\in\four)$     &  $q_{a}$\\
\hline $\mathsf{\Lambda}_{1,2,3}$     &  $q_{4}q_{1,2,3}$ \\
\hline $\mathsf{\Lambda}_{\mathsf{I}_{A}},\mathsf{\Lambda}_{\mathsf{J}_{A}}$ & \small{ $q_{a},q_{b}^{-1}$,\, $\{a,b\}=\bar{A}$ } \\
\hline
\end{tabular}
\eea
In our choice of the $E,J$-terms in \eqref{eq:2SUSYJEterm-spiked}, one of the Fermi superfields $\Lambda_{\mathsf{I}_{A}}$ will transform under $\U(1)^{3}$ as $q_{a}$ and the other $\Lambda_{\mathsf{J}_{A}}$ will transform as $q_{b}^{-1}$, where $\{a,b\}=\bar{A}$. 

The quiver diagram with flavor charges can be written similar to the previous cases, and for the setup when we have only one stack of D5$_{12}$-branes, we have
\bea\label{eq:flavorquiver-spiked}
\adjustbox{valign=c}{
\begin{tikzpicture}[decoration={markings,mark=at position \arrowHeadPosition with {\arrow{latex}}}]
 \tikzset{
        box/.style={draw, minimum width=0.7cm, minimum height=0.7cm, text centered,thick},
        ->-/.style={decoration={
        markings,mark=at position #1 with {\arrow[scale=1.5]{>}}},postaction={decorate},line width=0.5mm},
        -<-/.style={decoration={
        markings,
        mark=at position #1 with {\arrow[scale=1.5]{<}}},postaction={decorate},line width=0.5mm}    
    }
\begin{scope}[xshift=4cm]
    \draw[fill=black!10!white,thick](0,0) circle(0.4cm);
    \node at (0,0) {$k$};
    \node[box,fill=black!10!white] at (0,1.6) {$n_{12}$};
    \draw[postaction={decorate},thick] (-0.2,1.25)--(-0.2,0.3);
    \draw[postaction={decorate},red,thick] (-0.07,1.25)--(-0.07,0.4);
    \draw[postaction={decorate},thick] (0.2,0.3)--(0.2,1.25);
    \draw[postaction={decorate},red,thick] (0.07,0.4)--(0.07,1.25);
    \foreach \ang in {90,145,215,270} {
    \begin{scope}[rotate=\ang]
        \chiralarc[postaction={decorate},thick](0,0.5)(-45:225:0.22:0.65)
    \end{scope}
    }
    \foreach \ang in {90,145,270} {
    \begin{scope}[rotate=\ang]
    \fermiarc[postaction={decorate},thick](0,0.5)(-45:225:0.1:0.5)
    \end{scope}
    \node[scale=1.5] at (4,0) {$\rightsquigarrow$};
    \node[left] at (-0.2,0.8) {$\mathsf{I},\textcolor{red}{\Lambda_{\mathsf{I}}}$};
    \node[right] at (0.2,0.8) {$\mathsf{J},\textcolor{red}{\Lambda_{\mathsf{J}}}$};
    \node[left] at (-1.5,0) {$\mathsf{B}_{2},\textcolor{red}{\Lambda_{2}}$};
    \node[right] at (1.6,0) {$\mathsf{B}_{1},\textcolor{red}{\Lambda_{1}}$};
    \node[below left] at (-0.9,-1){$\mathsf{B}_{3},\textcolor{red}{\Lambda_{3}}$};
    \node[below right] at (0.9,-1){$\mathsf{B}_{4}$};
    \draw[fill=black!10!white,thick](0,0) circle(0.4cm);
    \node at (0,0) {$k$};
    \node at (0,-3){$\begin{array}{l}
    E_{i}=[\mathsf{B}_{4},\mathsf{B}_{i}],\quad E_{\mathsf{I}}=\mathsf{B}_{4}\mathsf{I},\quad E_{\mathsf{J}}=-\mathsf{J}\mathsf{B}_{4},\\
    J_{1}=[\mathsf{B}_{2},\mathsf{B}_{3}],\quad J_{2}=[\mathsf{B}_{3},\mathsf{B}_{1}],\\
     J_{3}=[\mathsf{B}_{1},\mathsf{B}_{2}]+\mathsf{I}\mathsf{J},\quad J_{\mathsf{I}}=\mathsf{J}\mathsf{B}_{3},\quad J_{\mathsf{J}}=\mathsf{B}_{3}\mathsf{I}
\end{array}$};
    
    }
\end{scope}
\end{tikzpicture}} \quad   \quad 
\adjustbox{valign=c}{
\begin{tikzpicture}[decoration={markings,mark=at position \arrowHeadPosition with {\arrow{latex}}}]
 \tikzset{
        box/.style={draw, minimum width=0.7cm, minimum height=0.7cm, text centered,thick},
        ->-/.style={decoration={
        markings,mark=at position #1 with {\arrow[scale=1.5]{>}}},postaction={decorate},line width=0.5mm},
        -<-/.style={decoration={
        markings,
        mark=at position #1 with {\arrow[scale=1.5]{<}}},postaction={decorate},line width=0.5mm}    
    }
\begin{scope}[xshift=4cm]
    \draw[fill=black!10!white,thick](0,0) circle(0.4cm);
    \node at (0,0) {$k$};
    \node[box,fill=black!10!white] at (0,1.6) {$n_{12}$};
    \draw[postaction={decorate},thick] (-0.2,1.25)--(-0.2,0.3);
    \draw[postaction={decorate},red,thick] (-0.07,1.25)--(-0.07,0.4);
    \draw[postaction={decorate},thick] (0.2,0.3)--(0.2,1.25);
    \draw[postaction={decorate},red,thick] (0.07,0.4)--(0.07,1.25);
    \foreach \ang in {90,145,215,270} {
    \begin{scope}[rotate=\ang]
        \chiralarc[postaction={decorate},thick](0,0.5)(-45:225:0.22:0.65)
    \end{scope}
    }
    \foreach \ang in {90,145,270} {
    \begin{scope}[rotate=\ang]
    \fermiarc[postaction={decorate},thick](0,0.5)(-45:225:0.1:0.5)
    \end{scope}
    \node[left] at (-0.2,0.8) {$1,\textcolor{red}{q_{4}}$};
    \node[right] at (0.2,0.8) {$q_{12},\textcolor{red}{q_{124}=q_{3}^{-1}}$};
    \node[left] at (-1.5,0) {$q_{2},\textcolor{red}{q_{24}}$};
    \node[right] at (1.6,0) {$q_{1},\textcolor{red}{q_{14}}$};
    \node[below left] at (-0.9,-1){$q_{3},\textcolor{red}{q_{34}}$};
    \node[below right] at (0.9,-1){$q_{4}$};
    \draw[fill=black!10!white,thick](0,0) circle(0.4cm);
    \node at (0,0) {$k$};
    \node at (0,-3){ };
    \node at (0,-4){ };
    
    }
\end{scope}
\end{tikzpicture}}
\eea

\section{Supersymmetric quantum mechanics and Witten index}\label{sec:susyqm-index}
To obtain the gauge origami instanton partition functions, one needs to compute the elliptic genus of the 2d $\mathcal{N}=(0,2)$ quiver gauge theories. Namely, we take the $\mathbb{R}^{1,1}$ part to be a torus $\mathbb{T}^{2}$ and study the partition function of the 2d theory on it. General discussions on how to compute the elliptic genus of a given 2d $\mathcal{N}=(0,2)$ quiver gauge theory were done in \cite{Benini:2013xpa,Benini:2013nda}. When computing the elliptic genus, we have a parameter $p$ which is the exponent of the complex structure $\tau$ of $\mathbb{T}^{2}$. The elliptic genus is then a function of this parameter $p$ and furthermore it will be written in theta functions.

Instead of considering the D1-branes in the type IIB string theory, from now on we will take the T-duality along the $9$-direction and consider the D0--$\D(2p)$ ($p=2,3,4$) setup in this section. Namely, the magnificent four, tetrahedron instanton, and spiked instanton setups will be the D0--D8, D0--D6$_{\bar{a}}$, D0--D4$_{A}$ brane setups. The low energy field theory of the D0-branes will be a dimensional reduction of the 2d $\mathcal{N}=(0,2)$ quiver gauge theory which is the 1d $\mathcal{N}=2$ supersymmetric quantum mechanics. The partition function of the theory is then described by the flavored Witten indices of these supersymmetric quantum mechanics.

Before moving on to the explicit evaluation of the Witten indices for the gauge origami system, let us briefly summarize the basic components to compute the Witten indices for general quiver gauge theories \cite{Hori:2014tda, Hwang:2014uwa,Ohta:2014ria,Cordova:2014oxa}. 

\paragraph{Supersymmetric quantum mechanics}
The supersymmetric quantum mechanics of a given 2d $\mathcal{N}=(0,2)$ quiver gauge theory is obtained by the dimensional reduction of section~\ref{sec:2susyquivergaugetheory-general}. To do this, one simply assumes that all the fields are independent on the $x^{1}$-coordinate and lose the differential operators depending on $\partial_{1}$. We basically follow the notations in \cite{Hori:2014tda}. We denote the two supercharges and the fermionic coordinates as $\tQ_{+}\rightarrow \tQ$, $\otQ_{+}\rightarrow \otQ$, and $\theta^{+}\rightarrow \theta,\bar{\theta}^{+}\rightarrow\bar{\theta}$. In the superfield description, the differential operators are
\bea
\tQ=\frac{\partial}{\partial\theta}+\frac{i}{2}\overline{\theta}\frac{\partial}{\partial t},\quad \otQ=-\frac{\partial}{\partial\overline{\theta}}-\frac{i}{2}\theta\frac{\partial}{\partial t},\quad 
\tD=\frac{\partial}{\partial\theta}-\frac{i}{2}\overline{\theta}\frac{\partial}{\partial t},\quad \otD=-\frac{\partial}{\partial \overline{\theta}}+\frac{i}{2}\theta \frac{\partial }{\partial t}.
\eea
The matter components we are interested in are the same as the 2d case: vector superfield, chiral superfield, and Fermi-superfield. 
\begin{itemize}[topsep=0pt, partopsep=0pt, itemsep=0pt]
\item Vector superfield and field strength $V=(v_{t},\sigma, \lambda,\bar{\lambda},D),\Upsilon$: 
\bea
V&=(v_{t}-\sigma) -i\theta\bar{\lambda}-i\bar{\theta}\lambda+\theta\bar{\theta}D,\\
\Upsilon&=-\lambda+\theta(D_{t}\sigma+iD)+\frac{i}{2}\theta\bar{\theta}D_{t}^{(+)}\lambda.
\eea 
where $D_{t}=\partial_{t}+iv_{t},\,\, D_{t}^{(\pm)}=\partial_{t}+i(v_{t}\pm \sigma)$ and $v_{0}\rightarrow v_{t},v_{1}\rightarrow \sigma$ and $\sigma$ will be a one-dimensional scalar field.
 \item Gauge covariant chiral superfield $\Phi=(\phi,\psi)$:
    \bea
    \overline{\mathcal{D}}\Phi=0,\quad \Phi=\phi+\theta \psi-\frac{i}{2}\theta\overline{\theta}D^{(+)}_{t}\phi
    \eea
    where the gauge covariant superderivatives are defined as 
    \bea
    \overline{\mathcal{D}}&=-\partial_{\bar{\theta}}+\frac{i}{2}\theta D_{t}^{(+)},\quad
\mathcal{D}=\partial_{\theta}-\frac{i}{2}\overline{\theta} D_{t}^{(+)}
    \eea
    under the WZ gauge.
    \item Gauge covariant Fermi superfield $\Lambda=(\eta,G)$:
    \bea
\overline{\mathcal{D}} \Lambda=E(\Phi),\quad \Lambda=\eta-\theta G-\frac{i}{2}\theta\bar{\theta}D_{t}^{(+)}\eta-\bar{\theta} E(\Phi).
\eea
\end{itemize}
We can also introduce the $J$-terms and the generic Lagrangian of the 1d $\mathcal{N}=2$ SQM is written as
\bea\label{eq:SQMLagrangian}
L_{\text{gauge}}&=\int d\theta d\bar{\theta} \frac{1}{2e^{2}}\Tr\overline{\Upsilon}\Upsilon= \frac{1}{2g^{2}}\Tr\left[(D_{t}\sigma)^{2}+i\bar{\lambda}D_{t}^{(+)}\lambda+D^{2}\right]\\
L_{\text{FI}}&=\frac{ir}{2}\int d\theta d \bar{\theta}\Tr \Upsilon |_{\bar{\theta}^{+}=0}+\text{c.c.}=-r \Tr D\\
L_{\text{chiral}}&=\int d\theta d\bar{\theta}\,i\overline{\Phi}\mathcal{D}_{t}^{(-)}{\Phi}\,\,=D_{t}\bar{\phi}D_{t}\phi+i\bar{\psi}D_{t}^{(-)}\psi+\bar{\phi}\{D-\sigma^{2}\}\phi-i\bar{\phi}\lambda\psi+i\overline{\psi\lambda}\phi \\
L_{\text{Fermi}}&=\int d\theta d\bar{\theta} \bar{\Lambda}\Lambda =i\bar{\eta}D_{t}^{(+)}\eta+|G|^{2}-|E(\phi)|^{2}-\bar{\eta}\partial_{i}E(\phi)\psi^{i}-\bar{\psi}^{\bar{i}}\partial_{\bar{i}}\overline{E(\phi)}\eta,\\
L_{\text{J-term}}&=-\int d\theta\,\left(\Lambda^{\alpha}J_{\alpha}(\Phi)\right)_{\bar{\theta}=0}+\text{c.c}=\eta^{\alpha}\partial_{i}J_{\alpha}(\phi)\psi^{i}+G^{\alpha}J_{\alpha}(\phi)+\text{c.c}.
\eea

All the properties discussed in section~\ref{sec:2susyquivergaugetheory-general} appear also in the SQM. Namely, we also have the $J,E$-duality, the vacuum moduli space is described as \eqref{eq:scalarpotential-general}, we can also consider quiver gauge theories using the same quiver diagram, and the flavor symmetries are all the same as section~\ref{sec:2susy2d-flavorsymmetry}.

\paragraph{Witten index}Given the supersymmetric quantum mechanics, the Witten index is defined as
\bea
\mathcal{Z}(y)=\Tr\left[(-1)^{F}e^{-\beta H}\prod_{a}y_{a}^{T_{a}}\right]
\eea
where $F$ is the fermion number, $\{y_{a}=e^{u_{a}}\}$ are the fugacities of the flavor symmetry, $\{T_{a}\}$ are the Cartan generators of the flavor symmetry group, and $\beta$ is the circumference of $\mathbb{S}^{1}$. Let us consider a one-dimensional system described by a vector multiplet $V$ with gauge group $G$, chiral multiplets $\Phi_{i}$ transforming in a representation $V_{\text{chiral}}$ of $G\times G_{f}$, and Fermi multiplets $\Lambda_{\alpha}$ transforming in the representation $V_{\text{Fermi}}$ of $G\times G_{f}$. The Witten index can be derived by considering the path integral of the Lagrangian\footnote{We take the Wick rotation to perform this path integral.} \eqref{eq:SQMLagrangian}, which eventually breaks down to
\bea
\mathcal{Z}(y)=\frac{1}{|W_{G}|}\oint _{\text{JK}}Z_{\text{V}}\prod_{i}Z_{\Phi_{i}}\prod_{\alpha}Z_{\Lambda_{\alpha}}
\eea
where\footnote{Recall that $\sh(x)=2\sinh\left(\frac{x}{2}\right)$.}
\bea\label{eq:2susylocalization}
Z_{\text{V}}&=\prod_{I}\frac{d\phi_{I}}{2\pi i}\prod_{\alpha\in G}\sh\left(\beta \alpha\cdot \phi\right),\quad Z_{\Phi_{i}}=\prod_{\rho\in V_{\text{chiral}}}\frac{1}{\sh(\beta\rho\cdot \mu)},\quad Z_{\Lambda_{\alpha}}=\prod_{\rho\in V_{\text{Fermi}}}\sh(\beta \rho\cdot \mu)
\eea
and $|W_{G}|$ is the order of the Weyl group of $G$. The parameter $\phi$ parametrizes the Cartan subalgebra of $G$ and $\mu$ contains $\phi$ and $u_{a}$. Each contribution is understood as follows. We first consider the index with constant gauge fields parameterized by $\phi_{I},u_{a}$. The supersymmetric Lagrangian given above surprisingly obeys the property $S=\mathcal{Q}V$ where $\mathcal{Q}=\tQ+\otQ$ for some $V$ and it is a $\mathcal{Q}$-exact term. We then can use supersymmetric localization as explained in section~\ref{sec:YMtheory-instanton} and compute the path integral of the action. The quadratic terms of the chiral and Fermi superfields only contributes to the superdeterminant as given in \eqref{eq:susylocal-opVEV} and they give $Z_{\Phi_{i}}$ and $Z_{\Lambda_{\alpha}}$. The variables inside the $\sh(\beta x)$-function comes from the constant gauge fields and $\sh( \beta x)$ appears because we are considering the quantum mechanics on $\mathbb{S}^{1}$ giving Fourier modes such as $e^{i2n\pi x/\beta}$. Since we are considering the index of a gauge theory whose gauge group is $G$, the index must be invariant under the gauge transformation. Thus, finally, we need to take the integral over the $\phi$ variables. Such contribution gives the vector superfield contribution $Z_{V}$. We then arrive at the contour integral formula \eqref{eq:2susylocalization}. Since the contour integral will be generally contour integrals on multi-variables, we need to specify how to take them. Now it is known that the resulting formalism to take the contour integral formula is the JK-residue formalism explained in section~\ref{sec:pureSYM_JK}. See the original papers \cite{Benini:2013xpa,Benini:2013nda, Hori:2014tda} for details. 

In particular, we are interested in the Witten index of 1d $\mathcal{N}=2$ quiver gauge theories with unitary groups. Let $\overline{Q}=(\overline{Q}_{0},\overline{Q}_{1})$ be the 2 SUSY quiver diagram and $\overline{Q}_{0},\overline{Q}_{1}$ denoting the set of nodes and edges. For use in the gauge origami setup, we assume that there is a flavor $\U(1)^{3}$ symmetry and turn on the fugacities $q_{i}=e^{\epsilon_{i}}$. We denote the $\U(1)^{3}$ charges of the chiral, Fermi superfields as $q(\Phi)=e^{\epsilon(\Phi)}$, $q(\Lambda)=e^{\epsilon(\Lambda)}$. 

The quiver diagram will now be written as discussed in section~\ref{sec:2susy2d-flavorsymmetry}. We further have flavor groups drawn in square nodes and we denote the fugacities of them collectively as $\{v_{i}=e^{\mathfrak{a}_{i}}\}$. On the other hand, the exponent of the root of the gauge group will be denoted as $x=e^{\phi}$ and at the end the integral over it will be taken. Under this notations, the basic factors appearing in the Witten index are determined as follows.
\begin{itemize}[topsep=0pt, partopsep=0pt, itemsep=0pt]
    \item For each gauge group (circle node) of the theory, we have 
    \bea
\adjustbox{valign=c}{\begin{tikzpicture}[decoration={markings,mark=at position \arrowHeadPosition with {\arrow{latex}}}]
 \tikzset{
        box/.style={draw, minimum width=0.6cm, minimum height=0.6cm, text centered,thick},
        ->-/.style={decoration={
        markings,mark=at position #1 with {\arrow[scale=1.5]{>}}},postaction={decorate},line width=0.5mm},
        -<-/.style={decoration={
        markings,
        mark=at position #1 with {\arrow[scale=1.5]{<}}},postaction={decorate},line width=0.5mm}    
    }
\begin{scope}{xshift=0cm}
    \draw[fill=black!10!white,thick](0,0) circle(0.4cm);
    \node at (0,0){$N_{a}$};
\end{scope}
\end{tikzpicture}}\quad {\rightsquigarrow} \quad \prod_{I=1}^{N_{a}}\frac{d\phi_{I}^{(a)}}{2\pi i } \prod_{I\neq J}\sh(\phi_{I}^{(a)}-\phi_{J}^{(a)})
    \eea
where we simply set $\beta=1$.
\item The chiral superfield corresponding to an arrow from $\U(N_{a})$ to $\U(N_{b})$ gives the contribution
    \bea
   \adjustbox{valign=c}{ \begin{tikzpicture}[decoration={markings,mark=at position \arrowHeadPosition with {\arrow{latex}}}]
 \tikzset{
        box/.style={draw, minimum width=0.6cm, minimum height=0.6cm, text centered,thick},
        ->-/.style={decoration={
        markings,mark=at position #1 with {\arrow[scale=1.5]{>}}},postaction={decorate},line width=0.5mm},
        -<-/.style={decoration={
        markings,
        mark=at position #1 with {\arrow[scale=1.5]{<}}},postaction={decorate},line width=0.5mm}    
    }

\begin{scope}{}
\draw[fill=black!10!white,thick](4.7,-2) circle(0.4cm);
\node at (4.7,-2){$N_{b}$};
\draw[fill=black!10!white,thick](2.1,-2) circle(0.4cm);
\node at (2.1,-2){$N_{a}$};
\draw[postaction={decorate}, thick](2.5,-2)--(4.3,-2);
\node[above] at (3.4,-2){$q(\Phi_{a\rightarrow b})$};
\end{scope}
\end{tikzpicture}}\quad \rightsquigarrow\quad \prod_{I=1}^{N_{a}}\prod_{J=1}^{N_{b}}\frac{1}{\sh(\phi_{J}^{(b)}-\phi_{I}^{(a)}-\epsilon(\Phi_{a\rightarrow b}))}.
    \eea
    For chiral superfields in the adjoint representation, we set $a=b$. For square nodes, we simply replace $\{\phi_{I}^{(a)}\}_{I=1}^{N_{a}}$ with $\{\mathfrak{a}_{\alpha}^{(a)}\}_{\alpha=1}^{N_{a}}$.
    \item The Fermi superfield corresponding to an arrow from $\U(N_{a})$ to $\U(N_{b})$ gives the contribution 
    \bea
   \adjustbox{valign=c}{ \begin{tikzpicture}[decoration={markings,mark=at position \arrowHeadPosition with {\arrow{latex}}}]
 \tikzset{
        box/.style={draw, minimum width=0.6cm, minimum height=0.6cm, text centered,thick},
        ->-/.style={decoration={
        markings,mark=at position #1 with {\arrow[scale=1.5]{>}}},postaction={decorate},line width=0.5mm},
        -<-/.style={decoration={
        markings,
        mark=at position #1 with {\arrow[scale=1.5]{<}}},postaction={decorate},line width=0.5mm}    
    }

\begin{scope}{}
\draw[fill=black!10!white,thick](4.7,-2) circle(0.4cm);
\node at (4.7,-2){$N_{b}$};
\draw[fill=black!10!white,thick](2.1,-2) circle(0.4cm);
\node at (2.1,-2){$N_{a}$};
\draw[red, postaction={decorate}, thick](2.5,-2)--(4.3,-2);
\node[above] at (3.4,-2){\textcolor{red}{$q(\Lambda_{a\rightarrow b})$}};
\end{scope}
\end{tikzpicture}}\quad \rightsquigarrow\quad \prod_{I=1}^{N_{a}}\prod_{J=1}^{N_{b}}\sh(\phi_{J}^{(b)}-\phi_{I}^{(a)}-\epsilon(\Lambda_{a\rightarrow b})).
    \eea
    Similarly, we set $a=b$ for Fermi superfields in the adjoint representation, and for square nodes we simply replace $\{\phi_{I}^{(a)}\}_{I=1}^{N_{a}}$ with $\{\mathfrak{a}_{\alpha}^{(a)}\}_{\alpha=1}^{N_{a}}$.
\end{itemize}

\paragraph{Magnificent four}Given the quiver diagram in \eqref{eq:flavorquiver-magnificent}, we can write down the contour integral formula 
\bea
\adjustbox{valign=c}{\begin{tikzpicture}[decoration={markings,mark=at position 0.7 with {\arrow{latex}}}]
 \tikzset{
        box/.style={draw, minimum width=0.7cm, minimum height=0.7cm, text centered,thick},
        ->-/.style={decoration={
        markings,mark=at position #1 with {\arrow[scale=1.5]{>}}},postaction={decorate},line width=0.5mm},
        -<-/.style={decoration={
        markings,
        mark=at position #1 with {\arrow[scale=1.5]{<}}},postaction={decorate},line width=0.5mm}    
    }
\begin{scope}[xshift=4cm]
    \draw[fill=black!10!white,thick](0,0) circle(0.4cm);
    \node at (0,0) {$k$};
    \node[box,fill=black!10!white] at (0,1.6) {$n$};
    \draw[postaction={decorate},thick] (0,1.25)--(0,0.4);
    \foreach \ang in {90,145,215,270} {
    \begin{scope}[rotate=\ang]
        \chiralarc[postaction={decorate},thick](0,0.5)(-45:225:0.22:0.65)
    \end{scope}
    }
    \foreach \ang in {90,145,270} {
    \begin{scope}[rotate=\ang]
    \fermiarc[postaction={decorate},thick](0,0.5)(-45:225:0.1:0.5)
    \end{scope}
    \node[right] at (0,0.8) {$1$};
    \node[left] at (-1.5,0) {$q_{2},\textcolor{red}{q_{2}q_{4}}$};
    \node[right] at (1.6,0) {$q_{1},\textcolor{red}{q_{1}q_{4}}$};
    \node[below left] at (-0.9,-1){$q_{3},\textcolor{red}{q_{3}q_{4}}$};
    \node[below right] at (0.9,-1){$q_{4}$};
    \draw[fill=black!10!white,thick](0,0) circle(0.4cm);
    \node at (0,0) {$k$};
    \node at (0,-2){ };
    }
\end{scope}
\end{tikzpicture}}\quad {\rightsquigarrow}\quad \begin{dcases}
V:\,\, \prod_{I=1}^{k}\frac{d\phi_{I}}{2\pi i }\prod_{I\neq J}\sh(\phi_{I}-\phi_{J}),\\
\Lambda_{i}\,(i=1,2,3):\, \prod_{I,J=1}^{k}\sh(\phi_{I}-\phi_{J}-\epsilon_{i}-\epsilon_{4}) \\
\mathsf{I}:\,\, \prod_{I=1}^{k}\prod_{\alpha=1}^{n}\frac{1}{\sh(\phi_{I}-\mathfrak{a}_{\alpha})},\quad \mathsf{B}_{a}\,(a\in\four):\,\prod_{I,J=1}^{k} \frac{1}{\sh(\phi_{I}-\phi_{J}-\epsilon_{a})}.
\end{dcases}
\eea
\begin{proposition}\label{prop:M4noantiDbrane}
The partition function of the magnificent four setup is 
\bea
\mathcal{Z}^{\D8}_{\text{inst.}}=\sum_{k=0}^{\infty}\mathfrak{q}^{k}\mathcal{Z}_{k}^{\D8},\quad \mathcal{Z}_{k}^{\D8}=\frac{1}{k!}\oint_{\text{JK}}\prod_{I=1}^{k}\frac{d\phi_{I}}{2\pi i }\,\,\mu_{k}^{\D8\tbar\D0}(\mathfrak{a}_{\alpha},\phi_{I})
\eea
where
\bea
\mu_{k}^{\D8\tbar\D0}(\mathfrak{a}_{\alpha},\phi_{I})&=\left(\frac{\sh(-\epsilon_{14,24,34})}{\sh(-\epsilon_{1,2,3,4})}\right)^{k}\prod_{I=1}^{k}\prod_{\alpha=1}^{n}\frac{1}{\sh(\phi_{I}-\mathfrak{a}_{\alpha})}\prod_{I\neq J}^{k}\frac{\sh(\phi_{I}-\phi_{J})\sh(\phi_{I}-\phi_{J}-\epsilon_{14,24,34})}{\sh(\phi_{I}-\phi_{J}-\epsilon_{1,2,3,4})}.
\eea

\end{proposition}

For convenience, we will include a Fermi superfield whose contribution modifies the numerator contributions with a factor $\prod_{I=1}^{k}\prod_{\alpha=1}^{n}\sh(\phi_{I}-\mathfrak{b}_{\alpha})$, where $\{\mathfrak{b}_{\alpha}\}$ is generic from $\{\mathfrak{a}_{\alpha}\}$. 
\begin{proposition}\label{prop:M4contourJK}
The magnificent four partition function with $\D8\tbar\overline{\D8}$-branes is
\bea\label{eq:M4contourJK}
\mathcal{Z}_{k}^{\D8}&=\frac{1}{k!}\oint_{\text{JK}}\prod_{I=1}^{k}\frac{d\phi_{I}}{2\pi i }\,\,\mu^{\D8\tbar\D0}_{k}(\mathfrak{a}_{\alpha},\mathfrak{b}_{\alpha},\phi_{I}),\\
\mu^{\D8\tbar\D0}_{k}(\mathfrak{a}_{\alpha},\mathfrak{b}_{\alpha},\phi_{I})&=\left(\frac{\sh(-\epsilon_{14,24,34})}{\sh(-\epsilon_{1,2,3,4})}\right)^{k}\prod_{I=1}^{k}\prod_{\alpha=1}^{n}\frac{\sh(\phi_{I}-\mathfrak{b}_{\alpha})}{\sh(\phi_{I}-\mathfrak{a}_{\alpha})}\prod_{I\neq J}^{k}\frac{\sh(\phi_{I}-\phi_{J})\sh(\phi_{I}-\phi_{J}-\epsilon_{14,24,34})}{\sh(\phi_{I}-\phi_{J}-\epsilon_{1,2,3,4})}.
\eea
\end{proposition}
One motivation to include this extra Fermi multiplet contribution is to make the degrees of orders of the numerator and denominator to be the same. In the original paper \cite{Nekrasov:2017cih,Nekrasov:2018xsb}, this setup was interpreted as a setup of pairs of $\D8\tbar\overline{\D8}$-branes with D0-branes. The existence of anti D8-branes makes the group of the $\D8\tbar\overline{\D8}$ theory to look like a supergroup $\U(n|n)$. The parameters $\{\mathfrak{a}_{\alpha}\}$ parametrizes the Cartan subalgebra of the $\U(n\,|\,0)$ part while $\{\mathfrak{b}_{\alpha}\}$ parametrizes the Cartan subalgebra of $\U(0\,|\,n)$. Physically, the differences of the parameters $\mathfrak{a}_{\alpha},\mathfrak{b}_{\beta}$ are understood as the distance between the D8 and $\overline{\D8}$-branes. From the string theory viewpoint, the existence of the anti D-branes make the setup nonsupersymmetric. However, after tachyon condensation \cite{Sen:1998sm}, it is expected that the D8 and $\overline{\D8}$-branes annihilate each other eventually leaving a D6-brane extending in one of the three complex three-planes $\mathbb{C}^{3}_{\bar{a}}$. At the partition function level, this can be observed by tuning the parameters as $\mathfrak{b}_{\alpha}=\mathfrak{a}_{\alpha}+\epsilon_{a}$ for some $a\in\four$. Actually, under this restriction, the tetrahedron instanton partition function in Prop.~\ref{prop:tetracontourJK} can be reproduced.

The discussion above can be extended to the case when the appearing theory is a $\U(n|m)$ theory \cite{Billo:2021xzh}. In \cite{Billo:2021xzh}, the authors introduced a stack of $n$ D8-branes and $m$ D8'-branes to the system.\footnote{Strictly speaking, they were considering the D7--D$(-1)$ setup which is the T-duality of our setup here.} For the D8-branes, we introduce a constant background field, whose effect is equivalent to the $B$-field above. On the other hand, for the D8'-branes, we do not turn on any background field. The open string D0-D8 gives a chiral multiplet while the open string D0-D8' gives a Fermi multiplet (see section~\ref{sec:M4-openstring}). Setting $n=m$ actually reproduces the above partition function.


\paragraph{Tetrahedron instanton}Using the quiver diagram of the tetrahedron instanton, we can write down the contour integral formula
\bea
\adjustbox{valign=c}{
\begin{tikzpicture}[decoration={markings,mark=at position \arrowHeadPosition with {\arrow{latex}}}]
    \tikzset{
        cir/.style={circle,fill=black!10!white, draw, minimum size=0.8cm, text centered, thick},
        point/.style={circle, fill, inner sep=1.5pt},
        box/.style={draw, fill=black!10!white,minimum width=0.7cm, minimum height=0.7cm, text centered,thick},
        ->-/.style={decoration={
        markings,mark=at position #1 with {\arrow[scale=1.5]{>}}},postaction={decorate},line width=0.5mm},
        -<-/.style={decoration={
        markings,
        mark=at position #1 with {\arrow[scale=1.5]{<}}},postaction={decorate},line width=0.5mm}    
    }

    \node[cir] (i1) at ($(0,0)$) {$k$};
    \begin{scope}[rotate=0]
    \node[box,rotate=0] (n13) at ($(0,-2)$) {\rotatebox{-0}{$n_{\bar{4}}$}};
    \node[below] at (0.5,-1.) {$1,\textcolor{red}{q_{4}}$};

    \draw[postaction={decorate}, thick] ($(0.06,-1.65)$) -- ($(0.06,-0.4)$);
    \draw[postaction={decorate},red, thick] ($(-0.06,-1.65)$) -- ($(-0.06,-0.4)$);
    \end{scope}

    \begin{scope}[rotate=90]
    \node[box,rotate=90] (n12) at  ($(0,-2)$) {\rotatebox{-90}{$n_{\bar{1}}$}};
    \node[] at (0.3,-1.1) {$1,\textcolor{red}{q_{1}}$};
    \draw[postaction={decorate}, thick] ($(0.06,-1.65)$) -- ($(0.06,-0.4)$);
    \draw[postaction={decorate},red, thick] ($(-0.06,-1.65)$) -- ($(-0.06,-0.4)$);
    \end{scope}
    
    \begin{scope}[rotate=180]
    \node[box,rotate=180] (n12) at  ($(0,-2)$) {\rotatebox{-180}{$n_{\bar{2}}$}};
    \node[] at (0.4,-1.1) {$1,\textcolor{red}{q_{2}}$};
    
    \draw[postaction={decorate}, thick] ($(0.06,-1.65)$) -- ($(0.06,-0.4)$);
    \draw[postaction={decorate},red, thick] ($(-0.06,-1.65)$) -- ($(-0.06,-0.4)$);
    \end{scope}
    
    \begin{scope}[rotate=270]
    \node[box,rotate=270] (n23) at  ($(0,-2)$) {\rotatebox{-270}{$n_{\bar{3}}$}};
    \node[] at (0.3,-1.1) {$1,\textcolor{red}{q_{3}}$};
    
    \draw[postaction={decorate}, thick] ($(0.06,-1.65)$) -- ($(0.06,-0.4)$);
    \draw[postaction={decorate},red, thick] ($(-0.06,-1.65)$) -- ($(-0.06,-0.4)$);
    \end{scope}
    \begin{scope}[rotate=45]
    \chiralarc[postaction={decorate},thick](0,0.55)(-45:225:0.22:0.65)
    \fermiarc[postaction={decorate},thick](0,0.5)(-45:225:0.1:0.5)
    \end{scope}
    \begin{scope}[rotate=135]
    \chiralarc[postaction={decorate},thick](0,0.55)(-45:225:0.22:0.65)
    \fermiarc[postaction={decorate},thick](0,0.5)(-45:225:0.1:0.5)
    \end{scope}
    \begin{scope}[rotate=225]
    \chiralarc[postaction={decorate},thick](0,0.55)(-45:225:0.22:0.65)
    \end{scope}
    \begin{scope}[rotate=315]
    \chiralarc[postaction={decorate},thick](0,0.55)(-45:225:0.22:0.65)
    \fermiarc[postaction={decorate},thick](0,0.5)(-45:225:0.1:0.5)
    \end{scope}    
\end{tikzpicture}}\quad {\rightsquigarrow} \quad \begin{dcases}
    V:\,\prod_{I=1}^{k}\frac{d\phi_{I}}{2\pi i}\prod_{I\neq J}\sh(\phi_{I}-\phi_{J}),\quad \mathsf{B}_{a}:\,\prod_{I,J}\frac{1}{\sh(\phi_{I}-\phi_{J}-\epsilon_{a})}\\
    \Lambda_{i}:\, \prod_{I,J=1}^{k}\sh(\phi_{I}-\phi_{J}-\epsilon_{i}-\epsilon_{4}),\\
\mathsf{I}_{\bar{a}}:\prod_{I=1}^{k}\prod_{\alpha=1}^{n_{\bar{a}}}\frac{1}{\sh(\phi_{I}-\mathfrak{a}_{\bar{a},\alpha})},\quad \Lambda_{\bar{a}}:\,\prod_{I=1}^{k}\prod_{\alpha=1}^{n_{\bar{a}}}\sh(\phi_{I}-\mathfrak{a}_{\bar{a},\alpha}-\epsilon_{4})
\end{dcases}
\eea
\begin{proposition}\label{prop:tetracontourJK}
    The partition function of the tetrahedron instanton setup is
\bea
\mathcal{Z}^{\D6}_{\text{inst.}}=\sum_{k=0}^{\infty}\mathfrak{q}^{k}\mathcal{Z}_{k}^{\D6},\quad \mathcal{Z}_{k}^{\D6}=\frac{1}{k!}\oint_{\text{JK}} \prod_{I=1}^{k}\frac{d\phi_{I}}{2\pi i }\,\mu_{k}^{\D6\tbar\D0}(\mathfrak{a}_{\bar{a},\alpha},\phi_{I})
\eea
where
\bea\label{eq:tetracontourJK}
\mu_{k}^{\D6\tbar\D0}(\mathfrak{a}_{\bar{a},\alpha},\phi_{I})&=\left(\frac{\sh(-\epsilon_{14,24,34})}{\sh(-\epsilon_{1,2,3,4})}\right)^{k}\prod_{I=1}^{k}\prod_{a\in\four}\prod_{\alpha=1}^{n_{\bar{a}}}\frac{\sh(\phi_{I}-\mathfrak{a}_{\bar{a},\alpha}-\epsilon_{a})}{\sh(\phi_{I}-\mathfrak{a}_{\bar{a},\alpha})}\\
&\qquad \times \prod_{I\neq J}^{k}\frac{\sh(\phi_{I}-\phi_{J})\sh(\phi_{I}-\phi_{J}-\epsilon_{14,24,34})}{\sh(\phi_{I}-\phi_{J}-\epsilon_{1,2,3,4})}.
\eea
\end{proposition}

\paragraph{Spiked instanton}We can similarly do the same procedure for the spiked instanton case. Since the quiver diagram is complicated, let us first focus on the case for the setup with only one stack of D4$_{12}$-branes:
\bea
\adjustbox{valign=c}{\begin{tikzpicture}[decoration={markings,mark=at position \arrowHeadPosition with {\arrow{latex}}}]
 \tikzset{
        box/.style={draw, minimum width=0.7cm, minimum height=0.7cm, text centered,thick},
        ->-/.style={decoration={
        markings,mark=at position #1 with {\arrow[scale=1.5]{>}}},postaction={decorate},line width=0.5mm},
        -<-/.style={decoration={
        markings,
        mark=at position #1 with {\arrow[scale=1.5]{<}}},postaction={decorate},line width=0.5mm}    
    }
\begin{scope}[xshift=4cm]
    \draw[fill=black!10!white,thick](0,0) circle(0.4cm);
    \node at (0,0) {$k$};
    \node[box,fill=black!10!white] at (0,1.6) {$n_{12}$};
    \draw[postaction={decorate},thick] (-0.2,1.25)--(-0.2,0.3);
    \draw[postaction={decorate},red,thick] (-0.07,1.25)--(-0.07,0.4);
    \draw[postaction={decorate},thick] (0.2,0.3)--(0.2,1.25);
    \draw[postaction={decorate},red,thick] (0.07,0.4)--(0.07,1.25);
    \foreach \ang in {90,145,215,270} {
    \begin{scope}[rotate=\ang]
        \chiralarc[postaction={decorate},thick](0,0.5)(-45:225:0.22:0.65)
    \end{scope}
    }
    \foreach \ang in {90,145,270} {
    \begin{scope}[rotate=\ang]
    \fermiarc[postaction={decorate},thick](0,0.5)(-45:225:0.1:0.5)
    \end{scope}
    \node[left] at (-0.2,0.8) {$1,\textcolor{red}{q_{4}}$};
    \node[right] at (0.2,0.8) {$q_{12},\textcolor{red}{q_{124}=q_{3}^{-1}}$};
    \node[left] at (-1.5,0) {$q_{2},\textcolor{red}{q_{24}}$};
    \node[right] at (1.6,0) {$q_{1},\textcolor{red}{q_{14}}$};
    \node[below left] at (-0.9,-1){$q_{3},\textcolor{red}{q_{34}}$};
    \node[below right] at (0.9,-1){$q_{4}$};
    \draw[fill=black!10!white,thick](0,0) circle(0.4cm);
    \node at (0,0) {$k$};
    
    }
\end{scope}
\end{tikzpicture}}\quad {\rightsquigarrow} \quad \begin{dcases}
    V:\,\prod_{I=1}^{k}\frac{d\phi_{I}}{2\pi i}\prod_{I\neq J}\sh(\phi_{I}-\phi_{J}),\quad \mathsf{B}_{a}:\,\prod_{I,J}\frac{1}{\sh(\phi_{I}-\phi_{J}-\epsilon_{a})}\\
    \Lambda_{i}:\, \prod_{I,J=1}^{k}\sh(\phi_{I}-\phi_{J}-\epsilon_{i}-\epsilon_{4}),\\
\mathsf{I}:\prod_{I=1}^{k}\prod_{\alpha=1}^{n_{\bar{a}}}\frac{1}{\sh(\phi_{I}-\mathfrak{a}_{\bar{a},\alpha})},\quad \Lambda_{\mathsf{I}}:\,\prod_{I=1}^{k}\prod_{\alpha=1}^{n_{12}}\sh(\phi_{I}-\mathfrak{a}_{\alpha}-\epsilon_{4})\\
\mathsf{J}:\prod_{I=1}^{k}\prod_{\alpha=1}^{n_{12}}\frac{1}{\sh(\mathfrak{a}_{\alpha}-\epsilon_{12}-\phi_{I})} ,\quad \Lambda_{\mathsf{J}}:\prod_{I=1}^{k}\prod_{\alpha=1}^{n_{12}}\sh(\mathfrak{a}_{\alpha}-\phi_{I}+\epsilon_{3})
\end{dcases}
\eea
The partition function obtained here is the 5d $\mathcal{N}=1^{\ast}$ gauge theory on $\mathbb{C}^{2}_{12}\times \mathbb{S}^{1}$ with adjoint mass $\epsilon_{3}$ (the $\widehat{A}_{0}$ quiver gauge theory). 

Generally, the partition function of the spiked instanton is given as follows.
\begin{proposition}\label{prop:spikedcontourJK}
    The partition function of the spiked instanton setup is 
    \bea
    \mathcal{Z}_{\text{inst.}}^{\D4}=\sum_{k=0}^{\infty}\mathfrak{q}^{k}\mathcal{Z}_{k}^{\D4},\quad \mathcal{Z}_{k}^{\D4}=\frac{1}{k!}\oint_{\text{JK}} \prod_{I=1}^{k}\frac{d\phi_{I}}{2\pi i }\,\,\mu_{k}^{\D4\tbar\D0}(\mathfrak{a}_{A,\alpha},\phi_{I})
    \eea    
    where
    \bea\label{eq:spikedcontourJK}
    \mu_{k}^{\D4\tbar\D0}(\mathfrak{a}_{A,\alpha},\phi_{I})&=\left(\frac{\sh(-\epsilon_{14,24,34})}{\sh(-\epsilon_{1,2,3,4})}\right)^{k}\prod_{I=1}^{k}\prod_{A=\overbar{(ab)}\in\six}\prod_{\alpha=1}^{n_{A}}\frac{\sh(\phi_{I}-\mathfrak{a}_{A,\alpha}-\epsilon_{b})\sh(\mathfrak{a}_{A,\alpha}-\phi_{I}+\epsilon_{a})}{\sh(\phi_{I}-\mathfrak{a}_{A,\alpha})\sh(\mathfrak{a}_{A,\alpha}-\phi_{I}-\epsilon_{A})}\\
&\qquad \times \prod_{I\neq J}^{k}\frac{\sh(\phi_{I}-\phi_{J})\sh(\phi_{I}-\phi_{J}-\epsilon_{14,24,34})}{\sh(\phi_{I}-\phi_{J}-\epsilon_{1,2,3,4})}.
    \eea
\end{proposition}

\section{Gauge origami and multi-dimensional partitions}\label{sec:gaugeorigami-multidim-partition}
We collect basic properties of multi-dimensional partitions in section~\ref{sec:multi-dim-part}. Classifications of the poles are given in the following sections.

\subsection{Multi-dimensional partitions}\label{sec:multi-dim-part}
Before evaluating the contour integral formulas for the gauge origami partition functions, let us summarize some notations regarding multi-dimensional partitions. See \cite{Nekrasov:2017cih,Nekrasov:2018xsb,Nekrasov:2023nai} and also \cite[Sec.~2]{Kimura:2023bxy} for details.

\paragraph{Young diagrams}Generalizing the coordinates given in \eqref{eq:2dpartition-coordinate}, we introduce coordinates and multiplicative coordinates as
\bea
c_{ab}(\Bbox)&=\mathfrak{a}+(i-1)\epsilon_{a}+(j-1)\epsilon_{b},\\
\chi_{ab,u}(\Bbox)&=e^{c_{ab}(\Abox)}=uq_{a}^{i-1}q_{b}^{j-1},
\eea
where $u=e^{\mathfrak{a}}$. We call the multiplicative coordinates as $q$-coordinates.

\paragraph{Plane partitions}
The plane partition is a stack of cubes that obeys a generalization of the condition \eqref{eq:Youngcond}:
\bea\label{eq:planepartitionfigure}
\adjustbox{valign=c}{\includegraphics[width=5cm]{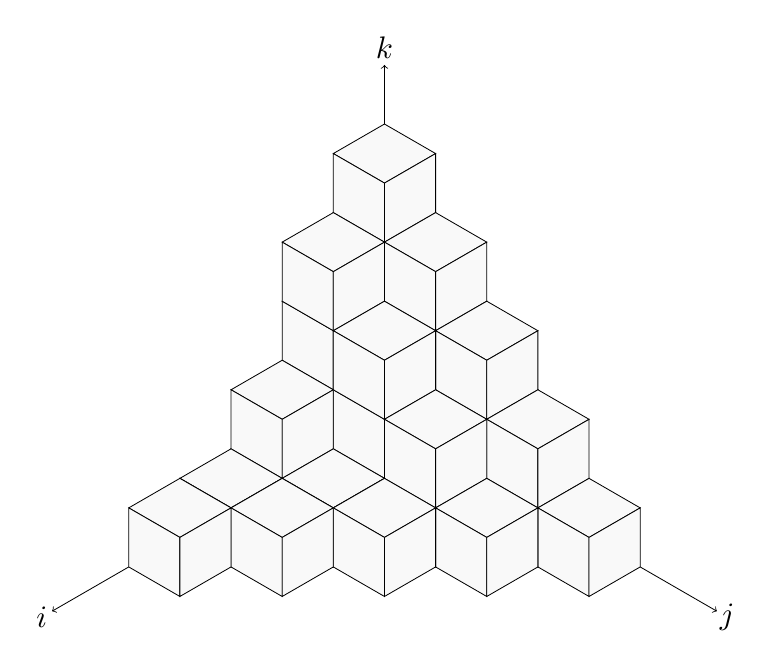}}
\eea
We denote the set of all possible plane partitions as $\mathcal{PP}$. 

There are two ways to describe the plane partition: $(2,1)$-type and $(1,2)$-type. The $(2,1)$-type description is to understand the plane partition $\pi$ as a 2d partition $\lambda_{\pi}$ where there is a map mapping each box $\Bbox=(i,j)\in\lambda_{\pi}$ a number $\pi_{i,j}$ obeying the condition
\bea\label{eq:planepartcond}
\adjustbox{valign=c}{\includegraphics[width=5cm]{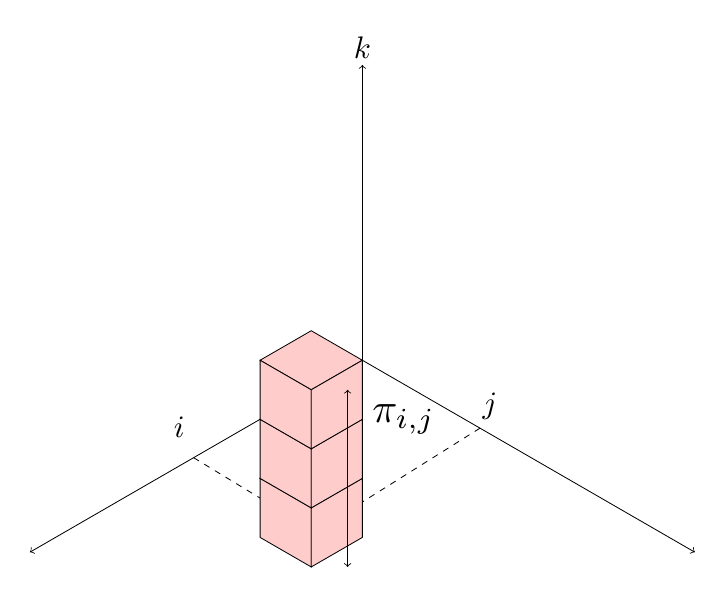}}\qquad \pi_{i,j}\geq \pi_{i+1,j},\quad \pi_{i,j}\geq \pi_{i,j+1}.
\eea
The size of the plane partition is defined as the number of cubes $|\pi|=\sum_{i,j}\pi_{i,j}$. Like the Young diagram case, a cube $\cube$ in the plane partition $\pi$ is assigned a coordinate $(i,j,k)\,(i,j,k\geq 1)$ as \eqref{eq:planepartitionfigure} and it obeys
\begin{equation}
    (i,j,k)\in\pi\,\,\Leftrightarrow \,\, 1\leq k\leq \pi_{i,j}.
\end{equation}

The other $(1,2)$-type description is to understand the plane partition as a non-increasing sequence of Young diagrams $\pi=(\Lambda^{(1)},\Lambda^{(2)},\ldots,\Lambda^{(h(\pi))})$:
\bea
\adjustbox{valign=c}{\includegraphics[width=5cm]{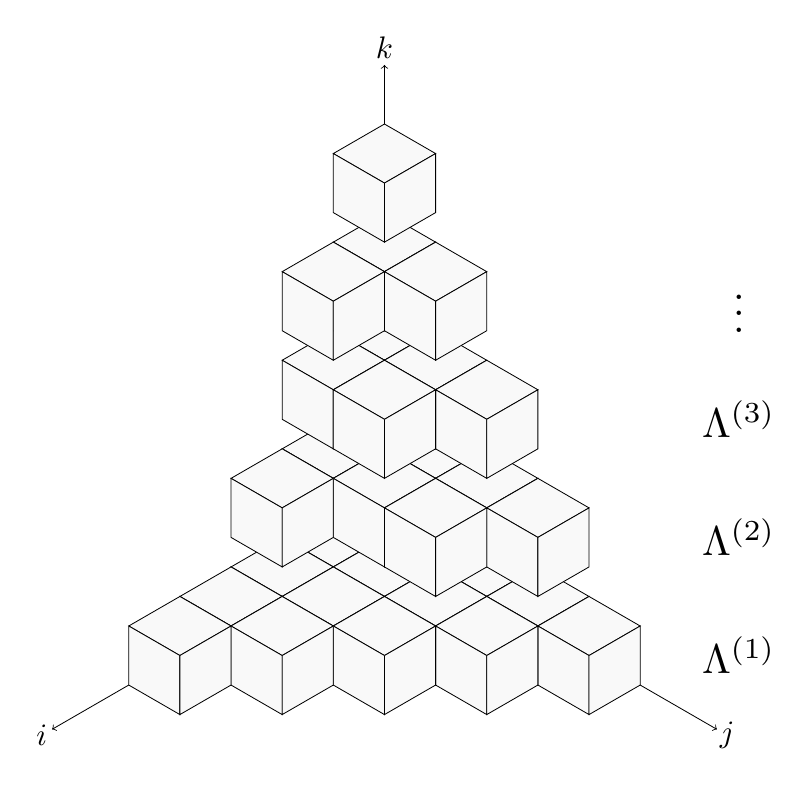}}\quad \quad \begin{array}{l}
\Lambda^{(k)}=(\Lambda_{1}^{(k)},\ldots,\Lambda_{i}^{(k)}\ldots)\\
\Lambda^{(k)}\succeq\Lambda^{(k+1)},\,\,\forall k
\end{array}
\eea
where $\Lambda^{(k)}\succeq \Lambda^{(k+1)}$ means $\forall (i,j)\in\Lambda^{(k+1)}\Rightarrow (i,j)\in\Lambda^{(k)}$ and $h(\pi)$ is the height of the plane partition which is defined as 
\begin{equation}
    h(\pi)=\min\{k\geq 0\,|\,(1,1,k+1)\not\in\pi\}.
\end{equation}
The size in this description is given as $|\pi|=\sum_{k}|\Lambda^{(k)}|$. Similarly, the $\cube=(i,j,k)$ in the plane partition obeys the condition
\begin{equation}
    (i,j,k)\in\pi\,\,\Leftrightarrow \,\, 1\leq j\leq \Lambda_{i}^{(k)}.
\end{equation}
\begin{figure}
    \centering
    \includegraphics[width=0.5\linewidth]{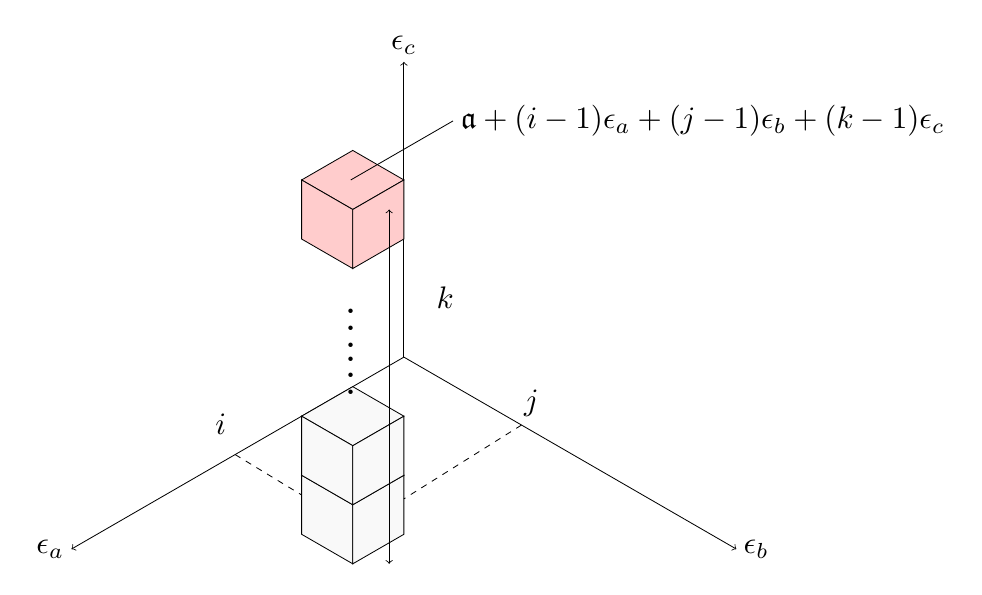}
    \caption{Coordinates of the plane partition}
    \label{fig:q-coord-planepartition}
\end{figure}
For later use, we introduce coordinates to each cubes in the plane partition. Given three parameters $\epsilon_{a,b,c}$, for a box in the plane partition $\cube=(i,j,k)$ we define
\bea
c_{abc}(\cube)=c_{abc}(\mathfrak{a},(i,j,k))=\mathfrak{a}+(i-1)\epsilon_{a}+(j-1)\epsilon_{b}+(k-1)\epsilon_{c}
\eea
where $\mathfrak{a}$ is some complex parameter. We also define the $q$-coordinates as
\bea
\chi_{abc,u}(\cube)=e^{c_{abc}(\scube)}=uq_{a}^{i-1}q_{b}^{j-1}q_{c}^{k-1}.
\eea

\paragraph{Solid partitions}A solid partition is a four-dimensional analog of the Young diagram and plane partition. It is a stack of hyper-cubes obeying similar conditions to \eqref{eq:Youngcond} and \eqref{eq:planepartcond}. We denote the set of all possible solid partitions as $\mathcal{SP}$. We have three ways to describe the solid partition: $(3,1)$, $(2,2)$, and $(1,3)$-types. We only discuss the $(3,1)$ and $(1,3)$-type descriptions.
\begin{enumerate}
    \item $(3,1)$-type:
    This description is similar to the plane partition's $(2,1)$-type description. We project the solid partition to a plane partition $\pi_{\rho}$ and for each cube $(i,j,k)\in\pi_{\rho}$, a height function $\rho_{i,j,k}$ is defined. The height function obeys the condition
    \begin{equation}
        \rho_{i,j,k}\geq \rho_{i+1,j,k},\quad \rho_{i,j,k}\geq \rho_{i,j+1,k},\quad \rho_{i,j,k}\geq \rho_{i,j,k+1}.
    \end{equation}
    The size is defined as
    \begin{equation}
        |\rho|=\sum_{(i,j,k)\in\pi_{\rho}}\rho_{i,j,k}.
    \end{equation}
    4-cubes $\hcube$ in the solid partition are assigned coordinates in a natural way $(i,j,k,l)$ and obey
    \begin{equation}
    (i,j,k,l)\in \rho \,\,\Leftrightarrow\,\, 1\leq l\leq \rho_{i,j,k}.
    \end{equation}
    
    \item $(1,3)$-type: 
    \begin{figure}
        \centering
        \includegraphics[width=13cm]{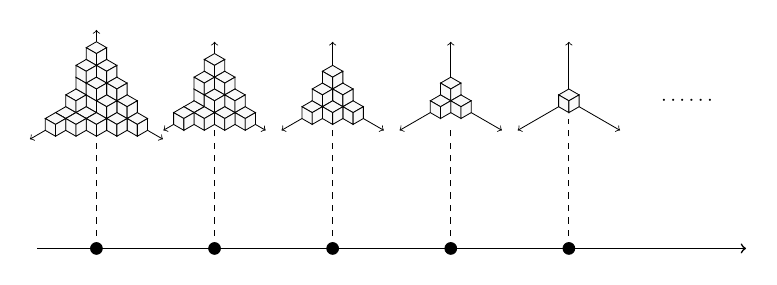}
        \caption{$(1,3)$-type description of the solid partition. The horizontal axis is one of the four axes of the solid partition. The solid partition is decomposed into multiple plane partitions $\Pi^{(1)},\Pi^{(2)},\ldots$.}
        \label{fig:solidplanedecomp}
    \end{figure}
    This description resembles the $(1,2)$-type description of plane partition. The solid partition is understood as non-increasing sequences of plane partitions (see Figure \ref{fig:solidplanedecomp}):
    \begin{equation}
        \rho=(\Pi^{(1)},\Pi^{(2)},\ldots,),\quad \Pi^{(l)}\succeq\Pi^{(l+1)}
    \end{equation}
    where $\Pi^{(l)}\succeq \Pi^{(l+1)}$ means
    \begin{equation}
    (i,j,k)\in \Pi^{(l+1)}\Rightarrow (i,j,k)\in\Pi^{(l)}.
    \end{equation}
    Under this description, the size is defined as
    \begin{equation}
        |\rho|=\sum_{l}|\Pi^{(l)}|
    \end{equation}
    and obviously,
    \begin{equation}
        (i,j,k,l)\in\rho\,\,\Leftrightarrow \,\, (i,j,k)\in\Pi^{(l)}.
    \end{equation}
   
\end{enumerate}

Similar to the previous cases, we define coordinates to the hypercubes in the solid partition as
\bea
c_{\four}(\hcube)&=c_{\four}(\mathfrak{a},(i,j,k,l))=\mathfrak{a}+(i-1)\epsilon_{1}+(j-1)\epsilon_{2}+(k-1)\epsilon_{3}+(l-1)\epsilon_{4},\\
\chi_{\four,u}(\hcube)&=e^{c_{\four}(\shcube)}=uq_{1}^{i-1}q_{2}^{j-1}q_{3}^{k-1}q_{4}^{l-1}.
\eea

\paragraph{Ordering of boxes}
For later use, let us define an ordering in the boxes included in the multi-dimensional partitions. The ordering for boxes in the solid partition is defined as
\bea\label{eq:ordering-partition}
\,&(i,j,k,l)<(i',j',k',l')\\
\Leftrightarrow &(l<l')\vee (l=l', k<k') \vee (l=l', k=k', j<j')\vee (l=l', k=k', j=j', i<i').
\eea
This ordering define a monomial ordering in the $q$-coordinates $q_{1}^{i-1}q_{2}^{j-1}q_{3}^{k-1}q_{4}^{l-1}$.

For the orderings of the boxes in the plane partition, we first embed the plane partition in the solid partition and induce the ordering defined on the solid partitions. There are three ways to embed the plane partition in the solid partition depending on which three axes the plane partition spans. When considering plane partitions, we implicitly assume that this embedding is also specified. We denote the set of plane partitions not extending in the $a\in\four$ direction as $\mathcal{PP}_{a}=\mathcal{PP}_{\bar{a}}$.

Similarly, for Young diagrams, we embed them into the solid partition and induce the ordering. There are six ways to embed the Young diagram to the solid partition depending on which two axes the Young diagram spans. Similarly, when considering the Young diagrams, we implicitly assume that this embedding is also specified. The set of Young diagrams extending in the $ab\in\six$ directions is denoted as $\mathcal{P}_{A}$ $(A\in\six)$.

\subsection{Magnificent four and solid partitions}
Let us evaluate \eqref{eq:M4contourJK} using the JK residue formalism. The discussion here is based on \cite{Nekrasov:2018xsb}. At the end of evaluation, we will see that similar to the pure SYM case, the poles giving non-zero JK-residues are classified by solid partitions. 

The denominator in \eqref{eq:M4contourJK} shows that the poles come from the hyperplanes
\bea
H_{I,\alpha}=\{\phi_{I}-\mathfrak{a}_{\alpha}=0\},\quad H_{IJ,a=1,2,3,4}=\{\phi_{I}-\phi_{J}-\epsilon_{a}=0\},
\eea
where the corresponding charge vectors are now $\{\bfe_{I}\}$ and $\{\bfe_{I}-\bfe_{J}\}$. Compared to the pure super Yang--Mills case in section~\ref{sec:pureSYM_JK}, the difference is that we do not have $\{-\bfe_{I}\}$ and thus Lemma~\ref{lem:JKantifund} is automatically satisfied.\footnote{This also means that if we set the reference vector to be $\eta=(-1,\cdots ,-1)$, the partition function vanishes, which is a property different from the pure SYM case. For the pure SYM case, even if we set the reference vector to be $\eta=(-1,\cdots, -1)$, the partition function does not vanish and actually we still have the same partition function.} Lemma~\ref{lem:JKtree-structure} tells us that the poles come from an oriented tree made by $\{\bfe_{I}\}$ and $\{\bfe_{I}-\bfe_{J}\}$. The main claim of this section is that the poles will be classified by solid partitions.

\begin{theorem}\label{thm:M4JKpoles}
    The poles of \eqref{eq:M4contourJK} are classified by $n$-tuples of solid partitions $\vec{\rho}=\{\rho^{(\alpha)}\}_{\alpha=1}^{n}$:
    \bea
    \{\phi_{I}\}_{I=1}^{k}\rightarrow \{c_{\four}(\hcube)\mid \hcube=(x_{1},x_{2},x_{3},x_{4})\in\rho^{(\alpha)}\}_{\alpha=1}^{n}
    \eea
    where $c_{\four}(\hcube)=\mathfrak{a}_{\alpha}+\sum_{a\in\four}(x_{a}-1)\epsilon_{a}$ and $k$ corresponds to the total number of boxes in the $n$-tuples of solid partitions: $k=\sum_{\alpha=1}^{n}|\rho^{(\alpha)}|$.
\end{theorem}
\begin{proof}
    We show this by induction. For the $k=1$ case, the pole will simply come from $\phi_{1}-\mathfrak{a}_{\alpha}=0$ for some $\alpha$ and it corresponds to configuration with one nonempty solid partition with only one box. Assume that the $k$ independent hyperplane equations determining the poles are $\{\bfQ_{1},\ldots \bfQ_{k}\}$ and they form a tree structure as Lemma~\ref{lem:JKtree-structure}. By the induction hypothesis, the $k-1$ poles are determined by the solid partition rule. We need to show that under these assumptions, the $k$-th pole will also obey the solid partition rule.

    When the trees have no branches, all the poles will be determined by $\{\phi_{I}-\mathfrak{a}_{\alpha}=0\}$. This is a situation when we have $k$ nonempty solid partitions where each of them has only one box. This case indeed obeys the solid partition rule.
    
    Let us consider the case when $\bfQ_{k}=\bfe_{I}-\bfe_{J}$ for some $I,J$ is a branch at the end of tree, and when the corresponding hyperplane determine the pole of $\phi_{I}$. Using the Weyl invariance, we can assume that $\{\phi_{1},\ldots,\phi_{k-1}\}$ are determined by $\{\bfQ_{1},\ldots, \bfQ_{k-1}\}$ and $\phi_{k}$ is the last pole determined by $\bfQ_{k}=\bfe_{k}-\bfe_{J}$ for some $J$. The induction hypothesis tells us that $\{\phi_{1},\ldots,\phi_{k-1}\}$ are classified by solid partitions $\vec{\rho}$ with $k-1$ boxes. Let us show that the possible choices of $\bfQ_{k}$ correspond to the possible choices to a add box to $\vec{\rho}$ giving ${\vec{\rho}\,'}=\vec{\rho}+\hcube$.

    We assume that $\bfe_{k}$ belongs to a tree whose root vertex corresponds to the pole at $\mathfrak{a}_{\ast}$ and denote the corresponding solid partitions $\rho_{\ast}, \rho'_{\ast}=\rho+\hcube$. When evaluating the pole $\phi_{k}$, poles at $\{\phi_{1},\ldots,\phi_{k-1}\}=\vec{\rho}$ are inserted and the related factor is 
    \bea
    \frac{\sh(\phi_{k}-\mathfrak{b}_{\ast})}{\sh(\phi_{k}-\mathfrak{a}_{\ast})}\prod_{\shcube\in\rho_{\ast}}\frac{\sh(\phi_{k}-c_{\four}(\hcube))\sh(\phi_{k}-c_{\four}(\hcube)-\epsilon_{14,24,34})}{\sh(\phi_{k}-c_{\four}(\hcube)-\epsilon_{1,2,3,4})}\\
    \times \prod_{\shcube\in\rho_{\ast}}\frac{\sh(c_{\four}(\hcube)-\phi_{k})\sh(c_{\four}(\hcube)-\phi_{k}-\epsilon_{14,24,34})}{\sh(c_{\four}(\hcube)-\phi_{k}-\epsilon_{1,2,3,4})},
    \eea
    where $c_{\four}(\hcube)=\mathfrak{a}_{\ast}+\sum_{a\in\four}(x_{a}-1)\epsilon_{a}$. Note that other factors appearing from \eqref{eq:M4contourJK} are regular and will not affect the pole structure and thus we omit them.
 
    Since $\bfQ_{k}=\bfe_{k}-\bfe_{J}$ determines the last pole $\phi_{k}$, the pole we are interested takes the form as $\phi_{k}=\phi_{J}+\epsilon_{1,2,3,4}$. For later use, we denote the position of the boxes corresponding to $\phi_{k},\phi_{J}$ as $x'=(x'_{a})_{a\in\four},\,x=(x_{a})_{a\in\four}$, respectively. We also note that the factor above is invariant under the quadrality of $\epsilon_{a}$ under the condition $\sum_{a\in\four}\epsilon_{a}=0$, because of
    \bea
       \prod_{\shcube\in\rho_{\ast}}\sh(\phi_{k}-c_{\four}(\hcube)-\epsilon_{14,24,34})\sh(c_{\four}(\hcube)-\phi_{k}-\epsilon_{14,24,34}) \propto  \prod_{\shcube\in\rho_{\ast}}\sh(\phi_{k}-c_{\four}(\hcube)-\epsilon_{14,24,34,12,23,13}).
    \eea
    
    Under this situation, we have the following properties.
    \begin{itemize}[topsep=1ex,itemsep=-0.5ex,partopsep=1ex,parsep=1ex]
        \item $\phi_{J}$ needs to be in the boundary\footnote{This means that $\phi_{J}+\epsilon_{a}$ for some $a$ will not be in the solid partition.} and $\phi_{k}$ is not included in the solid partition. If $\phi_{k}$ is already in the solid partition, then $\phi_{k}=\phi_{J}+\epsilon_{c}=\phi_{J'}$ for some $J'<k$. Namely, we have $x+\delta_{c}\in\rho_{\ast}$, where $\delta_{c}$ is a four dimensional vector which has $1$ at its $c$-component and the other components are zero\footnote{For example, $\delta_{4}=(0,0,0,1)$.}. The pole here is canceled by the numerator $\prod_{\shcube\in\rho_{\ast}}\sh(\phi_{k}-c_{\four}(\hcube))^{2}$:
        \bea
        \frac{\sh(\phi_{k}-c_{\four}(\mathfrak{a}_{\ast},x+\delta_{c}))\sh(c_{\four}(\mathfrak{a}_{\ast},x+\delta_{c})-\phi_{k})}{\sh(\phi_{k}-c_{\four}(\mathfrak{a}_{\ast},x)-\epsilon_{c})}\rightarrow 0.
        \eea

        \item Using the quadrality symmetry, let us fix the pole structure as $\phi_{k}=\phi_{J}+\epsilon_{4}$ and $(x'_{1},x'_{2},x'_{3},x'_{4})=(x_{1},x_{2},x_{3},x_{4}+1)$. When $x\in\rho_{\ast}$ belongs to the boundary and $x_{1}=x_{2}=x_{3}=1,\,x_{4}>1$, no pole cancellation occurs and we have a single pole. Namely, we can always add a box to the configuration. This situation is illustrated in the $(1,3)$-type description as
        \bea
        \adjustbox{valign=c}{\includegraphics[width=8cm]{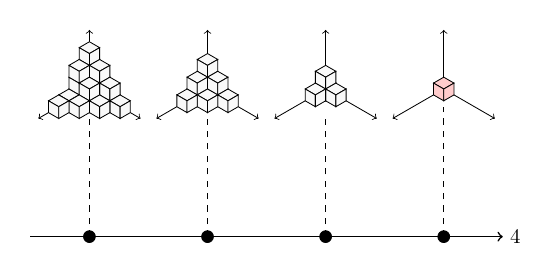}}
        \eea
        where the red box corresponds to the pole $\phi_{k}=c_{\four}(\mathfrak{a},(1,1,1,x_{4}))+\epsilon_{4}$. 
        \item When $x\in\rho_{\ast}$ belongs to the boundary and $x_{1}=x_{2}=1,x_{3}>1,x_{4}>1$ (other situations are obtained by the triality between $x_{1,2,3}$), the induction hypothesis tells us that $\rho_{\ast}$ is a solid partition and thus we have $(1,1,x_{3}-1,x_{4})\in\rho_{\ast}$. Note that this means that
        \bea
    \phi_{k}=c_{\four}(\mathfrak{a}_{\ast},x)+\epsilon_{4}=c_{\four}(\mathfrak{a}_{\ast},(1,1,x_{3}-1,x_{4}))+\epsilon_{3}+\epsilon_{4}.
        \eea
        We have two possibilities of the position of $(x_{3},x_{4}+1)$.
        \begin{itemize}
            \item If $(1,1,x_{3}-1,x_{4}+1)\notin\rho_{\ast}$, i.e. $\rho'_{\ast}$ is not a solid partition, then the pole vanishes and the JK residue is zero.
            \item If $(1,1,x_{3}-1,x_{4}+1)\in\rho_{\ast}$, i.e. $\rho'_{\ast}$ is a solid partition, then there is a single pole and the JK residue is non-zero.
        \end{itemize}
        This comes from
        \bea
        &\frac{\sh(\phi_{k}-c_{\four}(\mathfrak{a}_{\ast},(1,1,x_{3}-1,x_{4}))-\epsilon_{34})}{\sh(\phi_{k}-c_{\four}(\mathfrak{a}_{\ast},x)-\epsilon_{4})}\times \begin{dcases}
            1,\\ 
            \frac{1}{\sh(\phi_{k}-c_{\four}(\mathfrak{a}_{\ast},(1,1,x_{3}-1,x_{4}+1))-\epsilon_{3})} 
        \end{dcases}\\
        =&\begin{dcases}
        1,\quad (1,1,x_{3}-1,x_{4}+1)\notin\rho_{\ast}\\
        \frac{1}{\sh(\phi_{k}-c_{\four}(\mathfrak{a}_{\ast},x)-\epsilon_{4})},\quad (1,1,x_{3}-1,x_{4}+1)\in\rho_{\ast}
        \end{dcases}
        \eea
        where we used
        \bea
        c_{\four}(\mathfrak{a}_{\ast},x)+\epsilon_{4}=c_{\four}(\mathfrak{a}_{\ast},(1,1,x_{3}-1,x_{4}+1))+\epsilon_{3}.
        \eea
        This situation is illustrated as
        \bea
        \adjustbox{valign=c}{\includegraphics[width=8cm]{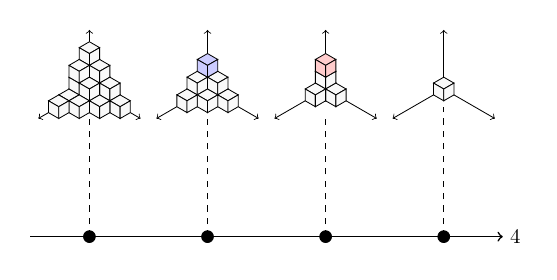}}
        \eea
        where the blue and red boxes correspond to the poles $\phi_{J}=c_{\four}(\mathfrak{a},(1,1,x_{3},x_{4}))$ and $\phi_{k}=c_{\four}(\mathfrak{a},(1,1,x_{3},x_{4}))+\epsilon_{4}$, respectively.

        \item When $x\in\rho_{\ast}$ belongs to the boundary and $x_{1}=1,x_{2,3,4}>1$ (other situations are obtained by the triality between $x_{1,2,3}$), the induction hypothesis gives $(1,x_{2}-1,x_{3},x_{4}), (1,x_{2},x_{3}-1,x_{4})\in\rho_{\ast}$, which means
        \bea
    \phi_{k}=c_{\four}(\mathfrak{a}_{\ast},x)+\epsilon_{4}=c_{\four}(\mathfrak{a}_{\ast},x-\delta_{2})+\epsilon_{24}=c_{\four}(\mathfrak{a},x-\delta_{3})+\epsilon_{34}.
        \eea
        The partition $\rho'_{\ast}$ is a solid partition if and only if
        \bea\label{eq:M4case3cond}
        (1,x_{2}-1,x_{3},x_{4}+1),\quad (1,x_{2},x_{3}-1,x_{4}+1)\in\rho_{\ast}.
        \eea
        Unless this situation, the JK residue will be zero:
        \bea
        &\frac{\sh(\phi_{k}-c_{\four}(\mathfrak{a}_{\ast},x-\delta_{2})-\epsilon_{24})\sh(\phi_{k}-c_{\four}(\mathfrak{a}_{\ast},x-\delta_{3})-\epsilon_{34})}{\sh(\phi_{k}-c_{\four}(\mathfrak{a}_{\ast},x)-\epsilon_{4})}\times \textcolor{blue}{\frac{1}{\sh(\phi_{k}-c_{\four}(\mathfrak{a}_{\ast},x-\delta_{2}+\delta_{4})-\epsilon_{2})}}\\
        &\times\textcolor{blue}{\frac{1}{\sh(\phi_{k}-c_{\four}(\mathfrak{a}_{\ast},x-\delta_{3}+\delta_{4})-\epsilon_{3})}}=\frac{1}{\sh(\phi_{k}-c_{\four}(\mathfrak{a}_{\ast},x)-\epsilon_{4})}
        \eea
        because the two numerators are canceled only if we have the blue terms which come from the condition \eqref{eq:M4case3cond}. This situation is illustrated as
        \bea
        \adjustbox{valign=c}{\includegraphics[width=8cm]{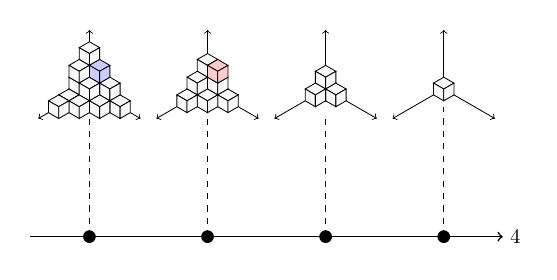}}
        \eea
         where the blue and red boxes correspond to the poles $\phi_{J}=c_{\four}(\mathfrak{a},x)$ and $\phi_{k}=c_{\four}(\mathfrak{a},x+\delta_{4})$, respectively.

        \item When $x\in\rho_{\ast}$ belongs to the boundary and $x_{1,2,3,4}>1$, the induction hypothesis gives
        \bea
        (x_{1}-1,x_{2},x_{3},x_{4}),\,\, (x_{1},x_{2}-1,x_{3},x_{4}),\,\,(x_{1},x_{2},x_{3}-1,x_{4})\in\rho_{\ast}
        \eea
        and we have
        \bea
        \phi_{k}=c_{\four}(\mathfrak{a}_{\ast},x)+\epsilon_{4}=c_{\four}(\mathfrak{a}_{\ast},x-\delta_{1})+\epsilon_{14}=c_{\four}(\mathfrak{a}_{\ast},x-\delta_{2})+\epsilon_{24}=c_{\four}(\mathfrak{a}_{\ast},x-\delta_{3})+\epsilon_{34}.
        \eea
        The partition $\rho'_{\ast}$ is a solid partition if and only if
        \bea\label{eq:M4case4cond}
        (x_{1}-1,x_{2},x_{3},x_{4}+1),\,\,(x_{1},x_{2}-1,x_{3},x_{4}+1),\,\,(x_{1},x_{2},x_{3}-1,x_{4}+1)\in\rho_{\ast}
        \eea
        and unless this situation, the JK residue will be zero:
        \bea
    &\frac{\prod\limits_{a=1}^{3}\sh(\phi_{k}-c_{\four}(\mathfrak{a}_{\ast},x-\delta_{a})-\epsilon_{a4})}{\sh(\phi_{k}-c_{\four}(\mathfrak{a}_{\ast},x)-\epsilon_{4}}\times \textcolor{blue}{\frac{1}{\prod\limits_{a=1}^{3}\sh(\phi_{k}-c_{\four}(\mathfrak{a}_{\ast},x-\delta_{a}+\delta_{4})-\epsilon_{a})}}=\frac{1}{\sh(\phi_{k}-c_{\four}(\mathfrak{a}_{\ast},x)-\epsilon_{4})}
        \eea
        because the blue terms that cancels the numerator only appear under the condition \eqref{eq:M4case4cond}. This situation is illustrated as 
        \bea
        \adjustbox{valign=c}{\includegraphics[width=8cm]{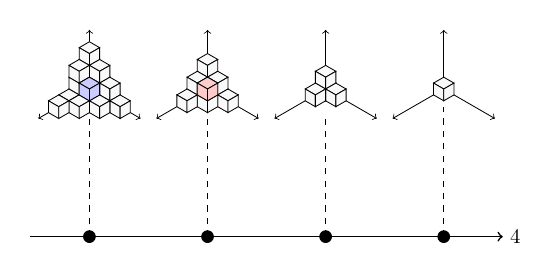}}
        \eea
         where the blue and red boxes correspond to the poles $\phi_{J}=c_{\four}(\mathfrak{a},x)$ and $\phi_{k}=c_{\four}(\mathfrak{a},x+\delta_{4})$, respectively.
    \end{itemize}
This concludes the inductive proof and the non-zero JK residues are indeed classified by $n$-tuples of solid partitions.

\end{proof}

Similar to the pure SYM case, from the proof of the classification, the JK-residue can be written as an iterative residue. For simplicity, let us focus on the $n=1$ case. Thm.~\ref{thm:M4JKpoles} claims that the poles giving the $k$-instanton sector are classified by a solid partition with $k$-boxes. The inductive process in $k-1\rightarrow k$ shows that after evaluating the residue at the level $k-1$, the position of the remaining pole at level $k$ is given by the possible positions to add a box to the solid partition. Moreover, each iterative process will obey the solid partition rule. Concretely, let $\{\phi_{1\ast},\ldots, \phi_{k\ast}\}$ be the sequence of the coordinates of the $k$-boxes in the solid partition ordered in the order they are stacked to the solid partition. Each subset $\{\phi_{1\ast},\ldots ,\phi_{i\ast}\}$ is a solid partition with $i$-boxes and the contour integral can then be evaluated as $\oint_{\phi_{k}=\phi_{k\ast}} d\phi_{k}\cdots \oint_{\phi_{1}=\phi_{1\ast}} d\phi_{1}$. Given a solid partition $\rho$, we may have multiple ways to stack boxes by keeping the solid partition condition for each step. However, the result does not depend on them. Thus, it is useful to fix one ordering to evaluate the contour integral and we choose the ordering given in \eqref{eq:ordering-partition}. Obviously, the ordering given there preserves the solid partition condition for each step. The $\U(1|1)$ magnificent four partition is then given as 
\bea
\mathcal{Z}^{\D8}_{\text{inst.}}=\sum_{\rho\in\mathcal{SP}}\mathfrak{q}^{|\rho|}\mathcal{Z}^{\D8}[\rho],\quad \mathcal{Z}^{\D8}[\rho]=\underset{\phi=\phi_{\rho}}{\Res}\mu^{\D8\tbar\D0}_{k}(\mathfrak{a}_{\alpha},\mathfrak{b}_{\alpha},\phi_{I})
\eea
where the iterated residue is defined as
\bea\label{eq:D8-iteratedresidue}
\underset{\phi=\phi_{\rho}}{\Res}\mu^{\D8\tbar\D0}_{k}(\mathfrak{a},\mathfrak{b},\phi_{I})\coloneqq\underset{\phi_{k}=\phi_{k\ast}}{\Res}\cdots \underset{\phi_{2}=\phi_{2\ast}}{\Res}\underset{\phi_{1}=\phi_{1\ast}}{\Res}\mu^{\D8\tbar\D0}_{k}(\mathfrak{a},\mathfrak{b},\phi_{I})
\eea
and $k=|\rho|$. The sequence $\{\phi_{1\ast},\ldots, \phi_{k\ast}\}$ is the coordinates of the boxes in the solid partition ordered in the ordering \eqref{eq:ordering-partition}. For higher rank cases, the discussion in the pure SYM case is straightforwardly applicable so we omit the discussion.

\subsection{Tetrahedron instantons and plane partitions}
Let us evaluate the contour integral \eqref{eq:tetracontourJK}. As mentioned before, this contour integral formula can be obtained from \eqref{eq:M4contourJK} by tuning the parameter $\mathfrak{b}_{\alpha}$ of the $\U(n|n)$ magnificent four setup. We first consider the rank $n=\sum_{a\in\four}n_{\bar{a}}$ magnificent four setup and then take the specialization 
\bea
\{\mathfrak{a}_{\alpha}\}_{\alpha=1}^{n}\rightarrow \{\mathfrak{a}_{\bar{a},\alpha}\}_{a\in\four,\alpha=1,\ldots,n_{\bar{a}}},\quad \{\mathfrak{b}_{\beta}\}_{\beta=1}^{n}\rightarrow \{\mathfrak{a}_{\bar{a},\alpha}+\epsilon_{a}\}_{a\in\four,\alpha=1,\ldots, n_{\bar{a}}}.
\eea

The classification of the poles for this setup was done in \cite{Pomoni:2021hkn}. Similar to the magnificent four setup, the poles come from the hyperplanes
\bea
H_{I,\bar{a},\alpha}=\{\phi_{I}-\mathfrak{a}_{\bar{a},\alpha}=0\},\quad H_{IJ,a=1,2,3,4}=\{\phi_{I}-\phi_{J}-\epsilon_{1,2,3,4}=0\}
\eea
and the charge vectors $\{\bfe_{I}\}, \{\bfe_{I}-\bfe_{J}\}$. Again, we do not have the anti-fundamental contributions $\{-\bfe_{I}\}$ and Lemma~\ref{lem:JKantifund} is automatically satisfied. The poles will then be classified by the oriented tree structure given in Lemma~\ref{lem:JKtree-structure}. The poles will actually be classified by plane partitions which are truncations of the solid partitions.
\begin{theorem}\label{thm:tetraJKpoles}
    The poles of \eqref{eq:tetracontourJK} are classified by plane partitions $\underline{\vec{\pi}}=(\vec{\pi}_{a})_{a\in\four}=(\pi^{(\alpha)}_{\bar{a}} )^{\alpha=1,\ldots,n_{\bar{a}}}_{a\in\four}$:
    \bea
    \{\phi_{I}\}_{I=1}^{k}\rightarrow \{c_{\bar{a}}(\cube)\mid \cube=(x_{i},x_{j},x_{k})\in \pi^{(\alpha)}_{\bar{a}},\,\,\,(i,j,k)\in\bar{a}\}_{a\in\four}^{\alpha=1,\ldots, n_{\bar{a}}}
    \eea
    where $c_{\bar{a}}(\cube)=\mathfrak{a}_{\bar{a},\alpha}+\sum_{i\in\bar{a}}(x_{i}-1)\epsilon_{i}$. Namely, the $n_{\bar{a}}$-tuples of plane partitions extend in the three-dimensional subspace $(x_{i},x_{j},x_{k})$, where $\bar{a}=ijk$. The total instanton number $k$ is identified with the number of boxes of the plane partitions $k=\sum_{a\in\four}\sum_{\alpha=1}^{n_{\bar{a}}}|\pi^{(\alpha)}_{\bar{a}}|$.
\end{theorem}
\begin{proof}
    The proof is done similarly as the magnificent four setup and we show the claim by induction. For $k=1$, the pole will come from $\phi_{1}-\mathfrak{a}_{\bar{a},\alpha}=0$ and it corresponds to a configuration with one nonempty plane partition with only one box.
    
    For $k>1$, we assume that the $k$ independent hyperplane equations come from $\{\bfQ_{1},\ldots, \bfQ_{k}\}$ and its associated tree structure. From the induction hypothesis, the $k-1$ poles are determined by the plane partition rule. When the trees have no branches, all the poles will come from $\{\phi_{I}-\mathfrak{a}_{\bar{a},\alpha}=0\}$ and we have $k$ nonempty plane partitions with one box. 

   The generic case is when we have a branch at the end of tree. Using the Weyl invariance, we can assume that $\{\phi_{1},\ldots,\phi_{k-1}\}$ are determined by $\{\bfQ_{1},\ldots,\bfQ_{k}\}$ and $\phi_{k}$ is the last pole determined by $\bfQ_{k}=\bfe_{k}-\bfe_{J}$ for some $J$. The induction hypothesis shows that $\{\phi_{1},\ldots,\phi_{k-1}\}$ are classified by plane partitions $\underline{\vec{\pi}}$ with $k-1$ boxes. We show that the possible choices of $\bfQ_{k}$ correspond to the possible choices to add a box to $\underline{\vec{\pi}}$ giving $\underline{\vec{\pi}}'=\underline{\vec{\pi}}+\cube$. 

   We can assume that $\bfe_{k}$ belongs to a tree whose root vertex comes from $\mathfrak{a}_{\ast}=\mathfrak{a}_{\bar{4},\alpha}$ for some $\alpha$ and denote the corresponding plane partitions $\pi_{\ast},\pi_{\ast}'=\pi_{\ast}+\cube$. In this case, the plane partition will span the $123$-subspace and the coordinates will be assigned using $\epsilon_{1,2,3}$. When the root vertex comes from other $\{\mathfrak{a}_{\bar{a},\alpha}\}_{a=1,2,3}$, the plane partition spans different three-dimensional subspace of $1234$ and we can simply use the quadrality symmetry. At the time we evaluate $\phi_{k}$, poles $\{\phi_{1},\ldots,\phi_{k-1}\}=\underline{\vec{\pi}}$ will be inserted and the related factor only comes from contribution of $\pi_{\ast}$ and other contributions will be regular:
    \bea
    \frac{\sh(\phi_{k}-\mathfrak{a}_{\ast}-\epsilon_{4})}{\sh(\phi_{k}-\mathfrak{a}_{\ast})}\prod_{\scube\in\rho_{\ast}}\frac{\sh(\phi_{k}-c_{\bar{4}}(\cube))\sh(\phi_{k}-c_{\bar{4}}(\cube)-\epsilon_{14,24,34})}{\sh(\phi_{k}-c_{\bar{4}}(\cube)-\epsilon_{1,2,3,4})}\\
    \times \prod_{\scube\in\rho_{\ast}}\frac{\sh(c_{\bar{4}}(\cube)-\phi_{k})\sh(c_{\bar{4}}(\cube)-\phi_{k}-\epsilon_{14,24,34})}{\sh(c_{\bar{4}}(\cube)-\phi_{k}-\epsilon_{1,2,3,4})},
    \eea
    where $c_{\bar{4}}(\cube)=\mathfrak{a}_{\ast}+\sum_{i=1}^{3}(x_{i}-1)\epsilon_{i}$. Using $\sum_{a\in\four}\epsilon_{a}=0$, this can be rewritten as
    \bea
    \frac{\sh(\phi_{k}-\mathfrak{a}_{\ast}+\epsilon_{123})}{\sh(\phi_{k}-\mathfrak{a}_{\ast})}\prod_{\scube\in\rho_{\ast}}\frac{\sh(\phi_{k}-c_{\bar{4}}(\cube))\sh(\phi_{k}-c_{\bar{4}}(\cube)+\epsilon_{12,23,31})}{\sh(\phi_{k}-c_{\bar{4}}(\cube)-\epsilon_{1,2,3})\sh(\phi_{k}-c_{\bar{4}}(\cube)+\epsilon_{123})}\\
    \times \prod_{\scube\in\rho_{\ast}}\frac{\sh(c_{\bar{4}}(\cube)-\phi_{k})\sh(c_{\bar{4}}(\cube)-\phi_{k}+\epsilon_{12,23,31})}{\sh(c_{\bar{4}}(\cube)-\phi_{k}-\epsilon_{1,2,3})\sh(c_{\bar{4}}(\cube)-\phi_{k}+\epsilon_{123})}.
    \eea
    The difference with the magnificent four setup is that the numerator $\sh(\phi_{k}-\mathfrak{a}_{\ast}-\epsilon_{4})$ gives an extra zero and some poles in the magnificent four case will be canceled. 

    The charge vector $\bfQ_{k}=\bfe_{k}-\bfe_{J}$ determines the pole as $\phi_{k}-\phi_{J}=+\epsilon_{1,2,3,4}$. For later use, we denote the position of the boxes $\phi_{k},\phi_{J}$ as $x'=(x'_{1},x'_{2},x'_{3}),\,\,x=(x_{1},x_{2},x_{3})$. The potential poles are then $x'=x+(1,0,0),x+(0,1,0),x+(0,0,1),x+(-1,-1,-1)$ where we used $\sum_{a\in\four}\epsilon_{a}=0$. Using the triality symmetry between $\epsilon_{1}\leftrightarrow \epsilon_{2}\leftrightarrow \epsilon_{3}$, it is enough to only consider the potential poles coming from $x'=x+(0,0,1),x+(-1,-1,-1)$. 
    
    We first consider the situation when the poles comes from $x'=x+(0,0,1)$.
    \begin{itemize}[topsep=1ex,itemsep=-0.5ex,partopsep=1ex,parsep=1ex]
        \item $\phi_{k}$ can not already be included in the plane partition. If this is so, then we have $\phi_{k}=\phi_{J}+\epsilon_{3}=\phi_{J'}$ for some $J'<k$, which gives
        \bea
        \frac{\sh(\phi_{k}-c_{\bar{4}}(\mathfrak{a}_{\ast},x+(0,0,1)))\sh(c_{\bar{4}}(\mathfrak{a}_{\ast},x+(0,0,1))-\phi_{k})}{\sh(\phi_{k}-c_{\bar{4}}(\mathfrak{a}_{\ast},x)-\epsilon_{3})}\rightarrow 0.
        \eea
        \item When $x\in\pi_{\ast}$ belongs to the boundary of the plane partition and $x_{1}=x_{2}=1,\,\,x_{3}>1$, no pole cancellation occurs and we have a single pole. This situation is illustrated as
        \bea
    \adjustbox{valign=c}{\includegraphics[width=4cm]{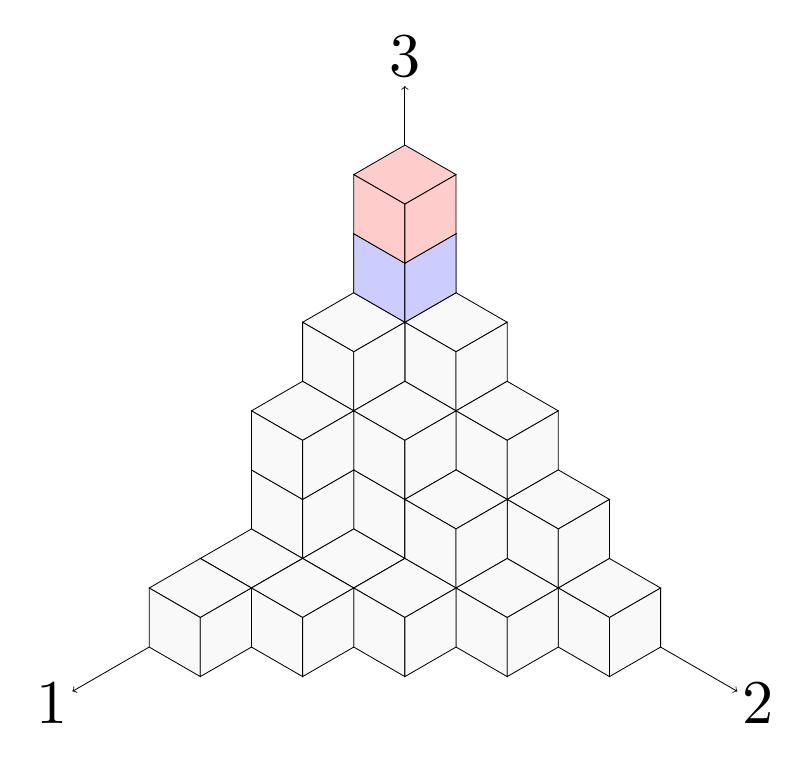}}
        \eea
        where the blue and red boxes correspond to the poles $\phi_{J},\phi_{k}$, respectively.
        \item When $x\in\pi_{\ast}$ belongs to the boundary of the plane partition and $x_{1}=1,\,\,x_{2},x_{3}>1$ (the other situation can be obtained by switching $x_{1}\leftrightarrow x_{2}$), we have $(1,x_{2}-1,x_{3})\in\pi_{\ast}$. The partition $\pi'_{\ast}$ is a plane partition if and only if $(1,x_{2}-1,x_{3}+1)\in\pi_{\ast}$:
        \bea
        \frac{\sh(c_{\bar{4}}(\mathfrak{a}_{\ast},x-(0,1,0))-\phi_{k}+\epsilon_{23})}{\sh(\phi_{k}-c_{\bar{4}}(\mathfrak{a}_{\ast},x)-\epsilon_{3})}\textcolor{blue}{\frac{1}{\sh(\phi_{k}-c_{\bar{4}}(\mathfrak{a}_{\ast},x+(0,-1,1))-\epsilon_{2})}}=-\frac{1}{\sh(\phi_{k}-c_{\bar{4}}(\mathfrak{a}_{\ast},x)-\epsilon_{3})}
        \eea
        where we used
        \bea
        \phi_{k}=c_{\bar{4}}(\mathfrak{a}_{\ast},x)+\epsilon_{3}=c_{\bar{4}}(\mathfrak{a}_{\ast},x-(0,1,0))+\epsilon_{23}=c_{\bar{4}}(\mathfrak{a}_{\ast},x+(0,-1,1))+\epsilon_{2}.
        \eea
        Since the blue term only appears when $\pi'_{\ast}$ is a plane partition, the JK residue is non-zero only when this condition is satisfied. This situation is illustrated as
        \bea
        \adjustbox{valign=c}{\includegraphics[width=4cm]{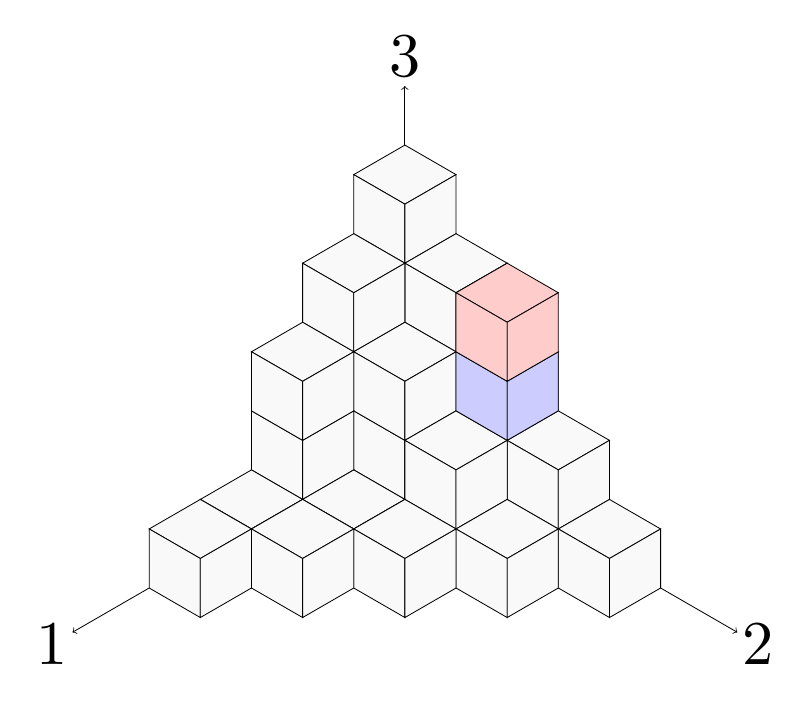}}
        \eea
        where the blue and red boxes correspond to the poles $\phi_{J},\phi_{k}$, respectively.

        \item When $x\in\pi_{\ast}$ belongs to the boundary of the plane partition and $x_{1,2,3}>1$, we have $(x_{1}-1,x_{2},x_{3}),(x_{1},x_{2}-1,x_{3})\in\pi_{\ast}$. The partition $\pi'_{\ast}$ is a plane partition if and only if
        \bea
        (x_{1}-1,x_{2},x_{3}+1),\,\,(x_{1},x_{2}-1,x_{3}+1)\in\pi_{\ast}.
        \eea
        Using 
        \bea
    \phi_{k}&=c_{\bar{4}}(\mathfrak{a}_{\ast},x)+\epsilon_{3}=c_{\bar{4}}(\mathfrak{a},x-(1,0,0))+\epsilon_{13}=c_{\bar{4}}(\mathfrak{a},x-(0,1,0))+\epsilon_{23}\\
    &=c_{\bar{4}}(\mathfrak{a}_{\ast},x+(-1,0,1))+\epsilon_{1}=c_{\bar{4}}(\mathfrak{a}_{\ast},x+(0,-1,1))+\epsilon_{2}
        \eea
        the related factor is
        \bea
        &\frac{\sh(c_{\bar{4}}(\mathfrak{a},x-(1,0,0))+\epsilon_{13}-\phi_{k})\sh(c_{\bar{4}}(\mathfrak{a},x-(0,1,0))+\epsilon_{23}-\phi_{k})}{\sh(\phi_{k}-c_{\bar{4}}(\mathfrak{a}_{\ast},x)-\epsilon_{3})}\\
        \times& \textcolor{blue}{\frac{1}{\sh(\phi_{k}-c_{\bar{4}}(\mathfrak{a}_{\ast},x+(-1,0,1))-\epsilon_{1})\sh(\phi_{k}-c_{\bar{4}}(\mathfrak{a}_{\ast},x+(0,-1,1))-\epsilon_{2})}}=\frac{1}{\sh(\phi_{k}-c_{\bar{4}}(\mathfrak{a}_{\ast},x)-\epsilon_{3})}.
        \eea
        The blue term only appears when $\pi_{\ast}'$ is a plane partition and gives a non-zero JK residue.
        \bea
        \adjustbox{valign=c}{\includegraphics[width=4cm]{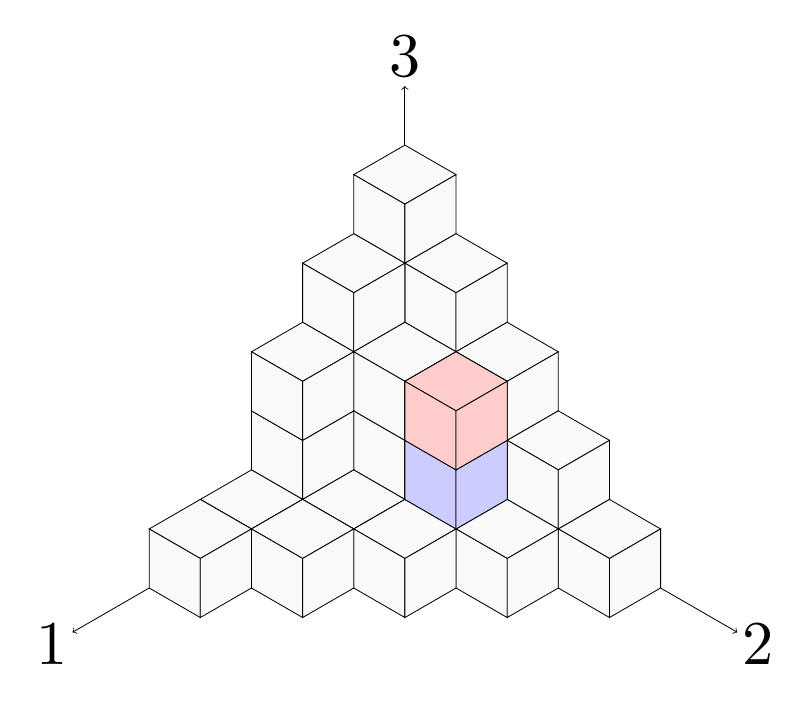}}
        \eea
        where the blue and red boxes correspond to the poles $\phi_{J},\phi_{k}$, respectively.
    \end{itemize}
    
Finally, let us then consider the case when the pole come from $\phi_{k}-\phi_{J}=\epsilon_{4}=-\epsilon_{123}$ $x'=x+(-1,-1,-1)$. The residue can be non-zero only if $(x_{1}-1,x_{2},x_{3})$ is not in $\pi_{\ast}$, since otherwise the numerator contain zeros from $\sh(\phi_{k}-c_{\bar{4}}(\mathfrak{a}_{\ast},x-(1,0,0))+\epsilon_{23})$. This also shows that $(x_{1},x_{2}-1,x_{3}),(x_{1},x_{2},x_{3}-1)$ are not in $\pi_{\ast}$. However, this contradicts with the assumption that $(x_{1},x_{2},x_{3})\in\pi_{\ast}$ and $\pi_{\ast}$ is a plane partition. Therefore, this contribution will always be canceled.

In summary, the poles giving non-zero JK residues are classified by a collection of plane partitions.
    
\end{proof}

Similar to the magnificent four case, the JK-residue can be written as an iterative residue. For simplicity, let us consider the case when $n_{\bar{4}}=1$ and $n_{\bar{a}}=0$ for $a\neq 4$. The discussion runs similarly and the conclusion is that we can choose a particular ordering for taking the residue and represent the JK-residue of a given plane partition by the iterative residue. Let $\{\phi_{1\ast},\cdots ,\phi_{k\ast}\}$ be the sequence of the coordinates of the $k$-boxes in the plane partition spanning the 123-directions. We can assume that the ordering is the one given in \eqref{eq:ordering-partition}. The iterative residue is then given by $\oint_{\phi_{k}=\phi_{k\ast}}\cdots \oint_{\phi_{1}=\phi_{1\ast}}$. The D6$_{123}$ $\U(1)$ partition function is then given by 
\bea
\mathcal{Z}^{\D6}_{\text{inst.}}=\sum_{\pi\in\mathcal{PP}_{4}}\mathfrak{q}^{|\pi|}\mathcal{Z}^{\D6}_{\bar{4}}[\pi],\quad \mathcal{Z}^{\D6}_{\bar{4}}[\pi]=\underset{\phi=\phi_{\pi}}{\Res}\mu_{k}^{\D6\tbar\D0}(\mathfrak{a},\phi_{I}),\quad k=|\pi|
\eea
where the iterated residue is given by 
\bea
\underset{\phi=\phi_{\pi}}{\Res}\mu_{k}^{\D6\tbar\D0}(\mathfrak{a},\phi_{I})=\underset{\phi_{k}=\phi_{k\ast}}{\Res}\cdots \underset{\phi_{1}=\phi_{1\ast}}{\Res}\mu_{k}^{\D6\tbar\D0}(\mathfrak{a},\phi_{I}).
\eea
The discussion for higher rank cases is straightforward so we omit it.

\subsection{Spiked instantons and Young diagrams}
Let us evaluate the contour integral formula given in Prop.~\ref{prop:spikedcontourJK}. If we specify to the case $n_{12}\neq 0,\,\,n_{A}=0\,(A\neq 12)$, then this partition function simply gives the instanton partition function of the 5d $\mathcal{N}=1$ $\U(n_{12})$ affine quiver gauge theory defined on $\mathbb{C}^{2}_{12}\times \mathbb{S}^{1}$ with adjoint mass $\epsilon_{3}$. Generalizing the discussion in section~\ref{sec:pureSYM_JK}, the poles of this contour integral formula actually are classified by Young diagrams. In the pure SYM case, the poles were labeled by Young diagrams spanning the $12$-direction whose coordinates are given by a linear combination of $\epsilon_{1,2}$, but in this case, we will have six types of Young diagrams spanning $A=ab\in\six$ whose coordinates come from $\epsilon_{a,b}$.
\begin{theorem}\label{thm:spikedJKpoles}
    The poles of \eqref{eq:spikedcontourJK} are classified by Young diagrams $\underline{\vec{\lambda}}=(\vec{\lambda}_{A})_{A\in\six}=(\lambda_{A}^{(\alpha)})^{\alpha=1,\ldots,n_{A}}_{A\in\six}$:
    \bea
    \{\phi_{I}\}_{I=1}^{k}\rightarrow \{c_{A}(\Bbox)\mid \Bbox=(x_{a},x_{b})\in \lambda_{A}^{(\alpha)},\quad A=ab\in\six\}
    \eea
    where $c_{ab}(\Bbox)=\mathfrak{a}_{A,\alpha}+(x_{a}-1)\epsilon_{a}+(x_{b}-1)\epsilon_{b}$. Namely, the $n_{A}$-tuples of Young diagrams extend in the two-dimensional space $(x_{a},x_{b})$ for $A=ab$. The total instanton number $k$ is identified with the number of boxes of the Young diagrams $k=\sum_{A\in\six}\sum_{\alpha=1}^{n_{A}}|\lambda_{A}^{(\alpha)}|$.
\end{theorem}
Since the proof is essentially the same with the pure SYM case, we only give a sketch of the induction process and point out different points compared with the pure SYM case. Lemma~\ref{lem:JKantifund} and \ref{lem:JKtree-structure} tell us that the poles picked up come from an oriented tree of charge vectors taking elements from $\{\bfe_{I}\}$ and $\{\bfe_{I}-\bfe_{J}\}$. The situation $k=1$ is trivial and we assume that for rank $k-1$, the Young diagram rule is true. We consider the situation when the root vertex corresponds to the pole at $\mathfrak{a}_{\ast}=\mathfrak{a}_{12,\alpha}$ for some $\alpha$. In this case, the Young diagram will expand in the $12$-direction as the pure SYM case and other situations are obtained by changing the equivariant parameters $\epsilon_{1,2}\rightarrow \epsilon_{a,b}$ for $A=ab\in\six$. The factor determining the pole at $\phi_{k}$ is
\bea
\frac{\sh(\phi_{k}-\mathfrak{a}_{\ast}-\epsilon_{4})\sh(\mathfrak{a}_{\ast}-\phi_{k}+\epsilon_{3})}{\sh(\phi_{k}-\mathfrak{a})\sh(\mathfrak{a}_{\ast}-\phi_{k}-\epsilon_{12})}\prod_{\Abox\in\lambda_{\ast}}\frac{\sh(\phi_{k}-c_{12}(\Bbox))\sh(\phi_{k}-c_{12}(\Bbox)-\epsilon_{14,24,34})}{\sh(\phi_{k}-c_{12}(\Bbox)-\epsilon_{1,2,3,4})}\\
    \times \prod_{\Abox\in\lambda_{\ast}}\frac{\sh(c_{12}(\Bbox)-\phi_{k})\sh(c_{12}(\Bbox)-\phi_{k}-\epsilon_{14,24,34})}{\sh(c_{12}(\Bbox)-\phi_{k}-\epsilon_{1,2,3,4})}.
\eea
One can show that the poles coming from $\phi_{k}=c_{12}(\Bbox)+\epsilon_{1,2}$ will always give non-zero JK residues as the pure SYM. The existence of the poles coming from $\phi_{k}=c_{12}(\Bbox)+\epsilon_{3,4}$ is the different part. Let us show that the poles coming from them will always give zero JK residues. Assume that $\phi_{J}=c_{12}(\mathfrak{a}_{\ast},(x,y))$ and $\phi_{k}=\phi_{J}+\epsilon_{3,4}$. We then have the following two cases.
\begin{itemize}[topsep=1ex,itemsep=-0.5ex,partopsep=1ex,parsep=1ex]
    \item When $\phi_{J}=\mathfrak{a}_{\ast}$, the poles will be canceled by the numerators:
    \bea
    \frac{\sh(\phi_{k}-\mathfrak{a}_{\ast}-\epsilon_{4})\sh(\mathfrak{a}_{\ast}-\phi_{k}+\epsilon_{3})}{\sh(\phi_{k}-\phi_{J}-\epsilon_{3})\sh(\phi_{k}-\phi_{J}-\epsilon_{4})}=1
    \eea
    \item When $\phi_{J}=c_{12}(\mathfrak{a}_{\ast},(x,y))$ for $x,y>1$ and $(x,y)\in\lambda_{\ast}$. This condition means that $(x-1,y),(x,y-1)\in\lambda_{\ast}$, which gives
    \bea
    \frac{\sh(\phi_{k}-c_{12}(\mathfrak{a}_{\ast},(x-1,y))-\epsilon_{14})\sh(c_{12}(\mathfrak{a}_{\ast},(x,y-1))-\phi_{k}-\epsilon_{14})}{\sh(\phi_{k}-\phi_{J}-\epsilon_{3})\sh(\phi_{k}-\phi_{J}-\epsilon_{3})}=-1
    \eea
    where we used
    \bea
    &c_{12}(\mathfrak{a}_{\ast},(x-1,y))+\epsilon_{14}=c_{12}(\mathfrak{a}_{\ast},(x,y))+\epsilon_{4},\\
    &c_{12}(\mathfrak{a}_{\ast},(x-1,y))-\epsilon_{14}=c_{12}(\mathfrak{a}_{\ast},(x,y))-\epsilon_{124}=c_{12}(\mathfrak{a}_{\ast},(x,y))+\epsilon_{3}.
    \eea
\end{itemize}
Therefore, the poles coming from $\phi_{k}=c_{12}(\Bbox)+\epsilon_{3,4}$ always disappear and the poles are classified by Young diagrams.

The iterative residue is given similarly as given in \eqref{eq:pureSYM-iterativeresidue} and in particlar for the $\U(1)$ D4$_{12}$ partition function, we have
\bea
\mathcal{Z}^{\D4}_{\text{inst.}}=\sum_{\lambda}\mathfrak{q}^{|\lambda|}\mathcal{Z}^{\D4}_{12}[\lambda],\quad \mathcal{Z}_{12}^{\D4}[\lambda]=\underset{\phi=\phi_{\lambda}}{\Res}\mu_{k}^{\D4\tbar\D0}(\mathfrak{a},\phi_{I}).
\eea
Explicitly we have
\bea
\mathcal{Z}_{12}^{\D4}[\lambda]&=\prod_{\Abox\in\lambda}\frac{\sh(\epsilon_{3}+(\ell_{\lambda}(\Bbox)+1)\epsilon_{1}-\epsilon_{2}a_{\lambda}(\Bbox))\sh(\epsilon_{3}-\epsilon_{1}\ell_{\lambda}(\Bbox)+\epsilon_{2}(a_{\lambda}(\Bbox)+1))}{\sh((\ell_{\lambda}(\Bbox)+1)\epsilon_{1}-\epsilon_{2}a_{\lambda}(\Bbox))\sh(-\epsilon_{1}\ell_{\lambda}(\Bbox)+\epsilon_{2}(a_{\lambda}(\Bbox)+1))}\\
&=\prod_{\Abox\in\lambda}\frac{\sh(\epsilon_{3}+(\ell_{\lambda}(\Bbox)+1)\epsilon_{1}-\epsilon_{2}a_{\lambda}(\Bbox))\sh(\epsilon_{4}+(\ell_{\lambda}(\Bbox)+1)\epsilon_{1}-\epsilon_{2}a_{\lambda}(\Bbox))}{\sh((\ell_{\lambda}(\Bbox)+1)\epsilon_{1}-\epsilon_{2}a_{\lambda}(\Bbox))\sh(\epsilon_{3}+\epsilon_{4}+(\ell_{\lambda}(\Bbox)+1)\epsilon_{1}-\epsilon_{2}a_{\lambda}(\Bbox))}.
\eea

\section{Description from equivariant index formalism}\label{sec:equiv-index}


The JK-residue for each multi-dimensional partition is explicitly derived in this section. We first introduce rational functions which we call structure functions and then rewrite the contour integral formulas using them. Following what we did in Chap.~\ref{chap:ADHM-localization}, we relate the JK-residues with characters and the index of them. Properties of these partition functions are also discussed.
\subsection{Structure functions}
Before moving on to the evaluation of the gauge origami partition function, let us introduce extra notations and definitions that frequently will be used from this chapter.

\paragraph{Data}
We denote the four complex coordinates of $\mathbb{C}^4$ as $z_{i}\,\,(i=1,2,3,4)$. The $\Omega$-deformation parameters are 
\bea
    q_{a}=e^{\epsilon_{a}},\quad \sum_{a\in\four}\epsilon_{a}=0,
\eea
where we simply omit the $\mathbb{S}^{1}$ radius. This is the Cartan torus of the $\SU(4)$ rotational symmetry, which acts on the four complex coordinates as
\bea
  (z_{1},z_{2},z_{3},z_{4})\mapsto (q_{1}z_{1},q_{2}z_{2},q_{3}z_{3},q_{4}z_{4})  
\eea
with the condition $q_{1}q_{2}q_{3}q_{4}=1$. We also introduce 
\bea
\bfQ_{a}=q_{a},\quad \bfP_{a}=1-q_{a},\quad \bfP_{a}^{\vee}=1-q_{a}^{-1},\quad a\in\four
\eea
and for any subset $S\subseteq\four$
\bea
\bfQ_{S}=\prod_{a\in S}\bfQ_{a},\quad \bfP_{S}=\prod_{a\in S}\bfP_{a}.
\eea
For example, we have $\bfP_{12}=(1-q_{1})(1-q_{2})$ and $\bfP_{123}=(1-q_{1})(1-q_{2})(1-q_{3})$. Some properties of the index are
\beq
q_{\four}=q_{\emptyset}=1,\quad \bfP_{\four}=\bfP_{1}\bfP_{2}\bfP_{3}\bfP_{4},\quad q_{\bar{S}}=q_{S}^{-1},\quad
\bfP^{\vee}_{S}=(-1)^{|S|}q_{S}^{-1}\bfP_{S},\label{eq:dualorigamiprop}
\eeq
for any subset $S\subseteq\four$. Note here, we denote $\bar{S}$ as the complement of the subset $S$. For later use, we also define $\bfP_{\bar{S}}=\prod_{a\in\bar{S}}\bfP_{a}$. This means that we have for example $\bfP_{\bar{4}}=\bfP_{123}$. For the subsets $S\subseteq\four$, we also use the following notation
\beq
S=\begin{dcases}
    a,\quad a=1,2,3,4,\\
    A,\quad A\in\six,\\
    \bar{a},\quad \bar{a}\in\{123,124,134,234\},\\
    \four=1234
\end{dcases}
\eeq
Some properties of the character are
\bea\label{eq:ch-squareroot}
\bfP_{\four}=\bfP_{\bar{a}}+\bfP_{\bar{a}}^{\vee},\quad \bfP_{\four}=\bfP_{\four}^{\vee}
\eea
for $a\in\four$.

\paragraph{Structure functions}
We introduce structure functions\footnote{\label{footnote:structure-function}The terminology \emph{structure functions} comes from an observation that if we take $\bar{a}=\bar{4}=123$, and take the limit $q_{4}\rightarrow 1$, the function $g_{\bar{a}}(x)$ will be \bea
    g_{123}(x)=\frac{(1-q_{1}x)(1-q_{2}x)(1-q_{3}x)(1-q_{4}^{-1}x)}{(1-x)(1-q_{4}^{-1}q_{1}^{-1}x)(q_{4}^{-1}q_{2}^{-1}x)(1-q_{4}^{-1}q_{3}^{-1}x)}\xrightarrow{q_{4}\rightarrow 1} \frac{(1-\sfq_{1}x)(1-\sfq_{2}x)(1-\sfq_{3}x)}{(1-\sfq_{1}^{-1}x)(1-\sfq_{2}^{-1}x)(1-\sfq_{3}^{-1}x)}
\eea where $q_{1,2,3}\rightarrow \sfq_{1,2,3}$ and $\sfq_{1}\sfq_{2}\sfq_{3}=1$. It resembles the structure function of the quantum toroidal $\mathfrak{gl}_{1}$ (see section~\ref{sec:QTgl1}).} as
\bea\label{eq:struct_funct}
\mathscr{V}_{a}(x)&=\mathbb{I}[-\bfP_{a}^{\vee}x^{\vee}]=\frac{1-q_{a}x}{1-x},\quad \mathscr{S}_{ab}(x)=\mathbb{I}[-\bfP_{ab}^{\vee}x^{\vee}]=\frac{(1-q_{a}x)(1-q_{b}x)}{(1-x)(1-q_{a}q_{b}x)},\\
g_{\bar{a}}(x)&=\mathbb{I}[-\bfP_{\bar{a}}^{\vee}x^{\vee}]=\frac{\prod\limits_{i\neq a}(1-q_{i}x)(1-q_{a}^{-1}x)}{(1-x)\prod\limits_{i\neq a}(1-q_{a}^{-1}q_{i}^{-1}x)},\quad \mathcal{A}_{\mathbb{C}^{4}}(x)=\mathbb{I}[-\bfP_{\four}^{\vee}x^{\vee}]=\frac{\prod\limits_{a\in\four}(1-q_{a}x)(1-q_{a}^{-1}x)}{(1-x)^{2}\prod\limits_{A\in\six}(1-q_{A}x)}
\eea
for $a\in\four,\,\,ab\in\six$. We also have the following properties:
\bea
&g_{abc}(x)=\frac{\mathscr{S}_{ab}(x)}{\mathscr{S}_{ab}(q_{c}x)},\quad \mathscr{S}_{ab}(x)=\frac{\mathscr{V}_{a}(x)}{\mathscr{V}_{a}(q_{b}x)},\quad \mathcal{A}_{\mathbb{C}^{4}}(x)=\frac{g_{\bar{a}}(x)}{g_{\bar{a}}(q_{a}x)}.
\eea

For later use, we also define
\bea
\mathscr{Y}^{A}_{\lambda,v}(x)&=\left(1-\frac{v}{x}\right)\prod_{\Abox\in\lambda}\mathscr{S}_{A}\left(\frac{\chi_{A,v}(\Bbox)}{x}\right)=\frac{\prod_{\Abox\in A(\lambda)}\left(1-\chi_{A,v}(\Bbox)/x\right)}{\prod_{\Abox\in R(\lambda)} (1-q_{A}\chi_{A,v}(\Bbox)/x)},\\
\mathscr{W}^{\bar{a}}_{\pi,v}(x)&=\left(1-\frac{v}{x}\right)\prod_{\scube\in \pi}g_{\bar{a}}\left(\frac{\chi_{\bar{a},v}(\cube)}{x}\right)\propto \prod_{\scube\in A(\pi)}\left(1-\frac{\chi_{\bar{a},v}(\cube)}{x}\right)\prod_{\scube\in R(\pi)}\left(1-q_{a}^{-1}\frac{\chi_{\bar{a},v}(\cube)}{x}\right).
\eea
and
\bea\mathscr{Y}^{A\,\vee}_{\lambda,v}(x)&=\left(1-\frac{x}{v}\right)\prod_{\Abox\in \lambda}\mathscr{S}_{A}\left(q_{A}^{-1}\frac{x}{\chi_{A,v}(\Bbox)}\right),\quad \mathscr{W}^{\bar{a}\,\vee}_{\pi,v}(x)&=\left(1-\frac{x}{v}\right)\prod_{\scube\in\pi}g_{\bar{a}}\left(q_{a}\frac{x}{\chi_{\bar{a},v}(\cube)}\right)^{-1}.
\eea

\paragraph{Reflection property}
Using the reflection property of the index in \eqref{eq:index_reflecprop}, we have
\bea
    \mathbb{I}[\bfP_{a}x]=q_{a}^{-1}\mathbb{I}[\bfP_{a}^{\vee}x^{\vee}],\quad \mathbb{I}[\bfP_{ab}x]=\mathbb{I}[\bfP_{ab}^{\vee}x^{\vee}],\quad \mathbb{I}[\bfP_{abc}x]=\mathbb{I}[\bfP_{abc}^{\vee}x^{\vee}],\quad \mathbb{I}[\bfP_{\four}x]=\mathbb{I}[\bfP_{\four}^{\vee}x^{\vee}] \label{eq:reflec_origamiindex}
\eea
for generic $x$. In terms of the structure functions, these are written as
\bea\label{eq:reflec_structfunc}
\mathscr{V}_{a}(x)&=q_{a}\mathscr{V}_{a}(q_{a}^{-1}x^{-1})^{-1},\quad 
\mathscr{S}_{ab}(x)=\mathscr{S}_{ab}(q_{a}^{-1}q_{b}^{-1}x^{-1}) \\
g_{\bar{a}}(x)&=g_{\bar{a}}(q_{a}x^{-1})^{-1},\quad 
\mathcal{A}_{\mathbb{C}^{4}}(x)=\mathcal{A}_{\mathbb{C}^{4}}(x^{-1}).
\eea

The properties \eqref{eq:reflec_origamiindex} can be generalized as follows. Let $\mathbf{X}=\sum_{i\in I}x_{i}$ where $I$ is a finite set ($|I|<\infty$). Then,
\bea\label{eq:reflectionprop2}
\mathbb{I}[\bfP_{a}\mathbf{X}]=q_{a}^{-|I|}\mathbb{I}[\bfP_{a}^{\vee}\mathbf{X}^{\vee}],&\quad \mathbb{I}[\bfP_{ab}\mathbf{X}]=\mathbb{I}[\bfP_{ab}^{\vee}\mathbf{X}^{\vee}],\\
\mathbb{I}[\bfP_{abc}\mathbf{X}]=\mathbb{I}[\bfP_{abc}^{\vee}\mathbf{X}^{\vee}],&\quad \mathbb{I}[\bfP_{\four}\mathbf{X}]=\mathbb{I}[\bfP_{\four}^{\vee}\mathbf{X}^{\vee}].
\eea

\paragraph{Sign rules}
The reflection properties mentioned above are true only when the roots $\{x_{i}\}$ are generic. When $\{x_{i}\}$ is not generic, the reflection property will give give extra sign factors. For example, let us consider the character $\bfP_{123}x$. When $x$ is generic, we simply have
\bea
\mathbb{I}[\bfP_{123}x]=\mathbb{I}[\bfP_{123}^{\vee}x^{-1}]
\eea
However, when $x\rightarrow 1$, the character will contain a $1$ term which give zeros after taking the index. Therefore, the reflection property should be modified as
\bea
\mathbb{I}[\bfP_{123}-1]=(-1)\mathbb{I}[\bfP_{123}^{\vee}-1].
\eea
This is nothing special and actually related with the residue property introduced in Prop.~\ref{prop:residuedual}.

\begin{proposition}\label{prop:reflection_sign}
    Let $\mathbf{X}=\sum_{i\in I}x_{i}$ be a character, where $I$ is a finite set and $\{x_{i}\}_{i\in I}$ are generic. Then, the reflection property is
    \bea
    \mathbb{I}\left[\bfP_{123}\mathbf{X}\right]=\mathbb{I}[\bfP_{123}^{\vee}\mathbf{X}^{\vee}].
    \eea
    Suppose that for example when $\{x_{i}\}$ are specialized and the character $\bfP_{123}\mathbf{X}$ is decomposed as
    \bea
    \bfP_{123}\mathbf{X}=\mathbf{Y}+\sum_{i\in I_{+}} 1-\sum_{i\in I_{-}}1
    \eea
    where $\mathbf{Y}$ does not contain any $\pm 1$ term and $I_{\pm}$ are finite sets giving $\pm1$ terms. Namely, 
    \bea
    \mathbf{Y}=\left[\bfP_{123}\bfX\right]^{(\neq 0)},\quad \sum_{i\in I_{+}} 1-\sum_{i\in I_{-}}1=\left[\bfP_{123}\bfX\right]^{(0)}.
    \eea
    The reflection property is then given as
    \bea
    \mathbb{I}[\bfP_{123}\mathbf{X}-|I_{+}|+|I_{-}|]=(-1)^{|I_{+}|-|I_{-}|}\mathbb{I}[\bfP_{123}^{\vee}\mathbf{X}^{\vee}-|I_{+}|+|I_{-}|].
    \eea
    Namely, the sign factor is determined by the unmovable terms (or the number of unmovable terms). Similar formulas can be obtained for others: $\bfP_{a}\bfX,\bfP_{A}\bfX,\bfP_{\four}\bfX$.
    
\end{proposition}

\begin{proposition}\label{prop:reflection-mod}
    Let $\mathbf{A},\mathbf{B}$ be a character where $\mathbf{A}+\mathbf{B}^{\vee}$ is movable. Then
    \bea
    \mathbb{I}\left[\mathbf{A}+\mathbf{B}^{\vee}\right]=(-1)^{\operatorname{rk}([\mathbf{B}]^{(\neq 0)})}\det [\mathbf{B}]^{(\neq 0)} \mathbb{I}\left[\mathbf{A}+\mathbf{B}\right]
    \eea
\end{proposition}
\begin{proof}
    Let $[\mathbf{B}]^{(0)}=m\in\mathbb{Z}$. Since $\mathbf{A}+\mathbf{B}^{\vee}$ is movable, using $\mathbf{A}=[\mathbf{A}]^{(0)}+[\mathbf{A}]^{(\neq0)}$, we have $[\mathbf{A}]^{(0)}=-m$. We then have
    \bea
    \mathbf{A}+\mathbf{B}^{\vee}=[\mathbf{A}]^{(\neq 0)}+[\mathbf{B}]^{(\neq 0)\vee}.
    \eea
    Since both terms are movable, we can safely use the reflection property and obtain the statement.

\end{proof}


\subsection{Magnificent four}\label{sec:magnificentindex}
Let us rewrite the contour integral formula given in Prop.~\ref{prop:M4contourJK} using the structure functions:
\bea
   \mathcal{Z}^{\D8}_{\text{inst.}}=\sum_{k=0}^{\infty}\mathfrak{q}^{k}\mathcal{Z}^{\D8}_{k},\quad  \mathcal{Z}^{\D8}_{k}=\frac{\mathcal{G}^{k}}{k!}\oint_{\text{JK}} 
\prod_{I=1}^{k}\frac{dx_{I}}{2\pi i x_{I}}\mathcal{Z}_{k}^{\D8\tbar\D0}(v_{\alpha},\bar{v}_{\alpha},x_{I})   \label{eq:D8integralmod}
\eea
where
\bea
&\mathcal{Z}_{k}^{\D8\tbar\D0}(v_{\alpha},\bar{v}_{\alpha},x_{I})=\prod_{I=1}^{k}\frac{\overline{P}(x_{I})}{P(x_{I})}\prod_{I\neq J}g_{\bar{4}}\left(\frac{x_{I}}{x_{J}}\right)^{-1}=\prod_{I=1}^{k}\frac{\overline{P}(x_{I})}{P(x_{I})}\prod_{I< J}\mathcal{A}_{\mathbb{C}^{4}}\left(\frac{x_{I}}{x_{J}}\right)^{-1},\\
&\mathcal{G}=\frac{(1-q_{12})(1-q_{13})(1-q_{23})}{(1-q_{1})(1-q_{2})(1-q_{3})(1-q_{123})},\quad P(x)=\prod_{\alpha=1}^{n}\left(1-\frac{v_{\alpha}}{x}\right),\quad \overline{P}(x)=\prod_{\alpha=1}^{n}\left(1-\frac{\bar{v}_{\alpha}}{x}\right)
\eea
and we redefined the topological terms. Note that the factor $\mathcal{G}$ is invariant under the quadrality symmetry.

Similar to the pure SYM case in section~\ref{sec:pureSYM_JK}, the rational function inside the contour integral can be rewritten using the index
\bea\label{eq:M4inst-character}
\mathcal{Z}_{k}^{\D8}=\frac{1}{k!}\oint_{\text{JK}}\prod_{I=1}^{k}\frac{dx_{I}}{2\pi ix_{I}}\mathbb{I}\left[\mathbf{v}^{\D8}_{\text{inst.}}-k\right],\quad  \mathbf{v}^{\D8}_{\text{inst.}}=-\bfN^{\vee}\bfK+\bfP_{123}^{\vee}\bfK^{\vee}\bfK
\eea
with
\bea
\label{eq:mag4ch}
    \mathbf{N}=\mathbf{n}-\bar{\mathbf{n}}=\sum_{\alpha=1}^{n}(e^{\mathfrak{a}_{\alpha}}-e^{\mathfrak{b}_{\alpha}})=\sum_{\alpha=1}^{n}(v_{\alpha}-\bar{v}_{\alpha}),\quad \mathbf{K}=\sum_{I=1}^{k}e^{\phi_{I}}=\sum_{I=1}^{k}x_{I}.
\eea
As shown in Thm.~\ref{thm:M4JKpoles}, the poles are classified with solid partitions:
\begin{align}
\begin{split}
&\vec{v}=(v_{\alpha})_{\alpha=1,\ldots,n},\quad \vec{\rho}=(\rho^{(\alpha)})_{\alpha=1,\ldots,n},\quad |\rho|=\sum_{\alpha=1}^{n}|\rho^{(\alpha)}|=k,\\
&\{x_{I}\}_{I=1,\ldots,k}\longrightarrow\{\chi_{\four,v_{\alpha}}(\hcube)=v_{\alpha}q_{1}^{i-1}q_{2}^{j-1}q_{3}^{k-1}q_{4}^{l-1}\}_{\alpha=1,\ldots,n,\,\shcube=(i,j,k,l)}
\end{split}
\end{align}
and at the fixed point $\vec{\rho}$, the character $\bfK$ will be
\bea\label{eq:D8Kch-fixedpoint}
\bfK|_{\vec{\rho}}=\sum_{\alpha=1}^{n}\sum_{\shcube\in\rho^{(\alpha)}}\chi_{\four,v_{\alpha}}(\hcube).
\eea
Since in our setup, a D8-brane will appear with a $\overline{\D8}$ as a pair, it is natural to use the following notation
\bea
    \bar{v}_{\alpha}=K_{\alpha}v_{\alpha},\quad \alpha=1,\ldots,n,\quad 
    \bfN=\sum_{\alpha=1}^{n}(1-K_{\alpha})v_{\alpha}
\label{eq:antiD8parameter}
\eea
where $K_{\alpha}$ physically represents the distance between the $\D8_{\alpha}$ and $\overline{\D8}_{\alpha}$ branes. Inserting this to \eqref{eq:mag4ch} and taking the index, we obtain
\bea
    \mathcal{Z}^{\D8}_{\four;4}[\vec{v},\vec{\rho}\,;\vv{K}]&=\prod_{\alpha=1}^{n}\mathcal{Z}_{\four;4}^{\D8}[\rho^{(\alpha)};K_{\alpha}]\prod_{\beta>\alpha}\mathcal{Z}^{\D8\tbar\D8}_{K_{\alpha}|K_{\beta}}(v_{\alpha},\rho^{(\alpha)}\,|\,v_{\beta},\rho^{(\beta)}),\\
    \mathcal{Z}^{\D8}_{\four;4}[\rho\,;K]&=\mathbb{I}\left[-(1-K^{-1})x^{-1}\bfK_{\rho}+\bfP_{123}^{\vee}\bfK_{\rho}^{\vee}\bfK_{\rho}\right],\\
    \mathcal{Z}^{\D8\tbar\D8}_{K_{1}|K_{2}}(x_{1},\rho^{(1)}\,|\,x_{2},\rho^{(2)})&=\mathbb{I}\left[-(1-K_{1}^{-1})x_{1}^{-1}\bfK_{\rho^{(2)}}-(1-K_{2}^{-1})x_{2}^{-1}\bfK_{\rho^{(1)}}+\bfP_{\four}\bfK^{\vee}_{\rho^{(1)}}\bfK_{\rho^{(2)}}\right],
\label{eq:mag4Nekrasovfact}
\eea
where we shortly wrote $\bfK_{\rho}=\sum_{\shcube\in\rho}\chi_{\four,x}(\hcube).$ 

Naively, one will expect that the partition function takes the form as
\bea
\mathcal{Z}_{\text{inst.}}^{\D8}=\sum_{\vec{\rho}}\mathfrak{q}^{|\vec{\rho}|}  \mathcal{Z}^{\D8}_{\four;4}[\vec{v},\vec{\rho}\,;\vv{K}].
\eea
However, actually this is \textit{not} correct. One reason for this is that obviously the character \eqref{eq:M4inst-character} is not quadrality invariant. Namely, we have other definitions such as
\bea
\mathbf{v}^{\D8}_{\text{inst.}}=-\bfN^{\vee}\bfK+\bfP_{\bar{a}}^{\vee}\bfK^{\vee}\bfK
\eea
and
\bea
\mathcal{Z}^{\D8}_{\four;a}[\rho\,;K]&=\mathbb{I}\left[-(1-K^{-1})x^{-1}\bfK_{\rho}+\bfP_{\bar{a}}^{\vee}\bfK_{\rho}^{\vee}\bfK_{\rho}\right],\label{eq:mag4Nekrasovfact2}
\eea
for $a\in\four$. For general $\rho$, one can explicitly check
\bea
\mathcal{Z}^{\D8}_{\four;a}[\rho\,;K]\neq \mathcal{Z}^{\D8}_{\four;b}[\rho\,;K]
\eea
for $a\neq b$ and actually both of them differ by extra \textit{signs} $\pm1$. For example, for lower levels, we have
\begin{subequations}
\begin{align}
    \rho=\{\{\{1\}\}\},&\quad \frac{\mathcal{Z}_{\four;1}^{\D8}[\rho,K]}{\mathcal{Z}_{\four;2}^{\D8}[\rho,K]}=\frac{\mathcal{Z}_{\four;2}^{\D8}[\rho,K]}{\mathcal{Z}_{\four;3}^{\D8}[\rho,K]}=\frac{\mathcal{Z}_{\four;3}^{\D8}[\rho,K]}{\mathcal{Z}_{\four;4}^{\D8}[\rho,K]}=\frac{\mathcal{Z}_{\four;4}^{\D8}[\rho,K]}{\mathcal{Z}_{\four;1}^{\D8}[\rho,K]}=1,\\
    \rho=\{\{\{2\}\}\},&\quad \frac{\mathcal{Z}_{\four;1}^{\D8}[\rho,K]}{\mathcal{Z}_{\four;2}^{\D8}[\rho,K]}=\frac{\mathcal{Z}_{\four;2}^{\D8}[\rho,K]}{\mathcal{Z}_{\four;3}^{\D8}[\rho,K]}=-\frac{\mathcal{Z}_{\four;3}^{\D8}[\rho,K]}{\mathcal{Z}_{\four;4}^{\D8}[\rho,K]}=-\frac{\mathcal{Z}_{\four;4}^{\D8}[\rho,K]}{\mathcal{Z}_{\four;1}^{\D8}[\rho,K]}=1,\\
    \rho=\{\{\{1,1\}\}\},&\quad \frac{\mathcal{Z}_{\four;1}^{\D8}[\rho,K]}{\mathcal{Z}_{\four;2}^{\D8}[\rho,K]}=-\frac{\mathcal{Z}_{\four;2}^{\D8}[\rho,K]}{\mathcal{Z}_{\four;3}^{\D8}[\rho,K]}=-\frac{\mathcal{Z}_{\four;3}^{\D8}[\rho,K]}{\mathcal{Z}_{\four;4}^{\D8}[\rho,K]}=\frac{\mathcal{Z}_{\four;4}^{\D8}[\rho,K]}{\mathcal{Z}_{\four;1}^{\D8}[\rho,K]}=1,\\
    \rho=\{\{\{1\},\{1\}\}\},&\quad -\frac{\mathcal{Z}_{\four;1}^{\D8}[\rho,K]}{\mathcal{Z}_{\four;2}^{\D8}[\rho,K]}=-\frac{\mathcal{Z}_{\four;2}^{\D8}[\rho,K]}{\mathcal{Z}_{\four;3}^{\D8}[\rho,K]}=\frac{\mathcal{Z}_{\four;3}^{\D8}[\rho,K]}{\mathcal{Z}_{\four;4}^{\D8}[\rho,K]}=\frac{\mathcal{Z}_{\four;4}^{\D8}[\rho,K]}{\mathcal{Z}_{\four;1}^{\D8}[\rho,K]}=1,
\end{align}
\end{subequations}
where we described the solid partition as 
\bea
\rho=\{\{\rho_{11},\rho_{12},\ldots\},\ldots,\{\rho_{i1},\rho_{i2}\ldots,\rho_{im_{i}}\},\ldots,\{\rho_{l1},\ldots,\rho_{lm_{l}}\}\},\quad \rho_{ij}=(\rho_{ij1},\ldots,\rho_{ijn_{ij}}).
\eea
These extra sign factors need to be incorporated to get the quadrality invariant partition function. 
\begin{theorem}[\cite{Nekrasov:2017cih,Nekrasov:2018xsb}]
The total instanton partition function of the $\U(n|n)$ magnificent four system is given by summation over $n$-tuple solid partitions, 
\bea
\mathcal{Z}^{\D8}_{\text{inst.}}=\sum_{\vec{\rho}\in\mathcal{SP}}\mathfrak{q}^{|\vec{\rho}|}(-1)^{\sigma_{a}(\vec{\rho})}\mathcal{Z}_{\four;a}^{\D8}[\vec{v},\vec{\rho}\,;\vv{K}],\quad a\in\four\label{eq:mag4inst}
\eea
where $\sigma_{a}(\vec{\rho})$ is a sign factor defined as
\bea\label{eq:signfactor-def}
\sigma_{a}(\vec{\rho})&=\sum_{\alpha=1}^{n}\sigma_{a}(\rho^{(\alpha)}),\quad \sigma_{a}(\rho)=\#\{(x_{1},x_{2},x_{3},x_{4})\in\rho\mid x_{b}=x_{c}=x_{d}<x_{a}\}.
\eea
\end{theorem}

\begin{remark}
Including these sign factors actually give 
\bea
(-1)^{\sigma_{a}(\vec{\rho})}\mathcal{Z}_{\four;a}^{\D8}[\vec{v},\vec{\rho}\,;\vv{K}]=(-1)^{\sigma_{b}(\vec{\rho})}\mathcal{Z}_{\four;b}^{\D8}[\vec{v},\vec{\rho}\,;\vv{K}],\quad a,b\in\four.
\eea
See Thm.~\ref{thm:D8signindep} and Thm.~\ref{thm:D8signruleJK} for discussion on the $n=1$ case.
\end{remark}

\begin{remark}
    The sum of the character \eqref{eq:M4inst-character} and its dual is
    \bea
    \mathbf{V}_{\text{inst.}}^{\D8}=\mathbf{v}^{\D8}_{\text{inst.}}+
\mathbf{v}^{\D8\vee}_{\text{inst.}}=-\bfN^{\vee}\bfK-\bfN\bfK^{\vee}+\bfP_{\four}\bfK^{\vee}\bfK
    \eea
    where we used \eqref{eq:ch-squareroot}. In this sense, we can write $\mathbf{v}_{\text{inst.}}^{\D8}=\sqrt{\mathbf{V}^{\D8}_{\text{inst.}}}$ and $\sqrt{\bfP_{\four}\bfK^{\vee}\bfK}=\bfP_{\bar{4}}^{\vee}\bfK^{\vee}\bfK$. At the level of contour integral, the square root can be also observed from 
    \bea
    \prod_{I<J}\mathcal{A}_{\mathbb{C}^{4}}\left(\frac{x_{I}}{x_{J}}\right)^{-1}.
    \eea

    After evaluating the poles and inserting \eqref{eq:D8Kch-fixedpoint} to the character, the square root part needs to be \textit{movable} (see Def.~\ref{app-def:movable}) so that we can safely take the index. Moreover, for the magnificent four case, the choice of this \textit{square root} part is crucial and the sign rules depend on this choice. For example, we can also use
    \bea
    \sqrt{\bfP_{\four}\bfK^{\vee}\bfK}=\bfP_{123}\bfK^{\vee}\bfK
    \eea
    for the square root part and the sign rule is modified to
    \bea
    \sigma_{4}(\rho)=\#\{(x_{1},x_{2},x_{3},x_{4})\in\rho\mid x_{1}=x_{2}=x_{3}\leq x_{4} \}.
    \eea
    Actually, this is the original choice of the square root part in \cite{Nekrasov:2017cih,Nekrasov:2018xsb}.

\end{remark}

\paragraph{One-loop perturbative factors}For the one-loop perturbative part, we define
\bea\label{eq:D8oneloop}
    \mathcal{Z}_{\text{1-loop}}^{\D8\tbar\D8}(v_{1},K_{1}\,|\,v_{2},K_{2})&=\exp\left(-\sum_{n>0}\frac{1}{n}\frac{(1-K_{2}^{-n})(1-K_{1}^{n})}{\bfP_{\four}^{[n]}}\left(\frac{v_{1}}{v_{2}}\right)^{n}\right).
\eea
In terms of character, it can be written as
\bea
\mathbb{I}[\mathring{\mathbf{v}}]&=\prod_{\alpha<\beta}\mathcal{Z}_{\text{1-loop}}^{\D8\tbar\D8}(v_{\alpha},K_{\alpha}\,|\,v_{\beta},K_{\beta})\eqqcolon\mathcal{Z}_{\text{1-loop}}^{\D8},
\eea
where
\bea
   &\mathring{\mathbf{v}}=\frac{1}{\mathbf{P}_{\four}}\sum_{v_{\beta}>v_{\alpha}}(1-K_{\alpha}^{\vee})(1-K_{\beta})v_{\beta}/v_{\alpha}
\eea
where we specified an order in the pair of indices $(\bar{a},\alpha)_{a\in\four}^{\alpha=1,\ldots,n_{\bar{a}}}$.

\paragraph{Movability}
Let us study the movability of the character $\mathbf{v}_{\text{inst.}}^{\D8}$ explicitly for the $\U(1|1)$ partition function.
\begin{theorem}\label{thm:D8movable}
    The character $\mathbf{v}_{\text{inst.}}^{\D8}$ is movable.
\end{theorem}
\begin{proof}
    It is enough to check the movability of $\mathbf{v}^{\D8}_{\text{inst.}}$ with $K\rightarrow \infty$:
    \bea
    \mathbf{v}_{\rho}=-\bfK_{\rho}+\bfP_{123}^{\vee}\bfK_{\rho}^{\vee}\bfK_{\rho},\quad \bfK_{\rho}=\sum_{\shcube\in\rho}\chi_{\four,1}(\hcube).
    \eea
    We prove it by induction of the size of the solid partition $|\rho|=k$. For $k=1$,
    \bea
    \mathbf{v}_{\rho}=-1+\bfP_{123}^{\vee}=-q_{1}^{-1}-q_{2}^{-1}-q_{3}^{-1}+q_{12}^{-1}+q_{13}^{-1}+q_{23}^{-1}-q_{123}^{-1}.
    \eea
    Assume that for $|\rho|\leq k$, $\mathbf{v}_{\rho}$ is movable and $\tilde{\rho}=\rho+\chi^{\mu}$ is a solid partition with an additional box $\chi^{\mu}=q_{1}^{\mu_{1}-1}q_{2}^{\mu_{2}-1}q_{3}^{\mu_{3}-1}q_{4}^{\mu_{4}-1}$. We then have
    \bea
    \mathbf{v}_{\tilde{\rho}}&=\mathbf{v}_{\rho}-\chi^{\mu}+\bfP_{123}^{\vee}\bfK_{\rho}^{\vee}\chi^{\mu}+\bfP_{123}^{\vee}\chi^{\mu\vee}\bfK_{\rho}+\bfP_{123}^{\vee}\\
    &=\mathbf{v}_{\rho}-\chi^{\mu}+\bfP_{123}^{\vee}\bfK_{\tilde{\rho}}^{\vee}\chi^{\mu}+\bfP_{123}^{\vee}\chi^{\mu\vee}\bfK_{\tilde{\rho}}-\bfP_{123}^{\vee}.
    \eea
    Decompose $\bfK_{\tilde{\rho}}$ into parts $\bfK_{\tilde{\rho}}=\bfK_{\rho'}+\bfK_{\rho''}$ and $\tilde{\rho}=\rho'\sqcup\rho''$ where
    \bea
\bfK_{\rho'}=\sum_{i=1}^{\mu_{1}}\sum_{j=1}^{\mu_{2}}\sum_{k=1}^{\mu_{3}}\sum_{l=1}^{\mu_{4}}q_{1}^{i-1}q_{2}^{j-1}q_{3}^{k-1}q_{4}^{l-1},\quad \bfK_{\rho''}=\sum_{\shcube\in\rho''}\chi_{\four,1}(\hcube).
    \eea
Note that $\rho''$ here contains boxes where the coordinate obeys $(i> \mu_{1})\vee (j> \mu_{2})\vee 
    (k> \mu_{3})\vee (l> \mu_{4})$ and $\rho''\subset \rho$. We first have
    \bea
    \left[\bfP_{123}^{\vee}\bfK_{\rho''}^{\vee}\chi^{\mu}\right]^{(0)}=0,\quad \left[\bfP_{123}^{\vee}\chi^{\mu\vee}\bfK_{\rho''}\right]^{(0)}=0.
    \eea
   The remaining terms are computed as
   \bea
    \left[\bfP_{123}^{\vee}\bfK_{\rho'}^{\vee}\chi^{\mu}\right]^{(0)}&=\left[(1-q_{1}^{-\mu_{1}})(1-q_{2}^{-\mu_{2}})(1-q_{3}^{-\mu_{3}})\sum_{l=1}^{\mu_{4}}q_{4}^{-l+1}\chi^{\mu}\right]^{(0)}\\
    &=\left[-(1-q_{1}^{\mu_{1}})(1-q_{2}^{\mu_{2}})(1-q_{3}^{\mu_{3}})\sum_{l=1}^{\mu_{4}}q_{4}^{\mu_{4}-l+1}\right]^{(0)}=\sum_{l=1}^{\mu_{4}}\delta_{\mu_{1},\mu_{2},\mu_{3},l},\\
    \left[\bfP_{123}^{\vee}\chi^{\mu\vee}\bfK_{\rho'}\right]^{(0)}&=\left[-q_{123}^{-1}(1-q_{1}^{\mu_{1}})(1-q_{2}^{\mu_{2}})(1-q_{3}^{\mu_{3}})\sum_{l=1}^{\mu_{4}}q_{4}^{l-1}\chi^{\mu\vee}\right]^{(0)}\\
    &=\left[(1-q_{1}^{-\mu_{1}})(1-q_{2}^{-\mu_{2}})(1-q_{3}^{-\mu_{3}})\sum_{l=1}^{\mu_{4}}q_{4}^{-\mu_{4}+l}\right]^{(0)}=1-\sum_{l=1}^{\mu_{4}-1}\delta_{\mu_{1},\mu_{2},\mu_{3},l},\\
    \left[-\chi^{\mu}-\bfP_{123}^{\vee}\right]^{(0)}&=-1-\delta_{\mu_{1},\mu_{2},\mu_{3},\mu_{4}}.
   \eea
Combining with $[\mathbf{v}_{\rho}]^{(0)}=0$, we can see that $\mathbf{v}_{\tilde{\rho}}$ is movable, which concludes the inductive step.
\end{proof}

\paragraph{JK residue and sign rules}As mentioned before, with proper sign rules, the partition function is unique.
\begin{theorem}[\cite{Nekrasov:2017cih, Nekrasov:2018xsb,Monavari:2022rtf}]\label{thm:D8signindep}
    The magnificent four partition function does not depend on the choice of the square root part:
    \bea
    (-1)^{\sigma_{a}(\rho)}\mathcal{Z}^{\D8}_{\four;a}[\rho;K]=(-1)^{\sigma_{b}(\rho)}\mathcal{Z}^{\D8}_{\four;b}[\rho;K].
    \eea
\end{theorem}
\begin{proof}
    We follow the proof of \cite{Nekrasov:2018xsb, Monavari:2022rtf}. It is enough to focus on $a=4,b=3$ without loss of generality. Define 
    \bea
    \mathbf{v}_{\rho}^{a}=-(1-K^{-1})\bfK_{\rho}+\bfP_{\bar{a}}^{\vee}\bfK_{\rho}^{\vee}\bfK_{\rho},\quad \bfK_{\rho}=\sum_{(i,j,k,l)\in\rho}q_{1}^{i-1}q_{2}^{j-1}q_{3}^{k-1}q_{4}^{l-1}
    \eea
    where we set the Coulomb moduli to be trivial because the partition function does not depend on it. We have
    \bea
    \mathbf{v}_{\rho}^{4}-\mathbf{v}_{\rho}^{3}&=(\bfP_{123}^{\vee}-\bfP_{124}^{\vee})\bfK_{\rho}^{\vee}\bfK_{\rho}=\bfP_{12}^{\vee}(q_{4}^{-1}-q_{3}^{-1})\bfK_{\rho}^{\vee}\bfK_{\rho}\\
    &=-\bfP_{12}^{\vee}q_{3}^{-1}\bfK_{\rho}^{\vee}\bfK_{\rho}+\bfP_{12}q_{3}\bfK_{\rho}^{\vee}\bfK_{\rho}.
    \eea
    Thus, we need to proof
    \bea
    \mathbb{I}[\mathbf{v}^{4}_{\rho}-\mathbf{v}^{3}_{\rho}]&=(-1)^{\sigma_{3}(\rho)-\sigma_{4}(\rho)}
    \eea
    which is equivalent to 
    \bea
    \mathbb{I}[\mathbf{v}^{4}_{\rho}-\mathbf{v}^{3}_{\rho}]&=
    -\sigma_{4}(\rho)+\sigma_{3}(\rho)&=[\bfP_{12}q_{3}\bfK_{\rho}^{\vee}\bfK_{\rho}]^{(0)}=[\bfP_{12}^{\vee}q_{3}^{-1}\bfK_{\rho}^{\vee}\bfK_{\rho}]^{(0)},\quad \mod2
    \eea
    after using Prop.~\ref{prop:reflection-mod}.
    Let us prove this by induction. For $|\rho|=1$, we have
    \bea
   \, [\bfP_{12}q_{3}]^{(0)}=0,\quad \sigma_{3}(\rho)=\sigma_{4}(\rho)=0.
    \eea
    Assume that for $|\rho|\leq k$ the claim is true and consider a solid partition $\tilde{\rho}=\rho+\chi^{\mu}$ where $\rho$ is a solid partition with size $k$ and $\chi^{\mu}=q_{1}^{\mu_{1}-1}q_{2}^{\mu_{2}-1}q_{3}^{\mu_{3}-1}q_{4}^{\mu_{4}-1},\,\,(\mu_{1,2,3,4}\in\mathbb{Z}_{\geq 1})$ is an additional box added to the solid partition. We then have
    \bea
    \mathbf{v}_{\tilde{\rho}}^{4}-\mathbf{v}^{3}_{\tilde{\rho}}&=\mathbf{v}_{\rho}^{4}-\mathbf{v}_{\rho}^{3}+\bfP_{12}^{\vee}(q_{4}^{-1}-q_{3}^{-1})\left(\chi^{\mu\vee}\bfK_{\tilde{\rho}}+\bfK_{\tilde{\rho}}^{\vee}\chi^{\mu}\right)-\bfP_{12}^{\vee}(q_{4}^{-1}-q_{3}^{-1}).
    \eea
    By the induction step, we have
    \bea
    \mathbb{I}[\mathbf{v}_{\rho}^{4}-\mathbf{v}^{3}_{\rho}]=(-1)^{\sigma_{3}(\rho)-\sigma_{4}(\rho)},\quad \mathbb{I}[-\bfP_{12}^{\vee}(q_{4}^{-1}-q_{3}^{-1})]=+1.
    \eea
    Focusing on the remaining terms, we need to study the following unmovable terms
    \bea
\,[\bfP_{12}^{\vee}q_{3}^{-1}\chi^{\mu\vee}\bfK_{\tilde{\rho}}+\bfP_{12}^{\vee}q_{4}^{-1}\bfK_{\tilde{\rho}}\chi^{\mu\vee}]^{(0)}.
    \eea
   Similar to the proof in Thm.~\ref{thm:D8movable}, decompose $\bfK_{\tilde{\rho}}$ into parts $\bfK_{\tilde{\rho}}=\bfK_{\rho'}+\bfK_{\rho''}$ and $\tilde{\rho}=\rho'\sqcup\rho''$. 
    Similar to the proof in Thm.~\ref{thm:D8movable}, we have
    \bea
\,    [\bfP_{12}^{\vee}q_{3}^{-1}\chi^{\mu\vee}\bfK_{\tilde{\rho}}]^{(0)}&=
    [\bfP_{12}^{\vee}q_{3}^{-1}\chi^{\mu\vee}\bfK_{\rho'}]^{(0)}=
    \left[\bfP_{12}^{\vee}q_{3}^{-1}\chi^{\mu\vee}\frac{(1-q_{1}^{\mu_{1}})(1-q_{2}^{\mu_{2}})}{\bfP_{12}}\sum_{k,l}q_{3}^{k-1}q_{4}^{l-1}\right]^{(0)}\\
    &=\left[(1-q_{1}^{-\mu_{1}})(1-q_{2}^{-\mu_{2}})\sum_{k=1}^{\mu_{3}}q_{3}^{-k}\sum_{l=0}^{\mu_{4}-1}q_{4}^{-l}\right]^{(0)}=\sum_{k=1}^{\mu_{3}}\sum_{l=0}^{\mu_{4}-1}\delta_{\mu_{1},\mu_{2},k,l}.
    \eea
    Similarly, we have
    \bea
\,    [\bfP_{12}^{\vee}q_{4}^{-1}\chi^{\mu\vee}\bfK_{\tilde{\rho}}]^{(0)}=\sum_{k=0}^{\mu_{3}-1}\sum_{l=1}^{\mu_{4}}\delta_{\mu_{1},\mu_{2},k,l}
    \eea
    and thus
    \bea
\relax[\bfP_{12}^{\vee}q_{3}^{-1}\chi^{\mu\vee}\bfK_{\tilde{\rho}}+\bfP_{12}^{\vee}q_{4}^{-1}\bfK_{\tilde{\rho}}\chi^{\mu\vee}]^{(0)}=\sum_{k=1}^{\mu_{3}-1}\delta_{\mu_{1},\mu_{2},k,\mu_{4}}+\sum_{l=1}^{\mu_{4}-1}\delta_{\mu_{1},\mu_{2},\mu_{3},l} \quad \mod2.
    \eea
    Therefore, we have $\sigma_{4}(\tilde{\rho})=\sigma_{4}(\rho)+\sum_{l=1}^{\mu_{4}-1}\delta_{\mu_{1},\mu_{2},\mu_{3},l}$ and $\sigma_{3}(\tilde{\rho})=\sigma_{3}(\rho)+\sum_{k=1}^{\mu_{3}-1}\delta_{\mu_{1},\mu_{2},k,\mu_{4}}$ which gives the claim.
\end{proof}

Since the partition function does not depend on the choice $a\in\four$, from now on, we will choose $a=4$ and frequently use the sign factor $\sigma_{4}(\rho)$ when necessary. The sign rule can be derived by comparing the result from the JK-residue and the index computation.
\begin{theorem}\label{thm:D8signruleJK}
The JK-residue of the integrand obeys the identity:
    \bea
\mathcal{G}^{k}\times \underset{x=x_{\rho}}{\Res}x^{-1}\mathcal{Z}_{k}^{\D8\tbar\D0}(v,Kv,x_{I})=(-1)^{\sigma_{a}(\rho)}\mathcal{Z}^{\D8}_{\four;a}[\rho,K].
    \eea
\end{theorem}
The residue of the left-hand side is understood as \eqref{eq:D8-iteratedresidue}:
\bea
\underset{x=x_{\rho}}{\Res}x^{-1}\mathcal{Z}_{k}^{\D8\tbar\D0}(v,Kv,x_{I})=\underset{x_{k}}{\Res}x_{k}^{-1}\underset{x_{k-1}}{\Res}x_{k-1}^{-1}\cdots \underset{x_{1}=v}{\Res}x_{1}^{-1}\mathcal{Z}_{k}^{\D8\tbar\D0}(v,Kv,x_{I})
\eea
and it is taken around poles in an ordering defined in \eqref{eq:ordering-partition}. The extra sign factors of the right hand side comes from the property discussed in Prop.~\ref{prop:residuegeneral}. Let us check low levels and explicitly see where the extra signs appear after evaluating the residue. 

For $|\rho|=1$, we only have one configuration $\rho=\{v\}$ where we identified the $q$-coordinates and the box position. The residue is given as
\bea\label{eq:D8one-inst}
\mathcal{G}\times \underset{x=v}{\Res}x^{-1}\mathcal{Z}_{1}^{\D8\tbar\D0}(v,Kv,x)&=\mathcal{G}\times \underset{x=v}{\Res}x^{-1}\frac{(1-Kv/x)}{(1-v/x)}=\mathcal{G}(1-K)\\
\mathbb{I}[-(1-K^{-1})+\bfP_{123}^{\vee}]&=\mathcal{G}(1-K),
\eea
and this is compatible with the fact that $\sigma_{4}(\rho)=0$.

A nontrivial example is when $|\rho|=2$ and $\rho=\{v,q_{4}v\}$. The residue is computed as
\bea
&\underset{x_{2}=q_{4}v}{\Res}x_{2}^{-1}\underset{x_{1}=v}{\Res}x_{1}^{-1}\mathcal{Z}_{1}^{\D8\tbar\D0}(v,Kv,x_{1,2})\\
=&\underset{x_{2}=q_{4}v}{\Res}x_{2}^{-1}\frac{(1-Kv/x_{2})}{(1-v/x_{2})}\underset{x_{1}=v}{\Res}x_{1}^{-1}\frac{(1-Kv/x_{1})}{(1-v/x_{1})}\mathcal{A}_{\mathbb{C}^{4}}\left(\frac{x_{1}}{x_{2}}\right)^{-1}\\
=&-\frac{(K-1)(K-q_{4})  (q_{1}-1)  (q_{2}-1)  (q_{3}-1) 
   q_{4}^2(q_{1} q_{2}-q_{4}) (q_{1}
   q_{3}-q_{4}) (q_{2} q_{3}-q_{4})}{(q_{4}+1)
   (q_{1}-q_{4}) (q_{1} q_{4}-1) (q_{4}-q_{2}) (q_{2}
   q_{4}-1) (q_{4}-q_{3}) (q_{3} q_{4}-1)}.
\eea
After computing the character and taking the index, one will obtain
\bea
\mathcal{G}^{2}\times\underset{x_{2}=q_{4}v}{\Res}x_{2}^{-1}\underset{x_{1}=v}{\Res}x_{1}^{-1}\mathcal{Z}_{1}^{\D8\tbar\D0}(v,Kv,x_{1,2})=\textcolor{red}{-}\mathcal{Z}_{\four;4}^{\D8}[\{v,q_{4}v\};K]
\eea
and an extra sign factor appears. This is compatible with 
\bea
\sigma_{4}(\{v,q_{4}v\})=\#\{(i,j,k,l)\in\rho\mid i=j=k<l\}=\#\{(1,1,1,2)\}=1.
\eea

Let us take a look more careful at how the residue is taken to see the origin of the sign factor. The first residue comes from the pole at the $(1-v/x_{1})$ in the denominator and thus $\underset{x_{1}=v}{\Res}x_{1}^{-1}(1-v/x_{1})^{-1}=+1$. After inserting $x_{1}=v$, the integrand is a function of $x_{2}$ which takes the form as
\bea
\underset{x_{1}=v}{\Res}x_{1}^{-1}\frac{(1-Kv/x_{1})}{(1-v/x_{1})}g_{\bar{4}}\left(\frac{x_{1}}{x_{2}}\right)^{-1}g_{\bar{4}}\left(\frac{x_{2}}{x_{1}}\right)^{-1}=\cdots \times \frac{1}{(1-q_{4}^{-1}v/x_{2})(1-q_{4}^{-1}x_{2}/v)}
\eea
and thus when taking the residue at $x_{2}=q_{4}v$, we have $\underset{x_{2}=q_{4}v}{\Res}x_{2}^{-1}(1-x_{2}/q_{4}v)^{-1}=-1$ which gives an extra sign factor. Note that to relate with the index using $\bfP_{123}^{\vee}$, we rewrote $\mathcal{A}_{\mathbb{C}^{4}}(x_{1}/x_{2})^{-1}=g_{\bar{4}}(x_{1}/x_{2})^{-1}g_{\bar{4}}(x_{2}/x_{1})^{-1}$.

For higher instantons, the discussion is generalized as follows. Let $x^{\ast}_{1},\ldots, x^{\ast}_{k}$ be an ordering \eqref{eq:ordering-partition} and each pole correspond to the $q$-coordinates of the boxes included in the solid partition.\footnote{Note that $x_{1}^{\ast}=v$ because it is the origin of the solid partition.} The residue is then computed as
\bea
\underset{x_{k}=x_{k}^{\ast}}{\Res}x_{k}^{-1}\underset{x_{k-1}=x^{\ast}_{k-1}}{\Res}x_{k-1}^{-1}\cdots \underset{x_{1}=x_{1}^{\ast}=v}{\Res}x_{1}^{-1}\mathcal{Z}_{k}^{\D8\tbar\D0}(v,Kv,\{x_{I}\}).
\eea
By definition, when taking the pole at $x_{I}=x_{I}^{\ast}$, the function only depends on the variables $x_{J}\,\,(J\geq I)$:
\bea
\underset{x_{I-1}=x^{\ast}_{I-1}}{\Res}x_{I-1}^{-1}\cdots \underset{x_{1}=x_{1}^{\ast}=v}{\Res}x_{1}^{-1}\mathcal{Z}_{k}^{\D8\tbar\D0}(v,Kv,\{x_{I}\})=\cdots \prod_{J<I}g_{\bar{4}}\left(\frac{x^{\ast}_{J}}{x_{I}}\right)^{-1}g_{\bar{4}}\left(\frac{x_{I}}{x^{\ast}_{J}}\right)^{-1}.
\eea
When taking the residue at $x_{I}=x_{I}^{\ast}$, Prop.~\ref{prop:residuegeneral} tells us that if the pole $x_{I}=x_{I}^{\ast}$ comes from the term $g_{\bar{4}}(x_{J}^{\ast}/x_{I})^{-1}$ then the residue gives a $+1$ sign, while if it comes from $g_{\bar{4}}(x_{I}/x_{J}^{\ast})^{-1}$, we get a $-1$ sign. By doing this procedure recursively, one can obtain the sign factor $\sigma_{4}(\rho)$ for a given solid partition $\rho$ and its corresponding sequence of poles $\{x^{\ast}_{i}\}_{i=1}^{k}$. See \cite[Sec.~2.4.2]{Nekrasov:2018xsb} for the original proof showing that the sign factor is really given by the formula in \eqref{eq:signfactor-def}. For completeness, the proof is reproduced in Appendix~\ref{app:D8signruleJK-proof} in our notation.

\paragraph{Plethystic exponential formula}
Another important property of the rank $n$ magnificent four partition function is that the partition function does not depend on the Coulomb branch parameters\footnote{For the moment, there is no mathematical proof for this statement yet but we still call it a theorem, a \textit{physical} theorem.} $v_{1},\ldots, v_{n}$ but only the product $\prod_{\alpha}K_{\alpha}$ \cite{Awata:2009dd,Nekrasov:2017cih,Nekrasov:2009JJM,Pomoni:2023nlf, Fasola:2023ypx}. 
\begin{theorem}\label{thm:M4PEformula}
The rank $n$ magnificent four partition function does not depend on Coulomb branch parameters:
\bea\label{eq:D8PE-generalrank}
\mathcal{Z}_{\text{inst.}}^{\D8}[\mathfrak{q},\{K_{\alpha}\}_{\alpha=1}^{n};q_{1,2,3,4}]=\PE\left[\frac{-q_{4}\prod_{i=1}^{3}(1-q_{i4}^{-1})}{\prod_{a\in\four}(1-q_{a})}\frac{1-\prod_{\alpha=1}^{n}K_{\alpha}^{-1}}{(1-\mathfrak{q})(1-\prod_{\alpha=1}^{n}K_{\alpha}^{-1}\mathfrak{q}^{-1})}\right]
\eea
where we explicitly wrote the $\mathfrak{q},\{K_{\alpha}\}_{\alpha=1}^{n},q_{1,2,3,4}$ dependence. The plethystic exponential operator here is defined as
\bea
\PE[f(x_{1},\ldots,x_{n})]=\exp\left(\sum_{\ell=1}^{\infty}\frac{1}{\ell}f(x_{1}^{l},\ldots,x_{n}^{\ell})\right).
\eea
\end{theorem}
For the $n=1$ case, we have
\bea\label{eq:D8PE-rank1}
\mathcal{Z}_{\text{inst.}}^{\D8}[\mathfrak{q},K;q_{1,2,3,4}]=\PE\left[\frac{-q_{4}\prod_{i=1}^{3}(1-q_{i4}^{-1})}{\prod_{a\in\four}(1-q_{a})}\frac{1-K^{-1}}{(1-\mathfrak{q})(1-K^{-1}\mathfrak{q}^{-1})}\right].
\eea
An interesting property of this relation is that the factor inside the PE operator is completely decomposed into parts depending on the equivariant parameters $\{q_{1,2,3,4}\}$ and the part depending on $K,\mathfrak{q}$. Note also that the coefficient factor not depending on $K,\mathfrak{q}$ is actually $\mathcal{G}$.

One can show the above identity \eqref{eq:D8PE-generalrank} for low levels explicitly. Let us focus first on the $n=1$ case and the one-instanton level:
\bea
\mathcal{Z}^{\D8}_{\text{inst.}}[\mathfrak{q},K;q_{1,2,3,4}]&=\PE\left[\frac{-q_{4}\prod_{i=1}^{3}(1-q_{i4}^{-1})}{\prod_{a\in\four}(1-q_{a})}\frac{1-K^{-1}}{(1-\mathfrak{q})(1-K^{-1}\mathfrak{q}^{-1})}\right]\\
&=1+\frac{(1-K)(1-q_{12})(1-q_{13})(1-q_{23})}{(1-q_{1})(1-q_{2})(1-q_{3})(1-q_{123})}\mathfrak{q}+\order{\mathfrak{q}^{2}}
\eea
where in the second line, we expanded the PE formula with respect to the powers of $\mathfrak{q}$. Obviously, this matches with the one-instanton contribution in \eqref{eq:D8one-inst}. After studying the higher instanton contributions, one will see that the extra sign factors given in \eqref{eq:signfactor-def} is crucial when using the equivariant index formalism to get this identity.

Let us move on to the higher rank case. The Coulomb branch parameter independence of the rank two magnificent four partition function can be confirmed at the one-instanton level as follows. The following two terms contribute:
\bea
\mathcal{Z}_{\four;4}^{\D8}[(v_{1},v_{2}),(\hcube,\emptyset);(K_{1},K_{2})]&=\mathcal{G}\times (1-K_{1})\frac{(1-K_{2}v_{2}/v_{1})}{1-v_{2}/v_{1}},\\
\mathcal{Z}_{\four;4}^{\D8}[(v_{1},v_{2}),(\emptyset,\hcube);(K_{1},K_{2})]&=\mathcal{G}\times (1-K_{2})\frac{(1-K_{1}v_{1}/v_{2})}{1-v_{1}/v_{2}}
\eea
and we have
\bea
\mathcal{Z}_{\four;4}^{\D8}[(v_{1},v_{2}),(\hcube,\emptyset);(K_{1},K_{2})]+\mathcal{Z}_{\four;4}^{\D8}[(v_{1},v_{2}),(\emptyset,\hcube);(K_{1},K_{2})]=\mathcal{G}\times (1-K_{1}K_{2}),
\eea
which indeed is independent of the Coulomb branch parameters. Another observation is that, comparing with \eqref{eq:D8one-inst}, the parameter $K$ is effectively modified to $K_{1}K_{2}$.

Assuming that the $\U(1|1)$ partition function obeys \eqref{eq:D8PE-rank1} and the Coulomb branch parameter independence of the partition function, we can derive the rank $n$ PE formula \eqref{eq:D8PE-generalrank}. The expanded version of the rank $n$ partition function is
\bea
\mathcal{Z}^{\D8}_{\text{inst.}}=\sum_{\{\rho^{(\alpha)}\}}\mathfrak{q}^{\sum_{\alpha}|\rho^{(\alpha)}|}\prod_{\alpha=1}^{n}(-1)^{\sigma_{4}(\rho^{(\alpha)})}\prod_{\alpha=1}^{n}\mathcal{Z}_{\four;4}^{\D8}[\rho^{(\alpha)};K_{\alpha}]\prod_{\beta>\alpha}\mathcal{Z}^{\D8\tbar\D8}_{K_{\alpha}|K_{\beta}}(v_{\alpha},\rho^{(\alpha)}\,|\,v_{\beta},\rho^{(\beta)}).
\eea
Using the Coulomb branch independence and the parametrization $v_{\alpha}=L^{N_{\alpha}}$ under the analytic region $|v_{\alpha}|<|v_{\beta}|$ for $\alpha<\beta$, which means $N_{\alpha}<N_{\beta}\,(\alpha<\beta)$, the limit $L\rightarrow \infty$ gives
\bea
\mathcal{Z}^{\D8\tbar\D8}_{K_{\alpha}|K_{\beta}}(v_{\alpha},\rho^{(\alpha)}\,|\,v_{\beta},\rho^{(\beta)})&=\prod_{\shcube'\in\rho^{(\beta)}}\left(\frac{1-K_{\alpha}v_{\alpha}/\chi_{\four,v_{\beta}}(\hcube')}{1-v_{\alpha}/\chi_{\four,v_{\beta}}(\hcube')}\right)\prod_{\shcube\in\rho^{(\alpha)}}\left(K_{\beta}\frac{1-K_{\beta}^{-1}\chi_{\four,v_{\alpha}}(\hcube)/v_{\beta}}{1-\chi_{\four,v_{\alpha}}(\hcube)/v_{\beta}}\right)\\
    &\qquad\qquad\times \prod_{\substack{\shcube\in\rho^{(\alpha)}\\\shcube'\in\rho^{(\beta)}}}\mathcal{A}_{\mathbb{C}^{4}}\left(\frac{\chi_{\four,v_{\alpha}}(\hcube)}{\chi_{\four,v_{\beta}}(\hcube')}\right)^{-1}\\
    &\longrightarrow K_{\beta}^{|\rho^{(\alpha)}|}
\eea
where we used
\bea
\mathcal{A}_{\mathbb{C}^{4}}(x)^{-1}\xrightarrow{x\rightarrow 0}1,\quad \frac{1-Kx}{1-x}\xrightarrow{x\rightarrow 0}1,\quad K\frac{1-K^{-1}x}{1-x}\xrightarrow{x\rightarrow 0}K.
\eea
Therefore, we have
\bea
\mathcal{Z}^{\D8}_{\text{inst.}}&=\lim_{L\rightarrow \infty}\sum_{\vec{\rho}}\mathfrak{q}^{|\vec{\rho}|}\prod_{\alpha=1}^{n}(-1)^{\sigma_{4}(\rho^{(\alpha)})}\prod_{\alpha=1}^{n}\mathcal{Z}_{\four;4}^{\D8}[\rho^{(\alpha)};K_{\alpha}]\prod_{\beta>\alpha}\mathcal{Z}^{\D8\tbar\D8}_{K_{\alpha}|K_{\beta}}(v_{\alpha},\rho^{(\alpha)}\,|\,v_{\beta},\rho^{(\beta)})\\
&=\sum_{\vec{\rho}}\prod_{\alpha}\mathfrak{q}^{|\rho^{(\alpha)}|}\prod_{\alpha=1}^{n}(-1)^{\sigma_{4}(\rho^{(\alpha)})}\prod_{\alpha=1}^{n}\mathcal{Z}_{\four;4}^{\D8}[\rho^{(\alpha)};K_{\alpha}]\prod_{\beta>\alpha}K_{\beta}^{|\rho^{(\alpha)}|}\\
&=\sum_{\vec{\rho}}\prod_{\alpha}\left(\mathfrak{q}\prod_{\beta=\alpha+1}^{n}K_{\beta}\right)^{|\rho^{(\alpha)}|}\prod_{\alpha=1}^{n}(-1)^{\sigma_{4}(\rho^{(\alpha)})}\prod_{\alpha=1}^{n}\mathcal{Z}_{\four;4}^{\D8}[\rho^{(\alpha)};K_{\alpha}]\\
&=\prod_{\alpha=1}^{n}\mathcal{Z}^{\D8}_{\text{inst.}}[\mathfrak{q}_{\alpha},K_{\alpha};q_{1,2,3,4}],
\eea
where we defined $\mathfrak{q}_{\alpha}=\mathfrak{q}\prod_{\beta>\alpha}K_{\beta}$. Using the PE formula \eqref{eq:D8PE-rank1}, we then have
\bea
\mathcal{Z}^{\D8}_{\text{inst.}}=\PE\left[\mathcal{G}\sum_{\alpha=1}^{n}\frac{1-K_{\alpha}^{-1}}{(1-\mathfrak{q}_{\alpha})(1-K_{\alpha}^{-1}\mathfrak{q}_{\alpha}^{-1})}\right]=\PE\left[\mathcal{G}\frac{1-\prod_{\alpha}K_{\alpha}^{-1}}{(1-\mathfrak{q})(1-\prod_{\alpha}K_{\alpha}^{-1}\mathfrak{q}^{-1})}\right]
\eea
where we used the identity (see Lem.~\ref{app:lemm-PEid})
\bea\label{eq:PEidentityproof}
\sum_{\alpha=1}^{n}\frac{1-K_{\alpha}^{-1}}{(1-\mathfrak{q}_{\alpha})(1-K_{\alpha}^{-1}\mathfrak{q}_{\alpha}^{-1})}=\frac{1-\prod_{\alpha=1}^{n}K_{\alpha}^{-1}}{(1-\mathfrak{q})(1-\prod_{\alpha=1}^{n}K_{\alpha}^{-1}\mathfrak{q}^{-1})}
\eea
which gives \eqref{eq:D8PE-generalrank}.

\paragraph{Specialization of PE formulas}
When we tune the parameters $K_{\alpha}$ to special values, the rank $n$ magnificent four partition simplifies. Using Thm.~\ref{thm:M4PEformula}, the rank $n$ partition function with variables $\{K_{\alpha}\}$ is equivalent to the rank one partition function with $K=\prod_{\alpha}K_{\alpha}$. Thus, it is enough to focus on the rank one magnificent four partition function. The PE formula \eqref{eq:D8PE-rank1} shows that when we set $K=1$, the partition function vanishes:
\bea
\mathcal{Z}^{\D8}_{\text{inst.}}[\mathfrak{q},1;q_{1,2,3,4}]=1.
\eea
Physically, this means that when we set the distance between the D8 and $\overline{\text{D8}}$ branes to be zero, which corresponds to $K=1$, the partition function vanishes. This property resembles the fact that when we place a pair of D-brane and ghost (\textit{not anti}) D-brane located at the same position, the system is equivalent to the closed string vacuum \cite{Okuda:2006fb} (see \cite{Vafa:2001qf,Dijkgraaf:2016lym,Kimura:2019msw,Kimura:2023iup} for related references). On the other hand, when we tune the parameter as $K=q_{a}\,(a\in\four)$, the system actually reduces to a D0-D6 system with a D6-brane wrapping $\mathbb{C}^{3}_{\bar{a}}$ (see section~\ref{sec:tetrahedronindex}). Moreover, if we choose $K=q_{1}^{N_{1}}q_{2}^{N_{2}}q_{3}^{N_{3}}q_{4}^{N_{4}}$, we will get a tetrahedron instanton system with $N_{a}$ D6$_{\bar{a}}$-branes. Note that due to the condition $q_{1}q_{2}q_{3}q_{4}=1$, the tetrahedron system with $N_{1}=N_{2}=N_{3}=N_{4}$ is also expected to be equivalent to the string vacuum. This mechanism was conjectured in \cite{Nekrasov:2017cih,Nekrasov:2018xsb} (see also \cite{Pomoni:2023nlf}) to be due to the tachyon condensation of the D8 and $\overline{\text{D8}}$ branes. Combining with the previous result, the magnificent four setup seems to contain both properties of the ghost D-brane and anti D-brane. The complete mechanism happening here is still not so clear and is left for future work.

\paragraph{Removing $\overline{\text{D8}}$-branes} Instead of keeping $K$ to be finite variables, one would like to ask if we can take the limit $K\rightarrow 0,\infty$. Since the parameter $K$ physically is the (exponentiated) distance between the D8 and $\overline{\text{D8}}$ branes, this limit corresponds to removing the $\overline{\text{D8}}$-brane from the system. Generally, such kind of limit might introduce Chern--Simons terms to the system. 

Let us first consider the rank one case and take the limit $K\rightarrow0$. Namely, we are interested in the system coming from the contour integral Prop.~\ref{prop:M4noantiDbrane}. The pole structure in Thm.~\ref{thm:M4JKpoles} does not change and the sign rule \eqref{eq:signfactor-def} is still true. In this case, no additional Chern--Simons term appears and the partition function is given as
\bea
\mathcal{Z}^{\D8}_{\text{inst.}}=\sum_{\rho}\mathfrak{q}^{|\rho|}(-1)^{\sigma_{a}(\rho)}\mathcal{Z}_{\four;a}^{\D8}[\rho],\quad \mathcal{Z}_{\four;a}^{\D8}[\rho]=\mathbb{I}\left[-x^{-1}\bfK_{\rho}+\bfP_{\bar{a}}^{\vee}\bfK_{\rho}^{\vee}\bfK_{\rho}\right]
\eea
where $\bfK_{\rho}=\sum_{\shcube\in\rho}\chi_{\four,x}(\hcube)$. One can show that we also have a PE formula 
\bea
\mathcal{Z}^{\D8}_{\text{inst.}}=\PE\left[\frac{-q_{4}\prod_{i=1}^{3}(1-q_{i4}^{-1})}{\prod_{a\in\four}(1-q_{a})}\frac{\mathfrak{q}}{(1-\mathfrak{q})}\right].
\eea
For example, this can be confirmed for the one-instanton level as
\bea
\mathcal{Z}^{\D8}_{\text{inst.}}&=\PE\left[\frac{-q_{4}\prod_{i=1}^{3}(1-q_{i4}^{-1})}{\prod_{a\in\four}(1-q_{a})}\frac{\mathfrak{q}}{(1-\mathfrak{q})}\right]=1+\frac{(1-q_{12})(1-q_{13})(1-q_{23})}{(1-q_{1})(1-q_{2})(1-q_{3})(1-q_{123})}\mathfrak{q}+\mathcal{O}(\mathfrak{q}^{2}).
\eea
Note that this result can be also obtained by simply setting $K=0$ in \eqref{eq:D8one-inst}.

The partition function of the $n$ D8-branes setup is given by
\bea
&\mathcal{Z}^{\D8}_{\text{inst.}}=\sum_{\vec{\rho}}\mathfrak{q}^{|\vec{\rho}|}(-1)^{\sigma_{4}(\vec{\rho})}\mathcal{Z}^{\D8}_{\four;4}[\vec{v},\vec{\rho}],\quad \mathcal{Z}^{\D8}_{\four;4}[\vec{v},\vec{\rho}]=\prod_{\alpha=1}^{n}\mathcal{Z}_{\four;4}^{\D8}[\rho^{(\alpha)}]\prod_{\beta>\alpha}\mathcal{Z}^{\D8\tbar\D8}(v_{\alpha},\rho^{(\alpha)}\,|\,v_{\beta},\rho^{(\beta)}),\\
    &\mathcal{Z}^{\D8\tbar\D8}(x_{1},\rho^{(1)}\,|\,x_{2},\rho^{(2)})=\mathbb{I}\left[-x_{1}^{-1}\bfK_{\rho^{(2)}}-x_{2}^{-1}\bfK_{\rho^{(1)}}+\bfP_{\four}\bfK^{\vee}_{\rho^{(1)}}\bfK_{\rho^{(2)}}\right].
\label{eq:mag4Nekrasovfact-Kzero}
\eea
Actually, one can show that this partition function also does not depend on the Coulomb branch parameters $\{v_{\alpha}\}_{\alpha=1}^{n}$. For example, the one-instanton level of the rank two case gives
\bea
\,&\mathcal{Z}_{\four;4}^{\D8}[(v_{1},v_{2}),(\hcube,\emptyset)]=\mathcal{G}\times \frac{1}{1-v_{2}/v_{1}},\\
\,&\mathcal{Z}_{\four;4}^{\D8}[(v_{1},v_{2}),(\emptyset,\hcube)]=\mathcal{G}\times \frac{1}{1-v_{1}/v_{2}},\\
\,&\mathcal{Z}_{\four;4}^{\D8}[(v_{1},v_{2}),(\hcube,\emptyset)]+\mathcal{Z}_{\four;4}^{\D8}[(v_{1},v_{2}),(\emptyset,\hcube)]=\mathcal{G},
\eea
which is independent from $v_{1},v_{2}$. Generally, one can show that the partition function of the $n$ D8-branes setup has the following PE formula
\bea
\mathcal{Z}^{\D8}_{\text{inst.}}=\sum_{\vec{\rho}}\mathfrak{q}^{|\vec{\rho}|}(-1)^{\sigma_{4}(\vec{\rho})}\mathcal{Z}^{\D8}_{\four;4}[\vec{v},\vec{\rho}]=\PE\left[\frac{-q_{4}\prod_{i=1}^{3}(1-q_{i4}^{-1})}{\prod_{a\in\four}(1-q_{a})}\frac{\mathfrak{q}}{(1-\mathfrak{q})}\right].
\eea
We have confirmed this equality at low instanton levels for various cases.\footnote{Checking the rank two case is already enough, because one can then recursively use this identity to show that the rank $n$ case gives the same formula.} A different way to obtain this relation is to simply set $K_{\alpha}=0$ in \eqref{eq:D8PE-generalrank} and the proof of Thm.~\ref{thm:M4PEformula}. This formula is surprising because it means that without the $\overline{\text{D8}}$-branes, no matter how many D8-branes we place, we effectively get the contribution of only one D8-brane. We leave the physical explanation of this phenomenon for future work.

\subsection{Tetrahedron instanton}\label{sec:tetrahedronindex}
The contour integral formula of the tetrahedron instanton system in Prop.~\ref{prop:tetracontourJK} is rewritten using the structure functions as
\bea
   \mathcal{Z}^{\D6}_{\text{inst.}}=\sum_{k=0}^{\infty}\mathfrak{q}^{k}\mathcal{Z}^{\D6}_{k},\quad  \mathcal{Z}^{\D6}_{k}=\frac{\mathcal{G}^{k}}{k!}\oint_{\text{JK}} 
\prod_{I=1}^{k}\frac{dx_{I}}{2\pi i x_{I}}\mathcal{Z}_{k}^{\D6\tbar\D0}(v_{\bar{a},\alpha},x_{I})   \label{eq:D6integralmod}
\eea
where
\bea
\mathcal{Z}_{k}^{\D6\tbar\D0}(v_{\bar{a},\alpha},x_{I})=\prod_{a\in\four}\prod_{\alpha=1}^{n_{\bar{a}}}\mathscr{V}_{a}\left(\frac{v_{\bar{a},\alpha}}{x_{I}}\right)\prod_{I<J}\mathcal{A}_{\mathbb{C}^{4}}\left(\frac{x_{I}}{x_{J}}\right)^{-1}.
\eea
In terms of character, we can write it as
\bea
\mathcal{Z}_{k}^{\D6}=\frac{1}{k!}\oint_{\text{JK}}\prod_{I=1}^{k}\frac{dx_{I}}{2\pi i x_{I}}\mathbb{I}\left[\mathbf{v}^{\D6}_{\text{inst.}}-k\right],\quad \mathbf{v}^{\D6}_{\text{inst.}}=-\sum_{a\in\four}\bfP_{a}^{\vee}\bfN_{\bar{a}}^{\vee}\bfK+\sqrt{\bfP_{\four}\bfK^{\vee}\bfK}
\eea
where 
\bea
\bfN_{\bar{a}}=\sum_{\alpha=1}^{n_{\bar{a}}}e^{\mathfrak{a}_{\bar{a},\alpha}}=\sum_{\alpha=1}^{n_{\bar{a}}}v_{\bar{a},\alpha},\quad \bfK=\sum_{I=1}^{k}e^{\phi_{I}}=\sum_{I=1}^{k}x_{I}
\eea
and the square root part is written schematically.

From Thm.~\ref{thm:tetraJKpoles}, the poles are classified by plane partitions:
\bea
&\underline{\vec{v}}=(\vec{v}_{\bar{a}})_{a\in\four}=(v_{\bar{a},\alpha})_{a\in\four}^{\alpha=1,\ldots,n_{\bar{a}}},\quad \underline{\vec{\pi}}=(\vec{\pi}_{\bar{a}})_{a\in\four}=(\pi^{(\alpha)}_{\bar{a}})_{a\in\four}^{\alpha=1,\ldots, n_{\bar{a}}},\quad |\underline{\vec{\pi}}|=\sum_{a\in\four}\sum_{\alpha=1}^{n_{\bar{a}}}|\pi_{\bar{a}}^{(\alpha)}|,\\
&\{x_{I}\}_{I=1,\ldots,k}\longrightarrow \{\chi_{\bar{a},v_{\bar{a},\alpha}}(\cube)\}_{a\in\four,\,\,\scube\,\in \pi_{\bar{a}}^{(\alpha)}}^{\alpha=1,\ldots,n_{\bar{a}}},\quad \chi_{abc,v}(\cube)=vq_{a}^{i-1}q_{b}^{j-1}q_{c}^{k-1}.
\eea
After evaluation of the poles, the character becomes
\bea
\bfK|_{\vec{\underline{\vec{\pi}}}}=\sum_{a\in\four}\bfK_{\bar{a}}|_{\vec{\pi}_{a}},\quad \left.\bfK_{\bar{a}}\right|_{\vec{\pi}_{\bar{a}}}=\sum_{\alpha=1}^{n_{\bar{a}}}\sum_{\scube\in\pi_{\bar{a}}^{(\alpha)}}\chi_{\bar{a},v_{\bar{a},\alpha}}(\cube),\quad a\in\four.
\eea
The correct choice of the square root part that does not give any sign factor (see Thm.~\ref{thm:D6signruleJK}) is
\bea
    \mathbf{v}^{\D6}_{\text{inst.}}&=\sum_{a\in\four}\left(-\mathbf{P}_{a}^{\vee}\bfN_{\bar{a}}^{\vee}\bfK_{\bar{a}}+\frac{\bfP_{\four}}{\bfP^{\vee}_{
a}}\bfK_{\bar{a}}^{\vee}\bfK_{\bar{a}}\right)-\sum_{a\in\four}\sum_{b\neq a}\bfP_{a}^{\vee}\bfN_{\bar{a}}^{\vee}\bfK_{\bar{b}}+\sum_{a<b}\bfP_{\four}\bfK_{\bar{a}}^{\vee}\bfK_{\bar{b}}
\eea
where
\bea
\sqrt{\bfP_{\four}\bfK^{\vee}\bfK}=\sum_{a\in\four}\frac{\bfP_{\four}}{\bfP^{\vee}_{
a}}\bfK_{\bar{a}}^{\vee}\bfK_{\bar{a}}+\sum_{a<b}\bfP_{\four}\bfK_{\bar{a}}^{\vee}\bfK_{\bar{b}}.
\eea
Note that we have
\bea
\left(\sum_{a\in\four}\frac{\bfP_{\four}}{\bfP^{\vee}_{
a}}\bfK_{\bar{a}}^{\vee}\bfK_{\bar{a}}+\sum_{a<b}\bfP_{\four}\bfK_{\bar{a}}^{\vee}\bfK_{\bar{b}}\right)+\left(\sum_{a\in\four}\frac{\bfP_{\four}}{\bfP^{\vee}_{
a}}\bfK_{\bar{a}}^{\vee}\bfK_{\bar{a}}+\sum_{a<b}\bfP_{\four}\bfK_{\bar{a}}^{\vee}\bfK_{\bar{b}}\right)^{\vee}=\bfP_{\four}\bfK^{\vee}\bfK.
\eea
The relationship between the JK residue and the character is given below.

\begin{proposition}
The JK residue of the D6--D0 integrand is related to the index of the character as
\bea
    \mathcal{G}^{k}\times \underset{x=x_{\underline{\vec{\pi}}}}{\Res}x^{-1}\mathcal{Z}_{k}^{\D6\tbar\D0}(v_{A,\alpha},x_{I})=\mathcal{Z}^{\D6}_{\text{tet.inst.}}[\vec{\underline{v}},\vec{\underline{\lambda}}],
\eea    
where
\bea\label{eq:D6tetinst_partfunct}
    \mathcal{Z}_{\text{tet.inst.}}^{\D6}[\vec{\underline{v}},\vec{\underline{\pi}}]&=\prod_{a\in\four}\prod_{\alpha=1}^{n_{\bar{a}}}\widetilde{\mathcal{Z}}^{\D6}_{\bar{a}}[\pi_{\bar{a}}^{(\alpha)}]\prod_{a\in\four}\prod_{1\leq\alpha<\beta\leq n_{\bar{a}}}\mathcal{Z}^{\D6\tbar\D6}_{\bar{a}|\bar{a}}(v_{\bar{a},\alpha},\pi_{\bar{a}}^{(\alpha)}\,|\,v_{\bar{a},\beta},\pi_{\bar{a}}^{(\beta)}) \\
    &\qquad \times \prod_{a<b}\prod_{\alpha=1}^{n_{\bar{a}}}\prod_{\beta=1}^{n_{\bar{b}}}\mathcal{Z}^{\D6\tbar\D6}_{\bar{a}|\bar{b}}(v_{\bar{a},\alpha},\pi_{\bar{a}}^{(\alpha)}\,|\,v_{\bar{b},\beta},\pi_{\bar{b}}^{(\beta)}),
\eea
and
\bea
    \widetilde{\mathcal{Z}}^{\D6}_{\bar{a}}[\pi]&=\mathbb{I}\left[-\bfP_{a}^{\vee}\bfN_{\bar{a}}^{\vee}\bfK_{\bar{a},\pi}+\bfP_{\bar{a}}^{\vee}\bfK_{\bar{a},\pi}^{\vee}\bfK_{\bar{a},\pi}\right],\\
    \mathcal{Z}_{\bar{a}|\bar{b}}^{\D6\tbar\D6}(v_{1},\pi^{(1)}\,|\,v_{2},\pi^{(2)})&=\mathbb{I}\left[-\bfP_{a}^{\vee}\bfN_{\bar{a}}^{\vee}\bfK_{\bar{b},\pi^{(2)}}-\bfP_{b}^{\vee}\bfN_{\bar{b}}^{\vee}\bfK_{\bar{a},\pi^{(1)}}+\bfP_{\four}\bfK_{\bar{a},\pi^{(1)}}^{\vee}\bfK_{\bar{b},\pi^{(2)}}\right].
\eea

\end{proposition}

\begin{theorem}
The total instanton partition function of the tetrahedron instanton system is
\bea
\mathcal{Z}_{\text{inst.}}^{\D6}=\sum_{\vec{\underline{\pi}}}\mathfrak{q}^{|\vec{\underline{\pi}}|}\mathcal{Z}^{\D6}_{\text{tet.inst.}}[\underline{\vec{v}},\vec{\underline{\pi}}].
\eea 
\end{theorem}

\begin{remark}
    Compared with the contour integral formula, the equivariant index formalism is rather complicated. The proper square root depends on which non-compact subspaces the D6-branes are spanning. For example, the $\U(1)$ part of the 7d theory comes from $\widetilde{\mathcal{Z}}_{\bar{4}}^{\D6}[\pi]$ whose collision term comes from $\bfP_{123}^{\vee}$ which corresponds to the fact that it is a theory on the subspace $\mathbb{C}^{3}_{123}\times \mathbb{S}^{1}$. Generally, one may use a different square root choice but then one needs to introduce the correct sign factors. It was mathematically proven that the above choice does not give any sign factor in~\cite{Fasola:2023ypx}. Later we will also see this property from the quantum algebraic view point.
\end{remark}


\paragraph{One-loop perturbative part}
Similar to the magnificent four system, we can define the one-loop perturbative part as
\bea
    \mathring{\mathbf{v}}=\sum_{(b,\beta)>(a,\alpha)}\frac{\bfP_{a}^{\vee}\bfP_{b}}{\bfP_{\four}}v_{\bar{b},\beta}/v_{\bar{a},\alpha}
\eea
where we specified an order in the pair of indices $(\bar{a},\alpha)_{a\in\four}^{\alpha=1,\ldots,n_{\bar{a}}}$. Taking the index gives 
\begin{equation}
\begin{split}
&\mathbb{I}[\mathring{\mathbf{v}}]=\prod_{(b,\beta)>(a,\alpha)}\mathcal{Z}^{\D6\tbar\D6}_{\text{1-loop}}(v_{\bar{a},\alpha},\bar{a}\,|\,v_{\bar{b},\beta},\bar{b})\eqqcolon\mathcal{Z}^{\D6}_{\text{1-loop}},\\
&\mathcal{Z}^{\D6\tbar\D6}_{\text{1-loop}}(x_{1},\bar{a}\,|\,x_{2},\bar{b})=\exp\left(-\sum_{n>0}\frac{1}{n}\frac{\bfP_{a}^{[n]}\bfP_{b}^{[-n]}}{\bfP_{\four}^{[n]}}\left(\frac{x_{1}}{x_{2}}\right)^{n}\right).
\label{eq:D6oneloop}
\end{split}
\end{equation}

\paragraph{JK residue and sign rules}
Similar to the D8 setup, let us relate the JK residue result and the partition function computed by the index formalism. Focusing on the 7d $\U(1)$ theory on $\mathbb{C}^{3}_{\bar{a}}\times \mathbb{S}^{1}$, we have the following identification.
\begin{theorem}\label{thm:D6signruleJK}
    The JK residue of the D6--D0 integrand obeys
    \bea
    \mathcal{G}^{k}\times\underset{x=x_{\pi}}{\Res}x^{-1}\mathcal{Z}_{\bar{a},k}^{\D6\tbar\D0}(v,x_{I})=\widetilde{\mathcal{Z}}^{\D6}_{\bar{a}}[\pi],\quad \pi\in\mathcal{PP}_{a}
    \eea
    where 
    \bea
    \mathcal{Z}_{\bar{a},k}^{\D6\tbar\D0}(v,x_{I})=\mathscr{V}_{a}\left(\frac{v}{x_{I}}\right)\prod_{I<J}\mathcal{A}_{\mathbb{C}^{4}}\left(\frac{x_{I}}{x_{J}}\right)^{-1}.
    \eea
\end{theorem}
 A direct proof can be obtained similar to the D8 case but we omit it (see \cite[Lemma B.1]{Kimura:2024osv} for example). Instead of directly proving this, we take a shortcut and obtain the above identity. First of all, let us set $a=4$. The D6$_{\bar{4}}$--D0 integrand is obtained by starting from a rank one magnificent four setup and setting $K=q_{4}$. Under this identification, the solid partition will terminate its growth in the 4th direction and the poles will be classified by plane partitions. With this restriction, from Thm.~\ref{thm:D8signruleJK}, we have
    \bea
     \mathcal{G}^{k}\times\underset{x=x_{\pi}}{\Res}x^{-1}\mathcal{Z}_{\bar{4},k}^{\D6\tbar\D0}(v,x_{I})=(-1)^{\sigma_{4}(\pi)}\mathcal{Z}_{\four;4}^{\D8}[\pi,q_{4}]
    \eea
    where we embedded the plane partition into the solid partition. We then have
    \bea
    \mathcal{Z}^{\D8}_{\four;4}[\pi,q_{4}]=\widetilde{\mathcal{Z}}^{\D6}_{\bar{4}}[\pi].
    \eea
    The sign factor is determined as
    \bea
    \sigma_{4}(\pi)=\#\{(x_{1},x_{2},x_{3},x_{4})\in\pi\mid x_{1}=x_{2}=x_{3}<x_{4}\}=\#\{(i,i,i,1)\in i<1\}=0
    \eea
because the plane partition $\pi$ does not extend into the 4th direction. We can do the same discussion for other $a\in\four$ by using Thm.~\ref{thm:D8signindep} and the quadrality symmetry and obtain
\bea
\mathcal{G}^{k}\times \underset{x=x_{\pi}}{\Res}x^{-1}\mathcal{Z}_{\bar{a},k}^{\D6\tbar\D0}(v,x_{I})=(-1)^{\sigma_{a}(\pi)}\mathcal{Z}_{\four;a}^{\D8}[\pi,q_{a}]=\widetilde{\mathcal{Z}}_{\bar{a}}^{\D6}[\pi],\quad \pi\in\mathcal{PP}_{a}.
\eea

Using the same discussion, we can also get a different relation between the JK residue result and the equivariant index result:
\bea
 \mathcal{G}^{k}\times \underset{x=x_{\pi}}{\Res}x^{-1}\mathcal{Z}_{\bar{4},k}^{\D6\tbar\D0}(v,x_{I})=(-1)^{\sigma_{b}(\pi)}\mathcal{Z}_{\four;b}^{\D8}[\pi,q_{b}],\quad b=1,2,3
\eea
where $\pi\in\mathcal{PP}_{4}$. The right hand side is the partition function obtained by using a different square root part such as $\bfP_{234}^{\vee}$ for the theory on $\mathbb{C}^{3}_{\bar{4}}\times \mathbb{S}^{1}$. However, the above equation says that one needs to introduce a proper sign factor. Note that the sign rule can be rewritten as
\bea
\sigma_{a}(\pi)&=\#\{(x_{1},x_{2},x_{3},x_{4})\in\pi\mid x_{b}=x_{c}=x_{4}<x_{a}\}\\
&=\#\{x_{b}=x_{c}=1<x_{a}\}
\eea
for $a=1,2,3$. Moreover, it is related to the height of the plane partition in the $a$-th direction as
\bea
\sigma_{a}(\pi)=\min\{h_{a}(\pi)-1,0\},\quad h_{a}(\pi)\coloneq\min\{x_{a}\geq 1\mid (x_{a},x_{b},x_{c})=(x_{a},1,1)\notin\pi\}.
\eea
Summarizing the above discussion, we also have the following identification.
\begin{corollary}\label{cor:D8-D6reduce-sign}
After specializing $K=q_{a}\,(a\in\four)$, we have the following identity:
\bea
(-1)^{\sigma_{4}(\pi)}\mathcal{Z}_{\four;4}^{\D8}[\pi,q_{a}]=\widetilde{\mathcal{Z}}^{\D6}_{\bar{a}}[\pi],\quad a\in\four
\eea
where $\pi\in\mathcal{PP}_{a}$. 
\end{corollary}

\paragraph{Plethystic exponential formula}
Similar to the magnificent four partition function, it is known that the tetrahedron instanton partition functions do not depend on the Coulomb branch parameters. Since the D6-brane setup can be obtained by tuning the parameters $K_{\alpha}$ of the magnificent four partition functions, Thm.~\ref{thm:M4PEformula} gives the following.
\begin{theorem}\label{thm:tetrahedronPEformula}
    The tetrahedron instanton partition function of $n_{\bar{a}}$ D6-branes wrapping $\mathbb{C}^{3}_{\bar{a}}\times \mathbb{S}^{1}$ does not depend on the Coulomb branch parameters and we have
    \bea
    \mathcal{Z}_{\text{inst.}}^{\D6}[\mathfrak{q},(n_{1},n_{2},n_{3},n_{4})\,;q_{1,2,3,4}]=\PE\left[\frac{-q_{4}\prod_{i=1}^{3}(1-q_{i4}^{-1})}{\prod_{a\in\four}(1-q_{a})}\frac{1-\prod_{a\in\four}q_{a}^{-n_{\bar{a}}}}{(1-\mathfrak{q})(1-\prod_{a\in\four}q_{a}^{-n_{\bar{a}}}\mathfrak{q}^{-1})}\right].
    \eea
\end{theorem}
Let us see this explicitly\footnote{See \cite{Nekrasov:2014nea} for a mathematical derivation of this formula.} for the $\U(1)$ theory on $\mathbb{C}^{3}_{\bar{4}}\times \mathbb{S}^{1}$. The PE formula first gives
\bea\label{eq:D6onebrane-PE}
\mathcal{Z}^{\D6}_{\bar{4}}[\mathfrak{q}]\coloneq\mathcal{Z}_{\text{inst.}}^{\D6}[\mathfrak{q},(0,0,0,1)\,;q_{1,2,3,4}]&=\PE\left[\frac{-q_{4}\prod_{i=1}^{3}(1-q_{i4}^{-1})}{\prod_{a\in\four}(1-q_{a})}\frac{1-q_{4}^{-1}}{(1-\mathfrak{q})(1-q_{4}^{-1}\mathfrak{q}^{-1})}\right]\\
&=\PE\left[-q_{4}\frac{(1-q_{12})(1-q_{23})(1-q_{13})}{(1-q_{1})(1-q_{2})(1-q_{3})}\frac{\mathfrak{q}}{(1-\mathfrak{q})(1-q_{4}\mathfrak{q})}\right]\\
&=1-q_{4}\frac{(1-q_{12})(1-q_{13})(1-q_{23})}{(1-q_{1})(1-q_{2})(1-q_{3})}\mathfrak{q}+\order{\mathfrak{q}^{2}}.
\eea
The one-instanton contribution comes from
\bea
\mathbb{I}[-\bfP_{4}^{\vee}+\bfP_{123}^{\vee}]=-q_{4}\frac{(1-q_{12})(1-q_{13})(1-q_{23})}{(1-q_{1})(1-q_{2})(1-q_{3})}
\eea
which matches with the PE formula.

Assuming that the PE formulas for the $\U(1)$ theories are \eqref{eq:D6onebrane-PE} and that the partition function do not depend on the Coulomb branch parameters, one can get the general case (Thm.~\ref{thm:tetrahedronPEformula}) by doing the same procedure done in the previous section. Since it is straightforward, we omit the discussion.

Let us study the case when there is one D6$_{\bar{4}}$-brane and discuss its physical interpretation. The perturbative part can be computed as
\bea
\mathcal{Z}^{\D6}_{\text{inst.}}[(0,0,0,1);q_{1,2,3,4}]=\PE\left[\mathcal{F}^{\text{pert}}_{(0,0,0,1)}(q_{1},q_{2},q_{3},q_{4})\right],\quad \mathcal{F}^{\text{pert}}_{(0,0,0,1)}(q_{1},q_{2},q_{3},q_{4})=-\frac{\bfP_{4}^{\vee}}{\bfP_{123}},
\eea
where $(0,0,0,1)$ denotes the number of D6-branes $\vec{n}=(n_{1},n_{2},n_{3},n_{4})$.
See the physical derivation of this perturbative part in \cite{Nekrasov:2009JJM}. One can also obtain this by setting $a=b=4$ and $x_{1}=x_{2}$ in \eqref{eq:D6oneloop}. Let us define 
\bea
\mathcal{F}_{\vec{n}}^{\text{inst.}}(q_{1},q_{2},q_{3},q_{4})=\frac{-q_{4}\prod_{i=1}^{3}(1-q_{i4}^{-1})}{\prod_{a\in\four}(1-q_{a})}\frac{1-\prod_{a\in\four}q_{a}^{-n_{\bar{a}}}}{(1-\mathfrak{q})(1-\prod_{a\in\four}q_{a}^{-n_{\bar{a}}}\mathfrak{q}^{-1})}
\eea
and then we have $\mathcal{Z}^{\D6}_{\text{inst.}}[\mathfrak{q},\vec{n};q_{1,2,3,4}]=\PE\left[\mathcal{F}_{\vec{n}}^{\text{inst.}}(q_{1,2,3,4})\right]$. It is intuitive to redefine the parameters as
\bea
\widetilde{\mathsf{q}}_{1,2,3}=q_{1,2,3},\quad \widetilde{\mathsf{q}}_{4}=\mathfrak{q}^{-1},\quad \widetilde{\mathsf{q}}_{5}=q_{4}\mathfrak{q}
\eea
which obeys 
\bea\label{eq:CY5cond}
\widetilde{\mathsf{q}}_{1}\widetilde{\mathsf{q}}_{2}\widetilde{\mathsf{q}}_{3}\widetilde{\mathsf{q}}_{4}\widetilde{\mathsf{q}}_{5}=1. 
\eea
The perturbative and instanton parts are then rewritten as
\bea
\mathcal{F}^{\text{pert.}}_{(0,0,0,1)}(\widetilde{\mathsf{q}}_{1,2,3,4,5})=-\frac{(1-\widetilde{\mathsf{q}}_{123})}{(1-\widetilde{\mathsf{q}}_{1})(1-\widetilde{\mathsf{q}}_{2})(1-\widetilde{\mathsf{q}}_{3})},\quad \mathcal{F}^{\text{inst.}}_{(0,0,0,1)}(\widetilde{\mathsf{q}}_{1,2,3,4,5})=\frac{\widetilde{\mathsf{q}}_{45}(1-\widetilde{\mathsf{q}}_{12})(1-\widetilde{\mathsf{q}}_{13})(1-\widetilde{\mathsf{q}}_{23})}{\prod_{i=1}^{5}(1-\widetilde{\mathsf{q}}_{i})}.
\eea
Using \eqref{eq:CY5cond}, the sum is computed as
\bea\label{eq:SUGRA}
\mathcal{F}^{\text{pert.}}_{(0,0,0,1)}+\mathcal{F}^{\text{inst.}}_{(0,0,0,1)}=\frac{\sum_{i=1}^{5}\widetilde{\mathsf{q}}_{i}}{\prod_{i=1}^{5}(1-\widetilde{\mathsf{q}}_{i})}+\frac{\sum_{i=1}^{5}\widetilde{\mathsf{q}}_{i}^{-1}}{\prod_{i=1}^{5}(1-\widetilde{\mathsf{q}}_{i}^{-1})}.
\eea
Actually, the right hand side of this equation coincides with the supergravity index of M-theory \cite{Nekrasov:2009JJM, Nekrasov:2014nea, Pomoni:2023nlf}.

\paragraph{Relation with supergravity index}
Let us introduce the M-theory index on a non-compact toric Calabi--Yau 5-fold $\mathcal{X}$. We follow the discussion of \cite[Sec.~5]{Pomoni:2023nlf} (see also \cite{Haupt:2008nu, Nekrasov:2009JJM, Nekrasov:2014nea,DelZotto:2021gzy}). M-theory preserves an 11-dimensional Majorana spinor with 32 real components. When we take the background space-time to be $\mathbb{S}^{1}\times \mathcal{X}$, two supercharges remain and the effective field theory on $\mathbb{S}^{1}$ is just a 1d $\mathcal{N}=2$ SQM. Note that since $\mathcal{X}$ is noncompact, we neglected the gravitational effects. The $\U(1)^{5}$ isometries of $\mathcal{X}$ play the roles of the global symmetries of this SQM. Among them, four linearly independent combinations of the Cartan generators of this global symmetry commutes with the supercharge.\footnote{This is similar to the story in $\mathbb{C}^{4}$ gauge origami, where we introduced $q_{1,2,3,4}$ with the condition $q_{1}q_{2}q_{3}q_{4}=1$.} Similar to what we did in section~\ref{sec:2susy2d-flavorsymmetry}, we introduce flavor fugacities $\widetilde{\mathsf{q}}_{i}$ and define the index as
\bea
\mathcal{Z}_{\mathcal{X}}(\widetilde{\mathsf{q}}_{1},\widetilde{\mathsf{q}}_{2},\widetilde{\mathsf{q}}_{3},\widetilde{\mathsf{q}}_{4},\widetilde{\mathsf{q}}_{5})=\Tr (-1)^{F}e^{-\beta H}\prod_{i=1}^{5}\widetilde{\mathsf{q}}_{i}^{T_{i}},\quad \prod_{i=1}^{5}\widetilde{\mathsf{q}}_{i}=1,
\eea
where $F$ is the fermion number operator, $\beta$ is the circumference of $\mathbb{S}^{1}$, and $T_{i}$ are the Cartan generators of the global symmetries. Nekrasov--Okounkov gave a simple expression for $\mathcal{Z}_{X}$:
\bea
\mathcal{Z}_{\mathcal{X}}(\widetilde{\mathsf{q}}_{1,2,3,4,5})=\PE\left[\mathcal{F}_{\mathcal{X}}(\widetilde{\mathsf{q}}_{1,2,3,4,5})\right]
\eea
where 
\bea
\mathcal{F}_{\mathcal{X}}(\widetilde{\mathsf{q}}_{1,2,3,4,5})=\int_{\mathcal{X}}\operatorname{ch}(T^{\ast}\mathcal{X}\ominus T\mathcal{X})\wedge \operatorname{td}(T\mathcal{X}).
\eea
For $\mathcal{X}=\mathbb{C}^{5}$, we have
\bea
\operatorname{ch}(T\mathbb{C}^{5})=\sum_{i=1}^{5}\widetilde{\mathsf{q}}_{i}=\sum_{i=1}^{5}e^{\tilde{\epsilon}_{i}},\quad \operatorname{td}(T\mathbb{C}^{5})=\prod_{i=1}^{5}\frac{\tilde{\epsilon}_{i}}{1-e^{-\tilde{\epsilon}_{i}}}
\eea
which reproduces the previous formula \eqref{eq:SUGRA}. A different way to derive this index was also performed in the context of twisted supergravity \cite{Raghavendran:2021qbh}.

The agreement of the two partition functions is understood as follows. Consider $n$ D6$_{\bar{4}}$-branes wrapping $\mathbb{C}^{3}_{\bar{4}}\times \mathbb{S}^{1}$ and D0-branes wrapping $\mathbb{S}^{1}$ in type IIA string theory. This setup is realized in M-theory by considering the geometry to be $\mathbb{S}^{1}\times \mathbb{C}^{3}\times \mathbb{TN}_{n}$, where $\mathbb{TN}_{n}$ is the $n$-centered multi-Taub-NUT space approaching $\mathbb{R}^{3}\times \mathbb{S}^{1}_{R}\simeq \mathbb{C}_{4}\times \mathbb{R}\times \mathbb{S}^{1}_{R}$ at the infinity, where $\mathbb{S}^{1}_{R}$ is the eleventh-dimensional direction and $R$ is the radius of it. In the limit $R\rightarrow \infty$, $\mathbb{TN}_{n}$ becomes $\widetilde{\mathbb{C}^{2}/\mathbb{Z}_{n}}$. Since both setups are related with string duality, it is natural to expect
\bea
\mathcal{F}_{(0,0,0,n)}^{\text{pert.}}+
\mathcal{F}_{(0,0,0,n)}^{\text{inst.}}=\mathcal{F}_{\mathbb{C}^{3}\times \widetilde{\mathbb{C}^{2}/\mathbb{Z}_{n}}}.
\eea
For $n=1$, this is just the equality given above. For discussion on general $n$, see \cite{Nekrasov:2009JJM, Nekrasov:2014nea}. For the case of the tetrahedron instanton setup, namely for general $\vec{n}$, the resulting $\mathcal{X}$ is conjectured to be obtained by superpositions of the Taub-NUT spaces (see \cite{Pomoni:2023nlf}).


\paragraph{Limits of PE formulas and MacMahon functions}
There are various limits when the PE formula in \eqref{eq:D6onebrane-PE} simplifies and connections with generating functions of plane partitions become much clearer. In this section, we will study three different limits and also briefly discuss application to physics. The three limits are the unrefined limit, refined limit, Macdonald refined limit.
\begin{itemize}[topsep=1.5pt,itemsep=0.5ex,partopsep=1ex,parsep=1ex] 
    \item \textbf{Unrefined limit}: The unrefined limit is obtained by taking the limit $q_{4}\rightarrow 1$. Under this limit, the PE formula transforms as
    \bea
    \mathcal{Z}_{\bar{4}}^{\D6}[\mathfrak{q}]\xrightarrow{q_{4}\rightarrow 1}\PE\left[\frac{\mathfrak{q}}{(1-\mathfrak{q})^{2}}\right].
    \eea
    On the other hand, the expanded version of the instanton partition function transforms as
    \bea
    \mathcal{Z}^{\D6}_{\text{inst.}}=\sum_{\pi\in\mathcal{PP}}\mathfrak{q}^{|\pi|}\widetilde{\mathcal{Z}}^{\D6}_{\bar{4}}[\pi]\xrightarrow{q_{4}\rightarrow 1} \sum_{\pi\in\mathcal{PP}}\mathfrak{q}^{|\pi|}
    \eea
    because all of the instanton contribution trivializes under this limit. We then have the following identity
\bea
M_{\text{3d}}(\mathfrak{q})\coloneq\PE\left[\frac{\mathfrak{q}}{(1-\mathfrak{q})^{2}}\right]=\prod_{n=1}^{\infty}\frac{1}{(1-\mathfrak{q}^{n})^{n}}=\sum_{\pi}\mathfrak{q}^{|\pi|}.
\eea
    This is the famous MacMahon formula giving the closed formula for the generating function of plane partitions and $M_{\text{3d}}(\mathfrak{q})$ is usually called the MacMahon function. This formula itself can be proved by the free fermions, vertex operators and the relation with Schur functions \cite{Okounkov:2003sp}. One may generalize the situation and consider the generating function of plane partitions with nontrivial asymptotic Young diagrams. In such cases, the generating function will be related with the so-called \textit{unrefined topological vertex} \cite{Iqbal:2003ds,Aganagic:2003db}. 
    
Compared to the MacMahon's formula, the original D6 partition is a generating function of plane partitions with nontrivial weights determined by $q_{1,2,3,4}$. The unrefined limit is a simplification that makes all this nontrivial weights trivial. Such kind of limit does not exist in the magnificent four case because a naive $q_{a}\rightarrow 1$ limit makes the PE formula diverge.

    \item \textbf{Refined limit}: Another intuitive limit is the refined limit. We first rescale the equivariant parameters as 
    \bea
    (q_{1},q_{2},q_{3},q_{4})\mapsto(z^{r_{1}}q_{1},z^{r_{2}}q_{2},z^{r_{3}}q_{3},q_{4})
    \eea 
    where
    \bea
    r_{1}+r_{2}+r_{3}=0,\quad r_{1}\gg r_{3}>0\gg r_{2}
    \eea
    and then take the limit $z\rightarrow 0$. Under this limit, we have
    \bea
    -q_{4}\frac{(1-z^{-r_{3}}q_{12})(1-z^{-r_{2}}q_{13})(1-z^{-r_{1}}q_{23})}{(1-z^{r_{1}}q_{1})(1-z^{r_{2}}q_{2})(1-z^{r_{3}}q_{3})}\xrightarrow{z\rightarrow 0}-q_{4}\frac{(-z^{-r_{3}}q_{12})(-z^{-r_{1}}q_{23})}{(-z^{r_{2}}q_{2})}=1.
    \eea
    We then obtain
    \bea
\mathcal{Z}_{\bar{4}}^{\D6}[\mathfrak{q}]\rightarrow \PE\left[\frac{\mathfrak{q}}{(1-\mathfrak{q})(1-q_{4}\mathfrak{q})}\right]=\PE\left[\frac{\mathfrak{q}}{(1-\mathfrak{q})(1-\mathfrak{t})}\right]\eqcolon M_{\text{3d}}(\mathfrak{q},\mathfrak{t})
    \eea
where we defined $\mathfrak{t}=q_{4}\mathfrak{q}$. The function $M_{\text{3d}}(\mathfrak{t},\mathfrak{q})$ is called the refined MacMahon function and it has an expanded formula as
\bea
M_{\text{3d}}(\mathfrak{t},\mathfrak{q})=\sum_{\pi}\mathfrak{t}^{\sum_{i=1}^{\infty}|\pi(-i)|}\mathfrak{q}^{\sum_{j=1}^{\infty}|\pi(j-1)|}=\sum_{\pi}\mathfrak{q}^{|\pi|}q_{4}^{\sum_{i=1}^{\infty}|\pi(-i)|}
\eea
where $\pi(i)\,(i\in\mathbb{Z})$ is a diagonal slice of the plane partition $\pi$ at $x_{2}+x_{1}=i$. This formula also can be proved by the transfer matrix method and it is related with Schur functions and Macdonald symmetric functions. Generalization to plane partitions with nontrivial asymptotic Young diagrams can be done and it will be related with the \textit{refined topological vertex} \cite{Iqbal:2007ii,Awata:2008ed,Taki:2007dh,Nekrasov:2014nea}.

Similar to the previous case, one can also show that under this limit, we have
\bea
\mathfrak{q}^{|\pi|}\widetilde{\mathcal{Z}}^{\D6}_{\bar{4}}[\pi]\rightarrow \mathfrak{t}^{\sum_{i=1}^{\infty}|\pi(-i)|}\mathfrak{q}^{\sum_{j=1}^{\infty}|\pi(j-1)|}.
\eea
We will not show this explicitly in this thesis, but one can implement all the formulas in a computer program and check it for lower levels.

    \item \textbf{Macdonald refined limit}: A different limit is the Macdonald refined limit 
    \bea
    (q_{1},q_{2},q_{3},q_{4})\mapsto (z^{r}q_{1},z^{-r}q_{2},q_{3},q_{4})
    \eea
    with $z\rightarrow 0$. Under this limit, we have 
    \bea
-q_{4}\frac{(1-q_{34}^{-1})(1-z^{r}q_{13})(1-z^{-r}q_{23})}{(1-z^{r}q_{1})(1-z^{-r}q_{2})(1-q_{3})}\xrightarrow{z\rightarrow 0}-q_{4}\frac{(1-q_{34}^{-1})(-z^{-r}q_{23})}{(1-q_{3})(-z^{-r}q_{2})}=\frac{(1-q_{34})}{(1-q_{3})}
    \eea
    which gives 
    \bea
    \mathcal{Z}_{\bar{4}}^{\D6}[\mathfrak{q}]\rightarrow \PE\left[\frac{1-q_{34}}{1-q_{3}}\frac{\mathfrak{q}}{(1-\mathfrak{q})(1-q_{4}\mathfrak{q})}\right].
    \eea
    This Macdonald refined limit has not been studied much compared with the previous two limits. Actualy, the PE formula above is related with the one given in \cite[Eq.~(1.7)--(1.8)]{Foda:2017tnv} 
    \bea
    \PE\left[\frac{1-t}{1-q}\frac{x}{(1-x)(1-y)}\right]
    \eea
     after the specialization $x=\mathfrak{q},\,\, y=q_{4}\mathfrak{q},\,\, q=q_{3},\,\, t=q_{34}$.      
     Historically, Vuleti\'c introduced a two parameter generalization of MacMahon's formula taking the form as
     \bea
     \sum_{\pi}F_{\pi}(q,t)x^{|\pi|}=\PE\left[\frac{1-t}{1-q}\frac{x}{(1-x)^{2}}\right]=\prod_{i,n+1=1}^{\infty}\left(\frac{1-x^{i}q^{n}t}{1-x^{i}q^{n}}\right)^{i}
     \eea
where the explicit formula for $F_{\pi}(q,t)$ is in \cite[Eq.~(1.4)]{Vuletic2007AGO}. Since the right hand side is directly related with the Macdonald kernel function, it is called the Macdonald refinement of MacMahon's function. Later, Cai--Wang--Wu--Yang studied the associated vertex operators in \cite{Cai2015TheVO}, and Foda--Wu introduced another refinement $x,y$ to it, taking the form as
\bea
\sum_{\pi}F_{\pi}(q,t)x^{\sum_{j=1}^{\infty}|\pi(j-1)|}y^{\sum_{i=1}^{\infty}|\pi(-i)|}=\PE\left[\frac{1-t}{1-q}\frac{x}{(1-x)(1-y)}\right],
\eea
and discussed applications to supersymmetric gauge theories and topological string theory in \cite{Foda:2017tnv} (see also \cite{Foda:2018jwz}). The difference between the limit of the D6 partition function and Foda--Wu's formula is that while we only have three independent deformation parameters $q_{3},q_{4},\mathfrak{q}$ here, theirs have four independent deformation parameters $x,y,q,t$. In this sense, the PE formula obtained from the D6 partition function is some specialization of their setup. We abuse the terminology and still call this limit the \textit{Macdonald refined limit}. The explicit correspondence
\bea
\mathfrak{q}^{|\pi|}\widetilde{\mathcal{Z}}^{\D6}_{\bar{4}}[\pi]\longrightarrow \left. F_{\pi}(q,t)x^{\sum_{j=1}^{\infty}|\pi(j-1)|}y^{\sum_{i=1}^{\infty}|\pi(-i)|} \right|_{x=\mathfrak{q},\,\, y=q_{4}\mathfrak{q},\,\, q=q_{3},\,\, t=q_{34}}
\eea
is rather difficult to show generally but one can implement everything in a computer program and explicitly check that they indeed correspond with each other.

\end{itemize}

\subsection{Spiked instanton}
The contour integral formula of the spiked instanton system in Prop.~\ref{prop:spikedcontourJK} is rewritten in the structure functions as
\bea\label{eq:D4integralmod}
\mathcal{Z}^{\D4}_{\text{inst.}}=\sum_{k=0}^{\infty}\mathfrak{q}^{k}\mathcal{Z}_{k}^{\D4},\quad \mathcal{Z}^{\D4}_{k}=\frac{\mathcal{G}^{k}}{k!}\oint_{\text{JK}}\prod_{I=1}^{k}\frac{dx_{I}}{2\pi ix_{I}}\mathcal{Z}_{k}^{\D4\tbar\D0}(v_{A,\alpha},x_{I})
\eea
where
\bea
\mathcal{Z}_{k}^{\D4\tbar\D0}(v_{A,\alpha},x_{I})=\prod_{A\in\six}\prod_{\alpha=1}^{n_{A}}\prod_{I=1}^{k}\mathscr{S}_{\bar{A}}\left(\frac{v_{A,\alpha}}{x_{I}}\right)\prod_{I<J}\mathcal{A}_{\mathbb{C}^{4}}\left(\frac{x_{I}}{x_{J}}\right)^{-1}.
\eea
In terms of character, we have
\beq
\mathcal{Z}_{k}^{\D4}=\frac{1}{k!}\oint_{\text{JK}}\prod_{I=1}^{k}\frac{dx_{I}}{2\pi i x_{I}}\mathbb{I}[\mathbf{v}_{\text{inst.}}^{\D4}-k],\quad \mathbf{v}_{\text{inst.}}^{\D4}=-\sum_{A\in\six}\bfP_{\bar{A}}^{\vee}\bfN_{A}^{\vee}\bfK+\sqrt{\bfP_{\four}\bfK^{\vee}\bfK}
\eeq
where
\begin{equation}
\bfN_{A}=\sum_{\alpha=1}^{n_{A}}e^{\mathfrak{a}_{A,\alpha}}=\sum_{\alpha=1}^{n_{A}}v_{A,\alpha},\quad \bfK=\sum_{I=1}^{k}e^{\phi_{I}}=\sum_{I=1}^{k}x_{I}
\end{equation}
and we schematically wrote the square root part.

Using Thm.~\ref{thm:spikedJKpoles}, the poles are classified by 2d partitions as
\bea
    &\vec{\underline{v}}=(v_{A,\alpha})^{\alpha=1,\ldots, n_{A}}_{A\in\six},\quad \vec{\underline{\lambda}}=(\vec{\lambda}_{A})_{A\in\six}=(\lambda_{A}^{(\alpha)})^{\alpha=1,\ldots,n_{A}}_{A\in\six},\quad |\vec{\underline{\lambda}}|=\sum_{A\in\six}\sum_{\alpha=1}^{n_{A}}|\lambda_{A}^{(\alpha)}|,\\
    &\{x_{I}\}_{I=1,\ldots,k}\longrightarrow \{\chi_{A,v_{A,\alpha}}(\Bbox)\}^{\alpha=1,\ldots,n_{A}}_{A\in\six,\,\Abox\in\lambda_{A,\alpha}},\quad \chi_{ab,v}(\Bbox)=vq_{a}^{i-1}q_{b}^{j-1}
\eea
and then the character $\bfK_{A}$ will be 
\bea
\bfK|_{\underline{\vec{\lambda}}}=\sum_{A\in\six}\left.\bfK_{A}\right|_{\vec{\lambda}_{A}},\quad \left.\bfK_{A}\right|_{\vec{\lambda}_{A}}=\sum_{\alpha=1}^{n_{A}}\sum_{\Abox\in\lambda_{A}^{(\alpha)}}\chi_{A,v_{A,\alpha}}(\Bbox),\quad A\in\six.
\eea
A choice for the square root part is
\beq
\sqrt{\bfP_{\four}\bfK^{\vee}\bfK}=\sum_{A\in\six}\bfP^{\vee}_{\text{inf}(\bar{A})}\bfP_{A}^{\vee}\bfK_{A}^{\vee}\bfK_{A}+\sum_{A<B}\bfP_{\four}\bfK_{A}^{\vee}\bfK_{B}
\eeq
and this choice does not have any sign. We define the Nekrasov factors as
\bea
\mathsf{N}_{A}(v_{1},\lambda^{(1)}\,|\,v_{2},\lambda^{(2)})=\prod_{\Abox\in\lambda^{(1)}}\left(1-\frac{q_{A}\chi_{A,v_{1}}(\Bbox)}{v_{2}}\right)\prod_{\Abox\in\lambda^{(2)}}\left(1-\frac{v_{1}}{\chi_{A,v_{2}}(\Bbox)}\right)\prod_{\substack{\Abox\in\lambda^{(1)}\\\AboxF\in\lambda^{(2)}}}\mathscr{S}_{A}\left(\frac{\chi_{A,v_{1}}(\Bbox)}{\chi_{A,v_{2}}(\BboxF)}\right).
\eea
The relation with the JK residue is then given as follows. 
\begin{proposition}
    The JK residue of the D4--D0 integrand is related to the index of the character as
    \bea
    \mathcal{G}^{k}\times \underset{x=x_{\underline{\vec{\lambda}}}}{\Res}x^{-1}\mathcal{Z}_{k}^{\D4\tbar\D0}(v_{A,\alpha},x_{I})=\mathcal{Z}^{\D4}_{\text{spk.inst.}}[\vec{\underline{v}},\vec{\underline{\lambda}}],
    \eea    
    where 
    \bea\label{eq:D4spikedpartition1}
    \mathcal{Z}^{\D4}_{\text{spk.inst.}}[\underline{\vec{v}},\vec{\underline{\lambda}}]&=\prod_{A\in\six}\prod_{\alpha=1}^{n_{A}}\widetilde{\mathcal{Z}}^{\D4}_{A}[\lambda_{A}^{(\alpha)}]\prod_{A\in\six}\prod_{\alpha<\beta}\mathcal{Z}^{\D4\tbar\D4}_{A|A}(v_{A,\alpha},\lambda_{A}^{(\alpha)}\,|\,v_{A,\beta},\lambda_{A}^{(\beta)}) \\
    &\qquad \times \prod_{A<B}\prod_{\alpha=1}^{n_{A}}\prod_{\beta=1}^{n_{B}}\mathcal{Z}_{A|B}^{\D4\tbar\D4}(v_{A,\alpha},\lambda_{A}^{(\alpha)}\,|\,v_{B,\beta},\lambda_{B}^{(\beta)}),
\eea
and
\bea\label{eq:D4spikedpartition2}
    \widetilde{\mathcal{Z}}^{\D4}_{A}[\lambda]&=q_{\text{inf}(\bar{A})}^{-|\lambda|}\frac{\mathsf{N}_{A}(q_{\text{inf}(\bar{A})}v,\lambda\,|\,v,\lambda)}{\mathsf{N}_{A}(v,\lambda\,|\,v,\lambda)},\\
    \mathcal{Z}_{A|B}^{\D4\tbar\D4}(v_{1},\lambda^{(1)}\,|\,v_{2},\lambda^{(2)})&=\prod_{\Abox\in\lambda^{(1)}}\mathscr{S}_{\bar{B}}\left(q_{B}\frac{\chi_{A,v_{1}}(\Bbox)}{v_{2}}\right)\prod_{\AboxF\in\lambda^{(2)}}\mathscr{S}_{\bar{A}}\left(\frac{v_{1}}{\chi_{B,v_{2}}(\BboxF)}\right)\prod_{\substack{\Abox\in\lambda^{(1)}\\\AboxF\in\lambda^{(2)}}}\mathcal{A}_{\mathbb{C}^{4}}\left(\frac{\chi_{A,v_{1}}(\Bbox)}{\chi_{B,v_{2}}(\BboxF)}\right)^{-1}.
\eea

\end{proposition}
The proof of this is similar to the pure SYM case. Since the adjoint mass contribution does not modify the structure of the poles, the iterative residue does not change, and the discussion in \eqref{eq:pureSYM-mult-residue} is applicable to here too. 

\begin{theorem}
The total instanton partition function of the spiked instanton system is given as
\bea
\mathcal{Z}_{\text{inst.}}^{\D4}=\sum_{\vec{\underline{\lambda}}}\mathfrak{q}^{|\vec{\underline{\lambda}}|}\mathcal{Z}^{\D4}_{\text{spk.inst.}}[\vec{\underline{v}},\vec{\underline{\lambda}}].
\eea
\end{theorem}

\paragraph{One-loop perturbative part}
The one-loop perturbative part for the spiked instanton setup is given by the character
\bea
    \mathring{\mathbf{v}}=\sum_{(B,\beta)>(A,\alpha)}\frac{\bfP_{\bar{A}}^{\vee}\bfP_{\bar{B}}}{\bfP_{\four}}v_{B,\beta}/v_{A,\alpha},
\eea
where we specified an order in $(A,\alpha)^{\alpha=1\ldots,n_{A}}_{A\in\six}$. Taking the index, we have 
\bea\label{eq:D4oneloop}
&\mathbb{I}[\mathring{\mathbf{v}}]=\prod_{(B,\beta)>(A,\alpha)}\mathcal{Z}^{\D4\tbar\D4}_{\text{1-loop}}(v_{A,\alpha},A\,|\,v_{B,\beta},B)\eqqcolon\mathcal{Z}^{\D4}_{\text{1-loop}},\\
&\mathcal{Z}_{\text{1-loop}}^{\D4\tbar\D4}(x_{1},A\,|\,x_{2},B)=\exp\left(-\sum_{n=1}^{\infty}\frac{1}{n}\frac{\bfP_{\bar{A}}^{[n]}\bfP_{\bar{B}}^{[-n]}}{\bfP_{\four}^{[n]}}\left(\frac{x_{1}}{x_{2}}\right)^{n}\right).
\eea

\paragraph{PE formula and limits}Let us consider the $\U(1)$ affine quiver gauge theory on $\mathbb{C}^{2}_{12}\times \mathbb{S}^{1}$. The PE formula for it is
\bea
\mathcal{Z}^{\D4}_{12}[\mathfrak{q}]=\sum_{\lambda}\mathfrak{q}^{|\lambda|}\widetilde{\mathcal{Z}}_{12}^{\D4}[\lambda]=\PE \left[\frac{\mathfrak{q}}{1-\mathfrak{q}}\mathscr{S}_{34}(q_{2})\right]=\PE \left[\frac{\mathfrak{q}}{1-\mathfrak{q}}\mathscr{S}_{34}(q_{1})\right].
\eea
We have two special limits.
\begin{itemize}[topsep=1.5pt,itemsep=0.5ex,partopsep=1ex,parsep=1ex] 
\item Vafa--Witten limit \cite{Vafa:1994tf}: The limit is $q_{4}\rightarrow 1$ with $q_{123}=1$:
\bea
\frac{\mathfrak{q}}{1-\mathfrak{q}}\mathscr{S}_{34}(q_{1})=\frac{\mathfrak{q}}{1-\mathfrak{q}}q_{4}^{-1}\frac{(1-q_{14})(1-q_{24})}{(1-q_{1})(1-q_{2})}\xrightarrow{q_{4}\rightarrow 1}\frac{\mathfrak{q}}{1-\mathfrak{q}}.
\eea
In this limit, we have $\widetilde{\mathcal{Z}}^{\D4}_{12}[\lambda]\rightarrow 1$ and it coincides with the character of the Young diagrams
\bea
\sum_{\lambda}\mathfrak{q}^{|\lambda|}=\PE\left[\frac{\mathfrak{q}}{1-\mathfrak{q}}\right].
\eea

\item Pure SYM limit: We first rescale $\mathfrak{q}\rightarrow q_{3}\mathfrak{q}$ and take the limit $q_{3},q_{4}\rightarrow 0,\infty$ with $q_{34}$ fixed:
\bea
(q_{3},q_{4})\mapsto (z^{r}q_{3},z^{-r}q_{4}),\quad z\rightarrow 0,\,\,r>0.
\eea
Under this limit, we have
\bea
q_{3}^{|\lambda|}\widetilde{\mathcal{Z}}^{\D4}_{12}[\lambda]\rightarrow \frac{1}{\mathsf{N}_{12}(v,\lambda\,|\,v,\lambda)}.
\eea
The PE formula transforms as
\bea
\frac{\mathfrak{q}q_{3}}{1-\mathfrak{q}q_{3}}\mathscr{S}_{34}(q_{1})\rightarrow \frac{\mathfrak{q}}{(1-q_{1})(1-q_{2})}
\eea
and we have
\bea
\sum_{\lambda}\mathfrak{q}^{|\lambda|}\frac{1}{\mathsf{N}_{12}(v,\lambda\,|\,v,\lambda)}=\PE\left[\frac{\mathfrak{q}}{(1-q_{1})(1-q_{2})}\right]
\eea

\end{itemize}

\chapter{Free field realizations and vertex operators}
\label{chap:freefield-vertexop}


In this chapter, we introduce vertex operators whose vacuum expectation value reproduces the contour integral formula of the Witten index of a low energy field theory. In section~\ref{sec:vertexop-quiver}, we introduce the concept of $q$-deformed quiver Cartan matrices (shortly qquiver Cartan matrix) and give a definition of them. Using this qquiver Cartan matrix, we introduce vertex operators using free bosons and derive the free field realization of the Witten index. In section~\ref{sec:freefieldintegral}, we apply this formalism to each gauge origami setup and show that they indeed reproduce the contour integral formula after taking the vacuum expectation value. These free field realizations imply the existence of an underlying quantum algebraic structure which will be discussed in Chap.~\ref{chap:quantum-algebra-BPSqq}. This chapter is mainly based on \cite{Kimura:2023bxy}.

\section{Vertex operators and quivers}\label{sec:vertexop-quiver}

As discussed in section~\ref{sec:susyqm-index}, the Witten index of an $\mathcal{N}=2$ supersymmetric quiver quantum mechanical system \eqref{eq:2susylocalization} is written as a contour integral over rational functions coming from the vector, chiral, and Fermi superfield contributions. Given a quiver $\overline{Q}=(\overline{Q}_{0},\overline{Q}_{1})$, we define a \textbf{$q$-deformed Cartan matrix} encoding the flavor charges of the component fields and show that the contour integral formula has a free field realization in vertex operators. The basic components are rewritten in the multiplicative variables as follows.
\begin{itemize}[topsep=0pt, partopsep=0pt, itemsep=0pt]
    \item For each gauge group (circle node) of the theory, we have 
    \bea\label{eq:index-vector-mod}
\adjustbox{valign=c}{\begin{tikzpicture}[decoration={markings,mark=at position \arrowHeadPosition with {\arrow{latex}}}]
 \tikzset{
        box/.style={draw, minimum width=0.6cm, minimum height=0.6cm, text centered,thick},
        ->-/.style={decoration={
        markings,mark=at position #1 with {\arrow[scale=1.5]{>}}},postaction={decorate},line width=0.5mm},
        -<-/.style={decoration={
        markings,
        mark=at position #1 with {\arrow[scale=1.5]{<}}},postaction={decorate},line width=0.5mm}    
    }
\begin{scope}{xshift=0cm}
    \draw[fill=black!10!white,thick](0,0) circle(0.4cm);
    \node at (0,0){$N_{a}$};
\end{scope}
\end{tikzpicture}}\quad {\large \rightsquigarrow} \quad \prod_{I=1}^{N_{a}}\frac{dx_{I}^{(a)}}{2\pi i x_{I}^{(a)}} \prod_{I\neq J}\left(\frac{x_{I}^{(a)}}{x_{J}^{(a)}}\right)^{1/2}\left(1-\frac{x_{J}^{(a)}}{x_{I}^{(a)}}\right)
    \eea
\item The chiral superfield corresponding to an arrow from $\U(N_{a})$ to $\U(N_{b})$ gives the contribution
    \bea\label{eq:index-chiral-mod}
   \adjustbox{valign=c}{ \begin{tikzpicture}[decoration={markings,mark=at position \arrowHeadPosition with {\arrow{latex}}}]
 \tikzset{
        box/.style={draw, minimum width=0.6cm, minimum height=0.6cm, text centered,thick},
        ->-/.style={decoration={
        markings,mark=at position #1 with {\arrow[scale=1.5]{>}}},postaction={decorate},line width=0.5mm},
        -<-/.style={decoration={
        markings,
        mark=at position #1 with {\arrow[scale=1.5]{<}}},postaction={decorate},line width=0.5mm}    
    }

\begin{scope}{}
\draw[fill=black!10!white,thick](4.7,-2) circle(0.4cm);
\node at (4.7,-2){$N_{b}$};
\draw[fill=black!10!white,thick](2.1,-2) circle(0.4cm);
\node at (2.1,-2){$N_{a}$};
\draw[postaction={decorate}, thick](2.5,-2)--(4.3,-2);
\node[above] at (3.4,-2){$q(\Phi_{a\rightarrow b})$};
\end{scope}
\end{tikzpicture}}\quad \rightsquigarrow\quad \prod_{I=1}^{N_{a}}\prod_{J=1}^{N_{b}}\left(\frac{q(\Phi_{a\rightarrow b})x_{I}^{(a)}}{x_{J}^{(b)}}\right)^{1/2}\frac{1}{\left(1-q(\Phi_{a\rightarrow b})x_{I}^{(a)}/x_{J}^{(b)}\right)}
    \eea
    \item The Fermi superfield corresponding to an arrow from $\U(N_{a})$ to $\U(N_{b})$ gives the contribution 
    \bea\label{eq:index-Fermi-mod}
   \adjustbox{valign=c}{ \begin{tikzpicture}[decoration={markings,mark=at position \arrowHeadPosition with {\arrow{latex}}}]
 \tikzset{
        box/.style={draw, minimum width=0.6cm, minimum height=0.6cm, text centered,thick},
        ->-/.style={decoration={
        markings,mark=at position #1 with {\arrow[scale=1.5]{>}}},postaction={decorate},line width=0.5mm},
        -<-/.style={decoration={
        markings,
        mark=at position #1 with {\arrow[scale=1.5]{<}}},postaction={decorate},line width=0.5mm}    
    }

\begin{scope}{}
\draw[fill=black!10!white,thick](4.7,-2) circle(0.4cm);
\node at (4.7,-2){$N_{b}$};
\draw[fill=black!10!white,thick](2.1,-2) circle(0.4cm);
\node at (2.1,-2){$N_{a}$};
\draw[red, postaction={decorate}, thick](2.5,-2)--(4.3,-2);
\node[above] at (3.4,-2){\textcolor{red}{$q(\Lambda_{a\rightarrow b})$}};
\end{scope}
\end{tikzpicture}}\quad \rightsquigarrow\quad \prod_{I=1}^{N_{a}}\prod_{J=1}^{N_{b}}\left(\frac{x_{J}^{(b)}}{x_{I}^{(a)}q(\Lambda_{a\textcolor{red}{\rightarrow} b})}\right)^{1/2}\left(1-\frac{q(\Lambda_{a\textcolor{red}{\rightarrow} b})x_{I}^{(a)}}{x_{J}^{(b)}}\right).
    \eea
\end{itemize}
For explicit construction of the free field realization, let us introduce the \textbf{half $q$-deformed quiver Cartan matrix}.\footnote{In the context of quiver W-algebras, a similar $q$-Cartan matrix appears (see for example \cite[Sec.2.2.3]{Kimura:2020jxl}). The quiver considered there is a 4 SUSY quiver while the one here is a 2 SUSY quiver.} We shortly call this the half qquiver\footnote{$q$-deformed + quiver} Cartan matrix.

\begin{definition}\label{def:halfqquiverCartan}
    Let $\overline{Q}=(\overline{Q}_{0},\overline{Q}_{1})$ be a 2 SUSY quiver. We assume that the flavor charges are associated with the chiral and Fermi superfields as mentioned in section~\ref{sec:2susy2d-flavorsymmetry}. Namely, for each chiral superfield $\Phi_{a\rightarrow b}$ and Fermi superfield $\Lambda_{a\textcolor{red}{\rightarrow} b}$, we associate the flavor charges $q(\Phi_{a\rightarrow b})$ and $q(\Lambda_{a\textcolor{red}{\rightarrow} b})$ respectively. The (positive) \textbf{half qquiver Cartan matrix} $c^{+}=(c^{+}_{ab})_{a,b\in\overline{Q}_{0}}$ is defined as follows. For circle nodes $a,b\in\overline{Q}_{0}$, it is defined as
    \bea
    c_{ab}^{+}=\delta_{ab}-\sum_{b\rightarrow a}q(\Phi_{b\rightarrow a})+\sum_{b\textcolor{red}{\rightarrow} a}q(\Lambda_{b\textcolor{red}{\rightarrow} a}),
    \eea
    where the sum is taken over all arrows connecting the two quiver nodes.
    
\end{definition}
The first term is understood as the contribution from the vector superfield. Note that when one of the nodes is a square node, the first term of the half qquiver Cartan matrix disappears. We omit the discussion when both nodes are square nodes. Such contributions are decoupled from the supersymmetric quantum mechanical system at the low energy limit.\footnote{In instanton contributions, they appear as one--loop perturbative contributions.}

\begin{definition}\label{def:qquiver-conjugate}
Given the (positive) half qquiver Cartan matrix $c^{+}=(c^{+}_{ab})_{a,b\in\overline{Q}_{0}}$, we define the \textit{Hermitian conjugate} (conjugate transpose)\footnote{Since we are taking the transpose and the dual operation which resembles the Hermitian conjugate operation $\dagger$, we abuse the terminology and call it in such way. The symbol $\dagger$ also comes from this motivation. } operation of it as
\bea
(c^{+\dagger})_{ab}=c_{ba}^{+\vee} 
\eea
where the dual operation $\vee$ simply changes the flavor charges to the inverse (which follows the convention in \eqref{eq:dualcharacter}). We sometimes call this the negative half qquiver Cartan matrix and denote it as $c^{-}_{ab}=c_{ba}^{+\vee}$.
\end{definition}

\begin{definition}
    The degree $n\in\mathbb{Z}$ half qquiver Cartan matrix is defined by taking the $n$-th Adams operation:
    \bea
    c_{ab}^{+[n]}=\delta_{ab}-\sum_{b\rightarrow a}q(\Phi_{b\rightarrow a})^{n}+\sum_{b\rightarrow a}q(\Lambda_{b\textcolor{red}{\rightarrow} a})^{n}
    \eea
    where the convention $[n]$ comes from \eqref{eq:pdegree-op}.
    
\end{definition}

Given the above ingredients, we define the total qquiver Cartan matrix.
\begin{definition}\label{def:totalqquiverCartan}
    The \textbf{total qquiver Cartan matrix} $c=(c_{ab})_{a,b\in\overline{Q}_{0}}$ is defined as the sum of positive and negative half qquiver Cartan matrices:
    \bea
    c_{ab}&=c_{ab}^{+}+c_{ab}^{-}=2\delta_{ab}-\left(\sum_{b\rightarrow a}q(\Phi_{b\rightarrow a})+\sum_{a\rightarrow b}q(\Phi_{a\rightarrow b})^{-1} \right)\\
    &\qquad +\left(\sum_{b\textcolor{red}{\rightarrow} a}q(\Lambda_{b\textcolor{red}{\rightarrow} a})+\sum_{a\textcolor{red}{\rightarrow} b}q(\Lambda_{a\textcolor{red}{\rightarrow} b})^{-1}\right).
    \eea
    We sometimes omit the \textit{total} and simply call it the qquiver Cartan matrix.
\end{definition}

An important property of this total qquiver Cartan matrix is that it is self-adjoint (\textit{Hermitian}):
\bea\label{eq:qquiverCartan-Hermite}
c_{ab}^{\dagger}=c_{ab}.
\eea
The $n$-th Adams operation is induced similarly as discussed above. The strategy to find free field realizations is then given as follows.
\begin{enumerate}
    \item For each quiver node $a\in\overline{Q}_{0}$, we associate a vertex operator assuming the form
    \bea
    \mathsf{V}_{a}(x)=\mathsf{v}_{a,0}(x)\exp\left(\sum_{\substack{n\neq 0\\n\in\mathbb{Z}}}\mathsf{v_{a,n}}x^{-n}\right),
    \eea
where $\mathsf{v}_{a,0}(x)$ is the zero-mode of the vertex operator.
    \item We impose the commutation relations of the non-zero modes $\{\mathsf{v}_{a,n}\}$ using the total qquiver Cartan matrix:
    \bea\label{eq:non-zero-OPE-qquiver}
    [\mathsf{v}_{a,n},\mathsf{v}_{b,m}]=-\frac{1}{n}\delta_{n+m,0}c_{ab}^{[n]}.
    \eea
    In this thesis, we do not give a general description on how to determine the qquiver Cartan matrix when both the nodes $a,b\in\overline{Q}_{0}$ are square nodes (flavor nodes). 
    \item We further need to impose the commutation relations on the zero-modes $\mathsf{v}_{a,0}(x)$ but in this thesis we do not give a systematic procedure to define it. Examples will be given in section~\ref{sec:zeromodes}.
    
    \item Suppose that the rank of the quiver node $a\in\overline{Q}_{0}$ is $N_{a}$. The free field realization of the Witten index is then given as
    \bea\label{eq:Wittenindex-freefield-general}
    \oint_{\text{JK}}\prod_{a\in\overline{Q}_{0}}\prod_{I=1}^{N_{a}}\frac{dx_{I}^{(a)}}{2\pi i x_{I}^{(a)}} \prod_{a\in\overline{Q}_{0}}\prod_{I=1}^{N_{a}}\mathsf{V}_{a}(x_{I}^{(a)}).
    \eea
   We assumed here for simplicity that all the quiver nodes are circle nodes. If some are square nodes, one will simply discard the contour integral over the flavor fugacities.
\end{enumerate}
Let us give some motivation of the above procedure. The operator product of the vertex operators $\mathsf{V}_{a}(x_{I})$ and $\mathsf{V}_{b}(x_{J})$ is given by 
\bea
\mathsf{V}_{a}(x_{I})\mathsf{V}_{b}(x_{J})&=\wick{\c{\mathsf{v}_{a,0}(x_{I})}\c{\mathsf{v}_{b,0}(x_{J})}}\exp\left(\sum_{n>0}-\frac{1}{n}c_{ab}^{[n]}\left(\frac{x_{J}}{x_{I}}\right)^{n}\right):\mathsf{V}_{a}(x_{I})\mathsf{V}_{b}(x_{J}):\\
&=\wick{\c{\mathsf{v}_{a,0}(x_{I})}\c{\mathsf{v}_{b,0}(x_{J})}}\frac{\left(1-x_{J}/x_{I}\right)^{2\delta_{ab}} \prod\limits_{b\textcolor{red}{\rightarrow} a} \left(1-q(\Lambda_{b\textcolor{red}{\rightarrow }a})x_{J}/x_{I} \right) \prod\limits_{a\textcolor{red}{\rightarrow}b} \left(1-q(\Lambda_{a\textcolor{red}{\rightarrow}b})^{-1}x_{J}/x_{I}\right)  }{\prod\limits_{b\rightarrow a}\left(1-q(\Phi_{b\rightarrow a})x_{J}/x_{I}\right)\prod\limits_{a\rightarrow b}\left(1-q(\Phi_{a\rightarrow b})^{-1}x_{J}/x_{I}\right)}\\
&\qquad \times :\mathsf{V}_{a}(x_{I})\mathsf{V}_{b}(x_{J}):,
\eea
where the operator product of the zero-modes are denoted as
\bea\label{eq:zeromode-Wick}
\mathsf{v}_{a,0}(x_{I})\mathsf{v}_{b,0}(x_{J})=\wick{\c{\mathsf{v}_{a,0}(x_{I})}\c{\mathsf{v}_{b,0}(x_{J})}}:\mathsf{v}_{a,0}(x_{I})\mathsf{v}_{b,0}(x_{J}):.
\eea
After proper definitions of the zero-modes, we can reproduce the contributions \eqref{eq:index-vector-mod}, \eqref{eq:index-chiral-mod}, and \eqref{eq:index-Fermi-mod}. We note that when chiral and Fermi superfields transforming under the adjoint representation of a gauge group exist in the theory considered, overall factors not depending on the fugacities $x_{I}^{(a)}$ appear. Such factors are not reproduced from the operator product of the vertex operators and we need to add them by hand.

Furthermore, when considering operator products, we need to specify the order of the operators because generally they might not commute with each other. In this thesis, we simply assume that we can introduce proper zero-modes so that the rational functions arising in the right hand side are the same after analytic continuation:
\bea
\mathsf{V}_{a}(x_{I})\mathsf{V}_{b}(x_{J})=\mathsf{V}_{b}(x_{J})\mathsf{V}_{a}(x_{I}).
\eea
Whether we can impose this condition and find an explicit expression generally is nontrivial and it will not be discussed in this thesis.

When understanding the operator product in \eqref{eq:Wittenindex-freefield-general}, we understand the products for vertex operators associated with flavor nodes using normal products:
\bea
:\prod_{a}\prod_{I=1}^{N_{a}}\mathsf{V}_{a}(v_{I}^{(a)}):
\eea
where $v_{I}^{(a)}$ are the flavor fugacities. In this way, we do not need the definition of the qquiver Cartan matrix associated with two flavor nodes. Physically, we are simply discarding the contributions which are not needed when evaluating the contour integral. In instanton counting, they will correspond to the one-loop perturbative factors.

\section{Free field realizations of gauge origami}\label{sec:freefieldintegral}
In Chap.~\ref{chap:gaugeorigamipartitionfunction}, we derived the contour integral formulas and non-perturbative partition functions of the magnificent four \eqref{eq:D8integralmod}, tetrahedron instantons \eqref{eq:D6integralmod}, spiked instantons \eqref{eq:D4integralmod}. All of the formulas are derived from an $\mathcal{N}=2$ quiver. In this section, we will explicitly find the free field realizations of the contour integral formula for each setup by following the general discussion of section~\ref{sec:vertexop-quiver}.



The main statements of this section are summarized as follows.
\begin{theorem}\label{thm:freefieldconclusion}
    For each D-brane (D0, D2, D4, D6, D8), we can define the corresponding vertex operators as
    \begin{align}
    \renewcommand\arraystretch{1.2}{
        \begin{tabular}{|c|c|c|c|}\hline
            D-brane & space-time & vertex operator & reference\\
           \hline\hline  D0-brane  & $\mathbb{S}^{1}$ & $\mathsf{A}(x)$ & \eqref{eq:D0op}, \eqref{eq:D0op2}\\
           \hline D2-brane   &  $\mathbb{C}_{a}\times \mathbb{S}^{1}$ ($a\in\four$)  & $\mathsf{S}_{a}(x)$&\eqref{eq:D2op}\\
           \hline D4-brane &  $\mathbb{C}^{2}_{A}\times \mathbb{S}^{1}$ ($A\in\six$) & $\mathsf{X}_{A}(x)$&\eqref{eq:D4op}, \eqref{eq:D4op2}\\
           \hline D6-brane &  $\mathbb{C}^{3}_{\bar{a}}\times \mathbb{S}^{1}$ ($a\in\four$) &  $\mathsf{W}_{\bar{a}}(x)$& \eqref{eq:D6op}, \eqref{eq:D6op2}\\
           \hline D8-brane & $\mathbb{C}^{4}\times \mathbb{S}^{1}$  & $\mathsf{Z}(x)$& \eqref{eq:D8op}, \eqref{eq:D8op2} \\\hline
         \end{tabular}}
    \end{align}
    In the contour integral formula, the D0-branes giving instanton contributions arise from\footnotemark $\mathsf{A}^{-1}(x)$, while other D-branes arise from $\mathsf{S}_{a}(x),\mathsf{X}_{A}(x),\mathsf{W}_{\bar{a}}(x),\mathsf{Z}(x)$. To include anti D-branes/negative D-branes, we need to reverse the power of the operators as $\mathsf{S}_{a}(x)^{-1},\mathsf{X}_{A}(x)^{-1},\mathsf{W}_{\bar{a}}(x)^{-1},\mathsf{Z}(x)^{-1}$. 
\end{theorem}\footnotetext{In this paper, when we write $\mathsf{V}(x)^{-1}$ for a vertex operator $\mathsf{V}(x)$, we are meaning $:\mathsf{V}(x)^{-1}:$. The normal ordering of them is implicitly imposed.}

\begin{theorem}
The contour integral formula of the $k$-instanton sector of the partition function of the gauge origami system takes the form as
\bea
    \mathcal{Z}_{k} = \frac{\mathcal{G}^{k}}{k!}\oint_{\text{JK}} \prod_{I=1}^{k}\frac{dx_{I}}{2\pi\iota x_{I}} \bra{0} \prod_{I=1}^{k}\mathsf{A}(x_{I})^{-1} :\prod_{i}\mathsf{V}_{i}(v_{i}): \ket{0},
\eea
where $\mathsf{V}_{i}(x)$ is an operator written from $\{\mathsf{X}_{A}(x),\mathsf{W}_{\bar{a}}(x),\mathsf{Z}(x)\}$.
\end{theorem}

\subsection{Magnificent four}\label{sec:M4LMNS}
Let us start with the magnificent four system. The vertex operator of this system was already introduced in \cite{Kimura:2022zsm} in the context of double quiver gauge theory (see also \cite{Kimura:2019hnw}). In this section, we derive the vertex operators that reproduce the contour integral formula \eqref{eq:D8integralmod} by performing the procedure discussed in section~\ref{sec:vertexop-quiver} in a more explicit way. We then discuss how to directly derive the vertex operators using the quiver structure \eqref{eq:flavorquiver-magnificent} and the qquiver Cartan matrix in Def.~\ref{def:totalqquiverCartan}.

The goal of this section is to reproduce the rational function $\mathcal{Z}^{\D8\tbar\D0}(v_{\alpha},\bar{v}_{\alpha},x_{I})$ given in \eqref{eq:D8integralmod}. One can observe that it factorizes in a nice way that rational functions only appear and the fundamental objects are
\bea
\frac{\left(1-\bar{v}_{\alpha}/x_{I}\right)}{\left(1-v_{\alpha}/x_{I}\right)},\qquad \mathcal{A}_{\mathbb{C}^{4}}\left(\frac{x_{I}}{x_{J}}\right)^{-1}.
\eea
Let us first start from the $\mathcal{A}_{\mathbb{C}^{4}}(x)^{-1}$ part. Let $\mathsf{A}(x)$ be some operator whose operator product is given as
\begin{equation}
    \mathsf{A}(x)\mathsf{A}(x')=\mathcal{A}_{\mathbb{C}^{4}}\left(\frac{x'}{x}\right)^{-1}:\mathsf{A}(x)\mathsf{A}(x'):,\label{eq:D0ope}
\end{equation}
Under this ansatz, indeed the $\mathcal{A}_{\mathbb{C}^{4}}(x)^{-1}$ part is reproduced as follows:
\bea
\mathsf{A}(x_{k})^{-1}\cdots \mathsf{A}(x_{1})^{-1}=\prod_{I<J}\mathcal{A}_{\mathbb{C}^{4}}\left(\frac{x_{I}}{x_{J}}\right)^{-1}{:\prod_{I=1}^{k}\mathsf{A}(x_{I})^{-1}:}.
\eea
Note here the inverse of the operator is understood as $\mathsf{A}^{-1}(x)={:\mathsf{A}^{-1}(x):}$. Taking the inverse here is just for convention so that the result matches with results in the literature.

We then assume a free field realization which means that the operator $\mathsf{A}(x)$ is a vertex operator written using free bosons:
\bea\label{eq:D0op}
 \mathsf{A}(x)=\mathsf{a}_{0}(x):\exp\left(\sum_{n\neq 0}\mathsf{a}_{n}x^{-n}\right):,
\eea
where $\mathsf{a}_{0}(x)$ is some zero-mode and $\mathsf{a}_{n}$ are the free bosons. So that we can reproduce \eqref{eq:D0ope}, we need the commutation relation
\bea\label{eq:D0op2}
[\mathsf{a}_{n},\mathsf{a}_{m}]=-\frac{1}{n}\bfP_{\four}^{[n]}\delta_{n+m,0},
\eea
and
\bea
\wick{\c{\mathsf{a}_{0}(x)}\c{\mathsf{a}_{0}(x')}}=1
\eea
where we used the convention \eqref{eq:zeromode-Wick}. For the zero modes $\mathsf{a}_{0}(z)$, we further impose the condition that the OPE factor will be the same rational function. Namely, using \eqref{eq:reflec_structfunc}, we impose 
\begin{equation}
    \wick{\c{\mathsf{a}_{0}(x)}\c{\mathsf{a}_{0}(x')}}=\wick{\c{\mathsf{a}_{0}(x')}\c{\mathsf{a}_{0}(x)}}\label{eq:D0D0zero}
\end{equation}
and then we have the OPE factor symmetric in $x$ and $x'$. The derivation and explicit form of the zero modes are discussed in section~\ref{sec:zeromodes}.

From the supersymmetric quiver quantum mechanics view point, the parameters $\{x_{I}\}_{I=1}^{k}$ are the exponent of the Cartan subalgebra of $\U(k)$ which comes from the $k$ D0-branes in the gauge origami setup. In this sense, the vertex operator $\mathsf{A}(x)^{-1}$ introduced here is associated with the D0-brane (see Chap.~\ref{chap:gauge-origami}) and so we call it the \textbf{D0-brane vertex operator}. In the algebraic context, it is called the \emph{root current} because it is associated with the roots of the quiver \cite{Frenkel:1997CMP,Frenkel:1998ojj}. In the context of quiver W-algebra \cite{Kimura:2015rgi} and double quiver gauge theory \cite{Kimura:2022zsm}, this is just the root current of the affine quiver W-algebra which is denoted by the $(\widehat{A}_{0},\widehat{A}_{0})$ theory in the terminology of \cite{Kimura:2022zsm}.

Let us move on to reproducing the rational function part including $v_{\alpha},\bar{v}_{\alpha}$. To do this, we introduce a vertex operator $\mathsf{Z}(x)$
\begin{equation}
    \mathsf{Z}(x)=\mathsf{z}_{0}(x):\exp\left(\sum_{n\neq 0}\mathsf{z}_{n}x^{-n}\right):\label{eq:D8op}
\end{equation}
with the OPE relation
\bea
    \mathsf{Z}(x)\mathsf{A}(x')&=\wick{\c{\mathsf{z}_{0}(x)}\c{\mathsf{a}_{0}(x')}}\left(1-x'/x\right):\mathsf{Z}(x)\mathsf{A}(x'):\\
    \mathsf{A}(x')\mathsf{Z}(x)&=\wick{\c{\mathsf{a}_{0}(x')}\c{\mathsf{z}_{0}(x)}}\left(1-x/x'\right):\mathsf{A}(x')\mathsf{Z}(x):.
\eea
The commutation relations are then determined as
\begin{equation}\label{eq:D8op2}
    [\mathsf{a}_{n},\mathsf{z}_{m}]=-\frac{1}{n}\delta_{n+m,0},\quad \mathsf{z}_{n}=\frac{\mathsf{a}_{n}}{\bfP_{\four}^{[n]}}
\end{equation}
which also gives
\bea\label{eq:D8op3}
[\mathsf{z}_{n},\mathsf{z}_{m}]=-\frac{1}{n}\frac{1}{\bfP_{\four}^{[n]}}\delta_{n+m,0}.
\eea
Using this vertex operators, we then have 
\bea
\mathsf{A}^{-1}(x_{I}){:\prod_{\alpha=1}^{n}\frac{\mathsf{Z}(v_{\alpha})}{\mathsf{Z}(\bar{v}_{\alpha})}:}=\prod_{\alpha=1}^{n}\frac{\wick{\c{\mathsf{a}_{0}(x_{I})}\c{\mathsf{z}_{0}(\bar{v}_{\alpha})}}}{\wick{\c{\mathsf{a}_{0}(x_{I})}\c{\mathsf{z}_{0}(v_{\alpha})}}}\prod_{\alpha=1}^{n}\frac{(1-\bar{v}_{\alpha}/x_{I})}{(1-v_{\alpha}/x_{I})}{:\mathsf{A}^{-1}(x_{I})\prod_{\alpha=1}^{n}\frac{\mathsf{Z}(v_{\alpha})}{\mathsf{Z}(\bar{v}_{\alpha})}:}
\eea
which indeed reproduces the rational functions including parameters $\{v_{\alpha},\bar{v}_{\alpha}\}$. Similar to the previous discussion, since $\mathsf{Z}(x)^{\pm1}$ reproduces the contributions from the flavor nodes with fugacities $\{v_{\alpha},\bar{v}_{\alpha}\}$, we call it the \textbf{D8-brane vertex operator}.

Since the magnificent four system we are interested in includes the same number of D8 and $\overbar{\text{D8}}$-branes \eqref{eq:M4contourJK}, it is convenient to introduce a vertex operator representing the coupled system with one set of D8 and $\overbar{\text{D8}}$ branes as (see also \eqref{eq:antiD8parameter}):
\begin{equation}
    \mathsf{Z}(K,x)\coloneqq{:\frac{\mathsf{Z}(x)}{\mathsf{Z}(K x)}:}=\tilde{\mathsf{z}}^{K}_{0}(x):\exp\left(\sum_{n\neq 0}\tilde{\mathsf{z}}^{K}_{n}x^{-n}\right):,\quad \tilde{\mathsf{z}}^{K}_{n}=(1-K^{-n})\mathsf{z}_{n}=\frac{1-K^{-n}}{\bfP_{\four}^{[n]}}\mathsf{a}_{n},\label{eq:D8D8barop}
\end{equation}
where $K\in\mathbb{C}^{\times}$ is a generic parameter. This parameter is the parameter introduced in \eqref{eq:antiD8parameter} and corresponds to the distance between the D8 and $\overbar{\text{D8}}$ branes physically. This gives 
\bea
    \mathsf{Z}(K,x)\mathsf{A}(x')&=\wick{\c{\widetilde{\mathsf{z}}^{K}_{0}(x)}\c{\mathsf{a}_{0}(x')}}\frac{1-x'/x}{1-K^{-1}x'/x}:\mathsf{Z}(K,x)\mathsf{A}(x'):,\\
    \mathsf{A}(x')\mathsf{Z}(K,x)&=\wick{\c{\mathsf{a}_{0}(x')}\c{\widetilde{\mathsf{z}}^{K}_{0}(x)}}\frac{1-x/x'}{1-Kx/x'}:\mathsf{Z}(K,x)\mathsf{A}(x'):
\eea
For the zero modes, we impose that the contraction with the root current will be the same rational function after analytic continuation:
\begin{equation}
    \wick{\c{\tilde{\mathsf{z}}^{K}_{0}(x)}\c{\mathsf{a}_{0}(x')}}=K^{-1}\wick{\c{\mathsf{a}_{0}(x')}\c{\tilde{\mathsf{z}}^{K}_{0}(x)}}.\label{eq:D0D8zero}
\end{equation}
We impose $\wick{\c{\mathsf{a}_{0}(x')}\c{\tilde{\mathsf{z}}^{K}_{0}(x)}}=1$ (see section \ref{sec:zeromodes}). Note that we are relaxing the conditions and only imposing conditions on the zero modes of the brane anti-brane coupled vertex operator $\widetilde{\mathsf{Z}}^{K}$ but not the $\mathsf{Z}$-operator itself. 

\begin{proposition}
The contour integral formula of the magnificent four system shown in \eqref{eq:D8integralmod} has a free field realization as
\begin{equation}\label{eq:D8op_integral}
    \mathfrak{q}^{k}\mathcal{Z}_{k}^{\D8}=\frac{\mathfrak{q}^{k}\mathcal{G}^{k}}{k!}\oint_{\text{JK}} \prod_{I=1}^{k}\frac{dx_{I}}{2\pi i x_{I}}\langle \mathsf{A}_{k}^{-1}\widetilde{\mathsf{Z}}^{\underline{K}}_{\underline{n}}\rangle
\end{equation}
where we used
\begin{equation}
\begin{split}
    &\mathsf{A}_{k}^{-1}=\prod_{I=1}^{k}\mathsf{A}(x_{I})^{-1},\quad \widetilde{\mathsf{Z}}^{\underline{K}}_{\underline{n}}=:\prod_{\alpha=1}^{n}\mathsf{Z}(K_{\alpha},v_{\alpha}):,\quad \langle \mathsf{A}_{k}^{-1}\widetilde{\mathsf{Z}}^{\underline{K}}_{\underline{n}}\rangle=\prod_{\alpha=1}^{n}\prod_{I=1}^{k}\frac{1-K_{\alpha}v_{\alpha}/x_{I}}{1-v_{\alpha}/x_{I}}\prod_{I<J}\mathcal{A}_{\mathbb{C}^{4}}\left(\frac{x_{I}}{x_{J}}\right)^{-1}
\end{split}
\end{equation}
and $\langle \mathcal{O}\rangle=\bra{0}\mathcal{O}\ket{0}$.     
\end{proposition}
Note that we need to determine an order in the operators when doing explicit computations. In most of the cases, we simply assume $\prod_{I=1}^{k}\mathcal{O}(x_{I})=\mathcal{O}(x_{k})\cdots\mathcal{O}(x_{2}) \mathcal{O}(x_{1})$ for an operator $\mathcal{O}(x)$. Actually, because of the zero modes conditions \eqref{eq:D0D0zero} and \eqref{eq:D0D8zero}, the order of the operators is not relevant after analytic continuation.

\paragraph{Quiver interpretation}
Let us rederive the commutation relations \eqref{eq:D0op2}, \eqref{eq:D8op2} directly from the quiver diagram $\eqref{eq:flavorquiver-magnificent}$ by using the general description proposed in section~\ref{sec:vertexop-quiver}. Let us first focus on the $\mathsf{A}\mathsf{A}$ commutation relation \eqref{eq:D0op2} coming from the adjoint superfields of $\U(k)$. Given the following quiver, we can define the positive half qquiver Cartan matrix in Def.~\ref{def:halfqquiverCartan} as
\bea
\adjustbox{valign=c}{\begin{tikzpicture}[decoration={markings,mark=at position 0.7 with {\arrow{latex}}}]
 \tikzset{
        box/.style={draw, minimum width=0.7cm, minimum height=0.7cm, text centered,thick},
        ->-/.style={decoration={
        markings,mark=at position #1 with {\arrow[scale=1.5]{>}}},postaction={decorate},line width=0.5mm},
        -<-/.style={decoration={
        markings,
        mark=at position #1 with {\arrow[scale=1.5]{<}}},postaction={decorate},line width=0.5mm}    
    }
\begin{scope}[xshift=4cm]
    \draw[fill=black!10!white,thick](0,0) circle(0.4cm);
    \node at (0,0) {$k$};
       \foreach \ang in {90,145,215,270} {
    \begin{scope}[rotate=\ang]
        \chiralarc[postaction={decorate},thick](0,0.5)(-45:225:0.22:0.65)
    \end{scope}
    }
    \foreach \ang in {90,145,270} {
    \begin{scope}[rotate=\ang]
    \fermiarc[postaction={decorate},thick](0,0.5)(-45:225:0.1:0.5)
    \end{scope}
    \node[left] at (-1.5,0) {$q_{2},\textcolor{red}{q_{2}q_{4}}$};
    \node[right] at (1.6,0) {$q_{1},\textcolor{red}{q_{1}q_{4}}$};
    \node[below left] at (-0.9,-1){$q_{3},\textcolor{red}{q_{3}q_{4}}$};
    \node[below right] at (0.9,-1){$q_{4}$};
    \draw[fill=black!10!white,thick](0,0) circle(0.4cm);
    \node at (0,0) {$k$};
    }
\end{scope}
\end{tikzpicture}}\quad{\rightsquigarrow}\quad c^{+}_{\text{D0},\text{D0}}=1-(q_{1}+q_{2}+q_{3}+q_{4})+\textcolor{red}{(q_{14}+q_{24}+q_{34})},
\eea
where we represent the quiver node by D0 and only draw the subquiver related. The term $1$ corresponds to the vector superfield, $q_{1,2,3,4}$ correspond to the chiral superfield, and $q_{14,24,34}$ correspond to the Fermi superfield contribution, respectively.

We further define the negative half quiver Cartan matrix as
\bea
c_{\text{D0},\text{D0}}^{-}=c_{\text{D0},\text{D0}}^{+\vee}=1-(q_{1}^{-1}+q_{2}^{-1}+q_{3}^{-1}+q_{4}^{-1})+(q_{14}^{-1}+q_{24}^{-1}+q_{34}^{-1})
\eea
and the total quiver Cartan matrix is given as
\bea
c_{\text{D0},\text{D0}}=c_{\text{D0},\text{D0}}^{+}+c_{\text{D0},\text{D0}}^{-}=(1-q_{1})(1-q_{2})(1-q_{3})(1-q_{4}).
\eea
Using \eqref{eq:non-zero-OPE-qquiver}, we observe that
\bea
\relax[\mathsf{a}_{n},\mathsf{a}_{m}]=-\frac{1}{n}\delta_{n+m,0}c_{\text{D0},\text{D0}}^{[n]}.
\eea

Let us then move on to the $\mathsf{Z}\mathsf{A}$ commutation relation \eqref{eq:D8op2}. The half qquiver Cartan matrices are given as
\bea
\adjustbox{valign=c}{\begin{tikzpicture}[decoration={markings,mark=at position 0.7 with {\arrow{latex}}}]
 \tikzset{
        box/.style={draw, minimum width=0.7cm, minimum height=0.7cm, text centered,thick},
        ->-/.style={decoration={
        markings,mark=at position #1 with {\arrow[scale=1.5]{>}}},postaction={decorate},line width=0.5mm},
        -<-/.style={decoration={
        markings,
        mark=at position #1 with {\arrow[scale=1.5]{<}}},postaction={decorate},line width=0.5mm}    
    }
\begin{scope}[xshift=4cm]
    \draw[fill=black!10!white,thick](0,0) circle(0.4cm);
    \node at (0,0) {$k$};
    \node[box,fill=black!10!white] at (0,1.6) {$n$};
    \draw[postaction={decorate},thick] (0,1.25)--(0,0.4);
    \node[right] at (0,0.8) {$1$};
\end{scope}
\end{tikzpicture}}\quad \rightsquigarrow\quad \begin{dcases}c_{\text{D0},\text{D8}}^{+}=-1,\\c^{+}_{\text{D8},\text{D0}}=0.\end{dcases}
\eea
which gives the qquiver Cartan matrix
\bea
c_{\text{D0},\text{D8}}=c_{\text{D0},\text{D8}}^{+}+c^{+\vee}_{\text{D8},\text{D0}}=-1.
\eea
We represented the quiver node with rank $n$ as D8 because of its relation with the D8-brane. We then have
\bea
\relax[-\mathsf{a}_{n},\mathsf{z}_{m}]=-\frac{1}{n}c_{\text{D0},\text{D8}}^{[n]}\delta_{n+m,0}=\frac{1}{n}\delta_{n+m,0}
\eea
which indeed reproduces \eqref{eq:D8op3}. The extra sign $-\mathsf{a}_{n}$ appears here because we used $\mathsf{A}^{-1}$ in the contour integral to get the free field realization.

\paragraph{One--loop perturbative part}
Note also that the OPE of the $\mathsf{Z}(K,x)$ operators give the perturbative factor introduced in \eqref{eq:D8oneloop}:
\begin{equation}
    \mathsf{Z}(K_{n},v_{n})\cdots \mathsf{Z}(K_{1},v_{1})=\prod_{\beta>\alpha}\mathcal{Z}_{\text{1-loop}}^{\D8\tbar\D8}(v_{\alpha},K_{\alpha}\,|\,v_{\beta},K_{\beta}):\prod_{\alpha=1}^{n}\mathsf{Z}(K_{\alpha},v_{\alpha}):.
\end{equation}

\subsection{Tetrahedron instanton}\label{sec:tetraLMNS}
Let us find the free field realization of the contour integral formula given in \eqref{eq:D6integralmod}. The contributions coming from $\mathcal{A}_{\mathbb{C}^{4}}(x)^{-1}$ are the same with the D8-case and thus they are reproduced by the D0 vertex operator $\mathsf{A}(x)$. The different part is
\bea
\mathscr{V}_{a}\left(\frac{v_{\bar{a},\alpha}}{x_{I}}\right)
\eea
whose physical origin is the D6-branes in the setup. In this case, we have four types of D6-branes depending on which complex three-planes $\mathbb{C}^{3}_{\bar{a}}$ they wrap. We introduce four types of \textbf{D6-brane vertex operators} defined as
\begin{equation}
\mathsf{W}_{\bar{a}}(x)=\mathsf{w}_{\bar{a},0}(x):\exp\left(\sum_{n\neq 0}\mathsf{w}_{\bar{a},n}x^{-n}\right):,\quad a\in\four\label{eq:D6op}
\end{equation}
where $\mathsf{w}_{\bar{a},0}(x)$ are zero modes. The operator product with respect to the D0-brane vertex operators are imposed as
\begin{equation}
\begin{split}
    \mathsf{A}(x)\mathsf{W}_{\bar{a}}(x')&=\wick{\c{\mathsf{a}_{0}(x)}\c{\mathsf{w}_{\bar{a},0}(x')}}\mathscr{V}_{a}(x'/x)^{-1}:\mathsf{A}(x)\mathsf{W}_{\bar{a}}(x'):\\
    \mathsf{W}_{\bar{a}}(x')\mathsf{A}(x)&=\wick{\c{\mathsf{w}_{\bar{a},0}(x')}\c{\mathsf{a}_{0}(x)}}\mathscr{V}_{a}(q_{a}^{-1}x/x'):\mathsf{W}_{\bar{a}}(x')\mathsf{A}(x):.
\end{split}
\end{equation}
The commutation relation is then uniquely determined as
\begin{equation}
[\mathsf{a}_{n},\mathsf{w}_{\bar{a},m}]=-\frac{1}{n}\bfP_{a}^{[n]}\delta_{n+m,0},\quad \mathsf{w}_{\bar{a},n}=\frac{\mathsf{a}_{n}}{\bfP_{\bar{a}}^{[-n]}},\quad [\mathsf{w}_{\bar{a},n},\mathsf{w}_{\bar{b},m}]=-\frac{1}{n}\frac{\bfP_{\four}^{[n]}}{\bfP_{\bar{a}}^{[-n]}\bfP_{\bar{b}}^{[n]}}\delta_{n+m,0},\label{eq:D6op2}
\end{equation}
Using \eqref{eq:reflec_structfunc}, we impose the zero mode conditions as
\begin{equation}
    \wick{\c{\mathsf{a}_{0}(x)}\c{\mathsf{w}_{\bar{a},0}(x')}}=q_{a}\wick{\c{\mathsf{w}_{\bar{a},0}(x')}\c{\mathsf{a}_{0}(x)}}.\label{eq:D0D6zero}
\end{equation}
Explicitly, we impose $\wick{\c{\mathsf{a}_{0}(x)}\c{\mathsf{w}_{\bar{a},0}(x')}}=1$ (see section \ref{sec:zeromodes}). 
The free field realization of the contour integral formula of the tetrahedron instanton system is then given as follows.
\begin{proposition}
The tetrahedron instanton partition function~\eqref{eq:D6integralmod} is equivalent to the following vertex operator correlation function after analytic continuation,
\begin{equation}\label{eq:D6op_integral}
\mathfrak{q}^{k}\mathcal{Z}_{k}^{\D6}=\frac{\mathfrak{q}^{k} \mathcal{G}^{k}}{k!}\oint_{\text{JK}}\prod_{I=1}^{k}\frac{dx_{I}}{2\pi i x_{I}}\langle \mathsf{A}_{k}^{-1}\mathsf{W}_{\underline{n}}\rangle
\end{equation}
where
\begin{subequations}
\begin{align}
&\mathsf{A}_{k}=\prod_{I=1}^{k}\mathsf{A}(x_{I}),\quad \mathsf{W}_{\underline{n}}={:\prod_{a\in\four}\prod_{\alpha=1}^{n_{\bar{a}}}\mathsf{W}_{\bar{a}}(v_{\bar{a},\alpha}):},\quad \langle\mathsf{A}_{k}^{-1}\mathsf{W}_{\underline{n}}\rangle=\prod_{a\in\four}\prod_{\alpha=1}^{n_{\bar{a}}}\prod_{I=1}^{k}\mathscr{V}_{a}\left(\frac{v_{\bar{a},\alpha}}{x_{I}}\right)\prod_{I<J}\mathcal{A}_{\mathbb{C}^{4}}\left(\frac{x_{I}}{x_{J}}\right)^{-1}.
\end{align}
\end{subequations}
\end{proposition}
Note that the order of the operator product in $\mathsf{A}_{k}$ does not change the result.

\paragraph{One--loop perturbativ part}
The OPE of the $\mathsf{W}$-operators give the one loop perturbative part in \eqref{eq:D6oneloop}:
\begin{equation}
    \mathsf{W}_{\bar{b}}(v_{\bar{b},\beta})\mathsf{W}_{\bar{a}}(v_{\bar{a},\alpha})=\mathcal{Z}_{\text{1-loop}}^{\D6\tbar\D6}(v_{\bar{a},\alpha},\bar{a}\,|\,v_{\bar{b},\beta},\bar{b}):\mathsf{W}_{\bar{b}}(v_{\bar{b},\beta})\mathsf{W}_{\bar{a}}(v_{\bar{a},\alpha}):.
\end{equation}

\paragraph{Quiver interpretation}
Let us rederive the commutation relation \eqref{eq:D6op2}. Similar to the magnificent four case, this can be explicitly read from the quiver \eqref{eq:flavorquiver-tetrahedron}. Due to quadrality, it is enough to focus on one of the stacks of D6-branes which we take the D6$_{123}$-branes. The half qquiver Cartan matrices are given as
\bea
\adjustbox{valign=c}{
\begin{tikzpicture}[decoration={markings,mark=at position \arrowHeadPosition with {\arrow{latex}}}]
 \tikzset{
        box/.style={draw, minimum width=0.7cm, minimum height=0.7cm, text centered,thick},
        ->-/.style={decoration={
        markings,mark=at position #1 with {\arrow[scale=1.5]{>}}},postaction={decorate},line width=0.5mm},
        -<-/.style={decoration={
        markings,
        mark=at position #1 with {\arrow[scale=1.5]{<}}},postaction={decorate},line width=0.5mm}    
    }
\begin{scope}[xshift=4cm]
    \draw[fill=black!10!white,thick](0,0) circle(0.4cm);
    \node at (0,0) {$k$};
    \node[box,fill=black!10!white] at (0,1.6) {$n_{\bar{4}}$};
    \draw[postaction={decorate},thick] (-0.1,1.25)--(-0.1,0.4);
    \draw[postaction={decorate},red,thick] (0.1,1.25)--(0.1,0.4);
    \node[left] at (-0.1,0.8) {$1$};
    \node[right] at (0.1,0.8) {$\textcolor{red}{q_{4}}$};
    \draw[fill=black!10!white,thick](0,0) circle(0.4cm);
    \node at (0,0) {$k$};
\end{scope}
\end{tikzpicture}}\quad {\rightsquigarrow} \quad \begin{dcases}c_{\text{D0},\text{D6}}^{+}=-1+\textcolor{red}{q_{4}},\\c^{+}_{\text{D6},\text{D0}}=0.\end{dcases}
\eea
which gives the qquiver Cartan matrix
\bea
c_{\text{D0},\text{D6}}=c_{\text{D0},\text{D6}}^{+}+c_{\text{D6},\text{D0}}^{+\vee}=-1+\textcolor{red}{q_{4}}.
\eea
We then have
\bea
\relax[-\mathsf{a}_{n},\mathsf{w}_{\bar{4},m}]=-\frac{1}{n}c_{\text{D0},\text{D6}}^{[n]}\delta_{n+m,0}=\frac{1}{n}\bfP_{4}^{[n]}\delta_{n+m,0}
\eea
which indeed reproduces \eqref{eq:D6op2}.

\paragraph{Relation with magnificent four}
Let us focus on the 7d $\U(1)$ theory on $\mathbb{C}^{3}_{\bar{a}}\times \mathbb{S}^{1}$. Starting from the 9d $\U(1|1)$ theory of the magnificent four and tuning the parameter $K=q_{a}$, we have 
\begin{equation}\label{eq:D8D6reduction}
    \mathsf{Z}(q_{a},x)={:\frac{\mathsf{Z}(x)}{\mathsf{Z}(q_{a}x)}:}\simeq\mathsf{W}_{\bar{a}}(x),
\end{equation}
where the equality $\simeq$ is up to extra zero modes depending on the explicit form. In our notation in section \ref{sec:zeromodes}, this becomes an exact identity. Generally, starting from a 9d $\U(n|n)$ magnificent four theory with parameters $(K_{\alpha})_{\alpha=1}^{n}$, $n=\sum_{a\in\four}n_{\bar{a}}$ and tuning 
\begin{equation}
   (K_{\alpha})_{\alpha=1}^{n}\longrightarrow (K_{\bar{a},\alpha})_{\alpha=1}^{n_{\bar{a}}},\quad K_{\bar{a},\alpha}=q_{a}
\end{equation}
we have
\begin{equation}
    \langle \mathsf{A}_{k}^{-1}\widetilde{\mathsf{Z}}^{\underline{K}}_{\underline{n}}\rangle=\langle \mathsf{A}_{\underline{k}}^{-1}\mathsf{W}_{\underline{n}}\rangle,
\end{equation}
and thus, we obtain the tetrahedron instanton system. Note also that setting $K=q_{a}^{-1}$ gives 
\bea
\mathsf{Z}(q_{a}^{-1},x)={:\frac{\mathsf{Z}(x)}{\mathsf{Z}(q_{a}^{-1}x)}:}=\mathsf{W}_{\bar{a}}(q_{a}^{-1}x)^{-1}
\eea
which allows $\mathsf{W}_{\bar{a}}(x)$ to appear in the denominator. 

Physically, this property suggests that the D8-branes and anti-D8-branes annihilate in a specific distance under the $\Omega$-background and eventually reproduce the, generally intersecting, D6-branes system \cite{Nekrasov:2017cih,Nekrasov:2018xsb,Pomoni:2021hkn,Pomoni:2023nlf}, which is also interpreted as a tachyon condensation~\cite{Sen:1998sm}.

\paragraph{Supergroup generalization}
Following the construction of the magnificent four system in~\eqref{eq:D8D8barop}, we can write down the contour integral formula with D6 operators appearing in the denominator:
\beq\label{eq:D6supergroupdef}
{:\frac{\mathsf{W}_{\bar{a}}(v_{1})\cdots \mathsf{W}_{\bar{a}}(v_{n})}{\mathsf{W}_{\bar{a}}(u_{1})\cdots \mathsf{W}_{\bar{a}}(u_{m})}:}.
\eeq
The contour integral formula is proportional to
\beq\label{eq:D6supergroupLMNS1}
\oint \prod_{I=1}^{k}\frac{dx_{I}}{2\pi i x_{I}}\langle\mathsf{A}_{k}^{-1}\mathsf{W}_{\underline{n}|\underline{m}}\rangle
\eeq
where 
\begin{subequations}\label{eq:D6supergroupLMNS2}
\begin{align}
\mathsf{A}_{\underline{k}}&=\prod_{I=1}^{k}\mathsf{A}(x_{I}),\quad \mathsf{W}_{\underline{n}|\underline{m}}={:\prod_{a\in\four}\frac{\prod_{\alpha=1}^{n_{\bar{a}}}\mathsf{W}_{\bar{a}}(v_{\bar{a},\alpha})}{\prod_{\beta=1}^{m_{\bar{a}}}\mathsf{W}_{\bar{a}}(u_{\bar{a},\beta})}:},\\
\langle \mathsf{A}_{k}^{-1}\mathsf{W}_{\underline{n}|\underline{m}}\rangle&=\prod_{a\in\four}\prod_{\alpha=1}^{n_{\bar{a}}}\prod_{I=1}^{k}\mathscr{V}_{a}\left(\frac{v_{\bar{a},\alpha}}{x_{I}}\right)\prod_{a\in\four}\prod_{\alpha=1}^{m_{\bar{a}}}\prod_{I=1}^{k}\mathscr{V}_{a}\left(\frac{u_{\bar{a},\alpha}}{x_{I}}\right)^{-1}\prod_{I<J}\mathcal{A}_{\mathbb{C}^{4}}\left(\frac{x_{I}}{x_{J}}\right)^{-1}.
\end{align}
\end{subequations}
We expect that the operators in the denominators of \eqref{eq:D6supergroupdef} correspond to $\overbar{\D6}$-branes similar to the situation of the magnificent four system. We leave a detailed analysis of the evaluation of this contour integral formula and its relation with the 7d supergroup gauge theory for future work \cite{Noshita-Nawata}. We will see in later sections, that after tuning the parameters $\{u_{\bar{a},\beta}\}$ to special values, we can further reduce the system and obtain the contour integral formula of the spiked instanton system.

\subsection{Spiked instanton}\label{sec:spikedLMNS}
Let us next consider the spiked instanton system where D4-branes wrapping $\mathbb{C}^{2}_{A}\times \mathbb{S}^{1}\,(A\in\six)$ appear and find the free field realization of \eqref{eq:D4integralmod}. In this case, we need to deal with the rational function
\bea
\mathscr{S}_{\bar{A}}\left(\frac{v_{A,\alpha}}{x_{I}}\right).
\eea
The derivation is similar to the previous cases and the \textbf{D4-brane vertex operators} are defined as
\begin{equation}
\begin{split}
    &\mathsf{X}_{A}(x)=\mathsf{x}_{A,0}(x):\exp\left(\sum_{n\neq 0}\mathsf{x}_{A,n}x^{-n}\right):,\quad [\mathsf{x}_{A,n},\mathsf{x}_{B,m}]=-\frac{1}{n}\frac{\bfP_{\four}^{[n]}}{\bfP_{A}^{[-n]}\bfP_{B}^{[n]}}\delta_{n+m,0},
\end{split}\label{eq:D4op}
\end{equation}
where $\mathsf{x}_{A,0}(x)$ is the zero mode and the commutation relation with the $\mathsf{A}$-operator is
\begin{equation}\label{eq:D4op2}
    \mathsf{x}_{A,n}=\frac{\mathsf{a}_{n}}{\bfP_{A}^{[-n]}},\qquad [\mathsf{a}_{n},\mathsf{x}_{A,m}]=-\frac{1}{n}\bfP_{\bar{A}}^{[n]}\delta_{n+m,0}.
\end{equation}
Explicitly, the contraction formulas are 
\begin{equation}
\begin{split}
    \mathsf{A}(x)\mathsf{X}_{A}(\nu)&=\mathscr{S}_{\bar{A}}(\nu/x)^{-1}\wick{\c{\mathsf{a}_{0}(x)}\c{\mathsf{x}_{A,0}(x')}}:\mathsf{A}(x)\mathsf{X}_{A}(\nu):,\\
    \mathsf{X}_{A}(\nu)\mathsf{A}(x)&=\mathscr{S}_{\bar{A}}(q_{A}x/\nu)^{-1}\wick{\c{\mathsf{x}_{A,0}(x')}\c{\mathsf{a}_{0}(x)}}:\mathsf{X}_{A}(\nu)\mathsf{A}(x):.
\end{split}
\end{equation}
We impose the following condition on the zero modes so that the operator product of the right-hand side is the same rational function after analytic continuation \eqref{eq:reflec_structfunc}:
\begin{equation}
\begin{split}
    \wick{\c{\mathsf{a}_{0}(x)}\c{\mathsf{x}_{A,0}(x')}}=\wick{\c{\mathsf{x}_{A,0}(x')}\c{\mathsf{a}_{0}(x)}}.
\end{split}\label{eq:D0D4zero}
\end{equation}
We will use $\wick{\c{\mathsf{a}_{0}(x)}\c{\mathsf{x}_{A,0}(x')}}=1$ (see section \ref{sec:zeromodes} for explicit forms).

\begin{proposition}
The free field realization of the contour integral formula~\eqref{eq:D4integralmod} is
\begin{equation}\label{eq:D4op_integral}
\mathfrak{q}^{k}\mathcal{Z}_{k}^{\D4}=\frac{\mathfrak{q}^{k}\mathcal{G}^{k}}{k!}\oint_{\text{JK}}\prod_{I=1}^{k}\frac{dx_{I}}{2\pi i x_{I}}\langle\mathsf{A}^{-1}_{k}\mathsf{X}_{\underline{n}}\rangle
\end{equation}
where
\begin{subequations}
\begin{align}
&\mathsf{A}_{k}=\prod_{I=1}^{k}\mathsf{A}(x_{I}),\quad \mathsf{X}_{\underline{n}}={:\prod_{A\in\six}\prod_{\alpha=1}^{n_{A}}\mathsf{X}_{A}(v_{A,\alpha}):},\quad \langle\mathsf{A}^{-1}_{k}\mathsf{X}_{\underline{n}}\rangle=\prod_{A\in\six}\prod_{I=1}^{k}\prod_{\alpha=1}^{n_{A}}\mathscr{S}_{\bar{A}}\left(\frac{v_{A,\alpha}}{x_{I}}\right)\prod_{I<J}\mathcal{A}_{\mathbb{C}^{4}}\left(\frac{x_{I}}{x_{J}}\right)^{-1}.
\end{align}
\end{subequations}

\end{proposition}
The one-loop perturbative part in \eqref{eq:D4oneloop} is obtained by the OPE of the $\mathsf{X}$ operators:
\begin{equation}
    \mathsf{X}_{B}(v_{B,\beta})\mathsf{X}_{A}(v_{A,\alpha})=\mathcal{Z}^{\D4\tbar\D4}_{\text{1-loop}}(v_{A,\alpha},A\,|\,v_{B,\beta},B):\mathsf{X}_{B}(v_{B,\beta})\mathsf{X}_{A}(v_{A,\alpha}):.
\end{equation}

\paragraph{Quiver interpretation}Let us rederive the commutation relation in \eqref{eq:D4op2} directly from the quiver in \eqref{eq:flavorquiver-spiked}. Due to quadrality, it is enough to focus on the D4$_{12}$-branes. The half qquiver Cartan matrices are determined as
\bea
\adjustbox{valign=c}{
\begin{tikzpicture}[decoration={markings,mark=at position \arrowHeadPosition with {\arrow{latex}}}]
 \tikzset{
        box/.style={draw, minimum width=0.7cm, minimum height=0.7cm, text centered,thick},
        ->-/.style={decoration={
        markings,mark=at position #1 with {\arrow[scale=1.5]{>}}},postaction={decorate},line width=0.5mm},
        -<-/.style={decoration={
        markings,
        mark=at position #1 with {\arrow[scale=1.5]{<}}},postaction={decorate},line width=0.5mm}    
    }
\begin{scope}[xshift=4cm]
    \draw[fill=black!10!white,thick](0,0) circle(0.4cm);
    \node at (0,0) {$k$};
    \node[box,fill=black!10!white] at (0,1.6) {$n_{12}$};
    \draw[postaction={decorate},thick] (-0.2,1.25)--(-0.2,0.3);
    \draw[postaction={decorate},red,thick] (-0.07,1.25)--(-0.07,0.4);
    \draw[postaction={decorate},thick] (0.2,0.3)--(0.2,1.25);
    \draw[postaction={decorate},red,thick] (0.07,0.4)--(0.07,1.25);
    
    \node[left] at (-0.2,0.8) {$1,\textcolor{red}{q_{4}}$};
    \node[right] at (0.2,0.8) {$q_{12},\textcolor{red}{q_{124}=q_{3}^{-1}}$};
    
    \draw[fill=black!10!white,thick](0,0) circle(0.4cm);
    \node at (0,0) {$k$};

\end{scope}
\end{tikzpicture}}\quad {\large \rightsquigarrow} \quad \begin{dcases}
    c_{\text{D0},\text{D4}}^{+}=-1+\textcolor{red}{q_{4}},\\
    c_{\text{D4},\text{D0}}^{+}=-q_{12}+\textcolor{red}{q_{3}^{-1}}
\end{dcases}
\eea
which gives the qquiver Cartan matrix
\bea
c_{\text{D0},\text{D4}}=c_{\text{D0},\text{D4}}^{+}+c_{\text{D4},\text{D0}}^{+\vee}=-(1+q_{12}^{-1})+\textcolor{red}{(q_{4}+q_{3})}=-(1-q_{3})(1-q_{4}).
\eea
We then obtain
\bea
\,[-\mathsf{a}_{n},\mathsf{x}_{12,m}]=-\frac{1}{n}\delta_{n+m,0}c_{\text{D0},\text{D4}}^{[n]}=\frac{1}{n}\bfP_{34}^{[n]}\delta_{n+m,0}
\eea
which indeed reproduces the commutation relation \eqref{eq:D4op2}.

\paragraph{Relation with tetrahedron instanton}
Similar to the situation in \eqref{eq:D8D6reduction}, where the D8--$\overbar{\text{D8}}$ coupled system is reduced to the tetrahedron system, using \eqref{eq:D6supergroupdef} and specializing the parameters, we can obtain the contour integral formula of the spiked instanton system. This comes from the following relation:
\begin{equation}
    \mathsf{X}_{ab}(x)\simeq {:\frac{\mathsf{W}_{abc}(x)}{\mathsf{W}_{abc}(q_{c}x)}:}
\end{equation}
where $\simeq $ means they are equivalent up to zero modes. In our notation of the zero modes, it will become an exact identity (see \eqref{eq:oprelation1}). For example, we have 
\bea
{:\prod_{\alpha=1}^{n_{12}}\mathsf{X}_{12}(v_{12,\alpha}):}={:\prod_{\alpha=1}^{n_{12}}\frac{\mathsf{W}_{123}(v_{12,\alpha})}{\mathsf{W}_{123}(q_{3}v_{12,\alpha})}:}
\eea
which means by considering the D6--$\overbar{\text{D6}}$ system spanning $\mathbb{C}^{3}_{123}\times \mathbb{S}^{1}$ and tuning the positions of them with the parameter $q_{3}$, we can obtain the 5d theory on $\mathbb{C}_{12}^{2}\times \mathbb{S}^{1}$.

\subsection{D2-brane vertex operators}\label{sec:cplvortLMNS}
Given the D0-brane vertex operator, the non-zero modes of the D4 \eqref{eq:D4op}, D6 \eqref{eq:D6op}, and D8 \eqref{eq:D8op} vertex operators are written using $\mathsf{a}_{n}$ in a similar way:
\bea\label{eq:relationwithD0}
\mathsf{a}_{n}=\bfP_{\mathcal{S}}^{[-n]}\mathcal{A}_{n},\quad \mathcal{A}_{n}=\mathsf{x}_{A,n},\,\,\mathsf{w}_{\bar{a},n},\,\,\mathsf{z}_{n}
\eea
where $\mathcal{S}$ is a subset of $\four$ depending on the subspace the D-branes are spanning. Explicitly, the D6$_{123}$-brane vertex operator comes from $\bfP_{123}$. Therefore, it is natural to define a \textbf{D2-brane vertex operator} defined as 
\bea
    \mathsf{S}_{a}(x)=\mathsf{s}_{a,0}(x):\exp\left(\sum_{n\neq 0}\mathsf{s}_{a,n}x^{-n}\right):,\quad [\mathsf{s}_{a,n},\mathsf{s}_{b,m}]=-\frac{1}{n}\frac{\bfP_{\four}^{[n]}}{\bfP_{a}^{[-n]}\bfP_{b}^{[n]}}\delta_{n+m,0},\label{eq:D2op}
\eea
where $\mathsf{s}_{a,0}(x)$ is the zero modes and the relation with the $\mathsf{A}$-operator is
\bea\label{eq:D2op2}
\relax[\mathsf{a}_{n},\mathsf{s}_{a,m}]=-\frac{1}{n}\bfP_{\bar{a}}^{[n]}\delta_{n+m,0},\quad \mathsf{s}_{a,n}=\frac{\mathsf{a}_{n}}{\bfP_{a}^{[-n]}}.
\eea
The operator product formulas are
\bea
    \mathsf{A}(x)\mathsf{S}_{a}(x')&=\wick{\c{\mathsf{a}_{0}(x)}\c{\mathsf{s}_{a,0}(x')}}g_{\bar{a}}(x'/x)^{-1}: \mathsf{A}(x)\mathsf{S}_{a}(x'):,\\
    \mathsf{S}_{a}(x')\mathsf{A}(x)&=\wick{\c{\mathsf{s}_{a,0}(x')}\c{\mathsf{a}_{0}(x)}}g_{\bar{a}}(q_{a}x/x'): \mathsf{A}(x)\mathsf{S}_{a}(x'):,
\eea
where we impose the zero modes condition as
\bea
    \wick{\c{\mathsf{a}_{0}(x)}\c{\mathsf{s}_{a,0}(x')}}=\wick{\c{\mathsf{s}_{a,0}(x')}\c{\mathsf{a}_{0}(x)}}.\label{eq:D0D2zero}
\eea
so that the rational function arising on the right-hand side will be the same after using \eqref{eq:reflec_structfunc}. Explicitly, we have $\wick{\c{\mathsf{a}_{0}(x)}\c{\mathsf{s}_{a,0}(x')}}=1$ (see section \ref{sec:zeromodes}).

These D2-brane vertex operators are actually related to the screening currents of quiver W-algebras~\cite{Kimura:2015rgi}. Let us consider the two screening currents $\mathsf{S}_{1}(x)$ and $\mathsf{S}_{2}(x)$. Focusing on $\mathsf{S}_{2}(x)$, we have 
\bea
\relax    [\mathsf{s}_{2,n},\mathsf{s}_{2,m}]=-\frac{1}{n}\frac{1-q_{1}^{n}}{1-q_{2}^{-n}}(1-q_{3}^{n})(1-q_{4}^{n})\delta_{n+m,0}
\eea
which gives one of the screening currents of the affine quiver W-algebra \cite[eq.~(3.33)]{Kimura:2015rgi}. The screening current $\mathsf{S}_{1}(x)$ gives the other screening current. The other two screening currents $\mathsf{S}_{3}(x),\,\mathsf{S}_{4}(x)$ are introduced in a symmetric way using the quadrality. Thus, using two screening currents $\mathsf{S}_{a}(x),\,\mathsf{S}_{b}(x')\,(a\neq b)$ we will obtain six affine quiver W-algebras. 

\subsection{Zero modes conditions}\label{sec:zeromodes}
Let us impose some conditions on the zero modes $\mathsf{a}_{0}(x),\widetilde{\mathsf{z}}_{0}^{K}(x),\mathsf{w}_{\bar{a},0}(x),\mathsf{x}_{A,0}(x),\mathsf{s}_{a,0}(x)$ and determine the free field realizations of them. Using the observation \eqref{eq:relationwithD0}, we can see that the operator product with operators associated with D-branes intersecting only at a point will give rational functions. We impose the zero modes so that the operator product will be the same rational functions after the analytic continuation. For the cases, when the $\mathsf{A}$-operator is involved, the conditions are given in \eqref{eq:D0D0zero}, \eqref{eq:D0D8zero}, \eqref{eq:D0D6zero}, \eqref{eq:D0D4zero}, \eqref{eq:D0D2zero}. For the cases when the $\mathsf{S}$-operators are involved, we impose the following conditions:
\begin{itemize}[topsep=0pt, partopsep=0pt, itemsep=0pt]
    \item D2$_{a}$-D2$_{b}$ $(a\neq b)$: 
    \begin{subequations}
    \begin{align}
        \mathsf{S}_{a}(x)\mathsf{S}_{b}(x')&=\wick{\c{\mathsf{s}_{a,0}(x)}\c{\mathsf{s}_{b,0}(x')}}\mathscr{S}_{\overbar{ab}}(q_{a}x'/x):\mathsf{S}_{a}(x)\mathsf{S}_{b}(x'):,\\
        \mathsf{S}_{b}(x')\mathsf{S}_{a}(x)&=\wick{\c{\mathsf{s}_{b,0}(x')}\c{\mathsf{s}_{a,0}(x)}}\mathscr{S}_{\overbar{ab}}(q_{b}x'/x):\mathsf{S}_{a}(x)\mathsf{S}_{b}(x'):
    \end{align}
    \end{subequations}
    which gives the zero mode condition
    \bea
        \wick{\c{\mathsf{s}_{a,0}(x)}\c{\mathsf{s}_{b,0}(x')}}=\wick{\c{\mathsf{s}_{b,0}(x')}\c{\mathsf{s}_{a,0}(x)}},\quad a\neq b.\label{eq:D2D2zero}
    \eea
    \item D4$_{A}$-D2$_{c}$ $(c,d\in\bar{A})$:
    \begin{subequations}
    \begin{align}
        \mathsf{X}_{A}(x)\mathsf{S}_{c}(x')&=\wick{\c{\mathsf{x}_{A,0}(x)}\c{\mathsf{s}_{c,0}(x')}}\frac{1-q_{A}x'/x}{1-q_{A}q_{d}x'/x}: \mathsf{X}_{A}(x)\mathsf{S}_{c}(x'):,\\
        \mathsf{S}_{c}(x')\mathsf{X}_{A}(x)&=\wick{\c{\mathsf{s}_{c,0}(x')}\c{\mathsf{x}_{A,0}(x)}}\frac{1-q_{A}^{-1}x/x'}{1-q_{A}^{-1}q_{d}^{-1}x/x'}:\mathsf{X}_{A}(x)\mathsf{S}_{c}(x'):
    \end{align}
    \end{subequations}
    which gives the zero mode condition
    \bea
    \wick{\c{\mathsf{x}_{A,0}(x)}\c{\mathsf{s}_{c,0}(x')}}=q_{A}^{-1}q_{c}^{-1}\wick{\c{\mathsf{s}_{c,0}(x')}\c{\mathsf{x}_{A,0}(x)}}.\label{eq:D2D4zero}
    \eea
    \item D2$_{a}$-D6$_{\bar{a}}$:
    \begin{subequations}
    \begin{align}
        \mathsf{W}_{\bar{a}}(x)\mathsf{S}_{a}(x')&=\wick{\c{\mathsf{w}_{\bar{a},0}(x)}\c{\mathsf{s}_{a,0}(x')}}\frac{1}{1-q_{a}^{-1}x'/x}:\mathsf{W}_{\bar{a}}(x)\mathsf{S}_{a}(x'):,\\
        \mathsf{S}_{a}(x')\mathsf{W}_{\bar{a}}(x)&=\wick{\c{\mathsf{s}_{a,0}(x')}\c{\mathsf{w}_{\bar{a},0}(x)}}\frac{1}{1-q_{a}x/x'}:\mathsf{W}_{\bar{a}}(x)\mathsf{S}_{a}(x'):
    \end{align}
    \end{subequations}
    which gives the zero mode condition
    \begin{equation}
        \wick{\c{\mathsf{w}_{\bar{a},0}(x)}\c{\mathsf{s}_{a,0}(x')}}=\left(-\frac{x'}{q_{a}x}\right)\wick{\c{\mathsf{s}_{a,0}(x')}\c{\mathsf{w}_{\bar{a},0}(x)}}.\label{eq:D2D6zero}
    \end{equation}
\end{itemize}

We can do the same analysis for the D4-brane operators as
\begin{subequations}
\begin{align}
    \mathsf{X}_{A}(x)\mathsf{X}_{\bar{A}}(x')=\wick{\c{\mathsf{x}_{A,0}(x)}\c{\mathsf{x}_{\bar{A},0}(x')}}\left(1-q_{A}\frac{x'}{x}\right):\mathsf{X}_{A}(x)\mathsf{X}_{\bar{A}}(x'):,\\
    \mathsf{X}_{\bar{A}}(x')\mathsf{X}_{A}(x')=\wick{\c{\mathsf{x}_{\bar{A},0}(x')}\c{\mathsf{x}_{A,0}(x)}}\left(1-q_{\bar{A}}\frac{x}{x'}\right):\mathsf{X}_{A}(x)\mathsf{X}_{\bar{A}}(x'):
\end{align}
\end{subequations}
and impose the condition
\beq
   \mathsf{x}_{A,0}(x)\mathsf{x}_{\bar{A},0}(x')=\left(-q_{A}^{-1}\frac{x}{x'}\right)\mathsf{x}_{\bar{A},0}(x')\mathsf{x}_{A,0}(x)\label{eq:D4D4zero}
\eeq
but to make the discussion simple, we do not impose this condition.\footnote{\label{note:D4footnote}This condition only affects the zero modes when we are considering the quadratic relations of the $qq$-characters which will be derived in section \ref{sec:D4quadraticrelation}.}

Under these conditions, the free field realizations of the zero modes are given as 
\bea\label{eq:zeromodes1}
    &\mathsf{a}_{0}(x)=e^{\mathsf{a}_{0}},\quad \mathsf{s}_{a,0}(x)=x^{\mathsf{s}_{a,0}}e^{\widetilde{\mathsf{s}}_{a,0}},\quad \mathsf{w}_{\bar{a}}(x)=x^{\mathsf{w}_{\bar{a},0}}e^{\widetilde{\mathsf{w}}_{\bar{a},0}}e^{\widetilde{\widetilde{\mathsf{w}}}_{\bar{a},0}},\\
    &\mathsf{x}_{A,0}(x)=e^{\mathsf{x}_{A,0}}, \quad \widetilde{\mathsf{z}}_{0}^{K}(x)=x^{\mathsf{z}^{K}_{0}}e^{\widetilde{\mathsf{z}}^{K}_{0}}e^{\widetilde{\widetilde{\mathsf{z}}}^{K}_{0}}
\eea
with
\bea\label{eq:zeromodes2}
& \mathsf{a}_{0}=\mathsf{t}_{0},\quad \mathsf{s}_{a,0}=-(\log q_{a})^{-1}\mathsf{t}_{0},\quad \widetilde{\mathsf{s}}_{a,0}=-(\log q_{a})^{-1}\widetilde{\partial}_{\mathsf{t}},\\
    &\mathsf{w}_{\bar{a},0}=-\log q_{a}\,\widetilde{\mathsf{t}}_{0},\quad \widetilde{\mathsf{w}}_{\bar{a},0}=-\log q_{a}\log (-q_{a})\,\widetilde{\mathsf{t}}_{0},\quad \widetilde{\widetilde{\mathsf{w}}}_{\bar{a},0}=-\log q_{a}\partial_{\mathsf{t}},\\
    &\mathsf{x}_{A,0}=\log q_{c}\log q_{d}\,\widetilde{\mathsf{t}}_{0},\,\,\quad (\bar{A}=cd),\\
    &\mathsf{z}^{K}_{0}=-\log K \tilde{\mathsf{t}}_{0},\quad \widetilde{\mathsf{z}}^{K}_{0}=-\log K\log(-K)\widetilde{\mathsf{t}}_{0},\quad \tilde{\tilde{\mathsf{z}}}^{K}_{0}=-\log K\partial_{\mathsf{t}}
\eea
where we introduced two independent sets of zero modes
\beq
    [\partial_{\mathsf{t}},\mathsf{t}_{0}]=[\widetilde{\partial}_{\mathsf{t}},\tilde{\mathsf{t}}_{0}]=1,\quad [\mathsf{t}_{0},\tilde{\mathsf{t}}_{0}]=[\partial_{\mathsf{t}},\widetilde{\partial}]=[\mathsf{t}_{0},\widetilde{\partial}_{\mathsf{t}}]=[\tilde{\mathsf{t}}_{0},\partial_{\mathsf{t}}]=0.\label{eq:zeromodesdef}
\eeq
The normal ordering is defined as 
\beq
{:\partial_{\mathsf{t}}\,\mathsf{t}_{0}:}=\mathsf{t}_{0}\partial_{\mathsf{t}},\quad {:\tilde{\partial}_{\mathsf{t}}\,\tilde{\mathsf{t}}_{0}:}=\tilde{\mathsf{t}}_{0}\tilde{\partial}_{\mathsf{t}}.
\eeq
See \cite{Kimura:2023bxy} for the explicit derivation.

Under the above conditions, we actually have extra relations 
\beq
    \mathsf{a}_{0}(x)={:\frac{\mathsf{s}_{a,0}(x)}{\mathsf{s}_{a,0}(q_{a}x)}:},\quad \mathsf{x}_{ab,0}(x)={:\frac{\mathsf{w}_{abc,0}(x)}{\mathsf{w}_{abc,0}(q_{c}x)}:},\quad \mathsf{w}_{\bar{a},0}(x)=\tilde{\mathsf{z}}^{q_{a}}_{0}(x)\label{eq:zeromodesrelation}
\eeq
which imply
\beq
    \mathsf{A}(x)={:\frac{\mathsf{S}_{a}(x)}{\mathsf{S}_{a}(q_{a}x)}:},\quad \mathsf{X}_{ab}(x)={:\frac{\mathsf{W}_{abc}(x)}{\mathsf{W}_{abc}(q_{c}x)}:},\quad \mathsf{W}_{\bar{a}}(x)={:\frac{\mathsf{Z}(x)}{\mathsf{Z}(q_{a}x)}:}=\widetilde{\mathsf{Z}}^{q_{a}}(x).\label{eq:oprelation1}
\eeq
In our notation, the relation between the D2 and D4 operators are 
\begin{align}
    \mathsf{S}_{a}(x)=\mathsf{s}_{a,0}(x):\frac{\mathsf{X}_{ab}(x)}{\mathsf{X}_{ab}(q_{b}x)}:
\end{align}
where extra zero modes appear in front. Using these relations, we also have 
\begin{subequations}
\begin{align}
\mathsf{A}(x)&=\mathsf{a}_{0}(x):\frac{\mathsf{X}_{ab}(x)\mathsf{X}_{ab}(q_{ab}x)}{\mathsf{X}_{ab}(q_{a}x)\mathsf{X}_{ab}(q_{b}x)}:,\label{eq:oprelationwithD0-D4}\\
&=\mathsf{a}_{0}(x):\frac{\mathsf{W}_{abc}(x)\mathsf{W}_{abc}(q_{ab}x)\mathsf{W}_{abc}(q_{ac}x)\mathsf{W}_{abc}(q_{bc}x)}{\mathsf{W}_{abc}(q_{a}x)\mathsf{W}_{abc}(q_{b}x)\mathsf{W}_{abc}(q_{c}x)\mathsf{W}_{abc}(q_{abc}x)}:\label{eq:oprelationwithD0-D6}\\
&=\mathsf{a}_{0}(x):\frac{\mathsf{Z}(x)^{2}\prod_{a<b}\mathsf{Z}(q_{ab}x)}{\prod_{a\in\four}\mathsf{Z}(q_{a}x)\prod_{a\in\four}\mathsf{Z}(q_{a}^{-1}x)}:\label{eq:oprelationwithD0-D8}
\end{align}\label{eq:oprelationwithD0}
\end{subequations}



\chapter{Quantum algebraic aspects of BPS \texorpdfstring{$qq$}{qq}-characters}\label{chap:quantum-algebra-BPSqq}
In this chapter, we introduce the BPS $qq$-characters and BPS quiver W-algebras\footnote{We do not distinguish them.} of the gauge origami system and study the quantum algebraic properties of them. In section~\ref{sec:BPSqq-intro}, we give some explanations of what a $qq$-character is and the motivation to study them. In section~\ref{sec:D4qqcharacter}, we study basic aspects of the D4 $qq$-character. We study the relation with screening charges and dicuss the BPS/CFT correspondence of the spiked instanton system. We then move on to D6 $qq$-characters in section~\ref{sec:D6qqcharacter}. The $qq$-characters introduced here are novel types of $qq$-characters originally introduced in the author's paper \cite{Kimura:2023bxy}. We study the relation with plane partitions and discuss the BPS/CFT correspondence of the tetrahedron instanton system. We also discuss the relation with the D4 $qq$-characters. Generalizations of the D6 $qq$-characters are discussed too. We then move on to the D8 $qq$-characters in section~\ref{sec:D8qq}. Compared to D4 and D6 $qq$-characters, we do not have a screening charge so we can not derive them by the commutativity. Instead, we show that infinite products of the D6 $qq$-characters produce the D8 $qq$-characters. We also discuss the quadratic relation of the D6 and D8 $qq$-characters and its physical interpretation. In the final section~\ref{sec:toroidal_alg}, we point out observations and discuss the relation with quantum toroidal $\mathfrak{gl}_{1}$. This chapter is mainly based on \cite{Kimura:2023bxy} and some proofs come from \cite{Kimura:2024osv}.

\section{BPS \texorpdfstring{$qq$}{qq}-characters and BPS quiver W-algebras}\label{sec:BPSqq-intro}
The free field realization of the Witten index given in Chap.~\ref{chap:freefield-vertexop} implies the existence of some quantum algebraic structure. Suppose we have the following free field realization for the $k$-instanton sector of a supersymmetric gauge theory:
\bea
\mathsf{T}_{k}(v_{i},q_{a})&\coloneqq \frac{\mathcal{G}^{k}}{k!}\oint _{\text{JK}}\prod_{I=1}^{k}\frac{dx_{I}}{2\pi i x_{I}}\prod_{I=1}^{k}\mathsf{A}(x_{I})^{-1}:\prod_{i}\mathsf{V}(v_{i}):\\
&=\frac{\mathcal{G}^{k}}{k!}\oint _{\text{JK}}\prod_{I=1}^{k}\frac{dx_{I}}{2\pi i x_{I}}\mathcal{Z}(v_{i},x_{I}):\prod_{I=1}^{k}\mathsf{A}(x_{I})^{-1}\prod_{i}\mathsf{V}(v_{i}):,
\eea
where we took the contraction of the vertex operators at the second line and assumed that $\mathcal{Z}(v_{i},x_{I})$ is some rational function. The normal product part is regular and so pole contributions only come from $\mathcal{Z}(v_{i},x_{I})$. Using the JK-residue prescription, we can evaluate this contour integral and take the residues. This time, the existence of the vertex operator part makes $\mathsf{T}_{k}(v_{i},q_{a})$ to be a summation of operators. Moreover, after taking the sum over the instanton number, we obtain
\bea
\mathsf{T}(v_{i})=\sum_{k=0}^{\infty}\mathfrak{q}^{k}\mathsf{T}_{k}(v_{i},q_{a}).
\eea
Taking the vacuum expectation value of $\mathsf{T}(v_{i})$, one will simply get the partition function. This duality relating correlation function of vertex operators and instanton partition functions is usually called the \textbf{BPS/CFT correspondence} \cite{Nekrasov:2015wsu,Nekrasov:2016qym,Nekrasov:2016ydq}.



Given the operator $\mathsf{T}(v_{i})$, the following questions naturally arise.\footnote{The two algebras appearing here generally might not be the same.}
\begin{itemize}
    
    \item Since they are sums of operators, they should be related with the \textit{character} of some algebra. What is the algebra whose character gives the operator $\mathsf{T}(v_{i})$? How is that algebra related with the monomial terms of $\mathsf{T}(v_{i})$?
    \item What is the algebraic structure of these operator valued Witten indices? Do they form a closed algebra? 
\end{itemize}
The operators $\mathsf{T}(v_{i})$ are usually called $qq$-characters \cite{Nekrasov:2015wsu} or generators of deformed W-algebras \cite{Kimura:2015rgi} in the literature. The former terminology comes from the first motivation while the latter terminology comes from the second motivation. In this thesis, we do not distinguish the terminology. See \cite{Kimura:2020jxl} for an excellent review of $qq$-characters and their relation with deformed W-algebras.

Instead of giving a detailed review of the $qq$-characters, let us study the pure SYM explained in Chap.~\ref{chap:ADHM-localization} to be more concrete. Starting from the contour integral formula for the $\U(1)$ partition function and doing the prescription given in the previous chapter, one will find the following vertex operators
\bea
\mathsf{X}(x)=\mathsf{x}_{0}(x):\exp\left(\sum_{n\neq 0}\mathsf{x}_{n}x^{-n}\right):,\quad \tilde{\mathsf{A}}(x)=\tilde{\mathsf{a}}_{0}(x):\exp\left(\sum_{n\neq 0}\widetilde{\mathsf{a}}_{n}x^{-n}\right):
\eea
with
\bea\label{eq:A1-quiver}
[\tilde{\mathsf{a}}_{n},\tilde{\mathsf{a}}_{m}]=-\frac{1}{n}\delta_{n+m,0}(1-q_{1}^{n})(1-q_{2}^{n})(1+q_{12}^{-n}),\quad \mathsf{x}_{n}=\frac{\tilde{\mathsf{a}}_{n}}{(1-q_{1}^{-n})(1-q_{2}^{-n})}
\eea
giving
\bea
\mathsf{T}_{k}(x)=\frac{1}{k!}\left(\frac{(1-q_{12})}{(1-q_{1})(1-q_{2})}\right)^{k}\oint_{\text{JK}}\prod_{I=1}^{k}\frac{dx_{I}}{2\pi i x_{I}}\prod_{I=1}^{k}\tilde{\mathsf{A}}(x_{I})^{-1}\mathsf{X}(x)
\eea
where we omitted the discussion of zero-modes. The VEV of this operator actually reproduces the contour integral formula in Prop.~\ref{prop:LMNSformula-mod} after redefinitions of the topological term. Strictly speaking, the zero-modes are strongly related with the topological term and the Chern--Simons term. This is because, focusing on the zero-modes of $\mathsf{A}^{-1}(x)$, the zero-modes taking the form $e^{\mathsf{d}}$, where $\mathsf{d}$ is some zero-mode operator, give $:e^{k\mathsf{d}}:$ or $:e^{|\lambda|\mathsf{d}}:$ (after localization), while the zero-modes taking the form $x^{\mathsf{d}}$ give $:\prod_{I=1}^{k}x_{I}^{\mathsf{d}}:$ or $:\left(\prod_{\Abox\in\lambda}\chi_{12,x}(\Bbox)\right)^{\mathsf{d}}:$ (after localization). Roughly speaking, they will give operator lift ups of the topological and Chern--Simons terms.

After evaluating the JK residues, the operator will be
\bea
\mathsf{T}_{\text{SYM}}(x)=\sum_{k=0}^{\infty}\mathfrak{q}^{k}\mathcal{Z}_{\text{SYM}}[\lambda]:\mathsf{X}(x)\prod_{\Abox\in\lambda}\tilde{\mathsf{A}}(\chi_{12,x}(\Bbox))^{-1}:
\eea
whose expectation value gives the $\U(1)$ partition function. Actually, after redefinitions of operators, this $qq$-character is related with the AFS intertwiner \cite{Awata:2011ce} (the intertwiner of quantum toroidal $\mathfrak{gl}_{1}$) which is also equivalent to the refined topological vertex \cite{Iqbal:2007ii,Awata:2008ed}. Somehow the deformed W-algebra interpretation of this $qq$-character is still missing in the literature.

Actually, we can introduce other types of vertex operators such as
\bea
\mathsf{Y}(x)=\mathsf{y}_{0}(x):\exp\left(\sum_{n\neq 0}\mathsf{y}_{n}x^{-n}\right),\quad \tilde{\mathsf{S}}_{i}(x)=\tilde{\mathsf{s}}_{i,0}(x):\exp\left(\sum_{n\neq 0}\tilde{\mathsf{s}}_{i,n}x^{-n}\right):\,\,(i=1,2)
\eea
with 
\bea
\mathsf{y}_{n}=\frac{\tilde{\mathsf{a}}_{n}}{(1+q_{12}^{n})},\quad \tilde{\mathsf{s}}_{i,n}=\frac{\tilde{\mathsf{a}}_{n}}{1-q_{i}^{-n}}.
\eea
By defining 
\bea
\tilde{\mathsf{T}}_{k}(x)&=\frac{1}{k!}\left(\frac{1-q_{12}}{(1-q_{1})(1-q_{2})}\right)^{k}\oint_{\text{JK}}\prod_{I=1}^{k}\frac{dx_{I}}{2\pi i x_{I}}\prod_{I=1}^{k}\prod_{I=1}^{k}\tilde{\mathsf{A}}(x_{I})^{-1}\mathsf{Y}(x)\\
&=\frac{1}{k!}\left(\frac{1-q_{12}}{(1-q_{1})(1-q_{2})}\right)^{k}\oint_{\text{JK}}\prod_{I=1}^{k}\frac{dx_{I}}{2\pi i x_{I}}\prod_{I=1}^{k}\mathscr{S}_{12}\left(\frac{v}{x_{I}}\right)\prod_{I\neq J}\mathscr{S}_{12}\left(\frac{x_{I}}{x_{J}}\right)^{-1}:\prod_{I=1}^{k}\prod_{I=1}^{k}\tilde{\mathsf{A}}(x_{I})^{-1}\mathsf{Y}(x):
\eea
and evaluating the poles, one will see that this will be non-zero only when $k=0,1$. We then have
\bea
\widetilde{\mathsf{T}}(x)=\mathsf{Y}(x)+\mathfrak{q}:\mathsf{Y}(x)\mathsf{A}(x)^{-1}:.
\eea
Actually this operator is just the generator of the $q$-Virasoro \cite{Shiraishi:1995rp} obeying the quadratic relation
\bea
\mathsf{f}\left(\frac{x'}{x}\right)\widetilde{\mathsf{T}}(x)\widetilde{\mathsf{T}}(x')-\mathsf{f}\left(\frac{x}{x'}\right)\widetilde{\mathsf{T}}(x')\widetilde{\mathsf{T}}(x)=-\frac{(1-q_{1})(1-q_{2})}{1-q_{12}}\left(\delta\left(q_{12}\frac{x'}{x}\right)-\delta\left(q_{12}^{-1}\frac{x'}{x}\right)\right),
\eea
where
\bea
\mathsf{f}(x)=\exp\left(\sum_{n=1}^{\infty}\frac{1}{n}\frac{(1-q_{1}^{n})(1-q_{2}^{n})}{(1+q_{12}^{n})}x^{n}\right).
\eea
Another interesting property is the commutativity with the screening charges:
\bea
\relax[\widetilde{\mathsf{T}}(x),\widetilde{\mathscr{Q}}_{1,2}(x')]=0
,\quad \widetilde{\mathscr{Q}}_{1,2}(x)=\sum_{k\in\mathbb{Z}}\widetilde{\mathsf{S}}_{1,2}(q_{1,2}^{k}x).
\eea
Moreover, we actually have
\bea
\relax[\widetilde{\mathsf{T}}(x),\mathsf{T}_{\text{SYM}}(x')]=0.
\eea

Let us summarize what we did here. Starting from the same root current $\widetilde{\mathsf{A}}(x)$, we gave multiple vertex operators,\footnote{The vertex operators here were introduced by classifying the possible vertex operators whose contraction with the root current give a rational function. From the physical viewpoint, the root current $\widetilde{\mathsf{A}}(x)$ comes from the gauge nodes of the low energy theory, while the vertex operators $\mathsf{X}(x),\mathsf{Y}(x),\widetilde{\mathsf{S}}(x)$ correspond to the flavor nodes/framing nodes. This phenomenon is general for other cases such as the gauge origami system of $\mathbb{C}^{4}$. In other words, gauge origami is a theory of the classification of the framing nodes.} and then introduced $qq$-characters and screening charges. Given such $qq$-characters and screening charges, we studied the quadratic relations and commutativity of them. In this way, we managed to find the deformed W-algebra structure.

The generator of the $q$-Virasoro $\widetilde{\mathsf{T}}(x)$ is sometimes called the $A_{1}$ $qq$-character. Focusing on the non-zero modes, we actually have the relation
\bea
:\mathsf{Y}(x)\mathsf{Y}(q_{12}^{-1}x):\simeq \tilde{\mathsf{A}}(x)
\eea
up to zero-modes. The $qq$-character is rewritten as
\bea
\widetilde{\mathsf{T}}(x)\simeq\mathsf{Y}(x)+\frac{\mathfrak{q}}{\mathsf{Y}(q_{12}^{-1}x)}
\eea
where we are not being careful with the zero-modes again. Since the operator sum resembles the two-dimensional representation character of $\SL(2)$
\bea
\chi(y)=y+\frac{1}{y}
\eea
it is called the $A_{1}$ $qq$-character. Although we will not discuss in detail, taking the Nekrasov--Shatashvili (NS) limit $q_{2}\rightarrow 1$, the vacuum expectation value $\bra{0}\widetilde{\mathsf{T}}(x)\widetilde{\mathsf{S}}_{2}(x')\ket{0}$ becomes the $q$-character of the two-dimensional representation of $\U_{q}(\widehat{\mathfrak{sl}}_{2})$ and each monomial term corresponds with the basis of the two-dimensional representation
\bea
\mathsf{Y}(x) \,\,\leftrightarrow \,\, \adjustbox{valign=c}{\begin{tikzpicture}
        \draw[->] (0,0)--(0,1);
        \draw[->] (0,0)--(1,0);
        \end{tikzpicture}
        }
\quad \quad :\mathsf{Y}(x)\mathsf{A}^{-1}(x): \,\,\leftrightarrow \,\, \adjustbox{valign=c}{\begin{tikzpicture}
        \draw[->] (0,0)--(1,0);
        \draw (0.5,0)--(0.5,0.5)--(0,0.5);
        \draw[->] (0,0)--(0,1);
        \end{tikzpicture}
        }
\eea

We will apply this procedure to the gauge origami system and study the associated deformed W-algebra structure and the character structure in the following sections. Before doing so, let us summarize what kind of $qq$-characters arise from the free field realization of the gauge origami system given in the previous chapter:
\begin{itemize}
    \item D4 $qq$-characters:
    \bea
    \mathsf{T}_{A}(x)=\mathsf{X}_{A}(x)+\cdots,\quad A\in\six
    \eea
    \item D6 $qq$-characters:
    \bea
    \mathsf{T}_{\bar{a}}(x)=\mathsf{W}_{\bar{a}}(x)+\cdots,\quad a\in\four
    \eea
    \item D8 $qq$-characters:
    \bea
    \mathsf{T}_{\four}^{K}(x)=\mathsf{Z}(K,x)+\cdots
    \eea
\end{itemize}
We only list here the first term, i.e., the zero-instanton sector. This zero-instanton sector term is usually called the \textit{highest weight} in the literature.

\section{D4 \texorpdfstring{$qq$}{qq}-characters}\label{sec:D4qqcharacter}

In this section, we study basic properties of the D4 $qq$-characters. These $qq$-characters were introduced in~\cite{Nekrasov:2015wsu} and the corresponding algebraic structure is known to be the generator of the affine quiver W-algebra~\cite{Kimura:2015rgi}. 

In section~\ref{sec:D2qqcharacter}, we introduce the screening charges and study the commutation relations of them. We also discuss the relation with one-dimensional partitions. In section~\ref{sec:D4qqaffinequiver}, we introduce six types of D4 $qq$-characters corresponding to the six possible configurations of D4-branes in $\mathbb{C}^{4}$ in an equal footing. The commutativity with the screening charges is discussed too. Using them we show that we can establish the BPS/CFT correspondence of the spiked instanton system in section~\ref{sec:D4qqandspikedinst}. We also show that the D4 $qq$-characters can be obtained by infinite products of the screening currents in section~\ref{sec:spikedandscreening}. The quadratic relations of the D4 $qq$-characters are discussed in section~\ref{sec:D4quadraticrelation}.




\subsection{Screening charges}\label{sec:D2qqcharacter}

\begin{definition}
    The \textbf{screening charges} are defined as
    \beq
        \mathscr{Q}_{a}(x)=\sum_{k\in\mathbb{Z}}\mathsf{S}_{a}(q_{a}^{k}x),\quad a\in\four.
    \eeq
  
\end{definition}
Based on the observation that the D2-brane vertex operator plays role here, we also call them the \textbf{D2-brane $qq$-characters}. Shortly, we call them D2 $qq$-characters.

\begin{theorem}\label{thm:D2qq-commute}
    The screening charges with different indices commute with each other: 
    \beq
        \relax[\mathscr{Q}_{a}(x),\mathscr{Q}_{b}(x')]=0,\quad a\neq b \in\four.
    \eeq
\end{theorem}
 \begin{proof}
     Let us focus on the commutation relation between $\mathscr{Q}_{1}(x)$ and $\mathscr{Q}_{4}(x)$. Other cases are obtained by using the quadrality. The operator products between $\mathsf{S}_{1}(x)$ and $\mathsf{S}_{4}(x)$ are 
     \begin{align}
     \begin{split}
    \mathsf{S}_{4}(x)\mathsf{S}_{1}(x')&=\mathscr{S}_{23}\left(q_{4}x'/x\right):\mathsf{S}_{4}(x)\mathsf{S}_{1}(x'):,\\
    \mathsf{S}_{1}(x')\mathsf{S}_{4}(x)&=\mathscr{S}_{23}(q_{1}x/x'):\mathsf{S}_{1}(x')\mathsf{S}_{4}(x):.
    \end{split}
\end{align}
We then have 
\begin{align}
\begin{split}
    &[\mathsf{S}_{4}(q_{4}^{k}x),\mathsf{S}_{1}(x')]\\
    =&\frac{(1-q_{2})(1-q_{3})}{(1-q_{2}q_{3})}\left(\delta\left(\frac{x'}{q_{4}^{k-1}x}\right)\textcolor{red}{:\mathsf{S}_{4}(q_{4}^{k}x)\mathsf{S}_{1}(q_{4}^{k-1}x):}-\delta\left(\frac{x'}{q_{1}q_{4}^{k}x}\right):\mathsf{S}_{4}(q_{4}^{k}x)\mathsf{S}_{1}(q_{1}q_{4}^{k}x):\right)\\
    =&\frac{(1-q_{2})(1-q_{3})}{(1-q_{2}q_{3})}\left(\delta\left(\frac{x'}{q_{4}^{k-1}x}\right)\textcolor{red}{:\mathsf{S}_{4}(q_{4}^{k-1}x)\mathsf{A}^{-1}(q_{4}^{k-1}x)\mathsf{S}_{1}(q_{4}^{k-1}x):}-\delta\left(\frac{x'}{q_{1}q_{4}^{k}x}\right):\mathsf{S}_{4}(q_{4}^{k}x)\mathsf{S}_{1}(q_{1}q_{4}^{k}x):\right)\\
    =&\frac{(1-q_{2})(1-q_{3})}{(1-q_{2}q_{3})}\left(\delta\left(\frac{x'}{q_{4}^{k-1}x}\right)\textcolor{red}{:\mathsf{S}_{4}(q_{4}^{k-1}x)\mathsf{S}_{1}(q_{1}q_{4}^{k-1}x):}-\delta\left(\frac{x'}{q_{1}q_{4}^{k}x}\right):\mathsf{S}_{4}(q_{4}^{k}x)\mathsf{S}_{1}(q_{1}q_{4}^{k}x):\right)
\end{split}
\end{align}
where we used 
\beq
    \mathsf{A}(x)={:\frac{\mathsf{S}_{1}(x)}{\mathsf{S}_{1}(q_{1}x)}:}={:\frac{\mathsf{S}_{4}(x)}{\mathsf{S}_{4}(q_{4}x)}:}.
\eeq
The commutation relation between the screening charge $\mathscr{Q}_{4}(x)$ and $\mathsf{S}_{1}(x')$ is 
\beq
    \relax[\mathscr{Q}_{4}(x),\mathsf{S}_{1}(x')]=\frac{(1-q_{2})(1-q_{3})}{(1-q_{2}q_{3})}\sum_{k\in\mathbb{Z}}\left(\delta\left(\frac{x'}{q_{4}^{k}x}\right)-\delta\left(\frac{x'}{q_{1}q_{4}^{k}x}\right)\right):\mathsf{S}_{4}(q_{4}^{k}x)\mathsf{S}_{1}(q_{1}q_{4}^{k}x):
\eeq
which is a total difference. Thus, we finally have 
\beq
   \relax [\mathscr{Q}_{4}(x),\mathscr{Q}_{1}(x')]=0.
\eeq
\end{proof}


Each term of the screening charge has a nice pictorial interpretation using the one-dimensional partition labeled by $k\in\mathbb{Z}$:
\bea
\mathsf{S}_{a}(q_{a}^{k}x),\,\,\,k\in\mathbb{Z} \quad \Longleftrightarrow \quad 
   \adjustbox{valign=c}{ \begin{tikzpicture}
        \draw[->] (-2,0)--(4,0);
        \draw[thick]   (-0.5,-0.5)--(-0.5,1);
        \draw (-2,0.7)--(3,0.7);
        \draw (3,0)--(3,0.7);
        \draw (0.2,0)--(0.2,0.7);
        \draw (0.9,0)--(0.9,0.7);
         \draw (1.6,0)--(1.6,0.7);
          \draw (2.3,0)--(2.3,0.7);
          \draw (-1.2,0)--(-1.2,0.7);
          \node at (-1.6,0.35) {$\cdots$};
          \node at (4,0) [right]{$q_{a}$};
          \node at (-0.15,0)[below]{1};
          \node at (0.55,0)[below]{2};
          \node at (1.25,0)[below]{$\cdots$};
          \node at (1.95,0)[below]{$\cdots$};
          \node at (2.65,0)[below]{$k$};
    \end{tikzpicture}}\label{eq:D2vectorfigure1}
\eea
The one-dimensional partition here is a one-dimensional sequence of boxes extending semi-infinitely from the left of the border. 

The screening charge $\mathscr{Q}_{a}(x)$ is understood as a collection of possible one-dimensional partitions labeled by $k\in\mathbb{Z}$ with the $q$-coordinate of the origin $x$:
\bea
\mathscr{Q}_{a}(x) \quad \Longleftrightarrow \left\{\quad\left. 
   \adjustbox{valign=c}{ \begin{tikzpicture}
        \draw[->] (-2,0)--(4,0);
        \draw[thick]   (-0.5,-0.5)--(-0.5,1);
        \draw (-2,0.7)--(3,0.7);
        \draw (3,0)--(3,0.7);
        \draw (0.2,0)--(0.2,0.7);
        \draw (0.9,0)--(0.9,0.7);
         \draw (1.6,0)--(1.6,0.7);
          \draw (2.3,0)--(2.3,0.7);
          \draw (-1.2,0)--(-1.2,0.7);
          \node at (-0.15,0.35) {$x$};
          \node at (-1.6,0.35) {$\cdots$};
          \node at (4,0) [right]{$q_{a}$};
          \node at (-0.15,0)[below]{1};
          \node at (0.55,0)[below]{2};
          \node at (1.25,0)[below]{$\cdots$};
          \node at (1.95,0)[below]{$\cdots$};
          \node at (2.65,0)[below]{$k$};
    \end{tikzpicture}}\quad \right|\,\, k\in\mathbb{Z} \right\}\label{eq:D2vectorfigure2}
\eea
Using \eqref{eq:oprelation1}, we have 
\bea
    \mathsf{S}_{a}(q_{a}^{k}x)={:\prod_{i=-\infty}^{k}\mathsf{A}^{-1}(q_{a}^{i-1}x):}
\eea
which means we can interpret each box of the one-dimensional partition with the $q$-coordinates $xq_{a}^{i-1}$ as the operator $\mathsf{A}^{-1}(xq_{a}^{i-1})$. The screening charge $\mathscr{Q}_{a}(x)$ is then understood as a collection of the possible one-dimensional partitions labeled by $k\in\mathbb{Z}$.

\subsection{D4 \texorpdfstring{$qq$}{qq}-characters and affine quiver W-algebra}\label{sec:D4qqaffinequiver}
The main claim of this section is that the D4$_{A}$ $qq$-characters commute with screening charges $\mathscr{Q}_{a}(x')$ for $a\in\bar{A}$. Conversely, the commutativity uniquely determines the whole formula of the D4 $qq$-character.

Let us consider the $qq$-character generated by the D4 vertex operator $\mathsf{X}_{A}(x)\,(A\in\six)$. The D4 $qq$-character generated here is identified with the generator of the affine quiver W-algebra of $\Gamma = \widehat{A}_0$ in \cite{Kimura:2015rgi}. Each term of the $qq$-character is decomposed into two parts, the operator part and the coefficient part. The operator part is determined by the iWeyl reflection \cite{Nekrasov:2012xe,Nekrasov:2013xda} which is defined as follows. 
\begin{definition}
    The iWeyl reflection of the D4 vertex operator $\mathsf{X}_{A}(x)\,\,(A\in\six)$ is 
    \begin{equation}
        \mathsf{X}_{A}(x)\longrightarrow {:\mathsf{X}_{A}(x)\mathsf{A}^{-1}(x):},\quad A\in\six.
    \end{equation}
    Using the property in \eqref{eq:oprelationwithD0-D4}, we have 
\beq
    \mathsf{X}_{ab}(x)\longrightarrow {:\mathsf{a}_{0}(x)^{-1}\frac{\mathsf{X}_{ab}(q_{a}x)\mathsf{X}_{ab}(q_{b}x)}{\mathsf{X}_{ab}(q_{ab}x)}:},\quad a\neq b\in\four.
\eeq
\end{definition}
The operator part of the $qq$-character is obtained by changing the root current to the $\mathsf{X}$-operators and applying the iWeyl reflection recursively.\footnote{The iWeyl reflection will act only on the numerators.} The operators will be classified by two-dimensional Young diagrams as
\beq\label{eq:D4monomialterm}
    {:\mathsf{X}_{ab}(x)\prod_{\Abox\in\lambda}\mathsf{A}^{-1}(\chi_{ab,x}(\Bbox)):}={:\mathsf{a}_{0}(x)^{-1}\frac{\prod_{\Abox\in A(\lambda)}\mathsf{X}_{ab}(\chi_{ab,x}(\Bbox))}{\prod_{\Abox\in R(\lambda)}\mathsf{X}_{ab}(q_{ab}\chi_{ab,x}(\Bbox))}:}.
\eeq
Following the correspondence in \eqref{eq:D2vectorfigure1} and \eqref{eq:D2vectorfigure2}, we can visualize each monomial terms of the D4 $qq$-character using the Young diagrams:
\bea\label{eq:D4Youngcorrespondence}
{:\mathsf{X}_{ab}(x)\prod_{\Abox \in\lambda}\mathsf{A}^{-1}(\chi_{ab,x}(\Bbox)):}\quad \longleftrightarrow \quad \adjustbox{valign=c}{\begin{tikzpicture}
 \fill[red!20!white] (0.9,1.4)--(1.6,1.4)--(1.6,2.1)--(0.9,2.1)--(0.9,1.4);
        \draw[->] (-1,0)--(4,0);
        \node[above] at (-0.5,4){$q_{b}$};
        \node [right] at (4,0){$q_{a}$};
        \node[below] at (-0.15,0) {$1$};
        \node [below] at (0.55,0){$\cdots$};
        \node [below] at (1.25,0){$i$};
         \draw[->]   (-0.5,-0.5)--(-0.5,4);
         \draw (0.2,3.5)--(0.2,0.7);
         \draw (0.9,2.8)--(0.9,0.7);
         \draw (1.6,2.1)--(1.6,0.7);
         \draw (2.3,1.4)--(2.3,0.7);
         \draw (2.3,0.7)--(-0.5,0.7);
         \draw (2.3,1.4)--(-0.5,1.4);
         \draw (1.6,2.1)--(-0.5,2.1);
         \draw (0.9,2.8)--(-0.5,2.8);
         \draw (0.2,3.5)--(-0.5,3.5);
        \draw (-0.5,0.7)--(3,0.7);
        \draw (3,0)--(3,0.7);
        \draw (0.2,0)--(0.2,0.7);
        \draw (0.9,0)--(0.9,0.7);
         \draw (1.6,0)--(1.6,0.7);
          \draw (2.3,0)--(2.3,0.7);
          \draw (-0.15,0.35)--++(-0.7,-1);
          \node[left] at (-0.85,-0.65){$x$};
          \node [left] at (-0.5,0.35) {$1$};
          \node [left] at (-0.5,1.05){$\vdots$};
          \node [left] at (-0.5,1.75){$j$};
           \draw  (1.25,1.75)--++(0.9,0.9);
          \node[right] at (2.2,2.65) {$xq_{a}^{i-1}q_{b}^{j-1}$};
        \end{tikzpicture}
        }
\eea
Similar to the D2-case, each $\mathsf{A}^{-1}(\chi_{ab,x}(\Bbox))$ corresponds to the box of the Young diagram. The operator $\mathsf{X}_{ab}(x)$ defines the vacuum and physically gives the one-loop perturbative part (see section \ref{sec:D4qqandspikedinst}). In the algebraic context, it is called the \emph{highest weight} of the $qq$-character. Moreover, as will be shown in \ref{thm:D4qq-commute}, it uniquely determines the $qq$-character, we also call it the \emph{generating current} of the $qq$-character.

The $qq$-character is defined by adding the monomial terms in \eqref{eq:D4monomialterm} for all possible Young diagrams with specific coefficients. The coefficients are determined by the commutativity with the screening currents.
\begin{definition}
    We define the D4 $qq$-character for $A\in\six$ as
    \beq
        \mathsf{T}_{A}(x)=\sum_{\lambda\in\mathcal{P}}\widetilde{\mathcal{Z}}_{A}^{\D4}[\lambda]:\Lambda_{A,\lambda}(x):
    \eeq
    where and 
    \beq
       \Lambda_{A,\lambda}(x)={:\mathsf{X}_{A}(x)\prod_{\Abox\in\lambda}\mathsf{A}^{-1}(\chi_{A,x}(\Bbox)):}.
    \eeq
    Note that the coefficients $\widetilde{\mathcal{Z}}^{\D4}_{A}[\lambda]$ do not depend on the choice of $a\in\bar{A}$. 
    One may redefine the zero-modes of the root currents as $\mathsf{A}(x)\rightarrow \mathfrak{q}^{-1}\mathsf{A}(x)$ and obtain
\beq
    \mathsf{T}_{A}(x)=\sum_{\lambda\in\mathcal{P}}\mathfrak{q}^{|\lambda|}\widetilde{\mathcal{Z}}_{A}^{\D4}[\lambda]:\Lambda_{A,\lambda}(x):,\quad A\in\six.
\eeq
\end{definition}

\begin{theorem}\label{thm:D4qq-commute}
The D4 $qq$-characters commute with the screening charges associated with the transverse directions
\beq
    [\mathsf{T}_{A}(x),\mathscr{Q}_{a}(x')]=0,\quad \forall a\in\bar{A}.
\eeq
\end{theorem}
\begin{proof}
    Let us focus on $A=12$ and $a=4$ and derive the D4 $qq$-character. Other cases are obtained using the quadrality. Using the formulas in \eqref{eq:app-contractions}, we have 
    \bea\label{eq:D4-D2contraction}
        \Lambda_{12,\lambda}(x)\mathsf{S}_{4}(x')=\left[q_{3}^{-1}\frac{\mathscr{Y}^{12}_{\lambda,x}(q_{12}x')}{\mathscr{Y}^{12}_{\lambda,x}(q_{123}x')}\right]^{x'}_{-} :\Lambda_{12,\lambda}(x)\mathsf{S}_{4}(x'):,\\
        \mathsf{S}_{4}(x')\Lambda_{12,\lambda}(x)=\left[q_{3}^{-1}\frac{\mathscr{Y}^{12}_{\lambda,x}(q_{12}x')}{\mathscr{Y}^{12}_{\lambda,x}(q_{123}x')}\right]^{x'}_{+}:\Lambda_{12,\lambda}(x)\mathsf{S}_{4}(x'):
    \eea
    where $\left[\,f(x)\,\right]^{x}_{\pm}$ means expansions of $f(x)$ in ${x}^{\mp}$ respectively. Assume that the $qq$-character takes the form of 
    \beq
        \mathsf{T}_{12}(x)=\sum_{\lambda\in\mathcal{P}}F^{\D4}_{12}[\lambda]\Lambda_{12,\lambda}(x),\quad F_{12}[\emptyset]=1
    \eeq
    where $F^{\D4}_{12}[\lambda]$ are some coefficients, then the commutation relation is 
    \begin{align}
    \begin{split}
    &[\mathsf{T}_{12}(x),\mathsf{S}_{4}(x')]\\
    =&-q_{3}^{-1}\left[\sum_{\lambda}F^{\D4}_{12}[\lambda]\sum_{\Abox\in A(\lambda)}\delta\left(\frac{\chi_{12,x}(\Bbox)}{q_{4}^{-1}x'}\right)\underset{x'=q_{3}^{-1}\chi_{12,x}(\Abox)}{\Res}{x'}^{-1}\frac{\mathscr{Y}_{\lambda,x}^{12}(x')}{\mathscr{Y}_{\lambda,x}^{12}(q_{3}x')}:\Lambda_{12,\lambda}(x)\mathsf{S}_{4}(x'):\right.\\
    &+\left.\sum_{\lambda}F^{\D4}_{12}[\lambda]\sum_{\Abox\in R(\lambda)}\delta\left(\frac{\chi_{12,x}(\Bbox)}{x'}\right)\underset{x'=q_{34}^{-1}\chi_{12,x}(\Abox)}{\Res}{x'}^{-1}\frac{\mathscr{Y}_{\lambda,x}^{12}(x')}{\mathscr{Y}_{\lambda,x}^{12}(q_{3}x')}:\Lambda_{12,\lambda}(x)\mathsf{S}_{4}(x'):\right].
\end{split}
\end{align}
Shifting the second term as $\lambda=\lambda'+\Bbox$, it will be
\beq
    -q_{3}^{-1}\sum_{\lambda'}\sum_{\Abox\in A(\lambda')}F^{\D4}_{12}[\lambda'+\Bbox\,]\delta\left(\frac{\chi_{12,x}(\Bbox)}{x'}\right)\underset{x'=q_{34}^{-1}\chi_{12,x}(\Abox)}{\Res}{x'}^{-1}\frac{\mathscr{Y}_{\lambda'+\Abox,x}^{12}(x')}{\mathscr{Y}_{\lambda'+\Abox,x}^{12}(q_{3}x')}:\Lambda_{12,\lambda'}(x)\mathsf{A}^{-1}(\chi_{12,x}(\Bbox))\mathsf{S}_{4}(x').
\eeq
Using 
\beq
    {:\Lambda_{12,\lambda'}(x)\mathsf{A}^{-1}(\chi_{12,x}(\Bbox))\mathsf{S}_{4}(\chi_{12,x}(\Bbox)):}={:\Lambda_{12,\lambda'}(x)\mathsf{S}_{4}(q_{4}\chi_{12,x}(\Bbox)):}
\eeq
and imposing the condition (see the recursion formula in Thm. \ref{app:thm-D4recursion})
\beq
    \frac{F^{\D4}_{12}[\lambda+\Bbox]}{F^{\D4}_{12}[\lambda]}=-\frac{\underset{x'=q_{3}^{-1}\chi_{12,x}(\Abox)}{\Res}{x'}^{-1}\frac{\mathscr{Y}_{\lambda,x}^{12}(x')}{\mathscr{Y}_{\lambda,x}^{12}(q_{3}x')}}{\underset{x'=q_{34}^{-1}\chi_{12,x}(\Abox)}{\Res}{x'}^{-1}\frac{\mathscr{Y}_{\lambda+\Abox,x}^{12}(x')}{\mathscr{Y}_{\lambda+\Abox,x}^{12}(q_{3}x')}}=\frac{\widetilde{\mathcal{Z}}_{12}^{\D4}[\lambda+\Bbox]}{\widetilde{\mathcal{Z}}_{12}^{\D4}[\lambda]}
\eeq
we then obtain
\bea
\relax[\mathsf{T}_{12}(x),\mathsf{S}_{4}(x')]&=-q_{3}^{-1}\sum_{\lambda\in\mathcal{P}}\sum_{\Abox\in A(\lambda)}F^{\D4}_{12}[\lambda]\underset{x'=q_{3}^{-1}\chi_{12,x}(\Abox)}{\Res}{x'}^{-1}\frac{\mathscr{Y}_{\lambda,x}^{12}(x')}{\mathscr{Y}_{\lambda,x}^{12}(q_{3}x')}:\Lambda_{12,\lambda}(x)\mathsf{S}_{4}(q_{4}\chi_{12,x}(\Bbox)):\\
&\qquad\qquad \times \left(\delta\left(\frac{q_{4}\chi_{12,x}(\Bbox)}{x'}\right)-\delta\left(\frac{\chi_{12,x}(\Bbox)}{x'}\right)\right)
\eea
which is a total difference. Therefore, under the condition $F_{12}[\lambda]=\widetilde{\mathcal{Z}}_{12}^{\D4}[\lambda]$, the $qq$-character commutes with the screening charge $\mathscr{Q}_{4}(x')$:
\beq
    [\mathsf{T}_{12}(x),\mathscr{Q}_{4}(x')]=0.
\eeq
\end{proof}

\subsection{Spiked instantons and D4 \texorpdfstring{$qq$}{qq}-characters}\label{sec:D4qqandspikedinst}
Similar to the integral formula, the expanded version of the spiked instanton partition function can be expressed using the D4 $qq$-characters. 
\begin{lemma}\label{lemm:D4ope}
The operator product of $\{\Lambda_{A,\lambda}(x)\}_{A\in\six}$ is
\beq
    \Lambda_{B,\mu}(x')\Lambda_{A,\lambda}(x)=\mathcal{Z}^{\D4\text{-}\D4}_{\text{1-loop}}(x,A\,|\,x',B)\mathcal{Z}^{\D4\text{-}\D4}_{A|B}(x,\lambda\,|\,x',\mu):\Lambda_{B,\mu}(x')\Lambda_{A,\lambda}(x):.
\eeq
When $A=B$, it gives the vector and adjoint hypermultiplet contributions, while when $A\neq B$, it gives the bifundamental-like contributions connecting gauge theories defined on different subspaces. 
\end{lemma}
\begin{theorem}\label{thm:spiked-qq-BPSCFT}
The gauge origami partition function of the spiked instantons is written using the D4 $qq$-characters:
\bea
    \mathcal{Z}^{\D4}_{\text{1-loop}}\mathcal{Z}^{\D4}_{\text{inst.}}=\sum_{\underline{\vec{\lambda}}}\mathfrak{q}^{|\underline{\vec{\lambda}}|}\mathcal{Z}^{\D4}_{\text{1-loop}}\mathcal{Z}_{\text{spk.inst.}}^{\D4}[\underline{\vec{v}},\underline{\vec{\lambda}}]=\bra{0}\prod_{A\in\six}\prod_{\alpha=1}^{n_{A}}\mathsf{T}_{A}(v_{A,\alpha})\ket{0}
\eea    
\end{theorem}
Depending on the value of $A,B\in\six$, we get different instanton contributions:
\begin{itemize}
    \item $A=B$: instantons in $\widehat{A}_0$ quiver gauge theory
    \item $A\cap B \in\four$: folded instantons
    \item $B=\bar{A}$: crossed instantons
\end{itemize}
To summarize, we have the following table of BPS/CFT correspondence.
\begin{align}
    \renewcommand\arraystretch{1.2}{
    \begin{tabular}{c|c}\toprule
        BPS &  CFT\\
     \hline  5d $\mathcal{N}=1^{\ast}$ U(1) on $\mathbb{C}^{2}_{ab}\times \mathbb{S}^{1}$ ($\D4_{ab}\times 1$) & $\bra{0}\mathsf{T}_{ab}(v)\ket{0}$ \\
      5d $\mathcal{N}=1^{\ast}$ U($n_{ab}$) on $\mathbb{C}^{2}_{ab}\times \mathbb{S}^{1}$ ($\D4_{ab}\times n_{ab}$)   & $\bra{0}\mathsf{T}_{ab}(v_{n_{ab}})\cdots\mathsf{T}_{ab}(v_{2})\mathsf{T}_{ab}(v_{1})\ket{0}$\\
      crossed instanton: D4$_{12}$-D4$_{34}$-D0 & $\bra{0}\mathsf{T}_{12}(v)\mathsf{T}_{34}(v')\ket{0}$\\
      folded instanton: D4$_{12}$-D4$_{13}$-$\D0$ & $\bra{0}\mathsf{T}_{12}(v)\mathsf{T}_{13}(v')\ket{0}$\\
      gauge origami of spiked instantons & $\bra{0}\prod\limits_{A\in\six}\prod\limits_{\alpha=1}^{n_{A}}\mathsf{T}_{A}(v_{A,\alpha})\ket{0}$
      \\ \bottomrule
    \end{tabular}}
\end{align}

\subsection{Fusion of screening charges}\label{sec:spikedandscreening}

A different way to derive the D4 $qq$-characters is to take infinite products of screening charges which we call the \textit{fusion process}.
\begin{theorem}\label{thm:D2toD4fusion}
    The D4 $qq$-characters are infinite products of D2 $qq$-characters (screening charges):
    \beq
        \mathsf{T}_{ab}(x)\simeq  \overrightarrow{\prod_{i=1}^{\infty}}\mathscr{Q}_{b}(xq_{a}^{i-1}),\quad a\neq b\in\four,
    \eeq
    where the product is $\overrightarrow{\prod\limits_{i=1}^{\infty}}f(x_{i})=f(x_{1})f(x_{2})\cdots$.
\end{theorem}
\begin{proof}
We only give a sketch of the proof (see \cite{Kimura:2015rgi} and \cite[eq.~(6.2.23)--(6.2.27)]{Kimura:2020jxl} for a similar discussion). The vertex operator $\Lambda_{ab,\lambda}(x)$ of the D4 $qq$-character satisfies
\beq
    \Lambda_{ab,\lambda}(x)\sim {:\prod_{i=1}^{\infty}\mathsf{S}_{b}(xq_{a}^{i-1}q_{b}^{\lambda_{i}}):}
\eeq
where the equality is up to zero modes. By direct computation, one can show that up to one-loop perturbative factors,\footnote{When we say one-loop perturbative factors, we are meaning the operator product coming from the highest weight vertex operators. For this case, the operator product coming from $\mathsf{S}_{b}(xq_{a}^{i-1})$. Strictly speaking, the $\mathsf{S}_{a}(x)$ operators are not highest weights in the usual sense, but we abuse the terminology and still call it the highest weight.} we have\footnote{Strictly speaking, we have to be careful of the zero-modes appearing on both hand sides. Moreover, the infinite product of the screening currents needs to be regularized properly.}
\beq
    \mathsf{T}_{ab}(x)\simeq\sum_{\lambda\in\mathcal{P}}\overrightarrow{\prod_{i=1}^{\infty}}\mathsf{S}_{b}(xq_{a}^{i-1}q_{b}^{\lambda_{i}}).\label{eq:D4screeningrep}
\eeq
Using the property that for $i\leq j$ and $\lambda_{i}<\lambda_{j}$
\beq
    \mathsf{S}_{b}(xq_{a}^{i-1}q_{b}^{\lambda_{i}})\mathsf{S}_{b}(xq_{a}^{j-1}q_{b}^{\lambda_{j}})=0,
\eeq
we obtain
\beq
    \mathscr{Q}_{b}(x)\mathscr{Q}_{b}(q_{a}x)=\sum_{k\in\mathbb{Z}}\mathsf{S}_{b}(xq_{b}^{k})\sum_{l\in\mathbb{Z}}\mathsf{S}_{b}(xq_{a}q_{b}^{l})=\sum_{k\geq l}\mathsf{S}_{b}(xq_{b}^{k})\mathsf{S}_{b}(xq_{a}q_{b}^{l}).
\eeq
By computing the nontrivial coefficients appearing after taking the contractions of the screening currents, one obtains the statement.
\end{proof}
\begin{corollary}
    The total partition function is written using infinite products of screening charges
    \beq
        \mathcal{Z}^{\D4}_{\text{1-loop}}\mathcal{Z}^{\D4}_{\text{inst.}}=\sum_{\underline{\vec{\lambda}}}\mathfrak{q}^{|\underline{\vec{\lambda}}|}\mathcal{Z}^{\D4}_{\text{1-loop}}\mathcal{Z}_{\text{spk.inst.}}^{\D4}[\underline{\vec{v}},\underline{\vec{\lambda}}]\simeq \bra{0}\prod_{A\in\six}\prod_{\alpha=1}^{n_{A}}\overrightarrow{\prod_{i=1}^{\infty}}\mathscr{Q}_{\bar{s}(A)}(v_{A,\alpha}q_{s(A)}^{i-1})\ket{0}.
    \eeq
\end{corollary}

\subsection{Quadratic relations of crossed instantons}\label{sec:D4quadraticrelation}
The $qq$-characters $\{\mathsf{T}_{A}(x)\}_{A\in\six}$ are expected to generate a larger algebra than the affine quiver W-algebra. As usual deformed W-algebras, studying the quadratic relations of the generators is another way to understand the algebraic structure. Since deriving the complete quadratic relations of the $\{\mathsf{T}_{A}(x)\}_{A\in\six}$ is beyond the scope of this paper, we just give a part of the quadratic relations which can be derived easily.
\begin{theorem}
    The D4 $qq$-characters $\mathsf{T}_{A}(x)$ and $\mathsf{T}_{\bar{A}}(x)$ for $A\in\six$ (anti-)commute with each other up to trivial zero-modes factors:
    \begin{equation}
       x\mathsf{T}_{A}(x)\mathsf{T}_{\bar{A}}(x')+q_{A}x'\mathsf{T}_{\bar{A}}(x')\mathsf{T}_{A}(x)=0.
    \end{equation}
\end{theorem}
\begin{proof}
Let us consider the commutation relations between $\mathsf{T}_{12}(x)$ and $\mathsf{T}_{34}(x')$:
\begin{align}
    \mathsf{T}_{12}(x)&=\sum_{\lambda}F^{\D4}_{12}[\lambda]:\Lambda_{12,\lambda}(x):,\\
    \mathsf{T}_{34}(x')&=\sum_{\mu}F^{\D4}_{34}[\mu]:\Lambda_{34,\mu}(x'):.
\end{align}
The contraction formulas give 
\bea
    \Lambda_{12,\lambda}(x)\Lambda_{34,\mu}(x')&=\left[\mathcal{Z}_{34\,|\,12}^{\text{tot}}(x',\mu\,|\,x,\lambda)\right]_{|x'/x|\ll1}:\Lambda_{12,\lambda}(x)\Lambda_{34,\mu}(x'):,\\
    \Lambda_{34,\mu}(x')\Lambda_{12,\lambda}(x)&=\left[\mathcal{Z}_{12\,|\,34}^{\text{tot}}(x,\lambda\,|\,x',\mu)\right]_{|x/x'|\ll1}:\Lambda_{12,\lambda}(x)\Lambda_{34,\mu}(x'):,
\eea
where
\begin{equation}
\mathcal{Z}_{\bar{A}|A}^{\text{tot.}}(x',\mu|x,\lambda)=\mathcal{Z}_{\text{1-loop}}^{\D4\tbar\D4}(x',\bar{A}|x,A)\mathcal{Z}_{\bar{A}|A}^{\D4\tbar\D4}(x',\mu|x,\lambda).
\end{equation}
Noting that the one-loop factors of the crossed instantons are rational functions
\beq
    \mathcal{Z}_{\text{1-loop}}^{\D4\tbar\D4}(x',\bar{X}|x,X)=\left(1-q_{A}x'/x\right),
\eeq
we then obtain
\begin{equation}
\begin{split}
    &\mathsf{T}_{12}(x)\mathsf{T}_{34}(x')-\left(-q_{12}\frac{x'}{x}\right)\mathsf{T}_{34}(x')\mathsf{T}_{12}(x)\\
    =&\sum_{\lambda,\mu}F^{\D4}_{12}[\lambda]F^{\D4}_{34}[\mu]\left\{\sum_{\substack{\AboxF\in A(\mu)\\\Abox\in R(\lambda)}}\underset{\substack{\chi_{12,x}(\Abox)\\=\chi_{34,x'}(\AboxF)}}{\Res}\left(\frac{x}{x'}\right)^{-1}\mathcal{Z}_{34\,|\,12}^{\text{tot}}(x',\mu\,|\,x,\lambda)\delta\left(\frac{\chi_{12,x}(\Bbox)}{\chi_{34,x'}(\BboxF)}\right):\Lambda_{12,\lambda}(x)\Lambda_{34,\mu}(x'):\right.\\
    &\left.+\sum_{\substack{\AboxF\in R(\mu)\\\Abox\in A(\lambda)}}\underset{\substack{\chi_{12,x}(\Abox)\\=\chi_{34,x'}(\AboxF)}}{\Res}\left(\frac{x}{x'}\right)^{-1}\mathcal{Z}_{34\,|\,12}^{\text{tot}}(x',\mu\,|\,x,\lambda)\delta\left(\frac{\chi_{12,x}(\Bbox)}{\chi_{34,x'}(\BboxF)}\right):\Lambda_{12,\lambda}(x)\Lambda_{34,\mu}(x'):\right\}.
\end{split}
\end{equation}
Similar to the proof in Thm.~\ref{thm:D4qq-commute}, we redefine the sum of the first term as $\lambda\rightarrow \lambda=\lambda'+\Bbox$, $\mu\rightarrow \mu=\mu'+\BboxF$. After this, the first term will be 
\begin{equation}
\begin{split}
    &\sum_{\lambda',\mu'}\sum_{\substack{\AboxF\in R(\mu')\\\Abox\in A(\lambda')}} F^{\D4}_{12}[\lambda'+\Bbox]F^{\D4}_{34}[\mu'-\BboxF]\underset{\substack{\chi_{12,x}(\Abox)\\=\chi_{34,x'}(\AboxF)}}{\Res}\left(\frac{x}{x'}\right)^{-1}\mathcal{Z}_{34\,|\,12}^{\text{tot}}(x',\mu'-\BboxF\,|\,x,\lambda'+\Bbox)\\
    &\qquad\times\delta\left(\frac{\chi_{12,x}(\Bbox)}{\chi_{34,x'}(\BboxF)}\right):\Lambda_{12,\lambda'+\Abox}(x)\Lambda_{34,\mu'-\AboxF}(x'):.
\end{split}
\end{equation}
Using 
\begin{equation}
    \frac{F^{\D4}_{12}[\lambda+\Bbox]F^{\D4}_{34}[\mu-\BboxF]}{F^{\D4}_{12}[\lambda]F^{\D4}_{34}[\mu]}=-\frac{\underset{\substack{\chi_{12,x}(\Abox)\\=\chi_{34,x'}(\AboxF)}}{\Res}\left(\frac{x}{x'}\right)^{-1}\mathcal{Z}_{34\,|\,12}^{\text{tot}}(x',\mu\,|\,x,\lambda)}{\underset{\substack{\chi_{12,x}(\Abox)\\=\chi_{34,x'}(\AboxF)}}{\Res}\left(\frac{x}{x'}\right)^{-1}\mathcal{Z}_{34\,|\,12}^{\text{tot}}(x',\mu-\BboxF\,|\,x,\lambda+\Bbox)}.
\end{equation}
and
\begin{equation}
    \delta\left(\frac{\chi_{12,x}(\Bbox)}{\chi_{34,x'}(\BboxF)}\right):\Lambda_{12,\lambda+\Abox}(x)\Lambda_{34,\mu-\AboxF}(x'):=\delta\left(\frac{\chi_{12,x}(\Bbox)}{\chi_{34,x'}(\BboxF)}\right):\Lambda_{12,\lambda}(x)\Lambda_{34,\mu}(x'):
\end{equation}
we obtain the claim.
\end{proof}

\begin{remark}
    The extra factors in front of the $qq$-characters come from the operator product of the zero modes $\mathsf{x}_{A,0}(x)$. We may modify the zero modes so that they obey \eqref{eq:D4D4zero} (see also footnote \ref{note:D4footnote}) and then get an exact commuting relation $[\mathsf{T}_{A}(x),\mathsf{T}_{\bar{A}}(x')]=0$. Instead of doing that, we simply relax the commutativity condition and say that operators $\mathcal{O}(x)$ and $\mathcal{O}'(x')$ commute when they satisfy 
    \begin{align}
        \mathcal{O}(x)-f(x,x')\mathcal{O'}(x')=0\label{eq:weakcommute}
    \end{align}
    where $f(x,x')$ are zero modes factors.
\end{remark}

\section{D6 \texorpdfstring{$qq$}{qq}-characters}\label{sec:D6qqcharacter}
We introduce the D6 $qq$-characters by using the commutativity with the screening charges in section~\ref{sec:D6qqchdef}. To be concrete, we show that monomial terms of the $qq$-character with the highest weight $\mathsf{W}_{\bar{a}}(x)$ are classified by plane partitions and that it commutes with the screening charge $\mathscr{Q}_{a}(x')$. We then show that the D6 $qq$-characters reproduce the tetrahedron instanton partition function in section~\ref{sec:tetrainstD6qq}. The D6 $qq$-characters can be obtained by fusion of an infinite number of D4 $qq$-characters (see section~\ref{sec:D4fusiontoD6}).



\subsection{D6 \texorpdfstring{$qq$}{qq}-characters and plane partition}\label{sec:D6qqchdef}
Following the construction of the D4 $qq$-characters and the affine quiver W-algebra, let us construct the D6 $qq$-characters. The D6 $qq$-characters are $qq$-characters whose generating currents are the D6 vertex operators $\mathsf{W}_{\bar{a}}(x)\,(a\in\four)$ in \eqref{eq:D6op}. Similar to the D4 case, each term of the D6 $qq$-character is decomposed into the vertex operator and the coefficient parts. The operator part is determined by the iWeyl reflection of the operator part defined as the following.
\begin{definition}
    The iWeyl reflection of the D6 vertex operator $\mathsf{W}_{\bar{a}}(x)\,(a\in\four)$ is
    \beq
        \mathsf{W}_{\bar{a}}(x)\longrightarrow {:\mathsf{W}_{\bar{a}}(x)\mathsf{A}^{-1}(x):},\quad a\in\four.
    \eeq
    Using the property in \eqref{eq:oprelationwithD0-D6}, the root current is rewritten in the D6 vertex operators as
\beq
    \mathsf{W}_{\bar{a}}(x)\longrightarrow {:\mathsf{a}_{0}(x)^{-1}\mathsf{W}_{\bar{a}}(q_{a}^{-1}x)\frac{\prod\limits_{i\in\four\setminus \{a\}}\mathsf{W}_{\bar{a}}(q_{i}x)}{\prod\limits_{i\in\four\setminus \{a\}}\mathsf{W}_{\bar{a}}(q_{i}^{-1}q_{a}^{-1}x)}:}
\eeq
\end{definition}
The operator part is obtained by changing the root currents to the D6-operators and applying the iWeyl reflection to the numerators recursively. After iWeyl reflections, the operators are classified by plane partitions. Let us study some terms after applying the iWeyl reflection. The operators obtained after $n$-times of iWeyl reflections are called operators at level $n$. We focus on the case $\mathsf{W}_{\bar{4}}(x)$ and then have
\beq
    \mathsf{W}_{\bar{4}}(x)\longrightarrow {:\mathsf{a}_{0}(x)^{-1}\frac{\mathsf{W}_{\bar{4}}(q_{123}x)\prod_{i=1}^{3}\mathsf{W}_{\bar{4}}(q_{i}x)}{\prod_{1\leq i<j\leq 3}\mathsf{W}_{\bar{4}}(q_{i}q_{j}x)}:}.
\eeq
\begin{itemize}
    \item Level 0: We only have one operator 
\beq
    \mathsf{W}_{\bar{4}}(x).
\eeq
We can associate this operator with a vacuum configuration of the plane partition in the space (1,2,3) where there is no box:
\bea
    \mathsf{W}_{\bar{4}}(x)\quad\longleftrightarrow\quad  \adjustbox{valign=c}{\begin{tikzpicture}[scale=0.5]
    \draw[->] (0,0)--(-30:2); 
    \draw[->] (0,0)--(210:2);
    \draw[->] (0,0)--(90:2);
    \node at (210:2.4){$1$};
    \node at (-30:2.4){$2$};
    \node at (90:2.4){$3$};
\end{tikzpicture}}
\eea
The spectral parameter $x$ corresponds to the $q$-coordinates in the origin. The operator $\mathsf{W}_{\bar{4}}(x)$ represents the addable box of this plane partition configuration which is the box in the origin.
\item Level 1: After applying the iWeyl reflection once, we have 
\beq
    {:\mathsf{W}_{\bar{4}}(x)\mathsf{A}^{-1}(x):}={:\mathsf{a}_{0}(x)^{-1}\frac{\mathsf{W}_{\bar{4}}(q_{123}x)\prod_{i=1}^{3}\mathsf{W}_{\bar{4}}(q_{i}x)}{\prod_{i<j\leq 3}\mathsf{W}_{\bar{4}}(q_{i}q_{j}x)}:}.
\eeq
The level 1 current can be described as 
\bea
:\mathsf{a}_{0}(x)^{-1}\frac{\textcolor{blue}{\mathsf{W}_{\bar{4}}(q_{4}^{-1}x)}\textcolor{red}{\prod_{i=1}^{3}\mathsf{W}_{\bar{4}}(q_{i}x)}}{\prod_{i<j\leq 3}\mathsf{W}_{\bar{4}}(q_{i}q_{j}x)}:\quad \longleftrightarrow \quad
\adjustbox{valign=c}{
\begin{tikzpicture}[scale=0.5]
    \draw[->] (0,0)--(-30:2); 
    \draw[->] (0,0)--(210:2);
    \draw[->] (0,0)--(90:2);
    \planepartition{{1}} 
    \node at (210:2.4){$1$};
    \node at (-30:2.4){$2$};
    \node at (90:2.4){$3$};
\end{tikzpicture}}
\eea
An observation is that the red terms $\mathsf{W}_{\bar{4}}(q_{i}x)\,(i=1,2,3)$ correspond to the addable boxes of this plane partition configuration. The variables $q_{i}x\,(i=1,2,3)$ correspond to the $q$-coordinates of the addable boxes. The blue term $\mathsf{W}_{\bar{4}}(q_{4}^{-1}x)$ corresponds to the removable box of the configuration, which is the box in the origin with coordinate $x$.

\item Level 2: Since we have three numerators in the level 1 operator, we will have three possible level 2 operators depending on which term we do the iWeyl reflection.

\begin{enumerate}
    \item $\pi_{1,1}=1,\,\,\pi_{2,1}=1$
    \bea
        {:\mathsf{W}_{\bar{4}}(x)\mathsf{A}^{-1}(x)\mathsf{A}^{-1}(q_{1}x):}\propto \frac{\textcolor{red}{(2)(3)(1^{2})}\textcolor{blue}{(14^{-1})}}{(1^{2}2)(1^{2}3)}\quad \longleftrightarrow \quad \adjustbox{valign=c}{
\begin{tikzpicture}[scale=0.5]
    \draw[->] (0,0)--(-30:2); 
    \draw[->] (0,0)--(210:2.8);
    \draw[->] (0,0)--(90:2);
    \planepartition{{1},{1}} 
    \node at (210:3.1){$1$};
    \node at (-30:2.4){$2$};
    \node at (90:2.4){$3$};
\end{tikzpicture}}
     \eea
    where we simply denote $\mathsf{W}_{\bar{4}}(q_{1}^{a}q_{2}^{b}q_{3}^{c}q_{4}^{d}x)$ as $(1^{a}2^{b}3^{c}4^{d})$. The $q$-coordinates of the addable boxes of this configuration are $q_{1}^{2}x, q_{2}x,q_{3}x$ and correspond to the red terms in the numerator. The $q$-coordinate of the removable box of this configuration is $xq_{1}$ and correspond to the blue term with variable $q_{4}^{-1}q_{1}x$.
    
    \item $\pi_{1,1}=1,\pi_{1,2}=1$
    \bea
    {:\mathsf{W}_{\bar{4}}(x)\mathsf{A}^{-1}(x)\mathsf{A}^{-1}(q_{2}x):}\propto \frac{\textcolor{red}{(1)(3)(2^{2})}\textcolor{blue}{(12^{2}3)}}{(12^{2})(2^{2}3)}\quad \longleftrightarrow \quad \adjustbox{valign=c}{
\begin{tikzpicture}[scale=0.5]
    \draw[->] (0,0)--(-30:2.8); 
    \draw[->] (0,0)--(210:2);
    \draw[->] (0,0)--(90:2);
    \planepartition{{1,1}} 
    \node at (210:2.4){$1$};
    \node at (-30:3.1){$2$};
    \node at (90:2.4){$3$};
\end{tikzpicture}}
    \eea
    Similarly, the red terms correspond to the addable boxes and the blue terms correspond to the removable boxes of the plane partition configuration.
    \item $\pi_{1,1}=2$
    \bea
    {:\mathsf{W}_{\bar{4}}(x)\mathsf{A}^{-1}(x)\mathsf{A}^{-1}(q_{3}x):}\propto \frac{\textcolor{red}{(1)(2)(3^{2})}\textcolor{blue}{(123^{2})}}{(13^{2})(23^{2})}\quad \longleftrightarrow \quad \adjustbox{valign=c}{
\begin{tikzpicture}[scale=0.5]
    \draw[->] (0,0)--(-30:2); 
    \draw[->] (0,0)--(210:2);
    \draw[->] (0,0)--(90:2.6);
    \planepartition{{2}} 
    \node at (210:2.4){$1$};
    \node at (-30:2.4){$2$};
    \node at (90:3.1){$3$};
\end{tikzpicture}}
    \eea
    Again, the red terms correspond to the addable boxes and the blue terms correspond to the removable boxes of the plane partition configuration.
\end{enumerate}
\end{itemize}
We can do this procedure recursively and then obtain the following statement.
\begin{lemma}\label{lemma:D6_iWeyl_ref}
The operators generated from the iWeyl reflection starting from $\mathsf{W}_{\bar{a}}(x)$ are classified by plane partitions:
\begin{equation}
    \Lambda_{\bar{a},\pi}(x)\coloneqq {:\mathsf{W}_{\bar{a}}(x)\prod_{\scube\in\pi}\mathsf{A}^{-1}(\chi_{\bar{a},x}(\cube)):},\quad a\in\four.
\end{equation}
Converting the root currents into the D6-operators, we have 
\begin{equation}
    \mathsf{W}_{\bar{a}}(x)\prod\limits_{\scube\in\pi}\mathsf{A}^{-1}(\chi_{\overbar{a},x}(\cube))\propto \frac{\prod\limits_{\scube\in A(\pi)}\mathsf{W}_{\bar{a}}(\chi_{\overbar{a},x}(\cube))\prod\limits_{\scube\in R(\pi)}\mathsf{W}_{\bar{a}}(q_{a}^{-1}\chi_{\overbar{a},x}(\cube))}{\#}
\end{equation}
where we have some extra denominators determined recursively.\footnotemark
\end{lemma}
\footnotetext{The explicit form of the right-hand side is related to the shell formula of the plane partition \cite{Feigin2011plane}. In this note, the information of the denominator is not necessary so we will not write the explicit form. }  

The D6 $qq$-character is a sum of the vertex operators $\Lambda_{\bar{a},\pi}(x)$ with some specific coefficients:
\beq
    \mathsf{T}_{\bar{a}}(x)=\sum_{\pi\in\mathcal{PP}}F^{\D6}_{\bar{a}}[\pi]\Lambda_{\bar{a},\pi}(x).
\eeq
We impose the condition that this $qq$-character commutes with the screening charge $\mathscr{Q}_{a}(x)$. After imposing this condition, the coefficient is determined uniquely.
\begin{definition}\label{def:D6_qq-ch}
    We define the D6 $qq$-character for $a\in\four$ as 
    \beq
        \mathsf{T}_{\bar{a}}(x)=\sum_{\pi\in\mathcal{PP}}\widetilde{\mathcal{Z}}^{\D6}_{\bar{a}}[\pi]\Lambda_{\bar{a},\pi}(x),\quad a\in\four,
    \eeq
    where the coefficients $\widetilde{\mathcal{Z}}^{\D6}_{\bar{a}}[\pi]$ are identified with the $\U(1)$ partition function of the 7d gauge theory on $\mathbb{C}^{3}_{\bar{a}}\times \mathbb{S}^{1}$ in \eqref{eq:D6tetinst_partfunct}. Rescaling the zero-modes $\mathsf{A}(x)\rightarrow \mathfrak{q}^{-1}\mathsf{A}(x)$, we have 
    \beq
        \mathsf{T}_{\bar{a}}(x)=\sum_{\pi\in\mathcal{PP}}\mathfrak{q}^{|\pi|}\widetilde{\mathcal{Z}}^{\D6}_{\bar{a}}[\pi]\Lambda_{\bar{a},\pi}(x).
    \eeq
    
\end{definition}

\begin{theorem}\label{thm:D6qq-commute}
    The D6 $qq$-characters $\mathsf{T}_{\bar{a}}(x)\,(a\in\four)$ commutes with the screening charge $\mathscr{Q}_{a}(x)$:
    \beq
       \relax [\mathsf{T}_{\bar{a}}(x),\mathscr{Q}_{a}(x')]=0.
    \eeq
\end{theorem}
\begin{proof}
Let us focus on $\mathsf{T}_{\bar{4}}(x)$ and see how the commutativity appears. Using the formulas in Thm. \ref{eq:app-contractions} and the property in \eqref{eq:oprelation1}, we have
\begin{align}
\begin{split}
    [\mathsf{W}_{\bar{4}}(x),\mathsf{S}_{4}(x')]&=q_{4}x\delta\left(q_{4}x/x'\right):\mathsf{W}_{\bar{4}}(x)\mathsf{S}_{4}(q_{4}x):\\
    [{:\mathsf{W}_{\bar{4}}(x)\mathsf{A}^{-1}(x):},\mathsf{S}_{4}(x')]&=x\delta\left(x'/x\right)\frac{\prod_{i=1}^{3}(1-q_{i})}{\prod_{1\leq i<j\leq 3}(1-q_{i}q_{j})}:\mathsf{W}_{123}(x)\mathsf{S}_{4}(q_{4}x):+\cdots,
    \end{split}
\end{align}
where for the second term, we only extracted the pole coming from $x'=x$. Thus, for the pole coming from $\mathsf{W}_{\bar{4}}(x)$ to disappear up to a total difference, we need the combination
\beq
    \mathsf{W}_{\bar{4}}(x)-q_{4}\frac{\prod_{1\leq i<j\leq 3}(1-q_{i}q_{j})}{\prod_{i=1}^{3}(1-q_{i})}:\mathsf{W}_{\bar{4}}(x)\mathsf{A}^{-1}(x):
\eeq
where the coefficient is $\widetilde{\mathcal{Z}}_{\bar{4}}^{\D6}[\,\cube\,]$. Generally, using 
\begin{align}
\begin{split}
  \Lambda_{\bar{a},\pi}(x)\mathsf{S}_{4}(x') &=-q_{4}x\left[\mathscr{W}^{\overbar{4}}_{\pi,x}(q_{4}^{-1}x')^{-1}\right]^{x'}_{-}:\Lambda_{\bar{a},\pi}(x)\mathsf{S}_{4}(x'):,\\
  \mathsf{S}_{4}(x')\Lambda_{\bar{a},\pi}(x)&=-q_{4}x\left[\mathscr{W}_{\pi,x}^{\overbar{4}}(q_{4}^{-1}x')^{-1}\right]^{x'}_{+}:\Lambda_{\bar{a},\pi}(x)\mathsf{S}_{4}(x'):
\end{split}
\end{align}
and the property in \eqref{eq:D6Nekrasov-shell}, we obtain
\begin{align}
\begin{split}
    &[\mathsf{T}_{\bar{4}}(x),\mathsf{S}_{4}(x')]\\
=&q_{4}x\sum_{\pi\in\mathcal{PP}}\widetilde{\mathcal{Z}}_{\bar{4}}^{\text{D6}}[\pi]\left(\sum_{\scube\,\in A(\pi)}\underset{x'=q_{4}\chi_{\overbar{4},x}(\scube)}{\Res}{x'}^{-1}\mathscr{W}^{\overbar{4}}_{\pi,x}(q_{4}^{-1}x')^{-1}\delta\left(\frac{x'}{q_{4}\chi_{\overbar{4},x}(\cube)}\right):\Lambda_{\bar{4},\pi}(x)\mathsf{S}_{4}(q_{4}\chi_{\overbar{4},x}(\cube)):\right.\\
    &\qquad\left.+\sum_{\scube\in R(\pi)}\underset{x'=\chi_{\overbar{4},x}(\scube)}{\Res}{x'}^{-1}\mathscr{W}^{\overbar{4}}_{\pi,x}(q_{4}^{-1}x')^{-1}\delta\left(\frac{x'}{\chi_{\overbar{4},x}(\cube)}\right):\Lambda_{\bar{4},\pi}(x)\mathsf{S}_{4}(\chi_{\overbar{4},x}(\cube)):\right).
\end{split}
\end{align}
Shifting the second term as $\pi'=\pi-\scube$, the second term will be rewritten as
\bea
\sum_{\pi'\in\mathcal{PP}}\widetilde{\mathcal{Z}}_{\bar{4}}^{\text{D6}}[\pi'+\cube]\sum_{\scube\in A(\pi')}\underset{x'=\chi_{\overbar{4},x}(\scube)}{\Res}{x'}^{-1}\mathscr{W}^{\overbar{4}}_{\pi'+\scube,x}(q_{4}^{-1}x')^{-1}\delta\left(\frac{x'}{\chi_{\overbar{4},x}(\cube)}\right):\Lambda_{\bar{4},\pi'+\scube}(x)\mathsf{S}_{4}(\chi_{\overbar{4},x}(\cube)):.
\eea
Using Thm.~\ref{eq:app-thm-D6U1recursionformula}, we have
\begin{align}
\begin{split}
    {:\Lambda_{\bar{4},\pi'+\scube}(x)\mathsf{S}_{4}(\chi_{\overbar{4},x}(\cube)):}={:\Lambda_{\bar{4},\pi'}(x)\mathsf{S}_{4}(q_{4}\chi_{\bar{4},x}(\cube)):},\\
    \frac{\widetilde{\mathcal{Z}}^{\D6}_{\bar{4}}[\pi+\cube]}{\widetilde{\mathcal{Z}}^{\D6}_{\bar{4}}[\pi]}=-\frac{\underset{x'=q_{4}\chi_{\overbar{4},x}(\scube)}{\Res}{x'}^{-1}\mathscr{W}^{\overbar{4}}_{\pi,x}(q_{4}^{-1}x')^{-1}}{\underset{x'=\chi_{\overbar{4},x}(\scube)}{\Res}{x'}^{-1}\mathscr{W}^{\overbar{4}}_{\pi+\scube,x}(q_{4}^{-1}x')^{-1}}
\end{split}
\end{align}
and then we obtain
\begin{align}
\begin{split}
    &[\mathsf{T}_{123}(x),\mathsf{S}_{4}(x')]\\
    =&q_{4}x\sum_{\pi\in\mathcal{PP}}\widetilde{\mathcal{Z}}_{\bar{4}}^{\text{D6}}[\pi]\sum_{\scube\,\in A(\pi)}\underset{x'=q_{4}\chi_{\overbar{4},x}(\scube)}{\Res}{x'}^{-1}\mathscr{W}^{\overbar{4}}_{\pi,x}(q_{4}^{-1}x')^{-1}:\mathsf{W}_{\bar{4}}(x)\prod_{\scube\in\pi}\mathsf{A}^{-1}(\chi_{\overbar{4},x}(\cube))\mathsf{S}_{4}(q_{4}\chi_{\overbar{4},x}(\cube))\\
    &\times \left(\delta\left(\frac{x'}{\chi_{\overbar{4},x}(\cube)}\right)-\delta\left(\frac{x'}{q_{4}\chi_{\overbar{4},x}(\cube)}\right)\right)
\end{split}
\end{align}
which gives the claim.
\end{proof}
We have shown here that the D6 $qq$-character $\mathsf{T}_{\bar{a}}(x)$ is uniquely determined by the commutativity with the screening charge $\mathscr{Q}_{a}(x)$ and that the coefficient factor is the $\U(1)$ partition function of the 7d theory on $\mathbb{C}^{3}_{\bar{a}}\times \mathbb{S}^{1}$. Compared with the magnificent four system where we need to take care of the sign problem~\cite{Nekrasov:2017cih,Nekrasov:2018xsb,Nekrasov:2023nai}, our discussion here gives an algebraic proof showing that there will be no sign problem for the 7d case. This is indeed compatible with a mathematical proof given in~\cite{Fasola:2023ypx}.

Thm.~\ref{thm:D2qq-commute}, \ref{thm:D4qq-commute}, \ref{thm:D6qq-commute} are summarized in the following theorem.
\begin{theorem}\label{thm:tetrascreening}
Let $\mathscr{T}$ be the tetrahedron corresponding to the $\mathbb{C}^{4}$ geometry (see Figure \ref{fig:complex}). We denote the set of vertices, edges, and faces of $\mathscr{T}$ as $\mathsf{v}=\{a\mid a\in\four\}$, $\mathsf{e}=\{A\mid A\in\six\}$, $\mathsf{f}=\{\bar{a}\mid a\in\four\}$ respectively. We also introduce the union of them as $\mathscr{S}=\mathsf{v}\cup\mathsf{e}\cup\mathsf{f}$. For each element of $i\in\mathscr{S}$, we can associate a $qq$-character. If $i\in\mathsf{v}$, we associate the D2 $qq$-character (screening charge), if $i\in\mathsf{e}$, we associate the D4 $qq$-character, if $i\in\mathsf{f}$, we associate the D6 $qq$-character. The $qq$-character associated with the elements $i,j\in\mathscr{T}$ commute with each other up to trivial zero modes (see \eqref{eq:weakcommute}) when $i$ and $j$ do not intersect in $\mathscr{T}$:
\beq
    \mathsf{T}_{i}(x)\mathsf{T}_{j}(x')-f_{ij}(x,x')\mathsf{T}_{j}(x')\mathsf{T}_{i}(x)=0 \quad \Longleftrightarrow \quad i \cap j =\emptyset,
\eeq
where $f_{ij}(x,x')$ are zero modes.
\vspace{-0.5cm}\begin{align*}\adjustbox{valign=c}{\includegraphics[width=4.5cm]{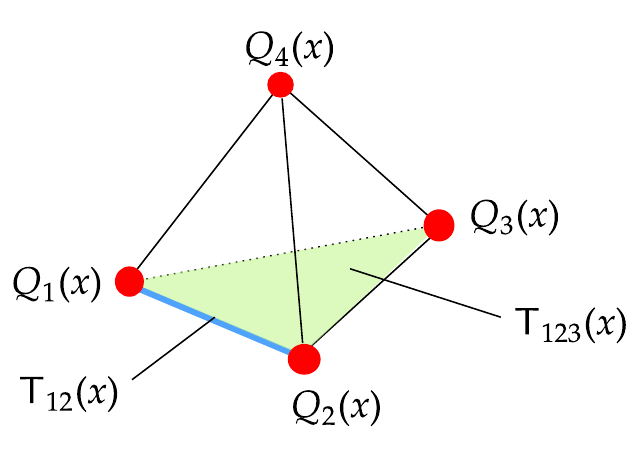}}\end{align*}
\end{theorem}
Choosing the screening charge $\mathscr{Q}_{4}(x)$, Thm.~\ref{thm:tetrascreening} claims that the kernel of $\mathscr{Q}_{4}(x)$ is generated by $\mathscr{Q}_{1,2,3}(x)$, $\mathsf{T}_{12,23,13}(x)$, $\mathsf{T}_{123}(x)$. We expect that this will give a larger algebra compared with the affine quiver W-algebra given by the kernel of two screening charges. We leave a detailed analysis of this for future work.

\subsection{Tetrahedron instantons and D6 \texorpdfstring{$qq$}{qq}-characters}\label{sec:tetrainstD6qq}
The expanded version of the tetrahedron instanton partition function can be expressed using the D6 $qq$-characters.
\begin{lemma}\label{lem:D6operatorproduct}
    The operator product of $\{\Lambda_{\bar{a},\pi}(x)\}_{a\in\four}$ is 
    \bea
    \Lambda_{\bar{b},\pi^{(2)}}(x_{2})\Lambda_{\bar{a},\pi^{(1)}}(x_{1})&=\mathcal{Z}^{\D6\tbar\D6}_{1\tbar\text{loop}}(x_{1},\bar{a}\,|\,x_{2},\bar{b})\mathcal{Z}^{\D6\tbar\D6}_{\bar{a}\,|\,\bar{b}}(x_{1},\pi^{(1)}\,|\,x_{2},\pi^{(2)}):\Lambda_{\bar{a},\pi^{(1)}}(x_{1})\Lambda_{\bar{b},\pi^{(2)}}(x_{2}):.
\eea
When $a=b$, it gives the vector multiplet contributions, while when $a\neq b$, it gives the bifundamental contribution connecting gauge theories defined on different $\mathbb{C}^{3}$ subspaces.
\end{lemma}
\begin{theorem}\label{thm:tetra-origamiBPSCFT}
    The gauge origami partition function of the tetrahedron instanton is written using the D6 $qq$-characters:
    \bea
        \mathcal{Z}_{\text{1-loop}}^{\D6}\mathcal{Z}_{\text{inst.}}^{\D6}=\sum_{\underline{\vec{\pi}}}\mathfrak{q}^{|\underline{\vec{\pi}}|}\mathcal{Z}^{\D6}_{\text{1-loop}}\mathcal{Z}^{\D6}_{\text{tet.inst.}}[\underline{\vec{v}},\underline{\vec{\pi}}]=\bra{0}\prod_{a\in\four}\prod_{\alpha=1}^{n_{\bar{a}}}\mathsf{T}_{\bar{a}}(v_{\bar{a},\alpha})\ket{0}.
    \eea
   \end{theorem}
 Explicitly, we have the following table of BPS/CFT correspondence.
    \begin{align}
    \renewcommand\arraystretch{1.2}{
    \begin{tabular}{c|c}\toprule
        BPS &  CFT\\
     \hline  7d U(1) theory on $\mathbb{C}^{3}_{abc}\times \mathbb{S}^{1}$ ($\D6_{abc}\times 1$) & $\bra{0}\mathsf{T}_{abc}(v)\ket{0}$ \\
      7d U($n_{ab}$) theory on $\mathbb{C}^{3}_{abc}\times \mathbb{S}^{1}$ ($\D6_{abc}\times n_{abc}$)   & $\bra{0}\mathsf{T}_{abc}(v_{n_{abc}})\cdots\mathsf{T}_{abc}(v_{2})\mathsf{T}_{abc}(v_{1})\ket{0}$\\
      generalized folded instantons: D6$_{123}$-D6$_{234}$-$\D0$ & $\bra{0}\mathsf{T}_{123}(v)\mathsf{T}_{234}(v')\ket{0}$\\
      gauge origami of tetrahedron instantons & $\bra{0}\prod\limits_{a\in\four}\prod\limits_{\alpha=1}^{n_{\bar{a}}}\mathsf{T}_{\bar{a}}(v_{\bar{a},\alpha})\ket{0}$ \\ \toprule
    \end{tabular}}
\end{align}

\subsection{Fusion of D4 \texorpdfstring{$qq$}{qq}-characters to D6 \texorpdfstring{$qq$}{qq}-characters}\label{sec:D4fusiontoD6}
Let us see that the D6 $qq$-characters are obtained by fusion of the D4 $qq$-characters. By studying the zeros and pole structure, one can show the following lemma. 
\begin{lemma}\label{lem:D6planecond}
 When the Young diagrams obey $\lambda^{(2)}\npreceq\lambda^{(1)}$ with the parameters as $x_{2}=q_{a}x_{1}\,(a\in\bar{A})$, we have 
 \beq
     \mathcal{Z}_{A|A}^{\D4\tbar\D4}(x_{1},\lambda^{(1)}\,|\,q_{a}x_{1},\lambda^{(2)})=0,\quad a\in\bar{A}
 \eeq
 which gives
 \beq
     \Lambda_{A,\lambda^{(2)}}(x_{2})\Lambda_{A,\lambda^{(1)}}(x_{1})=0.\label{eq:D4fusionproperty}
 \eeq
 where we used Lemma~\ref{lemm:D4ope}.
\end{lemma}
Using the above lemma, finite products of the D4 $qq$-characters are given as 
\begin{align}
\begin{split}
     \mathsf{T}_{A}(q_{a}^{N-1}x)\cdots\mathsf{T}_{A}(q_{a}x)\mathsf{T}_{A}(x)&=\sum_{\lambda^{(N)}\preceq\cdots \lambda^{(2)}\preceq\lambda^{(1)}}\mathfrak{q}^{\sum_{i=1}^{N}|\lambda^{(i)}|}\widetilde{\mathcal{Z}}^{\D4}_{A}[\lambda^{(1)}]\cdots\widetilde{\mathcal{Z}}^{\D4}_{A}[\lambda^{(N)}]\\
     &\qquad\times\Lambda_{A,\lambda^{(N)}}(q_{a}^{N-1}x)\cdots\Lambda_{A,\lambda^{(2)}}(q_{a}x)\Lambda_{A,\lambda^{(1)}}(x)
\end{split}
\end{align}
for $\forall a\in\bar{A}$. Using the operator product and extracting the one-loop perturbative factor, we define the renormalized $N$-fusion D4 $qq$-character as 
\beq
    \overbar{\mathsf{T}}^{(N)}_{A;a}(x)=\sum_{\lambda^{(N)}\preceq\cdots \preceq\lambda^{(1)}}\mathfrak{q}^{\sum_{i=1}^{N}|\lambda^{(i)}|}\prod_{i=1}^{N}\widetilde{\mathcal{Z}}^{\D4}_{A}[\lambda^{(i)}]\prod_{i<j}\mathcal{Z}_{A|A}^{\D4\tbar\D4}(q_{a}^{i-1}x,\lambda^{(i)}\,|\,q_{a}^{j-1}x,\lambda^{(j)}):\prod_{i=1}^{N}\Lambda_{A,\lambda^{(i)}}(q_{a}^{i-1}x):
\eeq
Taking the limit $N\rightarrow \infty$ and considering infinite products of the $qq$-characters, we can see that they are related to the D6 $qq$-character.
\begin{theorem}\label{thm:D4toD6fusion}
    Taking the limit $N\rightarrow \infty$ of the renormalized $N$-fusion D4 $qq$-characters $\overbar{\mathsf{T}}^{(N)}_{ab;c}(x)$ give the D6 $qq$-character $\mathsf{T}_{abc}(x)$:
    \beq
        \overbar{\mathsf{T}}^{(N)}_{ab;c}(x)\xrightarrow{N\rightarrow \infty} \mathsf{T}_{abc}(x).
    \eeq
    Equivalently, we have 
    \beq
        \overleftarrow{\prod_{i=1}^{\infty}}\mathsf{T}_{ab}(xq_{c}^{i-1})\simeq \mathsf{T}_{abc}(x)
    \eeq
    where the symbol ``$\simeq$'' means the equality is true up to one-loop perturbative factors.
\end{theorem}
\begin{proof}
Let us focus on the case $\overbar{\mathsf{T}}^{(N)}_{12;3}(x)$ and $\mathsf{T}_{123}(x)$. Since the infinite products diverge, we need to regularize it properly (see for example \cite{Awata:2018svb}). Moreover, we need to take the inductive limit so that at large $N$, $\overbar{\mathsf{T}}_{12;3}^{(N)}(x)=\overbar{\mathsf{T}}^{(N+1)}_{12;3}(x)$ and we denote this $\overbar{\mathsf{T}}^{(\infty)}_{12;3}(x)$. We only give a sketch of the proof so see \cite{Awata:2018svb} for details. The operator $\overbar{\mathsf{T}}^{(\infty)}_{12;3}(x)$ is expanded as 
\bea
    \overbar{\mathsf{T}}^{(\infty)}_{12;3}(x)=\sum_{\emptyset\preceq\cdots \preceq\lambda^{(2)}\preceq\lambda^{(1)}}\mathfrak{q}^{\sum_{i=1}^{\infty}|\lambda^{(i)}|}\prod_{i=1}^{\infty}\widetilde{\mathcal{Z}}^{\D4}_{12}[\lambda^{(i)}]\prod_{i<j}\mathcal{Z}_{12|12}^{\D4\tbar\D4}(x_{i},\lambda^{(i)}\,|\,x_{j},\lambda^{(j)}):\prod_{i=1}^{\infty}\Lambda_{12,\lambda^{(i)}}(q_{3}^{i-1}x):
\eea
where $x_{i}=q_{3}^{i-1}x$. Since the right-hand side is expanded in all possible Young diagrams $\{\lambda^{(i)}\}$ obeying the condition $\lambda^{(i)}\succeq\lambda^{(i+1)}$, the right-hand side is expanded in all possible plane partitions (see the $(1,2)$ description in section \ref{sec:multi-dim-part}): $\pi=(\lambda^{(1)},\lambda^{(2)},\ldots,\emptyset,\emptyset,\ldots)$. Let us show that the coefficients and the operator part indeed can be written using the plane partition. The coefficient part will be 
\beq
    \prod_{i=1}^{\infty}\widetilde{\mathcal{Z}}^{\D4}_{12}[\lambda^{(i)}]\prod_{i<j}\mathcal{Z}_{12|12}^{\D4\tbar\D4}(x_{i},\lambda^{(i)}\,|\,x_{j},\lambda^{(j)})=\widetilde{\mathcal{Z}}_{\bar{4}}^{\D6}[\pi].
\eeq
The operator part is obtained from
\beq
    {:\prod_{i=1}^{\infty}\Lambda_{12,\lambda^{(i)}}(q_{3}^{i-1}x):}={:\prod_{i=1}^{\infty}\mathsf{X}_{12}(q_{3}^{i-1}x)\prod_{i=1}^{\infty}\prod_{\Abox\in\lambda^{(i)}}\mathsf{A}(\chi_{12,x_{i}}(\Bbox))^{-1}:}={:\mathsf{W}_{\bar{4}}(x)\prod_{\scube\in\pi}\mathsf{A}(\chi_{\bar{4},x}(\cube))^{-1}:}
\eeq
where we used 
\beq
    {:\prod_{i=1}^{\infty}\mathsf{X}_{12}(q_{3}^{i-1}x):}={:\prod_{i=1}^{\infty}\frac{\mathsf{W}_{\bar{4}}(q_{3}^{i-1}x)}{\mathsf{W}_{\bar{4}}(q_{3}^{i}x)}:}=\mathsf{W}_{\bar{4}}(x).
\eeq
We then get the identity $\overbar{\mathsf{T}}^{(\infty)}_{12;3}(x)=\mathsf{T}_{123}(x)$.
\end{proof}

\begin{remark}
The fusion procedure can be visualized as the case of the fusion of D2 to D4. Using the correspondence in \eqref{eq:D4Youngcorrespondence} (see also section \ref{sec:multi-dim-part}), we have 
\bea
\adjustbox{valign=c}{\includegraphics[width=5cm]{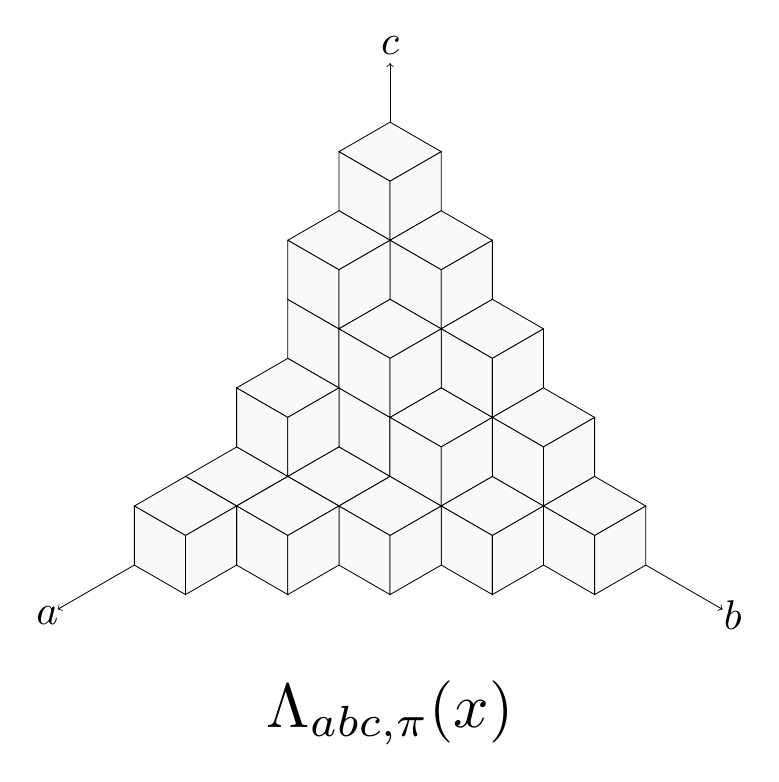}}\quad \longleftrightarrow \quad \adjustbox{valign=c}{\includegraphics[width=7cm]{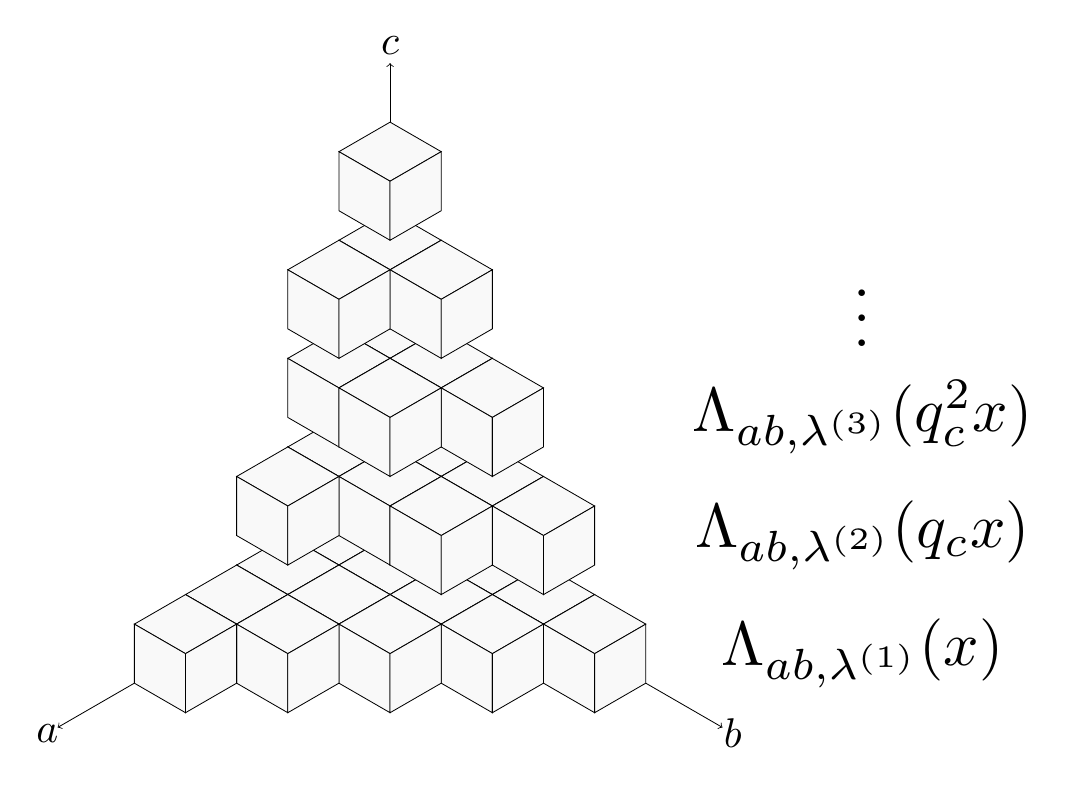}}
\eea

\end{remark}

\subsection{General D6 \texorpdfstring{$qq$}{qq}-characters}\label{sec:generalD6qq}
We can consider higher rank analogs of the D6 $qq$-characters
\beq
    \mathsf{T}^{(n)}_{\bar{a}}(\underline{x})={:\mathsf{W}_{\bar{a}}(x_{1})\cdots \mathsf{W}_{\bar{a}}(x_{n}):}+\cdots.
\eeq
The coefficients appearing in the expansion of the right-hand side correspond to the partition function of the $\U(n)$ gauge theory on the D6-brane on $\mathbb{C}^{3}_{\bar{a}}\times \mathbb{S}^{1}$. 

A further generalization is to add negative weights as the highest weights. General D6 $qq$-characters including negative weights are written as
\beq
    \mathsf{T}^{(n|m)}_{\bar{a}}(\underline{x}\,|\,\underline{y})={:\frac{\mathsf{W}_{\bar{a}}(x_{1})\cdots \mathsf{W}_{\bar{4}}(x_{n})}{\mathsf{W}_{\bar{a}}(y_{1})\cdots \mathsf{W}_{\bar{a}}(y_{m})}:}+\cdots.\label{eq:D6supergroupqqcharacter}
\eeq
The explicit coefficients are then obtained by imposing the commutativity with the screening charge $\mathscr{Q}_{a}(x')$. 

Noticing that 
\beq
    \mathsf{W}_{\bar{a}}(x)^{-1}\mathsf{S}_{a}(x')={x'}^{-1}(1-q_{a}^{-1}x'/x):\mathsf{W}_{\bar{a}}(x)^{-1}\mathsf{S}_{a}(x'):
\eeq
gives no new poles, $\mathsf{W}_{a}(x)^{-1}$ and $\mathsf{S}_{a}(x')$ commute with each other. The right-hand side of \eqref{eq:D6supergroupqqcharacter} is expanded by the plane partitions generated by $\mathsf{W}_{\bar{a}}(x_{i})\,(i=1,\ldots,n)$ and no plane partitions generated by $\mathsf{W}_{\bar{a}}(y_{j})^{-1}\,(j=1,\ldots,m)$. We expect the coefficients appearing on the right-hand side give the instanton partition function of the 7d $\U(n|m)$ theory on $\mathbb{C}^{3}_{\bar{a}}\times \mathbb{S}^{1}$. We leave a detailed analysis of this for future work \cite{Noshita-Nawata}.

For later use, let us consider the case when there is only one positive and one negative weight:
\beq
    \mathsf{T}_{\bar{a}}^{(1|1)}(x\,|\,Kx)={:\frac{\mathsf{W}_{\bar{a}}(x)}{\mathsf{W}_{\bar{a}}(Kx)}:}+\cdots
\eeq
where $K$ is a generic parameter. The contraction $\mathsf{W}_{\bar{4}}(Kx)^{-1}$ with the screening current $\mathsf{S}_{4}(x')$ will give a pole free rational function but when $K$ is generic, no poles will be canceled and the coefficients are only modified:
\begin{align}
\begin{split}
    &\mathsf{S}_{a}(x')\Lambda_{\bar{a},\pi}^{K}(x)=K^{-1}\left[\left(1-\frac{Kq_{a}x}{x'}\right)\mathscr{W}_{\pi,x}^{\bar{a}}(q_{a}^{-1}x')^{-1}\right]_{+}:\mathsf{S}_{a}(x')\Lambda_{\bar{a},\pi}^{K}(x):,\\
    &\mathscr{W}_{\pi,x}^{\bar{a}}(x')\rightarrow \frac{1}{(1-Kx/x')}\mathscr{W}_{\pi,x}^{\bar{a}}(x')\coloneqq \mathscr{W}_{\pi,x}^{\bar{a},K}(x'),\quad \widetilde{\mathcal{Z}}^{\D6}_{\bar{a}}[\pi]\rightarrow \prod_{\scube\in\pi}\frac{1-Kx/\chi_{\bar{a},x}(\cube)}{1-Kq_{a}x/\chi_{\bar{a},x}(\cube)}\widetilde{\mathcal{Z}}^{\D6}_{\bar{a}}[\pi],
\end{split}
\end{align}
where
\bea
    &\Lambda_{\bar{a},\pi}^{K}(x)={:\frac{\mathsf{W}_{\bar{a}}(x)}{\mathsf{W}_{\bar{a}}(Kx)}\prod_{\scube\in\pi}\mathsf{A}^{-1}(\chi_{\bar{a},x}(\cube)):}.
\eea
Thus, we will obtain
\bea
    \mathsf{T}_{\bar{a}}^{(1|1)}(x\,|\,Kx)&=\sum_{\pi\in\mathcal{PP}}\widetilde{\mathcal{Z}}_{\bar{a}}^{\D6}[K,\pi]\Lambda_{\bar{a},\pi}^{K}(x),\\
    \widetilde{\mathcal{Z}}_{\bar{a}}^{\D6}[K,\pi]&=\prod_{\scube\in\pi}\frac{\left(1-Kx/\chi_{\bar{a},x}(\cube)\right)\left(1-q_{a}x/\chi_{\bar{a},x}(\cube)\right)}{\left(1-Kq_{a}x/\chi_{\bar{a},x}(\cube)\right)\left(1-x/\chi_{\bar{a},x}(\cube)\right)}\prod_{\substack{\scube\in\pi\\\scubeF\in\pi}}g_{\bar{a}}\left(\frac{\chi_{\bar{a},x}(\cube)}{\chi_{\bar{a},x},(\cubeF)}\right)^{-1}.
\eea
Note that after rescaling $\mathsf{A}(x)\rightarrow \mathfrak{q}^{-1}\mathsf{A}(x)$ we can change the topological term to $\mathfrak{q}^{|\pi|}$. For later use, let us list some properties of these operators. The operator products of $\{\Lambda_{\bar{a},\pi}^{K}(x)\}$ are 
\begin{align}\label{eq:D6supergroupcontraction1}
\begin{split}
        \Lambda_{\bar{b},\pi^{(2)}}^{K_{2}}(x_{2})\Lambda^{K_{1}}_{\bar{a},\pi^{(1)}}(x_{1})&=\mathcal{Z}^{\D6\tbar\D6}_{\text{1-loop}}(x_{1},\bar{a},K_{1}\,|\,x_{2},\bar{b},K_{2})\mathcal{Z}^{\D6\tbar\D6}_{\bar{a};K_{1}|\bar{b};K_{2}}(x_{1},\pi^{(1)}\,|\,x_{2},\pi^{(2)})\\
        &\qquad\times :\Lambda_{\bar{b},\pi^{(2)}}^{K_{2}}(x_{2})\Lambda^{K_{1}}_{\bar{a},\pi^{(1)}}(x_{1}):
\end{split}
\end{align}
where 
\begin{subequations}\label{eq:D6supergroupcontraction2}
\begin{align}
    \mathcal{Z}^{\D6\tbar\D6}_{\text{1-loop}}(x_{1},\bar{a},K_{1}\,|\,x_{2},\bar{b},K_{2})&=\frac{\mathcal{Z}_{\text{1-loop}}^{\D6\tbar\D6}(x_{1},\bar{a}\,|\,x_{2},\bar{b})\mathcal{Z}_{\text{1-loop}}^{\D6\tbar\D6}(K_{1}x_{1},\bar{a}\,|\,K_{2}x_{2},\bar{b})}{\mathcal{Z}_{\text{1-loop}}^{\D6\tbar\D6}(K_{1}x_{1},\bar{a}\,|\,x_{2},\bar{b})\mathcal{Z}_{\text{1-loop}}^{\D6\tbar\D6}(x_{1},\bar{a}\,|\,K_{2}x_{2},\bar{b})},\\
    \mathcal{Z}^{\D6\tbar\D6}_{\bar{a};K_{1}|\bar{b};K_{2}}(x_{1},\pi^{(1)}\,|\,x_{2},\pi^{(2)})&=\mathcal{Z}_{\bar{a}|\bar{b}}^{\D6\tbar\D6}(x_{1},\pi^{(1)}\,|\,x_{2},\pi^{(2)})\nonumber\\
    &\qquad \times\prod_{\scube\in\pi^{(1)}}\left(q_{b}^{-1}\mathscr{V}_{b}\left(\frac{\chi_{\bar{a},x_{1}}(\cube)}{q_{b}K_{2}x_{2}}\right)\right)\prod_{\scubeF\in\pi^{(2)}}\mathscr{V}_{a}\left(\frac{K_{1}x_{1}}{\chi_{\bar{b},x_{2}}(\cubeF)}\right)^{-1}.
\end{align}
\end{subequations}

\paragraph{Pit reduction of D6 $qq$-characters}
When the parameter $K$ is generic, the coefficients are modified slightly without changing the structure of the $qq$-character. However, when we tune $K$ to specific values, the zeros appearing will cancel the poles, the iWeyl reflection will be restricted, and the right-hand side will not be expanded with arbitrary plane partitions but only by specified plane partitions. This can be understood also from the extra factors in the coefficients. When $K$ is tuned, the coefficient $\mathcal{Z}_{\bar{a}}^{\D6}[K,\pi]$ will be zero for some plane partition configurations and such terms will disappear from the $qq$-character. This procedure is well known in the literature of MacMahon representations and is called the \emph{pit reduction} \cite{Feigin2011plane,bershtein2018plane}. A plane partition with a pit~P is a plane partition that does not contain a box at the position~P.

Let us focus on the case when $a=4$. For example, when we tune the parameter $K$ as $K=q_{3}x$, we have 
\begin{align}
    {:\frac{\mathsf{W}_{\bar{4}}(x)}{\mathsf{W}_{\bar{4}}(q_{3}x)}:}=\mathsf{X}_{12}(x)
\end{align}
and the D6 $qq$-character will reduce to the D4 $qq$-character. This process is just placing a pit in $q_{3}x$ and reducing the plane partition in $(123)$ and restricting it to a Young diagram in the $(12)$-plane:
\begin{align}
    \adjustbox{valign=c}{\includegraphics[width=5cm]{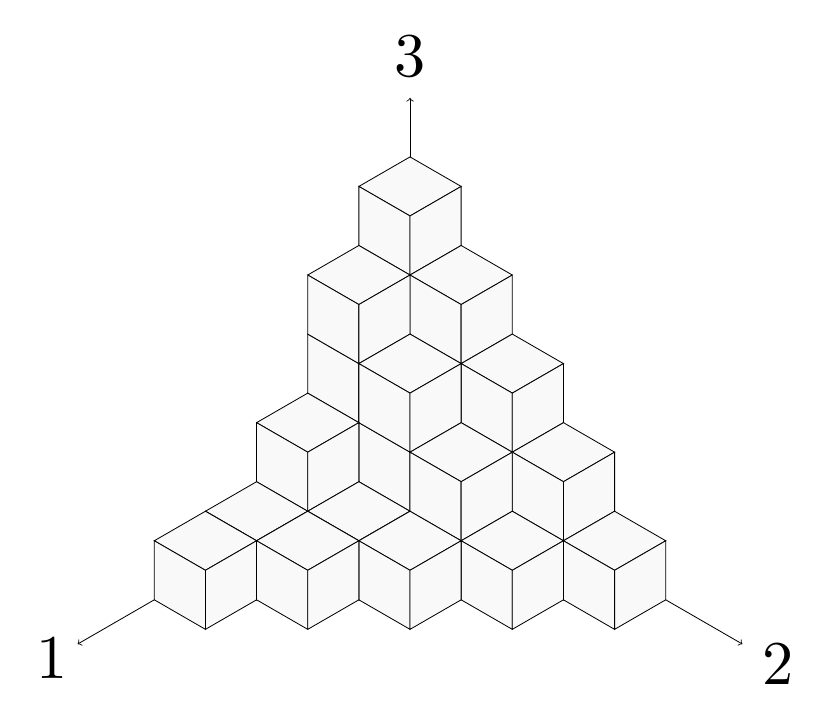}} \quad \Longrightarrow \quad \adjustbox{valign=c}{\includegraphics[width=5cm]{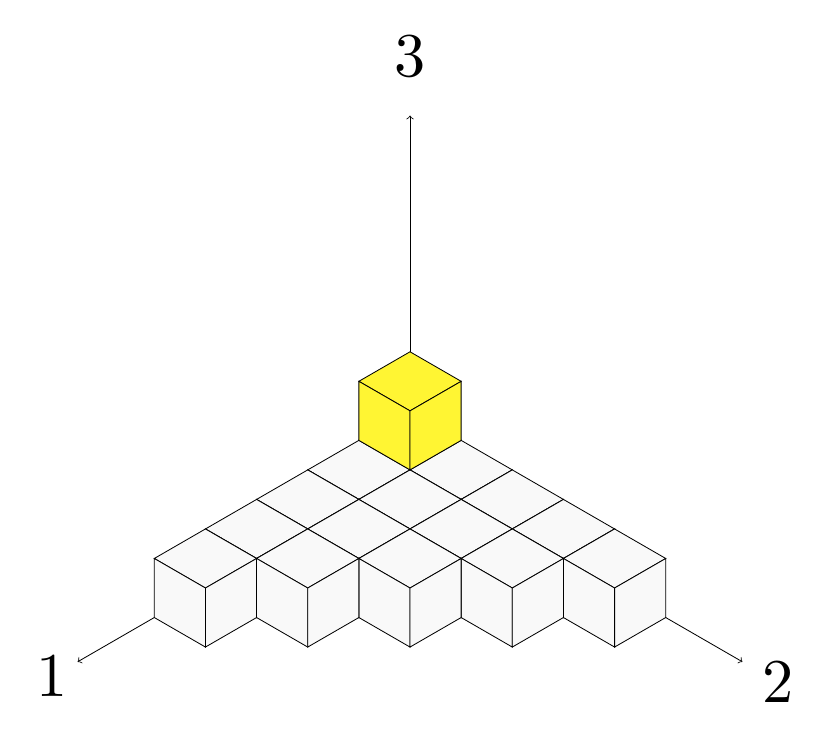}} 
\end{align}
Physically, this is interpreted as the Higgs mechanism and also the tachyon condensation as discussed in section~\ref{sec:tetraLMNS}.

When $K=\chi_{\bar{4},x}(\pitcube)/x=q_{1}^{L-1}q_{2}^{M-1}q_{3}^{N-1}$, we will get a pit reduction of the plane partition:
\begin{align}
\begin{split}
    \mathsf{T}_{\bar{4}}^{\pitcube}(x)&={:\frac{\mathsf{W}_{\bar{4}}(x)}{\mathsf{W}_{\bar{4}}(\chi_{\bar{4},x}(\hspace{-0.5mm}\scalebox{1.2}{\pitcube}))}:}+\cdots=\sum_{\substack{\pi:\text{plane partitions}\\\text{with a pit at $\pitcube$} }}\cdots
\end{split}
\end{align}
The highest weight has the following structure
\begin{align}
\begin{split}
    {:\frac{\mathsf{W}_{\bar{4}}(x)}{\mathsf{W}_{\bar{4}}(q_{1}^{L-1}q_{2}^{M-1}q_{3}^{N-1}x)}:}&={:\prod_{i=1}^{L}\frac{\mathsf{W}_{\bar{4}}(q_{1}^{i-1}q_{2}^{M-1}q_{3}^{N-1}x)}{\mathsf{W}_{\bar{4}}(q_{1}^{i}q_{2}^{M-1}q_{3}^{N-1}x)}\prod_{j=1}^{M}\frac{\mathsf{W}_{\bar{4}}(q_{2}^{j-1}q_{3}^{N-1}x)}{\mathsf{W}_{\bar{4}}(q_{2}^{j}q_{3}^{N-1}x)}\prod_{k=1}^{N}\frac{\mathsf{W}_{\bar{4}}(xq_{3}^{k-1})}{\mathsf{W}_{\bar{4}}(xq_{3}^{k})}:}\\
    &={:\prod_{i=1}^{L}\mathsf{X}_{23}(q_{1}^{i-1}q_{2}^{M-1}q_{3}^{N-1}x)\prod_{j=1}^{M}\mathsf{X}_{13}(q_{2}^{j-1}q_{3}^{N-1}x)\prod_{k=1}^{N}\mathsf{X}_{12}(xq_{3}^{k-1}):}
\end{split}
\end{align}
which implies the pit reduction of the plane partition is related to higher rank generalizations of D4 $qq$-characters. This fact is nothing special from the plane partition viewpoint. This is because we can pile $L$ Young diagrams spanning the $(23)$-plane, $M$ Young diagrams spanning the $(13)$-plane, and $N$ Young diagrams spanning the $(12)$-plane on top of each other to obtain all possible plane partitions with a pit at $(L,M,N)$. Note also that this decomposition in Young diagrams is not unique and so we have multiple descriptions in the D4 $qq$-characters. Physically, the system corresponding to the D6 $qq$-character with a pit-reduced plane partition is just the gauge origami system with folded instantons where the Coulomb branch parameters are tuned in a specific way. 


Generally, we may add another pit to the plane partition and this is called the double-constrained plane partition which was introduced in \cite{Harada:2018bkb} (see also the references there) to discuss minimal models of W-algebras. Let $(L_{1},M_{1},N_{1})$ and $(L_{2},M_{2},N_{2})$ be the coordinates of the two pits. The parameter $K$ needs to obey the conditions of the two pits
\begin{align}
    K=q_{1}^{L_{1}-1}q_{2}^{M_{1}-1}q_{3}^{N_{1}-1}=q_{1}^{L_{2}-1}q_{2}^{M_{2}-1}q_{3}^{N_{2}-1}.
\end{align}
Imposing this condition causes the $q$-parameters to be not generic anymore. The physical meaning of these types of $qq$-characters and their relation with minimal models are still unclear for the moment. We note that the condition above is just the Burge condition~\cite{Belavin:2015ria,Alkalaev:2014sma}, and thus the BPS/CFT correspondence arising should be an analog of the AGT dual of minimal models of W-algebras. See also \cite{Kimura:2022spi} where some examples of these truncations were studied from the $qq$-character viewpoint.

\section{D8 \texorpdfstring{$qq$}{qq}-characters}\label{sec:D8qq}
In this section, we derive the D8 $qq$-characters by infinite products of the D6 $qq$-characters in section~\ref{sec:fusionD6toD8}. In section~\ref{sec:D8qqsignrule}, we then study the commutativity of the D6 and D8 $qq$-characters and show that they commute with each other. We also discuss the physical interpretation of the commutativity. At the end the $qq$-character constructed here reproduce the magnificent four partition function which establishes the BPS/CFT correspondence.

\subsection{Fusion of D6 \texorpdfstring{$qq$}{qq}-characters to D8 \texorpdfstring{$qq$}{qq}-characters}\label{sec:fusionD6toD8}
After constructing $\D2,\D4,\D6$ $qq$-characters, we would like to obtain D8 $qq$-characters that reproduce the magnificent four partition function in \eqref{eq:mag4Nekrasovfact}. In the lower-dimensional cases, thanks to Thm.~\ref{thm:tetrascreening}, we can define the $qq$-characters by choosing the highest weight and imposing the commutativity with the screening charges. The highest weight and the corresponding screening charges were chosen so that the associated subspaces do not intersect in the $\mathbb{C}^{4}$ geometry. However, for the D8 operators $\mathsf{Z}(x),\widetilde{\mathsf{Z}}^{K}(x)$, the only operator that makes the operator products become rational functions is the root current $\mathsf{A}(x)$. In this sense, it is natural to construct a screening charge related to the D0 operator. For the moment, we do not know how to construct such kind of screening currents. Instead, we will use the fusion process discussed in Thm.~\ref{thm:D2toD4fusion}, \ref{thm:D4toD6fusion} to define the D8 $qq$-characters. Since we are interested in studying the relation between the magnificent four system, where D8 and anti D8-branes appear, we use $\widetilde{\mathsf{Z}}^{K}(x)$ as the highest weight.

Let us give the definition of the D8 $qq$-character first.
\begin{definition}[D8 $qq$-character]
The D8 $qq$-character is defined as
\bea
\mathsf{T}^{K}_{\four}(x)=\sum_{\rho\in\mathcal{SP}}\mathfrak{q}^{|\rho|}(-1)^{\sigma_{a}(\rho)}\mathcal{Z}^{\D8}_{\four;a}[\rho,K]\Lambda^{K}_{\four,\rho}(x),\quad a\in\four,
\eea
where $\mathcal{SP}$ denotes the set of arbitrary solid partitions extending in the four directions $1,2,3,4$. The vertex operator part is given as
\bea
\Lambda^{K}_{\four,\rho}(x)={:\mathsf{Z}(K,x)\prod_{\shcube\in\rho}\mathsf{A}^{-1}(\chi_{\four,x}(\hcube)):},
\eea
and the coefficients are the $\U(1)$ partition functions of the magnificent four system obtained. The explicit formula for the sign factor $\sigma_{a}(\rho)$ when $a=4$ is given in \eqref{eq:signfactor-def} and other formulas are obtained by using the quadrality symmetry.
\end{definition}
We note that actually the total coefficient $(-1)^{\sigma_{a}(\rho)}\mathcal{Z}^{\D8}_{\four;a}[\rho,K]$ does not depend on the choice of $a\in\four$ (see Thm.~\ref{thm:D8signindep}) and thus the above D8 $qq$-character is a unique $qq$-character.

The main claim of this section is the following theorem.
\begin{theorem}\label{thm:D6toD8}
The D8 $qq$-character is obtained as
\bea
\mathsf{T}^{K}_{\four}(x)\simeq \overleftarrow{\prod_{i=1}^{\infty}}\mathsf{T}_{\bar{a}}(xq_{a}^{i-1}\mid xq_{a}^{i-1}K),
\eea
where the equality is up to one--loop perturbative factors.
\end{theorem}
Let us focus on $a=4$ and give a proof of this theorem step by step. 
\begin{lemma}
    The contraction of the operators $\Lambda^{K}_{\bar{a},\Pi}(x)$ are
    \bea
    \Lambda^{K_{2}}_{\bar{b},\Pi^{(2)}}(x_{2})\Lambda^{K_{1}}_{\bar{a},\Pi^{(1)}}(x_{1})&=\mathcal{Z}^{\D6\tbar\D6}_{\text{1-loop}}(x_{1},\bar{a},K_{1}\mid x_{2},\bar{b},K_{2})\mathcal{Z}^{\D6\tbar\D6}_{\bar{a};K_{1}\mid \bar{b};K_{2}}(x_{1},\Pi^{(1)}\mid x_{2},\Pi^{(2)})\\
    &\quad \times :\Lambda^{K_{2}}_{\bar{b},\Pi^{(2)}}(x_{2})\Lambda^{K_{1}}_{\bar{a},\Pi^{(1)}}(x_{1}):
    \eea
    where
    \bea
    \mathcal{Z}^{\D6\tbar\D6}_{\text{1-loop}}(x_{1},\bar{a},K_{1}\mid x_{2},\bar{b},K_{2})&=\exp\left(-\sum_{n>0}\frac{1}{n}\frac{\bfP_{a}^{[n]}\bfP_{b}^{[-n]}}{\bfP_{\four}^{[n]}}(1-K_{1}^{n})(1-K_{2}^{-n})\left(\frac{x_{1}}{x_{2}}\right)^{n}\right),\\
    \mathcal{Z}^{\D6\tbar\D6}_{\bar{a};K_{1}\mid \bar{b};K_{2}}(x_{1},\Pi^{(1)}\mid x_{2},\Pi^{(2)})&=\mathbb{I}\left[-\bfP_{a}^{\vee}(1-K_{1}^{-1})x_{1}^{-1}\bm{\Pi}^{(2)}-\bfP_{4}(1-K_{2})x_{2}\bm{\Pi}^{(i)\vee}+\bfP_{\four}\bm{\Pi}^{(i)\vee}\bm{\Pi}^{(j)}\right].
    \eea
    We introduced $\bm{\Pi}^{(1,2)}=\sum_{\scube\in\Pi^{(i)}}\chi_{\bar{a},\bar{b}}(\cube)$ for plane partitions $\Pi^{(1,2)}$ and for the moment $x_{1,2}$ are generic here.
\end{lemma}

\begin{lemma}
Given two plane partitions $\Pi^{(1)},\Pi^{(2)}$ and parameters $x_{2}=q_{4}x_{1}$, we have
\bea
\mathcal{Z}^{\D6\tbar\D6}_{\bar{4};K\mid \bar{4};K}(x_{1},\Pi^{(1)}\mid x_{2},\Pi^{(2)})=0
\eea
for $\Pi^{(2)}\npreceq \Pi^{(1)}$.
\end{lemma}

Combining these lemmas, the infinite products of D6 $qq$-characters are expanded as
\bea
\overleftarrow{\prod_{i=1}^{\infty}}\mathsf{T}_{\bar{4}}(x_{i}\mid Kx_{i})&\simeq \sum_{\cdots \Pi^{(i+1)}\preceq \Pi^{(i)}\cdots}\mathfrak{q}^{\sum_{i}|\Pi^{(i)}|}\prod_{i=1}^{\infty}\widetilde{\mathcal{Z}}^{\D6}_{\bar{4}}[K,\Pi^{(i)}]\prod_{i<j}\mathcal{Z}^{\D6\tbar\D6}_{\bar{4};K\mid \bar{4};K}(x_{i},\Pi^{(i)}\mid x_{j},\Pi^{(j)}) :\prod_{i=1}^{\infty}\Lambda^{K}_{\bar{4},\Pi^{(i)}}(x_{i}):
\eea
where $x_{i}=xq_{4}^{i-1}$ and the perturbative factors are all omitted. Given a \textit{finite} solid partition $\rho$, we can decompose it into non-increasing finite plane partitions: $\rho=(\Pi^{(1)},\Pi^{(2)},\Pi^{(3)},\cdots)$ with $\Pi^{(i)}\succeq \Pi^{(i+1)}$. Since only finite numbers of $\Pi^{(i)}$ will be nonempty, the topological term is $\mathfrak{q}^{|\rho|}=\mathfrak{q}^{\sum_{i}|\Pi^{(i)}|}$. Moreover, by direct computation, one can easily show that the operator part obeys
\bea
:\prod_{i=1}^{\infty}\Lambda^{K}_{\bar{4},\Pi^{(i)}}(x_{i}):=\Lambda^{K}_{\four,\rho}(x).
\eea
Thus, the nontrivial part is how to obtain the coefficient part.

\begin{proposition}\label{prop:D6toD8vacuumproof}
    The coefficient part obeys:
    \bea
\prod_{i=1}^{\infty}\widetilde{\mathcal{Z}}^{\D6}_{\bar{4}}[K,\Pi^{(i)}]\prod_{i<j}\mathcal{Z}^{\D6\tbar\D6}_{\bar{4};K\mid \bar{4};K}(x_{i},\Pi^{(i)}\mid x_{j},\Pi^{(j)})=(-1)^{\sigma_{4}(\rho)}\mathcal{Z}^{\D8}_{\four;4}[\rho,K].
    \eea
\end{proposition}
The left hand side comes from the following character
\bea\label{eq:D8vacuum-proof1}
\mathbf{v}_{\D6\rightarrow \D8}&=\sum_{k=1}^{\infty}\left(-\bfP_{4}^{\vee}(1-K^{-1})x_{k}^{-1}\bm{\Pi}^{(k)} +\bfP_{123}^{\vee} \bm{\Pi}^{(k)\vee}\bm{\Pi}^{(k)}\right)\\
&+\sum_{i<j}\left(-\bfP_{4}^{\vee}(1-K^{-1})x_{i}^{-1}\bm{\Pi}^{(j)}-\bfP_{4}(1-K)x_{j}\bm{\Pi}^{(i)\vee}+\bfP_{\four}\bm{\Pi}^{(i)\vee}\bm{\Pi}^{(j)} \right).
\eea
First of all, one can show that this character is movable (see \cite{Nekrasov:2018xsb} and \cite{Monavari:2022rtf} for example). The term $-\bfP_{4}(1-K)x_{j}\bm{\Pi}^{(i)\vee}$ is also movable. Let $\chi=x_{i}q_{1}^{A-1}q_{2}^{B-1}q_{3}^{C-1}\,\,(A,B,C\in\mathbb{Z}_{\geq 1})$ be a term included in $\bm{\Pi}^{(i)}$, i.e. $(A,B,C)\in\Pi^{(i)}$. The unmovable part is  
\bea
\left[-\bfP_{4}(1-K)x_{j}\chi^{-1}\right]^{(0)}&=\left[-\bfP_{4}q_{4}^{j-i}q_{1}^{-A+1}q_{2}^{-B+1}q_{3}^{-C+1}\right]^{(0)}\\
&=-\delta_{A=B=C=i-j+1}+\delta_{A=B=C=i-j}
\eea
but since $j-i>0$ and $A,B,C\in\mathbb{Z}_{\geq 1}$ there is no unmovable part and we can safely use the reflection property \eqref{eq:reflectionprop2}. The character $\mathbf{v}_{\D6\rightarrow \D8}$ above is equivalent to the following character
\bea\label{eq:D8vacuum-proof2}
\mathbf{v}'_{\D6\rightarrow \D8}&=-\bfP_{4}^{\vee}(1-K^{-1})\sum_{i}x_{i}^{-1}\sum_{j}\bm{\Pi}^{(j)}+\sum_{k=1}^{\infty}\bfP_{123}^{\vee} \bm{\Pi}^{(k)\vee}\bm{\Pi}^{(k)}+\sum_{i<j}\bfP_{\four}\bm{\Pi}^{(i)\vee}\bm{\Pi}^{(j)}\\
&=-(1-K^{-1})\sum_{j}\bm{\Pi}^{(j)}+\sum_{k=1}^{\infty}\bfP_{123}^{\vee} \bm{\Pi}^{(k)\vee}\bm{\Pi}^{(k)}+\sum_{i<j}\bfP_{\four}\bm{\Pi}^{(i)\vee}\bm{\Pi}^{(j)}
\eea
where in the second line, we formally regularized $\sum_{i}x_{i}=1/\bfP_{4}$. Note that $\mathbb{I}[\mathbf{v}_{\D6\rightarrow \D8}]=\mathbb{I}\left[\mathbf{v}'_{\D6\rightarrow \D8}\right]$ and no sign factor appears at this part. The nontrivial sign factor actually comes from the $\sum_{j>i}\bfP_{\four}\bm{\Pi}^{(i)\vee}\bm{\Pi}^{(j)}$ part.
\begin{definition}\label{def:signfactorrgeneral}
Let $\bm{\Upsilon}^{(i)}$ be a character whose terms take the form as $xq_{4}^{i-1}q_{1}^{\geq0}q_{2}^{\geq 0}q_{3}^{\geq 0}$. We define the sign factor as
\bea
s\left(\Upsilon\right)\coloneqq\left[\sum_{i<j}\bfP_{123}\bm{\Upsilon}^{(i)\vee}\bm{\Upsilon}^{(j)}\right]^{(0)}
\eea
\end{definition}
Using $\bfP_{\four}=\bfP_{123}^{\vee}+\bfP_{123}$ and Prop.~\ref{prop:reflection-mod}, we have
\bea
\mathbb{I}\,[\,\mathbf{v}_{\D6\rightarrow \D8}]=
\mathbb{I}\,[\,\mathbf{v}'_{\D6\rightarrow \D8}]&=\mathbb{I}\left[\cdots +\sum_{i<j}\bfP_{\four}\bm{\Pi}^{(i)\vee}\bm{\Pi}^{(j)}\right]=\mathbb{I}\left[\cdots +\sum_{i<j}\bfP_{123}^{\vee}\bm{\Pi}^{(i)\vee}\bm{\Pi}^{(j)}+\sum_{i<j}\bfP_{123}\bm{\Pi}^{(i)\vee}\bm{\Pi}^{(j)}\right]\\
&=(-1)^{s(\Pi)}\mathbb{I}\left[\cdots +\sum_{i<j}\bfP_{123}^{\vee}\bm{\Pi}^{(i)\vee}\bm{\Pi}^{(j)}+\sum_{i<j}\bfP_{123}^{\vee}\bm{\Pi}^{(i)}\bm{\Pi}^{(j)\vee}\right]\\
&=(-1)^{s(\Pi)}\mathbb{I}\left[-(1-K^{-1})\sum_{j}\bm{\Pi}^{(j)}+\bfP_{123}^{\vee} \sum_{i}\bm{\Pi}^{(i)\vee}\sum_{j}\bm{\Pi}^{(j)}\right].
\eea
The index part indeed gives $\mathcal{Z}^{\D8}_{\four;4}[\rho,K]$ of \eqref{eq:mag4Nekrasovfact}. Combining with the following proposition, we get Thm.~\ref{thm:D6toD8}.

\begin{proposition}\label{prop:sign_proof}
The sign factor $s(\Pi)$ is
\bea
s(\Pi)=\sigma_{4}(\rho) \mod2
\eea
where $\sigma_{4}(\rho)=\#\{(i,i,i,j)\in\rho\mid i<j\}$ and we have
\bea
(-1)^{s(\Pi)}=(-1)^{\sigma_{4}(\rho)}.
\eea

\end{proposition}

\begin{proof}
    Let us consider the unmovable part of
    \bea
    \sum_{i=1}^{j-1}\bfP_{123}\bm{\Pi}^{(i)\vee}\bm{\Pi}^{(j)},\quad \bm{\Pi}^{(i)}=\sum_{\scube\in\Pi^{(i)}}q_{4}^{i-1}\chi_{\bar{4},x}(\cube).
    \eea
    We fix a term $\eta=x_{j}q_{1}^{A-1}q_{2}^{B-1}q_{3}^{C-1},\,\,(A,B,C)\in\Pi^{(j)}$ and $A,B,C\geq 1$. Since the plane partitions obey $\Pi^{(i)}\succeq \Pi^{(j)}$ for $j>i$, we have $(A,B,C)\in\Pi^{(i)}$. From the plane partition condition, the character $\bm{\Pi}^{(i)}$ can be decomposed as
    \bea
    \bm{\Pi}^{(i)}=\Delta_{i}(\eta)+\bm{\Pi}^{(i)}(\eta),\quad \Delta_{i}(\eta)=\sum_{a=1}^{A}\sum_{b=1}^{B}\sum_{c=1}^{C}xq_{4}^{i-1}q_{1}^{a-1}q_{2}^{b-1}q_{3}^{c-1}
    \eea
where $\bm{\Pi}^{i}(\eta)$ contains terms expressed as $xq_{4}^{i-1}q_{1}^{a-1}q_{2}^{b-1}q_{3}^{c-1}$ with $a>A$ or $b>B$ or $c>C$. 

First of all, the term $\bfP_{123}\bm{\Pi}^{(i)\vee}(\eta)\eta$ is movable. Focusing on
\bea
\bfP_{123}\xi^{\vee}\eta=\bfP_{123}q_{4}^{j-i}q_{1}^{A-a}q_{2}^{B-b}q_{3}^{C-c}=\bfP_{123}q_{1}^{A-a+i-j}q_{2}^{B-b+i-j}q_{3}^{C-c+i-j}
\eea
for $\forall\xi=xq_{4}^{i-1}q_{1}^{a-1}q_{2}^{b-1}q_{3}^{c-1}\in\bm{\Pi}^{(i)}(\eta)$, the term $q_{1}^{A-a+i-j}q_{2}^{B-b+i-j}q_{3}^{C-c+i-j}$ is strictly negative in the $q_{1}^{\leq -2},q_{2}^{\leq -2},q_{3}^{\leq -2}$ because $a-A,b-B,c-C\geq 1,\,\, j-i\geq 1$. Since $\bfP_{123}$ only contains terms where the degrees with respect to $q_{1,2,3}$ are $0$ or $1$, there are no unmovable terms.

It is then enough to focus on $\bfP_{123}\Delta^{\vee}_{i}(\eta)\eta$:
\bea
\bfP_{123}\Delta^{\vee}_{i}(\eta)\eta&=\bfP_{123}\frac{(1-q_{1}^{-A})(1-q_{2}^{-B})(1-q_{3}^{-C})}{\bfP^{\vee}_{123}}q_{4}^{j-i}q_{1}^{A-1}q_{2}^{B-1}q_{3}^{C-1}\\
&=(1-q_{1}^{A})(1-q_{2}^{B})(1-q_{3}^{C})(q_{1}q_{2}q_{3})^{i-j}.
\eea
The unmovable term only comes from $-q_{1}^{A}q_{2}^{B}q_{3}^{C}(q_{123})^{i-j}$ with $A=B=C=j-i>0$.

Combining all of these, we finally have
\bea
s(\Pi)=\left[\sum_{j}\sum_{i=1}^{j-1}\bfP_{123}\bm{\Pi}^{(i)\vee}\bm{\Pi}^{(j)}\right]^{(0)}&=\left[\sum_{j}\sum_{i=1}^{j-1}\sum_{\eta\in\Pi^{(j)}} \bfP_{123}\Delta^{\vee}_{i}(\eta)\eta  \right]^{(0)}\\
&=-\sum_{j}\sum_{i=1}^{j-1}\sum_{\eta=(A,B,C)\in\Pi^{(j)}}\left[q_{1}^{A}q_{2}^{B}q_{3}^{C}(q_{123})^{i-j}\right]^{(0)}\\
&=-\sum_{j}\sum_{i=1}^{j-1}\sum_{\eta=(A,B,C)\in\Pi^{(j)}}\delta_{A=B=C=j-i}\\
&=-\sum_{j}\sum_{\eta=(A,B,C)\in\Pi^{(j)}}\delta_{A=B=C<j}\\
&=\sigma_{4}(\rho)\quad \mod2.
\eea

\end{proof}

\paragraph{BPS/CFT correspondence} Using the D8 $qq$-characters, we can indeed reproduce the correct $\U(1|1)$ magnificent four partition function including the sign factors. The BPS/CFT correspondence for higher rank magnificent four theory with $\U(n|n)$ is
\bea\label{eq:D8qqBPSCFT}
\bra{0}\mathsf{T}^{K_{n}}_{\four}(x_{n})\cdots \mathsf{T}_{\four}^{K_{1}}(x_{1})\ket{0}&=\prod_{\beta>\alpha}\mathcal{Z}_{\text{1-loop}}^{\D8\tbar\D8}(x_{\alpha},K_{\alpha}\mid x_{\beta},K_{\beta})\\
&\times \sum_{\rho^{(1)},\cdots ,\rho^{(n)}}\mathfrak{q}^{|\vec{\rho}|}\prod_{\alpha=1}^{n}(-1)^{\sigma_{4}(\rho^{(\alpha)})}\mathcal{Z}_{\four;4}^{\D8}[\rho^{(\alpha)},K_{\alpha}]\prod_{\beta>\alpha}\mathcal{Z}^{\D8\tbar\D8}_{K_{\alpha},K_{\beta}}(x_{\alpha},\rho^{(\alpha)}\mid x_{\beta},\rho^{(\beta)}).
\eea

\subsection{Commutativity of D6 and D8 \texorpdfstring{$qq$}{qq}-characters}\label{sec:D8qqsignrule}
In the previous section, we derived the complete D8 $qq$-character including proper sign rules using the infinite products of lower dimensional D6 $qq$-characters. Given such $qq$-characters, one would like to find the quadratic relations of them and determine the quantum algebraic relations. In section~\ref{sec:D4qqcharacter} and \ref{sec:D6qqcharacter}, we gave a set of quadratic relations of the D2, D4, D6 $qq$-characters and showed that when the $qq$-characters are related with D-branes spanning transverse subspaces, they commute with each other. In this section, we give a conjecture of the quadratic relations of the D8 $qq$-characters and D6 $qq$-characters and discuss the physical implication of them.

We focus on the following D8 $qq$-character 
\bea
\mathsf{T}^{K}_{\four}(x)=\sum_{\rho\in\mathcal{SP}}\mathfrak{q}^{|\rho|}(-1)^{\sigma_{4}(\rho)}\mathcal{Z}^{\D8}_{\four;4}[\rho,K]\Lambda^{K}_{\four,\rho}(x).
\eea
The composition of the operators $\Lambda_{\four,\rho}^{K}(x)$ are
\bea
\Lambda_{\four,\rho^{(2)}}^{K_{2}}(x_{2})\Lambda_{\four,\rho^{(1)}}^{K_{1}}(x_{1})&=\mathcal{Z}^{\D8\tbar\D8}_{\text{1-loop}}(x_{1},K_{1}\,|\,x_{2},K_{2})\mathcal{Z}^{\D8\tbar\D8}_{K_{1}|K_{2}}(x_{1},\rho^{(1)}\,|\,x_{2},\rho^{(2)}):\Lambda_{\four,\rho^{(2)}}^{K_{2}}(x_{2})\Lambda_{\four,\rho^{(1)}}^{K_{1}}(x_{1}):.
\eea
For later use, we define
\bea
\mathsf{f}^{K_{1}|K_{2}}_{\four\four}\left(x_{1}/x_{2}\right)=\mathcal{Z}^{\D8\tbar\D8}_{\text{1-loop}}(x_{1},K_{1}\,|\,x_{2},K_{2})^{-1}
\eea
and then the composition of the $qq$-characters are given as
\bea
\mathsf{f}^{K_{1}|K_{2}}_{\four\four}(x_{1}/x_{2})\mathsf{T}^{K_{2}}_{\four}(x_{2})\mathsf{T}^{K_{1}}_{\four}(x_{1})
=&\sum_{k=0}^{\infty}\mathfrak{q}^{k}\sum_{|\rho^{(1)}|+|\rho^{(2)}|=k}(-1)^{\sigma_{4}(\rho^{(1)})+\sigma_{4}(\rho^{(2)})}\mathcal{Z}^{\D8}_{\four;4}[\rho^{(1)},K_{1}]\mathcal{Z}^{\D8}_{\four;4}[\rho^{(2)},K_{2}]\\
&\qquad \times \mathcal{Z}^{\D8\tbar\D8}_{K_{1}|K_{2}}(x_{1},\rho^{(1)}\,|\,x_{2},\rho^{(2)}):\Lambda_{\four,\rho^{(2)}}^{K_{2}}(x_{2})\Lambda_{\four,\rho^{(1)}}^{K_{1}}(x_{1}):\\
\eqqcolon &\sum_{k=0}^{\infty}\mathfrak{q}^{k}\,\mathsf{F}_{k}(x_{1},K_{1}\,|\, x_{2}, K_{2}).
\eea
Our main claim of this section is the following conjecture.
\begin{conjecture}\label{conj:D8commutativity}
The quadratic relations of the D8 $qq$-characters are
    \bea
\mathsf{f}^{K_{2}|K_{1}}_{\four\four}\left(x_{2}/x_{1}\right)\mathsf{T}_{\four}^{K_{1}}(x_{1})\mathsf{T}^{K_{2}}_{\four}(x_{2})-\mathsf{f}^{K_{1}|K_{2}}_{\four\four}\left(x_{1}/x_{2}\right)\mathsf{T}_{\four}^{K_{2}}(x_{2})\mathsf{T}^{K_{1}}_{\four}(x_{1})=0,
    \eea
    where $K_{1},K_{2}$ are arbitrary. Moreover, this commutativity uniquely determines the sign factor $(-1)^{\sigma_{4}(\rho)}$ up to a global $\mathbb{Z}_{2}$ symmetry. 
\end{conjecture}
To fix this global $\mathbb{Z}_{2}$ symmetry, we can impose the sign of the first instanton contribution by hand. Imposing the sign factor of the one-instanton to be $(-1)^{\sigma_{4}(\rho)}=1$, we have
    \bea
 \sigma_{4}(\rho)=\#\left\{(i,j)\mid (i,i,i,j)\in\rho,\quad i<j\right\},
    \eea
    which is the same with \eqref{eq:signfactor-def}. For the moment, we do not have a proof of this conjecture, but we have checked it up to five instantons $(k\leq 5)$ by using a computer program. 

At higher instanton levels $k\geq 4$, actually the factor $\mathcal{Z}^{\D8\tbar\D8}_{K_{1}|K_{2}}(x_{1},\rho^{(1)}\,|\,x_{2},\rho^{(2)})$ has poles with higher orders.
\begin{conjecture}
    The factor $\mathcal{Z}^{\D8\tbar\D8}_{K_{1}|K_{2}}(x_{1},\rho^{(1)}\,|\,x_{2},\rho^{(2)})$ has higher order poles only at $x_{1}=x_{2}$.
\end{conjecture}
Strictly speaking, to consider the quadratic relations, one needs to deal with these poles. However, to confirm the conjecture, we excluded such poles and only focus on the contribution from poles which are single order.  To include higher order poles, we expect that we need to deal with differentiated operators. Somehow, in the context of quantum algebra, discussions on higher order poles are not fully done. Well-known examples where this kind of phenomenon appears are the $qq$-characters associated with geometries with D, E-type singularities or higher rank $qq$-characters with non-trivial limits of spectral parameters. In such cases, the iWeyl reflection procedure needs to be modified and derivatives of the vertex operators will appear (see \cite{Nekrasov:2016ydq, Kimura:2015rgi} for example). In this paper, we will not make an attempt to discuss all of these aspects and leave it for future work.

\paragraph{Low instanton computations}
To see the commutativity, focusing on the $k=0,1$ sectors is already intuitive. For the zero--instanton case, the commutativity is trivial because
\bea
\mathsf{f}^{K_{2}|K_{1}}_{\four\four}(x_{2}/x_{1})\mathsf{Z}(K_{1},x_{1})\mathsf{Z}(K_{2},x_{2})&=:\mathsf{Z}(K_{2},x_{2})\mathsf{Z}(K_{1},x_{1}):\\
\mathsf{f}^{K_{1}|K_{2}}_{\four\four}(x_{1}/x_{2})\mathsf{Z}(K_{2},x_{2})\mathsf{Z}(K_{1},x_{1})&=:\mathsf{Z}(K_{2},x_{2})\mathsf{Z}(K_{1},x_{1}):.
\eea

For the one--instanton case, the sign factor $(-1)^{\sigma_{4}(\{\{\{1\}\}\})}=1,-1$ is arbitrary because it is an overall factor when considering the quadratic relation. The possible configurations are $(\rho^{(1)},\rho^{(2)})=(\{\{\{1\}\}\},\emptyset),(\emptyset,\{\{\{1\}\}\})$. The contribution from $(\rho^{(1)},\rho^{(2)})=(\{\{\{1\}\}\},\emptyset)$ is
\bea
\mathsf{f}^{K_{1}| K_{2}}_{\four\four}(x_{1}/x_{2})\mathsf{Z}(K_{2},x_{2}):\mathsf{Z}(K_{1},x_{1})\mathsf{A}^{-1}(x_{1}):&=K_{2}\frac{1-K_{2}^{-1}x_{1}/x_{2}}{1-x_{1}/x_{2}}:\mathsf{Z}(K_{2},x_{2})\mathsf{Z}(K_{1},x_{1})\mathsf{A}^{-1}(x_{1}):,\\
\mathsf{f}_{\four\four}^{K_{2}|K_{1}}(x_{2}/x_{1}):\mathsf{Z}(K_{1},x_{1})\mathsf{A}^{-1}(x_{1}):\mathsf{Z}(K_{2},x_{2})&=\frac{1-K_{2}\,x_{2}/x_{1}}{1-x_{2}/x_{1}}:\mathsf{Z}(K_{2},x_{2})\mathsf{Z}(K_{1},x_{1})\mathsf{A}^{-1}(x_{1}):
\eea
which gives
\bea
&\mathsf{f}^{K_{1}| K_{2}}_{\four\four}(x_{1}/x_{2})\mathsf{Z}(K_{2},x_{2}):\mathsf{Z}(K_{1},x_{1})\mathsf{A}^{-1}(x_{1}):-\mathsf{f}_{\four\four}^{K_{2}|K_{1}}(x_{2}/x_{1}):\mathsf{Z}(K_{1},x_{1})\mathsf{A}^{-1}(x_{1}):\mathsf{Z}(K_{2},x_{2})\\
=&-(1-K_{2})\delta\left(x_{1}/x_{2}\right):\mathsf{Z}(K_{1},x_{1})\mathsf{Z}(K_{2},x_{1})\mathsf{A}^{-1}(x_{1}):.
\eea
The contribution from $(\rho^{(1)},\rho^{(2)})=(\emptyset,\{\{\{1\}\}\})$ is obtained by switching the parameters as $K_{1}\leftrightarrow K_{2},\,x_{1}\leftrightarrow x_{2}$:
\bea
\mathsf{f}_{\four\four}^{K_{1}|K_{2}}(x_{1}/x_{2}):\mathsf{Z}(K_{2},x_{2})\mathsf{A}^{-1}(x_{2}):\mathsf{Z}(K_{1},x_{1})&=\frac{1-K_{1}\,x_{1}/x_{2}}{1-x_{1}/x_{2}}:\mathsf{Z}(K_{1},x_{1})\mathsf{Z}(K_{2},x_{2})\mathsf{A}^{-1}(x_{2}):,\\
\mathsf{f}^{K_{2}| K_{1}}_{\four\four}(x_{2}/x_{1})\mathsf{Z}(K_{1},x_{1}):\mathsf{Z}(K_{2},x_{2})\mathsf{A}^{-1}(x_{2}):&=K_{1}\frac{1-K_{1}^{-1}x_{2}/x_{1}}{1-x_{2}/x_{1}}:\mathsf{Z}(K_{1},x_{1})\mathsf{Z}(K_{2},x_{2})\mathsf{A}^{-1}(x_{2}):
\eea
which gives
\bea
&\mathsf{f}_{\four\four}^{K_{1}|K_{2}}(x_{1}/x_{2}):\mathsf{Z}(K_{2},x_{2})\mathsf{A}^{-1}(x_{2}):\mathsf{Z}(K_{1},x_{1})-\mathsf{f}^{K_{2}| K_{1}}_{\four\four}(x_{2}/x_{1})\mathsf{Z}(K_{1},x_{1}):\mathsf{Z}(K_{2},x_{2})\mathsf{A}^{-1}(x_{2}):\\
=&(1-K_{1})\delta(x_{1}/x_{2}):\mathsf{Z}(K_{1},x_{1})\mathsf{Z}(K_{2},x_{1})\mathsf{A}^{-1}(x_{1}):.
\eea
Using
\bea
\mathcal{Z}^{\D8}_{\four;4}[\{\{\{1\}\}\},K]=\frac{(1-K)(1-q_{12})(1-q_{13})(1-q_{23})}{(1-q_{1})(1-q_{2})(1-q_{3})(1-q_{123})}
\eea
we have the following one-instanton contribution
\bea
&\mathsf{F}_{1}(x_{1},K_{1}\,|\,x_{2}, K_{2})-\mathsf{F}_{1}(x_{2},K_{2}\,|\, x_{1}, K_{1})\\
=&-\mathcal{Z}^{\D8}_{\four;4}[{\{\{\{1\}\}\}}\,, K_{1}](1-K_{2})\delta\left(x_{1}/x_{2}\right):\mathsf{Z}(K_{1},x_{1})\mathsf{Z}(K_{2},x_{1})\mathsf{A}^{-1}(x_{1}):\\
&+\mathcal{Z}^{\D8}_{\four;4}[{\{\{\{1\}\}\}},K_{2}](1-K_{1})\delta(x_{1}/x_{2}):\mathsf{Z}(K_{1},x_{1})\mathsf{Z}(K_{2},x_{1})\mathsf{A}^{-1}(x_{1}):\\
=&\frac{(1-K_{1})(1-K_{2})(1-q_{12})(1-q_{13})(1-q_{23})}{(1-q_{1})(1-q_{2})(1-q_{3})(1-q_{123})}\\
&\qquad \times \left(:\mathsf{Z}(K_{1},x_{1})\mathsf{Z}(K_{2},x_{1})\mathsf{A}^{-1}(x_{1}):-:\mathsf{Z}(K_{1},x_{1})\mathsf{Z}(K_{2},x_{1})\mathsf{A}^{-1}(x_{1}):\right)\delta\left(x_{1}/x_{2}\right)\\
=&0.
\eea
Therefore, at the one-instanton level, the D8 $qq$-characters commute with each other.

\paragraph{Reduction to D6 \texorpdfstring{$qq$}{qq}-characters}
The D6 $qq$-characters can be obtained by reductions of the D8 $qq$-characters. The relation \eqref{eq:D8D6reduction} gives $\mathsf{Z}(q_{a},x)=\mathsf{W}_{\bar{a}}(x)$. Under this specialization, we actually can show that the solid partition stops its growth in one of the four directions:
\bea
    \mathcal{Z}^{\D8}_{\four;4}[\rho;q_{a}]=0,\quad \rho\in\mathcal{PP}_{a}
    \eea
    where recall that $\mathcal{PP}_{a}$ is the set of plane partitions not extending in the $q_{a}$ direction. Combining with Cor.~\ref{cor:D8-D6reduce-sign}, we have the following.
\begin{proposition}
The D8 $qq$-character reduces to an expansion of plane partitions of $\mathcal{PP}_{a}$. Using 
    \bea
    \Lambda^{q_{a}}_{\four,\pi}(x)=\Lambda_{\bar{a},\pi}(x),\quad \pi\in\mathcal
{PP}_{a},
    \eea
    we have
    \bea
    \mathsf{T}_{\four}^{q_{a}}(x)&=\sum_{\pi\in\mathcal{PP}_{a}}\mathfrak{q}^{|\pi|}(-1)^{\sigma_{4}(\pi)}\mathcal{Z}^{\D8}_{\four;4}[\pi,q_{a}]\Lambda_{\bar{a},\pi}(x)=\mathsf{T}_{\bar{a}}(x).
    \eea
\end{proposition}
Combining with Conj.~\ref{conj:D8commutativity}, we obtain the commutativity of the D6 $qq$-characters.

\begin{corollary}\label{cor:D6commutativity}
    The D6 $qq$-characters all commute with each other:
    \bea
\mathsf{f}^{q_{b}|q_{a}}_{\four\four}\left(x_{2}/x_{1}\right)\mathsf{T}_{\bar{a}}(x_{1})\mathsf{T}_{\bar{b}}(x_{2})-\mathsf{f}^{q_{a}|q_{b}}_{\four\four}\left(x_{1}/x_{2}\right)\mathsf{T}_{\bar{b}}(x_{2})\mathsf{T}_{\bar{a}}(x_{1})=0,\quad a,b\in\four.
    \eea
    Moreover, they also commute with the D8 $qq$-characters:
    \bea
    \mathsf{f}^{q_{a}|K}_{\four\four}\left(x_{2}/x_{1}\right)\mathsf{T}_{\four}^{K}(x_{1})\mathsf{T}_{\bar{a}}(x_{2})-\mathsf{f}^{K|q_{a}}_{\four\four}\left(x_{1}/x_{2}\right)\mathsf{T}_{\bar{a}}(x_{2})\mathsf{T}^{K}_{\four}(x_{1})=0,\quad a\in\four.
    \eea
\end{corollary}
Conj.~\ref{conj:D8commutativity} and Cor.~\ref{cor:D6commutativity} show that all the D6 and D8 $qq$-characters commute with each other. This property is surprising from the quantum algebraic viewpoint because usually when the generators of the deformed W-algebra commute with each other, one would expect it to be a trivial algebra.

\paragraph{Relation with the plethystic exponential formulas}
\remred{}
Although, the quantum algebraic meaning of the commutativity of the D6 and D8 $qq$-characters is not so clear for the moment, it has a physical meaning. The commutativity of the D6 and D8 $qq$-characters is actually a consequence of the fact that the partition functions of the tetrahedron instantons and magnificent four do \textit{not} depend on the Coulomb branch parameters (spectral parameters) and have a beautiful plethystic exponential formula as explained in Thm.~\ref{thm:M4PEformula}.

Using \eqref{eq:D8qqBPSCFT}, we have
\bea
\prod_{\beta>\alpha}\mathsf{f}_{\four\four}^{K_{\alpha}|K_{\beta}}(x_{\alpha}/x_{\beta})\bra{0}\mathsf{T}^{K_{n}}_{\four}(x_{n})\cdots \mathsf{T}_{\four}^{K_{1}}(x_{1})\ket{0}=\mathcal{Z}^{\D8}_{\text{inst.}}\,\left[\mathfrak{q},\{K_{\alpha}\}_{\alpha=1}^{n};q_{1,2,3,4}\right],
\eea
where the right hand side given in Thm.~\ref{thm:M4PEformula}. Since the right hand side does not depend on $\{x_{\alpha}\}_{\alpha=1}^{n}$, the left hand side should also not depend on them. Focusing on $n=2$, this gives exactly the commutativity in Conj.~\ref{conj:D8commutativity}. Similarly, this discussion is applicable to the tetrahedron instanton case by tuning the parameters $\{K_{\alpha}\}_{\alpha=1}^{n}$.

We stress that we are \textit{not} saying that the commutativity of the D6, D8 $qq$-characters proves the independence of the Coulomb branch parameters nor the existence of a plethystic formula. We are saying that if such kind of properties exist, then we should have the commutativity of the $qq$-characters and indeed for the D6, D8 $qq$-characters, it is true. 

We also note that this commutativity is \textit{not} satisfied for D4 and D6 $qq$-characters. As mentioned in section~\ref{sec:generalD6qq}, we can introduce a $\text{D6}\tbar\overline{\text{D6}}$ $qq$-character which after tachyon condensation, we can reproduce the spiked instanton D4 $qq$-characters. The commutation relation of such $\text{D6}\tbar\overline{\text{D6}}$ $qq$-characters actually reproduces extra terms which imply nontrivial quadratic relations for the affine quiver W-algebra. We postpone a detailed study of such cases for future work.


\section{Relation with quantum toroidal \texorpdfstring{$\mathfrak{gl}_{1}$}{gl(1)}}
\label{sec:toroidal_alg}
In this section, we review the quantum toroidal $\mathfrak{gl}_{1}$ and point out observations regarding the $qq$-characters we introduced in the previous sections. We will show that the $\D2,\D4,\D6$ $qq$-characters are related with the vertical representations of the quantum toroidal $\mathfrak{gl}_{1}$. At the end, we have the following correspondence:
\begin{equation*}
    \renewcommand\arraystretch{1.2}{
    \begin{tabular}{|c|c|c|}\hline
       section & $qq$-characters &  quantum toroidal $\mathfrak{gl}_{1}$\\
     \hline \ref{sec:vectorrep}  & D2 $qq$-character & vector representation \\
      \ref{sec:verticalFockrep} &D4 $qq$-character & Fock representation \\
      \ref{sec:MacMahonrep} &D6 $qq$-character & MacMahon representation \\\hline
    \end{tabular}}
\end{equation*}


\subsection{Quantum toroidal \texorpdfstring{$\mathfrak{gl}_{1}$}{gl(1)}}\label{sec:QTgl1}
The quantum toroidal $\mathfrak{gl}_{1}$ is an infinite-dimensional quantum algebra with two independent deformation parameters \cite{ding1997generalization,miki2007q,FFJMM1,Feigin2011plane,Feigin2011}. We follow the notations in \cite{DIMreview} (see also \cite[section 5.2, 5.3]{Noshita:2022otp} for a review).

\begin{definition}
Let $\mathsf{q}_{1},\mathsf{q}_{2},\mathsf{q}_{3}$ be the deformation parameters\footnote{In the literature, the deformation parameters of the algebra are denoted $q_{1},q_{2},q_{3}$. We use a different notation to prevent confusion with the parameters $q_{1},q_{2},q_{3},q_{4}$ introduced in this paper. As mentioned in footnote \ref{footnote:structure-function}, the structure function $\sfg(z)$ is related to the structure function $g_{\bar{4}}(z)$ after taking the limit $q_{4}\rightarrow 1$.}  with the condition $\sfq_{1}\sfq_{2}\sfq_{3}=1$. The quantum toroidal $\mathfrak{gl}_{1}$, which is denoted $\mathcal{E}$, is generated by three Drinfeld currents 
\begin{equation}
    E(z)=\sum_{m\in\mathbb{Z}}E_{m}z^{-m},\quad F(z)=\sum_{m\in\mathbb{Z}}F_{m}z^{-m},\quad K^{\pm}(z)=K^{\pm}\exp\left(\sum_{r>0}\mp\frac{\kappa_{r}}{r}H_{\pm r}z^{\mp r}\right)
\end{equation}
and central elements 
\begin{equation}
    C,\quad K^{-}=(K^{+})^{-1}.
\end{equation}
The defining relations are
\begin{align}
\begin{split}
    E(z)E(w)=\sfg(z/w)E(w)E(z),&\quad F(z)F(w)=\sfg(z/w)^{-1}F(w)F(z),\\
    K^{\pm}(z)K^{\pm}(w)=K^{\pm}(w)K^{\pm}(z),&\quad K^{-}(z)K^{+}(w)=\frac{\sfg(C^{-1}z/w)}{\sfg(Cz/w)}K^{+}(w)K^{-}(z),\\
    K^{\pm}(C^{(1\mp 1)/2}z)E(w)&=\sfg(z/w)E(w)K^{\pm}(C^{(1\mp1)/2}z),\\
    K^{\pm}(C^{(1\pm 1)/2}z)F(w)&=\mathsf{g}(z/w)^{-1}F(w)K^{\pm}(C^{(1\pm 1)/2}z),\\
    [E(z),F(w)]=\tilde{g}&\left(\delta\left(\frac{Cw}{z}\right)K^{+}(z)-\delta\left(\frac{Cz}{w}\right)K^{-}(w)\right)
\end{split}
\end{align}
where 
\begin{equation}\label{eq:gl1structurefunction}
    \sfg(z)=\frac{\prod_{i=1}^{3}(1-\sfq_{i}z)}{\prod_{i=1}^{3}(1-\sfq_{i}^{-1}z)},\quad \kappa_{r}=\prod_{i=1}^{3}(\sfq_{i}^{r/2}-\sfq_{i}^{-r/2}),
\end{equation}
and $\tilde{g}=1/\kappa_{1}$.
\end{definition}

Additionally, one needs the so-called Serre relations which are cubic relations of $E(z),F(z)$. The representations used in this thesis automatically satisfy them and so we omit the discussions (see \cite{DIMreview}). The function $\sfg(z)$ is called the \emph{structure function} of the quantum toroidal $\mathfrak{gl}_{1}$.

The quantum toroidal $\mathfrak{gl}_{1}$ has a Hopf algebraic structure. We only list down the coproduct structure:
\begin{align}\label{eq:coproduct}
\begin{split}
\Delta E(z)&=E(z)\otimes 1+K^{-}(C_{1}z)\otimes E(C_{1}z),\\
\Delta F(z)&=F(C_{2}z)\otimes K^{+}(C_{2}z)+1\otimes F(z),\\
\Delta K^{+}(z)&=K^{+}(z)\otimes K^{+}(C_{1}^{-1}z),\\
\Delta K^{-}(z)&=K^{-}(C_{2}^{-1}z)\otimes K^{-}(z),\\
\Delta(X)&=X\otimes X,\quad X=C,K^{-},
\end{split}
\end{align}
where $C_{1}=C\otimes 1$ and $C_{2}=1\otimes C$. Using this coproduct, we can construct tensor product representations.

The representations of the quantum toroidal $\mathfrak{gl}_{1}$ are obtained by determining the values of the central elements $C,K^{-}$. We have two classes of representations called \emph{vertical representations} and \emph{horizontal representations}. Vertical representations are representations when the central element $C$ is trivial: $C=1$. We have three types of them.
\begin{itemize}
    \item Vector representation \cite{FFJMM1}: $(C,K^{-})=(1,1)$
    \item Fock representation \cite{Feigin2011}: $(C,K^{-})=(1,\sfq_{c}^{1/2})\,(c=1,2,3)$
    \item MacMahon representation \cite{Feigin2011plane}: $(C,K^{-})=(1,K^{1/2})\,(K\in\mathbb{C}^{\times})$
\end{itemize}
Multi-dimensional partitions appear as the bases of the representation spaces of these representations. For the vector representation, 1d partitions labeled by integers appear. For the Fock and MacMahon representations, 2d and 3d partitions appear respectively (see \cite{DIMreview} and \cite[section 5.3.1]{Noshita:2022otp} for the derivations). 

On the other hand, horizontal representations are representations where the central charges are $(C,K^{-})=(\sfq_{c}^{1/2},1)\,(c=1,2,3)$ \cite{miki2007q,bershtein2018plane,FHSSY:2010,Kojima2019,Kojima2021,Harada:2021xnm}. Drinfeld currents are represented in vertex operators for these representations. See for example \cite{DIMreview} and \cite[section 5.3.2]{Noshita:2022otp} for the explicit derivation of these representations.
\subsection{Vector representation and D2 \texorpdfstring{$qq$}{qq}-character}\label{sec:vectorrep}
There are three types of vector representations with central charges $(C,K^{-})=(1,1)$ and the actions of the Drinfeld currents are
\begin{align}\label{eq:vectorrep}
  \begin{split}
        K^{\pm}(z)[u]^{(c)}_{j}=&\left[\Psi_{[u]^{(c)}_{j}}(z)\right]^{z}_{\pm}[u]^{(c)}_{j}\eqqcolon[S_{c}\left(u\sfq_{c}^{j}/z\right)]_{\pm}[u]^{(c)}_{j},\\
        E(z)[u]^{(c)}_{j}=&\mathcal{E}\delta\left(u\sfq_{c}^{j}/z\right)[u]^{(c)}_{j+1},\\
        F(z)[u]^{(c)}_{j}=&\mathcal{F}\delta\left(u\sfq_{c}^{j-1}/z\right)[u]^{(c)}_{j-1},\quad c=1,2,3,\quad j\in\mathbb{Z}
   \end{split}
\end{align}
where
\begin{equation}
    \mathcal{E}\mathcal{F}=\Tilde{g}\frac{(1-\sfq_{c+1}^{-1})(1-\sfq_{c-1}^{-1})}{(1-\sfq_{c})},\quad S_{c}(z)=\frac{(1-\sfq_{c-1}z)(1-\sfq_{c+1}z)}{(1-z)(1-\sfq_{c-1}\sfq_{c+1}z)}.
\end{equation}
We denote these representations $\mathcal{V}_{c}(u),\,(c=1,2,3)$. The bases $\{[u]_{j}^{(c)}\}_{j\in\mathbb{Z}}$ are represented by 1d partitions:
\begin{equation}
    \begin{tikzpicture}[scale=1.5]
\node at (-2,0.35) {$[u]_{j}^{(c)}=$};
\draw[thick] (0,1)--(0,-0.2);
    \draw[->] (-1,0)--(4.5,0);
    \node at (4.5,0) [right] {$\sfq_{c}$};
    \draw (-1,0.7)--(3.5,0.7);
    \draw (3.5,0.7)--(3.5,0);
    \draw (-0.7,0.7)--(-0.7,0);
    \draw (0.7,0.7)--(0.7,0);
    \draw (1.4,0.7)--(1.4,0);
    \draw (2.1, 0.7)--(2.1,0);
    \node at (1.8, 0.35) {$\cdots$};
    \node at (2.5, 0.35) {$\cdots$};
    \draw (2.8, 0.7)--(2.8,0);
    \node at (0.35, 0) [below] {$1$};
    \node at (1.05,0)[below] {$2$};
    \node at (1.75,0)[below] {$\cdots$};
    \node at (2.45,0)[below] {$\cdots$};
    \node at (3.15,0)[below] {$j$};
    \draw[->] (3.85,1.05)--(3.15,0.35);
    \node at (3.85,1.05)[right] {$u\sfq_{c}^{j-1}$};
    \draw[->] (0.35, 1.05)--(0.35,0.45);
    \node at (0.35, 1.05)[above] {$u$};
    \node at (0.35, 0.20) [right]{$\longrightarrow$};
    \node at (1.05, 0.20) [right,above] {$\sfq_{c}$};
\end{tikzpicture}
\end{equation}
The operator $K^{\pm}(z)$ acts diagonally, $E(z)/F(z)$ adds/removes boxes to/from the configuration. 

To relate the D2 $qq$-characters with the vector representations, we choose one specific direction, which is $\mathbb{C}_{4}$, in the gauge origami system. Let us study the relation of the $qq$-characters included in the $\mathbb{C}^{3}_{123}\times \mathbb{S}^{1}$. The operator products of $\mathsf{S}_{c}(z)\,(c=1,2,3)$ with $\mathsf{S}_{4}(q_{4}z)$ are
\begin{subequations}
\begin{align}
    \mathsf{S}_{c}(q_{c}^{j}u)\mathsf{S}_{4}(q_{4}z)&=\left[\mathscr{S}_{\overbar{c4}}(uq_{1}^{j}/z)\right]^{z}_{-}:\mathsf{S}_{c}(q_{c}^{j}u)\mathsf{S}_{4}(q_{4}z):\\
    \mathsf{S}_{4}(q_{4}z)\mathsf{S}_{c}(uq_{c}^{j})&=\left[\mathscr{S}_{\overbar{c4}}(uq_{c}^{j}/z)\right]^{z}_{+}:\mathsf{S}_{c}(q_{c}^{j}u)\mathsf{S}_{4}(q_{4}z):.
\end{align}
\end{subequations}
After taking the limit $q_{4}\rightarrow 1$, we can see that we have
\begin{equation}
    \mathscr{S}_{\overbar{c4}}(uq_{c}^{j}/z)\rightarrow S_{c}(u\sfq_{c}^{j}/z),\quad q_{1},q_{2},q_{3}\rightarrow \sfq_{1},\sfq_{2},\sfq_{3}.
\end{equation}
Comparing with \eqref{eq:vectorrep}, after taking the limit $q_{4}\rightarrow 1$, we can relate the monomial terms of the D2 $qq$-characters with the bases of the vector representation of quantum toroidal $\mathfrak{gl}_{1}$ as\footnote{Note that this is not a strict correspondence. In the limit $q_{4}\rightarrow 1$, some of the vertex operators will diverge and they do not obey the defining relations of the quantum toroidal $\mathfrak{gl}_{1}$. Moreover, for the moment, we do not know how to relate the other operators $E(z),\,F(z)$.}
\begin{align}
    \mathsf{S}_{4}(q_{4}z)\rightarrow K^{\pm}(z),\quad \mathsf{S}_{1}(uq_{1}^{j})\rightarrow [u]_{j}^{(1)}.
\end{align}
This correspondence strengthens the interpretation in \eqref{eq:D2vectorfigure1}. Due to this observation, we can call the D2 $qq$-characters the \textbf{vector $qq$-characters}.

\subsection{Fock representation and D4 \texorpdfstring{$qq$}{qq}-character}\label{sec:verticalFockrep}


Fock representations are representations with central charges $(C,K^{-})=(1,\sfq_{c}^{1/2}),\,(c=1,2,3)$. We denote them $\mathcal{F}_{c}(u),\, (c=1,2,3)$, respectively. The actions of the Drinfeld currents are given
\begin{align}\label{eq:Fockrep}
\begin{split}
    K^{\pm}(z)\ket{u,\lambda}^{(c)}=&\left[\Psi_{\lambda,u}^{(c)}(z)\right]^{z}_{\pm}\ket{u,\lambda}^{(c)}=\left[\sfq_{c}^{-1/2}\frac{\mathcal{Y}^{(c)}_{\lambda,u}(\sfq_{c}^{-1}z)}{\mathcal{Y}^{(c)}_{\lambda,u}(z)}\right]^{z}_{\pm}\ket{u,\lambda}^{(c)},\\
E(z)\ket{u,\lambda}^{(c)}=&\frac{1-\sfq_{c}}{\kappa_{1}}\sum_{\Abox\in A(\lambda)}\delta\left(\chi_{u}^{(c)}(\Bbox)/z\right)\underset{z=\chi_{u}^{(c)}(\Abox)}{\Res}z^{-1}\mathcal{Y}^{(c)}_{\lambda,u}(z)^{-1}\ket{u,\lambda+\Bbox}^{(c)},\\
    F(z)\ket{u,\lambda}^{(c)}=&-\frac{1-\sfq_{c}^{-1}}{\kappa_{1}}\sfq_{c}^{-1/2}\sum_{\Abox\in R(\lambda)} \delta\left(\chi_{u}^{(c)}(\Bbox)/z\right)\underset{z=\chi_{u}^{(c)}(\Abox)}{\Res}z^{-1}\mathcal{Y}_{\lambda,u}^{(c)}(\sfq_{c}^{-1}z)\ket{u,\lambda-\Bbox}^{(c)}
\end{split}
\end{align}
where
\begin{equation}
    \mathcal{Y}_{\lambda,u}^{(c)}(z)=(1-u/z)\prod_{\Abox\in\lambda}S_{c}(\chi^{(c)}_{u}(\Bbox)/z),\quad \chi^{(c)}_{u}(\Bbox)=u\sfq_{c+1}^{i-1}\sfq_{c-1}^{j-1}\,\,(i,j\geq1 ).\label{def-Yc}
\end{equation}
Note that the eigenvalue $\Psi_{\lambda,u}^{(c)}(z)$ can be rewritten as 
\begin{equation}\label{eq:FockCartan}
\Psi^{(c)}_{\lambda,u}(z)=\sfq_{c}^{-1/2}\frac{1-\sfq_{c}u/z}{1-u/z}\prod_{\Abox\in\lambda}\sfg\left(\frac{z}{\chi_{u}^{(c)}(\Bbox)}\right)
\end{equation}
The bases are represented by 2d partitions:
\begin{equation}
        \ket{u,\lambda}^{(c)}=\qquad \adjustbox{valign=c}{\begin{tikzpicture}[scale=0.7]
        \draw[->] (-1,0)--(4,0);
        \node[above] at (-0.5,4){$\sfq_{c-1}$};
        \node [right] at (4,0){$\sfq_{c+1}$};
        \node [below] at (1.25,0){$i$};
         \draw[->]   (-0.5,-0.5)--(-0.5,4);
         \draw (0.2,3.5)--(0.2,0.7);
         \draw (0.9,2.8)--(0.9,0.7);
         \draw (1.6,2.1)--(1.6,0.7);
         \draw (2.3,1.4)--(2.3,0.7);
         \draw (2.3,0.7)--(-0.5,0.7);
         \draw (2.3,1.4)--(-0.5,1.4);
         \draw (1.6,2.1)--(-0.5,2.1);
         \draw (0.9,2.8)--(-0.5,2.8);
         \draw (0.2,3.5)--(-0.5,3.5);
        \draw (-0.5,0.7)--(3,0.7);
        \draw (3,0)--(3,0.7);
        \draw (0.2,0)--(0.2,0.7);
        \draw (0.9,0)--(0.9,0.7);
         \draw (1.6,0)--(1.6,0.7);
          \draw (2.3,0)--(2.3,0.7);
          \draw (-0.15,0.35)--++(-0.7,-1);
          \node[left] at (-0.85,-0.65){$u$};
          \node [left] at (-0.5,1.75){$j$};
          \draw  (1.25,1.75)--++(0.9,0.9);
          \node[right] at (2.2,2.65) {$u\sfq_{c+1}^{i-1}\sfq_{c-1}^{j-1}$};
        \end{tikzpicture}
        }
    \end{equation}

Similar to the D2 case, we choose $\mathbb{C}_{4}$ to be a specific direction in the gauge origami system. Focusing on $\mathsf{T}_{12}(x)$ and using \eqref{eq:D4-D2contraction}, we obtain
\begin{subequations}
\begin{align}
     :\mathsf{X}_{12}(u)\prod_{\Abox\in\lambda}\mathsf{A}^{-1}(\chi_{12,u}(\Bbox)):\mathsf{S}_{4}(q_{4}z)&=\left[q_{3}^{-1}\frac{\mathscr{Y}^{12}_{\lambda,u}(q_{3}^{-1}z)}{\mathscr{Y}^{12}_{\lambda,u}(z)}\right]^{z}_{-} :\mathsf{X}_{12}(u)\prod_{\Abox\in\lambda}\mathsf{A}^{-1}(\chi_{12,u}(\Bbox))\mathsf{S}_{4}(q_{4}z):,\\
    \mathsf{S}_{4}(q_{4}z):\mathsf{X}_{12}(u)\prod_{\Abox\in\lambda}\mathsf{A}^{-1}(\chi_{12,u}(\Bbox)):&=\left[q_{3}^{-1}\frac{\mathscr{Y}^{12}_{\lambda,u}(q_{3}^{-1}z)}{\mathscr{Y}^{12}_{\lambda,u}(z)}\right]^{z}_{+} :\mathsf{X}_{12}(u)\prod_{\Abox\in\lambda}\mathsf{A}^{-1}(\chi_{12,u}(\Bbox))\mathsf{S}_{4}(q_{4}z):
\end{align}
\end{subequations}
and at the limit $q_{4}\rightarrow 1$, we have 
\begin{equation}
    \mathscr{Y}_{\lambda,u}^{12}(z)\rightarrow \mathcal{Y}^{(3)}_{\lambda,u}(z),\quad \chi_{12,u}(\Bbox)\rightarrow \chi_{u}^{(3)}(\Bbox).
\end{equation}
Thus, at the limit $q_{4}\rightarrow 1$, we can relate the monomial terms of the D4 $qq$-characters with the bases of the Fock representation of quantum toroidal $\mathfrak{gl}_{1}$ as 
\begin{equation}
    \mathsf{S}_{4}(q_{4}z)\longrightarrow K^{\pm}(z),\quad 
    {:\mathsf{X}_{12}(u)\prod_{\Abox\in\lambda}\mathsf{A}^{-1}(\chi_{12,u}(\Bbox)):}\longrightarrow \ket{u,\lambda}^{(3)}.
\end{equation}
In this sense, we can call the D4 $qq$-characters the \textbf{Fock $qq$-characters}.

\begin{remark}
    Actually, the Fock representation can be derived by an infinite number of tensor products of the vector representation \cite{FFJMM1}. The coproduct structure enables us to consider the action of the Drinfeld currents on tensor product representations $\otimes_{i=1}^{N}\mathcal{V}_{c-1}(u_{i})$. We tune the spectral parameter as $u_{i}=u\sfq_{c+1}^{i-1}$ and take the limit $N\rightarrow \infty$. After proper regularization of the infinite products, we obtain $\otimes_{i=1}^{\infty}\mathcal{V}_{c-1}(u\sfq_{c+1}^{i-1})\simeq \mathcal{F}_{c}(u) $. This property corresponds with the fact that the D4 $qq$-character can be obtained by the fusion process of the D2 $qq$-characters as discussed in Thm.~\ref{thm:D2toD4fusion}.
\end{remark}

\paragraph{Higher rank $qq$-characters}
Let us show that the higher rank D4 $qq$-characters correspond to the tensor product representations of the vertical Fock representations. We only focus on the $qq$-characters with no negative weights. The monomial terms appearing are
\begin{equation}
    :\prod_{\alpha=1}^{n_{12}}\Lambda_{12,\lambda_{\alpha}}(x_{12,\alpha})\prod_{\beta=1}^{n_{13}}\Lambda_{13,\mu_{\beta}}(x_{13,\beta})\prod_{\gamma=1}^{n_{23}}\Lambda_{23,\nu_{\gamma}}(x_{23,\gamma}):
\end{equation}
where $\lambda_{\alpha},\mu_{\beta},\nu_{\gamma}$ are Young diagrams. The coefficient appearing after taking the operator product with $\mathsf{S}_{4}(q_{4}z)$ is 
\begin{equation}
    \prod_{\alpha=1}^{n_{12}}q_{3}^{-1}\frac{\mathscr{Y}^{12}_{\lambda_{\alpha},x_{12,\alpha}}(q_{3}^{-1}z)}{\mathscr{Y}^{12}_{\lambda_{\alpha},x_{12,\alpha}}(z)}\prod_{\beta=1}^{n_{13}}q_{2}^{-1}\frac{\mathscr{Y}^{13}_{\mu_{\beta},x_{13,\beta}}(q_{2}^{-1}z)}{\mathscr{Y}^{13}_{\mu_{\beta},x_{13,\beta}}(z)}\prod_{\gamma=1}^{n_{23}}q_{1}^{-1}\frac{\mathscr{Y}^{23}_{\nu_{\gamma},x_{23,\gamma}}(q_{1}^{-1}z)}{\mathscr{Y}^{23}_{\nu_{\gamma},x_{23,\gamma}}(z)}.
\end{equation}
After taking the limit $q_{4}\rightarrow 1$, it corresponds with the Cartan eigenvalue of the tensor product representation $\bigotimes_{\alpha=1}^{n_{12}}\mathcal{F}_{3}(x_{12,\alpha})\otimes \bigotimes_{\beta=1}^{n_{13}}\mathcal{F}_{2}(x_{13,\beta})\otimes\bigotimes_{\gamma=1}^{n_{23}}\mathcal{F}_{1}(x_{23,\gamma})$. Note that the Cartan eigenvalues of tensor product representations are simply the products of the Cartan eigenvalue of each representation. This is because at $C=1$, using the coproduct structure \eqref{eq:coproduct}, we have $\Delta K^{\pm}(z)=K^{\pm}(z)\otimes K^{\pm}(z)$. Note also that the ordering of the tensor products does not matter because of the existence of the universal R-matrix\footnote{Recently there have been attempts to construct the $qq$-character using the R-matrix of the quantum toroidal $\mathfrak{gl}_1$~\cite{Liu:2022gwf,Bayindirli:2023byn}.
} \cite{miki2007q,Feigin:2015raa,Feigin:2016pld}.


\subsection{MacMahon representation and D6 \texorpdfstring{$qq$}{qq}-character}\label{sec:MacMahonrep}
MacMahon representations are representations with central charges $(C,K^{-})=(1,K^{1/2})$ where $K\in\mathbb{C}^{\times}$ is a generic parameter. The action of the Drinfeld currents is given as
\bea\label{eq:MacMahonrep}
    K^{\pm}(z)|u,\pi\rangle&=\left[\Psi_{\pi,u}(z)\right]^{z}_{\pm}|u,\pi\rangle,\\
    E(z)|u,\pi\rangle=&\sum_{\scube\in A(\pi)}\#\delta\left(\frac{z}{\chi_{u}(\cube)}\right)
    \sqrt{\underset{z=\chi_{u}(\scube)}{\mathrm{Res}}z^{-1}\Psi_{\pi,u}(z)}\,
    |u,\pi+\cube\rangle,\\
    F(z)|u,\pi\rangle&=\sum_{\scube\in R(\pi)}\#\delta\left(\frac{z}{\chi_{u}(\cube)}\right)
    \sqrt{\underset{z=\chi_{u}(\scube)}{\mathrm{Res}} z^{-1}\Psi_{\pi,u}(z)}\,
    |u,\pi-\cube\rangle,
\eea 
where $\#$ is some coefficient factor and 
\begin{equation}\label{eq:CartanMacMahon}
    \chi_{u}(\cube)=u\sfq_{1}^{i-1}\sfq_{2}^{j-1}\sfq_{3}^{k-1},\quad  \Psi_{\pi,u}(z)=K^{-1/2}\frac{1-Ku/z}{1-u/z}\prod_{\scube\in\pi}\sfg\left(\frac{z}{\chi_{u}(\cube)}\right).
\end{equation}
 We denote this representation $\mathcal{M}(u,K)$. The explicit coefficients of the right-hand side of $E(z),F(z)$ are omitted. The $qq$-character $\mathsf{T}_{123}(x)$ and screening current $\mathsf{S}_{4}(x')$ gives 
\begin{subequations}
\begin{align}
  :\frac{\mathsf{W}_{\bar{4}}(u)}{\mathsf{W}_{\bar{4}}(Ku)}\prod_{\scube\in\pi}\mathsf{A}^{-1}(\chi_{\overbar{4},u}(\cube)):\mathsf{S}_{4}(q_{4}z)
  &=-q_{4}u\left[\mathscr{W}^{\bar{4},K}_{\pi,u}(z)^{-1}\right]^{z}_{-}:\frac{\mathsf{W}_{\bar{4}}(u)}{\mathsf{W}_{\bar{4}}(Ku)}\prod_{\scube\in\pi}\mathsf{A}^{-1}(\chi_{\overbar{4},u}(\cube))\mathsf{S}_{4}(q_{4}z):,\\
  \mathsf{S}_{4}(q_{4}z):\frac{\mathsf{W}_{\bar{4}}(u)}{\mathsf{W}_{\bar{4}}(Ku)}\prod_{\scube\in\pi}\mathsf{A}^{-1}(\chi_{\overbar{4},u}(\cube)):&=-q_{4}u\left[\mathscr{W}_{\pi,x}^{\bar{4},K}(z)^{-1}\right]^{z}_{+}:\frac{\mathsf{W}_{\bar{4}}(u)}{\mathsf{W}_{\bar{4}}(Ku)}\prod_{\scube\in\pi}\mathsf{A}^{-1}(\chi_{\overbar{4},u}(\cube))\mathsf{S}_{4}(q_{4}z):.
\end{align}
\end{subequations}
Taking the limit $q_{4}\rightarrow 1$, we have $\chi_{\bar{4},u}(\cube)\rightarrow \chi_{u}(\cube)$ and 
\begin{equation}
g_{\bar{4}}(z)\longrightarrow \sfg(z),\quad 
\mathscr{W}_{\pi,u}^{\bar{4},K}(z)^{-1}\longrightarrow \frac{1-Ku/z}{1-u/z}\prod_{\scube\in\pi}\sfg\left(\frac{z}{\chi_{u}(\cube)}\right)
\end{equation}
which gives the identification
\begin{equation}
\mathsf{S}_{4}(q_{4}z)\longrightarrow K^{\pm}(z),\quad 
:\mathsf{W}_{\bar{4}}(u)\prod_{\scube\in\pi}\mathsf{A}^{-1}(\chi_{\overbar{4},u}(\cube)):\,\longrightarrow \ket{u,\pi}.
\end{equation}
Thus, at the limit $q_{4}\rightarrow 1$, we can relate the monomial terms of the D6 $qq$-characters with the bases of the MacMahon representation of the quantum toroidal 
$\mathfrak{gl}_{1}$. In this sense, we call the D6 $qq$-characters the \textbf{MacMahon $qq$-characters}. Under this identification, we can see that the distance between the D6 and $\overline{\D6}$ branes, denoted as $K$, appear as the central charge of the MacMahon representation.

\begin{remark}
    Instead of considering the $\U(1|1)$ theory of D6-branes, we can consider the $\U(1)$ theory by taking the limit $K\rightarrow 0,\,\infty$. After taking the limit $q_{4}\rightarrow 1$, the eigenvalue $\Psi_{\pi,u}(z)|_{K\rightarrow 0,\infty}$ is no longer a rational function with the same degrees in the numerator and the denominator. Note that such kind of representations are not the representation of quantum toroidal $\mathfrak{gl}_{1}$ but of the \emph{shifted} quantum toroidal $\mathfrak{gl}_{1}$ (see \cite{Bourgine:2022scz,Galakhov:2021xum,Noshita:2021dgj}). 
\end{remark}
\begin{remark}
    Similar to the relation between the vector representations and the Fock representations, the MacMahon representation also can be obtained by infinite tensor products of the Fock representations \cite{Feigin2011plane}. Consider the tensor product $\otimes_{i=1}^{N}\mathcal{F}_{c}(u_{i})$ and tune the spectral parameters as $u_{i}=u\sfq_{c}^{i-1}$. We then take the limit $N\rightarrow \infty$ and regularize it. After this process, we get $\otimes_{i=1}^{\infty}\mathcal{F}_{c}(u\sfq_{c}^{i-1})\simeq \mathcal{M}(K,u)$. Note that the nontrivial parameter $K$ appears from the regularization process. This property corresponds to the fusion process of the D4 $qq$-characters to D6 $qq$-characters discussed in Thm.~\ref{thm:D4toD6fusion}.
\end{remark}

\paragraph{General D6 $qq$-characters}
Similar to the D4 case, a higher rank version of D6 $qq$-characters corresponds to the tensor products of the MacMahon representations:
\begin{align}
\begin{split}
    :\prod_{i=1}^{N}\Lambda^{K_{i}}_{\bar{4},\pi^{(i)}}(x_{i}):\mathsf{S}_{4}(q_{4}x)=\prod_{i=1}^{N}(-q_{4}x_{i})\prod_{i=1}^{N}\left[\mathscr{W}^{\bar{4},K_{i}}_{\pi^{(i)},x_{i}}(z)^{-1}\right]^{z}_{-}:\prod_{i=1}^{N}\Lambda^{K_{i}}_{\bar{4},\pi^{(i)}}(x_{i})\mathsf{S}_{4}(q_{4}x):,\\
    \mathsf{S}_{4}(q_{4}x){:\prod_{i=1}^{N}\Lambda^{K_{i}}_{\bar{4},\pi^{(i)}}(x_{i}):}=\prod_{i=1}^{N}(-q_{4}x_{i})\prod_{i=1}^{N}\left[\mathscr{W}^{\bar{4},K_{i}}_{\pi^{(i)},x_{i}}(z)^{-1}\right]^{z}_{+}:\prod_{i=1}^{N}\Lambda^{K_{i}}_{\bar{4},\pi^{(i)}}(x_{i})\mathsf{S}_{4}(q_{4}x):.
\end{split}
\end{align}
After taking the limit $q_{4}\rightarrow 1$, one can see that it matches with the Cartan eigenvalue of the tensor product $\bigotimes_{i=1}^{N}\mathcal{M}(x_{i},K_{i})$.

As discussed in section \ref{sec:generalD6qq}, after specifying the value $K$, we can obtain the pit reduction of the D6 $qq$-characters, which eventually gives the D4 $qq$-character. This situation is the same in the quantum toroidal $\mathfrak{gl}_{1}$. Setting $K=\sfq_{c}$ in \eqref{eq:CartanMacMahon}, the Cartan eigenvalue will be \eqref{eq:FockCartan}. The residue of $\Psi_{\pi,u}(z)|_{K=\sfq_{c}}$ at $z=\sfq_{c}u$ will vanish and thus the action of $E(z)$ will stop the growth of the plane partition in the direction $\sfq_{c}$. We then obtain the $\mathcal{F}_{c}(u)$ representation.\footnote{The coefficients of the action of $E(z),F(z)$ on the bases in the MacMahon representation after setting $K=\mathsf{q}_{c}$ is different from the coefficients in the Fock representation. This comes from the degree of freedom to rescale the bases.} For a general pit located at $(L,M,N)$ in the MacMahon representation, the central charge is $K=\sfq_{1}^{L-1}\sfq_{2}^{M-1}\sfq_{3}^{N-1}$. A similar analysis can be done and we will see that the MacMahon representation will be reduced to $\mathcal{F}_{1}^{\otimes L}\otimes \mathcal{F}_{2}^{\otimes M}\otimes \mathcal{F}_{3}^{\otimes N}$ with the spectral parameters tuned properly (see for example \cite[section 5.1.4]{DIMreview}).

\chapter{Conclusion}\label{chap:conclusion}
Gauge origami is a generalized supersymmetric gauge theory defined on several intersecting space-time components and it provides a systematic way to consider generalizations of instantons. We focused on the gauge origami system of $\mathbb{C}^{4}$: spiked instanton, tetrahedron instanton, and magnificent four system. From the string theoretic viewpoint, the instantons are understood as D1-branes probing D$(2p+1)$-branes. The low energy field theory of the D1-branes are uniformly given by a 2d $\mathcal{N}=(0,2)$ quiver gauge theory. Using the quiver structure and the associated flavor symmetries, one can compute the Witten index of it and it gives the instanton partition function. By using the same data (quiver and flavor charges), we introduced vertex operators and gave free field realizations of the partition function. Based on this free field realization, we managed to construct D2, D4, D6, and D8 $qq$-characters and show the BPS/CFT correspondence for the spiked instanton, tetrahedron instanton, and magnificent four systems.

The D2 $qq$-characters play the roles of screening charges and the $\D4,\D6$ $qq$-characters are uniquely determined after setting the highest weight and imposing the commutativity condition with the screening charge. An interesting property was that the monomial terms are classified by truncations of plane partitions and we have a one-to-one correspondence with the MacMahon representation and its generalizations (e.g. truncations, boundary conditions, etc.) of the quantum toroidal $\mathfrak{gl}_{1}$.

 For the D8 $qq$-characters, we do not have a screening charge and so we derived it by taking infinite products of the D6 $qq$-characters. At the end, we managed to derive the D8 $qq$-character that produce the magnificent four partition function \textit{including} the nontrivial sign factors. We therefore managed to establish the BPS/CFT correspondence of all the setups in the gauge origami system of $\mathbb{C}^{4}$.

\vspace{1.5cm}

\begin{figure}[h]
    \centering
    \includegraphics[width=18cm]{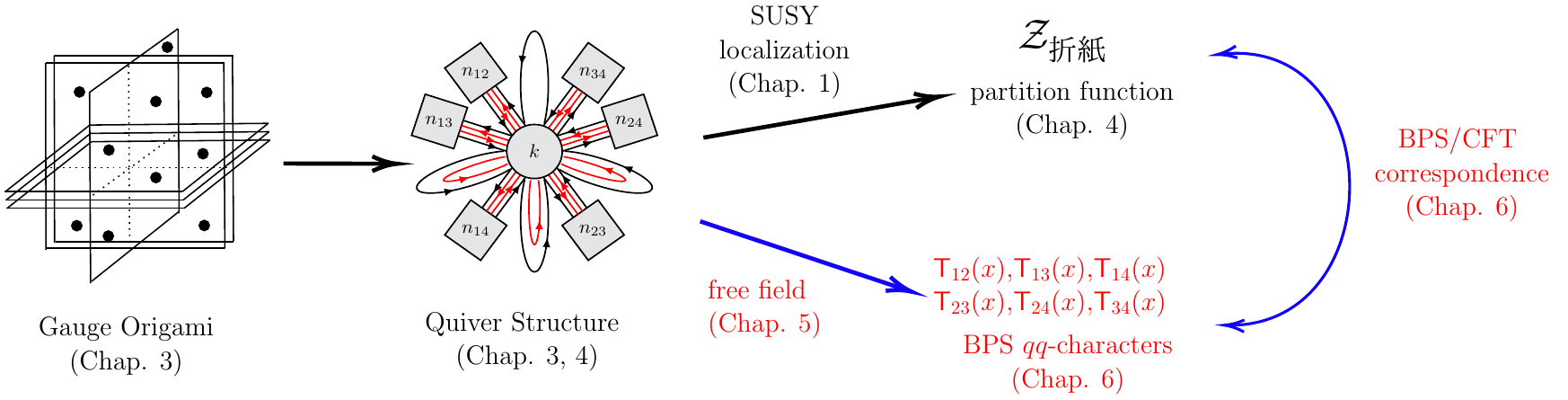}
    \label{fig:chap-struct}
\end{figure}

\renewcommand{\appendixname}{Acknowledgement}
\chapter*{Acknowledgement}



The first and deepest gratitude goes to my supervisor, Yutaka Matsuo, for advising me for five years. Throughout my time in graduate school, I received invaluable and insightful guidance from him in every aspect of my academic journey. As someone who tends to take a contrarian approach, there were many times when I did not follow that advice immediately. However, upon reflection, I now realize that much of that guidance was indeed correct. I am deeply grateful for his patience and unwavering support over the course of these five years.

The author would also like to thank Taro Kimura. The work in this thesis is based on collaborations with him. The author knew that Taro was working on gauge theory and deformed W-algebras before we met in person, but two years ago I had the chance to visit France and meet him in person. At that time I was at a loss as to what to do as a research. The lectures he gave and the discussions during the stay helped me to find out what to do next. We also happened to be from the same hometown, so we had a lot of fun talking about our hometowns. I am deeply grateful for the tremendous support I received from him, which enabled me to get through the three years of my Ph.D. studies.

The author is grateful to Satoshi Nawata and Rui-Dong Zhu too. The author had the opportunity to collaborate with them on a review of quantum toroidal algebras \cite{DIMreview}. During the collaboration, the author managed to broaden his understanding of quantum toroidal algebras. To be honest, the author feels that the journey to complete the paper was not necessarily enjoyable for him, as he is a bored person. On the other hand, looking back, I am sure that the patience I gained along the way will be invaluable to me in my future life. I would like to thank them for their patience with me until the end and for completing the review. I sincerely hope that more readers will become interested in the physical applications of quantum algebras after reading the review.

I would also like to express my deepest gratitude to Koichi Harada and Akimi Watanabe. They helped me a lot with my research right after I entered graduate school. I am sure that the basis of my research activities in the doctoral course was acquired through discussions with them.

I also thank all financial support and grants both from our government and a private company: Forefront Physics and Mathematics Program to Drive Transformation (FoPM), a World-leading Innovative Graduate Study (WINGS) Program, the University of Tokyo, JSPS fellowship for young students, MEXT, and JSR fellowship, the University of Tokyo. All these add up to the foundations of my learning and inspired me to become more innovative in my search for knowledge throughout my Ph.D. journey.

Furthermore, I am grateful to the members of the particle physics groups in Hongo, especially to Hajime Fukuda, Kohki Kawabata, Atsuya Niki, Kantaro Ohmori, and Juntaro Wada. I have learned to see the world from more... unconventional perspectives thanks to all the countless trivial conversations we have had.

The author is also grateful to many Japanese idols, streamers and VTubers. What they do is so interesting that I am afraid that the total time I spent escaping reality would be a frightening amount of time. Even when I was writing this dissertation, I was watching an online streaming for about 12 hours a day for a week or two to escape from reality.....

Finally, I would like to express my deepest gratitude to my family, who supported me with warmth and care until the very end. Without their support, I am certain I would not have been able to experience anything during these five years in graduate school.

\appendix

\renewcommand{\appendixname}{Appendix}

\chapter{Spinor representations and supersymmetry}\label{app:spinor_susy}
\section{Spinor representations}
We summarize the construction of spinor representations following \cite{Polchinski:1998rq,Polchinski:1998rr}. We focus only on the spinor representation of a $D$-dimensional spacetime where $D$ is even $D=2k+2$ and the metric is $\eta^{\mu\nu}=\operatorname{diag}(-1,+1,\ldots,+1)$. The Dirac matrix $\Gamma^{\mu}$ obeys
\bea
\{\Gamma^{\mu},\Gamma^{\nu}\}=2\eta^{\mu\nu}.
\eea
Let us construct the spinor representations explicitly. We first define
\bea
\Gamma^{0\pm}=\frac{1}{2}\left(\pm\Gamma^{0}+\Gamma^{1}\right),\quad \Gamma^{a\pm}=\frac{1}{2}\left(\Gamma^{2a}\pm i\Gamma^{2a+1}\right),\,\,a=1,\ldots, k
\eea
which obeys
\bea
\{\Gamma^{a+},\Gamma^{b-}\}=\delta^{ab},\quad \{\Gamma^{a+},\Gamma^{b+}\}=\{\Gamma^{a-},\Gamma^{b-}\}=0.
\eea
In particular, we have $(\Gamma^{a+})^{2}=(\Gamma^{a-})^{2}=0$. The operators $\Gamma^{a+},\Gamma^{a-}$ play roles of the creation and annihilation operator of fermions, respectively.

To construct the representation, we first define the vacuum $\ket{\Omega}$:
\bea
\Gamma^{a-}\ket{\Omega}=0,\quad \forall a.
\eea
The bases of the representation space are then constructed by applying the creation and annihilation operators:
\bea
\ket{\vec{s}\,}=\ket{s_{1},\ldots,s_{k}}\coloneqq (\Gamma^{k+})^{(s_{k}+1)/2}\cdots (\Gamma^{0+})^{(s_{0}+1)/2}\ket{\Omega},\quad s_{a}=\pm 1.
\eea
We then obtain a $2^{k+1}$ dimensional representation. Shortly, we omit the coefficient and only write the signs such as
\bea
\ket{\pm,\pm,\ldots,\pm}.
\eea
Note also that under this convention, the vacuum is $\Omega=\ket{-,-,\ldots,-}$.

\paragraph{$D=2$ Euclidean metric}Taking $\ket{\vec{s}\,}$ as bases, the Dirac matrices are just the Pauli matrices:
\bea
\Gamma^{2a}=\sigma_{1},\quad \Gamma^{2a+1}=\sigma_{2}
\eea
    where the Pauli matrices are defined as
    \bea
    \sigma_{1}=\begin{pmatrix}
   \,0   & \,\,1\,\\
   \,1  &\,\, 0\,
\end{pmatrix},\quad \sigma_{2}=\begin{pmatrix}
   \,0 & \,\, -i\,\\
   \,i & \,\, 0\,
\end{pmatrix},\quad \sigma_{3}=\begin{pmatrix}
  \, 1 & \,\,0\,\\
  \, 0 &  \,\,-1\,
\end{pmatrix}
    \eea
    obeying
    \bea
    \{\sigma_{i},\sigma_{j}\}=2\delta_{ij},\quad [\sigma_{i},\sigma_{j}]=2\,i\,\varepsilon_{ijk}\sigma_{k}.
    \eea

\paragraph{$D=2$ Minkowski metric} For the Dirac matrices with Minkowski metric, we have the following expression:
\bea
\Gamma^{0}=i\sigma_{2},\quad \Gamma^{1}=\sigma_{1},
\eea
where note that we have $\Gamma^{0}\Gamma^{1}=\sigma_{3}$.

\paragraph{General $D$ with Minkowski metric} The matrix representation for general even $D$ can be derived recursively. We have a $D=2k$ spinor $\xi_{(k-1)}$ and Gamma matrices $\gamma^{\mu}\,\,(\mu=0,1,\ldots,2k-1)$. The $k\rightarrow k+1$ process gives bases $\xi^{(k-1)}\otimes \ket{\pm}$. The $D=2k+2$ Gamma matrices $\Gamma^{\mu}$ can be defined as
\bea
\Gamma^{\mu}=\gamma^{\mu}\otimes \begin{pmatrix}
    1&0\\
    0&-1
\end{pmatrix}
\eea
for $\mu=0,1,2,\ldots, 2k-1$. Obviously, these Gamma matrices act as $\Gamma^{\mu}\xi_{(k-1)}\otimes \ket{\pm}=\pm \gamma^{\mu}\xi_{(k-1)}\otimes \ket{\pm}$.

The other two Gamma matrices $\Gamma^{2k,2k+1}$ then act on the remaining tensor component and reduces to the discussion of the 2d Euclidean case and we have
\bea
\Gamma^{2k}=1\otimes \sigma_{1},\quad \Gamma^{2k+1}=1\otimes \sigma_{2}.
\eea

Following this recursion relation, the Gamma matrices for 10-dimension have the following matrix representation:
\bea
\Gamma^{0}&=(i\sigma_{2})\otimes \sigma_{3}\otimes \sigma_{3}\otimes \sigma_{3}\otimes \sigma_{3},\quad   \Gamma^{5}=1\otimes 1\otimes \sigma_{2}\otimes \sigma_{3}\otimes \sigma_{3}     \\
\Gamma^{1}&=\sigma_{1}\otimes \sigma_{3}\otimes \sigma_{3}\otimes \sigma_{3}\otimes \sigma_{3},\quad \,\,\,\,\, \Gamma^{6}=1\otimes 1\otimes 1\otimes \sigma_{1}\otimes \sigma_{3} \\
\Gamma^{2}&=1\otimes \sigma_{1}\otimes \sigma_{3}\otimes \sigma_{3}\otimes \sigma_{3},\qquad \,\,\Gamma^{7}=1\otimes 1\otimes 1\otimes \sigma_{2}\otimes \sigma_{3} \\
\Gamma^{3}&=1\otimes \sigma_{2}\otimes \sigma_{3}\otimes \sigma_{3}\otimes \sigma_{3},\qquad \,\, \Gamma^{8}=1\otimes 1\otimes 1\otimes 1\otimes \sigma_{1}  \\
\Gamma^{4}&=1\otimes 1\otimes \sigma_{1}\otimes \sigma_{3}\otimes \sigma_{3},\qquad \,\,\,\,\,\Gamma^{9}=1\otimes 1\otimes 1\otimes 1\otimes \sigma_{2}.
\eea

\paragraph{Different notation}
Instead of the above expression, we may reverse the order of all of the tensor components and redefine the labels of the Gamma matrices. The result will be
\bea
\Gamma^{1}&=\sigma_{1}\otimes 1\otimes 1\otimes 1\otimes 1,\quad\qquad  \Gamma^{2}=\sigma_{2}\otimes 1\otimes 1\otimes 1\otimes 1 \\
\Gamma^{3}&=\sigma_{3}\otimes \sigma_{1}\otimes 1\otimes 1\otimes 1,\qquad \,\,\,\Gamma^{4}=\sigma_{3}\otimes \sigma_{2}\otimes 1\otimes 1\otimes 1,\\
\Gamma^{5}&=\sigma_{3}\otimes \sigma_{3}\otimes \sigma_{1}\otimes 1\otimes 1,\qquad \Gamma^{6}=\sigma_{3}\otimes \sigma_{3}\otimes \sigma_{2}\otimes 1\otimes 1,\\
\Gamma^{7}&=\sigma_{3}\otimes \sigma_{3}\otimes \sigma_{3}\otimes \sigma_{1}\otimes 1,\quad\,\,\, \Gamma^{8}=\sigma_{3}\otimes \sigma_{3}\otimes \sigma_{3}\otimes \sigma_{2}\otimes 1\\
\Gamma^{9}&=\sigma_{3}\otimes \sigma_{3}\otimes \sigma_{3}\otimes \sigma_{3}\otimes \sigma_{1},\quad\, \Gamma^{0}=\sigma_{3}\otimes \sigma_{3}\otimes \sigma_{3}\otimes \sigma_{3}\otimes (i\sigma_{2}).
\eea
In the main text, we will use this expression for convenience. The creation and annihilation operators are defined as
\bea
\Gamma^{0\pm}=\frac{1}{2}\left(\pm \Gamma^{0}+\Gamma^{9}\right),\quad \Gamma^{a\pm}=\frac{1}{2}(\Gamma^{2a-1}\pm i\Gamma^{2a}),\quad a=1,2,3,4
\eea
and we denote the bases as
\bea
\ket{s_{0},s_{1},s_{2},s_{3},s_{4}},\quad s_{a}\in\{\pm 1\}. 
\eea

Define 
\bea
\Sigma^{\mu\nu}=-\frac{i}{4}[\Gamma^{\mu},\Gamma^{\nu}]
\eea
and actually it obeys the Lorentz algebra:
\bea
\relax[M^{\mu\nu},M^{\rho\sigma}]=-i(\eta^{\mu\rho}M^{\nu\sigma}+\eta^{\nu\sigma}M^{\mu\rho}-\eta^{\nu\rho}M^{\mu\sigma}-\eta^{\mu\sigma}M^{\nu\rho}).
\eea
The generators $\Sigma^{09},\Sigma^{2a-1,2a}\,(a=1,\ldots, 4)$ commute with each other and are simultaneously diagonalizable. We then define
\bea
S_{0}=2i\,\Sigma^{09},\quad S_{a}=2\Sigma^{2a-1,2a},\quad (a=1,\ldots,4)
\eea
and we have 
\bea
S_{a}\ket{\vec{s}\,}=s_{a}\ket{\vec{s}\,},\quad s_{a}\in\{\pm 1\}.
\eea
The bases $\{\ket{\vec{s}\,}\}$ give the $2^{k}$-dimensional spinor representation of the Lorentz algebra but actually it is not irreducible because we have
\bea
\oGamma=\Gamma^{0}\Gamma^{1}\cdots \Gamma^{9}
\eea
obeys $\oGamma^{2}=1$, $\{\oGamma,\Gamma^{\mu}\}$, and $[\oGamma,\Sigma^{\mu\nu}]$ which means the representation space can be further decomposed into eigenspaces with eigenvalues $\pm1$. 

We further have
\bea
\Gamma=S_{0}S_{1}\cdots S_{4},\quad S_{0}=\Gamma^{0}\Gamma^{1},\quad S_{a}=-i\Gamma^{2a-1}\Gamma^{2a}\,\,(a=1,2,3,4),\quad S_{i}^{2}=1\,\,(i=0,1,2,3,4).
\eea
We also have
\bea
    \left(\prod_{i=1}^{a}\Gamma^{2i-1}\Gamma^{2i}\right)^{2}=(-1)^{a}.
\eea


\paragraph{Lorentz transformation of the spinors}
The Lorentz transformation acts on the spinor as
\bea
\xi \rightarrow \exp\left(\frac{i}{2}\omega_{\mu\nu}\Sigma^{\mu\nu}\right)\xi.
\eea
For example, the rotation in the Euclidean subspace gives
\bea
\exp\left(i\omega\Sigma^{2a-1,2a}\right)=\exp\left(\frac{\omega}{2}\Gamma^{2a-1}\Gamma^{2a}\right)=\exp\left(i\frac{\omega}{2} \sigma_{3}\right)
\eea
where $\omega$ is identified with the rotation angle of the Euclidean plane.

\section{Preserved supersymmetries}
We summarize the preserved supersymmetries of branes in M-theory and string theory. See \cite{Giveon:1998sr} for example.

In $D=11$, the irreducible spinor representation is the $\mathbf{32}$ spinor representation. The 32 supercharges belongs to this spinor representation. Denoting the Gamma matrices as
\bea
\Gamma^{\mu},\mu=0,1,\ldots,9,\quad \Gamma^{\natural}
\eea
where $\Gamma^{\natural}$ is the Gamma matrix of the M-theoretic 11-dimension, they obey
\bea
\Gamma^{0}\Gamma^{1}\cdots \Gamma^{\natural}=1.
\eea
If we write the 32 supercharges as $Q^{\alpha}$, the supercharge generally takes the form as $Q\sim \xi_{\alpha}Q^{\alpha}$ and $\xi_{\alpha}$ will be the 32 dimensional spinor.

Let us consider the supercharges preserved when we have M2 and M5-branes. By the rotational symmetry, the supercharges preserved take the form
\bea
\text{M2}(012):&\quad \Gamma^{0}\Gamma^{1}\Gamma^{2}\xi=\pm\xi\\
\text{M5}(01234):&\quad \Gamma^{0}\Gamma^{1}\Gamma^{2}\Gamma^{3}\Gamma^{4}\xi=\pm \xi
\eea
where the $\pm$ sign corresponds to the M-brane and anti M-brane, namely the orientation of the membranes. Note that $(\Gamma^{012})^{2}=1$ and thus such kind of spinors will always exist. The condition above projects the $\mathbf{32}$ spinor representation to one of the 16 dimensional spinor representation as $\mathbf{32}=\mathbf{16}\oplus\mathbf{16}'$ transforming under $\SO(1,2),\SO(1,4)$ acting on the M-branes. Namely, only 16 supercharges are preserved.

\paragraph{Type IIA string theory} In $D=10$, the Dirac matrix $\overline{\Gamma}\equiv\Gamma^{\natural}$ determines the chirality and the spinor representation is constructed by $\{\Gamma^{\mu}\}_{\mu=0}^{9}$. The 32 supercharges of the string theory will act as the $\mathbf{16}$ spinor representation. Decomposing the 32 supercharges into 16 supercharges as $Q_{\alpha},\tilde{Q}_{\beta}$, each of them needs to transform under the $\mathbf{16}$ representation of $\SO(1,9)$. The supercharges schematically takes the form as
\bea
Q=\xi_{\sL,\alpha}Q^{\alpha}+\xi_{\sR,\beta}\widetilde{Q}^{\beta},
\eea
Note here that the $\xi_{\sL(\sR),\alpha}$ are elements of the 16-dimensional irreducible Weyl representation. 

In type IIA string theory, the two sets of 16 supercharges have opposite chiralities:
\bea
\overline{\Gamma}\xi_{\sL}=\xi_{\sL},\quad \overline{\Gamma}\xi_{\sR}=-\xi_{\sR}.
\eea
The supersymmetries preserved by the branes are summarized as
\bea
\text{F1}(01)&:\quad \Gamma^{0}\Gamma^{1}\xi_{\sL}=\xi_{\sL},\quad \Gamma^{0}\Gamma^{1}\xi_{\sR}=-\xi_{\sR} \\
\text{NS5}(012345)&:\quad \Gamma^{012345}\xi_{\sL}=\xi_{\sL},\quad \Gamma^{012345}\xi_{\sR}=\xi_{\sR}\\
\text{D$p$}(01\cdots p)&: \quad  \Gamma^{01\cdots p}\xi_{\sR}=\xi_{\sL},
\eea
where we shortly wrote $\Gamma^{01\cdots p}=\Gamma^{0}\Gamma^{1}\cdots \Gamma^{p}$.

\paragraph{Type IIB string theory} To obtain the type IIB string theory, we can take one of the space direction and take T-duality. The two sets of 16 supercharges have the same chiralities:
\bea
\oGamma \xi_{\sL}=\xi_{\sL},\quad \oGamma\xi_{\sR}=\xi_{\sR}.
\eea
The preserved supersymmetries for each brane are given by
\bea
\text{F1}(01)&:\quad \Gamma^{01}\xi_{\sL}=\xi_{\sL},\quad \Gamma^{01}\xi_{\sR}=-\xi_{\sR} \\
\text{NS5}(012345)&:\quad \Gamma^{012345}\xi_{\sL}=\xi_{\sL},\quad \Gamma^{012345}\xi_{\sR}=-\xi_{\sR}\\
\text{D$p$}(01\cdots p)&: \quad  \Gamma^{01\cdots p}\xi_{\sR}=\xi_{\sL}.
\eea

\chapter{Open string spectrum}\label{app-chap:openstring}
In this chapter, we briefly review how to analyze the massless degrees of freedom coming from the open string spectrum in superstring theory. See \cite{Green:1987sp, Polchinski:1998rr} for standard references. 

\section{Bosonic Excitation}
Let us focus on the bosonic part of the action of the worldsheet theory:
\bea
S=\frac{1}{8\pi}\int d\tau\int_{0}^{\pi}d\sigma\left((\partial_{\tau}X)^{2}-(\partial_{\sigma}X)^{2}\right)
\eea
where $\tau,\sigma$ are the coordinates of the worldsheet. $\tau$ is the time direction and $\sigma$ is the space direction, and the ends of the string are at $\sigma=0,\pi$. We set the slope parameter $\alpha'=2$ for simplicity. When considering superstring theory, we have 10 bosonic variables denoted as $X^{\mu}$ but we focus on one bosonic component for the moment.

We introduce the  light-cone coordinates as $\sigma^{\pm}=\tau \pm \sigma$ and $\partial_{\pm}=(\partial_{\tau}\pm\partial_{\sigma})/2$. From the variations of the action, we have the  equation of motion
\bea
\partial_{+}\partial_{-}X=0
\eea
and the boundary condition
\bea
\left[\partial_{\sigma}X\delta X\right]_{\sigma=0}^{\sigma=\pi}=0.
\eea
We can impose two possible boundary conditions for each of the ends $\sigma=0,\pi$ and they are called the Neumann (N) and Dirichlet (D) boundary conditions:
\bea\label{eq:Dirichletbc-boson}
\text{N}:&\qquad \partial_{\sigma}X|_{\sigma=0,\pi}=0,\\
\text{D}:&\qquad \delta X|_{\sigma=0,\pi}=0\,\,\Leftrightarrow \,\, \partial_{\tau}X|_{\sigma=0,\pi}=0\,\,\Leftrightarrow \,\, X=\text{const.}
\eea
For the Neumann boundary condition, we have
\bea\label{eq:Neumannbc-boson}
\partial_{+}X|_{\sigma=0,\pi}=\partial_{-}X|_{\sigma=0,\pi}
\eea
because of $\partial_{\sigma}X|_{\sigma=0,\pi}=0$. For the Dirichlet boundary condition, we have
\bea
\partial_{+}X|_{\sigma=0,\pi}=-\partial_{-}X|_{\sigma=0,\pi}
\eea
because of $\partial_{\tau}X|_{\sigma=0,\pi}=0$.

The energy-momentum tensor is defined as
\bea
T_{++}=\frac{1}{2}\partial_{+}X\partial_{+}X,\quad T_{--}=\frac{1}{2}\partial_{-}X\partial_{-}X.
\eea
For the open string case, the boundary conditions relate the left-moving and right-moving modes of the boson and the energy-momentum tensor obeys the condition:
\bea
T_{++}=T_{--}.
\eea
Note that these energy-momentum tensors are defined only at the domain $\sigma=[0,\pi]$. Combining them and extending the domain to $\sigma=[0,2\pi]$ is useful; this method is known as the \textit{doubling trick}:
\bea
T_{++}(\tau,\sigma)\coloneqq \begin{dcases}
    T_{++}(\tau,\sigma),\quad\quad \quad  (0\leq \sigma \leq \pi)\\
    T_{--}(\tau,2\pi-\sigma),\quad (\pi\leq \sigma \leq 2\pi)
\end{dcases}.
\eea
Under this definition, we have $T_{++}(\tau, 2\pi)=T_{++}(\tau,0)$ and the mode expansions are given as
\bea
T_{++}(\tau,\sigma)=\sum_{n\in\mathbb{Z}}L_{n}e^{-in\sigma^{+}}.
\eea

Depending on which boundary condition we impose on the two ends, we have four possible cases.
\begin{itemize}[topsep=2pt, partopsep=0pt, itemsep=0pt]
    \item \textbf{NN boundary condition}: We impose N boundary conditions to both ends $\sigma=0,\pi$. Under this boundary condition, the mode expansions obeying the EOM is
    \bea
    X(\tau,\sigma)=x+4p\tau+2i\sum_{n\neq 0}\frac{\alpha_{n}}{n}e^{-in\tau}\cos n\sigma.
    \eea
    After quantization, we have
    \bea
    \relax[X(\sigma),\dot{X}(\sigma')]=4\pi i \delta(\sigma-\sigma'),\quad [x,p]=i,\quad [\alpha_{n},\alpha_{m}]=n\delta_{n+m,0}.
    \eea
    The current $\partial_{+}X$ is expanded as
    \bea
    \partial_{+}X=2p+\sum_{n\neq 0}\alpha_{n}e^{-in\sigma^{+}}=\sum_{n\in\mathbb{Z}}\alpha_{n}e^{-in\sigma^{+}}
    \eea
    where we defined $\alpha_{0}=2p$. The doubling trick can be performed as
    \bea
    \partial_{+}X(\tau,\sigma)=\begin{dcases}
    \partial_{+}X(\tau,\sigma)\quad \quad (0\leq \sigma\leq \pi)\\
    \partial_{-}X(\tau,2\pi-\sigma)\quad (\pi\leq \sigma\leq 2\pi)
    \end{dcases}.
    \eea
    The modes of the energy-momentum tensor are given as
    \bea
    L_{n}=\frac{1}{2}\sum_{m\in\mathbb{Z}}\alpha_{n-m}\alpha_{m}\,\,\,(n\neq 0),\quad L_{0}=\frac{1}{2}\alpha_{0}^{2}+\sum_{m=1}^{\infty}\alpha_{-m}\alpha_{m}.
    \eea

    \item \textbf{DD boundary condition}: D boundary conditions are imposed on both ends:
    \bea
    X(\tau,0)=v_{0},\quad X(\tau,\pi)=v_{\pi}.
    \eea
    The mode expansion is given as
    \bea
    X(\tau,\sigma)=v_{0}+\frac{v_{\pi}-v_{0}}{\pi}\sigma +2\sum_{n\neq 0}^{\infty}\frac{\alpha_{n}}{n}e^{-in\tau}\sin n\sigma.
    \eea
    After quantization, we have
    \bea
  \relax [\alpha_{n},\alpha_{m}]=n\delta_{n+m,0}.
    \eea
    For convenience, we can introduce $\alpha_{0}=(v_{\pi}-v_{0})/2\pi$ which gives
    \bea
    \partial_{+}X(\tau,\sigma)=\sum_{n\in \mathbb{Z}}\alpha_{n}e^{-in\sigma^{+}}.
    \eea
    The doubling trick is performed as
    \bea
    \partial_{+}X(\tau,\sigma)=\begin{dcases}
        \partial_{+}X(\tau,\sigma)\quad (0\leq \sigma\leq \pi)\\
        -\partial_{-}X(\tau,2\pi-\sigma)\quad (\pi\leq \sigma\leq 2\pi)
    \end{dcases}.
    \eea
    The modes of the energy-momentum tensor are 
    \bea
    L_{n}=\frac{1}{2}\sum_{m\in\mathbb{Z}}\alpha_{n-m}\alpha_{m}\,\, (n\neq 0),\quad L_{0}=\frac{1}{2}\alpha_{0}^{2}+\sum_{m=1}^{\infty}\alpha_{-m}\alpha_{m}.
    \eea
    \item \textbf{DN boundary condition}: The Neumann (Dirichlet) boundary condition is imposed at $\sigma=0\,\,(\pi)$:
    \bea
    X(\tau,0)=v_{0},\quad \partial_{\sigma}X(\tau,\pi)=0.
    \eea
    The mode expansions are given as
    \bea
    X(\tau,\sigma)=v_{0}+2\sum_{r\in\mathbb{Z}+\frac{1}{2}}\frac{\alpha_{r}}{r}e^{-ir\tau}\sin r\sigma
    \eea
    and after quantization, we have
    \bea
   \relax [\alpha_{r},\alpha_{s}]=r\delta_{r+s,0}. 
    \eea
    The doubling trick is performed as
    \bea
    \partial_{+}X(\tau,\sigma)=\begin{dcases}
        \partial_{+}X(\tau,\sigma)\qquad (0\leq \sigma\leq \pi)\\
        \partial_{-}X(\tau,2\pi-\sigma)\qquad (\pi\leq \sigma\leq 2\pi)
    \end{dcases}.
    \eea
    The modes of the energy-momentum tensor are
    \bea
    L_{n}=\frac{1}{2}\sum_{r\in\mathbb{Z}+1/2}\alpha_{n-r}\alpha_{r}\,\,(r\neq 0),\quad L_{0}=\sum_{r\in\mathbb{Z}+1/2,r>0}\alpha_{-r}\alpha_{r}+\frac{1}{16}.
    \eea
    \item \textbf{ND boundary condition}: The Dirichlet (Neumann) boundary condition is imposed at $\sigma=0$ $(\pi)$:
    \bea
    \partial_{\sigma}X(\tau,0)=0,\quad X(\tau,\pi)=v_{\pi}
    \eea
    This gives 
    \bea
    X(\tau,\sigma)=v_{\pi}+2i\sum_{r\in\mathbb{Z}+1/2}\frac{\alpha_{r}}{r}e^{-ir\tau}\cos r\sigma
    \eea
    and after quantization, we have
    \bea
   \relax [\alpha_{r},\alpha_{s}]=r\delta_{r+s,0}.
    \eea
    The doubling trick is performed as
    \bea
    \partial_{+}X(\tau,\sigma)=\begin{dcases}
        \partial_{+}X(\tau,\sigma)\qquad (0\leq \sigma\leq \pi)\\
        -\partial_{-}X(\tau,2\pi-\sigma)\qquad (\pi\leq \sigma\leq 2\pi)
    \end{dcases}.
    \eea
    The modes for the energy-momentum tensor are given as
    \bea
    L_{n}=\frac{1}{2}\sum_{r\in\mathbb{Z}+1/2}\alpha_{n-r}\alpha_{r}\quad (n\neq 0),\qquad L_{0}=\sum_{r\in\mathbb{Z}+1/2>0}\alpha_{-r}\alpha_{r}+\frac{1}{16}
    \eea

\end{itemize}

\paragraph{Zero-point energy}
Given the above mode expansions, an observation is that the current $\partial_{+}X(\tau,\sigma)$ obeys the periodic boundary condition $\partial_{+}X(\tau,\sigma)=\partial_{+}X(\tau,\sigma+2\pi)$ for the NN, DD boundary conditions, while it obeys the anti-periodic boundary condition $\partial_{+}X(\tau,\sigma)=-\partial_{+}X(\tau,\sigma+2\pi)$ for the ND, DN boundary conditions.

Let us compute the zero-point energy for the above cases. For periodic boundary conditions, we have
\bea
\frac{1}{2}\sum_{n\in\mathbb{Z}}\alpha_{-n}\alpha_{n}=L_{0}+\frac{1}{2}\sum_{n>0}n
\eea
and the zero-point energy is defined as
\bea
E_{0}=-\frac{1}{2}\sum_{n>0}n=-\frac{1}{2}\zeta(-1,1)=+\frac{1}{24}
\eea
where $\zeta(s,a)=\sum_{n=0}^{\infty}(n+a)^{-s}$ is the Hurwitz zeta function and we used
\bea
\zeta\left(-1,1\right)=-\frac{1}{12}.
\eea
For the anti-periodic boundary condition, we have
\bea
\frac{1}{2}\sum_{r\in\mathbb{Z}+1/2}\alpha_{-r}\alpha_{r}=\sum_{r>0}\alpha_{-r}\alpha_{r}+\frac{1}{2}\sum_{r\in\mathbb{Z}+1/2,r>0}r
\eea
and the zero-point energy is defined as
\bea
E_{0}=-\frac{1}{2}\sum_{r\in\mathbb{Z}+1/2,r>0}r=-\frac{1}{2}\zeta\left(-1,\frac{1}{2}\right)=-\frac{1}{48}.
\eea

\section{NS/R sector}
Let $\psi_{+}(\tau,\sigma)$ $\sigma\in[0,2\pi]$ be a fermion whose Lagrangian is given as
\bea
\mathcal{L}=\frac{i}{2\pi}\psi_{+}\partial_{-}\psi_{+}.
\eea
The EOM is given by
\bea
\partial_{-}\psi_{+}=0.
\eea
Depending on the boundary conditions, we have two sectors called the \textbf{NS sector} and the \textbf{R sector}:
\bea
\text{NS}:&\qquad \psi_{+}(\tau,\sigma)=-\psi_{+}(\tau,\sigma+2\pi)\\
\text{R}:&\qquad \psi_{+}(\tau,\sigma)=+\psi_{+}(\tau,\sigma+2\pi).
\eea
The mode expansions are given as
\bea\label{eq:NSRmodeexpand}
\psi_{+}(\tau,\sigma)=\sum_{r}\psi_{r}e^{-ir\sigma^{+}}\qquad \begin{dcases}
    r\in\mathbb{Z}+\frac{1}{2},\quad \text{NS sector}\\
    r\in\mathbb{Z},\qquad \text{R sector}
\end{dcases}.
\eea
After quantization, we obtain
\bea
\{\psi_{r},\psi_{s}\}=\delta_{r+s,0}.
\eea
Note that for the R sector, we have zero modes.

Let us then move on to the open superstring spectrum. The worldsheet action is given as
\bea
S=\frac{1}{2\pi}\int d^{2}\sigma (\partial_{+}X\partial_{-}X+i\psi_{+}\partial_{-}\psi_{+}+i\psi_{-}\partial_{+}\psi_{-})
\eea
where the domain of $\sigma$ is $\sigma\in[0,\pi]$. The supercurrents are defined as
\bea
G_{++}=\psi_{+}\partial_{+}X,\quad G_{--}=\psi_{-}\partial_{-}X.
\eea
For supersymmetry to be preserved at the boundary, we need
\bea
G_{++}=\pm G_{--}.
\eea
We can fix the boundary condition at the $\sigma=\pi$ as $G_{++}=G_{--}$ and we have two possible boundary conditions at $\sigma=0$
\bea\label{eq:supercurrent-bc}
\text{NS sector}:&\qquad G_{++}=-G_{--}\\
\text{R sector}:&\qquad G_{++}=G_{--}.
\eea
The doubling trick is then performed as
\bea
G_{++}(\tau,\sigma)=\begin{dcases}
G_{++}(\tau,\sigma)\qquad (0\leq \sigma\leq \pi)\\
G_{--}(\tau,2\pi-\sigma)\qquad (\pi\leq\sigma \leq 2\pi)
\end{dcases}.
\eea
After this doubling trick, the condition of the NS/R sector is simply given as $G_{++}(\tau,2\pi)=-G_{++}(\tau,0)$ for the NS sector and $G_{++}(\tau,2\pi)=+G_{++}(\tau,0)$ for the R sector.

The explicit boundary conditions for the fermions are obtained by using \eqref{eq:Dirichletbc-boson}, \eqref{eq:Neumannbc-boson} and \eqref{eq:supercurrent-bc}. Let us summarize the boundary conditions for the fermions for the NN, DD, DN, ND boundary conditions.
\begin{itemize}[topsep=2pt, partopsep=0pt, itemsep=0pt]
    \item NN boundary condition: We have
    \bea
    \text{NS sector}:&\quad \psi_{+}=-\psi_{-}\\
    \text{R sector}:& \quad \psi_{+}=\psi_{-}
    \eea
    for $\sigma=0$ and 
    \bea
    \psi_{+}=\psi_{-}
    \eea
    for $\sigma=\pi$. The doubling trick is performed as
    \bea
    \psi_{+}(\tau,\sigma)=\begin{dcases}
        \psi_{+}(\tau,\sigma)\qquad (0\leq \sigma\leq \pi)\\
        \psi_{-}(\tau,2\pi-\sigma)\qquad (\pi\leq \sigma\leq 2\pi)
    \end{dcases}
    \eea
    and the boundary conditions are summarized as
    \bea
    \text{NS sector}:&\quad  \psi_{+}(\tau,2\pi)=-\psi_{+}(\tau,0)\\
    \text{R sector}:&\quad \psi_{+}(\tau,2\pi)=+\psi_{+}(\tau,0).
    \eea

    \item DD boundary condition: We have
    \bea
    \text{NS sector}:&\quad \psi_{+}=+\psi_{-}\\
    \text{R sector}:& \quad \psi_{+}=-\psi_{-}
    \eea
    for $\sigma=0$ and 
    \bea
    \psi_{+}=-\psi_{-}
    \eea
    for $\sigma=\pi$. The doubling trick is performed as
    \bea
    \psi_{+}(\tau,\sigma)=\begin{dcases}
        \psi_{+}(\tau,\sigma)\qquad (0\leq \sigma\leq \pi)\\
        -\psi_{-}(\tau,2\pi-\sigma)\qquad (\pi\leq \sigma\leq 2\pi)
    \end{dcases}
    \eea
    and the boundary conditions are summarized as
    \bea
    \text{NS sector}:&\quad  \psi_{+}(\tau,2\pi)=-\psi_{+}(\tau,0)\\
    \text{R sector}:&\quad \psi_{+}(\tau,2\pi)=+\psi_{+}(\tau,0).
    \eea
    \item DN boundary condition: We have
    \bea
    \text{NS sector}:&\quad \psi_{+}=+\psi_{-}\\
    \text{R sector}:& \quad \psi_{+}=-\psi_{-}
    \eea
    for $\sigma=0$ and 
    \bea
    \psi_{+}=+\psi_{-}
    \eea
    for $\sigma=\pi$. The doubling trick is performed as
    \bea
    \psi_{+}(\tau,\sigma)=\begin{dcases}
        \psi_{+}(\tau,\sigma)\qquad (0\leq \sigma\leq \pi)\\
        \psi_{-}(\tau,2\pi-\sigma)\qquad (\pi\leq \sigma\leq 2\pi)
    \end{dcases}
    \eea
    and the boundary conditions are summarized as
    \bea
    \text{NS sector}:&\quad  \psi_{+}(\tau,2\pi)=\psi_{+}(\tau,0)\\
    \text{R sector}:&\quad \psi_{+}(\tau,2\pi)=-\psi_{+}(\tau,0).
    \eea

    \item ND boundary condition: We have
    \bea
    \text{NS sector}:&\quad \psi_{+}=-\psi_{-}\\
    \text{R sector}:& \quad \psi_{+}=+\psi_{-}
    \eea
    for $\sigma=0$ and 
    \bea
    \psi_{+}=-\psi_{-}
    \eea
    for $\sigma=\pi$. The doubling trick is performed as
    \bea
    \psi_{+}(\tau,\sigma)=\begin{dcases}
        \psi_{+}(\tau,\sigma)\qquad (0\leq \sigma\leq \pi)\\
        -\psi_{-}(\tau,2\pi-\sigma)\qquad (\pi\leq \sigma\leq 2\pi)
    \end{dcases}
    \eea
    and the boundary conditions are summarized as
    \bea
    \text{NS sector}:&\quad  \psi_{+}(\tau,2\pi)=\psi_{+}(\tau,0)\\
    \text{R sector}:&\quad \psi_{+}(\tau,2\pi)=-\psi_{+}(\tau,0).
    \eea
\end{itemize}
For each cases, the mode expansions can be obtained by using \eqref{eq:NSRmodeexpand}.

\paragraph{Zero-point energy}Due to the Grassmann property of the fermion, the zero-point energy for the periodic and anti-periodic boundary conditions are computed as follows:
\bea
\text{Periodic b.c.}&\quad E_{0}=+\frac{1}{2}\sum_{n>0}n\{\psi_{n},\psi_{-n}\}=+\frac{1}{2}\zeta(-1,1)=-\frac{1}{24}\\
\text{Anti-periodic b.c}&\quad E_{0}=+\frac{1}{2}\sum_{r\in\mathbb{Z}+1/2,r>0}r\{\psi_{r},\psi_{-r}\}=\frac{1}{2}\zeta\left(-1,\frac{1}{2}\right)=+\frac{1}{48}.
\eea

Combing with the boson case, the zero-point energies are summarized as follows.
\bea\renewcommand{\arraystretch}{1.1}
\begin{tabular}{|c|c|c|}\hline
    \text{boundary conditions} & \text{NS sector} &\text{R sector} \\\hline
    \text{NN, DD} & $+1/24+1/48=1/16$  & $1/24-1/24=0$\\
    \text{ND, DN} & $-1/48-1/24=-1/16$ &  $-1/48+1/48=0$ \\\hline
\end{tabular}
\eea

\section{Light Cone Gauge}\label{app:lightconegauge}
Up to the previous section, we were focusing only when we have one boson and one fermion. In superstring theory, we have 10 copies of them and the worldsheet action for the open string is given as
\bea
\mathcal{L}=\frac{1}{2\pi}\left(\partial_{+}X^{\mu}\partial_{-}X_{\mu}+i\psi^{\mu}_{+}\partial_{-}\psi_{+\mu}+i\psi^{\mu}_{-}\partial_{+}\psi_{-\mu}\right)
\eea
where $\mu=0,1\ldots,9$ and $\sigma=[0,\pi]$. After the doubling trick, we can extend the domain to $\sigma\in[0,2\pi]$ and the energy-momentum tensor and supercurrents are defined as
\bea
T_{++}=\frac{1}{2}\partial_{+}X^{\mu}\partial_{+}X_{\mu}+\frac{i}{2}\psi^{\mu}_{+}\partial_{+}\psi_{\mu},\quad G_{++}=\psi_{+}^{\mu}\partial_{+}X_{\mu}.
\eea
To study the spectrum, we take the light-cone gauge. To make the discussion explicit, let us consider a setup with one D$p$-brane spanning in the $0123\cdots p$ direction. In this setup, the boundary conditions are
\bea
\text{NN}:&\quad X^{0},\ldots, X^{p}\\
\text{DD}:&\quad X^{p+1},\ldots,X^{9}.
\eea
Under these boundary conditions and the doubling trick, the bosons are all periodic. The fermions are all anti-periodic at the NS sector while they are all periodic at the R sector. We define 
\bea
X^{\pm}=\frac{1}{\sqrt{2}}(X^{0}\pm X^{1}),\quad \psi^{\pm}_{+}=\frac{1}{\sqrt{2}}(\psi^{0}_{+}+\psi^{1}_{+})
\eea
We assume $p\geq 1$ to take the light cone gauge 
\bea
X^{+}=4p^{+}\tau,\quad \psi_{+}^{+}=0
\eea
and we can solve $T_{++}=G_{++}=0$. Let us focus on $T_{++}=0$:
\bea
\partial_{+}X^{-}=\frac{1}{2p^{+}}\left(\frac{1}{2}\partial_{+}X^{i}\partial_{+}X^{i}+\frac{i}{2}\psi_{+}^{i}\partial_{+}\psi_{i}\right),\quad i=2,\ldots,9.
\eea
Studying the zero-modes, we have
\bea
\alpha_{0}^{-}=\frac{1}{2p_{+}}\left(\frac{1}{2}\sum_{j=2}^{p}\alpha_{0}^{j}\alpha_{0}^{j}+\hat{N}-E_{0}\right),\quad \hat{N}=\sum_{n=1}^{\infty}\sum_{i=2}^{9
}\alpha_{-n}^{i}\alpha_{n}^{i}+\sum_{r>0}\sum_{i=2}^{9}r\psi^{i}_{-r}\psi^{i}_{r}.
\eea
The zero-point energy is given as
\bea
E_{0}=\begin{dcases}
   \frac{1}{16}\times 8=\frac{1}{2} \qquad \text{NS sector}\\
   0 \qquad \qquad \qquad \text{R sector}.
\end{dcases}
\eea
We finally have the mass of the $(p+1)$-dimensional particles
\bea
m^{2}\coloneqq 2p^{+}p^{-}-\sum_{j=2}^{p}p^{j}p^{j}=\frac{\widehat{N}-E_{0}}{2}.
\eea

Let us study the spectrum explicitly. For the NS spectrum, the Fock vacuum $\ket{k}$ where $k$ here is the momentum eigenvalue gives
\bea
m^{2}=-\frac{E_{0}}{2}=-\frac{1}{4}
\eea
which is tachyonic. This state is removed by GSO projection by imposing $(-1)^{F_{\text{NS}}}\ket{k}=-\ket{k}$ and using the projection $(1+(-1)^{F_{\text{NS}}})/2$. The first excited states given by $N=1/2$ are massless
\bea
m^{2}=0.
\eea
The GSO projection preserves these states and we have $\psi^{i}_{-1/2}\ket{k}$ ($i=2,\ldots,9$). They will decompose into the vector representation of $\SO(p-1)$ (states coming from the longitudinal directions $\psi^{l}_{-1/2}$ with $l=2,\ldots, p$) and the scalar representation of $\SO(p-1)$ (states coming from the transverse directions $\psi^{l}_{-1/2}$ with $l=p+1,\ldots 9$). The particles $\psi_{-1/2}^{l}\,(l=2,\ldots,p)$ give the vector fields of the $\U(1)$ gauge theory on the D-brane while the particles $\psi_{-1/2}^{l}\,(l=p+1,\ldots,9)$ are scalar fields.

Let us move on to the R spectrum. In this case, the level $N=0$ states give
\bea
m^{2}=0.
\eea
The zero-modes $\psi^{i}_{0}\,(i=2,\ldots,p)$ form the spinor representation of $\SO(p-1)$. After GSO projection half of them remains and we have 8 states.

The massless states given above actually comes from the dimensional reduction of the 10 dimensional super Yang--Mills theory with gauge group $\U(1)$. We can generalize the story to parallel D-branes and the result will be a $\U(N)$ gauge theory where $N$ is the number of D-branes.

\chapter{Special functions}\label{app:specialfunct}
\section{\texorpdfstring{$q$}{q}-functions}
The $q$-Pochhammer symbol is defined as
\bea
(x;q)_{n}=\prod_{m=0}^{n-1}(1-xq^{m}).
\eea
We can formally take the limit $n\rightarrow \infty$ for $|q|<1$ as
\bea
(x;q)_{\infty}=\prod_{m=0}^{\infty}(1-xq^{m}),
\eea
which is equivalent to
\bea
(x;q)_{\infty}=\exp\left(-\sum_{m=1}^{\infty}\frac{x^{m}}{m(1-q^{m})}\right).
\eea
For the region $|q|>1$, using the analytic continuation, we have
\bea
(x;q)_{\infty}=(xq^{-1};q^{-1})^{-1}_{\infty}.
\eea
Note that we also have the following relation:
\bea
(x;q)_{n}=\frac{(x;q)_{\infty}}{(xq^{n};q)_{\infty}}=\exp\left(-\sum_{m=1}^{\infty}\frac{x^{m}}{m}\frac{1-q^{mn}}{1-q^{m}}\right).
\eea

We can also define the multiple variable version of the $q$-shifted factorial for the analytic region $|q_{1,\ldots,k}|<1$ as
\bea
(x;q_{1},\ldots,q_{k})_{\infty}&=\prod_{0\leq n_{1},\ldots, n_{k}\leq\infty}(1-xq_{1}^{n_{1}}\cdots q_{k}^{n_{k}}).
\eea
For other analytic regions, we use the formula
\bea
(x;q_{1},\ldots,q_{k})_{\infty}&=\exp\left(-\sum_{m=1}^{\infty}\frac{x^{m}}{m}\frac{1}{(1-q_{1}^{m})\cdots (1-q_{k}^{m})}\right).
\eea

\section{Delta function and commutation relations}
\begin{definition}
    The delta function $\delta(x)$ is a Laurent polynomial defined as
    \bea
    \delta(x)=\sum_{m\in\mathbb{Z}}x^{m}.
    \eea
\end{definition}

We have the following properties.
\begin{proposition}
    Let $f(z)$ be a rational function and then we have
    \bea
    f(z)\delta(z/a)=f(a)\delta(z/a).
    \eea
\end{proposition}
\begin{proposition}
    The delta function obeys
    \bea\label{app:eq:delta-commutative}
\delta(x)=\frac{1}{1-x}+\frac{x^{-1}}{1-x^{-1}}.
    \eea
Generally, we have
    \bea\label{app:eq:general-delta-commutative}
    \frac{\prod_{i}(1-\alpha_{i}z)}{\prod_{j}(1-\beta_{j}z)}-z^{|\{i\}|-|\{j\}|}\frac{\prod_{i}(z^{-1}-\alpha_{i})}{\prod_{j}(z^{-1}-\beta_{j})}=\sum_{k}\frac{\prod_{i}(1-\alpha_{i}/\beta_{k})}{\prod_{j\neq k}(1-\beta_{j}/\beta_{k})}\delta(\beta_{k}z).
    \eea
\end{proposition}

The first term of \eqref{app:eq:delta-commutative} is understood as a rational function in the analytic region $|x|<1$, while the second term is understood as a rational function in the analytic region $|x|>1$. The above relation is obtained by the following expansion
\bea
\frac{1}{1-x}&=\sum_{n=0}^{\infty}x^{n},\quad |x|<1,\\
\frac{x^{-1}}{1-x^{-1}}&=\sum_{n=1}^{\infty}x^{-n},\quad |x|>1.
\eea
Namely, the delta function captures the singularity coming from the pole.

Focusing on the case of \eqref{app:eq:general-delta-commutative} when we only have one numerator and one denominator, the above relation can be obtained as
\bea
\frac{1-az}{1-bz}=(1-az)\sum_{n=0}^{\infty}(bz)^{n},&\quad \frac{a}{b}\frac{1-a^{-1}z^{-1}}{1-b^{-1}z^{-1}}=\frac{a}{b}(1-a^{-1}z^{-1})\sum_{n=0}^{\infty}(bz)^{-n},\\
\frac{1-az}{1-bz}-\frac{a}{b}\frac{1-a^{-1}z^{-1}}{1-b^{-1}z^{-1}}&=\left(1-\frac{a}{b}\right)\delta(bz).
\eea

We also use the following notations. 
\begin{definition}\label{app:def:rational-expansion}
Let $f(x)$ be a rational function of a variable $x$. We denote the Laurent expansion in $x^{\mp 1}$ as
\bea
\relax[f(x)]^{x}_{\pm}=\sum_{m=m_{\pm }}^{\pm \infty}f^{\pm}_{m} x^{\mp m},
\eea
for some $m_{\pm }\in\mathbb{Z}$.
    
\end{definition}
\begin{example}
    For example, we will be using the above notation as
    \bea
    \left[\frac{1-ax}{1-bx}\right]^{x}_{+}=(1-ax)\sum_{n=0}^{\infty}(bx)^{n},\quad \left[\frac{1-ax}{1-bx}\right]^{x}_{-}=\frac{a}{b}(1-a^{-1}x^{-1})\sum_{n=0}^{\infty}(bx)^{-n}.
    \eea
\end{example}

\chapter{Formulas for partition functions and vertex operators}\label{app-chap:Partition-formulas}

\section{Structure functions}\label{app:structurefunct}
The structure functions in \eqref{eq:struct_funct} were defined as 
\bea\label{app-eq:struct_funct}
\mathscr{V}_{a}(x)&=\mathbb{I}[-\bfP_{a}^{\vee}x^{\vee}]=\frac{1-q_{a}x}{1-x},\\
\mathscr{S}_{ab}(x)&=\mathbb{I}[-\bfP_{ab}^{\vee}x^{\vee}]=\frac{(1-q_{a}x)(1-q_{b}x)}{(1-x)(1-q_{a}q_{b}x)},\\
g_{\bar{a}}(x)&=\mathbb{I}[-\bfP_{\bar{a}}^{\vee}x^{\vee}]=\frac{\prod_{i\neq a}(1-q_{i}x)(1-q_{\bar{a}}x)}{(1-x)\prod_{i\neq a}(1-q_{a}^{-1}q_{i}^{-1}x)},\\
\mathcal{A}_{\mathbb{C}^{4}}(x)&=\mathbb{I}[-\bfP_{\four}^{\vee}x^{\vee}]=\frac{\prod_{i=1}^{4}(1-q_{i}x)\prod^{4}_{i=1}(1-q_{i}^{-1}x)}{(1-x)^{2}\prod_{i\neq j}(1-q_{i}q_{j}x)}.
\eea
We have the following properties
\bea
    \mathscr{S}_{ab}(x)=\frac{\mathscr{V}_{a}(x)}{\mathscr{V}_{a}(q_{b}x)},\quad g_{abc}(x)=\frac{\mathscr{S}_{ab}(x)}{\mathscr{S}_{ab}(q_{c}x)},\quad \mathcal{A}_{\mathbb{C}^{4}}(x)=\frac{g_{\bar{a}}(x)}{g_{\bar{a}}(q_{a}x)}
\eea
and the reflection formulas
\bea
\mathscr{V}_{a}(x)=q_{a}\mathscr{V}_{a}(q_{a}^{-1}x^{-1})^{-1},\quad \mathscr{S}_{ab}(x)=\mathscr{S}_{ab}(q_{ab}^{-1}x^{-1}),\quad g_{\bar{a}}(x)=g_{\bar{a}}(q_{a}x^{-1})^{-1},\quad \mathcal{A}_{\mathbb{C}^{4}}(x)=\mathcal{A}_{\mathbb{C}^{4}}(x^{-1}).
\eea


Structure functions that appear at the recursion relations of the partition functions are
\bea\label{eq:app-structurefunct}
\mathscr{U}^{a}_{k,v}(x)&=\left(1-\frac{v}{x}\right)\prod_{\Abox \in k}\mathscr{V}_{a}\left(\frac{\chi_{a,v}(\Bbox)}{x}\right)=\left(1-\frac{vq_{a}^{k}}{x}\right),\\
\mathscr{Y}^{A}_{\lambda,v}(x)&=\left(1-\frac{v}{x}\right)\prod_{\Abox\in\lambda}\mathscr{S}_{A}\left(\frac{\chi_{A,v}(\Bbox)}{x}\right)=\frac{\prod_{\Abox\in A(\lambda)}\left(1-\chi_{A,v}(\Bbox)/x\right)}{\prod_{\Abox\in R(\lambda)} (1-q_{A}\chi_{A,v}(\Bbox)/x)},\\
\mathscr{W}^{\bar{a}}_{\pi,v}(x)&=\left(1-\frac{v}{x}\right)\prod_{\scube\in \pi}g_{\bar{a}}\left(\frac{\chi_{\bar{a},v}(\cube)}{x}\right)\propto \prod_{\scube\in A(\pi)}\left(1-\frac{\chi_{\bar{a},v}(\cube)}{x}\right)\prod_{\scube\in R(\pi)}\left(1-q_{a}^{-1}\frac{\chi_{\bar{a},v}(\cube)}{x}\right).
\eea
and
\bea\label{eq:app-dualfunctions}
\mathscr{U}^{a\,\vee}_{k,v}(x)&=\left(1-\frac{x}{v}\right)\prod_{\Abox\in k}\mathscr{V}_{a}\left(q_{a}^{-1}\frac{x}{\chi_{a,v}(\Bbox)}\right)^{-1}=\left(1-\frac{x}{vq_{a}^{k}}\right),\\
\mathscr{Y}^{A\,\vee}_{\lambda,v}(x)&=\left(1-\frac{x}{v}\right)\prod_{\Abox\in \lambda}\mathscr{S}_{A}\left(q_{A}^{-1}\frac{x}{\chi_{A,v}(\Bbox)}\right),\\
\mathscr{W}^{\bar{a}\,\vee}_{\pi,v}(x)&=\left(1-\frac{x}{v}\right)\prod_{\scube\in\pi}g_{\bar{a}}\left(q_{a}\frac{x}{\chi_{\bar{a},v}(\cube)}\right)^{-1}.
\eea
Note that we have the following reflection formulas 
\bea\label{eq:app-dualreflection}
\mathscr{U}^{a\,\vee}_{k,v}(x)=\left(-\frac{x}{vq_{a}^{k}}\right)\mathscr{U}^{a}_{k,v}(x),&\qquad
\mathscr{Y}^{A\,\vee}_{\lambda,v}(x)=\left(-\frac{x}{v}\right)\mathscr{Y}^{A}_{\lambda,v}(x),\\
\mathscr{W}^{\bar{a}\,\vee}_{\pi,v}(x)=\left(-\frac{x}{v}\right)\mathscr{W}^{\bar{a}}_{\pi,v}(x),&\qquad
\mathscr{M}^{K\,\vee}_{\rho,v}(x)=K\mathscr{M}^{K}_{\rho,v}(x).
\eea

\section{Partition functions for 5d theory}

\paragraph{Nekrasov factors for 5d theory}
Using the quadrality in $q_{a}\,(a\in\four)$, we define similar Nekrasov factors which can be used to discuss gauge theories on $\mathbb{C}^{2}_{A}\times \mathbb{S}^{1}\,(A\in\six)$:
\bea
 \mathsf{N}_{A}(v_{1},\lambda^{(1)}\,|\,v_{2},\lambda^{(2)})&=\prod_{\Abox\in\lambda^{(1)}}\left(1-\frac{q_{A}\chi_{A,v_{1}}(\Bbox)}{v_{2}}\right)\prod_{\Abox\in\lambda^{(2)}}\left(1-\frac{v_{1}}{\chi_{A,v_{2}}(\Bbox)}\right)\prod_{\substack{\Abox\in\lambda^{(1)}\\\AboxF\in\lambda^{(2)}}}\mathscr{S}_{A}\left(\frac{\chi_{A,v_{1}}(\Bbox)}{\chi_{A,v_{2}}(\BboxF)}\right),\label{eq:D4Nekrasovfactor}
\eea
where we defined the box content as
\bea
    \chi_{ab,v}(\Bbox)=vq_{a}^{i-1}q_{b}^{j-1},\quad A=ab\in\six,\quad \Bbox=(i,j)\in\lambda.
\eea
\begin{lemma}
The recursion formulas of the 5d Nekrasov factors are given as follows,
\bea\label{eq:5dNekrasovrecursion}
    &\frac{\mathsf{N}_{A}(v_{1},\lambda^{(1)}+\Bbox\,|\,v_{2},\lambda^{(2)})}{\mathsf{N}_{A}(v_{1},\lambda^{(1)}\,|\,v_{2},\lambda^{(2)})}=\mathscr{Y}^{A\,\vee}_{\lambda^{(2)},v_{2}}(q_{A}\chi_{A,v_{1}}(\Bbox))\\
   &\frac{\mathsf{N}_{A}(v_{1},\lambda^{(1)}\,|\,v_{2},\lambda^{(2)}+\BboxF)}{\mathsf{N}_{A}(v_{1},\lambda^{(1)}\,|\,v_{2},\lambda^{(2)})}=\mathscr{Y}_{\lambda^{(1)},v_{1}}^{A}(\chi_{A,v_{2}}(\BboxF)),
\eea
where we define the $\mathscr{Y}$-functions,
\bea
   \mathscr{Y}_{\lambda,v}^{A}(x)&=\left(1-\frac{v}{x}\right)\prod_{\Abox\in\lambda}\mathscr{S}_{A}\left(\frac{\chi_{A,v}(\Bbox)}{x}\right)=\frac{\prod\limits_{\Abox\in A(\lambda)}\left(1-\chi_{A,v}(\Bbox)/x\right)}{\prod\limits_{\Abox\in R(\lambda)}\left(1-q_{A}\chi_{A,v}(\Bbox)/x\right)},\\
   \mathscr{Y}^{A\,\vee}_{\lambda,v}(x)&=\left(1-\frac{x}{v}\right)\prod_{\Abox\in \lambda}\mathscr{S}_{A}\left(q_{A}^{-1}\frac{x}{\chi_{A,v}(\Bbox)}\right)=\frac{\prod\limits_{\Abox\in A(\lambda)}\left(1-x/\chi_{A,v}(\Bbox)\right)}{\prod\limits_{\Abox\in R(\lambda)}\left(1-q_{A}^{-1}x/\chi_{A,v}(\Bbox)\right)}
\eea    
\end{lemma}


\begin{theorem}\label{app:thm-D4recursion}
The recursion relation of the $\U(1)$ partition function of the 5d $\mathcal{N}=1^{\ast}$ theory is written using the residues of the $\mathscr{Y}$-functions as
\begin{equation}
    \frac{\widetilde{\mathcal{Z}}^{\D4}_{A}[\lambda+\Bbox]}{\widetilde{\mathcal{Z}}^{\D4}_{A}[\lambda]}=-\frac{\underset{x=q_{\text{inf}(\bar{A})}^{-1}\chi_{A,v}(\Abox)}{\Res}{x}^{-1}\frac{\mathscr{Y}_{\lambda,v}^{A}(x)}{\mathscr{Y}_{\lambda,v}^{A}(q_{\text{inf}(\bar{A})}x)}}{\underset{x=q_{A}\chi_{A,v}(\Abox)}{\Res}{x}^{-1}\frac{\mathscr{Y}_{\lambda+\Abox,v}^{A}(x)}{\mathscr{Y}_{\lambda+\Abox,v}^{A}(q_{\text{inf}(\bar{A})}x)}}.
\end{equation}    
The recursion relation above can be written in a rather symmetric way: 
\bea\label{eq:app-D4U1recursionformula}
    \frac{\widetilde{\mathcal{Z}}^{\D4}_{ab}[\lambda+\Bbox]}{\widetilde{\mathcal{Z}}
^{\D4}_{ab}[\lambda]}&=-\underset{x\rightarrow q_{c}^{-1}\chi_{ab,v}(\Abox)}{\lim}\left(1-\frac{q_{c}^{-1}\chi_{ab,v}(\Bbox)}{x}\right)\frac{\mathscr{Y}_{\lambda,v}^{ab}(x)}{\mathscr{Y}^{ab}_{\lambda,v}(q_{c}x)}\\
&\qquad \times\left(\underset{x\rightarrow q_{ab}\chi_{ab,v}(\Abox)}{\lim}\left(1-\frac{q_{ab}\chi_{ab,v}(\Bbox)}{x}\right)\frac{\mathscr{Y}^{ab}_{\lambda+\Abox,v}(x)}{\mathscr{Y}^{ab}_{\lambda+\Abox,v}(q_{c}x)}\right)^{-1}\\
&=-\frac{\mathscr{Y}_{\lambda,v}^{ab}(q_{c}^{-1}\chi_{ab,v}(\Bbox))\mathscr{Y}^{ab}_{\lambda+\Abox,v}(q_{d}^{-1}\chi_{ab,v}(\Bbox))}{\mathscr{Y}^{ab}_{\lambda,v}(\chi_{ab,v}(\Bbox))\mathscr{Y}^{ab}_{\lambda+\Abox,v}(q_{cd}^{-1}\chi_{ab,v}(\Bbox))}
\eea
where $A=ab, \bar{A}=cd\,(c<d)$ and $\Bbox\in A(\lambda)$. The numerators come from the fact that they have no pole when taking the limit. For the denominators, both of them are singular but the singular part will cancel with each other and the above formula itself is well-defined.
\end{theorem}

\section{Partition functions for 7d theory}

\paragraph{Nekrasov factors for 7d theory}
The 7d analogues of the 5d Nekrasov factors are defined as
\beq
\mathsf{N}_{\bar{a}}(v_{1},\pi_{1}\,|\,v_{2},\pi_{2})=\frac{\prod_{\scube\in\pi_{2}}(1-v_{1}/\chi_{\bar{a},v_{2}}(\cube))}{\prod_{\scube\in\pi_{1}}(1-q_{a}^{-1}\chi_{\bar{a},v_{1}}(\cube)/v_{2})}\prod_{\substack{\scube\in\pi_{1}\\\scubeF\in\pi_{2}}}g_{\bar{a}}\left(\frac{\chi_{\bar{a},v_{1}}(\cube)}{\chi_{\bar{a},v_{2}}(\cubeF)}\right)
\eeq

\begin{lemma}
We have the following recursion relations for the 7d Nekrasov factor:
\bea\label{eq:D6Nekrasovrecursion}
    \frac{\mathsf{N}_{\bar{a}}(v_{1},\pi_{1}+\cube\,|\,v_{2},\pi_{2})}{\mathsf{N}_{\bar{a}}(v_{1},\pi_{1}\,|\,v_{2},\pi_{2})}&=\mathscr{W}^{\bar{a}\,\vee}_{\pi_{2},v_{2}}(q_{a}^{-1}\chi_{\bar{a},v_{1}}(\cube))^{-1}\\
    \frac{\mathsf{N}_{\bar{a}}(v_{1},\pi_{1}\,|\,v_{2},\pi_{2}+\cubeF)}{\mathsf{N}_{\bar{a}}(v_{1},\pi_{1}\,|\,v_{2},\pi_{2})}&=\mathscr{W}^{\bar{a}}_{\pi_{1},v_{1}}(\chi_{\bar{a},v_{2}}(\cubeF)),
\eea
where
\bea\label{eq:D6Nekrasov-shell}
    \mathscr{W}_{\pi,v}^{\bar{a}}(x)=(1-v/x)\prod_{\scube\in\pi}g_{\bar{a}}\left(\chi_{\bar{a},v}(\cube)/x\right)\propto \prod_{\cube\in A(\pi)}\left(1-\chi_{\bar{a},v}(\cube)/x\right)\prod_{\scube\in R(\pi)}\left(1-q_{a}^{-1}\chi_{\bar{a},v}(\cube)/x\right)
\eea
and 
\bea
\mathscr{W}^{\bar{a}\,\vee}_{\pi,v}(x)=\left(1-\frac{x}{v}\right)\prod_{\scube\in\pi}g_{\bar{a}}\left(q_{a}\frac{x}{\chi_{\bar{a},v}(\cube)}\right)^{-1}.
\eea
We also have 
\bea
    \frac{\mathsf{N}_{\bar{a}}(v,\pi+\cube\,|\,v,\pi+\cube)}{\mathsf{N}_{\bar{a}}(v,\pi\,|\,v,\pi)}=\left(\frac{q_{a}v}{\chi_{\bar{a},v}(\cube)}\right)\frac{\underset{x=\chi_{\bar{4},v}(\scube)}{\Res}x^{-1}\mathscr{W}^{\bar{4}}_{\pi+\scube,v}(q_{4}^{-1}x)^{-1}}{\underset{x=\chi_{\bar{4},v}(\scube)}{\Res}x^{-1}\mathscr{W}^{\bar{4}}_{\pi,v}(x)^{-1}}.
\eea    
\end{lemma}
We can define the vector $\U(1)$ contribution of the $\D6$ theory as 
\bea
    \mathcal{Z}^{\D6}_{\bar{a}}[\pi]=\frac{1}{\mathsf{N}_{\bar{a}}(v,\pi\,|\,v,\pi)}
\eea
which resembles the partition function of the pure SYM in the 5d theory. The two factors $\mathcal{Z}_{\bar{a}}^{\D6}[\pi]$ and $\widetilde{\mathcal{Z}}_{\bar{a}}^{\D6}[\pi]$ differ by extra Chern--Simons like term and topological term
\bea
    \mathcal{Z}_{\bar{a}}^{\D6}[\pi]=\prod_{\scube\in\pi}\left(-\frac{\chi_{\bar{a},v}(\cube)}{q_{a}v}\right)\widetilde{\mathcal{Z}}_{\bar{a}}^{\D6}[\pi].
\eea

\begin{theorem}\label{eq:app-thm-D6U1recursionformula}
    The recursion formula of the $\U(1)$ partition function $\widetilde{\mathcal{Z}}^{\D6}_{\bar{a}}[\pi]$ is related with the functions $\mathscr{W}^{\bar{a}}_{\pi,v}(x)$:
    \bea
        \frac{\widetilde{\mathcal{Z}}^{\D6}_{\bar{a}}[\pi+\cube]}{\widetilde{\mathcal{Z}}^{\D6}_{\bar{a}}[\pi]}=-\frac{\underset{x'=q_{a}\chi_{\bar{a},v}(\scube)}{\Res}x'^{-1}\mathscr{W}^{\bar{a}}_{\pi,v}(q_{a}^{-1}x')^{-1}}{\underset{x'=\chi_{\bar{a},v}(\scube)}{\Res}x'^{-1}\mathscr{W}^{\bar{a}}_{\pi+\scube,v}(q_{a}^{-1}x')^{-1}}
    \eea
\end{theorem}

\section{Operator products of vertex operators}
\begin{proposition}[Contraction formulas]\label{app:prop:contraction_formula}
Under the above free field realizations of the zero modes, the contraction formulas are
\begin{align}
\mathsf{A}(x)\mathsf{S}_{a}(x')&=g_{\bar{a}}\left(x'/x\right)^{-1}: \mathsf{A}(x)\mathsf{S}_{a}(x'):,\\
\mathsf{S}_{a}(x')\mathsf{A}(x)&=g_{\bar{a}}(q_{a}x/x'):\mathsf{A}(x)\mathsf{S}_{a}(x'):,\\
\StepSubequations
\mathsf{A}(x)\mathsf{X}_{A}(x')&=\mathscr{S}_{\bar{A}}(x'/x)^{-1}:\mathsf{A}(x)\mathsf{X}_{A}(x'):,\\
    \mathsf{X}_{A}(x')\mathsf{A}(x)&=\mathscr{S}_{\bar{A}}(q_{A}x/x')^{-1}:\mathsf{X}_{A}(x')\mathsf{A}(x):,\\
\StepSubequations
    \mathsf{A}(x)\mathsf{W}_{\bar{a}}(x')&=\mathscr{V}_{a}\left(x'/x\right)^{-1}:\mathsf{A}(x)\mathsf{W}_{\bar{a}}(x'):,\\
    \mathsf{W}_{\bar{a}}(x')\mathsf{A}(x)&=q_{a}^{-1}\mathscr{V}_{a}(q_{a}^{-1}x/x'):\mathsf{W}_{\bar{a}}(x')\mathsf{A}(x):\\
\StepSubequations
    \mathsf{A}(x)\mathsf{Z}(K,x')&=\frac{1-x'/x}{1-Kx'/x}:\mathsf{A}(x)\mathsf{Z}(K,x'):,\\
    \mathsf{Z}(K,x')\mathsf{A}(x)&=K^{-1}\frac{1-x/x'}{1-K^{-1}x/x'}:\mathsf{A}(x)\mathsf{Z}(K,x'):,\\
\StepSubequations
    \mathsf{S}_{a}(x)\mathsf{S}_{b}(x')&=\mathscr{S}_{\overbar{ab}}(q_{a}x'/x):\mathsf{S}_{a}(x)\mathsf{S}_{b}(x'):,\\
        \mathsf{S}_{b}(x')\mathsf{S}_{a}(x)&=\mathscr{S}_{\overbar{ab}}(q_{b}x/x'):\mathsf{S}_{a}(x)\mathsf{S}_{b}(x'):,\\
\StepSubequations
        \mathsf{X}_{A}(x)\mathsf{S}_{c}(x')&=\mathscr{V}_{d}\left(q_{A}x'/x\right)^{-1} : \mathsf{X}_{A}(x)\mathsf{S}_{c}(x'):  ,\\
        \mathsf{S}_{c}(x')\mathsf{X}_{A}(x)&=q_{d}^{-1}\mathscr{V}_{d}\left(q_{d}^{-1}q_{A}^{-1}x/x'\right): \mathsf{X}_{A}(x)\mathsf{S}_{c}(x'):,\\
\StepSubequations
\mathsf{W}_{\bar{a}}(x)\mathsf{S}_{a}(x')&=x'\frac{1}{1-q_{a}^{-1}x'/x}:\mathsf{W}_{\bar{a}}(x)\mathsf{S}_{a}(x'):,\\
\mathsf{S}_{a}(x')\mathsf{W}_{\bar{a}}(x)&=(-q_{a}x)\frac{1}{1-q_{a}x/x'}:\mathsf{W}_{\bar{a}}(x)\mathsf{S}_{a}(x'):
\end{align}\label{eq:app-contractions}
\end{proposition}

\chapter{Technical proofs}\label{app-chap:proofs}

\section{Proof of sign rules of the magnificent four system}\label{app:M4signproof}
In this section, we give the proof of Thm.~\ref{thm:D8signruleJK} following the original paper \cite{Nekrasov:2018xsb}.
\begin{theorem}\label{app:D8signruleJK-proof}
The JK-residue of the integrand obeys the identity:
    \bea
\mathcal{G}^{k}\times \underset{x=x_{\rho}}{\Res}x^{-1}\mathcal{Z}_{k}^{\D8\tbar\D0}(v,Kv,x_{I})=(-1)^{\sigma_{a}(\rho)}\mathcal{Z}^{\D8}_{\four;a}[\rho,K].
    \eea
\end{theorem}
\begin{proof}
We show this by induction. For the $|\rho|=1$ case, this is shown in \eqref{eq:D8one-inst}. We assume that the claim is true for $|\rho|=k$ and let us consider the case for $\tilde{\rho}=\rho+\chi$, where $\chi$ is an additional box to the solid partition $\rho$. We denote $\{x_{I}^{\ast}\}_{I=1}^{k}$ the $q$-coordinates of the boxes in the solid partition which is given in the monomial ordering \eqref{eq:ordering-partition}. The box $\chi\in\rho$ is identified with the $q$-coordinate of it. The left hand side is then given as 
\bea
\,&\mathcal{G}^{k+1}\underset{x=x_{\rho}}{\Res}x^{-1}\mathcal{Z}_{k}^{\D8\tbar\D0}(v,Kv,x_{I})\\
=&\underset{x_{k+1}=\chi}{\Res}x_{k+1}^{-1}\underset{x=x_{k}^{\ast}}{\Res}x_{k}^{-1}\cdots \underset{x=x_{1}^{\ast}}{\Res}x_{1}^{-1}\mathbb{I}\left[-(1-K^{-1})v^{-1}\sum_{I=1}^{k}x_{I}+\bfP_{123}^{\vee}\sum_{I,J=1}^{k}x_{I}^{-1}x_{J}^{-1}-k\right.\\
&\left.-(1-K^{-1})v^{-1}x_{k+1}+\bfP_{123}^{\vee}x_{k+1}^{-1}\sum_{I=1}^{k}x_{I}+\bfP_{123}^{\vee}\sum_{I=1}^{k}x_{I}^{-1}x_{k+1}+\bfP_{123}^{\vee}-1 \right]\\
&=(-1)^{\sigma_{4}(\rho)}\mathcal{Z}^{\D8}_{\four;4}[\rho,K]\times \mathbb{I}[\bfP_{123}^{\vee}-1]\\
&\qquad \times\underset{x_{k+1}=\chi}{\Res}x_{k+1}^{-1}\mathbb{I}\left[-(1-K^{-1})v^{-1}x_{k+1}+\bfP_{123}^{\vee}x_{k+1}^{-1}\sum_{I=1}^{k}x^{\ast}_{I}+\bfP_{123}^{\vee}\sum_{I=1}^{k}{x^{\ast}_{I}}^{-1}x_{k+1}\right].
\eea
Using Prop.~\ref{prop:residuegeneral}, the extra sign factor comes from the term $\bfP_{123}^{\vee}\sum_{I=1}^{k}x^{\ast}_{I}x_{k+1}^{-1}$. If this term is unmovable, then it means that it contains the pole $x_{k+1}=\chi$ and thus after taking the residue we obtain an extra sign factor. Let us determine the condition so that $\chi$ gives an extra sign factor.

We introduce
\bea
\bfK_{\rho}=\sum_{I=1}^{k}x_{I}^{\ast}=\sum_{i=1}^{\mu_{4}}vq_{4}^{i-1}\bm{\pi}_{i}(q_{1},q_{2},q_{3})
\eea
where $\bm{\pi}_{i}$ is the character of plane partitions $\{\pi_{i}\}$ obeying the non-increasing condition and $\mu_{4}\in\mathbb{Z}_{\geq 1}$ is the height of the solid partition in the 4th direction. Since the JK residue is taken in a monomial ordering, we have two possibilities on how $\chi$ is added to the solid partition:
\begin{enumerate}[topsep=1.5pt,itemsep=1ex,partopsep=1ex,parsep=1.5ex]
    \item $\chi=q_{4}^{\mu_{4}}v$: A new box is added to the solid partition and it increases the height of the solid partition.
    \item $\chi=vq_{4}^{\mu_{4}-1}\eta$: The box $\eta=q_{1}^{\mu_{1}-1}q_{2}^{\mu_{2}-1}q_{3}^{\mu_{3}-1}\,\,(\mu_{1}\mu_{2}\mu_{3}\geq 2)$ is a new box added to the plane partition $\pi_{h}$. Note that from the assumption of the monomial ordering \eqref{eq:ordering-partition}, it is enough to consider only at this level.
\end{enumerate}
The first situation is computed as
\bea
\bfP_{123}^{\vee}\chi^{-1}\bfK_{\rho}&=\bfP_{123}^{\vee}\sum_{i=1}^{\mu_{4}}q_{4}^{i-1-\mu_{4}}\bm{\pi}_{i}\\
&=-\bfP_{123}\sum_{i=1}^{\mu_{4}}(q_{123})^{\mu_{4}-i}\bm{\pi}_{i}.
\eea
Since $\bm{\pi}_{i}$ takes the form as $q_{1}^{\geq 0}q_{2}^{\geq 0}q_{3}^{\geq 0}$, when $1\leq i\leq \mu_{4}-1$, this term is movable. When $i=\mu_{4}$, since there is a box at $\chi=(1,1,1,\mu_{4}+1)$, we also have a box at $(1,1,1,\mu_{4})$ due to the solid partition condition and thus $\bm{\pi}_{\mu_{4}}=1+\cdots$. The unmovable term is given $[-\bfP_{123}\bm{\pi}_{\mu_{4}}]^{(0)}=-1$ and thus we have an extra sign factor.

For the second situation, it is computed as
\bea
\bfP_{123}^{\vee}\chi^{-1}\bfK_{\rho}&=\bfP_{123}^{\vee}\sum_{i=1}^{\mu_{4}}q_{4}^{i-\mu_{4}}\eta^{-1}\bm{\pi}_{i}.
\eea
Note that from the solid partition condition $\eta\in\bm{\pi}_{i}$ for $i\leq \mu_{4}-1$ and $\eta\notin\bm{\pi}_{\mu_{4}}$. We decompose $\bm{\pi}_{i}$ as the following
\bea
\bm{\pi}_{i}=\begin{dcases}
    \Delta_{i}(\eta)+\bm{\pi}_{i}(\eta)\quad i=1,\ldots, \mu_{4}-1,\\
    \Delta_{\mu_{4}}(\eta)+\bm{\pi}_{\mu_{4}}(\eta)-\eta,\quad i=\mu_{4}
\end{dcases}
\eea
where
\bea
\Delta_{i}(\eta)=\sum_{a=1}^{\mu_{1}}\sum_{b=1}^{\mu_{2}}\sum_{c=1}^{\mu_{3}}q_{1}^{a-1}q_{2}^{b-1}q_{3}^{c-1}
\eea
and $\bm{\pi}(\eta)$ contains terms with $\bm{\pi}_{i}\ni(a,b,c),\,(a>\mu_{1})\vee(b>\mu_{2})\vee(c>\mu_{3})$. First of all, the term coming from $\bm{\pi}_{i}(\eta)$ is movable. We first have
\bea
\bfP_{123}^{\vee}\sum_{i=1}^{\mu_{4}}(q_{123})^{\mu_{4}-i}\bm{\pi}(\eta)q_{1}^{-\mu_{1}+1}q_{2}^{-\mu_{2}+1}q_{1}^{-\mu_{3}+1}=-\bfP_{123}\sum_{i=1}^{\mu_{4}}\sum_{(a,b,c)\in\pi_{i}(\eta)}(q_{123})^{\mu_{4}-i-1}q_{1}^{a-\mu_{1}}q_{2}^{b-\mu_{2}}q_{3}^{c-\mu_{3}}.
\eea
For $i\leq \mu_{4}-1$, the terms have powers with strictly positive power in either $q_{1,2,3}$. For $i=\mu_{4}$:
\bea
\left[-\bfP_{123}\sum_{(a,b,c)\in\pi_{i}(\eta)}(q_{123})^{-1}q_{1}^{a-\mu_{1}}q_{2}^{b-\mu_{2}}q_{3}^{c-\mu_{3}}\right]^{(0)}=0.
\eea
This can be derived as follows. Since the plane partition condition obeys $(a,b,c)<(\mu_{1},\mu_{2},\mu_{3})$, without loss of generality we can assume $c\leq \mu_{3}$. Assuming $a>\mu_{1}$, we have $c=\mu_{3},b<\mu_{2}$ or $c<\mu_{3}$. For the case $c<\mu_{3}$, we have $q_{3}^{c-\mu_{3}-1\leq -2}$ and thus we have no unmovable terms. Similarly, for the case $c=\mu_{3},b<\mu_{2}$, it will be movable because of the term $q_{2}^{b-\mu_{2}-1\leq -2}$.

The remaining term is computed as 
\bea
\left[\bfP_{123}^{\vee}\eta^{-1}\sum_{i=1}^{\mu_{4}}q_{4}^{i-\mu_{4}}\Delta_{i}(\eta)-\bfP_{123}^{\vee}\right]^{(0)}&=\left[\sum_{i=1}^{\mu_{4}}q_{4}^{i-\mu_{4}}(1-q_{1}^{-\mu_{1}})(1-q_{2}^{-\mu_{2}})(1-q_{3}^{-\mu_{3}})-\bfP_{123}^{\vee}\right]^{(0)}\\
&=\left[-\sum_{i=1}^{\mu_{4}-1}q_{4}^{i-\mu_{4}}q_{1}^{-\mu_{1}}q_{2}^{-\mu_{2}}q_{3}^{-\mu_{3}}+1-\bfP_{123}^{\vee}\right]\\
&=-\sum_{i=1}^{\mu_{4}-1}\delta_{\mu_{1}=\mu_{2}=\mu_{3}=i}.
\eea

Combining both cases, the extra sign factor gives
\bea
\sigma_{4}(\tilde{\rho})=\sigma_{4}(\rho)+\sum_{i=1}^{\mu_{4}-1}\delta_{\mu_{1}=\mu_{2}=\mu_{3}=i}
\eea
which concludes the proof.

\end{proof}

\section{Proof of identity \texorpdfstring{\eqref{eq:PEidentityproof}}{(4.4.86)}}
\begin{lemma}\label{app:lemm-PEid}
We have the identity
\bea
\sum_{\alpha=1}^{n}\frac{1-K_{\alpha}^{-1}}{(1-\mathfrak{q}_{\alpha})(1-K_{\alpha}^{-1}\mathfrak{q}_{\alpha}^{-1})}=\frac{1-\prod_{\alpha=1}^{n}K_{\alpha}^{-1}}{(1-\mathfrak{q})(1-\prod_{\alpha=1}^{n}K_{\alpha}^{-1}\mathfrak{q}^{-1})},
\eea
where $\mathfrak{q}_{\alpha}=\mathfrak{q}\prod_{\beta=\alpha+1}^{n}K_{\beta}$.
\end{lemma}
\begin{proof}
    We show it by induction. When $n=1$, it is trivial. Assume that for all $\leq n$, the above identity is true. For $n+1$, we have
    \bea
    \sum_{\alpha=1}^{n+1}\frac{1-K_{\alpha}^{-1}}{(1-\tilde{\mathfrak{q}}_{\alpha})(1-K_{\alpha}^{-1}\tilde{\mathfrak{q}}_{\alpha}^{-1})}=\sum_{\alpha=1}^{n}\frac{1-K_{\alpha}^{-1}}{(1-\tilde{\mathfrak{q}}_{\alpha})(1-K_{\alpha}^{-1}\tilde{\mathfrak{q}}_{\alpha}^{-1})}+\frac{1-K_{n+1}^{-1}}{(1-\tilde{\mathfrak{q}}_{n+1})(1-K_{n+1}^{-1}\tilde{\mathfrak{q}}_{n+1}^{-1})}
    \eea
    where $\tilde{\mathfrak{q}}_{\alpha}=\mathfrak{q}\prod_{\beta=\alpha+1}^{n+1}K_{\beta}=\mathfrak{q}_{\alpha}K_{n+1}$ for $\alpha<n+1$ and $\tilde{\mathfrak{q}}_{n+1}=\mathfrak{q}$. Then, by the inductive hypothesis, we get
    \bea
    \,&\frac{1-\prod_{\alpha=1}^{n}K_{\alpha}^{-1}}{(1-\mathfrak{q}K_{n+1})(1-\prod_{\alpha=1}^{n+1}K_{\alpha}^{-1}\mathfrak{q}^{-1})}+\frac{1-K_{n+1}^{-1}}{(1-\mathfrak{q})(1-K_{n+1}^{-1}\mathfrak{q}^{-1})}\\
    =&\frac{1}{1-\mathfrak{q}K_{n+1}}\left(\frac{1-\prod_{\alpha=1}^{n}K_{\alpha}^{-1}}{1-\prod_{\alpha=1}^{n+1}K_{\alpha}^{-1}\mathfrak{q}^{-1}}  -\frac{1-K_{n+1}}{1-\mathfrak{q}^{-1}} \right)\\
    =&\frac{(1-\prod_{\alpha=1}^{n}K_{\alpha}^{-1})(1-\mathfrak{q}^{-1})-(1-K_{n+1})(1-\prod_{\alpha=1}^{n+1}K_{\alpha}^{-1}\mathfrak{q}^{-1})}{(1-\mathfrak{q}K_{n+1})(1-\prod_{\alpha=1}^{n+1}K_{\alpha}^{-1}\mathfrak{q}^{-1})(1-\mathfrak{q}^{-1})}\\
    =&\frac{(1-\prod_{\alpha=1}^{n+1}K_{\alpha}^{-1})}{(1-\mathfrak{q})(1-\prod_{\alpha=1}^{n+1}K_{\alpha}^{-1}\mathfrak{q}^{-1})}
    \eea
    which gives the claim.
\end{proof}

\bibliographystyle{ytamsalpha}
\baselineskip=.95\baselineskip
\bibliography{Worigami}

\newcommand{\etalchar}[1]{$^{#1}$}
\providecommand{\bysame}{\leavevmode\hbox to3em{\hrulefill}\thinspace}
\providecommand{\MR}{\relax\ifhmode\unskip\space\fi MR }
\providecommand{\MRhref}[2]{%
  \href{http://www.ams.org/mathscinet-getitem?mr=#1}{#2}
}
\providecommand{\href}[2]{#2}
\providecommand{\doihref}[2]{\href{#1}{#2}}
\providecommand{\arxivfont}{\tt}
\begin{thebibliography}{CWWY15}

\bibitem[AB84]{Atiyah:1984px}
M.~F. Atiyah and R.~Bott, \emph{{The Moment map and equivariant cohomology}}, \doihref{http://dx.doi.org/10.1016/0040-9383(84)90021-1}{Topology \textbf{23} (1984) 1--28}.

\bibitem[AB14]{Alkalaev:2014sma}
K.~B. Alkalaev and V.~A. Belavin, \emph{{Conformal blocks of $W_N$ minimal models and AGT correspondence}}, \doihref{http://dx.doi.org/10.1007/JHEP07(2014)024}{JHEP \textbf{07} (2014) 024}, \href{http://arxiv.org/abs/1404.7094}{{\arxivfont arXiv:1404.7094 [hep-th]}}.

\bibitem[ACNY87]{Abouelsaood:1986gd}
A.~Abouelsaood, C.~G. Callan, Jr., C.~R. Nappi, and S.~A. Yost, \emph{{Open strings in background gauge fields}}, \doihref{http://dx.doi.org/10.1016/0550-3213(87)90164-7}{Nucl. Phys. B \textbf{280} (1987) 599--624}.

\bibitem[AFS11]{Awata:2011ce}
H.~Awata, B.~Feigin, and J.~Shiraishi, \emph{{Quantum Algebraic Approach to Refined Topological Vertex}}, \doihref{http://dx.doi.org/10.1007/JHEP03(2012)041}{JHEP \textbf{03} (2012) 041}, \href{http://arxiv.org/abs/1112.6074}{{\arxivfont arXiv:1112.6074 [hep-th]}}.

\bibitem[AGT09]{Alday:2009aq}
L.~F. Alday, D.~Gaiotto, and Y.~Tachikawa, \emph{{Liouville Correlation Functions from Four-dimensional Gauge Theories}}, \doihref{http://dx.doi.org/10.1007/s11005-010-0369-5}{Lett. Math. Phys. \textbf{91} (2010) 167--197}, \href{http://arxiv.org/abs/0906.3219}{{\arxivfont arXiv:0906.3219 [hep-th]}}.

\bibitem[AHDM78]{Atiyah:1978ri}
M.~F. Atiyah, N.~J. Hitchin, V.~G. Drinfeld, and Y.~I. Manin, \emph{{Construction of Instantons}}, \doihref{http://dx.doi.org/10.1016/0375-9601(78)90141-X}{Phys. Lett. A \textbf{65} (1978) 185--187}.

\bibitem[AK08]{Awata:2008ed}
H.~Awata and H.~Kanno, \emph{{Refined BPS state counting from Nekrasov's formula and Macdonald functions}}, \doihref{http://dx.doi.org/10.1142/S0217751X09043006}{Int. J. Mod. Phys. \textbf{A24} (2009) 2253--2306},
\href{http://arxiv.org/abs/0805.0191}{{\arxivfont arXiv:0805.0191 [hep-th]}}.

\bibitem[AK09]{Awata:2009dd}
\bysame, \emph{{Quiver Matrix Model and Topological Partition Function in Six Dimensions}}, \doihref{http://dx.doi.org/10.1088/1126-6708/2009/07/076}{JHEP \textbf{07} (2009) 076}, \href{http://arxiv.org/abs/0905.0184}{{\arxivfont arXiv:0905.0184 [hep-th]}}.

\bibitem[AKM{\etalchar{+}}18]{Awata:2018svb}
H.~Awata, H.~Kanno, A.~Mironov, A.~Morozov, K.~Suetake, and Y.~Zenkevich, \emph{{The MacMahon $R$-matrix}}, \doihref{http://dx.doi.org/10.1007/JHEP04(2019)097}{JHEP \textbf{04} (2019) 097}, \href{http://arxiv.org/abs/1810.07676}{{\arxivfont arXiv:1810.07676 [hep-th]}}.

\bibitem[AKMV03]{Aganagic:2003db}
M.~Aganagic, A.~Klemm, M.~Marino, and C.~Vafa, \emph{{The Topological vertex}}, \doihref{http://dx.doi.org/10.1007/s00220-004-1162-z}{Commun. Math. Phys. \textbf{254} (2005) 425--478},
\href{http://arxiv.org/abs/hep-th/0305132}{{\arxivfont arXiv:hep-th/0305132}}.

\bibitem[AKOS95]{Awata:1995zk}
H.~Awata, H.~Kubo, S.~Odake, and J.~Shiraishi, \emph{{Quantum W(N) algebras and Macdonald polynomials}}, \doihref{http://dx.doi.org/10.1007/BF02102595}{Commun. Math. Phys. \textbf{179} (1996) 401--416}, \href{http://arxiv.org/abs/q-alg/9508011}{{\arxivfont arXiv:q-alg/9508011}}.

\bibitem[AKOS96]{Awata:1996dx}
H.~Awata, H.~Kubo, S.~Odake, and J.~Shiraishi, \emph{{Quantum deformation of the W(N) algebra}}, {Extended and Quantum Algebras and their Applications to Physics Tianjin, China, August 19-24, 1996}, 1996.
\href{http://arxiv.org/abs/q-alg/9612001}{{\arxivfont arXiv:q-alg/9612001}}.

\bibitem[AOS97]{Acharya:1997gp}
B.~S. Acharya, M.~O'Loughlin, and B.~J. Spence, \emph{{Higher dimensional analogs of Donaldson-Witten theory}}, \doihref{http://dx.doi.org/10.1016/S0550-3213(97)00515-4}{Nucl. Phys. B \textbf{503} (1997) 657--674}, \href{http://arxiv.org/abs/hep-th/9705138}{{\arxivfont arXiv:hep-th/9705138}}.

\bibitem[AY09]{Awata:2009ur}
H.~Awata and Y.~Yamada, \emph{{Five-dimensional AGT Conjecture and the Deformed Virasoro Algebra}}, \doihref{http://dx.doi.org/10.1007/JHEP01(2010)125}{JHEP \textbf{01} (2010) 125}, \href{http://arxiv.org/abs/0910.4431}{{\arxivfont arXiv:0910.4431 [hep-th]}}.

\bibitem[AY10]{Awata:2010yy}
\bysame, \emph{{Five-dimensional AGT Relation and the Deformed beta-ensemble}}, \doihref{http://dx.doi.org/10.1143/PTP.124.227}{Prog. Theor. Phys. \textbf{124} (2010) 227--262}, \href{http://arxiv.org/abs/1004.5122}{{\arxivfont arXiv:1004.5122 [hep-th]}}.

\bibitem[BBPT18]{Benini:2018hjy}
F.~Benini, G.~Bonelli, M.~Poggi, and A.~Tanzini, \emph{{Elliptic non-Abelian Donaldson-Thomas invariants of $\mathbb{C}^3$}}, \doihref{http://dx.doi.org/10.1007/JHEP07(2019)068}{JHEP \textbf{07} (2019) 068}, \href{http://arxiv.org/abs/1807.08482}{{\arxivfont arXiv:1807.08482 [hep-th]}}.

\bibitem[BDKZ23]{Bayindirli:2023byn}
M.~B. Bay\i{}nd\i{}rl\i{}, D.~N. Demirta\c{s}, C.~Koz\c{c}az, and Y.~Zenkevich, \emph{{On R-matrix formulation of qq-characters}}, \href{http://arxiv.org/abs/2310.02587}{{\arxivfont arXiv:2310.02587 [hep-th]}}.

\bibitem[BEHT13a]{Benini:2013nda}
F.~Benini, R.~Eager, K.~Hori, and Y.~Tachikawa, \emph{{Elliptic genera of two-dimensional N=2 gauge theories with rank-one gauge groups}}, \doihref{http://dx.doi.org/10.1007/s11005-013-0673-y}{Lett. Math. Phys. \textbf{104} (2014) 465--493}, \href{http://arxiv.org/abs/1305.0533}{{\arxivfont arXiv:1305.0533 [hep-th]}}.

\bibitem[BEHT13b]{Benini:2013xpa}
\bysame, \emph{{Elliptic Genera of 2d ${\mathcal{N}}$ = 2 Gauge Theories}}, \doihref{http://dx.doi.org/10.1007/s00220-014-2210-y}{Commun. Math. Phys. \textbf{333} (2015) 1241--1286}, \href{http://arxiv.org/abs/1308.4896}{{\arxivfont arXiv:1308.4896 [hep-th]}}.

\bibitem[BFF{\etalchar{+}}21]{Billo:2021xzh}
M.~Bill\`o, M.~Frau, F.~Fucito, L.~Gallot, A.~Lerda, and J.~F. Morales, \emph{{On the D(\textendash{}1)/D7-brane systems}}, \doihref{http://dx.doi.org/10.1007/JHEP04(2021)096}{JHEP \textbf{04} (2021) 096}, \href{http://arxiv.org/abs/2101.01732}{{\arxivfont arXiv:2101.01732 [hep-th]}}.

\bibitem[BFM15]{bershtein2018plane}
M.~Bershtein, B.~Feigin, and G.~Merzon, \emph{Plane partitions with a “pit”: generating functions and representation theory}, Selecta Mathematica \textbf{24} (2018) 21--62, \href{http://arxiv.org/abs/1512.08779}{{\arxivfont arXiv:1512.08779 [math]}}.

\bibitem[BFS15]{Belavin:2015ria}
V.~Belavin, O.~Foda, and R.~Santachiara, \emph{{AGT, N-Burge partitions and $ {\mathcal{W}}_N $ minimal models}}, \doihref{http://dx.doi.org/10.1007/JHEP10(2015)073}{JHEP \textbf{10} (2015) 073}, \href{http://arxiv.org/abs/1507.03540}{{\arxivfont arXiv:1507.03540 [hep-th]}}.

\bibitem[BKS97]{Baulieu:1997jx}
L.~Baulieu, H.~Kanno, and I.~M. Singer, \emph{{Special quantum field theories in eight-dimensions and other dimensions}}, \doihref{http://dx.doi.org/10.1007/s002200050353}{Commun. Math. Phys. \textbf{194} (1998) 149--175}, \href{http://arxiv.org/abs/hep-th/9704167}{{\arxivfont arXiv:hep-th/9704167}}.

\bibitem[Bou22]{Bourgine:2022scz}
J.-E. Bourgine, \emph{{Shifted Quantum Groups and Matter Multiplets in Supersymmetric Gauge Theories}}, \doihref{http://dx.doi.org/10.1007/s00220-023-04685-x}{Commun. Math. Phys. \textbf{401} (2023) 2051--2114}, \href{http://arxiv.org/abs/2205.01309}{{\arxivfont arXiv:2205.01309 [hep-th]}}.

\bibitem[BSY24]{Bao:2024ygr}
J.~Bao, R.-K. Seong, and M.~Yamazaki, \emph{{The origin of Calabi-Yau crystals in BPS states counting}}, \doihref{http://dx.doi.org/10.1007/JHEP03(2024)140}{JHEP \textbf{03} (2024) 140}, \href{http://arxiv.org/abs/2401.02792}{{\arxivfont arXiv:2401.02792 [hep-th]}}.

\bibitem[CKM19]{Cao:2019tvv}
Y.~Cao, M.~Kool, and S.~Monavari, \emph{{K-Theoretic DT/PT Correspondence for Toric Calabi\textendash{}Yau 4-Folds}}, \doihref{http://dx.doi.org/10.1007/s00220-022-04472-0}{Commun. Math. Phys. \textbf{396} (2022) 225--264}, \href{http://arxiv.org/abs/1906.07856}{{\arxivfont arXiv:1906.07856 [math.AG]}}.

\bibitem[CKM23]{Cao:2023gvn}
\bysame, \emph{{A Donaldson-Thomas crepant resolution conjecture on Calabi-Yau 4-folds}}, \doihref{http://dx.doi.org/10.1090/tran/9027}{Trans. Am. Math. Soc. \textbf{376} (2023) 8225--8268}, \href{http://arxiv.org/abs/2301.11629}{{\arxivfont arXiv:2301.11629 [math.AG]}}.

\bibitem[CS14]{Cordova:2014oxa}
C.~Cordova and S.-H. Shao, \emph{{An Index Formula for Supersymmetric Quantum Mechanics}}, \doihref{http://dx.doi.org/10.5427/jsing.2016.15b}{J. Singul. \textbf{15} (2016) 14--35}, \href{http://arxiv.org/abs/1406.7853}{{\arxivfont arXiv:1406.7853 [hep-th]}}.

\bibitem[CSS08]{Cirafici:2008sn}
M.~Cirafici, A.~Sinkovics, and R.~J. Szabo, \emph{{Cohomological gauge theory, quiver matrix models and Donaldson-Thomas theory}}, \doihref{http://dx.doi.org/10.1016/j.nuclphysb.2008.09.024}{Nucl. Phys. B \textbf{809} (2009) 452--518}, \href{http://arxiv.org/abs/0803.4188}{{\arxivfont arXiv:0803.4188 [hep-th]}}.

\bibitem[CWWY15]{Cai2015TheVO}
L.-Q. Cai, L.~Wang, K.~Wu, and J.~Yang, \emph{The vertex operator for a generalization of macmahon’s formula}, \href{https://api.semanticscholar.org/CorpusID:124268454}{International Journal of Modern Physics A \textbf{30} (2015) 1550176}.

\bibitem[CZ23]{Cao:2023lon}
Y.~Cao and G.~Zhao, \emph{{Quasimaps to quivers with potentials}}, \href{http://arxiv.org/abs/2306.01302}{{\arxivfont arXiv:2306.01302 [math.AG]}}.

\bibitem[Der64]{Derrick:1964ww}
G.~H. Derrick, \emph{{Comments on nonlinear wave equations as models for elementary particles}}, \doihref{http://dx.doi.org/10.1063/1.1704233}{J. Math. Phys. \textbf{5} (1964) 1252--1254}.

\bibitem[DHJV16]{Dijkgraaf:2016lym}
R.~Dijkgraaf, B.~Heidenreich, P.~Jefferson, and C.~Vafa, \emph{{Negative Branes, Supergroups and the Signature of Spacetime}}, \doihref{http://dx.doi.org/10.1007/JHEP02(2018)050}{JHEP \textbf{02} (2018) 050}, \href{http://arxiv.org/abs/1603.05665}{{\arxivfont arXiv:1603.05665 [hep-th]}}.

\bibitem[DI96]{ding1997generalization}
J.~Ding and K.~Iohara, \emph{{Generalization of Drinfeld quantum affine algebras}}, \doihref{http://dx.doi.org/10.1023/A:1007341410987}{Lett. Math. Phys. \textbf{41} (1997) 181--193}, \href{http://arxiv.org/abs/q-alg/9608002}{{\arxivfont q-alg/9608002}}.

\bibitem[Dou95]{Douglas:1995bn}
M.~R. Douglas, \emph{{Branes within branes}}, NATO Sci. Ser. C \textbf{520} (1999) 267--275, \href{http://arxiv.org/abs/hep-th/9512077}{{\arxivfont arXiv:hep-th/9512077}}.

\bibitem[Dou96]{Douglas:1996uz}
\bysame, \emph{{Gauge fields and D-branes}}, \doihref{http://dx.doi.org/10.1016/S0393-0440(97)00024-7}{J. Geom. Phys. \textbf{28} (1998) 255--262}, \href{http://arxiv.org/abs/hep-th/9604198}{{\arxivfont arXiv:hep-th/9604198}}.

\bibitem[DZNPZ21]{DelZotto:2021gzy}
M.~Del~Zotto, N.~Nekrasov, N.~Piazzalunga, and M.~Zabzine, \emph{{Playing With the Index of M-Theory}}, \doihref{http://dx.doi.org/10.1007/s00220-022-04479-7}{Commun. Math. Phys. \textbf{396} (2022) 817--865}, \href{http://arxiv.org/abs/2103.10271}{{\arxivfont arXiv:2103.10271 [hep-th]}}.

\bibitem[FFJ{\etalchar{+}}10]{Feigin2011}
B.~Feigin, E.~Feigin, M.~Jimbo, T.~Miwa, and E.~Mukhin, \emph{{Quantum continuous $gl(\infty)$ : Tensor products of Fock modules and $W_n$ characters}}, Kyoto Journal of Mathematics \textbf{51} (2011) 365--392, \href{http://arxiv.org/abs/1002.3113}{{\arxivfont arXiv:1002.3113 [math.QA]}}.

\bibitem[FHS{\etalchar{+}}10]{FHSSY:2010}
B.~Feigin, A.~Hoshino, J.~Shibahara, J.~Shiraishi, and S.~Yanagida, \emph{{Kernel function and quantum algebras}},
\href{http://arxiv.org/abs/1002.2485}{{\arxivfont arXiv:1002.2485 [math]}}.

\bibitem[FJMM10]{FFJMM1}
B.~Feigin, M.~Jimbo, T.~Miwa, and E.~Mukhin, \emph{{Quantum continuous $\mathfrak{gl}_\infty$: Semi-infinite construction of representations}}, \doihref{http://dx.doi.org/10.1215/21562261-1214375}{Kyoto J. Math. \textbf{51} (2011) 337--364},
\href{http://arxiv.org/abs/1002.3100}{{\arxivfont arXiv:1002.3100 [math.QA]}}.

\bibitem[FJMM11]{Feigin2011plane}
\bysame, \emph{Quantum toroidal $gl_1$-algebra: Plane partitions}, \doihref{http://dx.doi.org/10.1215/21562261-1625217}{Kyoto J. Math. \textbf{52} (2012) 621--659}, \href{http://arxiv.org/abs/1110.5310}{{\arxivfont arXiv:1110.5310}}.

\bibitem[FJMM15]{Feigin:2015raa}
\bysame, \emph{{Quantum toroidal $\mathfrak{g}{{\mathfrak{l}}_{1}}$ and Bethe ansatz}}, \doihref{http://dx.doi.org/10.1088/1751-8113/48/24/244001}{J. Phys. \textbf{A48} (2015) 244001},
\href{http://arxiv.org/abs/1502.07194}{{\arxivfont arXiv:1502.07194 [math.QA]}}.

\bibitem[FJMM16]{Feigin:2016pld}
\bysame, \emph{{Finite Type Modules and Bethe Ansatz for Quantum Toroidal ${\mathfrak{gl}_1}$}}, \doihref{http://dx.doi.org/10.1007/s00220-017-2984-9}{Commun. Math. Phys. \textbf{356} (2017) 285--327}, \href{http://arxiv.org/abs/1603.02765}{{\arxivfont arXiv:1603.02765 [math.QA]}}.

\bibitem[FM18]{Foda:2018jwz}
O.~Foda and M.~Manabe, \emph{{Macdonald topological vertices and brane condensates}}, \doihref{http://dx.doi.org/10.1016/j.nuclphysb.2018.10.001}{Nucl. Phys. B \textbf{936} (2018) 448--471}, \href{http://arxiv.org/abs/1801.04943}{{\arxivfont arXiv:1801.04943 [hep-th]}}.

\bibitem[FM23]{Fasola:2023ypx}
N.~Fasola and S.~Monavari, \emph{{Tetrahedron instantons in Donaldson-Thomas theory}}, \href{http://arxiv.org/abs/2306.07145}{{\arxivfont arXiv:2306.07145 [math.AG]}}.

\bibitem[FR97]{Frenkel:1997CMP}
E.~Frenkel and N.~Reshetikhin, \emph{{Deformations of {$\mathcal{W}$}-algebras associated to simple {L}ie algebras}}, Comm. Math. Phys. \textbf{197} (1998) 1--32, \href{http://arxiv.org/abs/q-alg/9708006}{{\arxivfont q-alg/9708006 [math.QA]}}.

\bibitem[FR98]{Frenkel:1998ojj}
\bysame, \doihref{http://dx.doi.org/10.1090/conm/248/03823}{\emph{{The {$q$}-characters of representations of quantum affine algebras and deformations of {$\mathcal{W}$}-algebras}}}, {Recent Developments in Quantum Affine Algebras and Related Topics}, Contemp. Math., vol. 248, Amer. Math. Soc., 1999, pp.~163--205. \href{http://arxiv.org/abs/math/9810055}{{\arxivfont math/9810055 [math.QA]}}.

\bibitem[Fra23]{Franco:2023tly}
S.~Franco, \emph{{4d crystal melting, toric Calabi-Yau 4-folds and brane brick models}}, \doihref{http://dx.doi.org/10.1007/JHEP03(2024)091}{JHEP \textbf{03} (2024) 091}, \href{http://arxiv.org/abs/2311.04404}{{\arxivfont arXiv:2311.04404 [hep-th]}}.

\bibitem[FW17]{Foda:2017tnv}
O.~Foda and J.-F. Wu, \emph{{A Macdonald refined topological vertex}}, \doihref{http://dx.doi.org/10.1088/1751-8121/aa7605}{J. Phys. A \textbf{50} (2017) 294003}, \href{http://arxiv.org/abs/1701.08541}{{\arxivfont arXiv:1701.08541 [hep-th]}}.

\bibitem[Gai09]{Gaiotto:2009we}
D.~Gaiotto, \emph{{N=2 dualities}}, \doihref{http://dx.doi.org/10.1007/JHEP08(2012)034}{JHEP \textbf{08} (2012) 034}, \href{http://arxiv.org/abs/0904.2715}{{\arxivfont arXiv:0904.2715 [hep-th]}}.

\bibitem[GK98]{Giveon:1998sr}
A.~Giveon and D.~Kutasov, \emph{{Brane Dynamics and Gauge Theory}}, \doihref{http://dx.doi.org/10.1103/RevModPhys.71.983}{Rev. Mod. Phys. \textbf{71} (1999) 983--1084}, \href{http://arxiv.org/abs/hep-th/9802067}{{\arxivfont arXiv:hep-th/9802067}}.

\bibitem[GLY21a]{Galakhov:2021xum}
D.~Galakhov, W.~Li, and M.~Yamazaki, \emph{{Shifted quiver Yangians and representations from BPS crystals}}, \doihref{http://dx.doi.org/10.1007/JHEP08(2021)146}{JHEP \textbf{08} (2021) 146}, \href{http://arxiv.org/abs/2106.01230}{{\arxivfont arXiv:2106.01230 [hep-th]}}.

\bibitem[GLY21b]{Galakhov:2021vbo}
\bysame, \emph{{Toroidal and elliptic quiver BPS algebras and beyond}}, \doihref{http://dx.doi.org/10.1007/JHEP02(2022)024}{JHEP \textbf{02} (2022) 024}, \href{http://arxiv.org/abs/2108.10286}{{\arxivfont arXiv:2108.10286 [hep-th]}}.

\bibitem[GR17]{Gaiotto:2017euk}
D.~Gaiotto and M.~Rap\v{c}\'ak, \emph{{Vertex Algebras at the Corner}}, \doihref{http://dx.doi.org/10.1007/JHEP01(2019)160}{JHEP \textbf{01} (2019) 160}, \href{http://arxiv.org/abs/1703.00982}{{\arxivfont arXiv:1703.00982 [hep-th]}}.

\bibitem[GSW88]{Green:1987sp}
M.~B. Green, J.~H. Schwarz, and E.~Witten, \emph{{SUPERSTRING THEORY. VOL. 1: INTRODUCTION}}, Cambridge Monographs on Mathematical Physics, 7 1988.

\bibitem[HKKP14]{Hwang:2014uwa}
C.~Hwang, J.~Kim, S.~Kim, and J.~Park, \emph{{General instanton counting and 5d SCFT}}, \doihref{http://dx.doi.org/10.1007/JHEP07(2015)063}{JHEP \textbf{07} (2015) 063}, \href{http://arxiv.org/abs/1406.6793}{{\arxivfont arXiv:1406.6793 [hep-th]}}. [Addendum: JHEP 04, 094 (2016)].

\bibitem[HKY14]{Hori:2014tda}
K.~Hori, H.~Kim, and P.~Yi, \emph{{Witten Index and Wall Crossing}}, \doihref{http://dx.doi.org/10.1007/JHEP01(2015)124}{JHEP \textbf{01} (2015) 124}, \href{http://arxiv.org/abs/1407.2567}{{\arxivfont arXiv:1407.2567 [hep-th]}}.

\bibitem[HLS08]{Haupt:2008nu}
A.~S. Haupt, A.~Lukas, and K.~S. Stelle, \emph{{M-theory on Calabi-Yau Five-Folds}}, \doihref{http://dx.doi.org/10.1088/1126-6708/2009/05/069}{JHEP \textbf{05} (2009) 069}, \href{http://arxiv.org/abs/0810.2685}{{\arxivfont arXiv:0810.2685 [hep-th]}}.

\bibitem[HM18]{Harada:2018bkb}
K.~Harada and Y.~Matsuo, \emph{{Plane partition realization of (web of) $ \mathcal{W} $ -algebra minimal models}}, \doihref{http://dx.doi.org/10.1007/JHEP02(2019)050}{JHEP \textbf{02} (2019) 050}, \href{http://arxiv.org/abs/1810.08512}{{\arxivfont arXiv:1810.08512 [hep-th]}}.

\bibitem[HMNW21]{Harada:2021xnm}
K.~Harada, Y.~Matsuo, G.~Noshita, and A.~Watanabe, \emph{{$q$-deformation of corner vertex operator algebras by Miura transformation}}, \doihref{http://dx.doi.org/10.1007/JHEP04(2021)202}{JHEP \textbf{04} (2021) 202}, \href{http://arxiv.org/abs/2101.03953}{{\arxivfont arXiv:2101.03953 [hep-th]}}.

\bibitem[IKV07]{Iqbal:2007ii}
A.~Iqbal, C.~Kozcaz, and C.~Vafa, \emph{{The Refined topological vertex}}, \doihref{http://dx.doi.org/10.1088/1126-6708/2009/10/069}{JHEP \textbf{10} (2009) 069},
\href{http://arxiv.org/abs/hep-th/0701156}{{\arxivfont arXiv:hep-th/0701156}}.

\bibitem[INOV03]{Iqbal:2003ds}
A.~Iqbal, N.~Nekrasov, A.~Okounkov, and C.~Vafa, \emph{{Quantum foam and topological strings}}, \doihref{http://dx.doi.org/10.1088/1126-6708/2008/04/011}{JHEP \textbf{04} (2008) 011}, \href{http://arxiv.org/abs/hep-th/0312022}{{\arxivfont arXiv:hep-th/0312022}}.

\bibitem[Jaf07]{Jafferis:2007sg}
D.~L. Jafferis, \emph{{Topological Quiver Matrix Models and Quantum Foam}}, \href{http://arxiv.org/abs/0705.2250}{{\arxivfont arXiv:0705.2250 [hep-th]}}.

\bibitem[JK93]{Jeffrey1993LocalizationFN}
L.~C. Jeffrey and F.~Kirwan, \emph{Localization for nonabelian group actions}, \href{https://api.semanticscholar.org/CorpusID:6929239}{Topology \textbf{34} (1993) 291--327}.

\bibitem[JKN{\etalchar{+}}24]{Noshita-Nawata}
J.~Jiang, T.~Kimura, S.~Nawata, G.~Noshita, and J.~Zheng, \emph{{In preparation}}.

\bibitem[JLN21]{Jeong:2021rll}
S.~Jeong, N.~Lee, and N.~Nekrasov, \emph{{Intersecting defects in gauge theory, quantum spin chains, and Knizhnik-Zamolodchikov equations}}, \doihref{http://dx.doi.org/10.1007/JHEP10(2021)120}{JHEP \textbf{10} (2021) 120}, \href{http://arxiv.org/abs/2103.17186}{{\arxivfont arXiv:2103.17186 [hep-th]}}.

\bibitem[JN18]{Jeong:2018qpc}
S.~Jeong and N.~Nekrasov, \emph{{Opers, surface defects, and Yang-Yang functional}}, \doihref{http://dx.doi.org/10.4310/ATMP.2020.v24.n7.a4}{Adv. Theor. Math. Phys. \textbf{24} (2020) 1789--1916}, \href{http://arxiv.org/abs/1806.08270}{{\arxivfont arXiv:1806.08270 [hep-th]}}.

\bibitem[Kan20]{Kanno:2020ybd}
H.~Kanno, \emph{{Quiver matrix model of ADHM type and BPS state counting in diverse dimensions}}, \doihref{http://dx.doi.org/10.1093/ptep/ptaa079}{PTEP \textbf{2020} (2020) 11B104}, \href{http://arxiv.org/abs/2004.05760}{{\arxivfont arXiv:2004.05760 [hep-th]}}.

\bibitem[Kim16]{Kim:2016qqs}
H.-C. Kim, \emph{{Line defects and 5d instanton partition functions}}, \doihref{http://dx.doi.org/10.1007/JHEP03(2016)199}{JHEP \textbf{03} (2016) 199}, \href{http://arxiv.org/abs/1601.06841}{{\arxivfont arXiv:1601.06841 [hep-th]}}.

\bibitem[Kim19]{Kimura:2019hnw}
T.~Kimura, \emph{{Integrating over quiver variety and BPS/CFT correspondence}}, \doihref{http://dx.doi.org/10.1007/s11005-020-01261-5}{Lett. Math. Phys. \textbf{110} (2020) 1237--1255}, \href{http://arxiv.org/abs/1910.03247}{{\arxivfont arXiv:1910.03247 [hep-th]}}.

\bibitem[Kim20]{Kimura:2020jxl}
\bysame, \doihref{http://dx.doi.org/10.1007/978-3-030-76190-5}{\emph{{Instanton Counting, Quantum Geometry and Algebra}}}, Springer, 7 2021. \href{http://arxiv.org/abs/2012.11711}{{\arxivfont arXiv:2012.11711 [hep-th]}}.

\bibitem[Kim22a]{Kimura:2022spi}
\bysame, \emph{{Higgsing $qq$-character and irreducibility}}, \href{http://arxiv.org/abs/2205.08312}{{\arxivfont arXiv:2205.08312 [math.QA]}}.

\bibitem[Kim22b]{Kimura:2022zsm}
\bysame, \emph{{Double Quiver Gauge Theory and BPS/CFT Correspondence}}, \doihref{http://dx.doi.org/10.3842/SIGMA.2023.039}{SIGMA \textbf{19} (2023) 039}, \href{http://arxiv.org/abs/2212.03870}{{\arxivfont arXiv:2212.03870 [hep-th]}}.

\bibitem[Kim23]{Kimura:2023iup}
\bysame, \emph{{Aspects of supergroup gauge theory}}, \doihref{http://dx.doi.org/10.1142/S0217751X23300016}{Int. J. Mod. Phys. A \textbf{38} (2023) 2330001}, \href{http://arxiv.org/abs/2301.05927}{{\arxivfont arXiv:2301.05927 [hep-th]}}.

\bibitem[KN23]{Kimura:2023bxy}
T.~Kimura and G.~Noshita, \emph{{Gauge origami and quiver W-algebras}}, \doihref{http://dx.doi.org/10.1007/JHEP05(2024)208}{JHEP \textbf{05} (2024) 208}, \href{http://arxiv.org/abs/2310.08545}{{\arxivfont arXiv:2310.08545 [hep-th]}}.

\bibitem[KN24]{Kimura:2024osv}
\bysame, \emph{{Gauge origami and quiver W-algebras III: Donaldson--Thomas $qq$-characters}}, \href{http://arxiv.org/abs/2411.01987}{{\arxivfont arXiv:2411.01987 [hep-th]}}.

\bibitem[Koj19]{Kojima2019}
T.~Kojima, \emph{{Quadratic relations of the deformed $W$-superalgebra ${\cal W}_{q, t}(\mathfrak{sl}(2|1))$}}, \doihref{http://dx.doi.org/10.1063/1.5142516}{J. Math. Phys. \textbf{62} (2021) 051702}, \href{http://arxiv.org/abs/1912.03096}{{\arxivfont arXiv:1912.03096 [math.QA]}}.

\bibitem[Koj21]{Kojima2021}
\bysame, \emph{{Quadratic relations of the deformed $W$-superalgebra ${\cal W}_{q, t}\bigl(A(M,N)\bigr)$}}, \doihref{http://dx.doi.org/10.1088/1751-8121/ac129f}{J. Phys. A \textbf{54} (2021) 335201}, \href{http://arxiv.org/abs/2101.01110}{{\arxivfont arXiv:2101.01110 [math.QA]}}.

\bibitem[KP15]{Kimura:2015rgi}
T.~Kimura and V.~Pestun, \emph{{Quiver W-algebras}}, \doihref{http://dx.doi.org/10.1007/s11005-018-1072-1}{Lett. Math. Phys. \textbf{108} (2018) 1351--1381}, \href{http://arxiv.org/abs/1512.08533}{{\arxivfont arXiv:1512.08533 [hep-th]}}.

\bibitem[KP16]{Kimura:2016dys}
\bysame, \emph{{Quiver elliptic W-algebras}}, \doihref{http://dx.doi.org/10.1007/s11005-018-1073-0}{Lett. Math. Phys. \textbf{108} (2018) 1383--1405}, \href{http://arxiv.org/abs/1608.04651}{{\arxivfont arXiv:1608.04651 [hep-th]}}.

\bibitem[KP17]{Kimura:2017hez}
\bysame, \emph{{Fractional quiver W-algebras}}, \doihref{http://dx.doi.org/10.1007/s11005-018-1087-7}{Lett. Math. Phys. \textbf{108} (2018) 2425--2451}, \href{http://arxiv.org/abs/1705.04410}{{\arxivfont arXiv:1705.04410 [hep-th]}}.

\bibitem[KP19a]{Kimura:2019msw}
\bysame, \emph{{Super instanton counting and localization}}, \href{http://arxiv.org/abs/1905.01513}{{\arxivfont arXiv:1905.01513 [hep-th]}}.

\bibitem[KP19b]{Kimura:2019xzj}
\bysame, \emph{{Twisted reduction of quiver W-algebras}}, \href{http://arxiv.org/abs/1905.03865}{{\arxivfont arXiv:1905.03865 [hep-th]}}.

\bibitem[LF20]{LeFloch:2020uop}
B.~Le~Floch, \emph{{A slow review of the AGT correspondence}}, \doihref{http://dx.doi.org/10.1088/1751-8121/ac5945}{J. Phys. A \textbf{55} (2022) 353002}, \href{http://arxiv.org/abs/2006.14025}{{\arxivfont arXiv:2006.14025 [hep-th]}}.

\bibitem[Liu22]{Liu:2022gwf}
H.~Liu, \emph{{A Representation-Theoretic Approach to $qq$-Characters}}, \doihref{http://dx.doi.org/10.3842/SIGMA.2022.090}{SIGMA \textbf{18} (2022) 090}, \href{http://arxiv.org/abs/2203.07072}{{\arxivfont arXiv:2203.07072 [math.QA]}}.

\bibitem[LNS97]{Losev:1997tp}
A.~Losev, N.~Nekrasov, and S.~L. Shatashvili, \emph{{Issues in topological gauge theory}}, \doihref{http://dx.doi.org/10.1016/S0550-3213(98)00628-2}{Nucl. Phys. B \textbf{534} (1998) 549--611}, \href{http://arxiv.org/abs/hep-th/9711108}{{\arxivfont arXiv:hep-th/9711108}}.

\bibitem[LNS98]{Lossev:1997bz}
\bysame, \emph{{Testing Seiberg-Witten solution}}, NATO Sci. Ser. C \textbf{520} (1999) 359--372, \href{http://arxiv.org/abs/hep-th/9801061}{{\arxivfont arXiv:hep-th/9801061}}.

\bibitem[LS16]{Litvinov:2016mgi}
A.~Litvinov and L.~Spodyneiko, \emph{{On W algebras commuting with a set of screenings}}, \doihref{http://dx.doi.org/10.1007/JHEP11(2016)138}{JHEP \textbf{11} (2016) 138}, \href{http://arxiv.org/abs/1609.06271}{{\arxivfont arXiv:1609.06271 [hep-th]}}.

\bibitem[LY20]{Li:2020rij}
W.~Li and M.~Yamazaki, \emph{{Quiver Yangian from Crystal Melting}}, \doihref{http://dx.doi.org/10.1007/JHEP11(2020)035}{JHEP \textbf{11} (2020) 035}, \href{http://arxiv.org/abs/2003.08909}{{\arxivfont arXiv:2003.08909 [hep-th]}}.

\bibitem[Mik07]{miki2007q}
K.~Miki, \emph{{A ($q, \gamma$) analog of the $W_{1+\infty}$ algebra}}, \doihref{http://dx.doi.org/10.1063/1.2823979}{J. Math. Phys \textbf{48} (2007) 3520}.

\bibitem[MNNZ23]{DIMreview}
Y.~Matsuo, S.~Nawata, G.~Noshita, and R.-D. Zhu, \emph{{Quantum toroidal algebras and solvable structures in gauge/string theory}}, \href{http://arxiv.org/abs/2309.07596}{{\arxivfont arXiv:2309.07596 [hep-th]}}.

\bibitem[MNS97]{Moore:1997dj}
G.~W. Moore, N.~Nekrasov, and S.~Shatashvili, \emph{{Integrating over Higgs branches}}, \doihref{http://dx.doi.org/10.1007/PL00005525}{Commun. Math. Phys. \textbf{209} (2000) 97--121}, \href{http://arxiv.org/abs/hep-th/9712241}{{\arxivfont arXiv:hep-th/9712241}}.

\bibitem[Mon22]{Monavari:2022rtf}
S.~Monavari, \emph{{Canonical vertex formalism in DT theory of toric Calabi-Yau 4-folds}}, \doihref{http://dx.doi.org/10.1016/j.geomphys.2022.104466}{J. Geom. Phys. \textbf{174} (2022) 104466}, \href{http://arxiv.org/abs/2203.11381}{{\arxivfont arXiv:2203.11381 [math.AG]}}.

\bibitem[MR08]{Mozgovoy2008OnTN}
S.~Mozgovoy and M.~Reineke, \emph{{On the noncommutative Donaldson-Thomas invariants arising from brane tilings}}, \doihref{http://dx.doi.org/10.1016/j.aim.2009.10.001}{Adv. Math. \textbf{223} (2010) 1521--1544}, \href{http://arxiv.org/abs/0809.0117}{{\arxivfont arXiv:0809.0117 [math.AG]}}.

\bibitem[Nak99]{Nakajima:1999}
H.~Nakajima, \doihref{http://dx.doi.org/10.1090/ulect/018}{\emph{{Lectures on Hilbert Schemes of Points on Surfaces}}}, American Mathematical Society, 1999.

\bibitem[Nek02]{Nekrasov:2002qd}
N.~Nekrasov, \emph{{Seiberg-Witten prepotential from instanton counting}}, \doihref{http://dx.doi.org/10.4310/ATMP.2003.v7.n5.a4}{Adv. Theor. Math. Phys. \textbf{7} (2003) 831--864}, \href{http://arxiv.org/abs/hep-th/0206161}{{\arxivfont arXiv:hep-th/0206161}}.

\bibitem[Nek09]{Nekrasov:2009JJM}
\bysame, \emph{{Instanton partition functions and M-theory}}, \href{https://doi.org/10.1007/s11537-009-0853-9}{Japanese J. Math. \textbf{4} (2009) 63--93}.

\bibitem[Nek15]{Nekrasov:2015wsu}
\bysame, \emph{{BPS/CFT correspondence: non-perturbative Dyson-Schwinger equations and qq-characters}}, \doihref{http://dx.doi.org/10.1007/JHEP03(2016)181}{JHEP \textbf{03} (2016) 181}, \href{http://arxiv.org/abs/1512.05388}{{\arxivfont arXiv:1512.05388 [hep-th]}}.

\bibitem[Nek16]{Nekrasov:2016qym}
\bysame, \emph{{BPS/CFT correspondence II: Instantons at crossroads, moduli and compactness theorem}}, \doihref{http://dx.doi.org/10.4310/ATMP.2017.v21.n2.a4}{Adv. Theor. Math. Phys. \textbf{21} (2017) 503--583}, \href{http://arxiv.org/abs/1608.07272}{{\arxivfont arXiv:1608.07272 [hep-th]}}.

\bibitem[Nek17a]{Nekrasov:2016ydq}
\bysame, \emph{{BPS/CFT Correspondence III: Gauge Origami partition function and qq-characters}}, \doihref{http://dx.doi.org/10.1007/s00220-017-3057-9}{Commun. Math. Phys. \textbf{358} (2018) 863--894}, \href{http://arxiv.org/abs/1701.00189}{{\arxivfont arXiv:1701.00189 [hep-th]}}.

\bibitem[Nek17b]{Nekrasov:2017rqy}
\bysame, \emph{{BPS/CFT correspondence IV: sigma models and defects in gauge theory}}, \doihref{http://dx.doi.org/10.1007/s11005-018-1115-7}{Lett. Math. Phys. \textbf{109} (2019) 579--622}, \href{http://arxiv.org/abs/1711.11011}{{\arxivfont arXiv:1711.11011 [hep-th]}}.

\bibitem[Nek17c]{Nekrasov:2017gzb}
\bysame, \emph{{BPS/CFT correspondence V: BPZ and KZ equations from qq-characters}}, \href{http://arxiv.org/abs/1711.11582}{{\arxivfont arXiv:1711.11582 [hep-th]}}.

\bibitem[Nek17d]{Nekrasov:2017cih}
\bysame, \emph{{Magnificent four}}, \doihref{http://dx.doi.org/10.4310/ATMP.2020.v24.n5.a4}{Adv. Theor. Math. Phys. \textbf{24} (2020) 1171--1202}, \href{http://arxiv.org/abs/1712.08128}{{\arxivfont arXiv:1712.08128 [hep-th]}}.

\bibitem[Nek23]{Nekrasov:2023xzm}
\bysame, \emph{{Analytic continuation and supersymmetry}}, \href{http://arxiv.org/abs/2310.01654}{{\arxivfont arXiv:2310.01654 [hep-th]}}.

\bibitem[NN08]{Nagao:2010kx}
K.~Nagao and H.~Nakajima, \emph{{Counting invariant of perverse coherent sheaves and its wall-crossing}}, \doihref{http://dx.doi.org/10.1093/imrn/rnq195}{Int. Math. Res. Not. (2010) }, \href{http://arxiv.org/abs/0809.2992}{{\arxivfont arXiv:0809.2992 [math.AG]}}.

\bibitem[NO03]{Nekrasov:2003rj}
N.~Nekrasov and A.~Okounkov, \emph{{Seiberg-Witten theory and random partitions}}, \doihref{http://dx.doi.org/10.1007/0-8176-4467-9_15}{Prog. Math. \textbf{244} (2006) 525--596}, \href{http://arxiv.org/abs/hep-th/0306238}{{\arxivfont arXiv:hep-th/0306238}}.

\bibitem[NO14]{Nekrasov:2014nea}
\bysame, \emph{{Membranes and sheaves.}}, \doihref{http://dx.doi.org/10.14231/AG-2016-015}{Algebr. Geom. \textbf{3} (2016) 320--369}, \href{http://arxiv.org/abs/1404.2323}{{\arxivfont arXiv:1404.2323 [math.AG]}}.

\bibitem[Nos22]{Noshita:2022otp}
G.~Noshita, \emph{{Brane tilings and crystal representations of quiver quantum toroidal algebras}}, Master's thesis, University of Tokyo, 1 2022. \url{https://inspirehep.net/literature/2149846}.

\bibitem[NP12]{Nekrasov:2012xe}
N.~Nekrasov and V.~Pestun, \emph{{Seiberg-Witten Geometry of Four-Dimensional $\mathcal N=2$ Quiver Gauge Theories}}, \doihref{http://dx.doi.org/10.3842/SIGMA.2023.047}{SIGMA \textbf{19} (2023) 047}, \href{http://arxiv.org/abs/1211.2240}{{\arxivfont arXiv:1211.2240 [hep-th]}}.

\bibitem[NP16]{Nekrasov:2016gud}
N.~Nekrasov and N.~S. Prabhakar, \emph{{Spiked Instantons from Intersecting D-branes}}, \doihref{http://dx.doi.org/10.1016/j.nuclphysb.2016.11.014}{Nucl. Phys. B \textbf{914} (2017) 257--300}, \href{http://arxiv.org/abs/1611.03478}{{\arxivfont arXiv:1611.03478 [hep-th]}}.

\bibitem[NP18]{Nekrasov:2018xsb}
N.~Nekrasov and N.~Piazzalunga, \emph{{Magnificent Four with Colors}}, \doihref{http://dx.doi.org/10.1007/s00220-019-03426-3}{Commun. Math. Phys. \textbf{372} (2019) 573--597}, \href{http://arxiv.org/abs/1808.05206}{{\arxivfont arXiv:1808.05206 [hep-th]}}.

\bibitem[NP23]{Nekrasov:2023nai}
\bysame, \emph{{Global magni$4$icence, or: 4G Networks}}, \href{http://arxiv.org/abs/2306.12995}{{\arxivfont arXiv:2306.12995 [hep-th]}}.

\bibitem[NPS13]{Nekrasov:2013xda}
N.~Nekrasov, V.~Pestun, and S.~Shatashvili, \emph{{Quantum geometry and quiver gauge theories}}, \doihref{http://dx.doi.org/10.1007/s00220-017-3071-y}{Commun. Math. Phys. \textbf{357} (2018) 519--567}, \href{http://arxiv.org/abs/1312.6689}{{\arxivfont arXiv:1312.6689 [hep-th]}}.

\bibitem[NW21a]{Noshita:2021ldl}
G.~Noshita and A.~Watanabe, \emph{{A note on quiver quantum toroidal algebra}}, \doihref{http://dx.doi.org/10.1007/JHEP05(2022)011}{JHEP \textbf{05} (2022) 011}, \href{http://arxiv.org/abs/2108.07104}{{\arxivfont arXiv:2108.07104 [hep-th]}}.

\bibitem[NW21b]{Noshita:2021dgj}
\bysame, \emph{{Shifted quiver quantum toroidal algebra and subcrystal representations}}, \doihref{http://dx.doi.org/10.1007/JHEP05(2022)122}{JHEP \textbf{05} (2022) 122}, \href{http://arxiv.org/abs/2109.02045}{{\arxivfont arXiv:2109.02045 [hep-th]}}.

\bibitem[NY03a]{Nakajima:2003pg}
H.~Nakajima and K.~Yoshioka, \emph{{Instanton counting on blowup. 1.}}, \doihref{http://dx.doi.org/10.1007/s00222-005-0444-1}{Invent. Math. \textbf{162} (2005) 313--355}, \href{http://arxiv.org/abs/math/0306198}{{\arxivfont arXiv:math/0306198}}.

\bibitem[NY03b]{Nakajima:2003uh}
\bysame, \emph{{Lectures on instanton counting}}, {CRM Workshop on Algebraic Structures and Moduli Spaces}, 11 2003. \href{http://arxiv.org/abs/math/0311058}{{\arxivfont arXiv:math/0311058}}.

\bibitem[NY05]{Nakajima:2005fg}
\bysame, \emph{{Instanton counting on blowup. II. K-theoretic partition function}}, \doihref{http://dx.doi.org/10.1007/s00031-005-0406-0}{Transform. Groups \textbf{10} (2005) 489--519}, \href{http://arxiv.org/abs/math/0505553}{{\arxivfont arXiv:math/0505553}}.

\bibitem[ORV03]{Okounkov:2003sp}
A.~Okounkov, N.~Reshetikhin, and C.~Vafa, \emph{{Quantum Calabi-Yau and classical crystals}}, \doihref{http://dx.doi.org/10.1007/0-8176-4467-9_16}{Prog. Math. \textbf{244} (2006) 597},
\href{http://arxiv.org/abs/hep-th/0309208}{{\arxivfont arXiv:hep-th/0309208}}.

\bibitem[OS14]{Ohta:2014ria}
K.~Ohta and Y.~Sasai, \emph{{Exact Results in Quiver Quantum Mechanics and BPS Bound State Counting}}, \doihref{http://dx.doi.org/10.1007/JHEP11(2014)123}{JHEP \textbf{11} (2014) 123}, \href{http://arxiv.org/abs/1408.0582}{{\arxivfont arXiv:1408.0582 [hep-th]}}.

\bibitem[OT06]{Okuda:2006fb}
T.~Okuda and T.~Takayanagi, \emph{{Ghost D-branes}}, \doihref{http://dx.doi.org/10.1088/1126-6708/2006/03/062}{JHEP \textbf{03} (2006) 062}, \href{http://arxiv.org/abs/hep-th/0601024}{{\arxivfont arXiv:hep-th/0601024}}.

\bibitem[OY08]{Ooguri:2009ijd}
H.~Ooguri and M.~Yamazaki, \emph{{Crystal Melting and Toric Calabi-Yau Manifolds}}, \doihref{http://dx.doi.org/10.1007/s00220-009-0836-y}{Commun. Math. Phys. \textbf{292} (2009) 179--199}, \href{http://arxiv.org/abs/0811.2801}{{\arxivfont arXiv:0811.2801 [hep-th]}}.

\bibitem[Pes07]{Pestun:2007rz}
V.~Pestun, \emph{{Localization of gauge theory on a four-sphere and supersymmetric Wilson loops}}, \doihref{http://dx.doi.org/10.1007/s00220-012-1485-0}{Commun. Math. Phys. \textbf{313} (2012) 71--129}, \href{http://arxiv.org/abs/0712.2824}{{\arxivfont arXiv:0712.2824 [hep-th]}}.

\bibitem[Pol07a]{Polchinski:1998rq}
J.~Polchinski, \doihref{http://dx.doi.org/10.1017/CBO9780511816079}{\emph{{String theory. Vol. 1: An introduction to the bosonic string}}}, Cambridge Monographs on Mathematical Physics, Cambridge University Press, 12 2007.

\bibitem[Pol07b]{Polchinski:1998rr}
\bysame, \doihref{http://dx.doi.org/10.1017/CBO9780511618123}{\emph{{String theory. Vol. 2: Superstring theory and beyond}}}, Cambridge Monographs on Mathematical Physics, Cambridge University Press, 12 2007.

\bibitem[PR17]{Prochazka:2017qum}
T.~Proch\'azka and M.~Rap\v{c}\'ak, \emph{{Webs of W-algebras}}, \doihref{http://dx.doi.org/10.1007/JHEP11(2018)109}{JHEP \textbf{11} (2018) 109}, \href{http://arxiv.org/abs/1711.06888}{{\arxivfont arXiv:1711.06888 [hep-th]}}.

\bibitem[PR18]{Prochazka:2018tlo}
\bysame, \emph{{$ \mathcal{W} $ -algebra modules, free fields, and Gukov-Witten defects}}, \doihref{http://dx.doi.org/10.1007/JHEP05(2019)159}{JHEP \textbf{05} (2019) 159}, \href{http://arxiv.org/abs/1808.08837}{{\arxivfont arXiv:1808.08837 [hep-th]}}.

\bibitem[PYZ21]{Pomoni:2021hkn}
E.~Pomoni, W.~Yan, and X.~Zhang, \emph{{Tetrahedron Instantons}}, \doihref{http://dx.doi.org/10.1007/s00220-022-04376-z}{Commun. Math. Phys. \textbf{393} (2022) 781--838}, \href{http://arxiv.org/abs/2106.11611}{{\arxivfont arXiv:2106.11611 [hep-th]}}.

\bibitem[PYZ23]{Pomoni:2023nlf}
\bysame, \emph{{Probing M-theory with tetrahedron instantons}}, \href{http://arxiv.org/abs/2306.06005}{{\arxivfont arXiv:2306.06005 [hep-th]}}.

\bibitem[PZ{\etalchar{+}}16]{Pestun:2016zxk}
V.~Pestun, M.~Zabzine, et~al., \emph{{Localization techniques in quantum field theories}}, \doihref{http://dx.doi.org/10.1088/1751-8121/aa63c1}{J. Phys. A \textbf{50} (2017) 440301}, \href{http://arxiv.org/abs/1608.02952}{{\arxivfont arXiv:1608.02952 [hep-th]}}.

\bibitem[RSW21]{Raghavendran:2021qbh}
S.~Raghavendran, I.~Saberi, and B.~R. Williams, \emph{{Twisted Eleven-Dimensional Supergravity}}, \doihref{http://dx.doi.org/10.1007/s00220-023-04745-2}{Commun. Math. Phys. \textbf{402} (2023) 1103--1166}, \href{http://arxiv.org/abs/2111.03049}{{\arxivfont arXiv:2111.03049 [math-ph]}}.

\bibitem[RSYZ18]{Rapcak:2018nsl}
M.~Rap\v{c}\'ak, Y.~Soibelman, Y.~Yang, and G.~Zhao, \emph{{Cohomological Hall algebras, vertex algebras and instantons}}, \doihref{http://dx.doi.org/10.1007/s00220-019-03575-5}{Commun. Math. Phys. \textbf{376} (2019) 1803--1873}, \href{http://arxiv.org/abs/1810.10402}{{\arxivfont arXiv:1810.10402 [math.QA]}}.

\bibitem[RSYZ20]{Rapcak:2020ueh}
\bysame, \emph{{Cohomological Hall algebras and perverse coherent sheaves on toric Calabi-Yau 3-folds}}, \href{http://arxiv.org/abs/2007.13365}{{\arxivfont arXiv:2007.13365 [math.QA]}}.

\bibitem[Sen98]{Sen:1998sm}
A.~Sen, \emph{{Tachyon condensation on the brane anti-brane system}}, \doihref{http://dx.doi.org/10.1088/1126-6708/1998/08/012}{JHEP \textbf{08} (1998) 012}, \href{http://arxiv.org/abs/hep-th/9805170}{{\arxivfont arXiv:hep-th/9805170}}.

\bibitem[Sha05]{Shadchin:2005mx}
S.~Shadchin, \emph{{On certain aspects of string theory/gauge theory correspondence}}, Other thesis, 2 2005. \href{http://arxiv.org/abs/hep-th/0502180}{{\arxivfont arXiv:hep-th/0502180}}.

\bibitem[SKAO95]{Shiraishi:1995rp}
J.~Shiraishi, H.~Kubo, H.~Awata, and S.~Odake, \emph{{A Quantum deformation of the Virasoro algebra and the Macdonald symmetric functions}}, \doihref{http://dx.doi.org/10.1007/BF00398297}{Lett. Math. Phys. \textbf{38} (1996) 33--51},
\href{http://arxiv.org/abs/q-alg/9507034}{{\arxivfont arXiv:q-alg/9507034}}.

\bibitem[ST22]{Szabo:2022zyn}
R.~J. Szabo and M.~Tirelli, \emph{{Noncommutative Instantons in Diverse Dimensions}}, \doihref{http://dx.doi.org/10.1140/epjs/s11734-023-00840-6}{Eur. Phys. J. Spec. Top. (2023) }, \href{http://arxiv.org/abs/2207.12862}{{\arxivfont arXiv:2207.12862 [hep-th]}}.

\bibitem[ST23]{Szabo:2023ixw}
\bysame, \emph{{Instanton Counting and Donaldson-Thomas Theory on Toric Calabi-Yau Four-Orbifolds}}, \href{http://arxiv.org/abs/2301.13069}{{\arxivfont arXiv:2301.13069 [hep-th]}}.

\bibitem[SV03]{Szenes2003ToricRA}
A.~Szenes and M.~Vergne, \emph{Toric reduction and a conjecture of batyrev and materov}, \href{https://api.semanticscholar.org/CorpusID:11930649}{Inventiones mathematicae \textbf{158} (2003) 453--495}.

\bibitem[Sze07]{Szendroi:2007nu}
B.~Szendroi, \emph{{Non-commutative Donaldson\textendash{}Thomas invariants and the conifold}}, \doihref{http://dx.doi.org/10.2140/gt.2008.12.1171}{Geom. Topol. \textbf{12} (2008) 1171--1202}, \href{http://arxiv.org/abs/0705.3419}{{\arxivfont arXiv:0705.3419 [math.AG]}}.

\bibitem[Tak07]{Taki:2007dh}
M.~Taki, \emph{{Refined Topological Vertex and Instanton Counting}}, \doihref{http://dx.doi.org/10.1088/1126-6708/2008/03/048}{JHEP \textbf{03} (2008) 048}, \href{http://arxiv.org/abs/0710.1776}{{\arxivfont arXiv:0710.1776 [hep-th]}}.

\bibitem[Tak14]{Taki:2014fva}
\bysame, \emph{{On AGT-W Conjecture and q-Deformed W-Algebra}}, \href{http://arxiv.org/abs/1403.7016}{{\arxivfont arXiv:1403.7016 [hep-th]}}.

\bibitem[Vaf01]{Vafa:2001qf}
C.~Vafa, \emph{{Brane / anti-brane systems and U($N|M$) supergroup}}, \href{http://arxiv.org/abs/hep-th/0101218}{{\arxivfont arXiv:hep-th/0101218}}.

\bibitem[Vul07]{Vuletic2007AGO}
M.~Vuleti'c, \emph{A generalization of macmahon's formula}, \href{https://api.semanticscholar.org/CorpusID:7337145}{Transactions of the American Mathematical Society \textbf{361} (2007) 2789--2804}.

\bibitem[VW94]{Vafa:1994tf}
C.~Vafa and E.~Witten, \emph{{A Strong coupling test of S duality}}, \doihref{http://dx.doi.org/10.1016/0550-3213(94)90097-3}{Nucl. Phys. B \textbf{431} (1994) 3--77}, \href{http://arxiv.org/abs/hep-th/9408074}{{\arxivfont arXiv:hep-th/9408074}}.

\bibitem[Wit93]{Witten:1993yc}
E.~Witten, \emph{{Phases of N=2 theories in two-dimensions}}, \doihref{http://dx.doi.org/10.1016/0550-3213(93)90033-L}{Nucl. Phys. B \textbf{403} (1993) 159--222}, \href{http://arxiv.org/abs/hep-th/9301042}{{\arxivfont arXiv:hep-th/9301042}}.

\bibitem[Wit94]{Witten:1994tz}
\bysame, \emph{{Sigma models and the ADHM construction of instantons}}, \doihref{http://dx.doi.org/10.1016/0393-0440(94)00047-8}{J. Geom. Phys. \textbf{15} (1995) 215--226}, \href{http://arxiv.org/abs/hep-th/9410052}{{\arxivfont arXiv:hep-th/9410052}}.

\bibitem[Wit95]{Witten:1995gx}
\bysame, \emph{{Small instantons in string theory}}, \doihref{http://dx.doi.org/10.1016/0550-3213(95)00625-7}{Nucl. Phys. B \textbf{460} (1996) 541--559}, \href{http://arxiv.org/abs/hep-th/9511030}{{\arxivfont arXiv:hep-th/9511030}}.

\bibitem[Wit00]{Witten:2000mf}
\bysame, \emph{{BPS Bound states of D0 - D6 and D0 - D8 systems in a B field}}, \doihref{http://dx.doi.org/10.1088/1126-6708/2002/04/012}{JHEP \textbf{04} (2002) 012}, \href{http://arxiv.org/abs/hep-th/0012054}{{\arxivfont arXiv:hep-th/0012054}}.

\bibitem[Wyl09]{Wyllard:2009hg}
N.~Wyllard, \emph{{$A_{(N-1)}$ conformal Toda field theory correlation functions from conformal N = 2 SU(N) quiver gauge theories}}, \doihref{http://dx.doi.org/10.1088/1126-6708/2009/11/002}{JHEP \textbf{11} (2009) 002}, \href{http://arxiv.org/abs/0907.2189}{{\arxivfont arXiv:0907.2189 [hep-th]}}.

\end{thebibliography}

\end{document}